\documentclass[11pt]{article}
\usepackage{amsfonts}
\usepackage{amsmath,amsthm}
\usepackage{graphicx,epstopdf,epsfig,multirow,epic,bm}
\usepackage{color}
\usepackage{multicol}
\usepackage{algorithm}
\usepackage{algpseudocode}
\usepackage{paralist}

\usepackage{makeidx}
\usepackage{showidx}
\usepackage{textcomp}

\theoremstyle{definition}
\newtheorem{thm}{Theorem}
\newtheorem{cor}[thm]{Corollary}
\newtheorem{lem}[thm]{Lemma}
\newtheorem{prop}[thm]{Proposition}
\theoremstyle{definition}
\newtheorem{defn}{Definition}
\theoremstyle{definition}
\newtheorem{rem}{Remark}

\theoremstyle{definition}
\newtheorem{problem}{Problem}

\theoremstyle{definition}
\newtheorem{conj}{Conjecture}
\theoremstyle{definition}
\newtheorem{example}{Example}
\theoremstyle{definition}

\DeclareGraphicsExtensions{.eps,.eps.gz}
\normalsize \rm
\usepackage{titlesec}

\newcommand{\qqed}{\hfill $\square$}
\newcommand{\paralled}{\hfill $\parallel$}

\usepackage[nottoc]{tocbibind} 

\begin{document}

\begin{center}
{\huge \textbf{Topological Authentication Technique In\\[12pt]
Topologically Asymmetric Cryptosystem}}\\[14pt]
{\Large \textbf{Bing YAO, Jing SU, Fei MA, Hongyu WANG, Chao YANG}\\[8pt]
(\today)}
\end{center}

\vskip 1cm

\pagenumbering{roman}
\tableofcontents

\newpage

\setcounter{page}{1}
\pagenumbering{arabic}

\begin{center}
{\huge \textbf{Topological Authentication Technique In\\[12pt]
Topologically Asymmetric Cryptosystem}}\\[14pt]
{\large \textbf{Bing Yao $^{1,\dagger}$~ Jing Su $^{2,\diamond}$~ Fei Ma $^{2,\ddagger}$~ Hongyu Wang $^{2,3,\star}$~ Chao Yang $^{3,4,\ast}$}}\\[12pt]
{1. College of Mathematics and Statistics, Northwest Normal University, Lanzhou, 730070 CHINA\\
2. School of Electronics Engineering and Computer Science, Peking University, Beijing, 100871, CHINA\\
3. National Computer Network Emergency Response Technical Team/Coordination Center of China, Beijing, 100029, CHINA \\
4. School of Mathematics, Physics and Statistics, Shanghai University of Engineering Science,
Shanghai, 201620, CHINA\\
5. Center of Intelligent Computing and Applied Statistics, Shanghai University of Engineering Science,
Shanghai, 201620, CHINA \footnote{~$^{\dagger}$ yybb918@163.com; $^{\diamond}$ 1099270659@qq.com; $^{\ddagger}$ mafei123987@163.com;\\ $^{\star}$ why5126@pku.edu.cn; $^{\star}$ yangchaomath0524@163.com}\\[8pt]
(\today)}
\end{center}

\vskip 1cm

\begin{quote}
\textbf{Abstract:} Making topological authentication from theory to practical application is an important and challenging task. More and more researchers pay attention on coming quantum computation, privacy data protection, lattices and cryptography. Research show the advantages of topological authentications through graph operations, various matrices, graph colorings and graph labelings are: related with two or more different mathematical areas, be not pictures, there are huge number of colorings and labelings, rooted on modern mathematics, diversity of asymmetric ciphers, simplicity and convenience, easily created, irreversibility, computational security, provable security, and so on. Topological authentications based on various graph homomorphisms, degree-sequence homomorphisms, graph-set homomorphisms. Randomly topological coding and topological authentications are based on Hanzi authentication, randomly adding-edge-removing operation, randomly leaf-adding algorithms, graph random increasing techniques, operation graphic lattice and dynamic networked models and their spanning trees and maximum leaf spanning trees. Realization of topological authentication is an important topic, we study: number-based strings generated from colored graphs, particular graphs (complete graphs, trees, planar graphs), some methods of generating public-keys. some techniques of topologically asymmetric cryptosystem are: W-type matching labelings, dual-type labelings, reciprocal-type labelings, topological homomorphisms, indexed colorings, graphic lattices, degree-sequence lattices, every-zero Cds-matrix groups of degree-sequences, every-zero graphic groups, graphic lattices having coloring closure property, self-similar networked lattices. We propose two problems: Parameterized Number-based String Partition Problem (PNBSPP) and Number-Based String No-order Decomposition Problem (NBSNODP). It is preliminarily concluded that the graph lattice based on PNBSPP and NBSNODP can resist quantum computing. We, also, condenser the complexity of the above techniques and algorithms, in which many of them are NP-problems and $O(2^n)$.
\end{quote}

\section{Introduction, Basic concepts and operations}

\subsection{Research background}

All civilizations created by human beings can be digitized and preserved, we are in a digital age. We are facing important researching topics of coming quantum computation, lattices and cryptography,  privacy computation, privacy computation and hypergraphs and hypernetworks.

\subsubsection{Coming quantum computation}

We are poised at a singular moment in the development of quantum computing, with small-scale to medium-scale quantum computers poised to become a reality. This has thrown up new algorithmic challenges from benchmarking and testing of quantum computers, the theory of interactive quantum device testing has already led to deep new connections with cryptography and computational complexity theory, and the resolution of some long-standing open questions in mathematics (ref. \cite{Bernstein-Buchmann-Dahmen-Quantum-2009}).

Bennett and DiVincenzo pointed \cite{Bennett-DiVincenzo-2000-Nature}: In information processing, as in physics, our classical world view provides an incomplete approximation to an underlying quantum reality. Quantum effects like interference and entanglement play no direct role in conventional information processing, but they can -- in principle now, but probably eventually in practice -- be harnessed to break codes, create unbreakable codes, and speed up otherwise intractable computations.

\subsubsection{Lattices and Cryptography}

As known, lattices have many significant applications in pure mathematics, particularly in connection to Lie algebras, number theory and group theory. They also arise in applied mathematics in connection with coding theory, in cryptography because of conjectured computational hardness of several lattice problems, and are used in various ways in the physical sciences (Ref. Wikipedia and \cite{Daniele-Micciancio-Oded-Regev}).

We have seen a report: Simons Institute, in the 2021 Quantum Computation cluster, brings together researchers from computer science, physics, chemistry, and mathematics to advance the science in these areas. Workshops in Simons Institute are:

(1) Lattices: Algorithms, Complexity and Cryptography Boot Camp;

(2) Lattices: Geometry, Algorithms and Hardness;

(3) Quantum Cryptanalysis of Post-Quantum Cryptography;

(4) Lattices: New Cryptographic Capabilities;

(5) Lattices: From Theory to Practice;

(6) Lattices: Algorithms, Complexity, and Cryptography Reunion.

Notice that the lattice difficulty problem is not only a classical number theory, but also an important research topic of computational complexity theory. Researchers have found that lattice theory has a wide range of applications in cryptanalysis and design. Many difficult problems in lattice have been proved to be NP-hard. So, this kind of cryptosystems are generally considered to have the characteristics of quantum attack resistance (Ref. \cite{Wang-Xiao-Yun-Liu-2014}).

Motivated from the traditional lattices and the topological authentications, the author, in \cite{Bing-Yao-2020arXiv}, introduced graphic lattices, graphic group lattices, Topcode-matrix lattices, matching-type graphic lattices, graphic lattice sequences and other type lattices made by various graph operations, matrix operations and group operations.

\subsubsection{Privacy computation}

\textbf{Privacy Computation} aims to resolve the contradiction between the right to use and the right to own data and deal with users' privacy data through encryption processing, multi-party computing and other methods. Data users (such as Internet platforms) no longer get the user's original data, but encrypted data.

As a solution to promote the safe and orderly flow of data, the core value of \textbf{Privacy Computation} is the ability to achieve ``Data can be used, but not visible'' and ``Data doesn't move, but Model moves'', with the ability to break data islands, strengthen privacy protection, and strengthen data security compliance.

\textbf{Privacy computing system} involves three key technologies: blockchain, federated learning and secure multi-party computation. Blockchain is a basic technology to better solve the problem of mutual trust among parties. It is widely believed that privacy computing will be deeply combined with blockchain technology. Federated learning solves the problem of data union modeling and secure multi-party computing solves the problem of multi-party data fusion, both of which are the basic technologies of privacy computing.

\textbf{Cloud computation} provides almost unlimited space for storing and managing data, and provides almost unlimited computing power for people to complete applications. In cryptography, there is no absolutely secure password, that is, there is no password that can not be broken, to break a password in a short time, it requires a large number of high-performance computers to work together. However, cloud computing has super unlimited expansion of computing capacity for solving the problem of insufficient computing capacity. It is clear that attackers can use the infinite expansion and computing power of cloud computing to brute force passwords.

There is no doubt that privacy protection is becoming more and more important as the world goes digital.

\subsubsection{Hypergraphs and hypernetworks}

In \cite{Bing-Yao-Fei-Ma-13354v1-2022}, the authors have applied set-colorings, intersected-graphs and intersected-networks to observe indirectly and characterize properties of hypergraphs and hypernetworks, and hope intersected-graphs and intersected-networks can be applied to real situations, such as network security, privacy protection, as well as anti-quantum computing in asymmetric cryptography, even metaverse today. Predictably, hypergraph theory will be an important application in the future resisting AI attacks equipped quantum computer, and will be developed quickly in ``Graphs are codes, and codes are graphs''.

\subsection{Examples from topological coding}

In this article, ``a topological authentication = a topological public-key + a topological private-key'' will be discussed by technology of graph theory.

\subsubsection{Examples of graph theory}

\begin{example}\label{exa:dynamic-11}
An uncolored graph $G$ and its own \emph{complement} (or \emph{inverse}) $\overline{G}$ form a \emph{topological authentication} $K_n$ such that $E(K_n)=E(G)\cup E(\overline{G})$ with $E(G)\cap E(\overline{G})=\emptyset$ and $V(K_n)=V(G)=V(\overline{G})$, where $n=|V(G)|=|V(\overline{G})|$, and $G$ is called a \emph{topological public-key} and $\overline{G}$ is called a \emph{topological private-key}, write this case as $\langle G, \overline{G}\mid K_n\rangle$.
\end{example}

\begin{example}\label{exa:dynamic-11}
A \emph{universal graph} $F$ has two edge-disjoint subgraphs $G$ and $H$, then $F$ forms a \emph{topological authentication} if $V(F)=V(G)\cup V(H)$ and $E(F)=E(G)\cup E(H)$ with $E(G)\cap E(H)=\emptyset $, where $G$ is called a \emph{topological public-key} and $H$ is called a \emph{topological private-key}, write this case as $\langle G, H\mid F\rangle$.
\end{example}

\begin{example}\label{exa:graph-split-families}
\textbf{Graph-split families}: We vertex-split a connected graph $G$ into vertex disjoint connected subgraphs $G_{i,1},G_{i,2},\dots ,G_{i,n_i}$ with $|E(G_{i,j})|\geq 1$ and $n_i\geq 2$, denoted as $\textbf{G}_i=(G_{i,1},G_{i,2},\dots $, $G_{i,n_i})$, and we call $\textbf{G}_i$ a \emph{graph-split family}, and write $\textbf{G}_i=\wedge \langle G\rangle $ (Ref. \cite{Wang-Su-Yao-2021-computer-science}). There is an integer $M_{vsp}(G)$ such that $G$ has just $M_{vsp}(G)$ different graph-split families, we use a set $S_{plit}(G)$ to collect all of different graph-split families of $G$, namely, $S_{plit}(G)=\{\textbf{G}_i:~i\in [1,M_{vsp}(G)]\}$. Do \emph{non-common neighbor vertex-coinciding operation} ``$\odot$'' defined in Definition \ref{defn:vertex-split-coinciding-operations} to a split-graph family $\textbf{G}_i$ for getting the original graph $G$, such that
\begin{equation}\label{eqa:graph-split-families}
G=\odot \langle \textbf{G}_i\rangle =\odot ^{n_i}_{j=1}G_{i,j}, \quad \textbf{G}_i\rightarrow G
\end{equation} where $\textbf{G}_i\rightarrow G$ is a \emph{graph homomorphism} from $\textbf{G}_i$ into $G$ (refer to Definition \ref{defn:definition-graph-homomorphism}). Here, the connected graph $G$ is as a \emph{topological public-key}, and a graph-split family $\textbf{G}_i$ with particular properties is just a \emph{topological private-key}, so $G$ is the desired \emph{topological authentication}. If $G$ admits a total coloring, so does each graph $G_{i,j}$ of each graph-split family $\textbf{G}_i$ too. As doing \emph{non-common neighbor vertex-coinciding operation} ``$\odot$'' defined in Definition \ref{defn:vertex-split-coinciding-operations}, two vertices from two vertex disjoint connected subgraphs $G_{i,j}$ and $G_{i,s}$ are vertex-coincided into one vertex if they are colored with the same color.
\end{example}

\begin{problem}\label{qeu:11-graph-split-families}
\textbf{Find} out all elements of the \emph{graph-split family set} $S_{plit}(G)$ defined in Example \ref{exa:graph-split-families}.
\end{problem}

\subsubsection{Examples from topological authentication}

We use $H_{abcd}$ to denote the graph of a Hanzi (also, Chinese latter) with number-based string $abcd$ defined in \cite{GB2312-80}, called a \emph{Hanzi-graph}, and a Hanzi-gpw is the result of coloring a Hanzi-graph $H_{abcd}$ by one of the existing $W$-type colorings and $W$-type labelings in graph theory.

\begin{figure}[h]
\centering
\includegraphics[width=13cm]{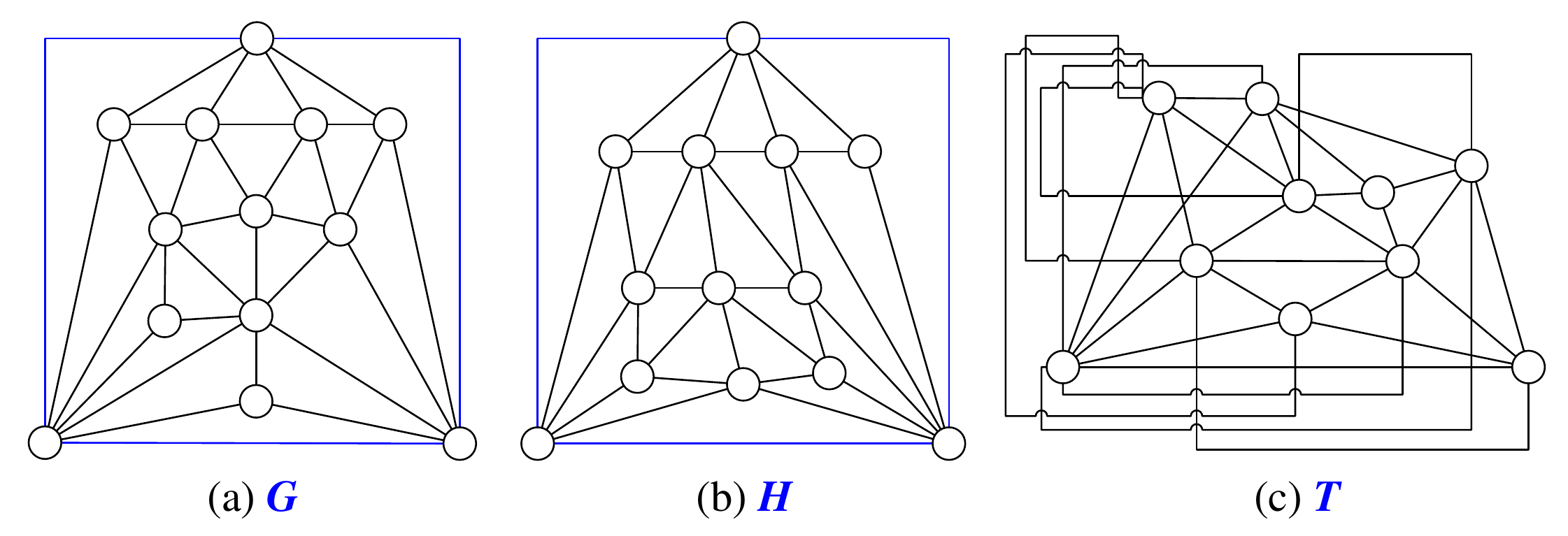}\\
\caption{\label{fig:graphs-33-edges}{\small (a) and (b) are two uncolored maximal planar graphs of 33 edges; (c) an uncolored non-planar graph having 33 edges.}}
\end{figure}

Each graph of three graphs $G$, $H$ and $T$ shown in Fig.\ref{fig:graphs-33-edges} has 33 edges in total, however, they are not isomorphic from each other, namely, $G\not\cong H$, $G\not\cong T$ and $H\not\cong T$. By the vertex-splitting operation and the vertex-coinciding operation defined in Definition \ref{defn:vertex-split-coinciding-operations}, each of them can corresponds a Chinese sentence made by some Hanzi-graphs shown in Fig.\ref{fig:6-hanzis}, and these Hanzi-graphs have 33 edges in total. For example, a Chinese sentence ``$H_{4043}H_{5448}H_{3164}H_{2201}$'' made by four Hanzi-graphs $H_{4043}$, $H_{5448}$, $H_{3164}$ and $H_{2201}$ shown in Fig.\ref{fig:6-hanzis} with 33 edges in total; two Hanzi-graphs $H_{3164}$ and $H_{1872}$ with 33 edges form a Chinese sentence ``$H_{3164}H_{1872}$'' that corresponds each of three graphs $G$, $H$ and $T$ shown in Fig.\ref{fig:graphs-33-edges}. And moreover, four Hanzi-graphs $H_{3164},H_{2201},H_{5181}$ and $H_{4043}$ produce a meaningful sentence ``$H_{3164}H_{2201}H_{5181}H_{4043}$'' in Chinese.

\begin{figure}[h]
\centering
\includegraphics[width=16.4cm]{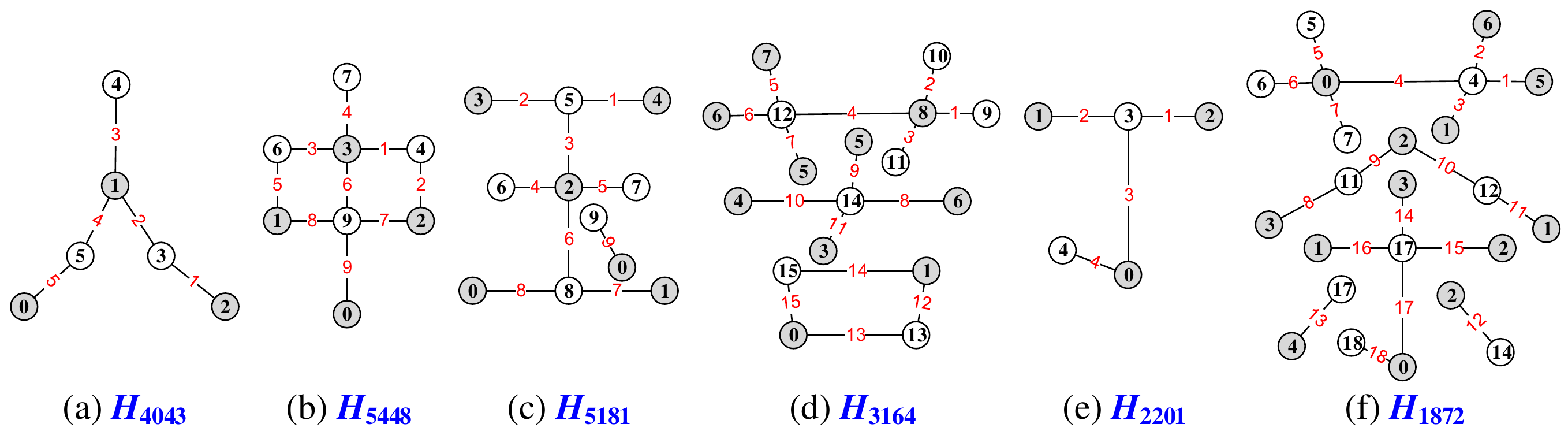}\\
\caption{\label{fig:6-hanzis}{\small (a)-(d) are six Hanzi-graphs $H_{4043}$, $H_{5448}$, $H_{3164}$, $H_{2201}$, $H_{1872}$ and $H_{5181}$.}}
\end{figure}

By a fixed rule, we get a number-based string $S_{3164}(65)$ shown in Eq.(\ref{eqa:rzykdc-00}) from the Topcode-matrix $T_{code}(H_{3164})$ shown in Eq.(\ref{eqa:Topcode-matrix-3164}), and get another number-based string $S^{odd}_{3164}(80)$ shown in Eq.(\ref{eqa:rzykdc-00}) from the Topcode-matrix $T_{code}(H^{odd}_{3164})$ shown in Eq.(\ref{eqa:Topcode-matrix-3164-odd}), where $S_{3164}(65)$ is consisted of 65 numbers, and $S^{odd}_{3164}(80)$ is made up of 80 numbers.

{\footnotesize
\begin{equation}\label{eqa:rzykdc-00}
{
\begin{split}
&S_{3164}(65)=91101128831212488512126767141485691414105411131312311315151401150\\
&S^{odd}_{3164}(80)=17119213161652323716169232311141213272715101217272719108212525236225292927023290
\end{split}}
\end{equation}
}

\begin{equation}\label{eqa:Topcode-matrix-3164}
\centering
{
\begin{split}
T_{code}(H_{3164})= \left(
\begin{array}{ccccccccccccccc}
8 & 8 & 8 & 8 & 7 & 6 & 5 & 6 & 5 & 4 & 3 & 1 & 0 & 1 & 0\\
1 & 2 & 3 & 4 & 5 & 6 & 7 & 8 & 9 & 10 & 11 & 12 & 13 & 14 & 15\\
9 & 10 & 11 & 12 & 12 & 12 & 12 & 14 & 14 & 14 & 14 & 13 & 13 & 15 & 15
\end{array}
\right)
\end{split}}
\end{equation}

\begin{equation}\label{eqa:Topcode-matrix-3164-odd}
\centering
{
\begin{split}
T_{code}(H^{odd}_{3164})= \left(
\begin{array}{ccccccccccccccc}
16 & 16 & 16 & 16 & 14 & 12 & 10 & 12 & 10 & 8 & 6 & 2 & 0 & 2 & 0\\
1 & 3 & 5 & 7 & 9 & 11 & 13 & 15 & 17 & 19 & 21 & 23 & 25 & 27 & 29\\
17 & 19 & 21 & 23 & 23 & 23 & 23 & 27 & 27 & 27 & 27 & 25 & 25 & 29 & 29
\end{array}
\right)
\end{split}}
\end{equation}

By the above preparation, we present an example for our topological authentication as follows:
\begin{quote}
$\bullet $ A \emph{topological public-key}: a number-based string $S_{3164}(65)$, a Topcode-matrix $T_{code}(H_{3164})$, a colored Hanzi-graph $H=H_{3164}$ admitting a graceful total coloring $f$.

$\bullet $ A \emph{topological private-key}: a number-based string $S^{odd}_{3164}(80)$, a Topcode-matrix $T_{code}(H^{odd}_{3164})$, a colored Hanzi-graph $G=H_{3164}$ admitting an odd-graceful total coloring $g$.

$\bullet $ A \emph{topological authentication}: $H \cong G$, there is a transformation function $\varphi$ holding $g(w\,')=\varphi (f(w))$ for $w\in V(H)\cup E(H)$ and $w\,'\in V(G)\cup E(G)$, such that $S^{odd}_{3164}(80)=\varphi (S_{3164}(65))$; there exists a fixed rule $\theta$ to produces the public-key $S_{3164}(65)$ from $T_{code}(H_{3164})$, and the private-key $S^{odd}_{3164}(80)$ from $T_{code}(H^{odd}_{3164})$, such that $S_{3164}(65)=\theta(T_{code}(H_{3164}))$ and $S^{odd}_{3164}(80)=\theta(T_{code}(H^{odd}_{3164}))$, since two Topcode-matrices $T_{code}(H_{3164})$ and $T_{code}(H^{odd}_{3164})$ both are the same order $3\times 15$.
\end{quote}

\begin{defn} \label{defn:topo-authentication-number-string}
$^*$ A \emph{topological authentication $T_{auth}$ of number-based strings} is defined as
\begin{equation}\label{eqa:topo-authentication-00}
{
\begin{split}
T_{auth}=&F\langle G\leftarrow \lambda(H), g\leftarrow \varphi(f),T_{code}(G)\leftarrow \varphi (T_{code}(H),\\
&S(m)=\theta(T_{code}(G))\leftarrow S(n)=\theta(T_{code}(H))\rangle
\end{split}}
\end{equation} based on a \emph{topological public-key} $P_{ub}\langle S(n), T_{code}(H), H,f\rangle $ and a \emph{topological private-key} $P_{ri}\langle S(m)$, $T_{code}(G)$, $ G,g\rangle $, where $\lambda$ is a graph operation, and $\varphi$ is a \emph{transformation function} on two colorings or two Topcode-matrices, and $\theta$ is a \emph{rule function} for producing number-based strings, as well as the public-key graph $H$ admitting a $W_i$-type coloring $f$ which induces a Topcode-matrix $T_{code}(H)$ of $H$, and the private-key graph $G$ admitting a $W_j$-type coloring $g$ which induces a Topcode-matrix $T_{code}(G)$ of $G$.\qqed
\end{defn}

\begin{rem}\label{rem:topo-authentication-rem}
About topological authentication of number-based strings, we have:

\begin{asparaenum}[\textbf{\textrm{Topau}}-1. ]
\item We write $T_{auth}$ defined in Definition \ref{defn:topo-authentication-number-string} by another shorter form
\begin{equation}\label{eqa:topo-authentication-11}
T_{auth}=\langle S(n), T_{code}(H), g,H\rangle \rightarrow _{\theta,\varphi,\lambda}\langle S(m),T_{code}(G),f,G\rangle
\end{equation} for simplicity. Moreover, it may happen one-vs-more case below
\begin{equation}\label{eqa:topo-authentication-k}
T^k_{auth}=\langle S(n), T_{code}(H), g,H\rangle \rightarrow _{\theta,\varphi,\lambda}\langle S_k(m),T_{code}(G_k),f_k,G_k\rangle,k\in [1,r]
\end{equation} One topological public-key $P_{ub}\langle S(n), T_{code}(H), H,f\rangle $ is topologically authenticated to two or more topological private-keys $P^k_{ri}\langle S_k(m)$, $T_{code}(G_k)$, $ G_k,g_k\rangle $ for $k\in [1,r]$ with $r\geq 2$.
\item The process of implementing a topological authentication of number-based strings defined in Definition \ref{defn:topo-authentication-number-string} is
{\small
\begin{equation}\label{eqa:process-topo-authentication}
S(n) \Rightarrow T_{code}(H) \Rightarrow H\textrm{ colored by } g\Rightarrow \lambda(H) \rightarrow G \textrm{ colored by }f \Rightarrow T_{code}(G) \Rightarrow S(m)
\end{equation}
}for $S(m)=\theta(T_{code}(G))$ and $S(n)=\theta(T_{code}(H))$. In fact, Eq.(\ref{eqa:process-topo-authentication}) is an algorithm, called \emph{topological authentication algorithm}, where the narrow ``$\Rightarrow$'' is ``find'', the graph homomorphism ``$\lambda(H) \rightarrow G$'' is a ``colored graph authentication''.
\item In general, the graph operation $\lambda$ of Definition \ref{defn:topo-authentication-number-string} is an isomorphic mapping $G\cong H$, sometimes, $\lambda$ is a graph homomorphism $G\rightarrow H$, or $\lambda$ is a graph anti-homomorphism $H\rightarrow_{anti} G$, and so on.
\item \textbf{A public-key composed of multiple variables.} A public-key is no longer a number-based string, or a text-based string, but it may be a graph, or it will be composed of multiple variables. \paralled
\end{asparaenum}
\end{rem}

For the Confidentiality, Integrity, Availability of Privacy computation, and the Data desensitization, Anonymization, Differential Privacy, Homomorphic encryption in Privacy data protection, we provide the following topological authentication with multiple variables:

\begin{defn} \label{defn:topo-authentication-multiple-variables}
$^*$ A \emph{topological authentication} $\textbf{T}_{\textbf{a}}\langle\textbf{X},\textbf{Y}\rangle$ \emph{of multiple variables} is defined as follows
\begin{equation}\label{eqa:topo-authentication-11}
\textbf{T}_{\textbf{a}}\langle\textbf{X},\textbf{Y}\rangle =P_{ub}(\textbf{X}) \rightarrow _{\textbf{F}} P_{ri}(\textbf{Y})
\end{equation} where $P_{ub}(\textbf{X})=(\alpha_1,\alpha_2,\dots ,\alpha_m)$ and $P_{ri}(\textbf{Y})=(\beta_1,\beta_2,\dots ,\beta_m)$ both are \emph{variable vectors}, in which both $\alpha_1$ and $\beta_1$ are two graphs or sets of graphs (resp. colored graphs, uncolored graphs), and $\textbf{F}=(\theta_1,\theta_2,\dots $, $\theta_m)$ is an \emph{operation vector}, $P_{ub}(\textbf{X})$ is a \emph{topological public-key vector} and $P_{ri}(\textbf{Y})$ is a \emph{topological private-key vector}, such that $\theta_k(\alpha_k)\rightarrow \beta_k$ for $k\in [1,m]$ with $m\geq 1$.\qqed
\end{defn}

\begin{rem}\label{rem:topo-authentication-multiple-complexity}
In Definition \ref{defn:topo-authentication-multiple-variables}, we have:

(i) The operation vector $\textbf{F}$ is consisted of graph operations and algebraic operations, as well as other abstract operations.

(ii) The graphs of the set $\alpha_1$ and the graphs of the set $\beta_1$ are belong to the same kind; $|\alpha_1|\neq |\beta_1|$ could happen.

(iii) The symbol ``$\rightarrow $'' in $\theta_k(\alpha_k)\rightarrow \beta_k$ means a transformation between variables $\alpha_k$ and $\beta_k$.

(iv) In the computation complexity of some $\theta_k(\alpha_k)\rightarrow \beta_k$ is $O(n^k)$ with variable $n$, $O(2^n)$ or $O(\textrm{NP})\in \{$NP-hard, NP-complete$\}$, then we say $\textbf{T}_{\textbf{a}}\langle\textbf{X},\textbf{Y}\rangle$ to be $O(n^k)$, or $O(2^n)$, or $O(\textrm{NP})$.

(v) We may meet $m$ topological authentications $\textbf{T}_{\textbf{a}}\langle\textbf{X}_i,\textbf{X}_{i+1}\rangle =P_{ub}(\textbf{X}_i) \rightarrow _{\textbf{F}} P_{ri}(\textbf{X}_{i+1})$ for $i\in [1,m]$, we call this sequence $\{\textbf{T}_{\textbf{a}}\langle\textbf{X}_i,\textbf{X}_{i+1}\rangle\}^{m-1}_1$ \emph{topological authentication chain}.\paralled
\end{rem}

In the topological authentication of number-based strings, we are facing the following questions:
\begin{asparaenum}[\textrm{\textbf{Ques}}-1. ]
\item \textbf{A large number of number-based strings generated from a Topcode-matrix.} There are $(45)!$ number-based strings, like $S_{3164}(65)$, made from the Topcode-matrix $T_{code}(H_{3164})$, so is for the Topcode-matrix $T_{code}(H^{odd}_{3164})$. Each of two number-based strings $S_{3164}(65)$ and $S^{odd}_{3164}(80)$ can be cut into $45$ segments for building up two Topcode-matrices $T_{code}(H_{3164})$ and $T_{code}(H^{odd}_{3164})$. For example, $S_{3164}(65)$ can be consecutively cut into segments 9, 1, 10, 11, 2, 8, 8, 3, 12, 12, 4, 8, 8, 5, 12, 12, 6, 7, 6, 7, 14, 14, 8, 5, 6, 9, 14, 14, 10, 5, 4, 11, 13, 13, 12, 3, 1, 13, 15, 15, 14, 0, 1, 15, 0, these segments can help us to rebuild up the desired Topcode-matrix $T_{code}(H_{3164})$ shown in Eq.(\ref{eqa:Topcode-matrix-3164}).
\item \textbf{No polynomial algorithm for decomposing a number-based string into the segments of a Topcode-matrix.} Two number-based strings $S_{3164}(65)^*$ and $S^{odd}_{3164}(80)^*$ shown in Eq.(\ref{eqa:rzykdc-00-mixed}) are obtained form two number-based strings $S_{3164}(65)$ and $S^{odd}_{3164}(80)$ shown in Eq.(\ref{eqa:rzykdc-00}) by rearranging the orders of numbers in $S_{3164}(65)$ and $S^{odd}_{3164}(80)$. Obviously, it is quite difficult to rebuild up the desired Topcode-matrices $T_{code}(H_{3164})$ and $T_{code}(H^{odd}_{3164})$ from two number-based strings $S_{3164}(65)^*$ and $S^{odd}_{3164}(80)^*$ directly, even impossible for number-based strings with rather long bytes.
\item \textbf{A Topcode-matrix may correspond two or more mutually non-isomorphic colored graphs.} Five colored graphs shown in Fig.\ref{fig:tea-graph-homomorphism} hold $T_{code}(H_{1872})=T_{code}(L_{i})$ for $i\in [1,4]$ true, and moreover $H_{1872}\not \cong L_{i}$, $L_{i}\not \cong L_{j}$ if $i\neq j$. In practical topological authentications, one will encounter the Graph Isomorphism Problem in finding colored public-key graphs and colored private-key graphs from Topcode-matrices, a NP-complete problem (Ref. Wikipedia).
\item \textbf{No polynomial algorithm for coloring a graph.} There are hundreds of kinds of colorings in topological coding. As reported, if a graph $G$ admits a $W$-type coloring $f$, then $G$ may admits hundreds of $W$-type colorings $f_i$, in which any two $W$-type colorings are different from each other.
\end{asparaenum}

{\footnotesize
\begin{equation}\label{eqa:rzykdc-00-mixed}
{
\begin{split}
&S_{3164}(65)^*=11011011011099888887766655555544444443333322222211111111111111111\\
&S^{odd}_{3164}(80)^*=00229914212131217621631112013277155273101272719132711912528231225236225616929723
\end{split}}
\end{equation}
}

\begin{figure}[h]
\centering
\includegraphics[width=16.4cm]{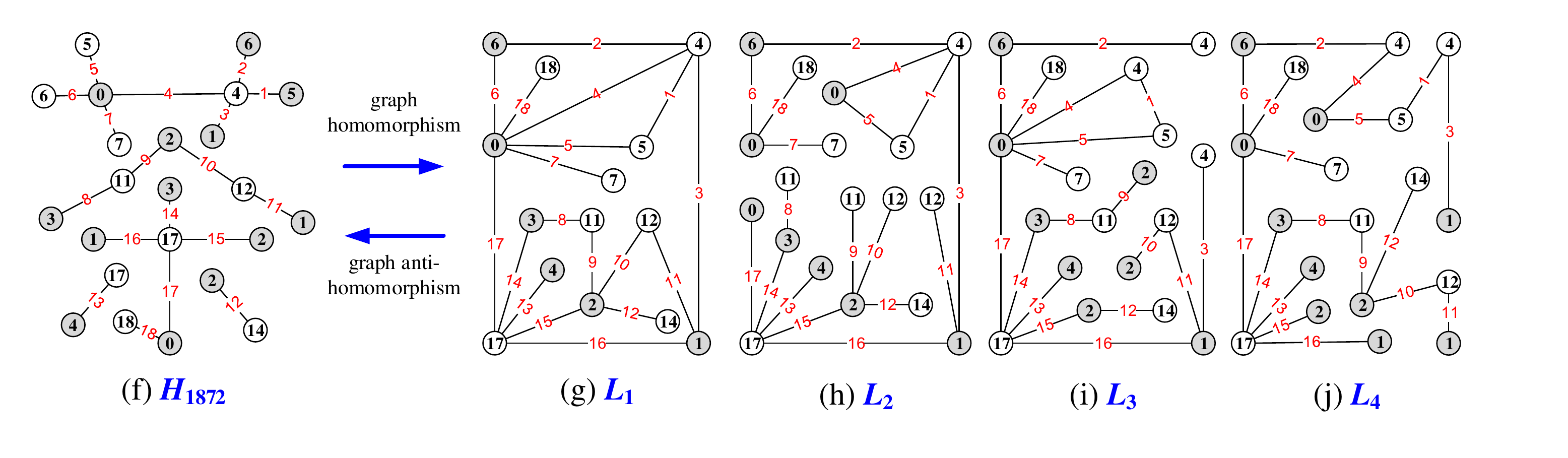}\\
\caption{\label{fig:tea-graph-homomorphism}{\small A colored Hanzi-graph $H_{1872}$ is graph homomorphism to each colored graph $L_{i}$ for $i\in [1,4]$, and each graph $L_{i}$ for $i\in [2,4]$ is graph homomorphism to $L_{1}$.}}
\end{figure}

\begin{problem}\label{problem:xxxxxxxxx}
As known, the number $G_{13}$ of graphs of $13$ vertices holds $G_{13}=50502031367952\approx 2^{46}$ true (Ref. Appendix and \cite{Harary-Palmer-1973}). \textbf{Using} the vertex-splitting operation defined in Definition \ref{defn:vertex-split-coinciding-operations} vertex-splits a connected graph of 13 vertices and 33 edges into Hanzi-graphs.
\end{problem}

\begin{defn} \label{defn:sets-topological-expression}
$^*$ Let $I(n)=\{a_1,a_2,\dots ,a_n\}$ be an integer set with $0\leq a_1<a_2<\dots <a_n$. If there is a graph $G$ admitting an existing coloring (resp. an existing labeling) $f$ such that each $a_k=f(w)$ for some $w\in V(G)\cup E(G)$, we say that $G$ is a \emph{topological expression} of the set $I(n)$, and $I(n)$ is \emph{topologically represented}. \qqed
\end{defn}

\begin{problem}\label{qeu:444444}
Let $T_{exp}(I(n))$ be the set of graphs, in which each graph is a topological expression of the set $I(n)$ defined in Definition \ref{defn:sets-topological-expression}. \textbf{Determine} the set $T_{exp}(I(n))$ for an integer set $I(n)=\{a_1,a_2,\dots ,a_n\}$ with $0\leq a_1<a_2<\dots <a_n$.
\end{problem}

\begin{problem}\label{problem:xxxxxxxxx}
\textbf{How} to express a finite integer set by topological structures?
\end{problem}

\subsubsection{Privacy data protection}

Key technologies for the data privacy protection are: \textbf{Data desensitization} -- reducing data sensitivity through distortion transformation; \textbf{Anonymization} -- privacy protection through ``de-identification''; \textbf{Differential Privacy} -- adding noise to resist differential attacks; \textbf{Homomorphic encryption} -- encrypting sensitive personal information directly, and then directly counting and machine learning on ciphertext data.

As known, \emph{Privacy data} can be divided into \emph{personal privacy data} and \emph{common privacy data}, and \emph{Personal privacy} data includes \emph{personal information} that can be used to identify or locate individuals (such as telephone number, address, credit card number, authentication information, etc.) and \emph{sensitive information} (such as personal health status, financial information, historical access records, important documents of the company, \emph{etc.}). \emph{Sensitive data} is presented in two forms: \emph{Structured sensitive data} exists in business applications, databases, enterprise resource planning (ERP) systems, storage devices, third-party service providers, backup media and external storage facilities of the enterprise; \emph{Unstructured sensitive data} is scattered in the whole infrastructure of the enterprise, including desktops, laptops on various removable hard disks and other endpoints.

\emph{Digital authentication certificate} is an encryption technology with the core digital certificate. It can encrypt and decrypt the information transmitted, digital signature and signature verification in networks, so as to ensure the security (\textbf{C}onfidentiality), integrity (\textbf{I}ntegrity) and availability (\textbf{A}vailability) of information transmitted on the Internet, called \textbf{CIA-bases}. With digital certificate, even if the information you send is intercepted by others in Internet, or even you lose your personal account, password and other information, you can still ensure the security of your account and funds.

\textbf{Confidentiality} is to the use of structured classification and classification guidelines to regulate access to personal data by individuals and third parties.

\textbf{Integrity} refers to ensuring that information is not tampered with (i.e. live data) as it travels from source to destination, and that stored information is not tampered with (i.e. static data).

\textbf{Availability} is to ensure that the data services are available.

\vskip 0.4cm

We present a new concept as follows:

\begin{defn} \label{defn:55-various-privacy-data}
$^*$ Let $P_{\textrm{ri-da}}(n)=\langle c_1,c_2,\dots ,c_n\rangle =\langle c_j\rangle ^n_{j=1}$ be a \emph{privacy data} (also \emph{static privacy data}),
\begin{equation}\label{eqa:555555}
\partial P_{\textrm{ri-da}}(n)=\langle \partial c_1,\partial c_2,\dots ,\partial c_n\rangle =\langle \partial c_j\rangle ^n_{j=1}
\end{equation} be a \emph{mathematical transformation} of $P_{\textrm{ri-da}}(n)$ such that each $\partial c_j\rightarrow c\,'_j$ guarantees that $c_j$ can not be found from $c\,'_j$, and let
\begin{equation}\label{eqa:555555}
R(P_{\textrm{ri-da}})(n)=\langle R(c_1),R(c_2),\dots ,R(c_n)\rangle =\langle R(c_j)\rangle ^n_{j=1}
\end{equation} be the \emph{privacy rank} of $P_{\textrm{ri-da}}(n)$ with $R(c_j)\in [0,M]$. If the privacy data $P_{\textrm{ri-da}}(n)$ holds completely the \emph{confidentiality}, the \emph{integrity} and the \emph{availability} (CIA-bases), we call $P_{\textrm{ri-da}}(n)$ a \emph{perfect privacy data}, otherwise, $P_{\textrm{ri-da}}(n)$ is called a \emph{minor-privacy data} if it holds part of CIA-bases. \qqed
\end{defn}

\begin{rem}\label{rem:55-various-privacy-data}
In Definition \ref{defn:55-various-privacy-data}, a privacy data $P_{\textrm{ri-da}}(n)$ is \emph{self-certificated} if it comes with a digital authentication certificate. In the real world, a privacy data changes over time, that is, we should consider dynamic privacy data
\begin{equation}\label{eqa:555555}
P_{\textrm{ri-da}}(t)=\langle c_1(t),c_2(t),\dots ,c_{n(t)}(t)\rangle=\langle c_j(t)\rangle ^{n(t)}_{j=1},~t\in [a,b]
\end{equation} with $c_i(t)\neq c_j(t)$ for $i\neq j$.\paralled
\end{rem}

As known, the privacy computation system involves three key technologies: blockchain, federated learning and secure multi-party computation. It is widely believed that privacy computing will be deeply combined with blockchain technology, since blockchain is a basic technology to better solve the problem of mutual trust among parties.

\begin{problem}\label{qeu:444444}
How do you get people to believe that your data is real, except blockchain technology, without the need of digital authentication certificate from the third-party authority?
\end{problem}

\begin{figure}[h]
\centering
\includegraphics[width=15cm]{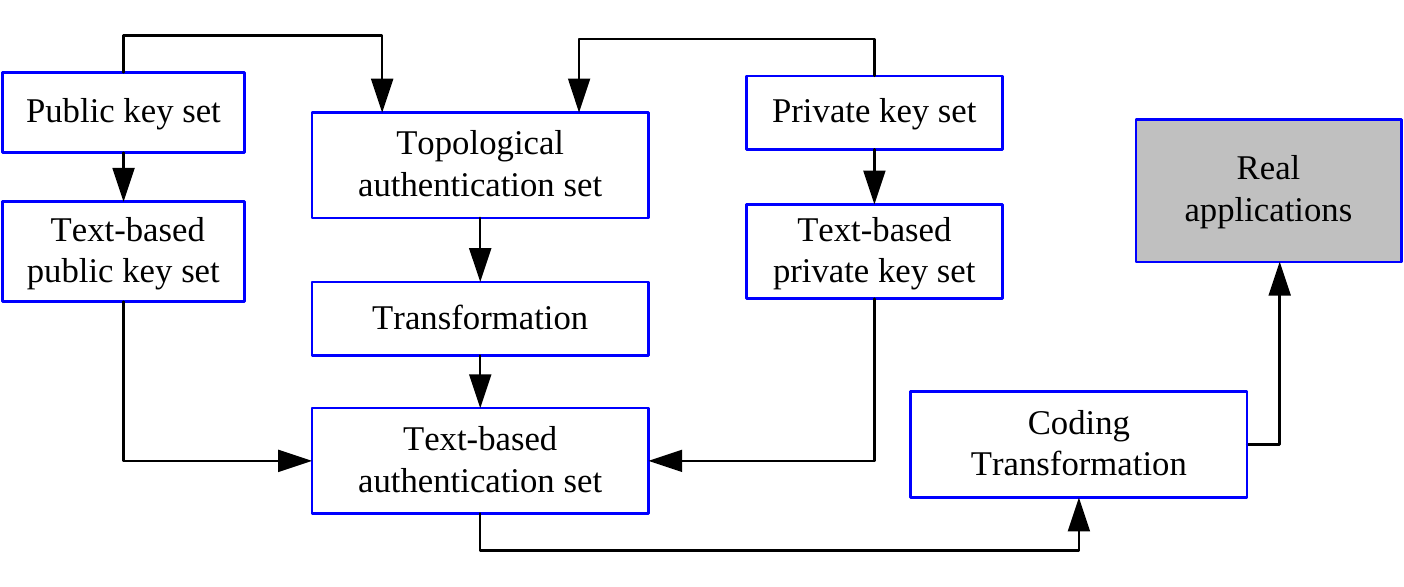}
\caption{\label{fig:authentication-system}{\small A topological authentication system, cited from \cite{Wang-Su-Yao-mixed-difference-2019}.}}
\end{figure}

\subsubsection{Advantages of topological authentications}

Our topological authentications realized by the Topsnut-gpws of topological coding, and the domain of topological coding involves millions of things, and a graph of topological coding connects things together to form a complete ``story'' under certain constraints. Topological authentications have the following advantages \cite{Yao-Wang-2106-15254v1}:
\begin{asparaenum}[\textbf{Advantages}-1.]
\item \textbf{Related with two or more different mathematical areas}. Topsnut-gpws were made of topological structures and mathematical restrictions. Here, ``mathematical restrictions'' are related with number theory, Set theory, Probability theory, Calculus, Algebra, Abstract algebra, Linear algebra \emph{etc}. However, ``topological structures'' differ from the above ``hard restrictions'' and are ``soft mathematical expressions'' belonging to Graph Theory, a mathematical branch. These characteristics will produce Topologically Asymmetric Cryptosystem for resisting the quantum computation and intelligent attacks equipped with quantum computer in future ear of quantum computer.
\item \textbf{Not pictures}. Topsnut-gpws run fast in communication networks because they are saved in computer by popular matrices rather than pictures and photos. With the present image recognition technology, it is easy to scan Topsnut-gpws drawn on the plane into computer, and obtain the adjacent matrix and other matrices of the graphs. In topological authentication, we have (A-1) Topsnut-gpws $\rightarrow$ (A-2) text-based (number-based) strings $\rightarrow$ (A-3) encrypt files. Conversely, (B-1) private-keys $\rightarrow$ (B-2) text-based (resp. number-based) strings $\rightarrow$ (B-3) Topsnut-gpws $\rightarrow$ (B-4) decrypt files, where the procedure of (B-2) $\rightarrow$ (B-3) is almost impossible.
\item \textbf{Huge numbers of graphs, colorings and labelings.} There are enormous numbers of graphs, graph colorings and labelings in graph theory (Ref. Appendix Table-1, Table-2 and Table-3). And new graph colorings (resp. labelings) come into being everyday. The number of one $W$-type different labelings (resp. colorings) of a graph maybe large, and no method is reported to find out all of such labelings (resp. colorings) for a graph.
\item \textbf{Based on modern mathematics.} For easy memory, some simpler operations like addition, subtraction, absolution and finite modular operations are applied. Number theory, algebra and graph theory are the strong support to Topsnut-gpws of topological coding. It is noticeable, a topological authentication consists of the following two aspects:

\quad (1) complex topological structures (vertices are different geometric figures, edges are undirected lines, or directed lines, or solid lines, or dotted line, \emph{ect.}) dyed by hundreds of visual colors, see Fig.\ref{fig:various-Topsnut-gpws}; and

\quad (2) various mathematical constraints.
\item \textbf{Diversity of asymmetric ciphers.} One key corresponds to more locks, or more keys correspond to one lock only. Topsnut-gpws realize the coexistence of two or more labelings and colorings on a graph, which leads to the problem of multiple labeling (resp. multiple coloring) decomposition of graphs, and brings new research objects and new problems to mathematics.
\item \textbf{Combinatorics}. There are connections between Topsnut-gpws and other type of passwords. For example, small circles in the Topsnut-gpws can be equipped with fingerprints and other biological information requirements, and users' pictures can be embedded in small circles, greatly reflects personalization.
\item \textbf{Simplicity and convenience.} Topsnut-gpws are suitable for people who need not learn new rules and are allowed to use their private knowledge in making Topsnut-gpws for the sake of remembering easily. For example, Chinese characters are naturally topological structures (also graphs in Topsnut-gpws) to produce Hanzi-graphs for topological coding. Chinese people can generate Topsnut-gpws simply by speaking and writing Chinese directly.
\item \textbf{Easily created.} Many labelings of trees are convertible from each other \cite{Yao-Liu-Yao-2017}. Tree-type structures can adapt to a large number of labelings and colorings, in addition to graceful labeling, no other labelings reported that were established in those tree-type structures having smaller vertex numbers by computer, almost no computer proof. Because construction methods are complex, this means that using computers to break down Topsnut-gpws will be difficult greatly.
\item \textbf{Irreversibility.} Topsnut-gpws can generate quickly number-based (resp. text-based) strings with bytes as long as desired, but these strings can not reconstruct the original Topsnut-gpws. The confusion of number-based (resp. text-based) strings and Topsnut-gpws can't be erased, as our knowledge.
\item \textbf{Computational security.} There are many non-polynomial algorithms in making Topsnut-gpws. For example, drawing non-isomorphic graphs is very difficult and non-polynomial. For a given graph, finding out all possible colorings (resp. labelings) are impossible, since these colorings (resp. labelings) are massive data, and many graph problems have been proven to be NP-complete.
\item \textbf{Provable security.} There are many mathematical conjectures (also, \emph{open problems}) in graph labelings and colorings, such as the famous graceful tree conjecture, odd-graceful tree conjecture, total coloring conjecture \emph{etc.}
\end{asparaenum}

\begin{figure}[h]
\centering
\includegraphics[width=16.4cm]{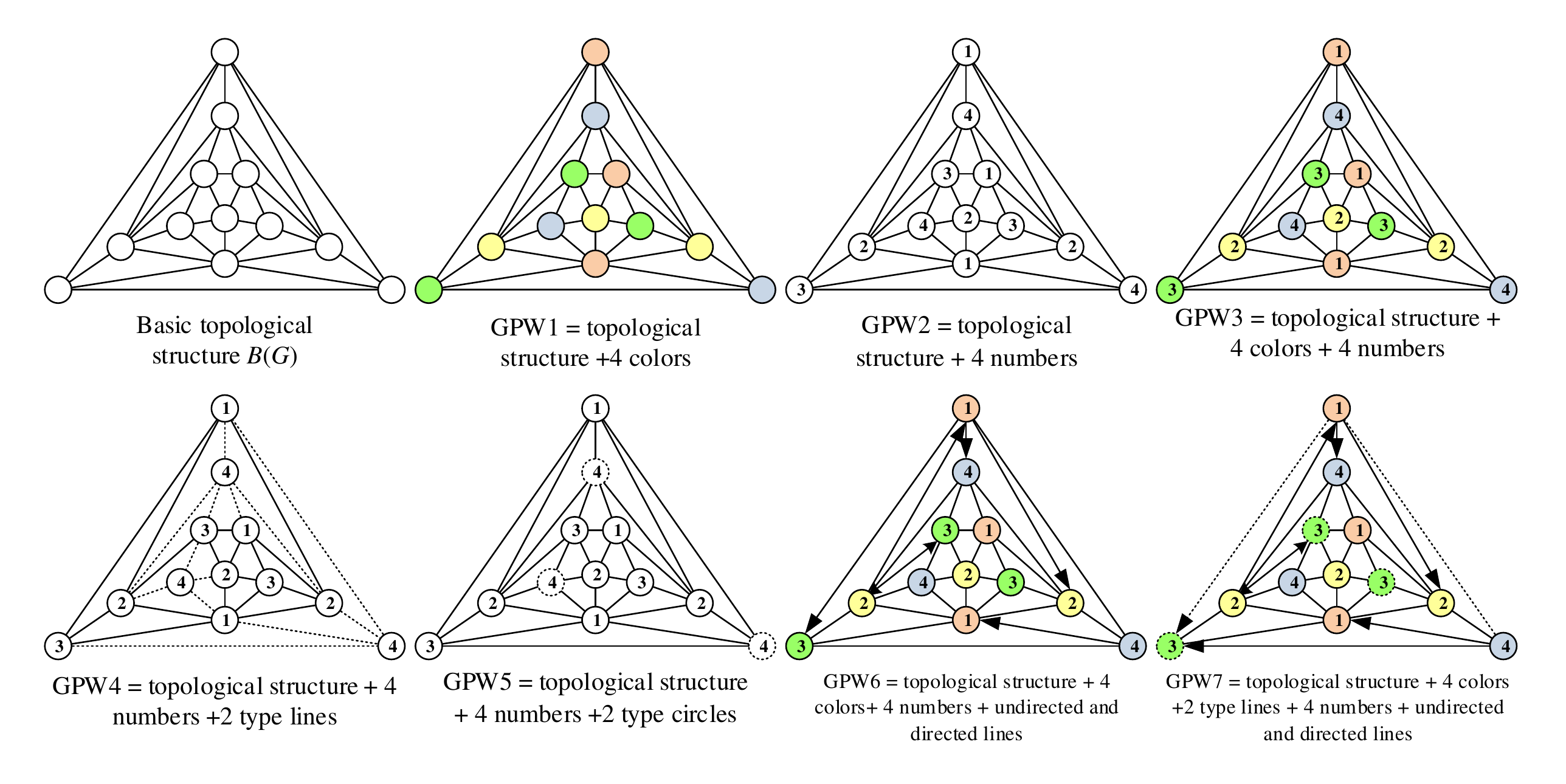}\\
\caption{\label{fig:various-Topsnut-gpws}{\small Various Topsnut-gpws, cited from \cite{Yao-Wang-2106-15254v1}.}}
\end{figure}

\subsubsection{Number-based String Problems}

We restate the following number-based string problems and their complex analysis.

\begin{problem}\label{qeu:PNBSPP}
\cite{Yao-Wang-2106-15254v1} \textbf{Parameterized Number-based String Partition Problem (PNBSPP).} Given a number-based string $s(n)=c_1c_2\cdots c_n$ with $c_i\in [0,9]$, partition it into $3q$ segments $c_1c_2\cdots c_n=a_1a_2\cdots a_{3q}$ with $a_j=c_{n_j}c_{n_j+1}\cdots c_{n_{j+1}}$ with $j\in [1,3q-1]$, where $n_1=1$ and $n_{3q}=n$. And use $a_k$ with $k\in [1,3q]$ to reform a Topcode-matrix $T_{code}(G)$ (refer to Definition \ref{defn:topcode-matrix-definition}), and moreover use the Topcode-matrix $T_{code}(G)$ to reconstruct all of Topsnut-gpws of $q$ edges. By the found Topsnut-gpws corresponding the common Topcode-matrix $T_{code}(G)$, \textbf{find} the desired \emph{public Topsnut-gpws} $H_i$ and the desired \emph{private Topsnut-gpws} $G_i$, such that each mapping $\varphi_i:V(G_i)\rightarrow V(H_i)$ forms a \emph{colored graph homomorphism} $G_i\rightarrow H_i$ with $i\in [1,r]$.
\end{problem}

\begin{problem}\label{problem:string-vs-no-order-partition}
\cite{Yao-Wang-2106-15254v1} \textbf{Number-Based String No-order Decomposition Problem (NBSNODP).} Decomposing a number-based string $s(n)=c_1c_2\cdots c_n$ with $c_j\in [0,9]$ (as a \emph{public-key}) produces smaller number-based strings $b_1,b_2,\dots, b_{k}$ such that

(1) $b_i=c_{i,1}c_{i,2}\cdots c_{i,m_i}$ with $i\in [1,k]$;

(2) if a number $c_j$ of $s(n)$ appears in $b_j$ for some $j$, then $c_j$ does not appear in any $b_i$ for $i\neq j$.

We want that these number-based strings $b_1,b_2,\dots, b_{k}$ can construct just one of a Topcode-matrix $T_{code}(G)$, a colored degree-sequence matrix $D_{sc}(\textbf{\textrm{d}})$, a degree-sequence $\textrm{deg}(G)$, an adjacent e-value matrix $E_{color}(G)$, and an adjacent ve-value matrix $VE_{color}(G)$ for some colored graph $G$ (as a \emph{private-key}).
\end{problem}

\textbf{Complexity analysis of PNBSPP and NBSNODP.} For a given number-based string $s(n)=c_1c_2\cdots c_n$ with $c_i\in [0,9]$, we have the following complexity analysis about number-based strings:

(a) The mathematical characteristics of the graph corresponding to a number-based string $s(n)$ are not unique. Each one of a Topcode-matrix $T_{code}(G)$, a colored degree-sequence matrix $D_{sc}(\textbf{\textrm{d}})$, a degree-sequence $\textrm{deg}(G)$, an adjacent e-value matrix $E_{color}(G)$, and an adjacent ve-value matrix $VE_{color}(G)$ for some colored graph $G$ can induce number-based strings, but there is no polynomial algorithm for writing all matrices corresponding a number-based string $s(n)$.

(b) A number-based string $s(n)$ can induce two or more degree-sequences, no polynomial algorithm reported is about this problem.

(c) The degree-sequence or Topcode-matrix induced from a number-based string $s(n)$ corresponds two or more graphs. Suppose that a number-based string $s(n)$ induces $m(s(n))$ degree-sequences $\textrm{deg}_k(s(n))=(a_{k,1},a_{k,2},\dots ,a_{k,m_k})$ with $k\in [1,m(s(n))]$, and we find out them all. However, determining non-isomorphic graphs induced by a degree-sequence $deg_k(s(n))$ will meet the graph isomorphism, it is a NP-hard problem. The amount of calculation for judging non-isomorphic graphs is exponential as $n$ is larger. For example, for judging non-isomorphic graphs of 24 vertices that are not less than $2^{197} $ different graphs from each other, it is reported that the number of sand particles on earth is about $2^{76}\sim 2^{78}$. At present, there is no polynomial algorithm for judging graph isomorphism.

(d) A number-based string $s(n)$ can induce two or more Topsnut-graphic codes. New graphical coding technologies are emerging everyday. As known, there are more than 350 coloring technologies (more than 200 labeling technologies listed in the literature \cite{Gallian2021}), and many coloring technologies contain unsolved mathematical problems, such as Graceful Tree Conjecture (Rosa, 1966), Total Coloring Conjecture (Behzad, 1965), Vertex-Distinguishing Edge-Coloring Conjecture (Burris and Schelp, \cite{Burris-Schelp}), Determining Ramsey Number, and so on (Ref. \cite{Yao-Yang-Yao-2020-distinguishing}).

(e) Solving or deciphering PNBSPP and NBSNODP involves three basic elements: password length, password element and topology. In addition, there is no unified technology to measure password strength. If the password element is a number, the work of solving or deciphering PNBSPP and NBSNODP needs to cross collide and combine back and forth in two different types of fields: Algebra (set theory, group) and topology. The above analyses (a), (b), (c) and (d) involve exponential calculations. Password length has become the direction of PNBSPP and NBSNODP solution or decoding. In \cite{Sheppard-D-A-1976}, Sheppard has proven: A tree of $n$ vertices admits $\frac{1}{2}n!$ different graceful labelings, Stirling's formula $n!=\sqrt{2\pi n}\left (\frac{n}{e}\right )^n$, it shows that determining the graceful labelings of trees is an exponential work.

(f) In Topsnut-graphic codes, suppose there are $k$ types of vertices and $l$ types of edges, $m$ colors used to color the vertices and edges of graphs, and a $(p,q)$-graph $G$ admits $n(G)$ types of colorings and labelings, for each type of colorings $\varepsilon\in [1,n(G)]$, the $(p,q)$-graph $G$ admits $o(G,\varepsilon )$ different $\varepsilon$-type colorings. Let $S_{vo}(G,p,q)$ be the number of elements of the Topsnut-graphic coding set obtained from the $(p,q)$-graph $G$, then
\begin{equation}\label{eqa:space-p-q}
{
\begin{split}
S_{vo}(G,p,q)&=2^{kp}\cdot 2^{lq}\cdot 2^{mp}\cdot 2^{mq}\sum^ {n(G)}_{\varepsilon=1}o(G,\varepsilon )=2^{(k+m)p+(l+m)q}\sum^ {n(G)}_{\varepsilon=1}o(G,\varepsilon )
\end{split}}
\end{equation} Immediately, the number of elements of the Topsnut-graphic coding set obtained from all $(p, q)$-graphs is $S_{vo}(p)=\sum^{G_p}_{q=1} S_{vo}(G,p,q)$.

Let $G$ be a $(p,p-1)$-tree, $n(G)=o(G,\varepsilon )=1$, and $m=k=l$, we get $S_{vo}(G,p,p-1)=2^{2m(2p-1)}$. Deo, Nikoloski and Suraweera reported \cite{Deo-Nikoloski-Suraweera-2002}: Trees with no more than 29 vertices are graceful trees, there are 5,469,566,585 trees of 29 vertices. When $m=2$, a tree of 29 vertices colored with white and black produces $2^{228}$ Topsnut-graphic codes. Since $5469566585\geq 2^{32}$, trees of 29 vertices produce at least $2^{260}$ Topsnut-graphic codes.

\vskip 0.4cm

To sum up, PNBSPP and NBSNODP are irreversible, that is, the public-key digital string $s(n)$ and a graphic code of the private-key do not have a mutual derivation relationship, although it is very easy to induce number-based strings from degree-sequences and various matrices of graphs. The security of graphic lattice cryptosystem is based on NP-hard and other difficult problems, and there are also various homomorphic relations between adjacent matrices, Topsnut-matrix homomorphism and Topsnut-matrices, degree sequences and graphs. It is preliminarily concluded that the graph lattice based on PNBSPP and NBSNODP can resist quantum computing.

\subsection{Terminology and notation}

Standard terminology and notation of graphs and digraphs used here are cited from \cite{Bang-Jensen-Gutin-digraphs-2007}, \cite{Bondy-2008} and \cite{Gallian2021}. All graphs mentioned here are \emph{simple} (no \emph{multiple-edges} and \emph{loops}), unless otherwise stated. All \emph{colorings} and \emph{labelings} mentioned in this article are on graphs of graph theory, very often, people say \emph{graph colorings} and \emph{graph labelings} \cite{Gallian2021}. The following notation and terminology will be used in the whole article:
\begin{asparaenum}[$\bullet$ ]
\item All non-negative integers are collected in the set $Z^0$, and all integers are in the set $Z$, so the positive integer set is $Z^+=Z^0\setminus \{0\}$.
\item A $(p,q)$-graph $G$ is a \emph{data structure} having $p$ vertices and $q$ edges, and $G$ has no multiple edge and directed-edge, such that its own vertex set $V(G)$ holds $|V(G)|=p$ and its own edge set $E(G)$ holds $|E(G)|=q$. And $\overline{G}$ stands for the \emph{complementary graph} of the $(p,q)$-graph $G$.
\item The \emph{cardinality} of a set $X$ is denoted as $|X|$, so the \emph{degree} of a vertex $x$ in a $(p,q)$-graph $G$ is denoted as $\textrm{deg}_G(x)=|N(x)|$, where $N(x)$ is the set of neighbors of the vertex $x$.
\item A vertex $x$ is called a \emph{leaf} if its degree $\textrm{deg}_G(x)=1$.
\item Adding the new edges of an edge set $E^*$ to a graph $G$ produces a new graph, denoted as $G+E^*$, where $E^*\cap E(G)=\emptyset$.
\item A symbol $[a,b]$ stands for an integer set $\{a,a+1,a+2,\dots, b\}$ with two integers $a,b$ subject to $a<b$, and $[a,b]^o$ denotes an \emph{odd-set} $\{a,a+2,\dots, b\}$ with odd integers $a,b$ holding $1\leq a<b$ true, and $[\alpha,\beta]^e$ is an \emph{even-set} $\{\alpha,\alpha+2,\dots, \beta\}$ with even integers $\alpha,\beta$ and $\alpha<\beta$.
\item The symbol $[a,b]^r$ is an \emph{interval} of real numbers.
\item A \emph{text-based password} is abbreviated as \emph{TB-paw} and is made by 52 English characters and numbers of $[0,9]$.
\item A \emph{number-based string} $S(n)=c_1c_2\cdots c_n$ with $c_j\in [0,9]$, where $n$ is the \emph{number} (\emph{length}) of $S(n)$, and moreover $S^{-1}(n)=(9-c_1)(9-c_2)\cdots (9-c_n)$ is the \emph{dual number-based string} of $S(n)$.
\item A \emph{number-based hyper-string} $R_n=r_1r_2\cdots r_n$ is defined as $r_j=c_{j,1}c_{j,2}\cdots c_{j,m_j}$ with $c_{j,k}\in [0,9]$ for $k\in [1,m_j]$ with $m_j\geq 2$ and $j\in [1,n]$.
\item A \emph{string} $D=t_1t_2\cdots t_m$ has its own \emph{reciprocal string} defined by $D^{-1}=t_mt_{m-1}\cdots t_2t_1$, also, we say that $D$ and $D^{-1}$ match with each other, where $m$ is the \emph{number} (\emph{length}) of the string $D$.
\item The set of all subsets of a set $S$ is denoted as $S^2=\{X:~X\subseteq S\}$, called \emph{power set} of $S$, and the power set $S^2$ contains no empty set at all. For example, for a given set $S=\{a,b,c,d,e\}$, so the power set $S^2$ has its own elements $\{a\}$, $\{b\}$, $\{c\}$, $\{d\}$, $\{e\}$, $\{a,b\}$, $\{a,c\}$, $\{a,d\}$, $\{a,e\}$, $\{b,c\}$, $\{b,d\}$, $\{b,e\}$, $\{c, d\}$, $\{c, e\}$, $\{d,e\}$, $\{a,b,c\}$, $\{a,b,d\}$, $\{a,b,e\}$, $\{a,c,d\}$, $\{a,c,e\}$, $\{a,d,e\}$, $\{b,c,d\}$, $\{b,c,e\}$, $\{a,b,c,d\}$, $\{a,b,c,e\}$, $\{a,c,d,e\}$, $\{b,c,d,e\}$ and $\{a,b,c,d,e\}$.
\item $[a,b]^2$ is the set of all subsets of an integer set $[a,b]$ for $a<b$ and $a,b\in Z^0$.
\item Let $S(k_i,d_i)^q_1=\{k_i, k_i + d_i, \dots , k_i +(q-1)d_i\}$ be an \emph{integer set} for integers $k_i\geq 1$ and $d_i\geq 1$, and let $\{(k_i,d_i)\}^m_1=\{(k_1,d_1), (k_2,d_2),\dots ,(k_m,d_m)\}$ be a \emph{matching-pair sequence}. Also, $S_{q-1,k,d}=\{k,k+d,k+2d,\dots, k+(q-1)d\}$ and $S^{\pm}_{q-1,k,d}=\{k,k-d,k-2d,\dots, k-(q-1)d\}\cup \{k,k+d,k+2d,\dots, k+(q-1)d\}$ for integers $q\geq 1$, $k\geq 0$ and $d\geq 1$.
\item For integers $a,k,m\geq 0$, $d\geq 1$ and $q\geq 1$, there are two parameterized sets
$$
S_{m,k,a,d}=\{k+ad,k+(a+1)d,\dots ,k+(a+m)d\},~O_{2q-1,k,d}=\{k+d,k+3d,\dots ,k+(2q-1)d\}
$$
with two cardinalities $|S_{m,k,a,d}|=m+1$ and $|O_{2q-1,k,d}|=q$.
\item \textbf{Topsnut-gpws}. Topsnut-gpws is the abbreviation of ``graphical passwords based on topological structure and number theory'' in Topological Coding, a mixture subbranch of graph theory and cryptography (Ref. \cite{Yao-Zhao-Zhang-Mu-Sun-Zhang-Yang-Ma-Su-Wang-Wang-Sun-arXiv2019}, \cite{Wang-Xu-Yao-2016}, \cite{Wang-Xu-Yao-Key-models-Lock-models-2016} and \cite{Yao-Sun-Zhang-Li-Zhang-Xie-Yao-2017-Tianjin-University}).
\item A \emph{sequence} $\textbf{\textrm{d}}=(m_1,m_2,\dots,m_n)=(m_k)^n_{k=1}$ consists of positive integers $m_1, m_2, \dots , m_n$. If a graph $G$ has its \emph{degree-sequence} ${\textrm{deg}}(G)=\textbf{\textrm{d}}$, then $\textbf{\textrm{d}}$ is \emph{graphical} (see Lemma \ref{thm:basic-degree-sequence-lemma}), and we call $\textbf{\textrm{d}}$ \emph{degree-sequence}, and each $m_i$ \emph{degree component}, and $n=L_{\textrm{ength}}(\textbf{\textrm{d}})$ \emph{length} of $\textbf{\textrm{d}}$.
\item A tree is a connected and acyclic graph. A \emph{caterpillar} $T$ is a tree, after removing all of leaves from the caterpillar $T$, the remainder is just a \emph{path}; a \emph{lobster} is tree, and the deletion of all leaves of the lobster produces a caterpillar.
\item $\textrm{gcd}(a,b)$ is the \emph{maximal common factor} between two positive integers $a$ and $b$.
\item An isomorphism $G\cong H$ is a configuration identity on two graphs $G$ and $H$, it has nothing to do with the colorings of these two graphs.
\end{asparaenum}

\begin{defn} \label{defn:111111}
A \emph{lattice} $\textrm{\textbf{L}}(\textbf{B})$ is a set of all integer combinations
\begin{equation}\label{eqa:traditional-lattice}
\textrm{\textbf{L}}(\textbf{B}) =\{x_1\textbf{\textrm{b}}_1+x_2\textbf{\textrm{b}}_2+\cdots +x_n\textbf{\textrm{b}}_n : x_i \in Z\}
\end{equation}
of $n$ linearly independent vectors of a base $\textbf{B}=(\textbf{\textrm{b}}_1,\textbf{\textrm{b}}_2,\dots $, $ \textbf{\textrm{b}}_n)$ in $R^m$ with $n\leq m$, where $Z$ is the integer set, $m$ is the \emph{dimension} and $n$ is the \emph{rank} of the lattice, and $\textbf{B}$ is called a \emph{lattice base}. In the view of geometry, a lattice is a set of discrete points with periodic structure in $R^m$. For no confusion, we call $\textrm{\textbf{L}}(\textbf{B})$ \emph{traditional lattice}. Particularly, each vector $A$ of the non-negative lattice $\textrm{\textbf{L}}^+(\textbf{B})$ holds $\textbf{\textrm{A}}=x_1\textbf{\textrm{b}}_1+x_2\textbf{\textrm{b}}_2+\cdots +x_n\textbf{\textrm{b}}_n$ as $x_i\geq 0$ with $i\in [1,n]$, and $\textbf{\textrm{b}}_k=(b_{k,1},b_{k,2},\dots,b_{k,m})$ with $b_{k,j}\geq 0$ for $j\in[1,m]$.\qqed
\end{defn}

\begin{lem} \label{thm:basic-degree-sequence-lemma}
\cite{Bondy-2008} (Erd\"{o}s-Galia Theorem) A sequence $\textbf{\textrm{d}}=(m_k)^n_{k=1}$ with $m_{i}\geq m_{i+1}\geq 0$ to be the \emph{degree-sequence} of a $(p,q)$-graph graph $G$ if and only if the sum $\sum^n_{i=1}m_i=2q$ and
\begin{equation}\label{eqa:basic-degree-sequence-lemma}
\sum^k_{i=1}m_i\leq k(k-1)+\sum ^n_{j=k+1}\min\{k,m_j\},~1\leq k\leq n-1
\end{equation}
\end{lem}

\begin{defn}\label{defn:totally-normal-labeling}
\cite{su-yan-yao-2018} A $(p,q)$-graph $G$ admits a bijection $f:V(G)\rightarrow [m,n]$ or $f: V(G)\cup E(G)\rightarrow [m,n]$, we denoted the set of colored vertices of $G$ as $f(V(G))=\{f(x):x\in V(G)\}$, the color set of edges of the graph $G$ as $f(E(G))=\{f(uv):uv\in E(G)\}$. If $|f(V(G))|=p$, then $f$ is called a \emph{vertex labeling} of $G$; when as $|f(E(G))|=q$, $f$ is called an \emph{edge labeling} of $G$; and when $|f(V(G)\cup E(G))|=p+q$, we call $f$ a \emph{total labeling}.\qqed
\end{defn}

\begin{defn}\label{defn:universal-label-set}
\cite{Bing-Yao-Cheng-Yao-Zhao2009} A graph $G$ admits a coloring $h:S\subseteq V(G)\cup E(G)\rightarrow [a,b]$, and write the color set $h(S)=\{h(x): x\in S\}$. The \emph{dual coloring} $h\,'$ of the coloring $h$ is defined as: $h\,'(z)=\max h(S)+\min h(S)-h(z)$ for $z\in S$. Moreover, $h(S)$ is called the \emph{vertex color set} if $S=V(G)$, $h(S)$ the \emph{edge color set} if $S=E(G)$, and $h(S)$ a \emph{universal color set} if $S=V(G)\cup E(G)$. Furthermore, if $G$ is a bipartite graph with its vertex set bipartition $(X,Y)$, and holds $\max h(X)<\min h(Y)$ true, we call $h$ a \emph{set-ordered coloring} (resp. \emph{set-ordered labeling}).\qqed
\end{defn}

\subsection{Graph operations}

Graph operation the soul of topological structures of graphs, they join vertices and edges into various wholes, and represent various number sets by topological expressions, and incarnate perfectly the origin of mathematics.

A \emph{graph operation} to a graph $G$ is an abstract mapping $\Gamma$, such that $\Gamma: G\rightarrow G^*$, write as $G^*=\Gamma(G)$ directly. For example, as $G\cong G^*$, and $\Gamma$ is a coloring, $G$ is an uncolored graph, but $G^*$ is colored by $\Gamma$. And moreover, an abstract mapping $\Gamma$ may join two or more graphs together, such that the resultant graph $H=\Gamma(G_1$, $G_2$, $\dots ,G_n)$. In the view of homomorphism, $G^*=\Gamma(G)$ and $H=\Gamma(G_1$, $G_2$, $\dots ,G_n)$ are graph homomorphisms or graph anti-homomorphisms, that are
\begin{equation}\label{eqa:555555}
G\rightarrow G^*, \quad (G_1, G_2, \dots ,G_n)\rightarrow H
\end{equation}

\subsubsection{Basic graph operations}

Removing an edge $uv$ from a graph $G$ produces an \emph{edge-removed graph} $G-uv$. Adding an edge $xy\not \in E(G)$ to $G$ makes an \emph{edge-added graph} $G+xy$. Similarly, $G-w$ is a \emph{vertex-removed graph} after deleting the vertex $w$ from $G$, and removing those edges with one end to be $w$.

A \emph{ve-added graph} $G+\{y,x_1y,x_2y,\dots ,x_sy\}$ is obtained by adding a new vertex $y$ to a graph $G$, and join $y$ with vertices $x_1,x_2,\dots ,x_s$ of $G$, respectively. By the edge-removed graph $G-uv$ and the vertex-removed graph $G-w$, we have a \emph{vertex-set-removed graph} $G-S$ for a vertex proper subset $S\subset V(G)$, as well as an \emph{edge-set-removed graph} $G-E^*$ for an edge subset $E^*\subset E(G)$.

Particularly, an \emph{edge-set-added graph} $G+E\,'$ is defined by adding each edge of an edge set $E\,'\subset E(\overline{G})$, where $\overline{G}$ is the complement of a graph $G$.

An \emph{edge symmetric graph} $G\ominus G\,'$ is obtained by join a vertex $u$ of a graph $G$ with its image $u\,'$ of a copy $G\,'$ of $G$ together by a new edge $uu\,'$.

\subsubsection{Vertex-splitting and Vertex-coinciding operations}

For the aim of convenient research, we present our operations in the manner of definition.

\begin{defn} \label{defn:vertex-split-coinciding-operations}
\cite{Yao-Sun-Zhang-Mu-Wang-Jin-Xu-2018, Yao-Zhang-Sun-Mu-Sun-Wang-Wang-Ma-Su-Yang-Yang-Zhang-2018arXiv} \textbf{Vertex-splitting operation.} We split a vertex $u$ of a graph $G$ with $\textrm{deg}(u)\geq 2$ into two vertices $u\,'$ and $u\,''$, such that $N(u)=N(u\,')\cup N(u\,'')$, $|N(u\,')|\geq 1$, $|N(u\,'')|\geq 1$ and $N(u\,')\cap N(u\,'')=\emptyset$, the resultant graph is denoted as $G\wedge u$, see Fig.\ref{fig:11-vertex-split-coin} (a)$\rightarrow$(b). Moreover, we select randomly a proper subset $S$ of $V(G)$, and implement the vertex-splitting operation to each vertex of $S$, the resultant graph is denoted as $G\wedge S$.

\textbf{Vertex-coinciding operation.} (Also, called the \emph{non-common neighbor vertex-coinciding operation}) A vertex-coinciding operation is the \emph{inverse} of a vertex-splitting operation, and vice versa. If two vertices $u\,'$ and $u\,''$ of a graph $H$ holds $|N(u\,')|\geq 1$, $|N(u\,'')|\geq 1$ and $N(u\,')\cap N(u\,'')=\emptyset$ true, we vertex-coincide $u\,'$ and $u\,''$ into one vertex $u=u\,'\odot u\,''$, such that $N(u)=N(u\,')\cup N(u\,'')$, the resultant graph is denoted as $G=H(u\,'\odot u\,'')$, see Fig.\ref{fig:11-vertex-split-coin} (b)$\rightarrow$(a).\qqed
\end{defn}

In Fig.\ref{fig:11-vertex-split-coin}, there are three graph homomorphisms $L_i\rightarrow _{\textrm{v-coin}}L_{i+1}$ for $i\in [1,3]$ obtained by the vertex-coinciding operation defined in Definition \ref{defn:vertex-split-coinciding-operations}, and moreover there are three graph anti-homomorphisms $L_k\rightarrow _{\textrm{v-split}}L_{k-1}$ for $k\in [2,4]$ obtained by the vertex-splitting operation defined in Definition \ref{defn:vertex-split-coinciding-operations}.

\begin{figure}[h]
\centering
\includegraphics[width=14cm]{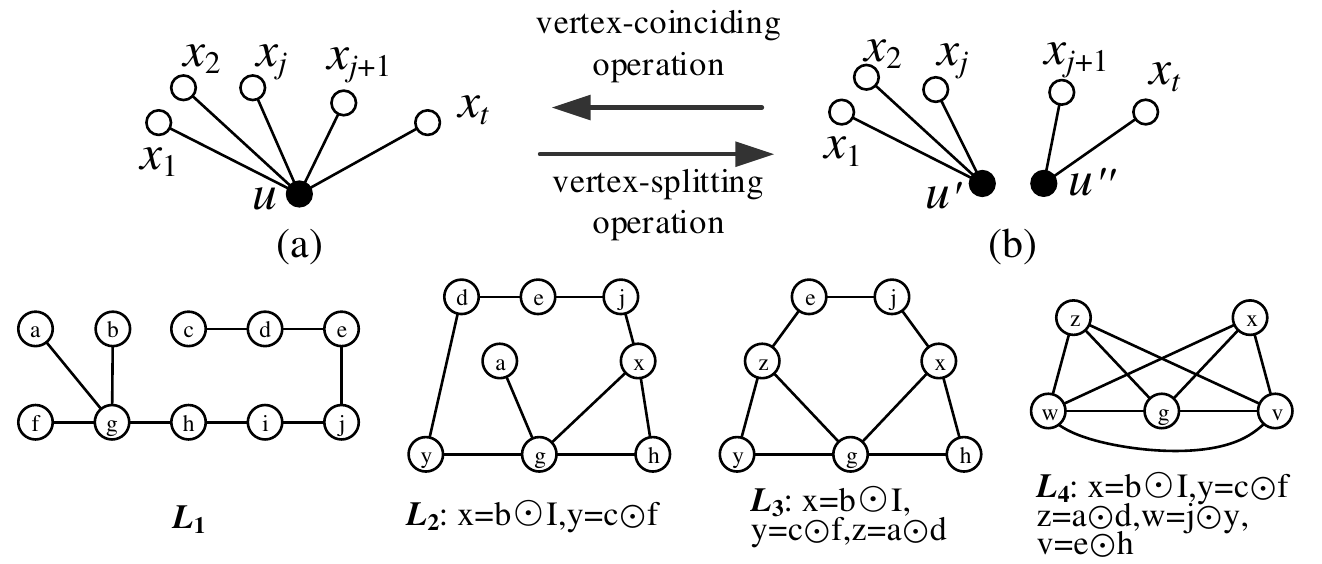}\\
\caption{\label{fig:11-vertex-split-coin}{\small The vertex-coinciding and the vertex-splitting operations defined in Definition \ref{defn:vertex-split-coinciding-operations}.}}
\end{figure}

If two vertex-disjoint colored graphs $G$ and $H$ have $k$ pairs of vertices with each pair of vertices colored with the same colors, then we, by the vertex-coinciding operation defined in Definition \ref{defn:vertex-split-coinciding-operations}, vertex-coincide each pair of vertices from $G$ and $H$ into one, the resultant graph is denoted as $\odot_k \langle G,H\rangle$ (or $G[\odot_k]H$, called \emph{vertex-coincided graph}), and moreover the vertex-coincided graph $\odot_k\langle G,H\rangle$ has $|V(\odot_k \langle G,H\rangle)|=|V(G)|+|V(H)|-k$ vertices and $|E(\odot_k \langle G,H\rangle)|=|E(G)|+|E(H)|$ edges. Clearly, the vertex-coincided graph $\odot_k \langle G,H\rangle$ holds for vertex-disjoint uncolored graphs $G$ and $H$ too.

\begin{thm}\label{thm:2-vertex-split-graphs-isomorphic}
\cite{Yao-Wang-2106-15254v1} Suppose that two connected graphs $G$ and $H$ admit a mapping $f:V(G)\rightarrow V(H)$. In general, a vertex-split graph $G\wedge u$ with $\textrm{deg}_G(u)\geq 2$ is not unique, so we have a vertex-split graph set $S_G(u)=\{G\wedge u\}$, similarly, we have another vertex-split graph set $S_H(f(u))=\{H\wedge f(u)\}$. If each vertex-split graph $L\in S_G(u)$ corresponds a unique vertex-split graph $T\in S_H(f(u))$ such that $L\cong T$, and vice versa, we write this fact as
\begin{equation}\label{eqa:2-vertex-split-graphs-isomorphic}
G\wedge u\cong H\wedge f(u)\,
\end{equation} Thereby, $G$ is \emph{isomorphic} to $H$, namely, $G\cong H$.
\end{thm}

\subsubsection{Edge-coinciding and edge-splitting operations}

\begin{defn} \label{defn:edge-split-coinciding-operations}
\cite{Yao-Zhang-Sun-Mu-Sun-Wang-Wang-Ma-Su-Yang-Yang-Zhang-2018arXiv} \textbf{Edge-splitting operation.} For an edge $uv$ of a graph $G$ with $\textrm{deg}(u)\geq 2$ and $\textrm{deg}(v)\geq 2$, remove the edge $uv$ from $G$ first; next vertex-split, respectively, two end vertices $u$ and $v$ of the edge $uv$ into vertices $u\,'$ and $u\,''$, $v\,'$ and $v\,''$; and then add a new edge $u\,'v\,'$ to join $u\,'$ and $v\,'$ together, and add another new edge $u\,''v\,''$ to join $u\,''$ and $v\,''$ together, respectively. The resultant graph is denoted as $G\wedge uv$, see Fig.\ref{fig:11-edge-leaf-split-coin} (a)$\rightarrow$(c). It is allowed that $|N(u\,')|=1$ and $|N(v\,'')|=1$, see Fig.\ref{fig:11-edge-leaf-split-coin}(a)$\rightarrow$(b), in this case, the process of obtaining $G\wedge uv$ is called \emph{leaf-splitting operation}, or a \emph{train-hook splitting operation}, they are \emph{particular cases} of the edge-splitting operation. The \emph{inverse} of a train-hook splitting operation is called a a \emph{train-hook coinciding operation}.

\textbf{Edge-coinciding operation.} For two edges $u\,'v\,'$ and $u\,''v\,''$ of a graph $H$, if $N(u\,')\cap N(u\,'')=\emptyset$ and $N(v\,')\cap N(v\,'')=\emptyset$, we edge-coincide two edges $u\,'v\,'$ and $u\,''v\,''$ into one edge $uv=u\,'v\,'\ominus u\,''v\,''$ with $u=u\,'\odot u\,''$ and $v=v\,'\odot v\,''$. The resultant graph $H(u\,'v\,'\ominus u\,''v\,'')$ is the result of doing the \emph{edge-coinciding operation} to $H$, see Fig.\ref{fig:11-edge-leaf-split-coin}(c)$\rightarrow$(a). Also, $H(u\,'v\,'\ominus u\,''v\,'')$ is the result of doing the \emph{leaf-coinciding operation} to $H$ as $|N(u\,'')|=1$ and $|N(v\,')|=1$, see Fig.\ref{fig:11-edge-leaf-split-coin}(b)$\rightarrow$(a).\qqed
\end{defn}

\begin{figure}[h]
\centering
\includegraphics[width=16.4cm]{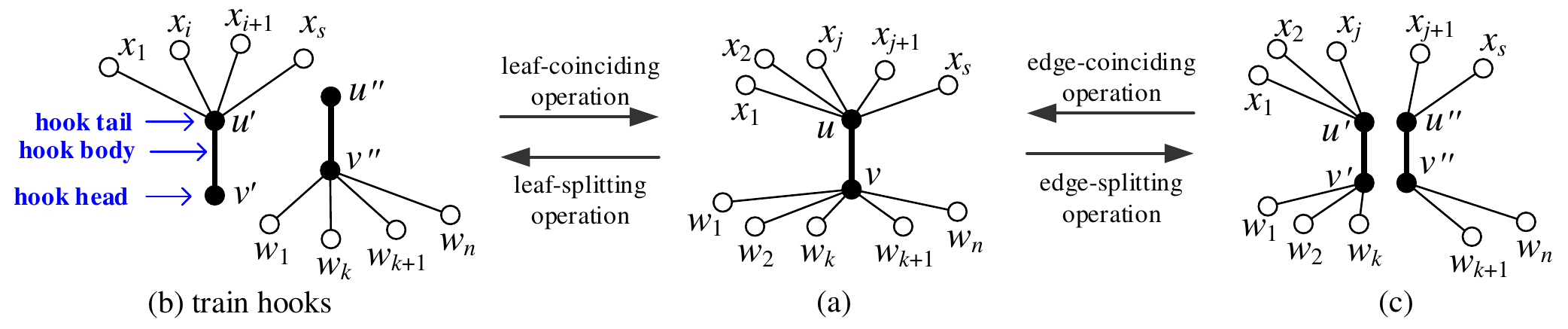}\\
\caption{\label{fig:11-edge-leaf-split-coin}{\small The edge-coinciding and the edge-splitting operations defined in Definition \ref{defn:edge-split-coinciding-operations}.}}
\end{figure}

\begin{problem}\label{qeu:vertex-splitting-set-auth}
Let $G$ be a connected $(p,q)$-graph. A set $V_{\textrm{split}}(G)$ collects all of connected graphs obtained by doing the vertex-splitting operation to $G$ (refer to Definition \ref{defn:vertex-split-coinciding-operations}), and the subset $T_{\textrm{vsplit}}(G)$ of $V_{\textrm{split}}(G)$ contains all of trees in $V_{\textrm{split}}(G)$. Clearly, each connected graph $H$ of $V_{\textrm{split}}(G)$ has $q$ edges in total, and $H$ is \emph{graph homomorphism} to $G$ under the vertex-coinciding operation defined in Definition \ref{defn:vertex-split-coinciding-operations}, namely, $H\rightarrow _{\textrm{v-coin}}G$. For a given connected public-key $G$, \textbf{find} graph homomorphisms for private-keys $H$ holding $H\rightarrow_{\textrm{v-coin}} G$ true.
\end{problem}

\begin{problem}\label{qeu:edge-splitting-set-auth}
Let $G$ be a connected $(p,q)$-graph, and let $E_{\textrm{split}}(G)$ be the set of connected graphs obtained by doing the edge-splitting operation defined in Definition \ref{defn:edge-split-coinciding-operations} to $G$, and the subset $T_{\textrm{esplit}}(G)$ of $E_{\textrm{split}}(G)$ contains all of trees in $V_{\textrm{split}}(G)$. Notice that each connected graph $H$ of $V_{\textrm{split}}(G)$ is \emph{graph homomorphism} to $G$ under the edge-coinciding operation defined in Definition \ref{defn:edge-split-coinciding-operations}, namely, $H\rightarrow _{\textrm{e-coin}}G$. \textbf{Find} all private-keys $H$ holding $H\rightarrow _{\textrm{e-coin}}G$ for a given connected public-key $G$ in order to form topological authentications.
\end{problem}

\subsubsection{$W$-coinciding and $W$-splitting operations}

\begin{defn} \label{defn:W-splitting-coinciding-operation}
$^*$ Let $W$ be a proper subgraph of a graph $G$. We do a $W$-splitting operation to $G$ in the following way \cite{Yao-Su-Sun-Wang-Graph-Operations-2021}:

(i) Removing the edges of $E(W)$ from $W$;

(ii) Vertex-split each vertex $x_i\in V(W)=\{x_i:i\in [1,m]\}$ into two vertices $x\,'_i$ and $x\,''_i$, such that $N(x_i)\setminus V(W)=N(x\,'_i)\cup N(x\,''_i)$ with $N(x\,'_i)\cap N(x\,''_i)=\emptyset$;

(iii) adding new edges to the vertex set $\big \{x\,'_i:i\in [1,m]\big \}$ produces a graph $H_1$ holding $H_1\cong W$ true, and then adding new edges to the vertex set $\big \{x\,''_i:i\in [1,m]\big \}$ makes another graph $H_2$ holding $H_2\cong W$ true, such that each edge $x_ix_j\in E(W)$ corresponds an edge $x\,'_ix\,'_j\in E(H_1)$ and an edge $x\,''_ix\,''_j\in E(H_2)$, and vice versa.

The resultant graph is written as $G\wedge W$, and it has the following properties:

(i) Both $W$-type graphs $H_1$ and $H_2$ are two vertex disjoint isomorphic subgraphs of the $W$-split graph $G\wedge W$, namely, $V(H_1)\cap V(H_2)=\emptyset$;

(ii) each $H_i$ is joined with a vertex $w_i\in V(G\wedge W)\setminus \big [V(H_1)\cup V(H_2)\big ]$ for $i=1,2$; and

(iii) no a common vertex $u^*\in V(G\wedge W)\setminus \big [V(H_1)\cup V(H_2)\big ]$ holds $u^*x_i\in E(H_i)$ for $i=1,2$.

We call the process of obtaining the \emph{$W$-split graph} $G\wedge W$ an \emph{$W$-splitting operation}. Conversely, the process of obtaining $G$ from $G\wedge W$ by the vertex-coinciding operation and the edge-coinciding operation defined in Definition \ref{defn:vertex-split-coinciding-operations} and Definition \ref{defn:edge-split-coinciding-operations} is called an \emph{$W$-coinciding operation}, since $N(x\,'_i)\cap N(x\,''_i)=\emptyset$ for $i\in [1,m]$. \qqed
\end{defn}

\begin{rem}\label{rem:333333}
If $G\wedge W$ is disconnected, so $G$ has two vertex-disjoint components $G_1$ and $G_2$ holding $W=H_i\subset G_i$ for $i=1,2$ in Definition \ref{defn:W-splitting-coinciding-operation}, we then write $G=G_1[\ominus ^W_k]G_2$.

For vertex disjoint graphs $T_s$ with $s\in [1,n]$, if each $T_s$ contains a subgraph $W$ of $k$ vertices, we get an $W$-coincided graph denoted as
\begin{equation}\label{eqa:555555}
[\ominus ^W_k]^n_{s=1}T_s=\Big (\cdots \big (T_1[\ominus ^W_k]T_2\big )[\ominus ^W_k]T_3\cdots \Big )[\ominus ^W_k]T_n
\end{equation} For a permutation $T_{j_1}, T_{j_2}, \dots ,T_{j_s}$ of $T_1,T_2,\dots ,T_n$, we have $[\ominus ^W_k]^n_{r=1}T_{j_r}$. So, there are $n!$ $W$-coincided graphs in total. \paralled
\end{rem}

By Definition \ref{defn:W-splitting-coinciding-operation}, we have the following particular cases:

\textbf{Case 1.} If $W$ is a cycle $C$ of $k$ vertices, we write $T_1[\ominus ^{cyc}_k]T_2$ by ``$[\ominus ^{cyc}_k]$'' replacing ``$[\ominus ^W_k]$'', similarly, $T_1[\ominus ^{path}_k]T_2$ if $W$ is a path of $k$ vertices, and $T_1[\ominus ^{tree}_k]T_2$ if $W$ is a tree of $k$ vertices, since cycles, paths and trees are \emph{linear-type graphs} in various applications.

\textbf{Case 2.} If $W$ is a complete graph $K_n$ of $n$ vertices, we have $T_1[\ominus ^{K}_n]T_2$ if $K_n$ is a subgraph of two vertex disjoint graphs $T_i$ for $i=1,2$.

\textbf{Case 3.} If $W$ is a cycle $C$ of $k$ vertices in a maximal planar graph $G$, the cycle-split graph $G\wedge C$ has just two vertex disjoint components $G^C_{out}$ and $G^C_{in}$, called \emph{semi-maximal planar graphs}, where $G^C_{out}$ is in the \emph{infinite plane}, and $G^C_{in}$ is inside of $G$. Thereby, we write $G=G^C_{out}[\ominus^C_k]G^C_{in}$ hereafter.

\textbf{Case 4.} If $W$ is a complete graph $K_1$ of one vertex, we write $T_1[\ominus ^{K}_1]T_2=\odot_1\langle T_1,T_2\rangle$, that is, the graph $W$ shrinks to a vertex.

\begin{problem}\label{problem:xxxxxxxxx}
\cite{Yao-Su-Sun-Wang-Graph-Operations-2021} \textbf{Characterize} the following particular cycle-coincided graphs:
\begin{asparaenum}[\textrm{Planep}-1. ]
\item A graph $G$ can be expressed as $G=H_i[\ominus ^{cyc}_{k_i}]G_i$ for $i\in [1,m]$ with $m\geq 1$ by the cycle-coinciding operation, where two vertex disjoint graphs $H_i$ and $G_i$ contain cycles with the same length $k_i$.

\item A cycle-coincided graph $L_i=H^*[\ominus ^{cyc}_{k_i}]G_i$ for $i\in [1,m]$, where any pair of vertex disjoint graphs $H^*$ and each $G_i$ contain cycles with the same length $k_i$, where $H^*$ is like a fixed ``point'' under the cycle-coinciding operation. Furthermore, we get a cycle-coincided graph

\begin{equation}\label{eqa:555555}
[\ominus ^{cyc}]^m_{j=1}L_j=\Big (\cdots \big (H^*[\ominus ^{cyc}_{k_1}]G_1\big )[\ominus ^{cyc}_{k_2}]G_2\cdots \Big )[\ominus ^{cyc}_{k_m}]G_m
\end{equation} the above graph is like a \emph{kaleidoscope}.

\item A cycle-coincided graph $B=H^*[\ominus ^{cyc}_{n}]^m_{i=1}G_i$ is like a ``\emph{super book}'', where $H^*$ and each $G_i$ contain cycles with the same length $n$, so each $G_i$ is a \emph{book page} and $H^*$ is the \emph{book back} of a super book.
\item If $W$ is a path of $k$ vertices, $H^*[\ominus ^{path}_k]^m_{i=1}G_i$ is ``\emph{topological-page book}'', where the book back $H^*$ and each topological-page $G_i$ contain paths of $k$ vertices.
\end{asparaenum}
\end{problem}

\begin{problem}\label{qeu:uniquely-4-colorable-mpgs}
Let $C$ be a $k$-cycle of a maximal planar graph $G$, so $G=G^{C}_{out}[\ominus^{cyc}_k]G^{C}_{in}$, and write $G^{C}_{out}=G_{out}$ (as a \emph{public-key}) and $G^{C}_{in}=G_{in}$ (as a \emph{private-key}) if there is no confusion. We have:
\begin{asparaenum}[\textrm{MPG}-1. ]
\item For each triangle $C=K_3$, $G=G_{out}[\ominus^{cyc}_3]G_{in}$ holds $G_{out}=G$ and $G_{in}=K_3$, we call $G$ a \emph{no-$3$-cycle split maximal planar graph}.
\item For each $k$-cycle $C$ with $4\leq k <|V(G)|$, if the edge-removed graph $G_{in}-E(C)$ in $G=G_{out}[\ominus^{cyc}_k]G_{in}$ is a tree $T$, we call $G_{in}$ a \emph{cycle-chord semi-maximal planar graph} if $V(C)=V(T)$, $G_{in}$ a \emph{tree-pure semi-maximal planar graph} if $|V(C)|<|V(T)|$, refer to \cite{Jin-Xu-Maximal-Science-Press-2019}.
\item For a maximal planar graph $G\neq K_4$, if $G=G_{out}(1)[\ominus^{cyc}_3]G_{in}(1)$ with $G_{out}(1)$ is a maximal planar graph being not $K_4$ and $G_{in}(1)=K_4$, we have $G_{out}(1)=G_{out}(2)[\ominus^{cyc}_3]G_{in}(2)$ with $G_{out}(2)$ is a maximal planar graph being not $K_4$ and $G_{in}(2)=K_4$, go on in this way, we get $G_{out}(k-1)=G_{out}(k)[\ominus^{cyc}_3]G_{in}(k)$ with $G_{out}(k)$ is a maximal planar graph being not $K_4$ and $G_{in}(k)=K_4$ for $k\in [1,m]$, where $G=G_{out}(0)$, $G_{out}(m-1)=G_{out}(m)[\ominus^{cyc}_3]G_{in}(m)$ with $G_{out}(m)=G_{in}(m)=K_4$. So, $G$ is a \emph{recursive maximal planar graph} and admits a proper vertex $4$-coloring $f$, such that $V(G)=\bigcup^4_{k=1} V_k(G)$ and $f(x)=k$ for $x\in V_k(G)$ with $k\in [1,4]$. Uniquely 4-colorable Maximal Planar Graph Conjecture \cite{Greenwell-Kronk-Uniquely-4c-1973}: A recursive maximal planar graph $G$ is uniquely 4-colorable, that is, each set $V_k(G)$ in $V(G)=\bigcup^4_{k=1} V_k(G)$ is not changed by any two $4$-colorings of $G$.
\item A $(k,l)$-recursive maximal planar graph has $k$ 3-degree vertices, and the distance between any two 3-degree vertices is just $l$. Chen and Li conjecture \cite{Xiang-en-CHEN-LI-2018}: For $k\geq 5$ and odd $l\geq 3$, there is no $(k,l)$-recursive maximal planar graph.
\end{asparaenum}
\end{problem}

Now, we define the so-called \emph{$W$-type graph-split connectivity} for a connected graph $G$:

\begin{defn} \label{defn:W-type-graph-split-connectivity}
$^*$ Let $H$ be a $W$-type proper subgraph of a connected graph $G$. If the $W$-split graph $G\wedge W$ is disconnected, we call the following parameter
$$\min \{|V(H)|:~G\wedge W\textrm{ is disconnected},~\textrm{and }H~ \textrm{is a $W$-type proper subgraph of}~G\}$$ \emph{$W$-type graph-split connectivity} of $G$, denoted as $\kappa_{W}(G)$. Here, ``$W$-type'' may be one of path, cycle, complete graph, tree, bipartite complete graph, particular graph, and so on.\qqed
\end{defn}

\begin{rem}\label{rem:333333}
In Definition \ref{defn:W-type-graph-split-connectivity}, if $H$ is a graph consisted of edges, then the $W$-type graph-split connectivity $\kappa_{W}(G)=\kappa\,'(G)$, the traditional \emph{edge connectivity} of graphs; and if $H$ is a graph consisted of vertices and edges, then the $W$-type graph-split connectivity $\kappa_{W}(G)=\kappa(G)$ or $\kappa_{W}(G)=\kappa\,''(G)$ for the traditional \emph{vertex connectivity} $\kappa(G)$, or the traditional \emph{total connectivity} $\kappa\,''(G)$. Notice that the $W$-split graph $G\wedge W$ differs from the vertex-removed graph $G-V(H)$, since $G\wedge W$ keeps all information of the original graph $G$.\paralled
\end{rem}

\begin{problem}\label{qeu:444444}
\textbf{Determine} some $W$-type graph-split connectivities $\chi_{path}$, $\chi_{cycle}$ and $\chi_{tree}$ for connected graphs.
\end{problem}

\subsubsection{$W$-expanding and $W$-contracting operations}

\begin{defn} \label{defn:expanding-W-contracting-operations}
$^*$ Suppose that $W$ is a proper subgraph of a graph $G$, and $V_{nei}=\{x_1,x_2$, $\dots$, $x_k\}$ is the subset of $V(G)$ such that $V(W)\cap V_{nei}=\emptyset $ and if a vertex $w_j\in V(W)$ is adjacent with a vertex $v\not \in V(W)$, then $v\in V_{nei}$, and each vertex $x_i\in V_{nei}$ is adjacent with a vertex of $W$, so we call $V_{nei}$ the \emph{neighbor set} of $W$. We delete the subgraph $W$ from $G$, and add a new vertex $u_0$ to join each vertex $x_i\in V_{nei}$ for forming new edges $u_0x_i$ with $i\in [1,k]$, the resultant graph is denoted as $G(W\rightarrow u_0)$, called \emph{$W$-contracted graph}, and the process of obtaining the $W$-contracted graph $G(W\rightarrow u_0)$ is called the \emph{$W$-contracting operation}. Notice that the vertex $u_0$ has degree $k$, and $k\leq |E^*\langle V(W), V_{nei}\rangle |$, where each edge of the edge subset $E^*\langle V(W), V_{nei}\rangle $ of $E(G)$ has one end in $V(W)$ and another end in $V_{nei}$.

Conversely, the graph $G$ can be obtained from the $W$-contracted graph $L^*=G(W\rightarrow u_0)$ by doing the $W$-expanding operation to be the inverse of the $W$-contracting operation, that is, $G=L^*(u_0\rightarrow W)$.\qqed
\end{defn}

\begin{example}\label{exa:8888888888}
In Fig.\ref{fig:expand-contract-operation}, doing the $W$-splitting operation to a maximal planar graph $T$ with $W=H$ produces the $W$-split graph $T\wedge H$ which has two components $H$ and $L$, so $T=L[\ominus^{cyc}_6]H$. By Definition \ref{defn:expanding-W-contracting-operations}, we have the $W$-contracted graph $J=L(C\rightarrow u_0)$, and $J=T(H\rightarrow u_0)$, as well as $L=J(u_0\rightarrow C)$ and $T=J(u_0\rightarrow H)$. If we use the maximal planar graph $J$ as a public-key, there are many different topological authentications like $T$ for a fixed private-key $H$ shown in Fig.\ref{fig:expand-contract-operation}, since there are many different semi-maximal planar graphs like $L$.

\begin{figure}[h]
\centering
\includegraphics[width=16.4cm]{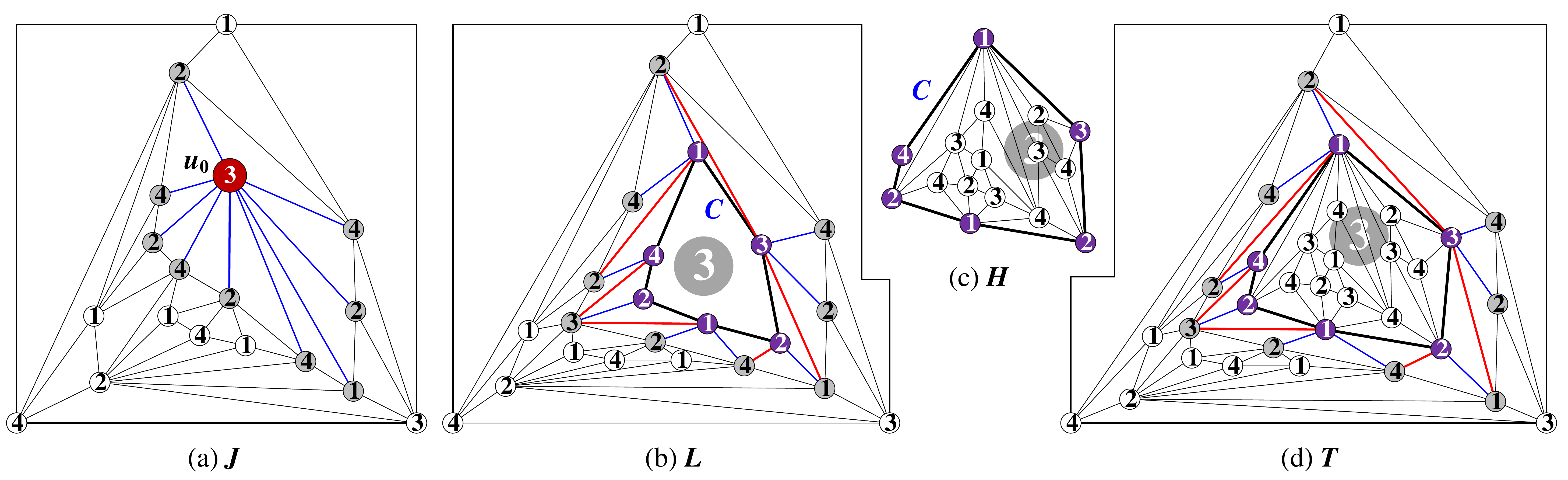}\\
\caption{\label{fig:expand-contract-operation}{\small A scheme for understanding the $W$-expanding operation and the $W$-contracting operation defined in Definition \ref{defn:expanding-W-contracting-operations}.}}
\end{figure}

In Fig.\ref{fig:2-color-cycle-exp-contrac}, a maximal planar graph $G_A$ is a \emph{public-key}, there are three \emph{private-keys} $G_B=G_A(w\rightarrow H)$, $G_C=G_A(y\rightarrow W)$ and $G_D=G_C(x\rightarrow L)$ by the $W$-expanding operation. Conversely, we can obtain $G_A=G_B(H\rightarrow w)$, $G_A=G_C(W\rightarrow y)$, $G_C=G_D(L\rightarrow x)$ by the $W$-contracting operation. Also, we have
\begin{equation}\label{eqa:555555}
G_A=G_D(L\rightarrow x; W\rightarrow y),\quad G_D=G_A(y\rightarrow W; x\rightarrow L)
\end{equation}
In general, we can do a series of $W$-expanding operations to a graph $T$ as follows:

(i) $T_1=T(u_1\rightarrow W_1)$, $T_2=T_1(u_2\rightarrow W_2)$, $\dots$, $T_n=T_{n-1}(u_n\rightarrow W_n)$, where each $u_i$ is a vertex of graph $T_{i-1}$ and $W_i$ is a graph of $n_i$ vertices with $n_i\geq 2$. Notice that some vertex $u_j$ may be not a vertex of $T$. Clearly, $T_{k-1}=T_{k}(W_k\rightarrow u_k)$ for $k\in [1,n]$, where $T_0=T$.

(ii) We select $m$ vertices $x_1,x_2,\dots ,x_m$ of $T$ to do the $W$-expanding operation to them one-by-one, and get
\begin{equation}\label{eqa:555555}
H=T(x_1\rightarrow W_1, x_2\rightarrow W_2,\dots , x_m\rightarrow W_m).
\end{equation} Conversely, we have
\begin{equation}\label{eqa:555555}
T=H(W_1\rightarrow x_1, W_2\rightarrow x_2,\dots , W_m\rightarrow x_m).
\end{equation}
\end{example}

\begin{figure}[h]
\centering
\includegraphics[width=16.4cm]{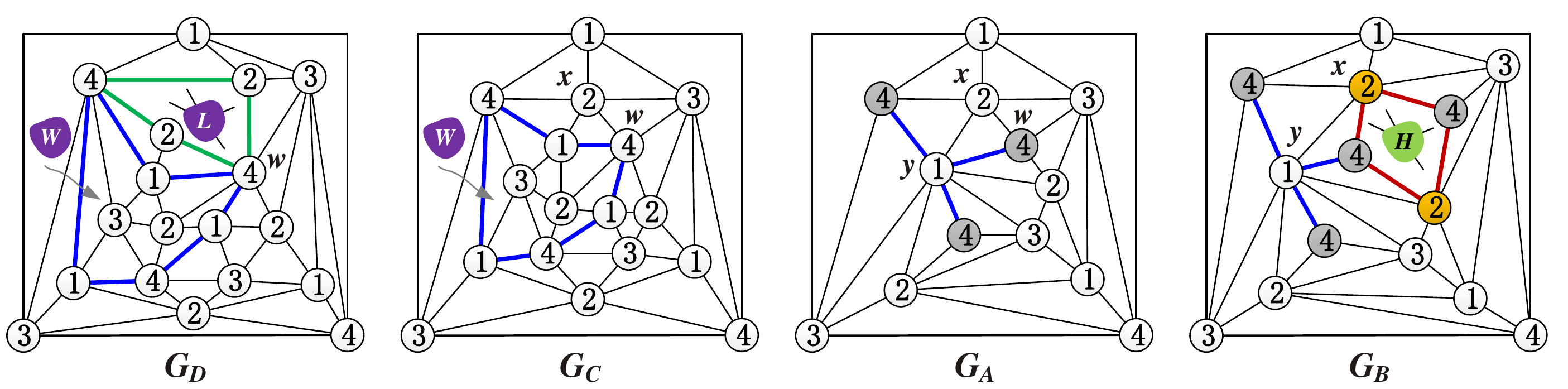}\\
\caption{\label{fig:2-color-cycle-exp-contrac} {\small Four colored maximal planar graphs, where $G_A$ is a \emph{public-key} and $G_B$, $G_C$ and $G_D$ are three \emph{private-keys}.}}
\end{figure}

\begin{rem}\label{rem:333333}
In practical application, a cryptographer transforms the public-key into a form known only to the decryptor and sends it to the decryptor. For example, any one of four uncolored maximal planar graphs $UG_A$, $UG_B$, $UG_C$ and $UG_D$ shown in Fig.\ref{fig:no-color-cycle-exp-contrac} can be selected as a \emph{public-key}. However, it is difficult to get four colored maximal planar graphs $G_A$, $G_B$, $G_C$ and $G_D$ from the uncolored maximal planar graphs $UG_A$, $UG_B$, $UG_C$ and $UG_D$, it is not even possible to decrypt the topological authentication in the valid time.\paralled
\end{rem}

\begin{figure}[h]
\centering
\includegraphics[width=16.4cm]{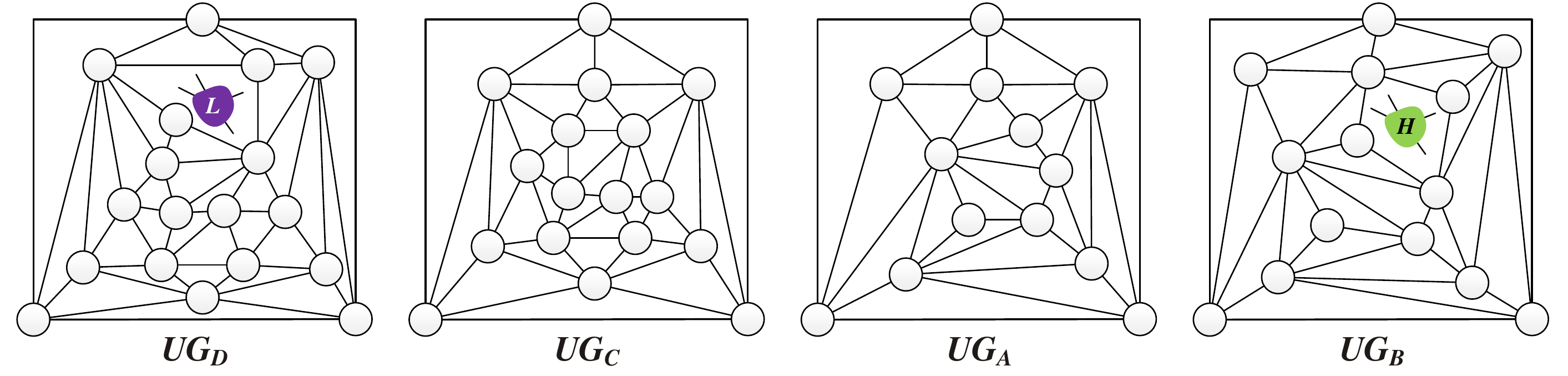}\\
\caption{\label{fig:no-color-cycle-exp-contrac} {\small Four uncolored maximal planar graphs $UG_A$, $UG_B$, $UG_C$ and $UG_D$.}}
\end{figure}

\subsubsection{Graph homomorphism}

\begin{defn}\label{defn:definition-graph-homomorphism}
\cite{Bondy-2008} A \emph{graph homomorphism} $G\rightarrow H$ from a graph $G$ into another graph $H$ is a mapping $f: V(G) \rightarrow V(H)$ such that $f(u)f(v)\in E(H)$ for each edge $uv \in E(G)$.\qqed
\end{defn}

\begin{thm}\label{thm:sequence-graph-homomorphisms}
\cite{Bing-Yao-Hongyu-Wang-arXiv-2020-homomorphisms} There are infinite graphs $G^*_n$ forming a sequence $\{G^*_n\}^{\infty}_{n=1}$, such that $G^*_n\rightarrow G^*_{n-1}$ is really a graph homomorphism for $n\geq 1$.
\end{thm}

\begin{defn}\label{defn:11-topo-auth-faithful}
\cite{Gena-Hahn-Claude-Tardif-1997} A graph homomorphism $\varphi: G \rightarrow H$ is called \emph{faithful} if $\varphi(G)$ is an induced subgraph of $H$, and called \emph{full} if $uv\in E(G)$ if and only if $\varphi(u)\varphi(v) \in E(H)$.\qqed
\end{defn}

\begin{thm}\label{thm:bijective-graph-homomorphism}
\cite{Gena-Hahn-Claude-Tardif-1997} A faithful bijective graph homomorphism is an isomorphism, that is $G \cong H$.
\end{thm}

\subsection{Various matrices of graphs and their operations}

A graph of graph theory is saved in computer by its adjacent matrix, or its incident matrix, in general. A colored Hanzi-graph $H_{abcd}$ can be saved in computer by one of its Topcode-matrix $T_{code}(H_{abcd})$ defined in Definition \ref{defn:topcode-matrix-definition}, adjacent matrix $A(H_{abcd})$ defined in \cite{Bondy-2008}, adjacent e-value matrix $E_{col}(H_{abcd})$ defined in Definition \ref{defn:e-value-matrix} and the adjacent ve-value matrix $VE_{col}(H_{abcd})$ defined in Definition \ref{defn:ve-value-matrix}. A colored graph $G$ drawn on planar paper can be scanned into computer by modern image recognition technology, and then switch them into various matrices for computation. Hanzi-graphs can be feeded into a computer by speaking, writing and keyboard inputs in Chinese.

\subsubsection{Various matrices of graphs}

Since two adjacent matrices $A(T)$ and $A(L)$ are not similar from each other for two non-isomorphic graphs $T$ and $L$, in other words, there is no matrix $B$ holding $B^{-1}A(T)B=A(L)$ true. Thereby, we consider a graph $G$ has its unique adjacent matrix $A(G)$.

\begin{defn}\label{defn:topcode-matrix-definition}
\cite{Yao-Sun-Zhao-Li-Yan-ITNEC-2017, Yao-Zhao-Zhang-Mu-Sun-Zhang-Yang-Ma-Su-Wang-Wang-Sun-arXiv2019} A \emph{Topcode-matrix} (or \emph{topological coding matrix}) is defined as
\begin{equation}\label{eqa:Topcode-matrix}
\centering
{
\begin{split}
T_{code}= \left(
\begin{array}{ccccc}
x_{1} & x_{2} & \cdots & x_{q}\\
e_{1} & e_{2} & \cdots & e_{q}\\
y_{1} & y_{2} & \cdots & y_{q}
\end{array}
\right)_{3\times q}=
\left(\begin{array}{c}
X\\
E\\
Y
\end{array} \right)=(X,~E,~Y)^{T}
\end{split}}
\end{equation} where \emph{v-vector} $X=(x_1, x_2, \dots, x_q)$, \emph{e-vector} $E=(e_1$, $e_2 $, $ \dots $, $e_q)$, and \emph{v-vector} $Y=(y_1, y_2, \dots, y_q)$ consist of non-negative integers $e_i$, $x_i$ and $y_i$ for $i\in [1,q]$. We say $T_{code}$ to be \emph{evaluated} if there exists a function $f$ such that $e_i=f(x_i,y_i)$ for $i\in [1,q]$, and call $x_i$ and $y_i$ to be the \emph{ends} of $e_i$, and $q$ the \emph{size} of $T_{code}$.\qqed
\end{defn}

\begin{rem}\label{rem:Topcode-matrix}
Let $\{w_i\}^p_{i=1}$ be the set of distinct elements of $x_1, x_2, \dots, x_q$, $y_1, y_2, \dots, y_q$ in $T_{code}$, and let $\textrm{deg}(w_i)$ be the time of $w_i$ appeared in $T_{code}$, and $n_d$ be the number of $w_i$ holding $d=\textrm{deg}(w_i)$, without loss of generality, $\textrm{deg}(w_i)\geq \textrm{deg}(w_{i+1})$ with $i\in [1,p-1]$. If the sequence $\{\textrm{deg}(w_i)\}^p_{i=1}$ holds Erd\"{o}s-Galia Theorem shown in Lemma \ref{thm:basic-degree-sequence-lemma}, we say that the Topcode-matrix $T_{code}$ is \emph{graphicable}.

For a fixed Topsnut-matrix $T_{code}$, we define its reciprocal as: $T^{-1}_{code}=(X^{-1},E^{-1},T^{-1})^T$ with $X^{-1}=(x_q,x_{q-1},\dots,x_2,x_1)$, $E^{-1}=(e_q,e_{q-1},\dots,e_2,e_1)$ and $Y^{-1}=(y_q,y_{q-1},\dots,y_2,y_1)$.\paralled
\end{rem}

\begin{defn}\label{defn:e-value-matrix}
\cite{Yao-Wang-2106-15254v1} An \emph{adjacent e-value matrix} $E_{col}(G)=(c_{i,j})_{p\times p}$ is defined on a $(p,q)$-graph $G$ admitting a total coloring $f: V(G)\cup E(G)\rightarrow [a,b]$ for $V(G)=\{x_1,x_2,\dots ,x_p\}$ in the way:

(i) $a_{i,j}=f(x_ix_j)$ with $i\in[1,p]$, $j\in[1,p]$ and $i\neq j$;

(ii) $a_{i,i}=0$ with $i\in[1,p]$.\qqed
\end{defn}

\begin{defn}\label{defn:ve-value-matrix}
\cite{Yao-Mu-Sun-Sun-Zhang-Wang-Su-Zhang-Yang-Zhao-Wang-Ma-Yao-Yang-Xie2019} An \emph{adjacent ve-value matrix} $VE_{col}(G)=(a_{i,j})_{(p+1)\times (p+1)}$ is defined on a $(p,q)$-graph $G$ admitting a total coloring $g: V(G)\cup E(G)\rightarrow [a,b]$ for $V(G)=\{x_1,x_2,\dots ,x_p\}$ in the way:

(i) $a_{1,1}=0$, $a_{1,j+1}=g(x_j)$ with $j\in[1,p]$, and $a_{k+1,1}=g(x_k)$ with $k\in[1,p]$;

(ii) $a_{i+1,i+1}=0$ with $i\in[1,p]$;

(iii) For an edge $x_{ij}=x_ix_j\in E(G)$ with $i,j\in[1,p]$, then $a_{i+1,j+1}=g(x_{ij})$, and $a_{i+1,j+1}=0$ otherwise.\qqed
\end{defn}

We use a colored Hanzi-graph $H_{4043}$ shown in Fig.\ref{fig:6-hanzis} to illustrating the Topcode-matrix $T_{code}(H_{4043})$, the adjacent matrix $A(H_{4043})$, the adjacent e-value matrix $E_{col}(H_{4043})$ and the adjacent ve-value matrix $VE_{col}(H_{4043})$ in Eq.(\ref{eqa:4-matrices-Hanzi-graph-H-4043-11}) and Eq.(\ref{eqa:4-matrices-Hanzi-graph-H-4043-22}), respectively.

\begin{equation}\label{eqa:4-matrices-Hanzi-graph-H-4043-11}
\centering
T_{code}(H_{4043})= \left(
\begin{array}{cccccc}
2 & 1 & 1 & 1 & 0 \\
1 & 2 & 3 & 4 & 5 \\
3 & 3 & 4 & 5 & 5
\end{array}
\right),~A(H_{4043})= \left(
\begin{array}{cccccc}
0 & 0 & 0 & 0 & 0 & 1 \\
0 & 0 & 0 & 1 & 1 & 1 \\
0 & 0 & 0 & 1 & 0 & 0 \\
0 & 1 & 1 & 0 & 0 & 0 \\
0 & 1 & 0 & 0 & 0 & 0 \\
1 & 1 & 0 & 0 & 0 & 0
\end{array}
\right)
\end{equation}

\begin{equation}\label{eqa:4-matrices-Hanzi-graph-H-4043-22}
\centering
E_{col}(H_{4043})= \left(
\begin{array}{cccccc}
0 & 0 & 0 & 0 & 0 & 5 \\
0 & 0 & 0 & 2 & 3 & 4 \\
0 & 0 & 0 & 1 & 0 & 0 \\
0 & 2 & 1 & 0 & 0 & 0 \\
0 & 3 & 0 & 0 & 0 & 0 \\
5 & 4 & 0 & 0 & 0 & 0
\end{array}
\right),~VE_{col}(H_{4043})= \left(
\begin{array}{ccccccc}
0 & 0 & 1 & 2 & 3 & 4 & 5 \\
0 & 0 & 0 & 0 & 0 & 0 & 5 \\
1 & 0 & 0 & 0 & 2 & 3 & 4 \\
2 & 0 & 0 & 0 & 1 & 0 & 0 \\
3 & 0 & 2 & 1 & 0 & 0 & 0 \\
4 & 0 & 3 & 0 & 0 & 0 & 0 \\
5 & 5 & 4 & 0 & 0 & 0 & 0
\end{array}
\right)
\end{equation}

\begin{defn}\label{defn:colored-degree-sequence-matrix}
\cite{Yao-Wang-2106-15254v1} Let $\textbf{\textrm{d}}=(a_1,a_2,\dots ,a_p)$ with $a_j\in Z^0\setminus \{0\}$ be a degree-sequence. There exists a coloring $f: \textbf{\textrm{d}}\rightarrow \{b_1,b_2,\dots ,b_p\}$ with $b_j\in Z^0$, and the following matrix
\begin{equation}\label{eqa:degree-sequence-matrix}
\centering
{
\begin{split}
D_{sc}(\textbf{\textrm{d}})= \left(
\begin{array}{ccccc}
a_1 & a_2 & \cdots & a_p\\
f(a_1) & f(a_2) & \cdots & f(a_p)\\
\end{array}
\right)_{2\times p}=(\textbf{\textrm{d}},~~f(\textbf{\textrm{d}}))^T
\end{split}}
\end{equation} is called a \emph{colored degree-sequence matrix} (Cds-matrix), where $f(\textbf{\textrm{d}})=(f(a_1), f(a_2), \dots , f(a_p))$ is a \emph{colored vector} of the degree sequence $\textbf{\textrm{d}}$.\qqed
\end{defn}

\begin{rem}\label{rem:333333}
Each one of the adjacent matrix $A(G)$, the adjacent e-value matrix $E_{col}(G)$ and the adjacent ve-value matrix $VE_{col}(G)$ corresponds a unique graph, but a Topcode-matrix $T_{code}$ (resp. Cds-matrix) may correspond two or more graphs, so does Cds-matrix. We have graphic Topcode-matrix and string Topcode-matrix introduced in the section ``Realization of topological authentication''.

There are other type Topcode-matrices in topological coding, for example,

(i) Set-type Topcode-matrices induced in \cite{YAO-SUN-WANG-SU-XU2018arXiv} can be served in hypergraphs.

(ii) Graph Topcode-matrices has its own elements to be graphs.

(iii) Matrix Topcode-matrices has its own elements to be matrices.

(iv) Hanzi Topcode-matrices with elements to be Chinese letters.

(v) Dynamic network Topcode-matrices with dynamic elements at each time step.

We can consider some particular Topcode-matrices having elements: text-based strings, letters, functions, sets, Hanzis, graphs, groups and so on.\paralled
\end{rem}

\begin{defn}\label{defn:pan-Topcode-matrix}
\cite{Yao-Wang-Ma-Su-Wang-Sun-2020ITNEC} Let $P_{code}=(X_{pan}, E_{pan}, Y_{pan}~)^{T}$ be a \emph{pan-Topcode-matrix}, where three vectors $X_{pan}=(\alpha_1, \alpha_2, \dots , \alpha_q)$, $E_{pan}=(\gamma_1, \gamma_2, \dots , \gamma_q)$ and $Y_{pan}=(\beta_1, \beta_2, \dots , \beta_q)$, and $\alpha_j,\beta_j$ are the ends of $\gamma_j$. $(X_{pan}Y_{pan})^*$ is the set of different elements of the union set $\{\alpha_1, \alpha_2, \dots , \alpha_q,\beta_1, \beta_2, \dots $, $\beta_q\}$, and $P^*_E$ is the set of different elements of the set $\{\gamma_1, \gamma_2, \dots , \gamma_q\}$. If there exits a function $F$ such that $\gamma_j=F\langle \alpha_j,\beta_j\rangle$ for each $j\in [1,q]$, then the pan-Topcode-matrix $P_{code}$ is \emph{valued}.\qqed
\end{defn}

\subsubsection{Operations on Topcode-matrices}

In \cite{Bing-Yao-2020arXiv}, the authors introduced the following operations on Topcode-matrices:

\vskip 0.2cm

\textbf{1. Union-sum operation ``$\uplus$'', union-sum Topcode-matrix and sub-Topcode-matrix.} Let $T^i_{code}=(X_i,~E_i$, $Y_i)^{T}_{3\times q_i}$ be Topcode-matrices, where
$$
X_i=(x^i_1, x^i_2, \dots, x^i_{q_i}),~E_i=(e^i_1, e^i_2, \dots ,e^i_{q_i}),~Y_i=(y^i_1, y^i_2,\dots,y^i_{q_i}),~i=1,2
$$ The union operation ``$\uplus$'' of two Topcode-matrices $T^1_{code}$ and $T^2_{code}$ is defined by
\begin{equation}\label{eqa:union-operation-2-matrices}
T^1_{code}\uplus T^2_{code}=(X_1\uplus X_2,~E_1\uplus E_2,~Y_1\uplus Y_2)^{T}_{3\times (q_1+q_2)}
\end{equation} with

$X_1\uplus X_2=(x^1_1, x^1_2, \dots, x^1_{q_1},x^2_1, x^2_2, \dots, x^2_{q_2})$, $E_1\uplus E_2=(e^1_1, e^1_2, \dots, e^1_{q_1},e^2_1, e^2_2, \dots, e^2_{q_2})$ and

$Y_1\uplus Y_2=(y^1_1, y^1_2, \dots, y^1_{q_1},y^2_1, y^2_2, \dots, y^2_{q_2})$,\\
and the process of obtaining the \emph{union-sum Topcode-matrix} $T^1_{code}\uplus T^2_{code}$ is called a \emph{union operation}, we call $T^i_{code}$ to be a \emph{sub-Topcode-matrix} of the union-sum Topcode-matrix $T^1_{code}\uplus T^2_{code}$, denoted as $T^i_{code}\subset T^1_{code}\uplus T^2_{code}$ $i=1,2$.

Moreover, we can generalize Eq.(\ref{eqa:union-operation-2-matrices}) to the following form
\begin{equation}\label{eqa:union-operation-more-matrices}
\uplus |^m_{i=1}T^i_{code}=T^1_{code}\uplus T^2_{code}\uplus \cdots \uplus T^m_{code}=(X^*,~E^*,~Y^*)^{T}_{3\times q^*}
\end{equation} where

$X^*=X_1\uplus X_2\uplus \cdots \uplus X_m$, $E^*=E_1\uplus E_2\uplus \cdots \uplus E_m$ and

$Y^*=Y_1\uplus Y_2\uplus \cdots \uplus Y_m$, as well as $q^*=q_1+q_2+\cdots +q_m$ with $m\geq 2$.\\
And each $T^i_{code}$ is a \emph{sub-Topcode-matrix} of the union-sum Topcode-matrix $\uplus |^m_{i=1}T^i_{code}$.

\vskip 0.2cm

\textbf{2. Subtraction operation ``$\setminus$''.} Let $T_{code}=T^1_{code}\uplus T^2_{code}$ be the union-sum Topcode-matrix defined in Eq.(\ref{eqa:union-operation-2-matrices}). We remove $T^1_{code}$ from $T_{code}$, and obtain
\begin{equation}\label{eqa:subtraction-operation-2-matrices}
{
\begin{split}
T_{code}\setminus T^1_{code}&=(X_1\uplus X_2\setminus X_1,~E_1\uplus E_2\setminus E_1,~Y_1\uplus Y_2\setminus Y_1)^{T}_{3\times (q_1+q_2-q_1)}\\
&=(X_2,~E_2,~Y_2)^{T}_{3\times q_2}=T^2_{code}
\end{split}}
\end{equation}
the process of obtaining $T^2_{code}$ from $T_{code}$ is called a \emph{subtraction operation}.

\vskip 0.2cm

\textbf{3. Intersection operation ``$\cap$''.} Let $T\,'_{code}=T^1_{code}\uplus H_{code}$ and $T\,''_{code}=T^2_{code}\uplus H_{code}$ be two Topcode-matrices obtained by doing the union operation to Topcode-matrices $T^1_{code}, H_{code}$ and $T^2_{code}$. Since the Topcode-matrix $H_{code}$ is a \emph{common sub-Topcode-matrix} of two Topcode-matrices $T\,'_{code}$ and $T\,''_{code}$, we take the largest common sub-Topcode-matrix of $T\,'_{code}$ and $T\,''_{code}$ and denote it as $T\,'_{code}\cap T\,''_{code}$ called \emph{intersected Topcode-matrix}, this is so-called the \emph{intersection operation}, so we have
\begin{equation}\label{eqa:555555}
H_{code}\subseteq T\,'_{code}\cap T\,''_{code},~T\,'_{code}\cap T\,''_{code}\subseteq T\,'_{code},~T\,'_{code}\cap T\,''_{code}\subseteq T\,''_{code}
\end{equation}

\vskip 0.2cm

\textbf{4. Union operation ``$\cup$''.} For the intersected Topcode-matrix $T\,'_{code}\cap T\,''_{code}$ based on two Topcode-matrices $T\,'_{code}$ and $T\,''_{code}$, we define the union operation ``$\cup$'' on $T\,'_{code}$ and $T\,''_{code}$ as
\begin{equation}\label{eqa:555555}
T\,'_{code}\cup T\,''_{code}=T\,'_{code}\uplus [T\,''_{code}\setminus T\,'_{code}\cap T\,''_{code}]=[T\,'_{code}\setminus T\,'_{code}\cap T\,''_{code}]\uplus T\,''_{code}
\end{equation}

\vskip 0.2cm

\textbf{5. Coinciding operation ``$\odot$''.} Let $T\,'_{code}=T^1_{code}\uplus H_{code}$ and $T\,''_{code}=T^2_{code}\uplus H_{code}$ be two Topcode-matrices obtained by doing the union operation to Topcode-matrices $T^1_{code}, H_{code}$ and $T^2_{code}$. So $H_{code}\subset T\,'_{code}$ and $H_{code}\subset T\,''_{code}$. We define the \emph{coinciding operation} ``$\odot$'' on two Topcode-matrices $T\,'_{code}$ and $T\,''_{code}$ as follows
\begin{equation}\label{eqa:intersection-operation-2-matrices}
[\odot H_{code}]\langle T\,'_{code},T\,''_{code}\rangle =T^1_{code}\uplus H_{code}\uplus T^2_{code}=[T\,'_{code}\setminus H_{code}]\uplus T\,''_{code}=[T\,''_{code}\setminus H_{code}]\uplus T\,'_{code}
\end{equation}

\vskip 0.2cm

\textbf{6. Splitting operation ``$\wedge $''.} We do the splitting operation to the sub-Topcode-matrix $H_{code}$ of a coincided Topcode-matrix $T^1_{code}\uplus H_{code}\uplus T^2_{code}$, such that the resultant Topcode-matrix consisted of two disjoint Topcode-matrices $T^1_{code}\uplus H_{code}$ and $T^2_{code}\uplus H_{code}$, denoted as
\begin{equation}\label{eqa:splitting-operation-2-matrices}
\{[\odot H_{code}]\langle T\,'_{code},T\,''_{code}\rangle \}\wedge H_{code}=[T^1_{code}\uplus H_{code}\uplus T^2_{code}]\wedge H_{code}
\end{equation}

\begin{example}\label{exa:8888888888}
For the following three Topcode-matrices
\begin{equation}\label{eqa:three-topcode-matrces}
\centering
{
\begin{split}
T_{code}= \left(
\begin{array}{ccccc}
x_{1} & x_{2} & \cdots & x_{n}\\
e_{1} & e_{2} & \cdots & e_{n}\\
y_{1} & y_{2} & \cdots & y_{n}
\end{array}
\right),~T^i_{code}= \left(
\begin{array}{ccccc}
x^i_{1} & x^i_{2} & \cdots & x^i_{m}\\
e^i_{1} & e^i_{2} & \cdots & e^i_{m}\\
y^i_{1} & y^i_{2} & \cdots & y^i_{m}
\end{array}
\right),~i=1,2
\end{split}}
\end{equation}
\noindent we have two union-sum Topcode-matrices

\begin{equation}\label{eqa:two-sum-topcode-matrces}
\centering
{
\begin{split}
T_{code}\uplus T^i_{code}= \left(
\begin{array}{cccccccc}
x_{1} & x_{2} & \cdots & x_{n} & x^i_{1} & x^i_{2} & \cdots & x^i_{m}\\
e_{1} & e_{2} & \cdots & e_{n} & e^i_{1} & e^i_{2} & \cdots & e^i_{m}\\
y_{1} & y_{2} & \cdots & y_{n} & y^i_{1} & y^i_{2} & \cdots & y^i_{m}
\end{array}
\right),~i=1,2
\end{split}}
\end{equation}
 A coincided Topcode-matrix is
\begin{equation}\label{eqa:matrces-coinciding-operation}
J_{code}=[\odot T^2_{code}]\langle T_{code}\uplus T^2_{code},T^1_{code}\uplus T^2_{code}\rangle
\end{equation}
and
\begin{equation}\label{eqa:topcode-matrces-coinciding-operation}
\centering
{
\begin{split}
J_{code}= \left(
\begin{array}{cccccccccccc}
x_{1} & x_{2} & \cdots & x_{n} & x^2_{1} & x^2_{2} & \cdots & x^2_{s} & x^1_{1} & x^1_{2} & \cdots & x^1_{m}\\
e_{1} & e_{2} & \cdots & e_{n} & e^2_{1} & e^2_{2} & \cdots & e^2_{s}& e^1_{1} & e^1_{2} & \cdots & e^1_{m}\\
y_{1} & y_{2} & \cdots & y_{n} & y^2_{1} & y^2_{2} & \cdots & y^2_{s}& y^1_{1} & y^1_{2} & \cdots & y^1_{m}
\end{array}
\right)=T_{code}\uplus T^2_{code}\uplus T^1_{code}
\end{split}}
\end{equation}
\end{example}

\begin{example}\label{exa:8888888888}
Given two Topcode-matrices
\begin{equation}\label{eqa:subtractive-union-intersection-1}
\centering
A= \left(
\begin{array}{ccccccccc}
\textbf{\textcolor[rgb]{0.00,0.00,1.00}{7}} & 5 & 7 &\textbf{\textcolor[rgb]{0.00,0.00,1.00}{1}}\\
\textbf{\textcolor[rgb]{0.00,0.00,1.00}{1}} & 3 & 5 &\textbf{\textcolor[rgb]{0.00,0.00,1.00}{7}} \\
\textbf{\textcolor[rgb]{0.00,0.00,1.00}{18}}& 18 & 14 &\textbf{\textcolor[rgb]{0.00,0.00,1.00}{18}}
\end{array}
\right)~
B= \left(
\begin{array}{ccccccccc}
\textbf{\textcolor[rgb]{0.00,0.00,1.00}{7}} &\textbf{\textcolor[rgb]{0.00,0.00,1.00}{1}} & 5 \\
\textbf{\textcolor[rgb]{0.00,0.00,1.00}{1}} &\textbf{\textcolor[rgb]{0.00,0.00,1.00}{7}}& 9\\
\textbf{\textcolor[rgb]{0.00,0.00,1.00}{18}} &\textbf{\textcolor[rgb]{0.00,0.00,1.00}{18}}& 12
\end{array}
\right)
\end{equation}
\noindent we have
\begin{equation}\label{eqa:subtractive-union-intersection-2}
\centering
A\uplus B= \left(
\begin{array}{ccccccccc}
\textbf{\textcolor[rgb]{0.00,0.00,1.00}{7}} & 5 & 7 &\textbf{\textcolor[rgb]{0.00,0.00,1.00}{1}}& \textbf{\textcolor[rgb]{0.00,0.00,1.00}{7}} &\textbf{\textcolor[rgb]{0.00,0.00,1.00}{1}} & 5 \\
\textbf{\textcolor[rgb]{0.00,0.00,1.00}{1}} & 3 & 5 &\textbf{\textcolor[rgb]{0.00,0.00,1.00}{7}}& \textbf{\textcolor[rgb]{0.00,0.00,1.00}{1}} &\textbf{\textcolor[rgb]{0.00,0.00,1.00}{7}}& 9\\
\textbf{\textcolor[rgb]{0.00,0.00,1.00}{18}}& 18 & 14 &\textbf{\textcolor[rgb]{0.00,0.00,1.00}{18}}& \textbf{\textcolor[rgb]{0.00,0.00,1.00}{18}} &\textbf{\textcolor[rgb]{0.00,0.00,1.00}{18}}& 12
\end{array}
\right),~A\cup B= \left(
\begin{array}{ccccccccc}
\textbf{\textcolor[rgb]{0.00,0.00,1.00}{7}} & 5 & 7 &\textbf{\textcolor[rgb]{0.00,0.00,1.00}{1}}& 5 \\
\textbf{\textcolor[rgb]{0.00,0.00,1.00}{1}} & 3 & 5 &\textbf{\textcolor[rgb]{0.00,0.00,1.00}{7}}& 9 \\
\textbf{\textcolor[rgb]{0.00,0.00,1.00}{18}}& 18 & 14 &\textbf{\textcolor[rgb]{0.00,0.00,1.00}{18}}& 12
\end{array}
\right)
\end{equation}

\begin{equation}\label{eqa:subtractive-union-intersection-3}
\centering
~B\setminus A= \left(
\begin{array}{ccccccccc}
 5 \\
9 \\
12
\end{array}
\right),~A\cap B= \left(
\begin{array}{ccccccccc}
\textbf{\textcolor[rgb]{0.00,0.00,1.00}{7}} &\textbf{\textcolor[rgb]{0.00,0.00,1.00}{1}} \\
\textbf{\textcolor[rgb]{0.00,0.00,1.00}{1}} &\textbf{\textcolor[rgb]{0.00,0.00,1.00}{7}}\\
\textbf{\textcolor[rgb]{0.00,0.00,1.00}{18}} &\textbf{\textcolor[rgb]{0.00,0.00,1.00}{18}}
\end{array}
\right),~A\setminus B= \left(
\begin{array}{ccccccccc}
 5 & 7 \\
 3 & 5 \\
18 & 14
\end{array}
\right)
\end{equation}

The split Topcode-matrix $(A\cup B)\wedge (A\cap B)$ consists of $A$ and $B$, and
\begin{equation}\label{eqa:555555}
{
\begin{split}
&A=(A\uplus B)\setminus B,~B=(A\uplus B)\setminus A,~(A\setminus B)\uplus (B\setminus A)=(A\cup B)\setminus (A\cap B)\\
&A\cup B=(A\setminus B)\uplus (B\setminus A)\uplus (A\cap B), ~A\cup B=[\odot (A\cap B)]\langle A, B\rangle\\
&[\odot (A\cap B)]\langle A, B\rangle=(A\setminus B)\uplus B=A\uplus (B\setminus A)
\end{split}}
\end{equation}
\end{example}

\subsubsection{Algebraic operations of real-valued Topcode-matrices}

We define some algebraic operations on Topcode-matrices \cite{Yao-Wang-Ma-Su-Wang-Sun-2020ITNEC}. In Definition \ref{defn:topcode-matrix-definition}, the Topcode-matrix $I_{code}=(X, E, Y)^{T}$ with $x_i=1$, $e_i=1$ and $y_i=1$ for $i\in [1,q]$ is called the \emph{unit Topcode-matrix}. For two Topcode-matrices $T^j_{code}=(X_j, E_j, Y_j)^{T}$ with $j=1,2$, where $X_j=(x^j_1, x^j_2, \cdots , x^j_q)$, $E_j=(e^j_1, e^j_2, \cdots , e^j_q)$ and $Y_j=(y^j_1, y^j_2, \cdots , y^j_q)$.

\vskip 0.2cm

\textbf{1.} The \emph{coefficient multiplication} of a function $f(x)$ and a Topcode-matrix $T^j_{code}$ is defined by
\begin{equation}\label{eqa:555555}
{
\begin{split}
f(x)\cdot T^j_{code}=f(x)\cdot(X_j, E_j, Y_j)^{T}=(f(x)\cdot X_j, f(x)\cdot E_j, f(x)\cdot Y_j)^{T}
\end{split}}
\end{equation}
where

$f(x)\cdot X_j=(f(x)\cdot x^j_1, f(x)\cdot x^j_2, \cdots , f(x)\cdot x^j_q)$,

$f(x)\cdot E_j=(f(x)\cdot e^j_1, f(x)\cdot e^j_2, \cdots , f(x)\cdot e^j_q)$ and

$f(x)\cdot Y_j=(f(x)\cdot y^j_1, f(x)\cdot y^j_2, \cdots , f(x)\cdot y^j_q)$.

\vskip 0.2cm

\textbf{2.} The \emph{addition operation} and \emph{subtraction operation} between two Topcode-matrices $T^1_{code}$ and $T^2_{code}$ is denoted as $T^1_{code}\pm T^2_{code}$, and
\begin{equation}\label{eqa:555555}
T^1_{code}\pm T^2_{code}=(X_1\pm X_2, E_1\pm E_2, Y_1\pm Y_2)^{T}
\end{equation}
where

$X_1\pm X_2=(x^1_1\pm x^2_1, x^1_2\pm x^2_2, \cdots , x^1_q\pm x^2_q)$,

$E_1\pm E_2=(e^1_1\pm e^2_1, e^1_2\pm e^2_2, \cdots , e^1_q\pm e^2_q)$ and

$Y_1\pm Y_2=(y^1_1\pm y^2_1, y^1_2\pm y^2_2, \cdots , y^1_q\pm y^2_q)$.

If some $e^1_j-e^2_j\leq 0$, or $x^1_i-x^2_i\leq 0$ and $y^1_k-y^2_k\leq 0$ in $T^1_{code}- T^2_{code}$, we can consider complex graphs introduced in the later sections to approach $T^1_{code}-T^2_{code}$.

\vskip 0.2cm

\textbf{3.} Let $\alpha(\varepsilon)$ and $\beta(\varepsilon)$ be two real-valued functions, we have a \emph{real-valued Topcode-matrix} $R_{code}$ defined as:
$R_{code}=\alpha(\varepsilon)T^1_{code}+\beta(\varepsilon)T^2_{code}$
and another real-valued Topcode-matrix
\begin{equation}\label{eqa:real-valued-topcode-matrix}
R_{code}(f_\varepsilon,G)=\alpha(\varepsilon)I_{code}+\beta(\varepsilon)T_{code}(G)
\end{equation}where $I_{code}$ is the unit Topcode-matrix, $T_{code}(G)$ is a Topcode-matrix of a $(p,q)$-graph $G$.

Clearly, the text-based passwords induced by the real-valued Topcode-matrix $R_{code}(f_\varepsilon,G)$ are complex than that induced by a Topcode-matrix of $G$, and have huge numbers, since two functions $\alpha(\varepsilon)$ and $\beta(\varepsilon)$ are real and various. There are some relationships between a Topcode-matrix $T_{code}(G)$ and a real-valued Topcode-matrix $R_{code}(f_\varepsilon,G)$ as follows:

(1) If $e_i=|x_i-y_i|$ in a Topcode-matrix $T_{code}(G)$ of a $(p,q)$-graph $G$, then we have $\alpha(\varepsilon)+\beta(\varepsilon)|x_i-y_i|$ in the real-valued Topcode-matrix $R_{code}(f_\varepsilon,G)$.

(2) If $x_i+e_i+y_i=k$ in $T_{code}(G)$, then we have $3\alpha(\varepsilon)+\beta(\varepsilon)\cdot k$ in $R_{code}(f_\varepsilon,G)$.

(3) If $e_i+|x_i-y_i|=k$ in $T_{code}(G)$, then we have $\alpha(\varepsilon)+\beta(\varepsilon)\cdot k$ in $R_{code}(f_\varepsilon,G)$.

(4) If $|x_i+y_i-e_i|=k$ in $T_{code}(G)$, then we have $\alpha(\varepsilon)+\beta(\varepsilon)\cdot k$ in $R_{code}(f_\varepsilon,G)$ if $e_i-(x_i+y_i)\geq 0$; otherwise $|x_i+y_i-e_i|=k$ corresponds $|\beta(\varepsilon)\cdot k-\alpha(\varepsilon)|$.

(5) If $\big ||x_i-y_i|-e_i\big |=k$ in $T_{code}(G)$, then we have $\alpha(\varepsilon)+\beta(\varepsilon)\cdot k$ in $R_{code}(f_\varepsilon,G)$ for $|x_i+y_i|-e_i< 0$; otherwise $\big ||x_i-y_i|-e_i\big |=k$ corresponds $|\beta(\varepsilon)\cdot k-\alpha(\varepsilon)|$.

\subsubsection{Topcode-matrix equations, Merge operation}

\textbf{1. Topcode-matrix equations}

There are Topcode-matrices
\begin{equation}\label{eqa:Topcode-matrix}
\centering
{
\begin{split}
T^*_{code}= \left(
\begin{array}{ccccc}
x^*_1 & x^*_2 & \cdots & x_{q}\\
e^*_1 & e^*_2 & \cdots & e_{q}\\
y^*_1 & y^*_2 & \cdots & y^*_{q}
\end{array}
\right)_{3\times q}=
\left(\begin{array}{c}
X\\
E\\
Y
\end{array} \right)=(X,~E,~Y)^{T}
\end{split}}
\end{equation} with $X=(x^*_1, x^*_2, \dots, x^*_q)$, $E=(e^*_1,e^*_2 , \dots ,e^*_q)$, and $Y=(y^*_1, y^*_2, \dots, y^*_q)$ consist of number-based strings $e^*_i$, $x^*_i$ and $y^*_i$ for $i\in [1,q]$, and
\begin{equation}\label{eqa:ith-topcode-matrix}
\centering
{
\begin{split}
T^i_{code}= \left(
\begin{array}{ccccc}
x_{i,1} & x_{i,2} & \cdots & x_{i,q}\\
e_{i,1} & e_{i,2} & \cdots & e_{i,q}\\
y_{i,1} & y_{i,2} & \cdots & y_{i,q}
\end{array}
\right)_{3\times q}=
\left(\begin{array}{c}
X_i\\
E_i\\
Y_i
\end{array} \right)=(X_i,~E_i,~Y_i)^{T}
\end{split}}
\end{equation}
We have a \emph{Topcode-matrix equation}
\begin{equation}\label{eqa:topcode-matrix-equations}
\sum ^m_{i=1}a_iT^i_{code}=T_{code}
\end{equation}
we get
\begin{equation}\label{eqa:ith-topcode-matrix-11}
\centering
{
\begin{split}
\sum ^m_{i=1}a_iT^i_{code}=\left(\begin{array}{c}
\sum ^m_{i=1}a_iX_i\\
\sum ^m_{i=1}a_iE_i\\
\sum ^m_{i=1}a_iY_i
\end{array} \right)=\left (\sum ^m_{i=1}a_iX_i,~\sum ^m_{i=1}a_iE_i,~\sum ^m_{i=1}a_iY_i\right )^{T}=(X,~E,~Y)^{T}
\end{split}}
\end{equation}
such that the elements of the Topcode-matrix $T_{code}$ hold
\begin{equation}\label{eqa:topcode-matrix-equation-parts}
x^*_k=\sum ^m_{i=1}a_ix_{i,k},~e^*_k=\sum ^m_{i=1}a_ie_{i,k},~y^*_k=\sum ^m_{i=1}a_iy_{i,k},~k\in [1,q]
\end{equation}
Thereby, we get a \emph{Topcode-matrix lattice} as follows
\begin{equation}\label{eqa:Topcode-matrix-lattices}
\textbf{\textrm{L}}(Z^0(\Sigma) \textbf{\textrm{B}}_{code})=\left \{\sum ^m_{i=1}a_iT^i_{code}:~a_i\in Z^0, T^i_{code}\in \textbf{\textrm{B}}_{code}\right \}
\end{equation} where $\textbf{\textrm{B}}_{code}=\big (T^1_{code},T^2_{code},\dots ,T^m_{code}\big )$ is the \emph{lattice base}.

Notice that judging for whether or not each matrix $\sum ^m_{i=1}a_iT^i_{code}$ in a Topcode-matrix lattice $\textbf{\textrm{L}}(Z^0(\Sigma) \textbf{\textrm{B}}_{code})$ corresponds a graph are given in Remark \ref{rem:Topcode-matrix}.

\vskip 0.4cm

\textbf{2. Number-based string merge operation, Topcode-matrix merge operation}

For each $j\in [1,m]$, we have a Topcode-matrix $H^j_{code}=(X_j,~E_j,~Y_j)^{T}_{3\times q}$ with
$$
X_j=(x_{j,1},x_{j,2},~\dots ,x_{j,q}),~E_j=(e_{j,1},e_{j,2},\dots ,e_{j,q}),~Y_j=(y_{j,1},y_{j,2},\dots ,y_{j,q})
$$ A Topcode-matrix $T_{code}=(X,~E,~Y)^{T}_{3\times q}$ has its elements to be number-based strings, and we use notation ``$S(=)S^*$'' to express the equality of two number-based strings $S$ and $S^*$, so

For $k\in [1,q]$, we have

(i) $x^*_k(=)\varepsilon _1 x_{1,k}+\varepsilon _2 x_{2,k}+\cdots +\varepsilon _m x_{m,k}$, where $\varepsilon _j=1$ if $x_{j,k}=x^*_k$, and $\varepsilon _j=0$ if $x_{j,k}\neq x^*_k$;

(ii) $e^*_k(=)\varepsilon _1 e_{1,k}+\varepsilon _2 e_{2,k}+\cdots +\varepsilon _m e_{m,k}$, where $\varepsilon _j=1$ if $e_{j,k}=e^*_k$, and $\varepsilon _j=0$ if $e_{j,k}\neq e^*_k$;

(iii) $y^*_k(=)\varepsilon _1 y_{1,k}+\varepsilon _2 y_{2,k}+\cdots +\varepsilon _m y_{m,k}$, where $\varepsilon _j=1$ if $y_{j,k}=y^*_k$, and $\varepsilon _j=0$ if $y_{j,k}\neq y^*_k$.

We define the number-based string merge operation ``$\sqcup$'' as follows:

(1) If $x_{j,k}=x^*_k$, then $x_{i,k}=0$ for $i\in [1,m]\setminus \{j\}$, then define the \emph{number-based string merge operation} as $x^*_k=x_{1,k} \sqcup x_{2,k} \sqcup \cdots \sqcup x_{j,k}\sqcup \cdots \sqcup x_{m,k}$;

(2) If $e_{j,k}=e^*_k$, then $e_{i,k}=0$ for $i\in [1,m]\setminus \{j\}$, and $e^*_k=e_{1,k} \sqcup e_{2,k} \sqcup \cdots \sqcup e_{j,k}\sqcup \cdots \sqcup e_{m,k}$; and

(3) If $y_{j,k}=y^*_k$, then $y_{i,k}=0$ for $i\in [1,m]\setminus \{j\}$, and $y^*_k=y_{1,k} \sqcup y_{2,k} \sqcup \cdots \sqcup y_{j,k}\sqcup \cdots \sqcup y_{m,k}$.

If one of $x_{j,k},e_{j,k},y_{j,k}$ is equal to one of $x^*_k,e^*_k,y^*_k$, we have $T_{code}\cap H^j_{code}\neq \emptyset $. Suppose that $T_{code}\cap H^j_{code}\neq \emptyset $ for each $j\in [1,m]$, then the Topcode-matrix merge operation ``$\sqcup$'' is defined by
\begin{equation}\label{eqa:555555}
T_{code}=H^1_{code} \sqcup H^2_{code}\sqcup \cdots \sqcup H^m_{code}=\sqcup ^m_{i=1}H^i_{code}
\end{equation}

\subsubsection{Transformations on Topcode-matrices}

We show the following operations on Topcode-matrices \cite{Yao-Sun-Zhao-Li-Yan-ITNEC-2017}:

(i) $^*$ \textbf{Dual operation.} Let $M_v=\max (XY)^*$ and $m_v=\min (XY)^*$ in an evaluated Topcode-matrix $T_{code}$ defined in Definition \ref{defn:topcode-matrix-definition}, $T_{code}=(X,E,Y)^T$ admits a function $f$ such that $e_i=f(x_i,y_i)$ for $i\in [1,q]$. We make two vectors $\overline{X}=(\overline{x}_1,\overline{x}_2,\dots ,~\overline{x}_q)$ with $\overline{x}_i=M_v+m_v-x_i$ for $i\in [1,q]$ and $\overline{Y}=(\overline{y}_1,\overline{y}_2,\dots ,\overline{y}_q)$ with $\overline{y}_j=M_v+m_v-y_j$ for $j\in [1,q]$, as well as $\overline{E}$ having each element $\overline{e}_i=f(\overline{x}_i,\overline{y}_i)$ or $\overline{e}_i=M_e+m_e-e_i$ for $i\in [1,q]$. The Topcode-matrix $\overline{T}_{code}=(\overline{X}, \overline{E}, \overline{Y})^{T}_{3\times q}$ is called the \emph{dual Topcode-matrix} of the Topcode-matrix $T_{code}$.

(ii) \textbf{Column-exchanging operation} \cite{Yao-Sun-Zhao-Li-Yan-ITNEC-2017}. We exchange the positions of two columns $(x_i,e_i,y_i)^{T}$ and $(x_j~e_j~y_j)^{T}$ in $T_{code}=(X,E,Y)^T$ defined in Definition \ref{defn:topcode-matrix-definition}, so we get another Topcode-matrix $T\,'_{code}=(X\,',E\,',Y\,')^T$. In mathematical symbol, the \emph{column-exchanging operation} $c_{(i,j)}(T_{code})=T\,'_{code}$ is defined by
$${
\begin{split}
c_{(i,j)}(X)=c_{(i,j)}(x_1 ~ x_2 ~ \cdots ~\textcolor[rgb]{0.00,0.00,1.00}{x_i}~ \cdots ~\textcolor[rgb]{0.00,0.00,1.00}{x_j}~ \cdots ~x_q)=(x_1 ~ x_2 ~ \cdots ~\textcolor[rgb]{0.00,0.00,1.00}{x_j}~ \cdots ~\textcolor[rgb]{0.00,0.00,1.00}{x_i}~ \cdots ~x_q)=X\,'
\end{split}}$$
$$
{
\begin{split}
c_{(i,j)}(E)=c_{(i,j)}(e_1 ~ e_2 ~ \cdots ~\textcolor[rgb]{0.00,0.00,1.00}{e_i}~ \cdots ~\textcolor[rgb]{0.00,0.00,1.00}{e_j}~ \cdots ~e_q)=(e_1 ~ e_2 ~ \cdots ~\textcolor[rgb]{0.00,0.00,1.00}{e_j}~ \cdots ~\textcolor[rgb]{0.00,0.00,1.00}{e_i}~ \cdots ~e_q)=E\,'
\end{split}}
$$
and
$$
{
\begin{split}
c_{(i,j)}(Y)= c_{(i,j)}(y_1 ~ y_2 ~ \cdots ~\textcolor[rgb]{0.00,0.00,1.00}{y_i}~ \cdots ~\textcolor[rgb]{0.00,0.00,1.00}{y_j}~ \cdots ~y_q)=(y_1 ~ y_2 ~ \cdots ~\textcolor[rgb]{0.00,0.00,1.00}{y_j}~ \cdots ~\textcolor[rgb]{0.00,0.00,1.00}{y_i}~ \cdots ~y_q)=Y\,'
\end{split}}
$$

(iii) \textbf{XY-exchanging operation} \cite{Yao-Sun-Zhao-Li-Yan-ITNEC-2017}. We exchange the positions of $x_i$ and $y_i$ of the $i$th column of $T_{code}=(X,E,Y)^T$ defined in Definition \ref{defn:topcode-matrix-definition} by an \emph{XY-exchanging operation} $l_{(i)}$ defined as:
$${
\begin{split}
l_{(i)}(X)=l_{(i)}(x_1 ~ x_2 ~ \cdots x_{i-1}~\textcolor[rgb]{0.00,0.00,1.00}{x_i}~x_{i+1} \cdots ~x_q)=(x_1 ~ x_2 ~ \cdots x_{i-1}~\textcolor[rgb]{0.00,0.00,1.00}{y_i}~x_{i+1} \cdots ~x_q)
\end{split}}
$$
and
$${
\begin{split}
l_{(i)}(Y)=l_{(i)}(y_1 ~ y_2 ~ \cdots y_{i-1}~\textcolor[rgb]{0.00,0.00,1.00}{y_i}~y_{i+1} \cdots ~y_q)=(y_1 ~ y_2 ~ \cdots y_{i-1}~\textcolor[rgb]{0.00,0.00,1.00}{x_i}~y_{i+1} \cdots ~y_q),
\end{split}}
$$
the resultant matrix is denoted as $l_{(i)}(T_{code})$.

Now, we do a series of column-exchanging operations $c_{(i_k,j_k)}$ with $k\in [1,a]$, and a series of XY-exchanging operations $l_{(i_s)}$ with $s\in [1,b]$ to a Topcode-matrix $T_{code}$ defined in Definition \ref{defn:topcode-matrix-definition}, the resultant Topcode-matrix is written by $C_{(c,l)(a,b)}(T_{code})$.

\begin{lem}\label{thm:one-topcode-matrix-one-graph}
A graph $G$ corresponds to an its own Topcode-matrix $T_{code}(G)$ and another its own Topcode-matrix $T\,'_{code}(G)$. Then there are a series of column-exchanging operations $c_{(i_k,j_k)}$ with $k\in [1,a]$ and a series of XY-exchanging operations $l_{(i_s)}$ with $s\in [1,b]$, such that
\begin{equation}\label{eqa:one-topcode-matrix-one-graph}
C_{(c,l)(a,b)}(T_{code}(G))=T\,'_{code}(G)
\end{equation}
\end{lem}

\begin{lem}\label{thm:grapgicable-no-isomorphic}
Suppose $T_{code}$ and $T\,'_{code}$ are defined in Definition \ref{defn:topcode-matrix-definition} and grapgicable, a graph $G$ corresponds to $T_{code}$ and another graph $H$ corresponds to $T\,'_{code}$. If
\begin{equation}\label{eqa:graphs-isomorphic-Topcode-matrices}
C_{(c,l)(a,b)}(T_{code}(G))=T\,'_{code}(H)
\end{equation}
then two graphs $G$ is graph homomorphism into $H$, that is, $G\rightarrow H$.
\end{lem}

Under the XY-exchanging operation and the column-exchanging operation, a colored graph $G$ has its own \emph{standard Topcode-matrix} (representative Topcode-matrix) $A_{vev}(G)=(X,~E,~Y)^{T}$ such that $e_i\leq e_{i+1}$ with $i\in[1,q-1]$, $x_j<y_j$ with $j\in[1,q]$, since there exists no case $x_j=y_j$. Thereby, this graph $G$ has $m$ standard Topcode-matrices if it admits $m$ different colorings/labelings of graphs in graph theory.

(iv) $^*$ \textbf{Union-sum operation.} A \emph{single matrix} $X_i$ is defined as
\begin{equation}\label{eqa:one-columme-matrix}
\centering
{
\begin{split}
X_i&= \left(
\begin{array}{ccccc}
x_{i,1}\\
x_{i,2}\\
\cdots \\
x_{i,n}
\end{array}
\right)=(x_{i,1},~x_{i,2},~\cdots ,~x_{i,n})^{T}_{n\times 1}
\end{split}}
\end{equation}
Thereby, doing the \emph{union-sum operation} on $X_i=(x_{i,1},~x_{i,2},~\cdots ,~x_{i,n})^{T}_{n\times 1}$ with $i\in [1,m]$ produces
\begin{equation}\label{eqa:one-columme-matrix}
{
\begin{split}
\uplus |^m_{i=1}X_i=X_1\uplus X_2\uplus \cdots \uplus X_m=\left(\begin{array}{ccccc}
x_{1,1} & x_{2,1} & \cdots & x_{m,1}\\
x_{1,2} & x_{2,2} & \cdots & x_{m,2}\\
\cdots & \cdots & \cdots & \cdots\\
x_{1,n} & x_{2,n} & \cdots & x_{m,n}
\end{array}
\right)_{n\times m}
\end{split}}
\end{equation}

Moreover, let each $T^i_{code}=(X_i,~E_i,~Y_i)^{T}_{3\times q_i}$ be a Topcode-matrix defined in Definition \ref{defn:topcode-matrix-definition}, where v-vector $X_i=(x_{i,1}, x_{i,2}$, $ \cdots $, $x_{i,q_i})$, e-vector $E_i=(e_{i,1},e_{i,2},\cdots ,e_{i,q_i})$, and v-vector $Y_i=(y_{i,1},y_{i,2},\cdots ,y_{i,q_i})$ for $i\in[1,m]$. We have a union-sum Topcode-matrix as follows
\begin{equation}\label{eqa:matrix-operation}
{
\begin{split}
T_{code}=\uplus |^m_{i=1}T^i_{code}=T^1_{code}\uplus T^2_{code}\uplus \cdots \uplus T^m_{code}=\left(\begin{array}{ccccc}
X_1 & X_2 & \cdots & X_m\\
E_1 & E_2 & \cdots & E_m\\
Y_1 & Y_2 & \cdots & Y_m
\end{array}
\right)_{3\times mA}
\end{split}}
\end{equation}
by the union-addition operation, where $A=\sum ^m_{i=1}q_i$.

\begin{thm}\label{thm:Topcode-matrix-union-graphicable}
If each Topcode-matrix $T^i_{code}$ with $i\in[1,m]$ is \emph{graphicable}, so is the union-sum Topcode-matrix $\uplus |^m_{i=1}T^i_{code}$.
\end{thm}

\vskip 0.2cm

Theorem \ref{thm:Topcode-matrix-union-graphicable} provides us techniques for constructing new Topsnut-gpws and new \emph{combinatorial labelings}.

(i) If a graph $G$ having the union $\uplus |^m_{i=1}T^i_{code}$ of Topcode-matrices $T^i_{code}$ with $i\in[1,m]$ is connected, then $G$ is a Topsnut-gpw and an authentication for the private keys $T^1_{code},T^2_{code},\dots, T^m_{code}$.

(ii) Suppose that each Topsnut-gpw $G_i$ having the Topcode-matrix $T^i_{code}$ admits a graph labeling/coloring $f_i$ with $i\in[1,m]$, then $G$ admits a graph labeling/coloring made by a combinatorial coloring/labeling $f=\uplus |^m_{i=1}f_i$.

By Lemma \ref{thm:one-topcode-matrix-one-graph} and Lemma \ref{thm:grapgicable-no-isomorphic}, we have a result as follows:

\begin{thm}\label{thm:one-topcode-matrix-graph-homomorphism}
For two uncolored graphs $G$ and $H$ with $q=|E(G)|=|E(H)|$, then $G$ is graph homomorphism into $H$, that is, $G\rightarrow H$ if and only if $T_{code}(G)=T_{code}(H)$.
\end{thm}
\begin{proof} \textbf{Necessary.} There is a mapping $f:V(G)\rightarrow V(H)$ such that $f(u_i)f(v_i)\in E(H)$ for each edge $u_iv_i\in E(G)=\{u_1v_1, u_2v_2, \dots ,u_qv_q\}$, so $(u_i, u_iv_i, v_i)^T=A_i\rightarrow B_i=(f(u_i), f(u_i)f(v_i), f(v_i))^T$ for $i\in [1,q]$. Thereby
$$T_{code}(G)=\uplus |^q_{i=1}A_i\rightarrow \uplus |^q_{i=1}B_i=T_{code}(H)$$
Since $G$ and $H$ are not colored and $q=|E(G)|=|E(H)|$, so $T_{code}(G)=T_{code}(H)$.

\textbf{Sufficiency.} Since two uncolored graphs $G$ and $H$ corresponds a common Topcode-matrix $T_{code}$, then we have a mapping $g:V(G)\rightarrow V(H)$ induced by $T_{code}$, such that $g(x)g(y)\in E(H)$ for each edge $xy\in E(G)$, also, $G\rightarrow H$.
\end{proof}

Seven graphs shown in Fig.\ref{fig:homo-vs-topsnut-gpws} corresponds one Topcode-matrix $T_{code}$ shown in Eq.(\ref{eqa:homo-vs-topsnut-gpws}) for understanding Theorem \ref{thm:one-topcode-matrix-graph-homomorphism}, we have six homomorphisms $H_1\rightarrow H_0$ and $H_k\rightarrow H_1$ for $k\in [2,6]$.

\begin{figure}[h]
\centering
\includegraphics[width=14.4cm]{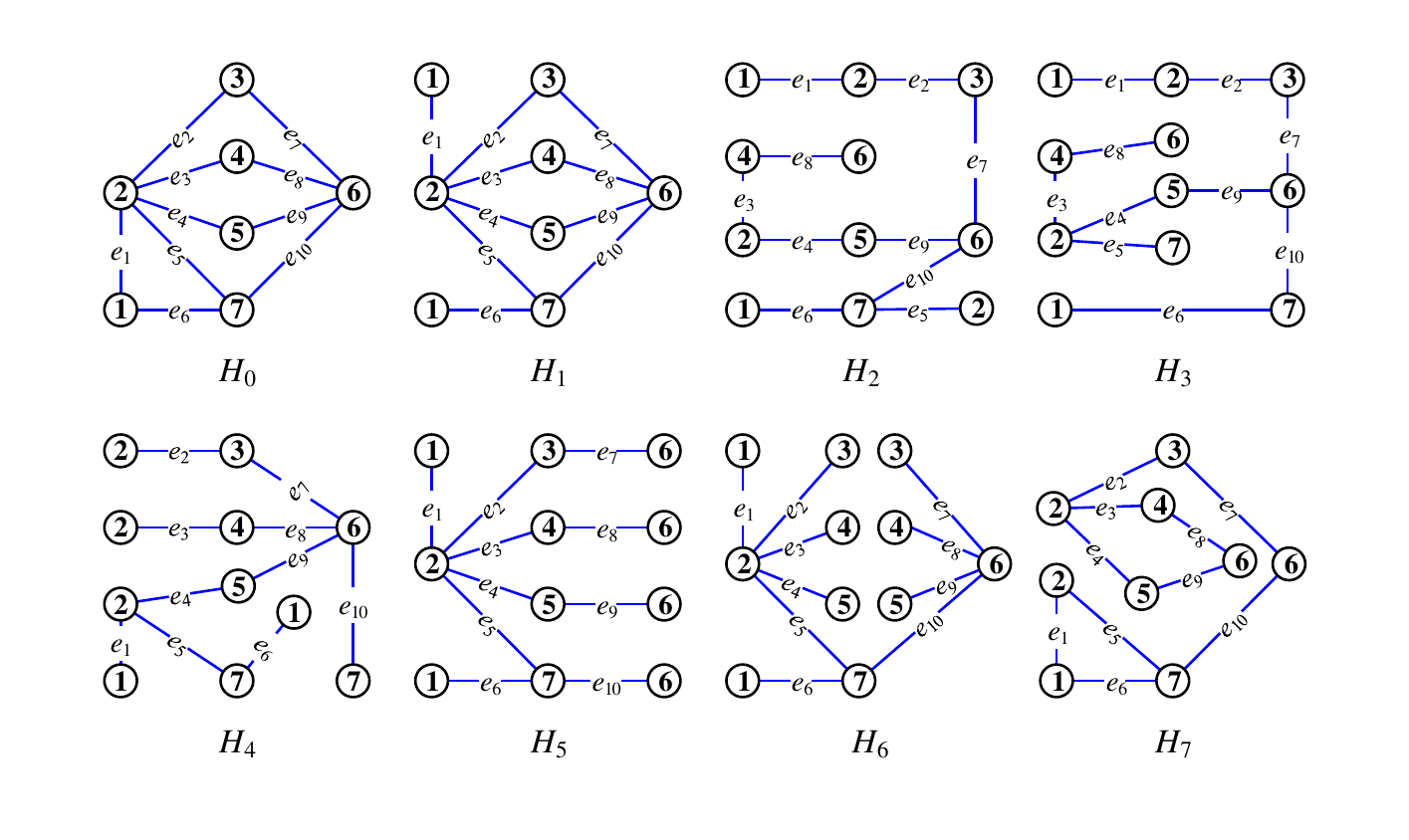}\\
\caption{\label{fig:homo-vs-topsnut-gpws}{\small Examples for understanding Theorem \ref{thm:one-topcode-matrix-graph-homomorphism}.}}
\end{figure}

\begin{defn} \label{defn:TM-degree-sequence}
$^*$ Let $(XY)^*=\{w_1,w_2,\dots ,w_p\}$ be the set of different elements of two vectors $X$ and $Y$ and let $E^*=\{e_1, e_2, \dots , e_q\}$ from the vector $E$ in a Topcode-matrix $T_{code}$ defined in Definition \ref{defn:topcode-matrix-definition}. The number of times each $w_i$ appears in $X$ and $Y$ is denoted as $d_i=\textrm{deg}(w_i)$, we call $\textrm{deg}(T_{code})=(d_1,d_2,\dots ,d_p)$ the \emph{TM-degree-sequence}.

If the TM-degree-sequence $\textrm{deg}(T_{code})$ of the Topcode-matrix $T_{code}$ obeys Erd\"{o}s-Galia Theorem in Lemma \ref{thm:basic-degree-sequence-lemma}, that is, $2q=\sum ^p_{i=1}d_i$ and $d_i\geq d_{i+1}$ for $i\in [1,p-1]$, as well as
$$\sum^k_{i=1}d_i\leq k(k-1)+\sum ^p_{j=k+1}\min\{k,d_j\},~1\leq k\leq p-1$$
then we say that $T_{code}$ is \emph{graphicable}. So, each graph having its own Topcode-matrix to be $T_{code}$ has the degree-sequence $\textrm{deg}(T_{code})=(d_1,d_2,\dots ,d_p)$. \qqed
\end{defn}

\begin{problem}\label{qeu:graph-homomorphism-uncontradicted}
Let $G_{raph}(T_{code})$ be the set of graphs having the common graphicable Topcode-matrix $T_{code}$ with the TM-degree-sequence $\textrm{deg}(T_{code})=(d_1,d_2,\dots ,d_p)$, refer to Definition \ref{defn:TM-degree-sequence}. If a graph $H^*\in G_{raph}(T_{code})$ holds $|V(H^*)|\leq |V(G)|$ for any $G\in G_{raph}(T_{code})$, see an example $H_0$ shown in Fig.\ref{fig:homo-vs-topsnut-gpws}, then we call $H^*$ to be \emph{graph homomorphism uncontradicted}, and ``$G\rightarrow H^*$'' with $|V(G)|\geq 1+|V(H^*)|$ is called a \emph{proper graph homomorphism}. Characterize graph homomorphism uncontradicted graphs $H^*\in G_{raph}(T_{code})$, how many graph homomorphism uncontradicted graphs does the set $G_{raph}(T_{code})$ have?
\end{problem}

\begin{equation}\label{eqa:homo-vs-topsnut-gpws}
\centering
{
\begin{split}
T_{code}= \left(
\begin{array}{ccccccccccc}
1&2&2&2&2&1&6&6&6&6\\
e_1&e_2&e_3&e_4&e_5&e_6&e_7&e_8&e_9&e_{10}\\
2&3&4&5&7&7&3&4&5&7
\end{array}
\right)
\end{split}}
\end{equation}

\begin{rem}\label{rem:333333}
In Problem \ref{qeu:graph-homomorphism-uncontradicted}, the case that each graph of $G_{raph}(T_{code})$ does not admit any coloring differs from the case of each graph of $G_{raph}(T_{code})$ admitting a $W$-type coloring, since the graph homomorphisms on uncolored graphs differ from that on colored graphs.\paralled
\end{rem}

\begin{thm}\label{thm:colored-one-topcode-matrix-graph-homomorphism}
Suppose that two graphs $G$ and $H$ with $q=|E(G)|=|E(H)|$ admit the same $W$-type colorings, then $G$ is graph homomorphism into $H$ if and only if there is a mapping $\varphi$ such that $\varphi(T_{code}(G))=T_{code}(H)$.
\end{thm}

\section{Graph colorings as authentication technique}

The vast majority of graph labelings mentioned here are cited from \cite{Gallian2021} and \cite{Yao-Wang-2106-15254v1}, basic graph colorings can be found in \cite{Bondy-2008}.

\subsection{Graph labelings serving topological authentications}

\begin{defn} \label{defn:basic-W-type-labelings}
\cite{Gallian2021, Yao-Sun-Zhang-Mu-Sun-Wang-Su-Zhang-Yang-Yang-2018arXiv, Bing-Yao-Cheng-Yao-Zhao2009, Zhou-Yao-Chen-Tao2012} Suppose that a connected $(p,q)$-graph $G$ admits a mapping $\theta:V(G)\rightarrow \{0,1,2,\dots \}$. For each edge $xy\in E(G)$, the induced edge color is defined as $\theta(xy)=|\theta(x)-\theta(y)|$. Write vertex color set by $\theta(V(G))=\{\theta(u):u\in V(G)\}$, and edge color set by
$\theta(E(G))=\{\theta(xy):xy\in E(G)\}$. There are the following constraint conditions:

B-1. $|\theta(V(G))|=p$;

B-2. $\theta(V(G))\subseteq [0,q]$, $\min \theta(V(G))=0$;

B-3. $\theta(V(G))\subset [0,2q-1]$, $\min \theta(V(G))=0$;

B-4. $\theta(E(G))=\{\theta(xy):xy\in E(G)\}=[1,q]$;

B-5. $\theta(E(G))=\{\theta(xy):xy\in E(G)\}=[1,2q-1]^o$;

B-6. $G$ is a bipartite graph with the bipartition $(X,Y)$ such that $\max\{\theta(x):x\in X\}< \min\{\theta(y):y\in Y\}$ ($\max \theta(X)<\min \theta(Y)$ for short);

B-7. $G$ is a tree having a perfect matching $M$ such that $\theta(x)+\theta(y)=q$ for each matching edge $xy\in M$; and

B-8. $G$ is a tree having a perfect matching $M$ such that $\theta(x)+\theta(y)=2q-1$ for each matching edge $xy\in M$.

\noindent \textbf{Then}:
\begin{asparaenum}[\textrm{Blab}-1.]
\item A \emph{graceful labeling} $\theta$ satisfies B-1, B-2 and B-4 at the same time.
\item A \emph{set-ordered graceful labeling} $\theta$ holds B-1, B-2, B-4 and B-6 true.
\item A \emph{strongly graceful labeling} $\theta$ holds B-1, B-2, B-4 and
B-7 true.
\item A \emph{set-ordered strongly graceful labeling} $\theta$ holds B-1, B-2, B-4, B-6 and B-7 true.
\item An \emph{odd-graceful labeling} $\theta$ holds B-1, B-3 and B-5 true.
\item A \emph{set-ordered odd-graceful labeling} $\theta$ abides B-1, B-3, B-5 and B-6.
\item A \emph{strongly odd-graceful labeling} $\theta$ holds B-1, B-3, B-5 and B-8, simultaneously.
\item A \emph{set-ordered strongly odd-graceful labeling} $\theta$ holds B-1, B-3, B-5, B-6 and B-8 true.\qqed
\end{asparaenum}
\end{defn}

\begin{figure}[h]
\centering
\includegraphics[width=16.4cm]{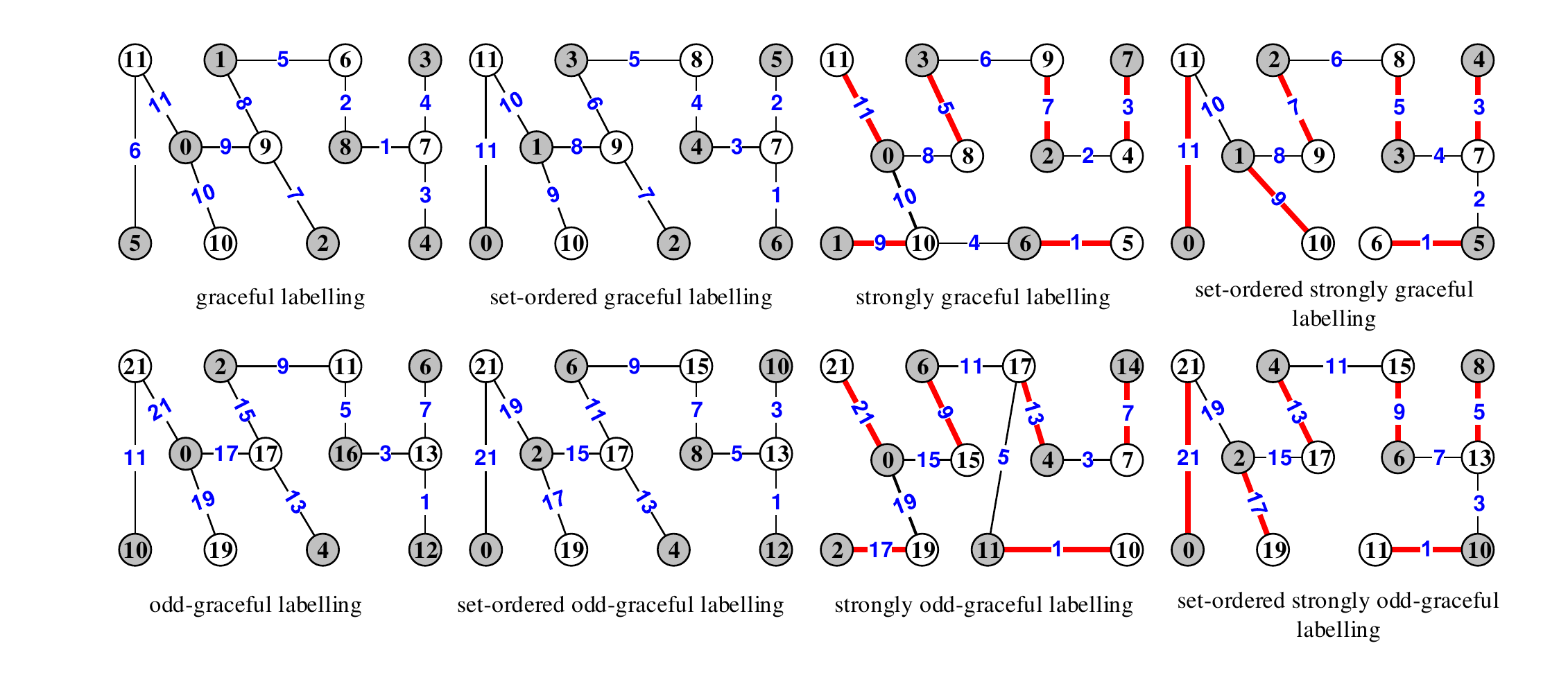}\\
\caption{\label{fig:basic-various-labelings}{\small Basic labelings for understanding Definition \ref{defn:basic-W-type-labelings}, cited from \cite{Yao-Wang-2106-15254v1}.}}
\end{figure}

\begin{defn} \label{defn:11-old-labelings-Gallian}
Let $G$ be a $(p,q)$-graph.

(1) \cite{Gallian2021} An \emph{edge-magic total labeling} $f$ of $G$ holds $f(V(G)\cup E(G))=[1,p+q]$ such that for any edge $uv\in E(G)$, $f(u)+f(uv)+f(v)=c$, where the magic constant $c$ is a fixed positive integer; and furthermore $f$ is \emph{a super edge-magic total labeling} if $f(V(G))=[1,p]$.

(2) \cite{Gallian2021} A \emph{felicitous labeling} $f$ of $G$ holds $f(V(G))\subset [0,q]$, $f(u)\neq f(v)$ for distinct $u,v\in V(G)$, and all edge colors $f(uv)=f(u)+f(v)~(\bmod~q)$ for $uv\in E(G)$ are distinct from each other; and furthermore $f$ is a \emph{super felicitous labeling} if $f(V(G))=[1,p]$.

(3) \cite{Zhou-Yao-Chen2013} An \emph{odd-elegant labeling} $f$ of $G$ holds $f(V(G))\subset [0,2q-1]$, $f(u)\neq f(v)$ for distinct $u,v\in V(G)$, and $f(E(G))=\{f(uv)=f(u)+f(v)~(\bmod~2q):uv\in E(G)\}=[1,2q-1]^o$.

(4) \cite{Gallian2021} A \emph{$k$-graceful labeling} $f$ of $G$ holds $f(V(G))\subset [0,q+k-1]$, $|f(V(G))|=p$ and $f(E(G))=\{f(uv)=|f(u)-f(v)|:uv\in E(G)\}=[k,q+k-1]$ true.\qqed
\end{defn}

\begin{defn}\label{defn:Marumuthu-edge-magic-graceful-labeling}
\cite{Marumuthu-G-2015} If there exists a constant $k\geq 0$, such that a $(p, q)$-graph $G$ admits a total labeling $f:V(G)\cup E(G)\rightarrow [1, p+q]$, each edge $uv\in E(G)$ holds
$|f(u)+f(v)-f(uv)|=k$ and $f(V(G)\cup E(G))=[1, p+q]$ true, we call $f$ an \emph{edge-magic graceful labeling} of $G$, and $k$ a \emph{magic constant}. Moreover, $f$ is called a \emph{super edge-magic graceful labeling} if $f(V(G))=[1, p]$.\qqed
\end{defn}

\begin{defn}\label{defn:edge-magic-total-graceful-labeling}
\cite{Yao-Mu-Sun-Zhang-Wang-Su-2018} An \emph{edge-difference total labeling} $g$ of a $(p,q)$-graph $G$ is defined as: $g: V(G)\cup E(G)\rightarrow [1,p+q]$ such that $g(x)\neq g(y)$ for any two elements $x,y\in V(G)\cup E(G)$, and each edge $uv\in E(G)$ holds $g(uv)+|g(u)-g(v)|=k$ with a constant $k$. Moreover, $g$ is \emph{super} if $\max g(E(G))<\min g(V(G))$ (or $\max g(V(G))<\min g(E(G))$).\qqed
\end{defn}

\begin{defn}\label{defn:pan-edge-magic-graceful-labeling-1}
\cite{Yao-Wang-2106-15254v1} For an edge-magic graceful labeling defined in Definition \ref{defn:Marumuthu-edge-magic-graceful-labeling}, we substitute the condition ``$|f(u)+f(v)-f(uv)|=k$ and $f(V(G)\cup E(G))=[1, p+q]$'' by ``$|f(u)+f(v)-f(uv)|=k$ or $(p+q)-|f(u)+f(v)-f(uv)|=k$ such that $f(V(G)\cup E(G))=[1, p+q]$''. The resulting labeling is called a \emph{pan-edge-magic graceful labeling}.\qqed
\end{defn}

\begin{defn}\label{defn:pan-edge-magic-graceful-labeling}
\cite{Yao-Wang-2106-15254v1} A \emph{pan-edge-difference total labeling} $g$ is obtained by replacing the restriction ``$g(uv)+|g(u)-g(v)|=k$'' by ``$g(uv)+|g(u)-g(v)|=k$ or $g(uv)+(p+q)-|g(u)-g(v)|=k$'' defined in Definition \ref{defn:edge-magic-total-graceful-labeling}.\qqed
\end{defn}

\begin{defn}\label{defn:pan-odd-edge-magic-graceful-labeling}
\cite{Yao-Wang-2106-15254v1} A \emph{pan-odd-graceful labeling} is obtained by replacing the restriction ``$f(uv)=|f(u)-f(v)|$ and $f(E(G))=[1, 2q-1]^o$'' by ``$f(uv)=|f(u)-f(v)|$ or $f(uv)=2q-1-|f(u)-f(v)|$ such that $f(E(G))=[1, 2q-1]^o$'' defined in Definition \ref{defn:basic-W-type-labelings}.\qqed
\end{defn}

\subsection{Multiple color-valued graphs for topological authentications}

\textbf{The multiple color-valued graphic authentication problem (MuCVGAP) \cite{Yao-Wang-Su-Ma-Wang-Sun-ITNEC2020}:} There are $(30)!$ number-based strings generated from the topological coding matrix $T_{code}(K_5)$ shown in Eq.(\ref{eqa:multiple-color-valued-graphic}). A Topsnut-gpw $G$ admitting a proper total coloring $g$ can be split into $G_1$, $G_2$, $\dots$, $G_m$ by the vertex-splitting operation defined in Definition \ref{defn:vertex-split-coinciding-operations}, such that $G_i\not \cong G_j$ for $i\neq j$, so, we get a split Topsnut-gpw set $S_{plt}(\wedge G; g)=\{G_i:~i\in [1,m]\}$. We can make a topological authentication as complex as we were looking for, for instance, we select randomly colored split graphs $G_{i_1},G_{i_2},\dots ,G_{i_n}$ from $S_{plt}(\wedge G; g)$, and want:

(i) Doing the vertex-coinciding operation defined in Definition \ref{defn:vertex-split-coinciding-operations} to each split Topsnut-gpw $G_{i_j}$, such that each resulting Topsnut-gpw is just the original Topsnut-gpw $G$, that is, $G_{i_j}\rightarrow G$.

(ii) This process of one-vs-more multiple authentication has very high complexity, since no polynomial algorithm judges whether graphs are isomorphic to each other up to now, and there is a variety of different kinds of total colorings in topological coding.

\begin{figure}[h]
\centering
\includegraphics[width=15cm]{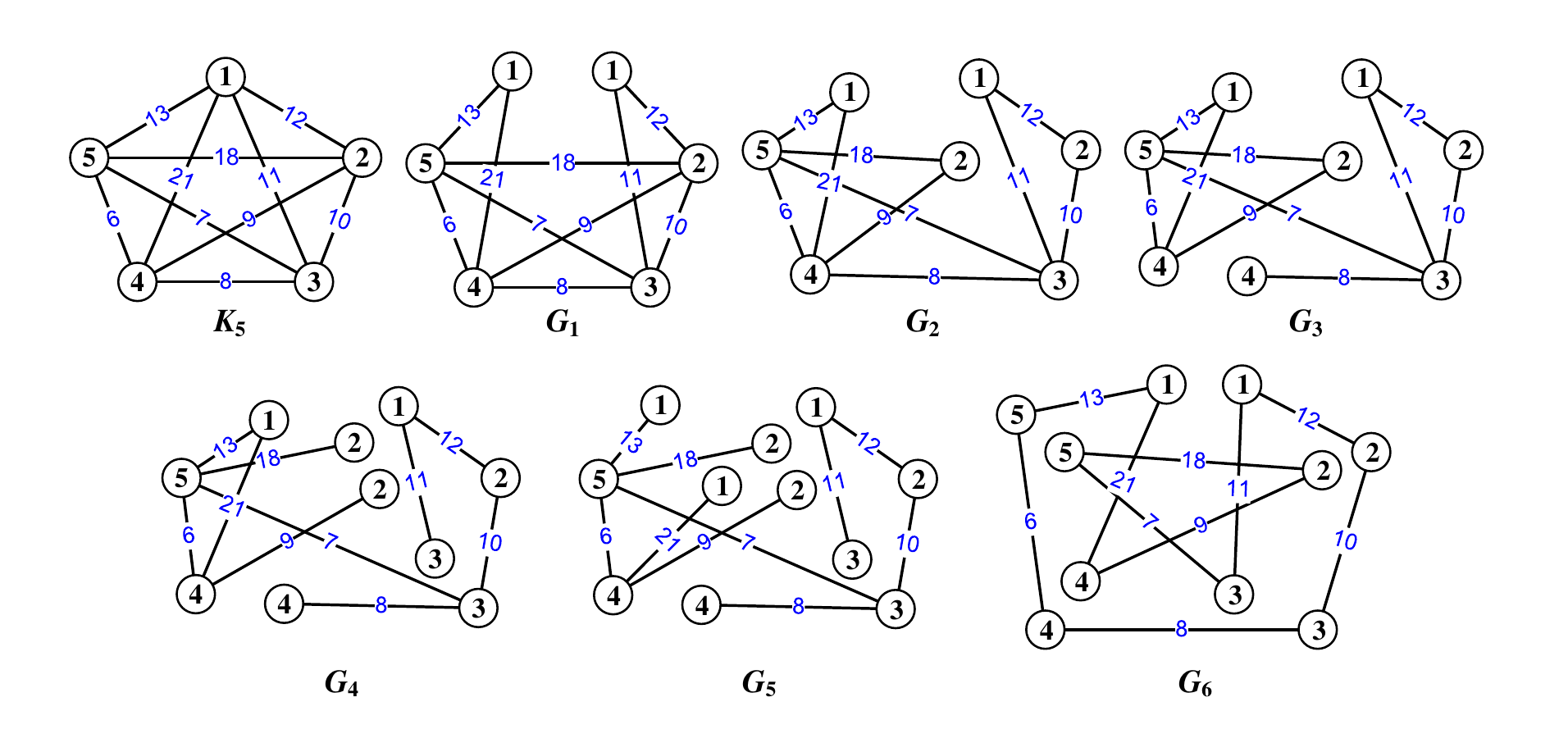}\\
\caption{\label{fig:22-k5-example-matrix} {\small A Topsnut-gpw $K_5$ and its vertex-split graphs $G_j$ for $j\in [1,6]$.}}
\end{figure}

In Fig.\ref{fig:22-k5-example-matrix}, we have graph homomorphisms $G_i\rightarrow K_5$ for $i\in [1,6]$, and $G_i\rightarrow G_{i+1}$ for $i\in [1,4]$, where $G_5$ is a tree, $G_6$ is a cycle of 10 vertices.

\begin{equation}\label{eqa:multiple-color-valued-graphic}
\centering
T_{code}(K_5)= \left(
\begin{array}{cccccccccc}
4 & 3 & 3 & 2 & 1 & 2 & 2 & 1 & 1 & 1\\
6 & 7 & 8 & 18 & 18 & 9 & 10 & 21 & 11 & 12\\
5 & 5 & 4 & 5 & 5 & 4 & 3 & 4 & 3 & 3
\end{array}
\right)
\end{equation}

\begin{problem}\label{qeu:22-MCVGAP}
In the multiple color-valued graphic authentication problem (MuCVGAP), \textbf{determine} vertex-split graph set $S_{plt}(\wedge G; g)$ for each proper total coloring $g$ of a Topsnut-gpw $G$. Is there $|S_{plt}(\wedge G; f)|=|S_{plt}(\wedge G; g)|$ for two different proper total colorings $f$ and $g$ of $G$?
\end{problem}

\begin{example}\label{exa:multiple-topological-authentication}
A \emph{multiple topological authentication} is shown in Fig.\ref{fig:22-multiple-authen}, where a lobster $T$ is as a topological public-key, and other lobsters $T_1,T_2,T_3,T_4,T_5,T_6$ form a group of topological private-keys as follows:

\begin{asparaenum}[(i) ]
\item $T_2$ admits a \emph{pan-edge-magic total labeling} $f_2$ holding $f_2(x_i)+f_2(e_i)+f_2(y_i)=16$ for each edge $e_i=x_iy_i\in E(T_2)$;

\item $T_3$ admits a \emph{pan-edge-magic total labeling} $f_3$ holding $f_3(x_i)+f_3(e_i)+f_3(y_i)=27$ for each edge $e_i=x_iy_i\in E(T_3)$;

\item $T_4$ admits a \emph{felicitous labeling} $f_4$ holding $f_4(e_i)=f_4(x_i)+f_4(y_i)~(\bmod~11)$ for each edge $e_i=x_iy_i\in E(T_4)$;

\item $T_5$ admits an \emph{edge-magic graceful labeling} $f_5$ holding $\big |f_5(x_i)+f_5(y_i)-f_5(e_i)\big |=4$ for each edge $e_i=x_iy_i\in E(T_5)$;

\item $T_6$ admits an \emph{edge-odd-graceful labeling} $f_6$ holding $\{f_6(x_i)+f_6(y_i)+f_6(e_i):e_i=x_iy_i\in E(T_6)\}=[16,26]$.
\end{asparaenum}
\end{example}

\begin{defn} \label{defn:22-one-v-multiple-e-labeling}
$^*$ Suppose that a graph $G$ admits a vertex labeling $f$. If there is a group of edge labelings $f_1,f_2,\dots,f_m$ induced by the vertex labeling $f$ such that each edge labeling $f_k$ holds an equation $e_k(f(u),f_k(uv),f(v))=0$ for each edge $uv\in E(G)$, we call $f$ a \emph{one-v multiple-e labeling} of $G$.\qqed
\end{defn}

\begin{figure}[h]
\centering
\includegraphics[width=16cm]{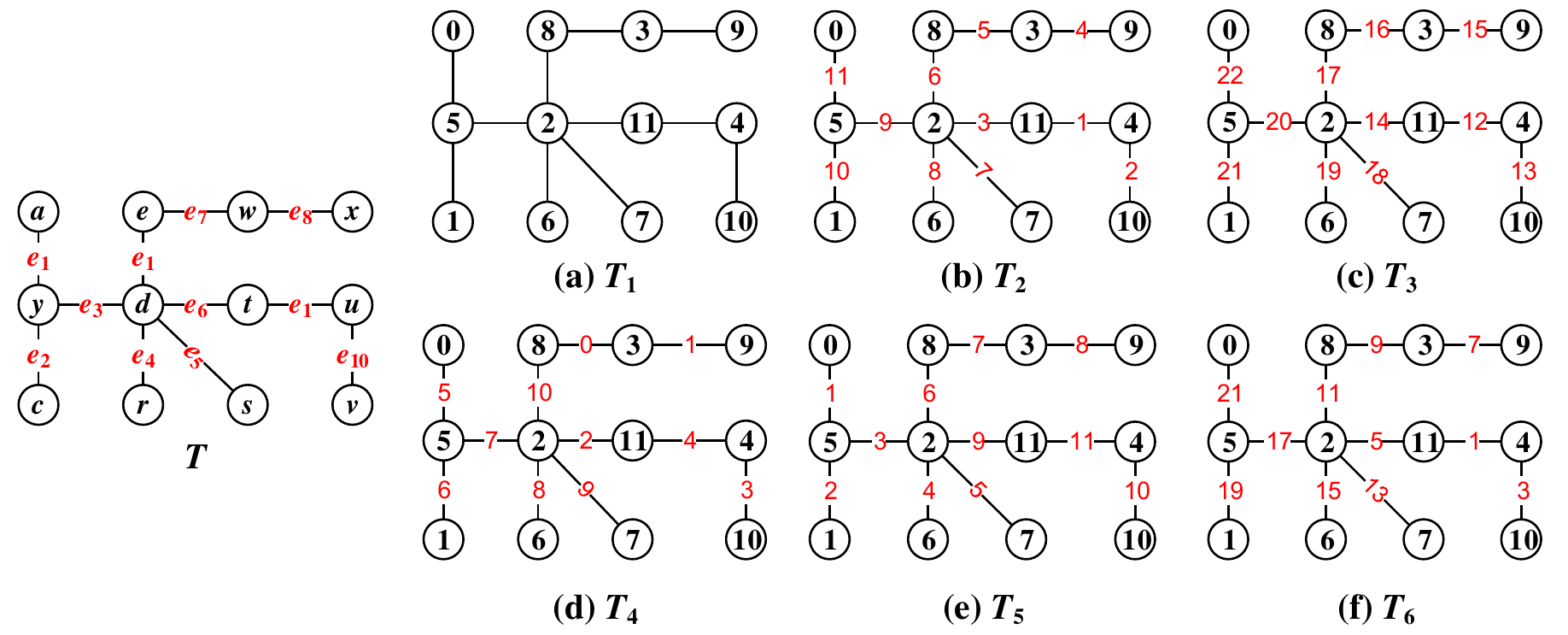}\\
\caption{\label{fig:22-multiple-authen}{\small A lobster $T_1$ admits a vertex labeling $f_1$ shown in (a), and other labelings from (b) to (f) are induced from the labeling $f_1$, cited from \cite{Yao-Wang-2106-15254v1}.}}
\end{figure}

\begin{example}\label{exa:8888888888}
In Fig.\ref{fig:one-vertex-vs-more-edges}, a connected $(10,11)$-graph $G$ admits a graceful labeling $f$ that induces a group of edge labelings $f_1,f_2,f_3,f_4,f_5$, and moreover we can observe $V(G)=X\cup Y$ with $X=\{x_1,x_2,x_3,x_4,x_5\}$ and $Y=\{y_1,y_2,y_3,y_4,y_5\}$, $\{f(x_i):i\in [1,5]\}=\max f(X)<\min f(Y)=\{f(y_i):i\in [1,5]\}$. According to Definition \ref{defn:22-one-v-multiple-e-labeling}, we have:

(1) $G_1$ admits a \emph{set-ordered graceful labeling} $f\cup f_1$, since the constraint $c_1$ is that $f_1(x_iy_j)=|f(x_i)-f(y_j)|=f(y_j)-f(x_i)$ for $x_iy_j\in E(G)$, and $f_1(E(G))=\{f_1(x_iy_j):x_iy_j\in E(G)\}=[1,11]$.

(2) $G_2$ admits a \emph{set-ordered edge-odd-graceful labeling} $f\cup f_2$, since the constraint $c_2$ is that $f_2(x_iy_j)=|2f(x_i)-2f(y_j)+1|=2f(y_j)-2f(x_i)-1$ for $x_iy_j\in E(G)$, and $f_2(E(G))=\{f_2(x_iy_j):x_iy_j\in E(G)\}=[1,21]^o$.

(3) $G_3$ admits an \emph{edge-difference total labeling} $f\cup f_3$, since the constraint $c_3$ is that $f_3(x_iy_j)+|f(y_j)-f(x_i)|=12$ for $x_iy_j\in E(G)$, and $f_3(E(G))=\{f_3(x_iy_j):x_iy_j\in E(G)\}=[1,11]$.

(4) $G_4$ admits a \emph{felicitous-difference total labeling} $f_4$ defined by $f_4(x_i)=5-f(x_i)$ and $f_4(y_j)=f(y_j)$, $f_4(x_iy_j)=f_1(x_iy_j)$ for $x_iy_j\in E(G)$, and $f_4(E(G))=\{f_4(x_iy_j):x_iy_j\in E(G)\}=[1,11]$.

(5) $G_5$ admits an \emph{edge-magic total labeling} $f_5$ defined by $f_5(x_i)=5-f(x_i)$, $f_5(y_j)=f(y_j)$ and $f_5(E(G))=\{f_5(x_iy_j):x_iy_j\in E(G)\}=[1,11]$ such that $f_5(x_i)+f_5(x_iy_j)+f_5(y_j)=17$ for $x_iy_j\in E(G)$.

Thereby, this connected $(10,11)$-graph $G$ admits a one-v multiple-e labeling $f$ inducing a group of edge labelings $f_1,f_2,f_3,f_4,f_5$.
\end{example}

\begin{figure}[h]
\centering
\includegraphics[width=16.4cm]{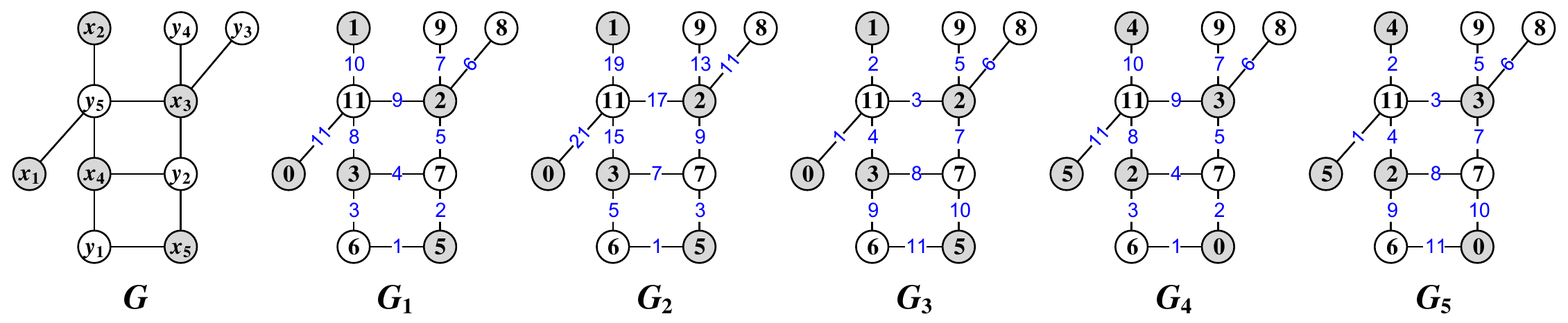}\\
\caption{\label{fig:one-vertex-vs-more-edges}{\small An example for illustrating Definition \ref{defn:22-one-v-multiple-e-labeling}.}}
\end{figure}

\begin{defn}\label{defn:multiple-graph-matching}
\cite{Yao-Sun-Zhang-Mu-Sun-Wang-Su-Zhang-Yang-Yang-2018arXiv} If a $(p,q)$-graph $G$ admits a labeling $f: V(G)\rightarrow [0, p-1]$, such that $G$ can be vertex-split into (spanning) graphs $G_1,G_2,\dots ,G_m$ with $m\geq 2$ and $E(G)=\bigcup^m_{i=1}E(G_i)$ with $E(G_i)\cap E(G_j)=\emptyset $ for $i\neq j$, and each graph $G_i$ admits a $W_i$-type labeling $f_i$ induced by the labeling $f$. We call $G$ a \emph{multiple-graph matching partition}, denoted as $G=\odot_f\langle G_i\rangle ^m_1$.\qqed
\end{defn}

\begin{thm}\label{thm:set-ordered-matchings-10-labelings}
\cite{Yao-Sun-Zhang-Mu-Sun-Wang-Su-Zhang-Yang-Yang-2018arXiv} If a tree $T$ admits a set-ordered graceful labeling $f$, then $T$ matches with a multiple-tree matching partition $\odot_f\langle T_i\rangle ^m_1$ with $m\geq 10$ (see Definition \ref{defn:multiple-graph-matching}).
\end{thm}

\subsection{$W$-type matching labelings}

The twin-type matching labeling has been applied in \cite{Tian-Li-Peng-Yang-2021-102212}.

\begin{defn}\label{defn:22-twin-odd-graceful-labeling}
\cite{Wang-Xu-Yao-2017-Twin} For two connected $(p_i,q)$-graphs $G_i$ with $i=1,2$, and let $p=p_1+p_2-2$, if a $(p,q)$-graph $G=\odot \langle G_1, G_2\rangle$ admits a vertex labeling $f$: $V(G)\rightarrow [0, q]$ such that

(i) $f$ is just an odd-graceful labeling of $G_1$, so $f(E(G_1))=\{f(uv)=|f(u)-f(v)|: uv\in E(G_1)\}=[1, 2q-1]^o$;

(ii) $f(E(G_2))=\{f(uv)=|f(u)-f(v)|: uv\in E(G_2)\}=[1,2q-1]^o$; and

(iii) $|f(V(G_1))\cap f(V(G_2))|=k\geq 0$ and $f(V(G_1))\cup f(V(G_2))\subseteq [0, 2q-1]$.\\
Then $f$ is called a \emph{twin odd-graceful labeling} (Tog-labeling) of $G$.\qqed
\end{defn}

\begin{defn} \label{defn:matching-type-topo-authen}
$^*$ In Definition \ref{defn:topo-authentication-multiple-variables}, if variables $\alpha_1\in P_{ub}(\textbf{X})$ and $\beta_1\in P_{ri}(\textbf{Y})$ form a \emph{topological matching} $\theta_1(\alpha_1,\beta_1)$, and others $\alpha_i\in P_{ub}(\textbf{X})$ and $\beta_i\in P_{ri}(\textbf{Y})$ form a \emph{$W_i$-type matching} $\theta_i(\alpha_i,\beta_i)$ for some $i\in [2,m]$ under the operation vector $\textbf{F}=(\theta_1,\theta_2,\dots $, $\theta_m)$, we call $\textbf{T}_{\textbf{a}}\langle\textbf{X},\textbf{Y}\rangle$ a \emph{matching topological authentication}.\qqed
\end{defn}

\begin{figure}[h]
\centering
\includegraphics[width=14.4cm]{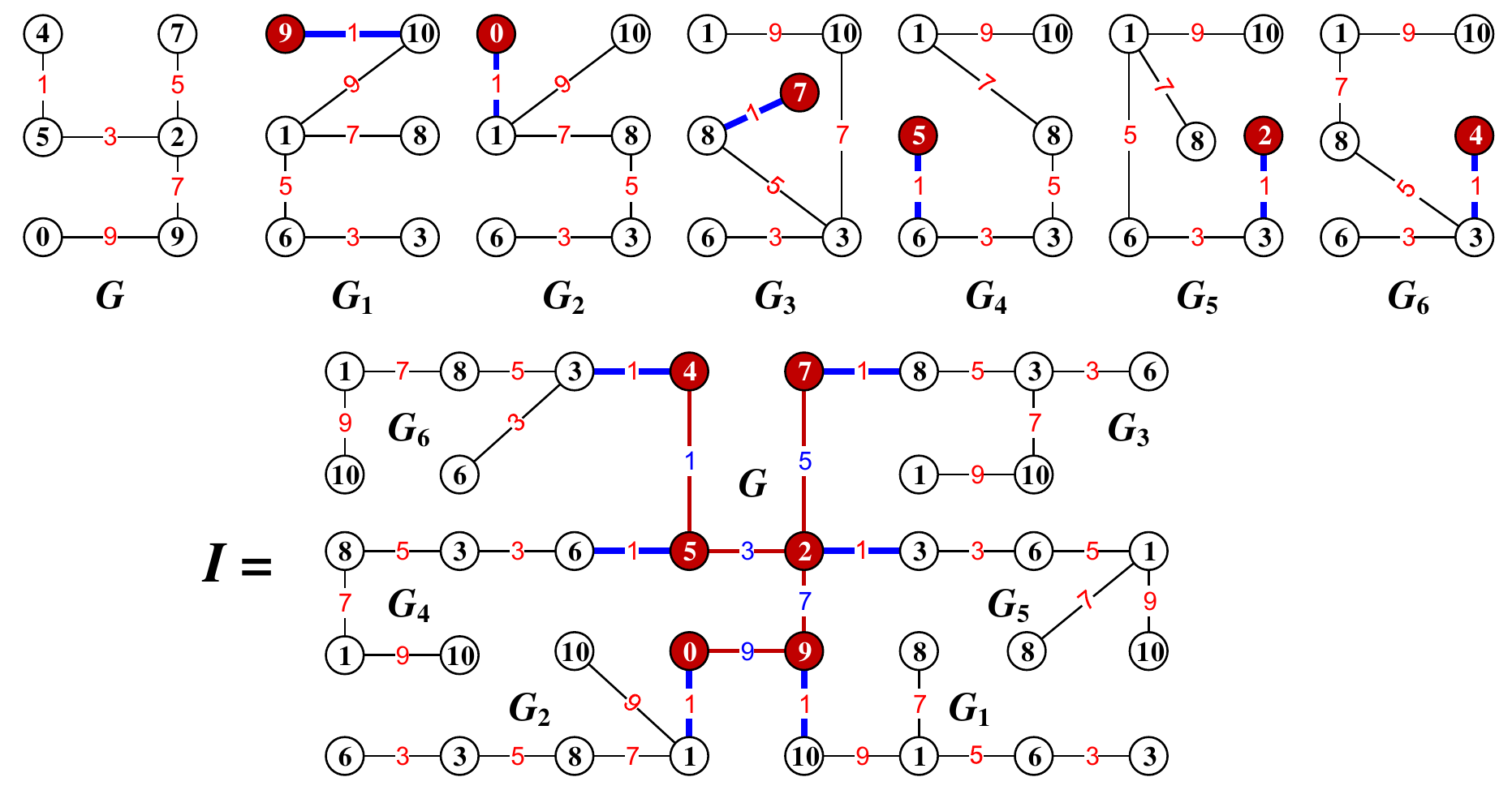}\\
\caption{\label{fig:public-private-compound-one}{\small A matching topological authentication $I$ for illustrating Definition \ref{defn:matching-type-topo-authen}.}}
\end{figure}

In Fig.\ref{fig:public-private-compound-one}, there are six edge-odd-graceful matchings $\odot_1\langle G,G_i\rangle = G\odot_1 G_i$ for $i\in [1,6]$, and moreover we obtain a matching topological authentication $I=G[\ominus^{K_1}_1]^6_{i=1} G_i$.

\begin{defn} \label{defn:e-odd-graceful-v-matching-labeling}
$^*$ If each $(p_i,q_i)$-graph $G_i$ admits a labeling $f_i$ such that $f_i(x)\neq f_i(y)$ for distinct vertices $x,y\in V(G_i)$, and each edge color set
$$
f_i(E(G_i))=[1,2q_i-1]^o=\{f_i(u_jv_j)=|f_i(u_j)-f_i(v_j)|:~u_jv_j\in E(G_i)\}
$$ and $\bigcup ^m_{i=1}f_i(V(G_i))=[0,M]$ with $m\geq 2$, then we say that the \emph{edge-odd-graceful graph base} $\textbf{B}=(G_1,G_2,\dots , G_m)$ admits an \emph{edge-odd-graceful vertex-matching labeling} defined by $F=\uplus ^m_{i=1}f_i$, and $|f_i(V(G_i))|=p_i$ for $i\in [1,m]$ since each $f_i$ is a vertex labeling (refer to Definition \ref{defn:totally-normal-labeling}).\qqed
\end{defn}

\begin{rem}\label{rem:333333}
About Definition \ref{defn:e-odd-graceful-v-matching-labeling}, we have the following particular situations:
\begin{asparaenum}[\textbf{Case}-1.]
\item If $q_1=q_2=\cdots =q_m=q$, we call the labeling $F=\uplus ^m_{i=1}f_i$ a \emph{uniformly edge-odd-graceful vertex-matching labeling} of the edge-odd-graceful graph base $\textbf{B}=(G_1,G_2,\dots , G_m)$.
\item If $m=2$ and $q_1=q_2=q$, the labeling $F=\uplus ^m_{i=1}f_i$, also, is called a \emph{twin odd-graceful labeling} introduced in Definition \ref{defn:22-twin-odd-graceful-labeling}.
\item If $|f_j(V(G_j))\cap f_{j+1}(V(G_{j+1}))|=a_j\geq 1$ for $j\in [1,m-1]$, so we get vertex-coincided graphs $\odot _{a_j}\langle G_j,G_{j+1}\rangle =G_j\odot _{a_j}G_{j+1}$ for $j\in [1,m-1]$, and get a graph
\begin{equation}\label{eqa:555555}
H=G_1\odot _{a_1}G_{2}\odot _{a_2}G_{3}\odot _{a_3}\cdots \odot _{a_{m-2}}G_{m-1}\odot _{a_{m-1}}G_m
\end{equation}
to be as a matching topological authentication.\paralled
\end{asparaenum}
\end{rem}

\begin{figure}[h]
\centering
\includegraphics[width=14.6cm]{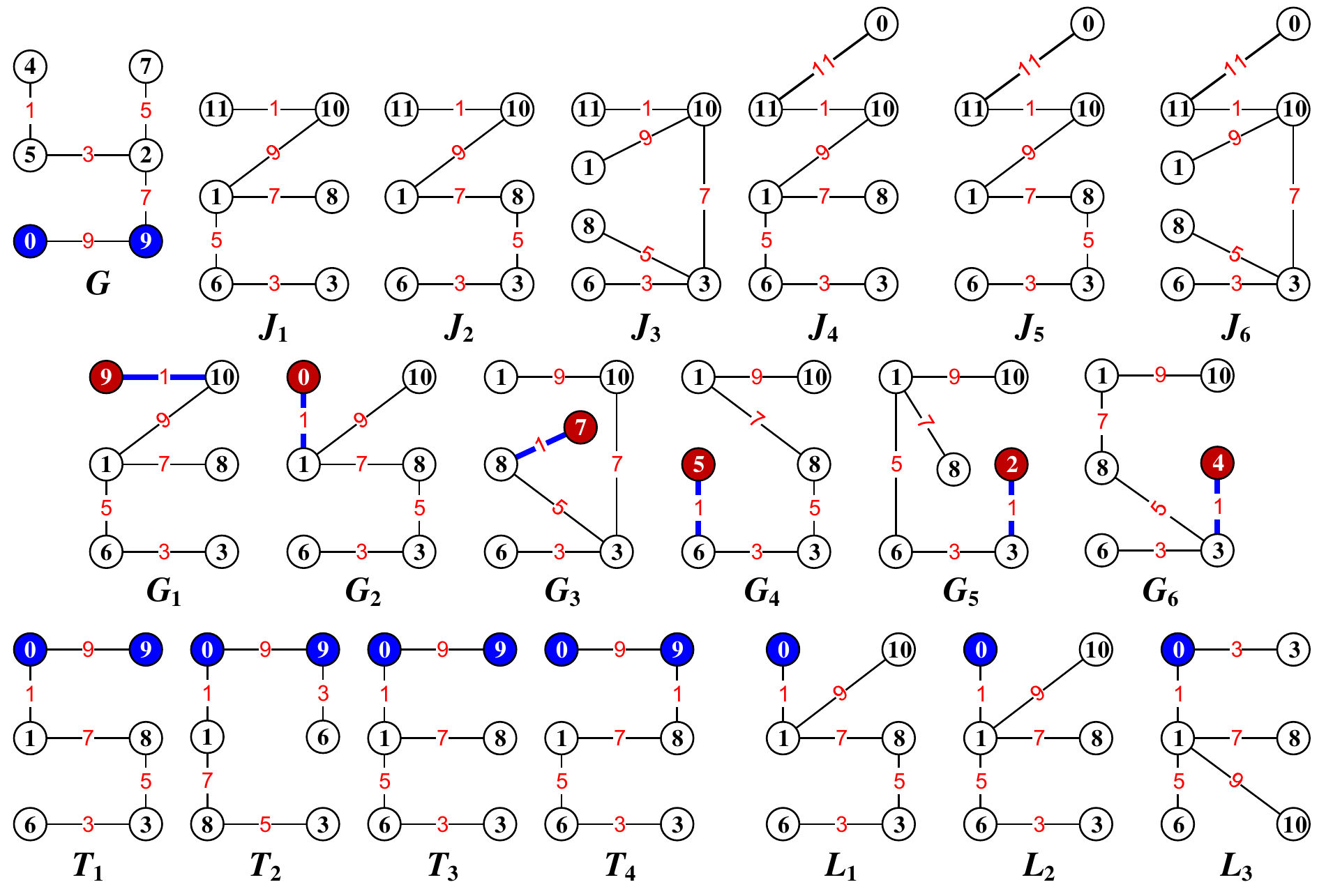}\\
\caption{\label{fig:1-public-vs-more-private}{\small Part of examples for understanding Definition \ref{defn:matching-type-topo-authen} and Definition \ref{defn:e-odd-graceful-v-matching-labeling}, where $G$ is a topological public-key with five edge-odd-graceful graph bases $\{J_1,J_2,J_3\}$, $\{J_4,J_5,J_6\}$, $\{G_1,G_2,G_3,G_4,G_5,G_6\}$, $\{T_1,T_2,T_3,T_4\}$ and $\{L_1,L_2,L_3\}$ .}}
\end{figure}

\begin{problem}\label{qeu:edge-odd-graceful-graph-base}
According to Definition \ref{defn:e-odd-graceful-v-matching-labeling}, for a topological public-key $G$, \textbf{find} all edge-odd-graceful graph bases of $G$.
\end{problem}

\begin{thm}\label{thm:odd-graceful-vertex-matchings}
$^*$ A connected graph $G$ admitting an odd-graceful labeling has at least an edge-odd-graceful graph base $\textbf{B}=(G_1,G_2,\dots , G_m)$ with $m\geq 2$.
\end{thm}

\begin{rem}\label{rem:odd-graceful-vertex-matchings}
A proof idea of Theorem \ref{thm:odd-graceful-vertex-matchings} is shown in Fig.\ref{fig:edge-odd-gracefu-v-matching}, the edge-odd-graceful graph base $\textbf{B}=(H_1,H_2,H_3)$ is obtained by adding one to the vertex colors of $H_{i-1}$ for $i\in [1,3]$, where $H_0=H$ with $H_i\cong H$ for $i=1,2,3$. A \emph{twin odd-graceful labeling matching} $\langle H,T\rangle $ induces three twin odd-graceful labeling matchings $\langle H_i,T_i\rangle $ for $i\in [1,3]$. Moreover, the edge-odd-graceful graph base $\textbf{B}=(H_1,H_2,H_3)$ matches with the edge-odd-graceful graph base $\textbf{A}=(T_1,T_2,T_3)$, also, $\langle \textbf{B},\textbf{A}\rangle $ is a \emph{twin odd-graceful graph base matching}.\paralled
\end{rem}

\begin{figure}[h]
\centering
\includegraphics[width=16.4cm]{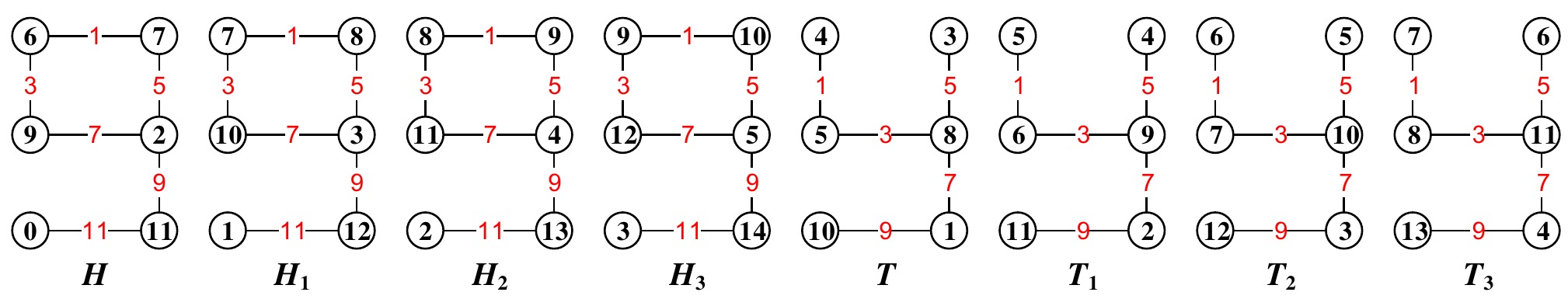}\\
\caption{\label{fig:edge-odd-gracefu-v-matching}{\small A scheme for understanding Theorem \ref{thm:odd-graceful-vertex-matchings}, where the edge-odd-graceful graph base $\textbf{B}=(H_1,H_2,H_3)$ defined in Definition \ref{defn:e-odd-graceful-v-matching-labeling}, and $\langle H,T\rangle $ is a \emph{twin odd-graceful labeling matching}.}}
\end{figure}

\begin{problem}\label{qeu:odd-graceful-vertex-matchings}
According to Theorem \ref{thm:odd-graceful-vertex-matchings} and Problem \ref{rem:odd-graceful-vertex-matchings}, \textbf{find} twin odd-graceful graph base matchings $\langle \textbf{B},\textbf{A}\rangle $ of a connected graph $G$ admitting an odd-graceful labeling.
\end{problem}

\begin{conj}\label{conj:graphs-twin-odd-graceful-labelings}
$^*$ For a bipartite connected graph $G$ admitting an odd-graceful labeling $f$, there exists at least a graph $H$ admitting a labeling $g$, such that $f(V(G))\cup g(V(H))=[0,M]$ with $M\leq 2\max\{|E(G)|,|E(H)|\}$,
$f(E(G))=[1,2|E(G)|-1]^o$ and
$$
g(E(H))=\{g(xy)=|g(x)-g(y)|:xy\in E(H)\}=[1,2|E(H)|-1]^o
$$ Or there exists a group of connected graphs $H_k$ admitting a labeling $g_k$ for $k\in [1,r]$, such that
$$
f(V(G))\cup g_1(V(H_1))\cup g_2(V(H_2))\cup \cdots \cup g_r(V(H_r))=[0,M]
$$ with $M\leq 2\max\{|E(G)|,|E(H_1)|,|E(H_2)|,\dots $, $|E(H_r)|\}$, $g(E(G))=[1,2|E(G)|-1]^o$ and
$$
g_k(E(H_k))=\{g_k(xy)=|g_k(x)-g_k(y)|:xy\in E(H_k)\}=[1,2|E(H_k)|-1]^o,~k\in [1,r]
$$
\end{conj}

\begin{rem}\label{rem:333333}
We can partly answer Conjecture \ref{conj:graphs-twin-odd-graceful-labelings} as follows: If a tree $G$ admits a set-ordered graceful labeling $f$ holding $f(w)=0$ (or $f(w)=|E(G)|-1$) for a leaf $w$ of $G$, then there exists at least a graph $H$ admitting a labeling $g$ holding $f(V(G))\cup g(V(H))=[0,M]$ with
$$
M\leq 2\max\{|E(G)|,|E(H)|\},~f(E(G))=[0,2|E(G)|-1]^o,~g(E(H))=[0,2|E(H)|-1]^o
$$ By the definition of a set-ordered graceful labeling, we have $V(G)=X\cup Y$ with $X\cap Y=\emptyset$, where $X=\{x_1,x_2,\dots ,x_s\}$ and $Y=\{y_1,y_2,\dots ,y_t\}$ with $s+t=p=|V(G)|$. Since $\max X<\min Y$, we have
$$
0=f(x_1)<f(x_2)<\cdots <f(x_s)<f(y_1)<f(y_2)<\cdots <f(y_t)=q=p-1
$$ so $f(x_i)=i-1$ for $i\in[1,s]$, $f(y_j)=s-1+j$ for $j\in[1,t]$.

(i) Doing a new laneling $f^+$ defined by setting $2f(x_i)=2i-2$ for $i\in[1,s]$ and
$$
2f(y_j)-1=2(s-1+j)-1=2s+2j-3,~j\in[1,t]
$$ to $G$ induces two color sets $\{0,2,\dots,2s-2\}$ and $\{2s-1,2s+1,\dots,2p-3\}$, the resultant tree is denoted as $G^+$ admitting the labeling $f^+$ holding $f^+(E(G^+))=[0,2p-3]^o$.

(ii) Doing a new laneling $f^-$ defined by letting $2f(x_i)-1=2i-3$ for $i\in[1,s]$ and
$$
2f(y_j)-2=2(s-1+j)-2=2s+2j-4,~j\in[1,t]
$$ to $G$ induces two color sets $\{-1,1,3,\dots,2s-3\}$ and $\{2s-2,2s,\dots,2p-4\}$, the resultant tree is denoted as $G^-$.

(iii) Delete the vertex $w$ colored with $-1$ from $G^-$, and add a new vertex $u$ to $G^--w$, and color $u$ with $2(p-1)$ and then join it with the vertex $v$ of the tree $G^--w$, where $v$ is colored with 1, write the last tree as $H=G^--w+v$, which admits a labeling $g$ holding $g(E(H))=[0,2p-3]^o$, such that $f^+(V(G))\cup g(V(H))=[0,2(p-1)]$.

There are many interesting researching works on Conjecture \ref{conj:graphs-twin-odd-graceful-labelings}. In Fig.\ref{fig:no-ordered-odd-graceful-matching}, two trees $A_1$ and $B_1$ admit $0$-rotatable odd-graceful labelings, but no set-ordered odd-graceful labelings. \paralled
\end{rem}

\begin{figure}[h]
\centering
\includegraphics[width=16cm]{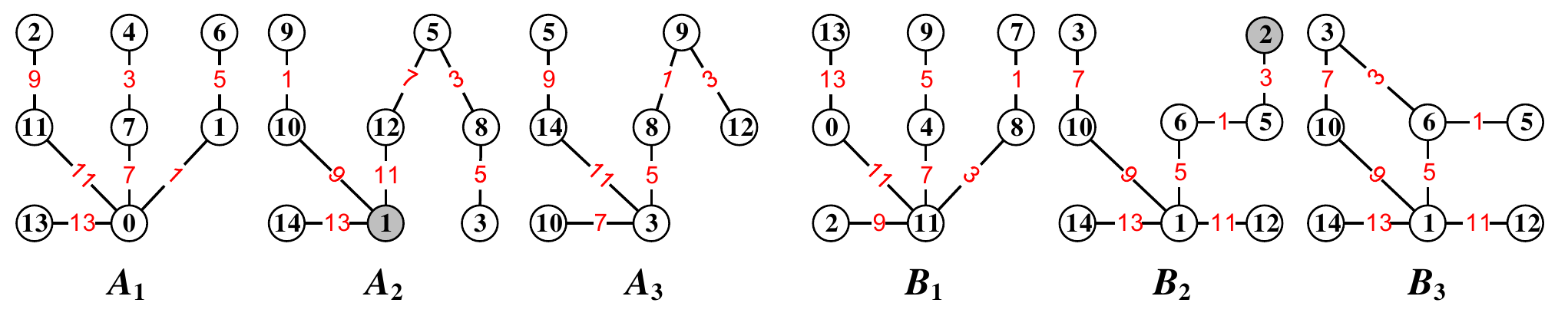}\\
\caption{\label{fig:no-ordered-odd-graceful-matching}{\small $A_1\cong B_1$, they admit $0$-rotatable odd-graceful labelings; $A_1$ has an edge-odd-graceful graph base $A_2,A_3$, and $B_1$ has an edge-odd-graceful graph base $B_2,B_3$.}}
\end{figure}

\subsubsection{Reciprocal-inverse matchings}

\begin{defn}\label{defn:6C-labeling}
\cite{Yao-Sun-Zhang-Mu-Sun-Wang-Su-Zhang-Yang-Yang-2018arXiv} A total labeling $f:V(G)\cup E(G)\rightarrow [1,p+q]$ for a bipartite $(p,q)$-graph $G$ is a bijection and holds:

(i) (e-magic) $f(uv)+|f(u)-f(v)|=k$;

(ii) (ee-difference) each edge $uv$ matches with another edge $xy$ holding one of $f(uv)=|f(x)-f(y)|$ and $f(uv)=2(p+q)-|f(x)-f(y)|$ true;

(iii) (ee-balanced) let $s(uv)=|f(u)-f(v)|-f(uv)$ for $uv\in E(G)$, then there exists a constant $k\,'$ such that each edge $uv$ matches with another edge $u\,'v\,'$ holding one of $s(uv)+s(u\,'v\,')=k\,'$ and $2(p+q)+s(uv)+s(u\,'v\,')=k\,'$ true;

(iv) (EV-ordered) $\min f(V(G))>\max f(E(G))$, or $\max f(V(G))<\min f(E(G))$, or $f(V(G))\subseteq f(E(G))$, or $f(E(G))$ $\subseteq f(V(G))$, or $f(V(G))$ is an odd-set and $f(E(G))$ is an even-set;

(v) (ve-matching) there exists a constant $k\,''$ such that each edge $uv$ matches with one vertex $w$ such that $f(uv)+f(w)=k\,''$, and each vertex $z$ matches with one edge $xy$ such that $f(z)+f(xy)=k\,''$, except the \emph{singularity} $f(x_0)=\lfloor \frac{p+q+1}{2}\rfloor $;

(vi) (set-ordered) $\max f(X)<\min f(Y)$ (resp. $\min f(X)>\max f(Y)$) for the bipartition $(X,Y)$ of $V(G)$.

We refer to $f$ as a \emph{6C-labeling} of $G$.\qqed
\end{defn}

\begin{defn}\label{defn:6C-complementary-matching}
\cite{Yao-Sun-Zhang-Mu-Sun-Wang-Su-Zhang-Yang-Yang-2018arXiv} For a given $(p,q)$-tree $G$ admitting a 6C-labeling $f$ defined in Definition \ref{defn:6C-labeling}, and another $(p,q)$-tree $H$ admits a 6C-labeling $g$, if both graphs $G$ and $H$ hold $f(V(G))\setminus X^*=g(E(H))$, $f(E(G))=g(V(H))\setminus X^*$ and $f(V(G))\cap g(V(H))=X^*=\{z_0\}$ with $z_0=\lfloor \frac{p+q+1}{2}\rfloor $, then both labelings $f$ and $g$ are pairwise \emph{reciprocal-inverse}. The vertex-coincided graph $\odot_1\langle G,H \rangle $ obtained by vertex-coinciding the vertex $x_0$ of $G$ having $f(x_0)=z_0$ with the vertex $w_0$ of $H$ having $g(w_0)=z_0 $ is called a \emph{6C-complementary matching}.\qqed
\end{defn}

\begin{defn}\label{defn:reciprocal-inverse}
\cite{Yao-Sun-Zhang-Mu-Sun-Wang-Su-Zhang-Yang-Yang-2018arXiv} Suppose that a $(p,q)$-graph $G$ admits a $W$-type labeling $f:V(G)\cup E(G)\rightarrow [1,p+q]$, and a $(q,p)$-graph $H$ admits another $W$-type labeling $g:V(H)\cup E(H)\rightarrow [1,p+q]$. If $f(E(G))=g(V(H))\setminus X^*$ and $f(V(G))\setminus X^*=g(E(H))$ for $X^*=f(V(G))\cap g(V(H))\neq \emptyset$, then both $W$-type labelings $f$ and $g$ are \emph{reciprocal-inverse} (resp. \emph{reciprocal complementary}) from each other, and $H$ (resp. $G$) is an \emph{inverse matching} of $G$ (resp. $H$).\qqed
\end{defn}

\begin{example}\label{exa:8888888888}
In Fig.\ref{fig:self-isomorphic-ve-image}, a tree $T$ admits a 6C-labeling $f_T$, and other trees $G_i$ admits a 6C-labeling $f_i$ for $i=[1,3]$, where

(i) $f_T(uv)+|f_T(u)-f_T(v)|=13$ for each edge $uv\in E(T)$;

(ii) $f_1(xy)-|f_1(x)-f_1(y)|=13$ for each edge $xy\in E(G_1)$;

(iii) $f_2(uv)-|f_2(u)-f_2(v)|=13$ for each edge $uv\in E(G_2)$;

(iv) $f_3(xy)+|f_3(x)-f_3(y)|=26$ for each edge $xy\in E(G_3)$.

Moreover, each 6C-labeling $f_i$ is a \emph{reciprocal-inverse labeling} of the 6C-labeling $f_T$, that is, $f_T(V(T))\setminus \{13\}=f_i(E(G_i))$ and $f_T(E(T))=f_i(V(G_i))\setminus \{13\}$ for $i=[1,3]$.

The 6C-labeling $f_T$ and its \emph{reciprocal-inverse labeling} $f_i$ form a \emph{6C-complimentary matching} $\langle f_T,f_i\rangle$ for $i=[1,3]$. Each vertex-coincided tree $H_i=\odot \langle T,G_i\rangle$ admits a \emph{6C-complimentary matching labeling} $g_i=\langle f_T,f_i\rangle$ for $i=[1,3]$, where $\odot \langle T,G_1\rangle$ is a self-isomorphic ve-image, since $T\cong G_1$; and

(v) $g_j(uv)+|g_j(u)-g_j(v)|=13$ or $g_j(uv)-|g_j(u)-g_j(v)|=13$ with $uv\in E(H_j)$ and $j=1,2$;

(vi) $g_3(xy)+|g_3(x)-g_3(y)|=13$ or $g_3(xy)+|g_3(x)-g_3(y)|=26$ with $xy\in E(H_3)$.
\end{example}

\begin{figure}[h]
\centering
\includegraphics[width=16.4cm]{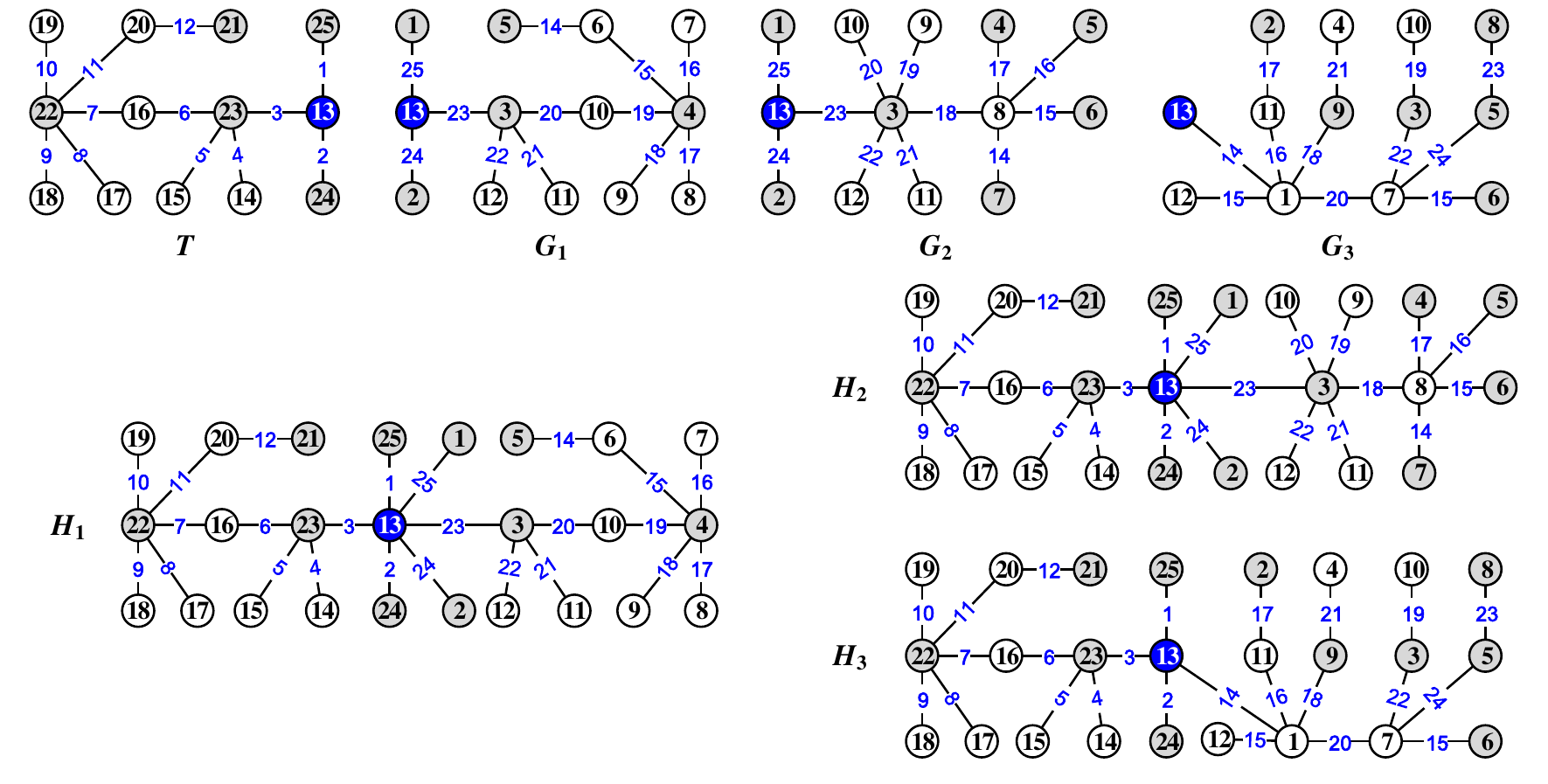}\\
\caption{\label{fig:self-isomorphic-ve-image}{\small Examples for illustrating the 6C-labeling defined in Definition \ref{defn:6C-labeling}, the reciprocal-inverse labeling defined in Definition \ref{defn:6C-complementary-matching} and the 6C-complimentary matching defined in Definition \ref{defn:reciprocal-inverse}.}}
\end{figure}

\begin{lem}\label{thm:lamma-set-ordered-vs-6C-labelling}
\cite{Yao-Sun-Zhang-Mu-Sun-Wang-Su-Zhang-Yang-Yang-2018arXiv} If a tree admits a set-ordered graceful labeling if and only if it admits a 6C-labeling.
\end{lem}

\begin{thm}\label{thm:set-ordered-vs-6C-matching}
\cite{Yao-Zhang-Sun-Mu-Sun-Wang-Wang-Ma-Su-Yang-Yang-Zhang-2018arXiv} If two trees of $p$ vertices admit set-ordered graceful labelings, then they form a 6C-complementary matching.
\end{thm}

\begin{problem}\label{problem:xxxxxx}
\textbf{Find 6C-complementary matchings.} For a given $(p,q)$-tree $G$ admitting a 6C-labeling $f$ defined in Definition \ref{defn:6C-labeling}, \textbf{find} a $(p,q)$-tree $H$ admitting a 6C-labeling $g$ such that the vertex-coincided graph $\odot\langle G,H \rangle $ is a 6C-complementary matching defined in Definition \ref{defn:6C-complementary-matching}.
\end{problem}
\begin{problem}\label{problem:xxxxxx}
\cite{Yao-Sun-Zhang-Mu-Sun-Wang-Su-Zhang-Yang-Yang-2018arXiv} \textbf{Find reciprocal-inverse matchings.} Suppose that a $(p,q)$-graph $G$ and a $(q,p)$-graph $H$ admit two edge-magic graceful labelings $f$ and $g$, respectively, and both $f$ and $g$ are reciprocal inverse from each other as if $f(E(G))=g(V(H))\setminus X^*$ and $f(V(G))\setminus X^*=g(E(H))$ for $X^*=f(V(G))\cap g(V(H))$. \textbf{Find} such pairs of graphs $G$ and $H$, and characterize them, refer to Definition \ref{defn:6C-complementary-matching}.
\end{problem}

\subsubsection{Image-type labelings}

\begin{defn}\label{defn:total-image-dual}
\cite{Yao-Zhang-Sun-Mu-Sun-Wang-Wang-Ma-Su-Yang-Yang-Zhang-2018arXiv} Suppose that a connected graph $G$ admits two $W_i$-type labelings $f_i$ for $i=1,2$. There are constants $M,k$ and $k\,'$ in the following constraint conditions:
\begin{asparaenum}[\textrm{C}-1. ]
\item \label{tt-edge-dual} $f_1(uv)+f_2(uv)=M_{edge}$ for each edge $uv\in E(G)$, where $M_{edge}=\max \{f_1(uv):uv\in E(G)\}+\min \{f_1(uv):uv\in E(G)\}$;
\item \label{tt-vertex-dual} $f_1(x)+f_2(x)=M_{vertex}$ for each vertex $x\in V(G)$, where $M_{vertex}=\max \{f_1(x):x\in V(G)\}+\min \{f_1(x):x\in V(G)\}$;
\item \label{tt-total-dual} $f_1(w)+f_2(w)=M$ for $w\in V(G)\cup E(G)$;
\item \label{tt-vertex-image} $f_1(x)+f_2(x)=k$ for each vertex $x\in V(G)$; and
\item \label{tt-edge-image} $f_1(uv)+f_2(uv)=k\,'$ for each edge $uv\in E(G)$.
\end{asparaenum}
\textbf{Then we have}:
\begin{asparaenum}[(i) ]
\item $f_2$ is an \emph{edge-dual labeling} of $f_{1}$ if C-\ref{tt-edge-dual} holds true.
\item $f_2$ is a \emph{vertex-dual labeling} of $f_{1}$ if C-\ref{tt-vertex-dual} holds true.
\item $f_i$ is a \emph{total dual labeling} of $f_{3-i}$ for $i=1,2$ if C-\ref{tt-total-dual} holds true.
\item $f_i$ is a \emph{vertex-image labeling} of $f_{3-i}$ for $i=1,2$ if C-\ref{tt-vertex-image} holds true.
\item $f_i$ is an \emph{edge-image labeling} of $f_{3-i}$ for $i=1,2$ if C-\ref{tt-edge-image} holds true.
\item $f_i$ is a \emph{totally image labeling} of $f_{3-i}$ for $i=1,2$ if C-\ref{tt-vertex-image} and C-\ref{tt-edge-image} hold true.\qqed
\end{asparaenum}
\end{defn}

\begin{defn}\label{defn:mirror-image-labeling}
\cite{Yao-Zhang-Sun-Mu-Sun-Wang-Wang-Ma-Su-Yang-Yang-Zhang-2018arXiv} Let $f_i:V(G)\rightarrow [a,b]$ be a labeling of a $(p,q)$-graph $G$ and let each edge $uv\in E(G)$ have its own color as $f_i(uv)=|f_i(u)-f_i(v)|$ with $i=1,2$. If each edge $uv\in E(G)$ holds $f_1(uv)+f_2(uv)=k$ true, where $k$ is a positive integer, we call $\langle f_1,f_2\rangle$ a \emph{matching image-labelings}, and $f_i$ a \emph{mirror-image} of $f_{3-i}$ with $i=1,2$.\qqed
\end{defn}

\begin{figure}[h]
\centering
\includegraphics[width=16.4cm]{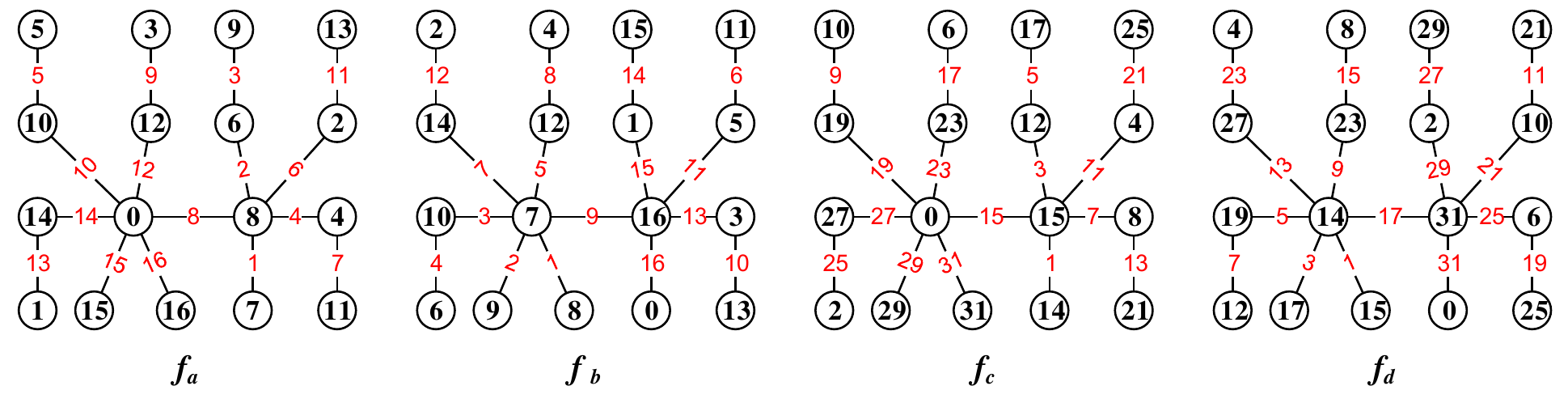}
\caption{\label{fig:image-graceful-aa} {\small A tree $T$ admits four labelings $f_a,f_b,f_c$ and $f_d$ for illustrating Definition \ref{defn:total-image-dual} and Definition \ref{defn:mirror-image-labeling}, where $\langle f_a,f_b\rangle $ is a matching of \emph{set-ordered graceful image-labelings} with $f_a(uv)+h_b(uv)=17$, and $\langle f_c,f_d\rangle $ is a matching of \emph{set-ordered odd-graceful image-labelings} with $f_c(uv)+h_d(uv)=32$, cited from \cite{Yao-Zhang-Sun-Mu-Sun-Wang-Wang-Ma-Su-Yang-Yang-Zhang-2018arXiv}.}}
\end{figure}

\begin{lem}\label{thm:graceful-image-labeling}
\cite{Yao-Zhang-Sun-Mu-Sun-Wang-Wang-Ma-Su-Yang-Yang-Zhang-2018arXiv} If a tree $T$ admits a set-ordered graceful labeling $f$, then $T$ admits another set-ordered graceful labeling $g$ such that $\langle f,g\rangle $ is a matching of image-labelings.
\end{lem}

\begin{defn}\label{defn:sujing-image-labelling7}
\cite{Su-Wang-Yao-Image-labelings-2021-MDPI} Let $f:V(G)\cup E(G)\rightarrow [a,b]$ and $g:V(G)\cup E(G)\rightarrow [a\,',b\,']$ be two labelings of a $(p,q)$-graph $G$, integers $a, b, a\,', b\,'$ satisfy $0\leq a<b$ and $0\leq a\,'<b\,'$.

(1) An equation $f(v)+g(v)=k\,'$ holds true for each vertex $v\in V(G)$, where $k\,'$ is a positive integer and it is called \emph{vertex-image coefficient}, then both labelings $f$ and $g$ are called a matching of \emph{vertex image-labelings} (abbreviated as \emph{v-image-labelings});

(2) An equation $f(uv)+g(uv)=k\,''$ holds true for every edge $uv\in E(G)$, and $k\,''$ is a positive integer, called \emph{edge-image coefficient}, then both labelings $f$ and $g$ are called a matching of \emph{edge image-labelings} (abbreviated as \emph{e-image-labelings}).\qqed
\end{defn}

\begin{figure}[h]
\centering
\includegraphics[width=12cm]{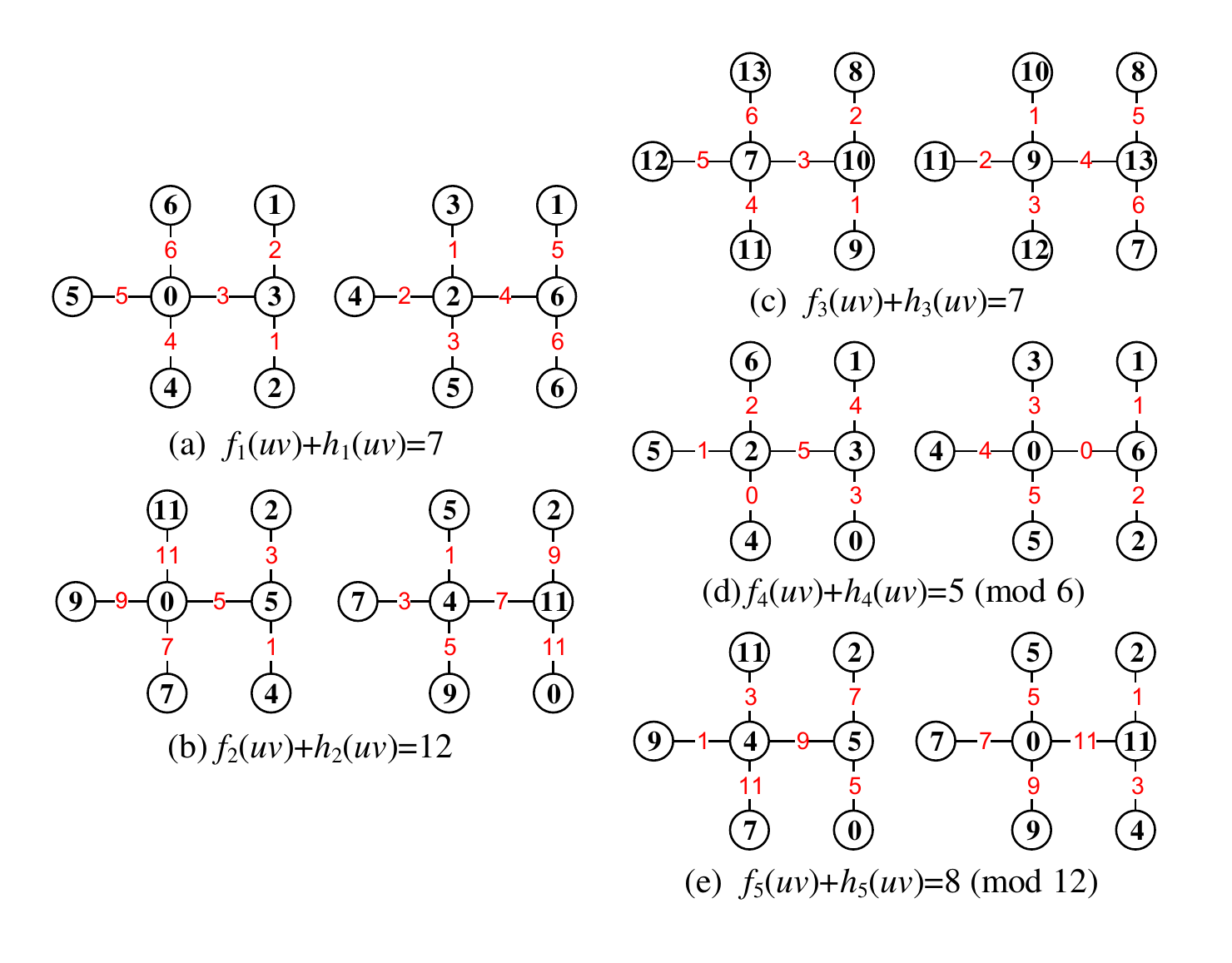}
\caption{\label{fig:image-graceful-11}{\small A tree $T$ admits (Ref. \cite{Gallian2021, Zhou-Yao-Chen-Tao2012, Zhou-Yao-Chen2013}): (a) a matching of set-ordered graceful image-labelings $f_1$ and $h_1$; (b) a matching of set-ordered odd-graceful image-labelings $f_2$ and $h_2$; (c) a matching of edge-magic graceful image-labelings $f_3$ and $h_3$; (d) a matching of set-ordered felicitous image-labelings $f_4$ and $h_4$; (e) a matching of set-ordered odd-elegant image-labelings $f_5$ and $h_5$.}}
\end{figure}

\begin{defn} \label{defn:twin-k-d-harmonious-image-labelings}
\cite{Yao-Zhang-Sun-Mu-Sun-Wang-Wang-Ma-Su-Yang-Yang-Zhang-2018arXiv} A $(p,q)$-graph $G$ admits two $(k,d)$-harmonious labelings $f_i:V(G)\rightarrow X_0\cup S_{q-1,k,d}$ with $i=1,2$, where $X_0=\{0,d,2d, \dots ,(q-1)d\}$ and $S_{q-1,k,d}=\{k,k+d,k+2d,\dots, k+(q-1)d\}$, such that each edge $uv\in E(G)$ is colored with $f_i(uv)-k=[f_i(uv)+f_i(uv)-k~(\textrm{mod}~qd)]$ with $i=1,2$. If $f_1(uv)+f_2(uv)=2k+(q-1)d$, then $\langle f_1, f_2\rangle $ is called a \emph{matching $(k,d)$-harmonious image-labelings} of $G$.\qqed
\end{defn}

\begin{figure}[h]
\centering
\includegraphics[width=15cm]{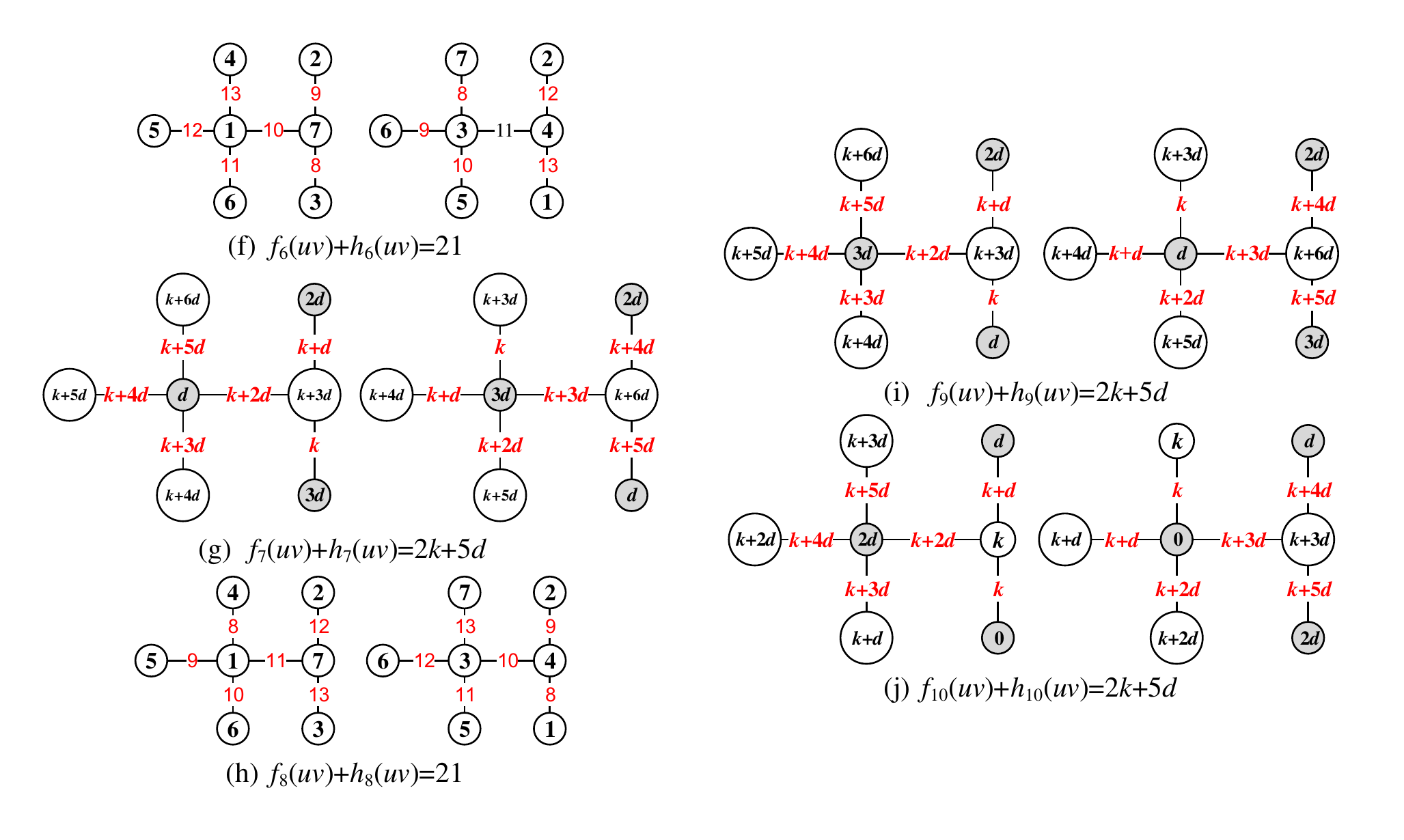}
\caption{\label{fig:image-graceful-22}{\small A tree $T$ admits (Ref. \cite{Gallian2021, Zhou-Yao-Chen-Tao2012, Zhou-Yao-Chen2013}): (f) a matching of super set-ordered edge-magic total image-labelings $f_6$ and $h_6$; (g) a matching of set-ordered $(k,d)$-graceful image-labelings $f_7$ and $h_7$; (h) a matching of super set-ordered edge-antimagic total image-labelings $f_8$ and $h_8$; (i) a matching of $(k,d)$-edge antimagic total image-labelings $f_9$ and $h_9$; (j) a matching of $(k,d)$-arithmetic image-labelings $f_{10}$ and $h_{10}$.}}
\end{figure}

\subsection{Dual-type labelings, Reciprocal-type labelings}

\subsubsection{Dual-type labelings of set-ordered graceful labelings}

Let $G$ be a connected bipartite $(p,q)$-graph admitting a set-ordered graceful labeling $f$, and let $(X,Y)$ be the bipartition of vertex set $V(G)$, where $X=\{x_1,x_2,\dots,x_s\}$ and $Y=\{y_1,y_2,\dots,y_t\}$ with $s+t=p$. Without loss of generality, we have
\begin{equation}\label{equ:set-ordered-graceful-labeling}
0=f(x_1)<f(x_2)<\cdots <f(x_s)<f(y_1)<f(y_2)<\cdots <f(y_t)=q
\end{equation} also, $\max f(X)<\min f(Y)$. See a connected bipartite $(8,9)$-graph $G_0$ admitting a set-ordered graceful labeling shown in Fig.\ref{fig:set-dual-new}(a).

We define the following dual-type labelings:

\textrm{\textbf{Dual-1.}} The \emph{total set-dual labeling} $f_{dual}$ of $f$ is defined as:
$$f_{dual}(w)=\max f(V(G))+\min f(V(G))-f(w)=q-f(w)
$$ for $w\in V(G)$, and the induced edge color of each edge $x_iy_j$ is
\begin{equation}\label{equ:2222222}
f_{dual}(x_iy_j)=|f_{dual}(x_i)-f_{dual}(y_j)|=|f(x_i)-f(y_j)|=f(y_j)-f(x_i)=f(x_iy_j)
\end{equation}
Then $f_{dual}(E(G))=f(E(G))=[1,q]$ and
\begin{equation}\label{equ:2222222}
{
\begin{split}
0=&f_{dual}(y_t)<f_{dual}(y_{t-1})<\cdots <f_{dual}(y_1)<f_{dual}(x_s)\\
&<f_{dual}(x_{s-1})<\cdots <f_{dual}(x_2)<\cdots <f_{dual}(x_1)\\
=&q
\end{split}}
\end{equation} also, the dual labeling $f_{dual}$ is a \emph{set-ordered graceful labeling} of $G$ too.

\begin{problem}\label{qeu:444444}
Suppose that a connected bipartite $(p,q)$-graph $G$ admitting a set-ordered graceful labeling $f$, and $f_{dual}$ is the dual labeling of $f$. Dose $f(V(G))\cup f_{dual}(V(G))=[0,q]$?
\end{problem}

Another dual labeling $f^*_{dual}$ of $f$ is defined as
$$
f^*_{dual}(w)=\max f(V(G))+\min f(V(G))-f(w)=q-f(w)
$$ for $w\in V(G)$, and the induced edge color of each edge $x_iy_j$ is defined by
$$
f^*_{dual}(x_iy_j)=\max f(E(G))+\min f(E(G))-f(x_iy_j)=q+1-f(x_iy_j)
$$ for $x_iy_j\in E(G)$, then $f^*_{dual}(E(G))=f(E(G))=[1,q]$. Because of
$$
f^*_{dual}(x_iy_j)+|f^*_{dual}(y_j)-f^*_{dual}(x_i)|=q+1-f(x_iy_j)+|f(y_j)-f(x_i)|=q+1
$$ so $f^*_{dual}$ is a set-ordered \emph{edge-difference total labeling} of $G$.

\begin{thm}\label{thm:666666}
A connected bipartite graph $G$ admits a set-ordered graceful labeling $f$ if and only if the dual labeling $f_{dual}$ of $f$ is a set-ordered graceful labeling and another dual labeling $f^*_{dual}$ of $f$ is a set-ordered edge-difference total labeling.
\end{thm}

\textrm{\textbf{Dual-2.}} The \emph{$XY$-set-dual labeling} $g_{setXY}$ of $f$ is defined as: $g_{setXY}(x_i)=\max f(X)+\min f(X)-f(x_i)$ for $x_i\in X$, $g_{setXY}(y_j)=\max f(Y)+\min f(Y)-f(y_j)$ for $y_j\in Y$, and the induced edge color of each edge $x_iy_j$ is defined by
\begin{equation}\label{equ:2222222}
{
\begin{split}
g_{setXY}(x_iy_j)&=|g_{setXY}(x_i)-g_{setXY}(y_j)|\\
&=\big |[\max f(X)+\min f(X)-f(x_i)]-[\max f(Y)+\min f(Y)-f(y_j)]\big |\\
&=[\max f(Y)+\min f(Y)]-[\max f(X)+\min f(X)]-f(x_iy_j)\\
&=q+\min f(Y)-\max f(X)-f(x_iy_j)
\end{split}}
\end{equation} and
\begin{equation}\label{equ:2222222}
{
\begin{split}
g_{setXY}(E(G))&=[q+\min f(Y)-\max f(X)-q,q+\min f(Y)-\max f(X)-1]\\
&=[\min f(Y)-\max f(X),\min f(Y)-\max f(X)+(q-1)]
\end{split}}
\end{equation} then the set-dual labeling $g_{setXY}$, when as $\min f(Y)-\max f(X)=1$, is a \emph{set-ordered graceful labeling} of $G$.

\begin{thm}\label{thm:666666}
A connected bipartite graph $G$ admits a set-ordered graceful labeling $f$ if and only if the set-dual labeling $g_{setXY}$ of $f$ is a set-ordered graceful labeling.
\end{thm}

And another set-dual labeling $g^*_{setXY}$ is defined as $g^*_{setXY}(w)=g_{setXY}(w)$ for $w\in V(G)$, and
$$
g^*_{setXY}(x_iy_j)=\max f(E(G))+\min f(E(G))-f(x_iy_j)=q+1-f(x_iy_j)
$$ for each edge $x_iy_j\in E(G)$, which induces edge color set $g^*_{setXY}(E(G))=f(E(G))=[1,q]$. Since $g^*_{setXY}(u)\neq g^*_{setXY}(v)$ for distinct vertices $u,v\in V(G)$, and
$$\label{eqa:555555}
{
\begin{split}
&\quad \big | |g^*_{setXY}(y_j)-g^*_{setXY}(x_i)|-g^*_{setXY}(x_iy_j)\big |\\
&=\big | [q+\min f(Y)-\max f(X)-f(x_iy_j)]-[q+1-f(x_iy_j)]\big |\\
&=\min f(Y)-\max f(X)-1
\end{split}}
$$ for each edge $x_iy_j\in E(G)$, so $g^*_{setXY}$ is a set-ordered \emph{graceful-difference total labeling} of $G$.

Here, $g^*_{setXY}$ has its own dual labeling $\alpha _{set}$ defined by
$$
\alpha _{set}(w)=\max g^*_{setXY}(V(G))+\min g^*_{setXY}(V(G))-\max g^*_{setXY}(w)
$$ for $w\in V(G)$, and the edge color of each edge $x_iy_j$ is $\alpha _{set}(x_iy_j)=\max g^*_{setXY}(x_iy_j)$, so it is not hard to show that $\alpha _{set}$ is a \emph{set-ordered graceful labeling} of $G$, see an example $G_5$ shown in Fig.\ref{fig:set-dual-new-11}(a) admitting a set-ordered graceful labeling $\alpha _{set}$ to be the dual labeling of $g^*_{setXY}$ admitted by $G_4$ shown in Fig.\ref{fig:set-dual-new}(e).

\vskip 0.4cm

See Fig.\ref{fig:set-dual-new} for understanding the labelings introduced in Dual-1 and Dual-2.

\begin{figure}[h]
\centering
\includegraphics[width=16.4cm]{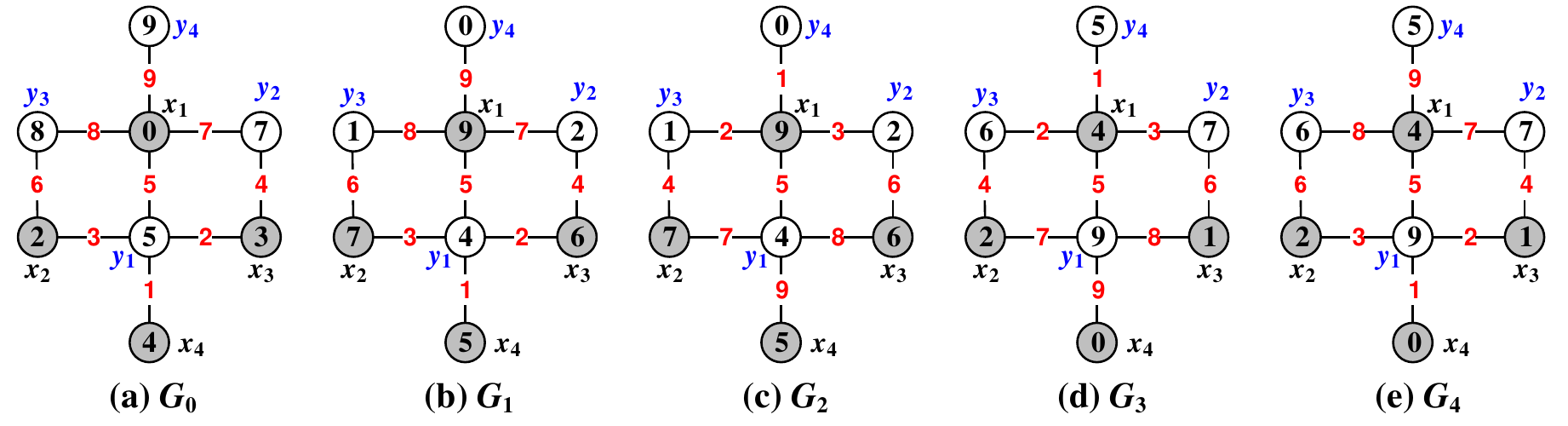}\\
\caption{\label{fig:set-dual-new} {\small (a) $G_0$ admits a set-ordered graceful labeling $f$; (b) $G_1$ admits a set-ordered graceful labeling $f_{dual}$, which is the total set-dual labeling of $f$; (c) $G_2$ admits a set-ordered edge-difference total labeling $f^*_{dual}$; (d) $G_3$ admits a $XY$-set-dual labeling $g_{setXY}$ of $f$; (e) $G_4$ admits a set-ordered edge-difference total labeling $g^*_{setXY}$.}}
\end{figure}

\textrm{\textbf{Dual-3.}} The \emph{$X$-set-dual labeling} $h_{setX}$ of $f$ is defined as: $h_{setX}(x_i)=\max f(X)+\min f(X)-f(x_i)=\max f(X)-f(x_i)$ for $x_i\in X$, $h_{setX}(y_j)=f(y_j)$ for $y_j\in Y$, and the edge color of each edge $x_iy_j$ is $h_{setX}(x_iy_j)=f(x_iy_j)$ for $x_iy_j\in E(G)$, so $h_{setX}(E(G))=f(E(G))=[1,q]$. Furthermore, we have
$${
\begin{split}
h_{setX}(x_i)+h_{setX}(y_j)-h_{setX}(x_iy_j)&=\max f(X)+\min f(X)-f(x_i)+f(y_j)-f(x_iy_j)\\
&=\max f(X)
\end{split}}
$$ so $h_{setX}$ is a set-ordered \emph{felicitous-difference total labeling} of $G$.

Moreover, we define another $X$-set-dual labeling $h^*_{setX}$ by $h^*_{setX}(w)=h_{setX}(w)$ for $w\in V(G)$, and
$$
h^*_{setX}(x_iy_j)=\max f(E(G))+\min f(E(G))-f(x_iy_j)=q+1-f(x_iy_j)
$$ for each edge $x_iy_j\in E(G)$, then $h^*_{setX}(E(G))=f(E(G))=[1,q]$, we claim that $h^*_{setX}$ is an edge-magic total labeling, since
\begin{equation}\label{eqa:555555}
{
\begin{split}
h^*_{setX}(x_i)+h^*_{setX}(x_iy_j)+h^*_{setX}(y_j)&=h_{setX}(x_i)+q+1-f(x_iy_j)+h_{setX}(y_j)\\
&=\max f(X)+h_{setX}(x_iy_j)+q+1-f(x_iy_j)\\
&=\max f(X)+f(x_iy_j)+q+1-f(x_iy_j)\\
&=q+1+\max f(X)
\end{split}}
\end{equation} which shows that $h^*_{setX}$ is a set-ordered \emph{edge-magic total labeling} of $G$.

\textrm{\textbf{Dual-4.}} The \emph{$Y$-set-dual labeling} $h_{setY}$ of $f$ is defined as: $h_{setY}(x_i)=f(x_i)$ for $x_i\in X$, $h_{setY}(y_j)=\max f(Y)+\min f(Y)-f(y_j)=q+\min f(Y)-f(y_j)$ for $y_j\in Y$, and the edge color of each edge $x_iy_j$ is $h_{setY}(x_iy_j)=f(x_iy_j)$ for $x_iy_j\in E(G)$, immediately, $h_{setY}(E(G))=f(E(G))=[1,q]$. Moreover, we confirm that $h_{setY}$ is an \emph{edge-magic total labeling} of $G$, since
\begin{equation}\label{eqa:555555}
{
\begin{split}
h_{setY}(x_i)+h_{setY}(x_iy_j)+h_{setY}(y_j)&=f(x_i)+f(x_iy_j)+q+\min f(Y)-f(y_j)\\
&=q+\min f(Y)
\end{split}}
\end{equation} for each edge $x_iy_j\in E(G)$.

And another case, we define another $Y$-set-dual labeling $h^*_{setY}$ by $h^*_{setY}(w)=h_{setY}(w)$ for $w\in V(G)$, and
$$
h^*_{setY}(x_iy_j)=\max f(E(G))+\min f(E(G))-f(x_iy_j)=q+1-f(x_iy_j)
$$ for $x_iy_j\in E(G)$, then $h^*_{setY}(E(G))=f(E(G))=[1,q]$. We omit the proof for $h^*_{setY}$ being a set-ordered \emph{felicitous-difference total labeling} of $G$.

See Fig.\ref{fig:set-dual-new-11} for understanding the labelings introduced in Dual-3 and Dual-4.
\begin{figure}[h]
\centering
\includegraphics[width=16.4cm]{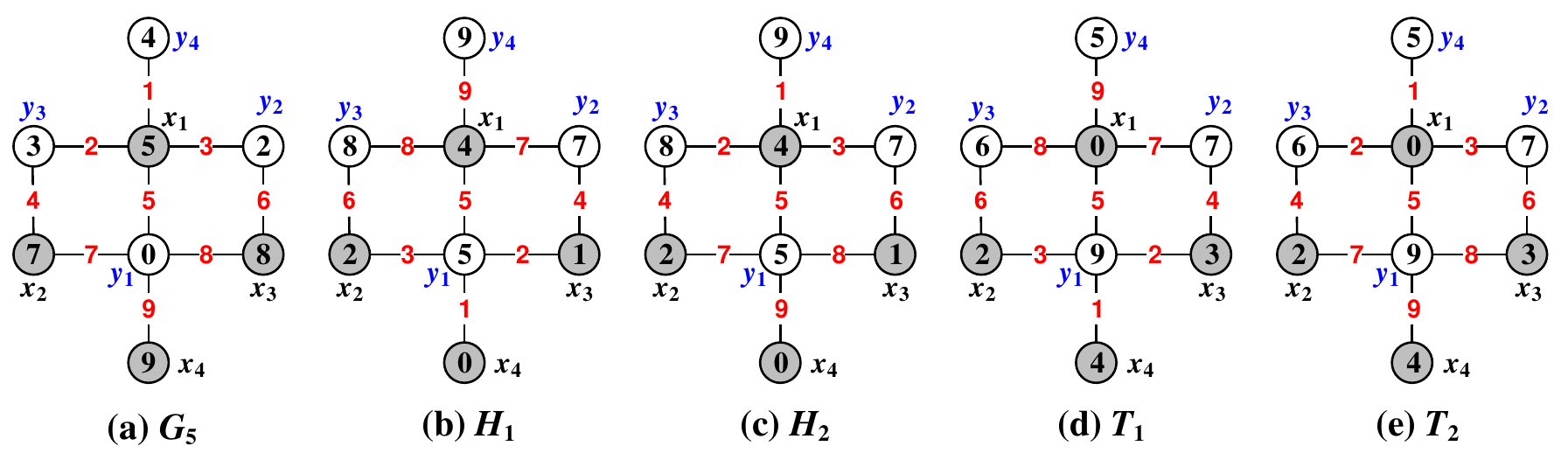}\\
\caption{\label{fig:set-dual-new-11} {\small (a) $G_5$ admits a set-ordered graceful labeling; (b) $H_1$ admits a set-ordered felicitous-difference total labeling $h_{setX}$, which is a $X$-set-dual labeling of the set-ordered graceful labeling $f$ of $G_0$ shown in Fig.\ref{fig:set-dual-new}(a); (c) $H_2$ admits a set-ordered edge-magic total labeling $h^*_{setX}$; (d) $T_1$ admits a set-ordered edge-magic total labeling $h_{setY}$, which is a $Y$-set-dual labeling of the set-ordered graceful labeling $f$ of $G_0$ shown in Fig.\ref{fig:set-dual-new}(a); (e) $T_2$ admits a set-ordered felicitous-difference total labeling $h^*_{setY}$, also, a $Y$-set-dual labeling of the set-ordered graceful labeling $f$ of $G_0$ shown in Fig.\ref{fig:set-dual-new}(a).}}
\end{figure}

\vskip 0.4cm

The above dual-type labelings from \textbf{Dual-1} to \textbf{Dual-4} produce the following matchings:

\begin{asparaenum}[\textrm{\textbf{Matching}}-1. ]
\item The set-ordered graceful matching $\langle f, f^*_{dual}\rangle $ holds $f(w)+f^*_{dual}(w)=q$ for $w\in V(G)$ and $f(x_iy_j)+f^*_{dual}(x_iy_j)=q+1$ for $x_iy_j\in E(G)$.

\item $\langle g_{setXY}, g^*_{setXY}\rangle $ is a \emph{matching of a set-ordered graceful labeling and a graceful-difference total labeling}.

\item $\langle h_{setX}, h^*_{setY}\rangle $ is a \emph{matching of two set-ordered edge-magic total labelings}.

\item $\langle h^*_{setX}, h_{setY}\rangle $ is a \emph{matching of two set-ordered felicitous-difference total labelings}.
\end{asparaenum}

\subsubsection{Reciprocal-type labelings}

A $W$-type labeling $f$ of a connected bipartite $(p,q)$-graph $G$ having its own bipartition $(S_X,S_Y)$ with $S_X=\{x_1,x_2,\dots,x_s\}$ and $S_Y=\{x_{s+1},x_{s+2},\dots,x_{s+t}\}$ with $s+t=p$, holds $\max f(S_X)<\min f(S_Y)$ and $0=f(x_1)<f(x_2)<\cdots <f(x_{s+t})=q$ true.

There are the following reciprocal-type labelings (see Fig.\ref{fig:reciprocal-labeling-11} and Fig.\ref{fig:reciprocal-labeling-22}):

\begin{figure}[h]
\centering
\includegraphics[width=16.4cm]{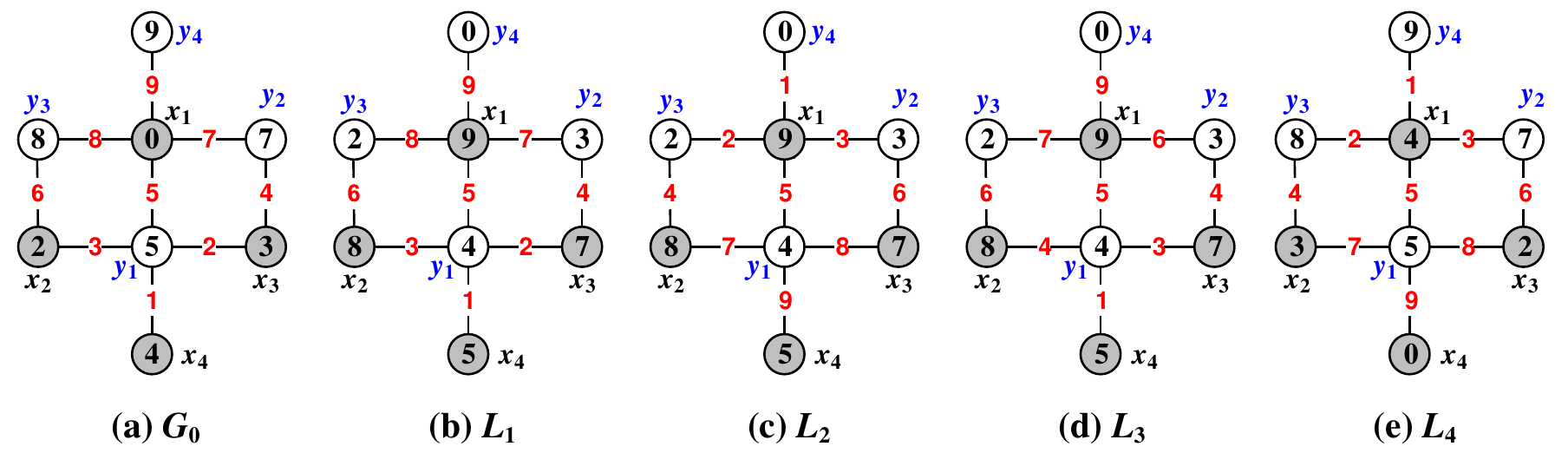}\\
\caption{\label{fig:reciprocal-labeling-11} {\small A scheme for reciprocal-type labelings.}}
\end{figure}

\begin{figure}[h]
\centering
\includegraphics[width=16.4cm]{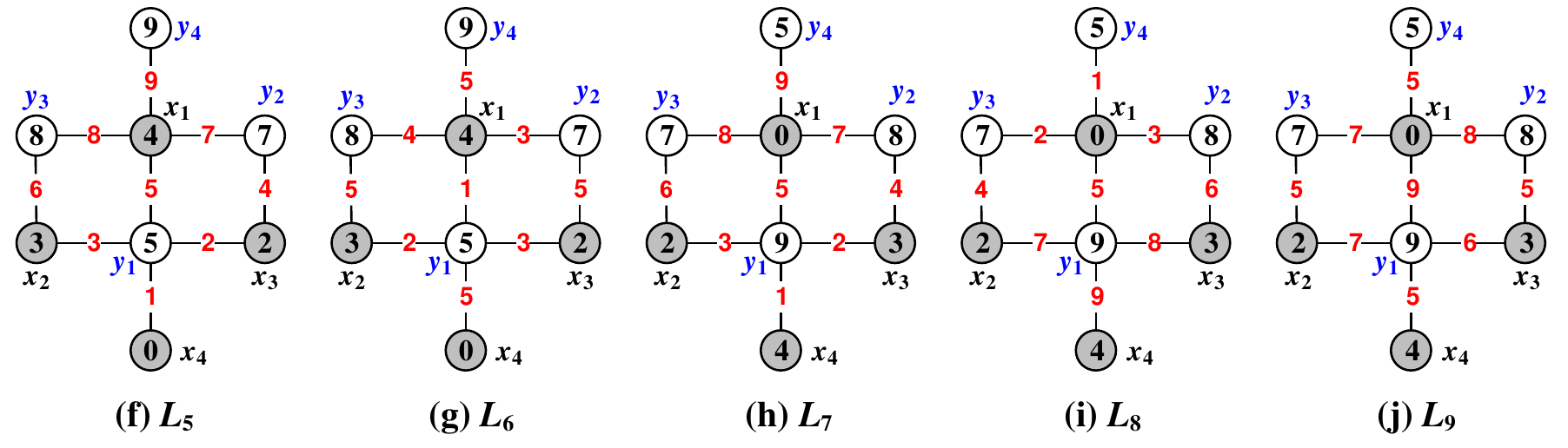}\\
\caption{\label{fig:reciprocal-labeling-22} {\small Another scheme for reciprocal-type labelings.}}
\end{figure}

\textbf{Reciprocal-1.} The $X$-reciprocal labeling $g_{recipX}$ of the $W$-type labeling $f$ is defined by $g_{recipX}(x_i)=f(x_{s+1-i})$ for $i\in [1,s]$, and $g_{recipX}(x_j)=f(x_j)$ for $j\in [s+1,s+t]$, the edge color of each edge $x_iy_j\in E(G)$ is defined as one of the following cases

(i) $g_{recipX}(x_iy_j)=f(x_iy_j)$;

(ii) $g_{recipX}(x_iy_j)=\max f(E(G))+\min f(E(G))-f(x_iy_j)=q+1-f(x_iy_j)$; and

(iii) $g_{recipX}(x_iy_j)=|g_{recipX}(x_i)-g_{recipX}(y_j)|$.

\textbf{Reciprocal-2.} The $Y$-reciprocal labeling $g_{recipY}$ of the $W$-type labeling $f$ is defined as: $g_{recipY}(x_i)=f(x_{x_i})$ for $i\in [1,s]$, and $g_{recipY}(x_j)=f(x_{2s+t+1-j})$ for $j\in [s+1,s+t]$, the edge color of each edge $x_iy_j\in E(G)$ is defined as one of the following cases

(1) $g_{recipY}(x_iy_j)=f(x_iy_j)$;

(2) $g_{recipY}(x_iy_j)=\max f(E(G))+\min f(E(G))-f(x_iy_j)=q+1-f(x_iy_j)$; and

(3) $g_{recipY}(x_iy_j)=|g_{recipY}(x_i)-g_{recipY}(y_j)|$.

\textbf{Reciprocal-3.} The total reciprocal labeling $h_{recip}$ of the $W$-type labeling $f$ is defined as one of $h_{recip}(x_r)=f(x_{s+t+1-r})$ for $r\in [1,s+t]$, and the edge color of each edge $x_iy_j\in E(G)$ is defined as: (a) $h_{recip}(x_iy_j)=f(x_iy_j)$; (b) $h_{recip}(x_iy_j)=\max f(E(G))+\min f(E(G))-f(x_iy_j)=q+1-f(x_iy_j)$; and (c) $h_{recip}(x_iy_j)=|h_{recip}(x_i)-h_{recip}(y_j)|$.

\begin{rem}\label{rem:333333}
It is noticeable, the reciprocal-type labelings of a tree are equal to its own dual-type labelings. The reciprocal-type labelings of a $W$-type labeling $f$ of a connected bipartite $(p,q)$-graph $G$ induce some Topcode-matrices differ from the Topcode-matrix $T_{code}(G)$, however, these Topcode-matrices induce number-based strings that are equal to that induced by $T_{code}(G)$.\paralled
\end{rem}

\section{Topological homomorphisms}

In this section, the topological authentication will be strengthened by technique of \emph{topological homomorphisms}, such as graph homomorphism, graph anti-homomorphism, Topcode-matrix homomorphism, and so on. Homomorphism technique is useful and important in privacy protection and cloud computation nowadays.

\subsection{Traditional graph homomorphisms and anti-homomorphisms}

\subsubsection{Uncolored graph homomorphisms}

\begin{defn} \label{defn:many-graphs-homo-anti-homo}
$^*$ Suppose that there is a mapping $f: V(G) \rightarrow V(H)$ on two graphs $G$ and $H$, so let $f(E(G))=\{f(u)f(v) \in E(H):uv\in E(G)\}$ be an induced edge subset based on $f$. There are the following restrictive conditions:
\begin{asparaenum}[\textrm{\textbf{Con}}-1. ]
\item \label{homo:Bondy-2008-homomorphism} $f(u)f(v)\in E(H)$ for $uv \in E(G)$.
\item \label{homo:full-isomorphic-homomorphism} $uv\in E(G)$ if and only if $f(u)f(v) \in E(H)$.
\item \label{homo:faithful-isomorphic-homomorphism} $f(V(G)\cup E(G))$ is an induced subgraph of $H$.
\item \label{homo:two-isomorphic-graphs} $G\cong H$.
\item \label{homo:graph-anti-homomorphism-G-H} $f(u)f(v)\not \in E(H)$ for $uv \in E(G)$.
\item \label{homo:graph-anti-homomorphism-H-G} $xy \not\in E(G)$ for $f(x)f(y)\in E(H)$.
\item \label{homo:public-key-m-subgraphs} $G$ contains vertex disjoint subgraphs $G_1,G_2,\dots,G_m$.
\item \label{homo:private-key-m-subgraphs}$H$ contains vertex disjoint subgraphs $H_1,H_2,\dots,H_m$.
\item \label{homo:H-k-homomorphism-to-G-k} $f:V(G_k)\rightarrow V(H_k)$ for $k\in [1,m]$, $G_k\rightarrow H_k$.
\item \label{homo:verteices-large} $|V(G)|>|V(H)|$ and $|E(H)|\geq |E(G)|$.
\item \label{homo:edges-large} $|V(G)|=|V(H)|$ and $|E(H)|> |E(G)|$.
\item \label{homo:verteices-edges-equal} $|V(G)|=|V(H)|$ and $|E(G)|=|E(H)|$.
\item \label{homo:H-no-loop-multiple} $H$ is simple.
\item \label{homo:H-multiple-edges} $H$ has multiple edges.
\item \label{homo:edges-subset-small} $|f(E(G))|<|E(H)|$.
\end{asparaenum}
\textbf{We have $G\rightarrow H$ from $G$ into $H$ to be}:
\begin{asparaenum}[\textrm{Homo}-1. ]
\item \cite{Bondy-2008} A \emph{graph homomorphism} if \textbf{Con}-\ref{homo:Bondy-2008-homomorphism} holds true.
\item \cite{Gena-Hahn-Claude-Tardif-1997} A \emph{full graph homomorphism}if \textbf{Con}-\ref{homo:full-isomorphic-homomorphism} holds true.
\item \cite{Gena-Hahn-Claude-Tardif-1997} A \emph{faithful graph homomorphism}if \textbf{Con}-\ref{homo:faithful-isomorphic-homomorphism} holds true.
\item \cite{Yao-Wang-2106-15254v1} A \emph{graph anti-homomorphism} if \textbf{Con}-\ref{homo:graph-anti-homomorphism-G-H} and \textbf{Con}-\ref{homo:graph-anti-homomorphism-H-G} hold true, meanwhile.
\item \cite{Yao-Wang-2106-15254v1} A \emph{self-isomorphic graph anti-homomorphism} if \textbf{Con}-\ref{homo:two-isomorphic-graphs}, \textbf{Con}-\ref{homo:graph-anti-homomorphism-G-H} and \textbf{Con}-\ref{homo:graph-anti-homomorphism-H-G} hold true, simultaneously.
\item $^*$ A \emph{simple graph homomorphism} if \textbf{Con}-\ref{homo:Bondy-2008-homomorphism} and \textbf{Con}-\ref{homo:H-no-loop-multiple} hold true.
\item $^*$ A \emph{multiple-edge graph homomorphism} if \textbf{Con}-\ref{homo:Bondy-2008-homomorphism} and \textbf{Con}-\ref{homo:H-multiple-edges} hold true.
\item $^*$ A \emph{proper graph homomorphism} if \textbf{Con}-\ref{homo:Bondy-2008-homomorphism} and \textbf{Con}-\ref{homo:verteices-large} hold true.
\item $^*$ An \emph{edge-proper graph homomorphism} if \textbf{Con}-\ref{homo:Bondy-2008-homomorphism} and \textbf{Con}-\ref{homo:edges-large} hold true.
\item $^*$ A \emph{ve-equal graph homomorphism} if \textbf{Con}-\ref{homo:Bondy-2008-homomorphism} and \textbf{Con}-\ref{homo:verteices-edges-equal} hold true.
\item $^*$ If \textbf{Con}-\ref{homo:public-key-m-subgraphs}, \textbf{Con}-\ref{homo:private-key-m-subgraphs} and \textbf{Con}-\ref{homo:H-k-homomorphism-to-G-k} hold concurrently true, we say that $G$ admits an \emph{$m$-subgraph homomorphism} to $H$, denoted as
 \begin{equation}\label{eqa:m-subgraph-homomorphism}
 \partial (G\rightarrow H)^m_{1}=\{G_k\}\rightarrow^m_{k=1}\{H_k\}
 \end{equation} also, a \emph{partial $m$-subgraph homomorphism} from $G$ into $H$.
\item $^*$ If \textbf{Con}-\ref{homo:public-key-m-subgraphs} and \textbf{Con}-\ref{homo:private-key-m-subgraphs} hold true, and moreover $f:V(G_k)\rightarrow V(H_k)$ with $f(u)f(v)\not \in E(H_k)$ for $uv \in E(G_k)$ and $xy \not\in E(G_k)$ for $f(x)f(y)\in E(H_k)$ for $k\in [1,m]$, this case is called an \emph{$m$-subgraph anti-homomorphism} from $H$ into $G$, denoted as
 \begin{equation}\label{eqa:m-subgraph-anti-homomorphism}
 \partial (H\rightarrow _{anti}G)^m_{1}=\{H_k\}^m_{1}\rightarrow_{anti}\{G_k\}^m_{1}
 \end{equation} also, a \emph{partial $m$-subgraph anti-homomorphism} from $H$ into $G$.\qqed
\end{asparaenum}
\end{defn}

\begin{defn} \label{defn:homomorphism-authenticationss}
$^*$ Suppose that there is a mapping $f: V(G) \rightarrow V(H)$ on a topological public-key $G$ and a topological private-key $H$, so let $f(E(G))=\{f(u)f(v) \in E(H):uv\in E(G)\}$ be an induced edge subset based on $f$, and $m=|\{f(u):u\in V(G)\}|$. By Definition \ref{defn:vertex-split-coinciding-operations} and Definition \ref{defn:many-graphs-homo-anti-homo}, we have:
\begin{asparaenum}[\textrm{\textbf{Hoau}}-1. ]
\item A \emph{graph homomorphism} $G\rightarrow H$ induces a topological authentication defined by a vertex-coincided graph $\odot_{m}\langle G,H\rangle$, such that $G$ and $H$ both are subgraphs of $\odot_{m}\langle G,H\rangle$.
\item A \emph{graph anti-homomorphism} $G\rightarrow _{anti}H$ induces a topological authentication defined by a vertex-coincided graph $\odot_{m}\langle G,H\rangle$ with $E(\odot_{m}\langle G,H\rangle)=E(G)\cup E(H)$ and $V(\odot_{m}\langle G,H\rangle)\subseteq V(H)$, and moreover $\odot_{m}\langle G,H\rangle$ is simple if $G$ and $H$ both are simple.
\item A \emph{full graph homomorphism} $G\rightarrow H$ induces a topological authentication defined by a $W$-coincided graph $G[\ominus ^W_{m}] H$, such that $H=G[\ominus ^W_{m}] H$ with $W=G$.
\item A \emph{faithful graph homomorphism} $G\rightarrow H$ induces a topological authentication defined by a $W$-coincided graph $G[\ominus ^W_{m}] H$ with $W=G$, such that $G[\ominus ^W_{m}] H$ is a proper subgraph of $H$.
\item A \emph{self-isomorphic graph anti-homomorphism} $G\rightarrow _{anti}G$ induces a topological authentication defined by a vertex-coincided graph $\odot _{p}\langle G, G\rangle $ with $p=|V(G)|$, such that $p=|V(G[\ominus ^G_{n}] G)|$ and $|E(G[\ominus ^G_{n}] G)|=2|E(G)|$. We can describe $\odot _{p}\langle G, G\rangle $ in this way: $\odot _{p}\langle G, G\rangle $ contains just two edge-disjoint subgraphs with the same vertex set, in which each subgraph is just a copy of $G$.
\item An \emph{$m$-subgraph anti-homomorphism} $\partial (G\rightarrow _{anti}H)^m_{1}=\{G_k\}^m_{1}\rightarrow_{anti}\{H_k\}^m_{1}$ induces a topological authentication defined by an $m$-partial vertex-coincided graph
\begin{equation}\label{eqa:555555}
U=\partial (G\odot H)^m_{1}=\odot^m_{k=1}\langle \{G_k\}, \{H_k\}\rangle
\end{equation} such that
 $|V(U)|=\big|V(G)\setminus \bigcup^m_{k=1}V(G_k)\big |+|V(H)|$ and $|E(U)|=|E(G)|+|E(H)|$.\qqed
\end{asparaenum}
\end{defn}

\subsubsection{Graph anti-homomorphisms}

\begin{defn}\label{defn:new-graph-anti-homomorphisms}
\cite{Yao-Wang-2106-15254v1} A \emph{graph anti-homomorphism} $G\rightarrow _{anti} H$ from a graph $G$ into another graph $H$ is a mapping $g: V(G) \rightarrow V(H)$ such that $f(u)f(v)\not \in E(H)$ for each edge $uv \in E(G)$, and $xy\not \in E(G)$ for each edge $f(x)f(y)\in E(H)$. Moreover, if $G\cong H$ in $G\rightarrow _{anti} H$, then we say that $G$ admits a \emph{self-isomorphic graph anti-homomorphism}. \qqed
\end{defn}

In Fig.\ref{fig:self-isomorphic-anti-homo}, we can see a mapping $f: V(A) \rightarrow V(B)$. The vertex-coincided graph $G=A[\odot_{12}]B$ has $|V(G)|=|V(A)|=|V(B)|$ and $|E(G)|=|E(A)|+|E(B)|$, since $A\cong B$.

\begin{figure}[h]
\centering
\includegraphics[width=15.6cm]{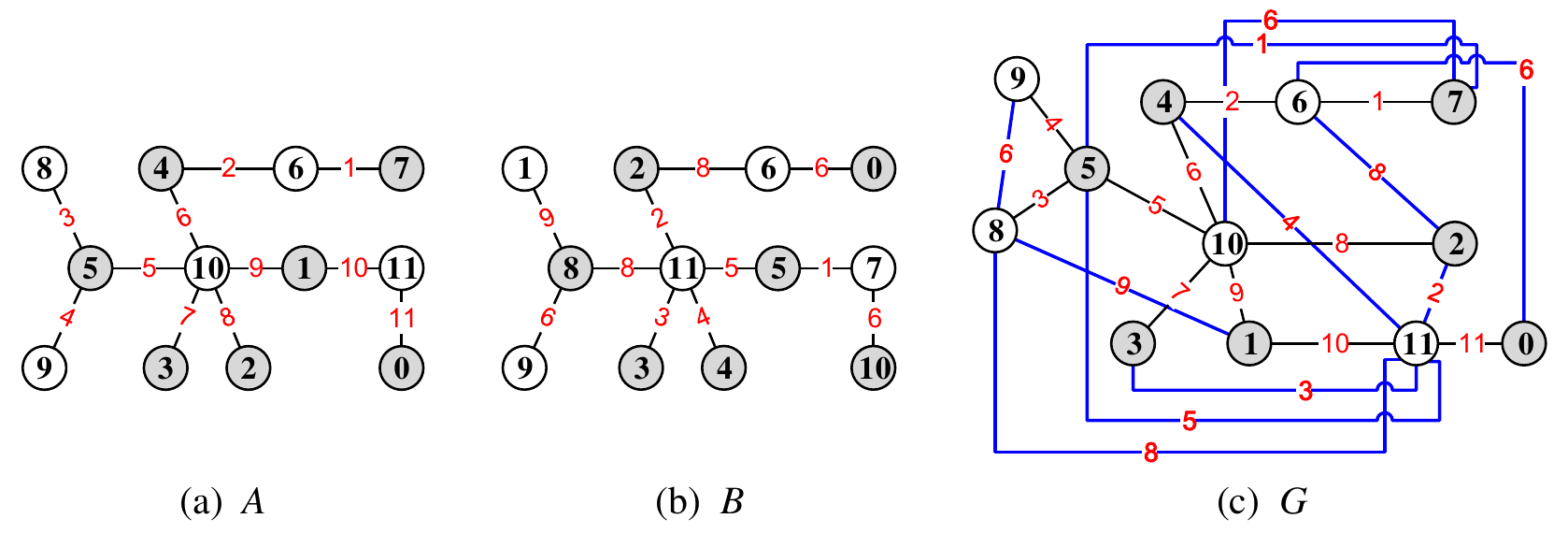}\\
\caption{\label{fig:self-isomorphic-anti-homo}{\small An example for understanding a self-isomorphic graph anti-homomorphism defined in Definition \ref{defn:new-graph-anti-homomorphisms}.}}
\end{figure}

\begin{problem}\label{problem:Charac-self-isomorphic-graph-anti-homomorphism}
\cite{Yao-Wang-2106-15254v1} \textbf{Characterize} graphs $G$ admitting \emph{self-isomorphic graph anti-homomorphisms} $G\rightarrow _{anti} G$ defined in Definition \ref{defn:new-graph-anti-homomorphisms}.
\end{problem}

This is an example of a topological public-key $G$ and a topological private-key $G$ being combined into one for forming a \emph{topological authentication}, also a \emph{$W$-coincided graph} $G[\ominus^H_{n}]G$ with $n=|V(G)|$ and $W=G$. We can partly answer Problem \ref{problem:Charac-self-isomorphic-graph-anti-homomorphism} in the following theorem:

\begin{thm}\label{thm:tree-self-isomorphic-graph-anti-homomorphism}
\cite{Yao-Su-Sun-Wang-Graph-Operations-2021} Each tree $T$ admits a self-isomorphic graph anti-homomorphism $T\rightarrow _{anti} T$ defined in Definition \ref{defn:new-graph-anti-homomorphisms} if its diameter $D(T)\geq 3$, such that $V\big (T[\ominus^T_{n}]T\big )=V(T)$ with $n=|V(T)|$, and $\big |E\big (T[\ominus^T_{n}]T\big )\big |=2|E(T)|$.
\end{thm}

\begin{defn}\label{defn:partly-graph-anti-homomorphisms}
\cite{Yao-Su-Sun-Wang-Graph-Operations-2021} A \emph{partial graph anti-homomorphism} $G\rightarrow _{panti} H$ from a graph $G$ into another graph $H$ is a mapping $g: V(G) \rightarrow V(H)$ such that

(i) there are two edge subsets $E_G\subset E(G)$ and $E_H\subset E(H)$ holding $f(u)f(v) \in E_H$ for each edge $uv \in E_G$ and $xy \in E_G$ for each edge $f(x)f(y)\in E_H$ true;

(ii) $f(u)f(v)\not \in E(H)\setminus E_H$ for each edge $uv \in E(G)\setminus E_G$, and $xy\not \in E(G)\setminus E_G$ for each edge $f(x)f(y)\in E(H)\setminus E_H$.

If $G\cong H$ in $G\rightarrow _{panti} H$, then we say that $G$ admits a \emph{partial self-isomorphic graph homomorphism}. \qqed
\end{defn}

\begin{problem}\label{problem:self-isomorphic-graph-homomorphism}
\cite{Yao-Su-Sun-Wang-Graph-Operations-2021} Each graph $G$ admits a partial self-isomorphic graph homomorphism $g: V(G) \rightarrow V(G)$ defined in Definition \ref{defn:partly-graph-anti-homomorphisms}, \textbf{determine} $\min_g\{|E_G|\}$ over all partial self-isomorphic graph homomorphism $g$ of $G$.
\end{problem}

\subsubsection{Colored graph homomorphisms}

\begin{defn}\label{defn:vertex-colored-graph-anti-homomorphisms}
$^*$ A \emph{vertex-colored graph anti-homomorphism} $G\rightarrow_{anti} H$, from a graph $G$ admitting a $W$-type vertex coloring $g$ into another graph $H$ admitting a $W\,'$-type vertex coloring $h$, is a bijection $f: V(G) \rightarrow V(H)$ such that $f(u)f(v)\not \in E(H)$ for each edge $uv \in E(G)$, and $xy\not \in E(G)$ for each edge $f(x)f(y)\in E(H)$, so
\begin{equation}\label{eqa:555555}
f: \{g(w):w\in V(G)\}=g(V(G))\rightarrow h(V(H))=\{h(z):z\in V(H)\}
\end{equation}and $h(f(u))h(f(v))\not \in E(H)$ for $uv \in E(G)$, and $g(f^{-1}(x))g(f^{-1}(y))\not \in E(G)$ for $xy\in E(H)$.\qqed
\end{defn}

\begin{defn}\label{defn:totally-colored-graph-homomorphisms}
$^*$ A \emph{totally-colored graph homomorphism} $G\rightarrow H$, from a graph $G$ admitting a $W$-type total coloring $f$ into another graph $H$ admitting a $W\,'$-type total coloring $g$, is a mapping $\varphi: V(G) \rightarrow V(H)$ with $\varphi(u)\varphi(v)\in E(H)$ for each edge $uv \in E(G)$, so $\varphi:f(V(G))\cup f(E(G))\rightarrow g(V(H))\cup g(E(H))$, where $f(V(G))=\{f(w):w\in V(G)\}$, $f(E(G))=\{f(uv):uv\in E(G)\}$, $g(V(H))=\{g(z):z\in V(H)\}$ and $g(E(H))=\{g(xy):xy\in E(H)\}$.\qqed
\end{defn}

\begin{defn}\label{defn:gracefully-graph-homomorphism}
\cite{Bing-Yao-Hongyu-Wang-arXiv-2020-homomorphisms} Let $G\rightarrow H$ be a graph homomorphism from a $(p,q)$-graph $G$ into another $(p\,',q\,')$-graph $H$ based on a mapping $\alpha: V(G) \rightarrow V(H)$ such that $\alpha(u)\alpha(v)\in E(H)$ for each edge $uv \in E(G)$. The graph $G$ admits a total coloring $f$, and the graph $H$ admits another total coloring $g$. Write $f(E(G))=\{f(uv):uv \in E(G)\}$ and $g(E(H))=\{g(\alpha(u)\alpha(v)):\alpha(u)\alpha(v)\in E(H)\}$, there are the following conditions:
\begin{asparaenum}[\textrm{C}-1. ]
\item \label{bipartite} $V(G)=X\cup Y$, each edge $uv \in E(G)$ holds $u\in X$ and $v\in Y$ true. $V(H)=W\cup Z$, each edge $\alpha(u)\alpha(v)\in E(G)$ holds $\alpha(u)\in W$ and $\alpha(v)\in Z$ true.
\item \label{edge-difference} $f(uv)=|f(u)-f(v)|$ for each $uv \in E(G)$, $g(\alpha(u)\alpha(v))=|g(\alpha(u))-g(\alpha(v))|$ for each $\alpha(u)\alpha(v)\in E(H)$.
\item \label{edge-homomorphism} $f(uv)=g(\alpha(u)\alpha(v))$ for each $uv \in E(G)$.
\item \label{vertex-color-set} $f(x)\in [1,q+1]$ for $x\in V(G)$, $g(y)\in [1,q\,'+1]$ with $y\in V(H)$.
\item \label{odd-vertex-color-set} $f(x)\in [1,2q+2]$ for $x\in V(G)$, $g(y)\in [1,2q\,'+2]$ with $y\in V(H)$.
\item \label{grace-color-set} $[1,q]=f(E(G))=g(E(H))=[1,q\,']$.
\item \label{odd-grace-color-set} $[1,2q-1]^o=f(E(G))=g(E(H))=[1,2q\,'-1]^o$.
\item \label{set-ordered} Set-ordered property: $\max f(X)<\min f(Y)$ and $\max g(W)<\min g(Z)$.
\end{asparaenum}

We say graph homomorphism $G\rightarrow H$ to be:

(i) a \emph{bipartitely graph homomorphism} if C-\ref{bipartite} holds true;

(ii) a \emph{graceful graph homomorphism} if C-\ref{edge-difference}, C-\ref{edge-homomorphism}, C-\ref{vertex-color-set} and C-\ref{grace-color-set} hold true;

(iii) a \emph{set-ordered graceful graph homomorphism} if C-\ref{edge-difference}, C-\ref{edge-homomorphism}, C-\ref{vertex-color-set}, C-\ref{grace-color-set} and C-\ref{set-ordered} hold true;

(iv) an \emph{odd-graceful graph homomorphism} if C-\ref{edge-difference}, C-\ref{edge-homomorphism}, C-\ref{odd-vertex-color-set} and C-\ref{odd-grace-color-set} hold true; and

(v) a \emph{set-ordered odd-graceful graph homomorphism} if C-\ref{edge-difference}, C-\ref{edge-homomorphism}, C-\ref{odd-vertex-color-set}, C-\ref{odd-grace-color-set} and C-\ref{set-ordered} hold true.\qqed
\end{defn}

\begin{rem} \label{rem:gracefully-graph-homomorphism}
\begin{asparaenum}[(i) ]
\cite{Bing-Yao-Hongyu-Wang-arXiv-2020-homomorphisms} There are the following issues about Definition \ref{defn:gracefully-graph-homomorphism}:
\item A graceful graph homomorphism $G\rightarrow H$ holds $|f(E(G))|=|g(E(H))|$ with $q=q\,'$ and $|f(V(G))|\geq |g(V(H))|$, in general.
\item Also, we call the totally-colored graph homomorphisms defined in Definition \ref{defn:gracefully-graph-homomorphism} as \emph{$W$-type totally-colored graph homomorphisms}, where a ``$W$-type totally-colored graph homomorphism'' is one kind of the totally-colored graph homomorphisms; and we say that the graph $G$ admits a $W$-type totally-colored graph homomorphism to $H$ in a $W$-type totally-colored graph homomorphism $G\rightarrow H$.
\item If $T\rightarrow H$ is a $W$-type totally-colored graph homomorphism, and so is $H\rightarrow T$, we say $T$ and $H$ are \emph{homomorphically equivalent} from each other, denoted as $T\leftrightarrow H$. If two graphs $G$ admitting a coloring $f$ and $H$ admitting a coloring $g$ hold $f(x)=g(\varphi(x))$ and $G\cong H$, we write this case by $G=H$.\paralled
\end{asparaenum}
\end{rem}

\begin{thm}\label{thm:parti-gracecular-bijective}
\cite{Bing-Yao-Hongyu-Wang-arXiv-2020-homomorphisms} If a (set-ordered) graceful graph homomorphism $\varphi: G \rightarrow H$ defined in Definition \ref{defn:gracefully-graph-homomorphism} holds $f(V(G))=g(V(H))$ and $f(E(G))=g(E(H))$, then $G=H$.
\end{thm}

By Definition \ref{defn:11-topo-auth-faithful} and Theorem \ref{thm:bijective-graph-homomorphism}, we have

\begin{thm}\label{thm:set-ordered-gracefully-bijective}
\cite{Bing-Yao-Hongyu-Wang-arXiv-2020-homomorphisms} If a (set-ordered) graceful graph homomorphism $\varphi: G \rightarrow H$ is faithful bijective, then $G=H$.
\end{thm}

\subsection{Graph operations for graph (anti-)homomorphisms}

\subsubsection{Graph-operations (anti-)homomorphisms of planar graphs}

\begin{defn} \label{defn:111111}
$^*$ If a graph $H$ is obtained by doing a graph operation ``$(\bullet)$'' to another graph $G$, we say that $G$ is \emph{graph-operation homomorphism} into $H$, denote this fact as $G\rightarrow _{oper}H$, or $H=(\bullet)\langle G \rangle $. Suppose that $(\bullet)^{-1}$ is the inverse operation of the graph operation ``$(\bullet)$'', so we have another \emph{graph-operation homomorphism} $H\rightarrow _{oper}G$, or $G=(\bullet)^{-1}\langle H \rangle $.\qqed
\end{defn}

Graph-operation (anti-)homomorphisms of planar graphs based on:

(i) wheel-expending operation, wheel-contracting operation \cite{Jin-Xu-55-56-configurations-arXiv-2107-05454v1}; and

(ii) rhombus-expending operation, rhombus-contracting operation \cite{Yao-Sun-Wang-Su-Maximal-Planar-Graphs-2021}.

In Fig.\ref{fig:graph-operation-homo}, a graph-operation homomorphism $B_1\rightarrow_{oper} B_2$ is obtained by contracting a 4-wheel $W_4$ of the center colored with 2, called the contracting 4-wheel operation in \cite{Jin-Xu-55-56-configurations-arXiv-2107-05454v1}, also, an \emph{anti-rhombus operation} in \cite{Yao-Sun-Wang-Su-Maximal-Planar-Graphs-2021}; two graph-operation homomorphisms $B_i\rightarrow_{oper} B_{i+1}$ is obtained by deleting a 3-degree vertex from $B_i$ for $i=2,3$; a graph-operation homomorphism $B_4\rightarrow_{oper} B_5$ is obtained by exchange the colors of some vertices of $B_4$. Conversely, two graph-operation homomorphisms $B_j\rightarrow_{oper} B_{j+1}$ is obtained by adding a 3-degree vertex in a inner triangle of $B_j$ for $j=5,6$, also, a \emph{3-wheel-expending operation} in \cite{Jin-Xu-55-56-configurations-arXiv-2107-05454v1}; the last graph-operation homomorphism $B_7\rightarrow_{oper} B_8$ is obtained by doing a \emph{rhombus operation} to $B_7$.

Since $B_k\cong B_{8-k+1}$ for $k\in [1,4]$, so we have four \emph{colored graph homomorphisms} $B_k\rightarrow B_{8-k+1}$ by Definition \ref{defn:vertex-colored-graph-anti-homomorphisms} for $k\in [1,4]$.

\begin{figure}[h]
\centering
\includegraphics[width=15cm]{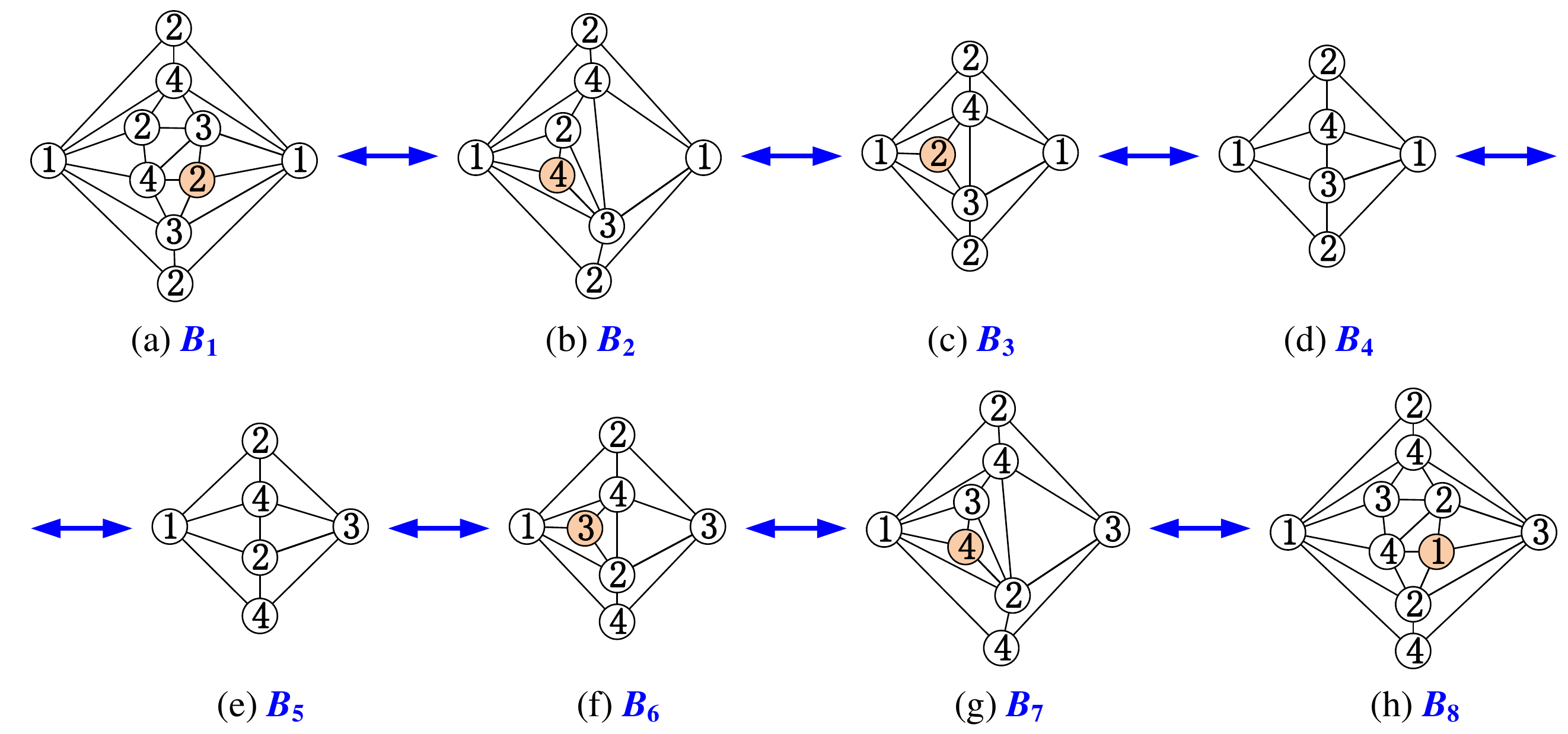}\\
\caption{\label{fig:graph-operation-homo}{\small An example for illustrating graph-operation homomorphisms, cited from \cite{Jin-Xu-55-56-configurations-arXiv-2107-05454v1}.}}
\end{figure}

\subsubsection{Uncolored-rhombus algorithms}

\textbf{Planar rhombus operation} is based on the $P_3$-\emph{subdivision operation} of planar graphs (Ref. \cite{Yao-Sun-Wang-Su-Maximal-Planar-Graphs-2021}, \cite{Hongyu-Wang-2018-Doctor-thesis}). Let $P_3=xwy$ be a path of three vertices $x,w,y$ in a planar graph $G$ (see Fig.\ref{fig:no-color-rhombus-operation}(a)). There are two processes:

(i) We do the vertex-splitting operation defined in Definition \ref{defn:vertex-split-coinciding-operations} to the middle vertex $w$ of the path $P_3$ by vertex-splitting $w$ into two vertices $w_1$ and $w_2$ holding the neighbor set $N_G(w)=N(w_1)\cup N(w_2)$, where two neighbor sets $N(w_1)=\{x,u_1,u_2,\dots, u_k,y\}$ and $N(w_1)=\{u_{k+1},u_{k+2},\dots, u_m\}$; next we add two new edges $xw_2,yw_2$ to join $w_2$ with both vertices $x,y$ respectively, such that the resultant graph, denoted as $G\wedge P_3$, contains a 4-cycle $C_4=xw_1yw_2z$ (see Fig.\ref{fig:no-color-rhombus-operation}(b)), we call the graph $G\wedge P_3$ to be a $P_3$-\emph{subdivision graph} after doing a $P_3$-\emph{subdivision operation} to the path $P_3=xwy$.

(ii) We put a new vertex $a$ into the face with boundary $C_4=xw_1yw_2x$ in the graph $G\wedge P_3$, and join the vertex $a$ with each vertex of the 4-cycle $C_4$ by four new edges, the resultant graph, denoted as $\langle\oplus\rangle(G)$, just contains a \emph{rhombus} $\langle\oplus\rangle xw_1yaw_2z$, see Fig.\ref{fig:no-color-rhombus-operation}(c).

The above two processes (i)$+$(ii) of obtaining the graph $\langle\oplus\rangle(G)$ is called a \emph{planar rhombus operation}, and the procedure from $\langle\oplus\rangle(G)$ to $G$ is called the \emph{anti-rhombus operation}. Here, two vertex numbers $|V(\langle\oplus\rangle(G))|=|V(G)|+2$ and two edge numbers $|E(\langle\oplus\rangle(G))|=|E(G)|+6$. Also, the planar rhombus operation is the UNCOLORED-RHOMBUS algorithm.

Suppose that $C_4=xw_1yw_2x$ is a 4-cycle of the graph $G\wedge P_3$ shown in Fig.\ref{fig:no-color-rhombus-operation}(b) after doing a $P_3$-\emph{subdivision} to the path $P_3=xwy$ of the planar graph $G$, then we have a planar graph $G\wedge P_3\langle -\rangle e$ containing an A-sub-rhombus $C_4+w_1w_2$ (see Fig.\ref{fig:no-color-rhombus-operation}(d)), and have a planar graph $G\wedge P_3\langle |\rangle e$ containing a B-sub-rhombus $C_4+xy$ shown in Fig.\ref{fig:no-color-rhombus-operation}(e). We call the procedures of obtaining A-sub-rhombus and B-sub-rhombus \emph{planar sub-rhombus operations}.

\begin{figure}[h]
\centering
\includegraphics[width=16cm]{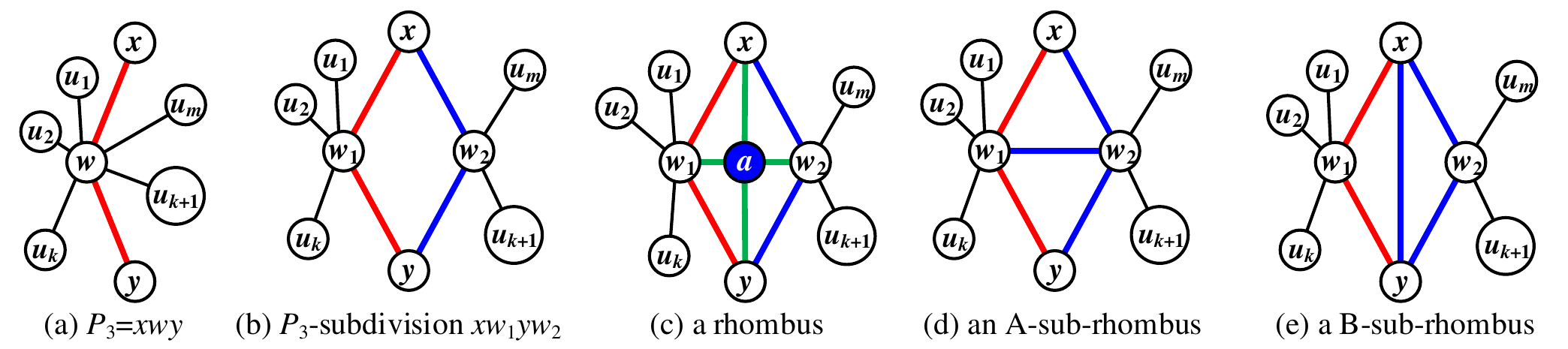}
\caption{\label{fig:no-color-rhombus-operation}{\small (a) A planar graph $G$; (b) a $P_3$-subdivision graph $G\wedge P_3$; (c) the resultant graph $\langle\oplus\rangle(G)$ after doing a planar rhombus operation to $G$; (d) a planar graph $G\wedge P_3\langle -\rangle e$ containing an A-sub-rhombus $C_4+w_1w_2$; (e) a planar graph $G\wedge P_3\langle |\rangle e$ containing a B-sub-rhombus $C_4+xy$, cited from \cite{Yao-Sun-Wang-Su-Maximal-Planar-Graphs-2021}.}}
\end{figure}

\subsubsection{Colored-rhombus algorithms}

\textbf{COLORED-RHOMBUS algorithm-I.} The COLORED-RHOMBUS algorithm-I is called \emph{$4$-colored rhombus operation} (Ref. \cite{Yao-Sun-Wang-Su-Maximal-Planar-Graphs-2021}). Doing the $P_3$-subdivision operation to a path $P_3=xwy$ of a $4$-colored planar graph $T$, so the resultant planar graph has an inner face $f_{ace}$ having boundary $xw_1yw_2x$, and adding a new vertex $c$ in the center of the face $f_{ace}$ and we use four new edges to join $c$ with four vertices $x,w_1,y,w_2$, respectively, and we get a new planar graph, denoted as $\langle\oplus\rangle(T)$. Suppose that the $4$-colored planar graph $T$ admits a proper vertex $4$-coloring $\alpha $, we define a proper vertex $4$-coloring $\beta $ of the new planar graph $\langle\oplus\rangle(T)$ by setting $\beta (t)=\alpha(t)$ for $t\in V(T)\setminus \{w\}$, $\beta(w_1)=\beta(w_2)=\alpha(w)$, and $\beta(c)\in \{1,2,3,4\}\setminus \{\beta(x),\beta(y),\beta(w)\}$.

\vskip 0.4cm

\textbf{COLORED-RHOMBUS algorithm-II.} Firstly, we see examples on $4$-colored planar graphs shown in Fig.\ref{fig:color-rhombus-operation}:

(a) A path $P_3=rst$ of a $4$-colored planar graph $H$;

(b) the $4$-colored planar graph $H_1$ is obtained by doing a \emph{colored $P_3$-subdivision operation} to $H$ and contains a 4-cycle $C\,'_4=rs_1ts_2r$;

(c) the $4$-colored planar graph $H_2=\langle\oplus\rangle(H)$ is obtained by doing a \emph{colored rhombus operation};

(d) two selected paths $P\,'_3=xwu$ and $P\,''_3=xwy$;

(e) the $4$-colored planar graph $L_2$ is obtained by doing the $P_3$-subdivision operation to the paths $P\,'_3=xwu$ and $P\,''_3=xwy$ respectively, and get a $4$-colored A-sub-rhombus $F_1=xw_1yw_2x+w_1w_2$ and another $4$-colored A-sub-rhombus $F_2=xw_3uw_1x+w_3w_1$.

Also, the $4$-colored planar graph $L_2$ is called a \emph{two-$P_3$ two-A-sub-rhombus $4$-colored planar graph}. Notice that two paths $P\,'_3=xwu$ and $P\,''_3=xwy$ have a common vertex $w$, so they form a star $K_{1,3}$ with its \emph{center} $w$, such that $K_{1,3}+uy$ is a \emph{funnel}; and moreover two $4$-colored A-sub-rhombus $F_1$ and $F_2$ have a common edge $xw_1$, so the graph $(F_1\cup F_2)+uy$ is just a \emph{wheel} $W_5$ having vertices $x,w_3,u,y,w_2,w_1$, where $w_1$ is the \emph{center} of the wheel $W_5$.

\begin{defn} \label{defn:111111}
Let $F_{unnel}$ be a funnel with vertex set $V(F_{unnel})=\{x,w,y,u\}$ and edge set $E(F_{unnel})=\{xw,wy,wu,yu\}$ in a colored planar graph $G$ admitting a proper vertex $4$-coloring $h$, where $h(x)=f(y)$, $h(x)\not =f(w)$, $h(w)\not =f(u)$ and $h(u)\not =f(y)$. We do the $P_3$-subdivision operation to the paths $P\,'_3=xwu$ and $P\,''_3=xwy$ respectively, so we get a wheel $W_5$ consisted of vertices $x,w_3,u,y,w_2,w_1$ and the center $w_1$, the resultant planar graph is denoted $L(\wedge F_{unnel})$, such that $L(\wedge F_{unnel})$ admits a proper vertex $4$-coloring $g$ defined as: $g(t)=h(t)$ for $V(G)\setminus \{w_1,w_2,w_3\}$, and $g(w_2)=g(w_3)=h(w)$, as well as $g(w_1)\in \{1,2,3,4\}\setminus \{g(x),g(y),g(u),g(w_2),g(w_3)\}$.\qqed
\end{defn}

\begin{figure}[h]
\centering
\includegraphics[width=16cm]{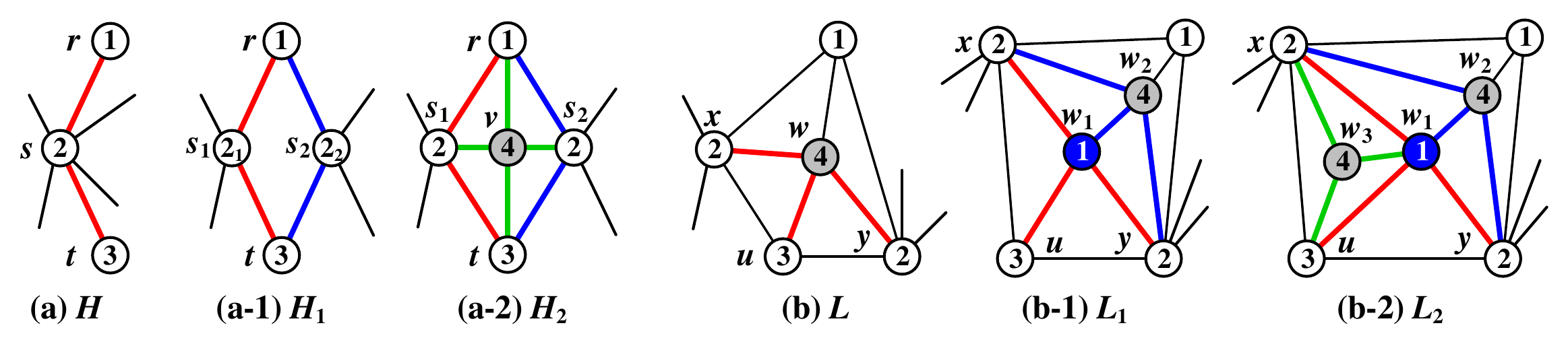}
\caption{\label{fig:color-rhombus-operation}{\small Examples for the colored rhombus operation, cited from \cite{Yao-Sun-Wang-Su-Maximal-Planar-Graphs-2021}.}}
\end{figure}

Obviously, our COLORED-RHOMBUS algorithm-I and COLORED-RHOMBUS algorithm-II keep the planarity and the proper vertex $4$-coloring to $4$-colored planar graphs participated in the planar rhombus operation and the $P_3$-subdivision operation. In Fig.\ref{fig:more-rhombus-operations}, we have two graph-operation homomorphism chains as follows:
$$G\rightarrow _{oper}H_1\rightarrow _{oper}H_2\rightarrow _{oper}H_3\rightarrow _{oper}H_4,~G\rightarrow _{oper}G_1\rightarrow _{oper}G_2\rightarrow _{oper}G_3\rightarrow _{oper}G_4.$$
The above two graph-operation homomorphism chains form two \emph{topological authentication chains}.

\begin{figure}[h]
\centering
\includegraphics[width=16.4cm]{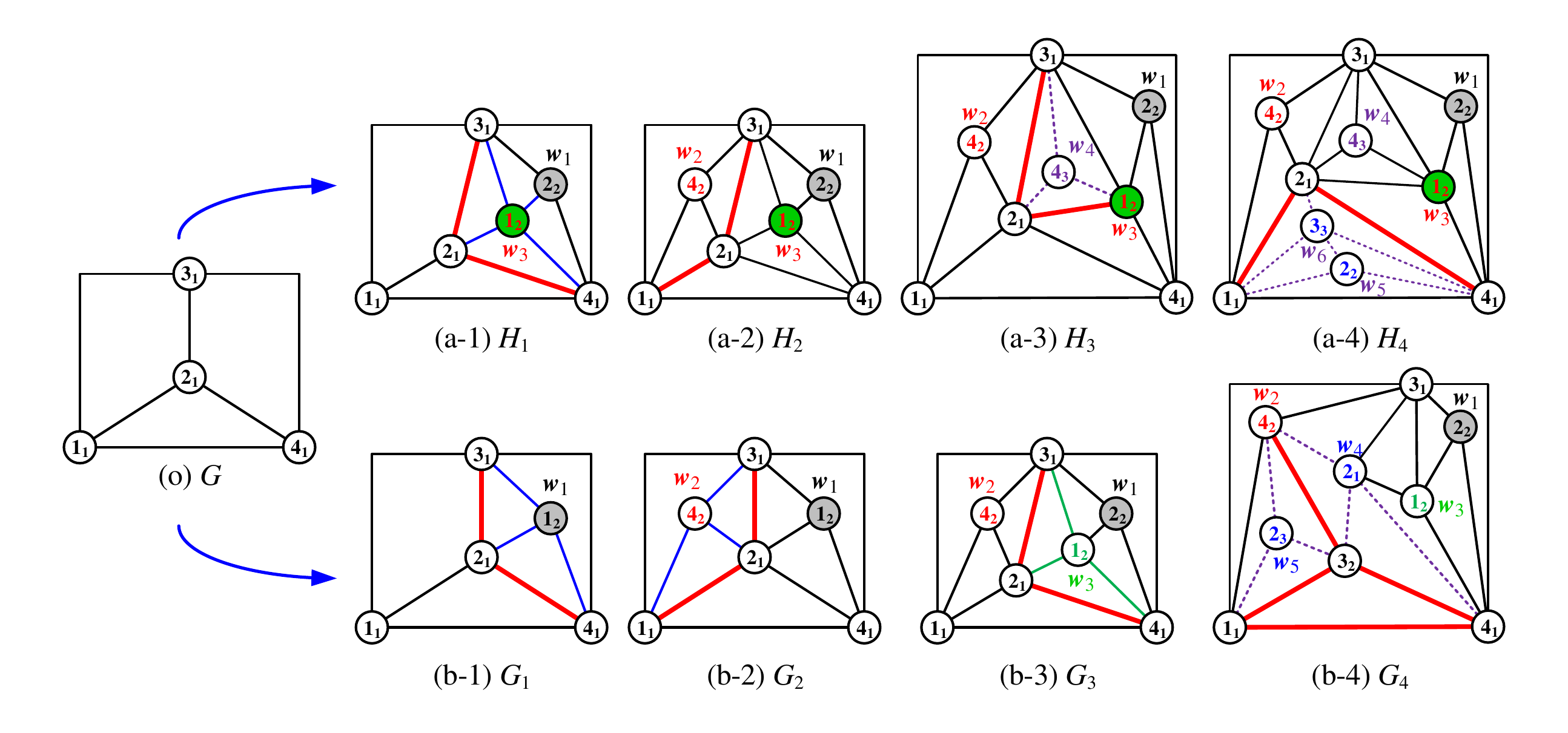}\\
\caption{\label{fig:more-rhombus-operations}{\small Examples for illustrating the rhombus operation, cited from \cite{Yao-Sun-Wang-Su-Maximal-Planar-Graphs-2021}.}}
\end{figure}

\begin{thm}\label{thm:color-number-up-bound}
\cite{Yao-Sun-Wang-Su-Maximal-Planar-Graphs-2021} Suppose that a maximal planar $(p,q)$-graph $G$ admits a proper vertex $4$-coloring $g$, such that vertex set $V(G)=\bigcup^4_{i=1} V_i(G)$, and vertex color sets $C_i(G)=\{g(x_{i,j}):x_{i,j}\in V_i(G)\}=\{i_j,i_j,\dots ,i_{m_i}\}$ for $i\in [1,4]$, then $|C_i(G)|=m_i\leq \frac{p}{2}$.
\end{thm}

\begin{figure}[h]
\centering
\includegraphics[width=15cm]{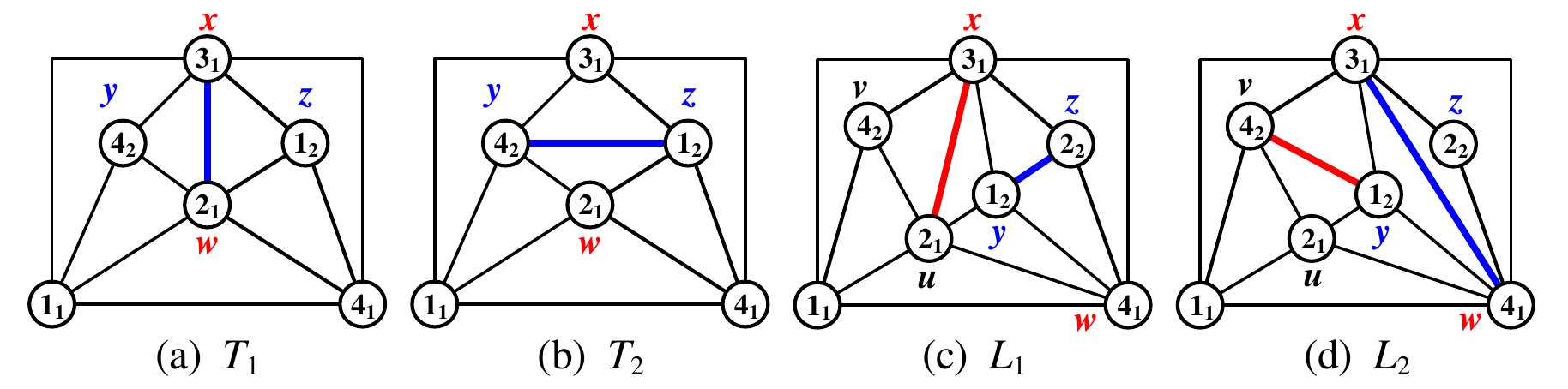}
\caption{\label{fig:edge-flip-operation}{\small Examples for understanding the edge-flip operation, cited from \cite{Yao-Sun-Wang-Su-Maximal-Planar-Graphs-2021}.}}
\end{figure}

\begin{problem}\label{problem:infinite-graph-homomorphism}
Are there infinite graph-operation homomorphisms?
\end{problem}

\subsubsection{Topcode-matrix homomorphism, Graph-set homomorphism}

According to a \emph{graph homomorphism} $L\rightarrow J$ on two graphs of $q$ edges, we have a \emph{Topcode-matrix homomorphism} $T_{code}(L)_{3\times q}\rightarrow T_{code}(J)_{3\times q}$ from a Topcode-matrix $T_{code}(L)_{3\times q}$ into another Topcode-matrix $T_{code}(J)_{3\times q}$.

\begin{thm}\label{thm:666666}
There is a graph set $S(T_{code}(L))=\{L_k\}^M_{k=1}$ with graph homomorphisms $L_k\rightarrow L_{k+1}$ for $k\in [1,M]$, such that $T_{code}(L_k)_{3\times q}=T_{code}(L)_{3\times q}$, $|V(L_k)|>|V(L_{k+1})|$ with $|V(L_1)|=2q$. We say $S(T_{code}(L))$ is the \emph{topological expression set} of the Topcode-matrix $T_{code}(L)_{3\times q}$, and each graph of $S(T_{code}(L))$ is a \emph{topological expression} of $T_{code}(L)$.
\end{thm}

If there is a graph set $S(T_{code}(J))=\{J_k\}^M_{k=1}$ with graph homomorphisms $J_k\rightarrow J_{k+1}$ for $k\in [1,M]$, such that $T_{code}(J_k)_{3\times q}=T_{code}(J)_{3\times q}$, then, we have a \emph{graph-set homomorphism} $S(T_{code}(L))\rightarrow S(T_{code}(J))$ obtained by graph homomorphisms $L_k\rightarrow J_k$ for $k\in [1,M]$.

\subsection{Authentications based on various graph homomorphisms}

If we select a Topsnut-gpw (a) as a \emph{public key} $G_{(\textrm{a})}$ in Fig.\ref{fig:one-matrix-more-graphs}, then we have at least five Topsnut-gpws $G_{(k)}$ with $k=$b,c,d,e,f, as \emph{private keys}, to form five set-ordered gracefully graph homomorphisms $G_{(k)}\rightarrow G_{(\textrm{a})}$. In the topological structure of view, $G_{(ii)}$ is not isomorphic to $G_{(jj)}$, that is, $G_{(ii)}\not\cong G_{(jj)}$ for $i,j=$a,b,c,d,e,f and $i\neq j$. We, by these six Topsnut-gpws, have a Topcode-matrix (Ref. \cite{Sun-Zhang-Zhao-Yao-2017, Yao-Sun-Zhao-Li-Yan-ITNEC-2017, Yao-Zhang-Sun-Mu-Sun-Wang-Wang-Ma-Su-Yang-Yang-Zhang-2018arXiv}) as follows:
\begin{equation}\label{eqa:Topcode-matrix-vs-6-Topsnut-gpws}
\centering
{
\begin{split}
T_{code}= \left(
\begin{array}{ccccccccccc}
6&5&6&6&6&1&1&1&1&1\\
1&2&3&4&5&6&7&8&9&10\\
7&7&9&10&11&7&8&9&10&11
\end{array}
\right)
\end{split}}
\end{equation} So, each of these six Topsnut-gpws corresponds the unique Topcode-matrix $T_{code}$ defined in Eq.(\ref{eqa:Topcode-matrix-vs-6-Topsnut-gpws}). Moreover, the Topcode-matrix $T_{code}$ can distribute us, in total, $30!$ \emph{number-based strings} $S_{k}$ with $k\in [1,30!]$ like the following number-based string
$$S_{1}=617725639104665117611678711891089111011$$
with $39$ numbers.
\begin{figure}[h]
\centering
\includegraphics[width=16.4cm]{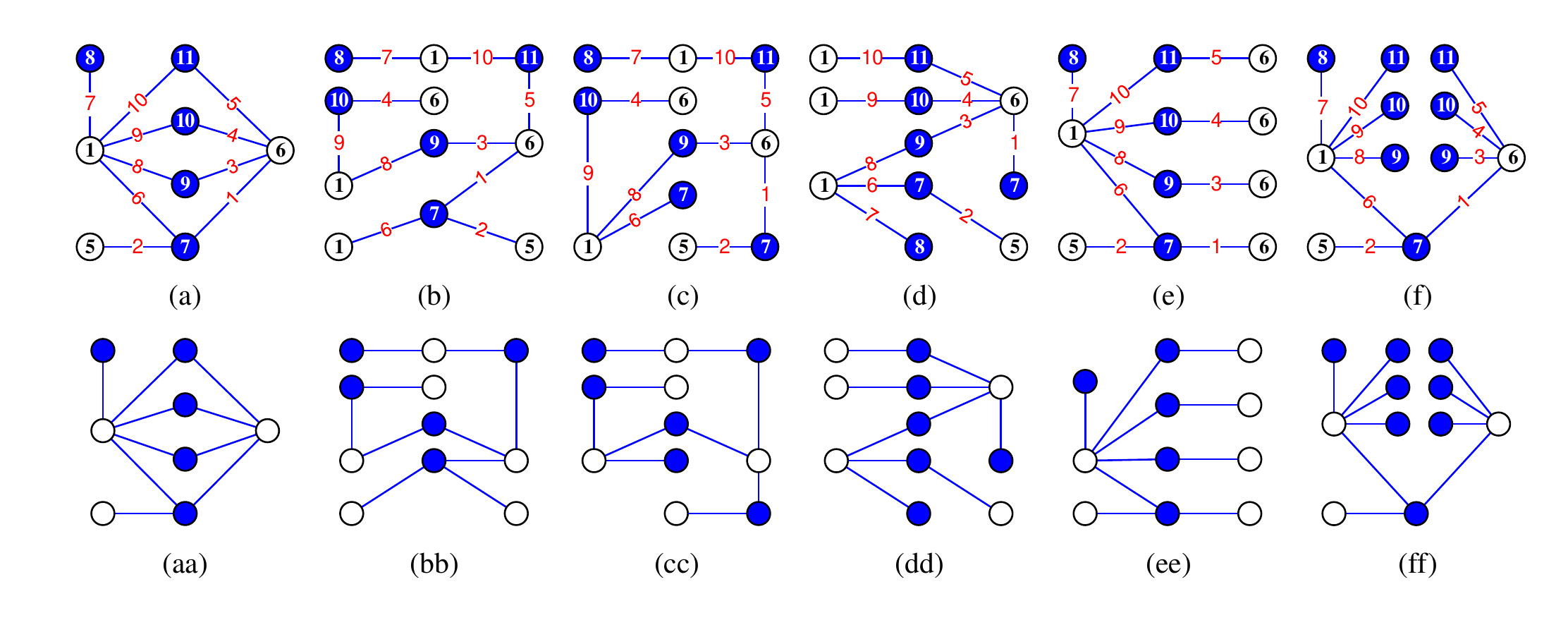}\\
\caption{\label{fig:one-matrix-more-graphs}{\small A Topcode-matrix shown in Eq.(\ref{eqa:Topcode-matrix-vs-6-Topsnut-gpws}) corresponds six Topsnut-gpws (a)-(f), cited from \cite{Bing-Yao-Hongyu-Wang-arXiv-2020-homomorphisms}.}}
\end{figure}

For the reason of authentications, we have to solve a problem as follows:

\vskip 0.4cm

\textbf{Number String Decomposition and Graph Homomorphism Problem (NSD-GHP).} \cite{Bing-Yao-Hongyu-Wang-arXiv-2020-homomorphisms} Given a number-based string $S_{1}=c_1c_2\cdots c_m$ with $c_i\in [0,9]$, decompose it into $30$ segments $c_1c_2\cdots c_m=a_1a_2\cdots a_{30}$ with $a_j=c_{n_j}c_{n_j+1}\cdots c_{n_{j+1}}$ with $j\in [1,29]$, $n_1=1$ and $n_{30}=m$. And use $a_k$ with $k\in [1,30]$ to reform the Topcode-matrix $T_{code}$ in Eq.(\ref{eqa:Topcode-matrix-vs-6-Topsnut-gpws}), and moreover reconstruct all Topsnut-gpws (like six Topsnut-gpws corresponds shown in Fig.\ref{fig:one-matrix-more-graphs}). By the found Topsnut-gpws corresponding the common Topcode-matrix $T_{code}$, \textbf{find} the public Topsnut-gpws $H_i$ and the private Topsnut-gpws $G_i$ as we desired, such that each mapping $\varphi_i:V(G_i)\rightarrow V(H_i)$ forms a graph homomorphism $G_i\rightarrow H_i$ with $i\in [1,a]$.

\vskip 0.4cm

So we have a topological authentication $\textbf{T}_{\textbf{a}}\langle\textbf{X},\textbf{Y}\rangle=P_{ub}(\textbf{X}) \rightarrow _{\textbf{F}} P_{ri}(\textbf{Y})$ of multiple variables defined in Definition \ref{defn:topo-authentication-multiple-variables} based on a common Topcode-matrix $T_{code}$ such that $P_{ub}(\textbf{X})=(H_1$, $H_2,\dots ,H_a)$ and $P_{ri}(\textbf{Y})=(G_1,G_2,\dots ,G_a)$ hold graph homomorphisms $G_i\rightarrow H_i$ with $i\in [1,a]$.

\textbf{Analysis of complexity of NSD-GHP:}
\begin{asparaenum}[\textrm{Comp}-1. ]
\item Since number-based strings are not integers, the well-known integer decomposition techniques can not be used to solve the number-based string decomposition problem.
\item No polynomial algorithm for cutting a number-based string $S=c_1c_2\cdots c_m$ with $c_i\in [0,9]$ and $S$ is not encrypted into $a_1a_2\cdots a_{3q}$, such that all $a_i$ to be correctly the elements of some matrix. As known, there are several kinds of matrices related with graphs, for example, graph adjacency matrix, Topsnut-matrix, Topcode-matrix and Hanzi-matrix, and so on.
\item If the matrix in problem has been found, it is difficult to guess the desired graphs, since it will meet NP-hard problems, such as, Graph Isomorphic Problem, and Hanzi-graph Decomposition Problem (Ref. \cite{Yao-Mu-Sun-Sun-Zhang-Wang-Su-Zhang-Yang-Zhao-Wang-Ma-Yao-Yang-Xie2019}), and so on.
\item If the desired graphs have been determined, we will face a large number of the existing graph colorings and graph labelings for coloring exactly the desired graphs, as well as unknown problems of graph theory, such as, Graph Total Coloring Problem, Graceful Tree Conjecture, \emph{etc.}
\end{asparaenum}

\subsection{Operations, homomorphisms and lattices on degree-sequences}

\subsubsection{Complex graphs}

Many contents in this subsection are cited from \cite{Yao-Zhang-Wang-Su-Integer-Decomposing-2021}.

\begin{defn}\label{defn:imaginary-complex-graph}
\cite{Wang-Yao-Su-Wanjia-Zhang-2021-IMCEC} A \emph{complex graph} $H\overline{\ominus} iG$ has its own vertex set $=V(H)$ and \emph{imaginary vertex set} $V(iG)$, and its own edge set $E(H\overline{\ominus} iG)=E(H)\cup E(iG)\cup E^{1/2}$, where $E^{1/2}=\{xy:x\in V(H), y\in V(iG)\}$ is \emph{half-half edge set}, $E(H)$ is \emph{popular edge set}, and $E(iG)$ is \emph{imaginary edge set}. Moreover, each vertex $x\in V(H)$ is called a \emph{vertex} again, each vertex $y\in V(iG)$ is called an \emph{imaginary vertex}; each edge of $E(H)$ is called an \emph{edge} again, each edge of $E(iG)$ is called an \emph{imaginary edge}, and each edge of the half-half edge set is called a \emph{half-half edge}.\qqed
\end{defn}

\begin{rem}\label{rem:333333}
A complex graph $H\overline{\ominus} iG$ with $p_r=|V(H)|$ and $p_i=|V(iG)|$ has its own \emph{complex degree-sequence}
\begin{equation}\label{eqa:complex-degree-sequence-express}
\textrm{C}_{\textrm{deg}}(H\overline{\ominus} iG)=\textrm{C}_{\textrm{deg}}(H)\overline{\ominus} \textrm{C}_{\textrm{deg}}(iG)i^2=(a_1,a_2, \dots ,a_{p_r},-b_1,-b_2, \dots ,-b_{p_i})
\end{equation} where $\textrm{C}_{\textrm{deg}}(H)=(a_1,a_2, \dots ,a_{p_r})$ and $\textrm{C}_{\textrm{deg}}(iG)=(b_1,b_2, \dots ,b_{p_i})$ with $i^2=-1$, and \emph{complex degree} $\textrm{C}_{\textrm{deg}}(x_j)=a_j\geq 0$ for each vertex $x_j\in V(H)$ with $j\in [1$, $p_r]$, and complex degree $\textrm{C}_{\textrm{deg}}(y_s)=b_s\geq 0$ for each imaginary vertex $y_s\in V(iG)$ with $s\in [1,p_i]$, and moreover the \emph{size} of the complex degree-sequence $\textrm{C}_{\textrm{deg}}(H\overline{\ominus} iG)$ is $|\textrm{C}_{\textrm{deg}}(H\overline{\ominus} iG)|=\sum^{p_r}_{i=1}a_{i}+\sum^{p_i}_{j=1}b_{j}$.

In general, the complex degree-sequence of a complex graph $H\overline{\ominus} iG$ is written as $\textrm{C}_{\textrm{deg}}(H\overline{\ominus} iG)=(c_1,c_2, \dots , c_{n})$ with $n=p_r+p_i$, where $c_1c_2 \cdots c_{n}$ is a permutation of $a_1a_2 \cdots a_{p_r} (-b_1)(-b_2) \cdots (-b_{p_i})$. In the view of vectors, a complex degree-sequence $\textrm{C}_{\textrm{deg}}(H\overline{\ominus} iG)$ is just a \emph{popular vector}, so we call $\textrm{C}_{\textrm{deg}}(H\overline{\ominus} iG)$ a \emph{graphic vector}. The \emph{proper degree-sequence} $\textrm{\textbf{d}}=(c\,'_{j_1},c\,'_{j_2}, \dots , c\,'_{j_n})$ with $c\,'_{j_s}=|c_{j_s}|$ for $c_{j_s}\in \textrm{C}_{\textrm{deg}}(H\overline{\ominus} iG)$ and $c\,'_{j_s}\geq c\,'_{j_{s+1}}$ for $s\in [1,n]$ holds Eq.(\ref{eqa:Erdos-Gallai-absolut}).

In the topological authentication of view, $H$ can be considered as a public-key, and $G$ is a private-key in a complex graph $H\overline{\ominus} iG$. \paralled
\end{rem}

\begin{defn} \label{defn:complex-graphs-homomorphisms}
$^*$ Let $H_1\overline{\ominus} iG_1$ and $H_2\overline{\ominus} iG_2$ be two complex graphs defined in Definition \ref{defn:imaginary-complex-graph}. If there is a mapping $\varphi: V(H_1\overline{\ominus} iG_1)\rightarrow V(H_2\overline{\ominus} iG_2)$ holding $\varphi: V(H_1)\rightarrow V(H_2)$ and $\varphi: V(iG_1)\rightarrow V(iG_2)$ true, such that $\varphi(s)\varphi(t)\in E(E^{1/2}_2)$ for $st\in E(E^{1/2}_1)$, $\varphi(x)\varphi(y)\in E(H_2)$ for $xy\in E(H_1)$ and $\varphi(u)\varphi(v)\in E(iG_2)$ for $uv\in E(iG_1)$. Then we say that the complex graph $H_1\overline{\ominus} iG_1$ is \emph{complex graphs homomorphism} into the complex graph $H_2\overline{\ominus} iG_2$, denoted as $H_1\overline{\ominus} iG_1\rightarrow H_2\overline{\ominus} iG_2$.\qqed
\end{defn}

\begin{example}\label{exa:8888888888}
In Fig.\ref{fig:4-complex-graphs}, we can see four \emph{complex degree-sequences}

$\textrm{C}_{\textrm{deg}}(L_1)$=(5, 5, 3, 3)$\overline{\ominus}$(4, 3, 2, 1)$i^2$, $\textrm{C}_{\textrm{deg}}(L_2)$=(4, 1, 4, 1, 3, 3)$\overline{\ominus}$(4, 3, 2, 1)$i^2$,

$\textrm{C}_{\textrm{deg}}(L_3)$=(5, 5, 3, 2, 1)$\overline{\ominus}$(2, 2, 3, 2, 1)$i^2$, $\textrm{C}_{\textrm{deg}}(L_4)$=(2, 2, 1, 2, 2, 1, 3, 3)$\overline{\ominus}$(2, 2, 3, 2, 1)$i^2$\\
by Eq.(\ref{eqa:complex-degree-sequence-express}). And moreover, we have three \emph{complex graph homomorphisms} $L_k\rightarrow L_{k-1}$ for $k\in [2,4]$ according to Definition \ref{defn:complex-graphs-homomorphisms}.
\end{example}

\begin{figure}[h]
\centering
\includegraphics[width=16cm]{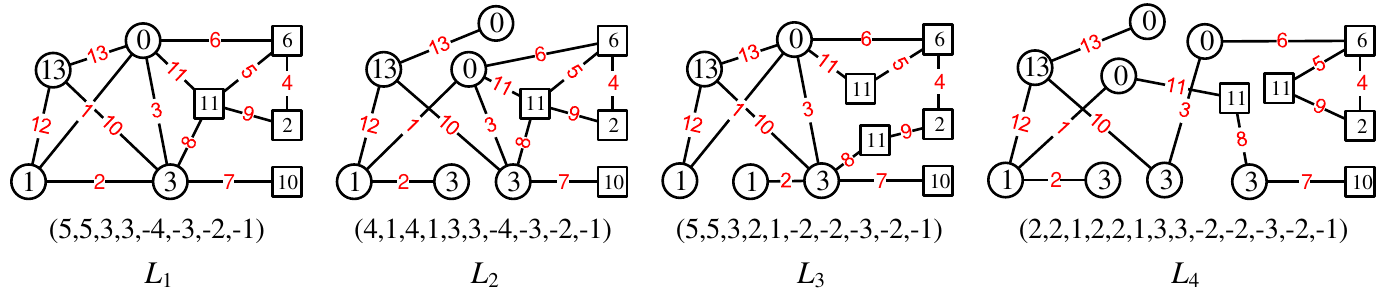}\\
\caption{\label{fig:4-complex-graphs}{\small Four colored complex graphs $L_1,L_2,L_3$ and $L_4$ correspond the same Topcode-matrix $T_{code}$ shown in Eq.(\ref{eqa:Three-complex-graphs-same-Topcode-matrix}) and different degree-sequences and, they admit gracefully total colorings, cited from \cite{Wang-Yao-Su-Wanjia-Zhang-2021-IMCEC}.}}
\end{figure}

\begin{equation}\label{eqa:Three-complex-graphs-same-Topcode-matrix}
\centering
{
\begin{split}
T_{code}=\left(
\begin{array}{ccccccccccccc}
0 & 1 & 0 & 2 & 6 & 0 & 3& 3 & 2 & 3 & 0 & 1 & 0\\
1 & 2 & 3 & 4 & 5 & 6 & 7& 8 & 9 & 10 & 11 & 12 & 13\\
1 & 3 & 3 & 6 & 11 &6 & 10 & 11 & 11 & 13 & 11 & 13 & 13
\end{array}
\right)
\end{split}}
\end{equation}

\begin{thm}\label{thm:complex-degree-sequence-Erdos}
An integer sequence $\textbf{\textrm{d}}=(c_1$, $ c_2, \dots , c_p)$ with $|c_{i}|\geq |c_{i+1}|>0$ for $i\in [1,p-1]$ is the \emph{complex degree-sequence} of a certain complex graph $G$ of $p$ vertices and $q$ edges if and only if $\sum^n_{i=1}|c_{i}|=2q$ and
\begin{equation}\label{eqa:Erdos-Gallai-absolut}
\sum^k_{i=1}|c_{i}|\leq k(k-1)+\sum ^p_{i=k+1}\min\{k,|c_{i}|\},~1\leq k\leq p-1
\end{equation}
\end{thm}
This result is as the same as that distributed by Erd\"{o}s and Gallai in 1960 \cite{Bondy-2008}, and we write $\textrm{C}_{\textrm{deg}}(G)=\textbf{\textrm{d}}$, and call $p=L_{ength}(\textbf{\textrm{d}})$ the \emph{length} of $\textbf{\textrm{d}}$ and $|\textrm{C}_{\textrm{deg}}(G)|=\sum^n_{i=1}|c_i|=2q$ \emph{size} of $\textrm{C}_{\textrm{deg}}(G)$.

\vskip 0.4cm

\textbf{Particular degree-sequences.} Let $\textrm{\textbf{d}}=(a_1, a_2, \dots, a_n)$ be a $n$-rank degree-sequence with $a_{j}\geq 1$ for $j\in [1, n]$, there are particular degree-sequences as follows:

\begin{asparaenum}[\textrm{\textbf{Par}}-1. ]
\item $\overline{\textbf{\textrm{d}}}=(n-1-a_1$, $n-1-a_2, \dots , n-1-a_n)$ is the \emph{complementary degree-sequence} of the degree-sequence $\textbf{\textrm{d}}$.
\item Partition integer sequence $(a_1, a_2, \dots, a_n)$ into $k$ groups of disjoint integer sequence $(a_{i,1},a_{i,1}, \dots, a_{i,n_i})$ with $i\in [1, k]$ with $k\geq 2$, such that $\{a_{i,1},a_{i,1}, \dots, a_{i,n_i}\}\cap \{a_{j,1},a_{j,1}, \dots, a_{j,n_j}\}=\emptyset $ for $i\neq j$, and $\{a_1, a_2, \dots, a_n\}=\bigcup^k_{i=1} \{a_{i,1},a_{i,1}, \dots, a_{i,n_i}\}$. We call these disjoint integer sequences $(a_{i,1},a_{i,1}, \dots, a_{i,n_i})$ with $i\in [1, k]$ to be a \emph{$k$-partition} of the degree-sequence $\textrm{\textbf{d}}$.
\item A \emph{perfect degree-sequence} $\textrm{\textbf{d}}=(a_1, a_2, \dots, a_n)$ holds: Select $a_{j_1}, a_{j_2}, \dots, a_{j_r}$ from $\textrm{\textbf{d}}$ with $r\geq 2$, such that $a_{j_i}\geq a_{j_{i+1}}$ and $a_{j_1}\leq r-1$, and the new integer sequence $(a_{j_1}, a_{j_2}, \dots, a_{j_r})$ is just a degree-sequence. Some examples of perfect degree-sequences are:

\quad (i) $a_1=n-1$, $a_2=\dots=a_n=k$;

\quad (ii) $\textrm{deg}(K_{m,n})$;

\quad (iii) $a_1=a_2=n-2$, $a_3=\dots=a_n=k$.
\item A \emph{unique graph degree-sequence} corresponds one graph only, for example, (5, 2, 2, 1, 1, 1).
\item A\emph{ prime degree-sequence} with prime number $a_i$ for $i\in [1,n]$, for instance, (5, 3, 3, 3, 3, 1, 1, 1, 1).
\item An \emph{Euler degree-sequence} with each $a_i$ is even for $i\in [1,n]$.
\item A \emph{Hamilton degree-sequence} holds $\min\{a_1, a_2, \dots, a_n\}\geq \frac{n}{2}$.
\end{asparaenum}

\begin{prop}\label{thm:unique-graph-degree-sequence}
If a degree-sequence $\emph{\textbf{\textrm{d}}}=(a_1, a_2, \dots, a_n)$ holds $a_1=n-1$ true, and let $\emph{\textbf{\textrm{d}}}\,'=(a_2-1, a_3-1, \dots, a_n-1)$. Then $\emph{\textbf{\textrm{d}}}$ is a unique graph degree-sequence if and only if $\emph{\textbf{\textrm{d}}}\,'$ is a unique graph degree-sequence too.
\end{prop}

\subsubsection{Operations on degree-sequences}

Since a complex graph $G$ has its own complex degree-sequence $\textrm{C}_{\textrm{deg}}(H\overline{\ominus} iG)=(a_1,a_2, \dots ,a_{p_r})\overline{\ominus} (b_1$, $b_2$, $\dots ,b_{p_i})i^2$ with $i^2=-1$ defined in Eq.(\ref{eqa:complex-degree-sequence-express}), we will show some \emph{complex degree-sequence operations} for degree-sequences $\textrm{\textbf{d}}=(a_1, a_2, \dots, a_n)$ with $a_{j}\geq 1$ for $j\in [1, n]$ in this subsection.

Part of operations on degree-sequences shown in the following have been introduced in \cite{Yao-Wang-Ma-Wang-Degree-sequences-2021}. Let $\textrm{\textbf{d}}=(a_1$, $a_2$, $\dots, a_n)$ be a degree-sequence with $a_{j}\geq 1$ for $j\in [1, n]$, and let $\textrm{\textbf{d}}\,'=(a\,'_1, a\,'_2, \dots, a\,'_m)$ be another degree-sequence with $a\,'_{j}\geq 1$ for $j\in [1, m]$.

\begin{asparaenum}[\textrm{\textbf{Dsop}}-1. ]
\item \textbf{Increasing (decreasing) degree component operation}: Increasing a new degree component $k$ to a $n$-rank degree-sequence $\textbf{\textrm{d}}_n=(a_1$, $a_2, \dots , a_n)$ produces a $(n+1)$-rank degree-sequence $\textbf{\textrm{d}}_{n+1}=(a\,'_1$, $a\,'_2, \dots , a\,'_n, k)$, where $a\,'_{j_i}=a_{j_i}+1$ for $i\in [1,k]$, others $a\,'_{j_r}=a_{j_r}$ if $r\not \in [1,k]$, denote $\textbf{\textrm{d}}_{n+1}=\textbf{\textrm{d}}_n\uplus (k)$. For example, (3, 2, 2, 1)$\uplus (4)=$(4, 3, 3, 2, 4), (3, 3, 3, 3)$\uplus (1)=$(4, 3, 3, 3, 1), (3, 3, 3, 3)$\uplus (2)=$(4, 4, 3, 3, 2), (3, 3, 3, 3)$\uplus (3)=$(4, 4, 4, 3, 3), and (3, 3, 3, 3)$\uplus (4)=$(4, 4, 4, 4, 4).

\quad The \emph{inverse} of the increasing degree component operation ``$\uplus$'' is denoted as ``$\uplus^{-1}$'', so $\textbf{\textrm{d}}_n=\textbf{\textrm{d}}_{n+1}\uplus^{-1}(a_i)$ defined by removing a degree component $a_i$ from $\textbf{\textrm{d}}_{n+1}$ and arbitrarily select $a_i$ degree components from the remainder to subtract one from each of them. So, ``$\uplus^{-1}$'' is the \textbf{decreasing degree component operation}.

\item \textbf{Degree-sequence union (subtraction) operation}:
$$\textbf{\textrm{d}}\cup \textbf{\textrm{d}}\,'=(a_1, a_2, \dots, a_n)\cup (a\,'_1, a\,'_2, \dots, a\,'_m)=(b_1,b_2, \dots, b_{n+m}),
$$ such that $a_i,a\,'_j\in \{b_1, b_2, \dots, b_{n+m}\}$ and $|\textbf{\textrm{d}}\cup \textbf{\textrm{d}}\,'|=|\textrm{\textbf{d}}|+|\textrm{\textbf{d}}\,'|$. Thereby, $\textbf{\textrm{d}}$ and $\textbf{\textrm{d}}\,'$ are two \emph{subdegree-sequences} of $\textbf{\textrm{d}}\cup \textbf{\textrm{d}}\,'$.

\quad The degree-sequence subtraction operation ``$-$'' is the inverse of the degree-sequence union operation: $\textbf{\textrm{d}}-\textbf{\textrm{d}}\,'=(c_1,c_2,\dots, c_t)$ such that $c_i\in \{a_1, a_2, \dots, a_n\}$ and $c_i\not \in \{a_1$, $a_2, \dots, a_n\}\cap \{a\,'_1, a\,'_2, \dots, a\,'_m\}$.
\item \textbf{Component-coinciding (component-splitting) operation}:
$$
\textbf{\textrm{d}}\odot_s \textbf{\textrm{d}}\,'=(b_1, b_2, \dots, b_{n+m-s})
$$ with $\min\{m,n\}\geq s\geq 1$, where $b_{i_j}=a_{r_j}+a\,'_{k_j}$ for $j\in [1,s]$, and others $b_{t_j}\in \{a_1,a_2,\dots$, $a_n$, $a\,'_1$, $a\,'_2$, $\dots $, $a\,'_m\}\setminus $ $\{a_{r_j}, a\,'_{k_j}:j\in [1,s]\}$. For example, (\textbf{\textcolor[rgb]{0.00,0.00,1.00}{4}}, \textbf{\textcolor[rgb]{0.00,0.00,1.00}{3}}, 2, 2, 1)$\odot_2$(\textbf{\textcolor[rgb]{0.00,0.00,1.00}{3}}, \textbf{\textcolor[rgb]{0.00,0.00,1.00}{3}}, 2, 2, 2)=(\textbf{\textcolor[rgb]{0.00,0.00,1.00}{7}}, \textbf{\textcolor[rgb]{0.00,0.00,1.00}{6}}, 2, 2, 2, 2, 2, 1), and (4, 3, 2, \textbf{\textcolor[rgb]{0.00,0.00,1.00}{2}}, \textbf{\textcolor[rgb]{0.00,0.00,1.00}{1}})$\odot_2$(3, 3, \textbf{\textcolor[rgb]{0.00,0.00,1.00}{2}}, \textbf{\textcolor[rgb]{0.00,0.00,1.00}{2}}, 2)=(4, \textbf{\textcolor[rgb]{0.00,0.00,1.00}{4}}, \textbf{\textcolor[rgb]{0.00,0.00,1.00}{3}}, 3, 3, 3, 2, 2).

\quad The degree component-splitting operation ``$\odot^{-1}_s$'' is the inverse of the degree component-coinciding operation ``$\odot$'': $\odot^{-1}_s(\textbf{\textrm{d}})=(\textbf{\textrm{d}}_1, \textbf{\textrm{d}}_2)$, such that $\textbf{\textrm{d}}=\textbf{\textrm{d}}_1\odot_s \textbf{\textrm{d}}_2$. For example, $\odot^{-1}_2$(\textbf{\textcolor[rgb]{0.00,0.00,1.00}{7}}, \textbf{\textcolor[rgb]{0.00,0.00,1.00}{6}}, 2, 2, 2, 2, 2, 1)$=\langle $(\textbf{\textcolor[rgb]{0.00,0.00,1.00}{4}}, \textbf{\textcolor[rgb]{0.00,0.00,1.00}{3}}, 2, 2, 1),(\textbf{\textcolor[rgb]{0.00,0.00,1.00}{3}}, \textbf{\textcolor[rgb]{0.00,0.00,1.00}{3}}, 2, 2, 2)$\rangle $.

\item \textbf{Component decomposition operation} ``$\wedge$'': Decompose some components $a_{i_j}$ into $b_{i_j,1},b_{i_j,2},\dots ,b_{i_j,n_j}$ with $j\in [1,s]$, $a_{i_j}=$ $b_{i_j,1}+b_{i_j,2}+\cdots +b_{i_j,n_j}$, and write $D_{co}=$ $\{a_1, a_2$, $\dots, a_n\}\setminus \{a_{i_1},a_{i_2},\dots ,a_{i_s}\}$, the new degree sequence is denoted as $$\wedge(\textbf{\textrm{d}})=(c_1,c_2,\dots,c_{n-s},b_{i_1,1}, b_{i_1,2},\dots,b_{i_1,n_1} , b_{i_2,1},b_{i_2,2},\dots ,b_{i_2,n_2},\dots ,b_{i_s,1},b_{i_s,2},\dots ,b_{i_s,n_s})
 $$ with $c_i\in $ $D_{co}$. For example, $d_1=$(4, 2, 2, 2, 2, 2, 1, 1, 1, 1), $d_2=$(5, 3, 2, 2, 2, 2, 2) and $d_3=$(5, 4, 3, 2, 2, 2), we have $\wedge(d_{i+1})=d_{i}$ for $i=1,2$.

\quad \textbf{Component compound operation} ``$\vee$'' is the inverse operation of the component-degree decomposition operation ``$\wedge$'': $\vee(\textbf{\textrm{d}})=(e_1,e_2,\dots, e_s)$, where $e_{t_j}=a_{t_j,1}+a_{t_j,2}+\cdots +a_{t_j,n_j}$ with $j\in [1,r]$, and $e_i\in \{a_1, a_2, \dots, a_n\}\setminus \{a_{t_j,1},a_{t_j,2},\dots ,a_{t_j,n_j}:j\in [1,r]\}$ for $i\neq t_j$. For example, we have $\vee(d_{i})=d_{i+1}$ for $i=1,2$.

\item \textbf{Degree-sequence direct-sum operation}: For $m\leq n$, our direct-sum operation $\textbf{\textrm{d}}+ \textbf{\textrm{d}}\,'=\textbf{\textrm{d}}\odot_m \textbf{\textrm{d}}\,'$. Especially, $\textbf{\textrm{d}}+ \overline{\textbf{\textrm{d}}}=(n-1,n-1, \dots , n-1)$, where $\overline{\textbf{\textrm{d}}}=(n-1-a_1, n-1-a_2, \dots, n-1-a_n)$ is the \emph{complementary degree-sequence} of $\textbf{\textrm{d}}$.
\end{asparaenum}

\begin{defn}\label{defn:self-contraction-operation}
\cite{Yao-Wang-Su-Jianmin-Xie-2021-Conference} Suppose that two sequences $\textbf{\textrm{d}}=(a_1,a_2,\dots ,a_p)$ and $\textbf{\textrm{d}}\,'=(b_1,b_2, \dots , b_{p-1})$ hold $a_1+a_2=b_1$ and $a_{i}=b_{i-1}$ for $i\in [3,p]$ true, we write $\textbf{\textrm{d}}\,'=\odot_1(\textbf{\textrm{d}})$, and call this operation a \emph{self-contraction operation}, and moreover $\textbf{\textrm{d}}=\wedge_1(\textbf{\textrm{d}}\,')$ is called a \emph{self-splitting degree-sequence} under the \emph{self-splitting operation}.\qqed
\end{defn}

\begin{rem}\label{rem:333333}
Doing the one-order self-contraction operation, we get $\textbf{\textrm{d}}_1=\odot_{1}(\textbf{\textrm{d}})$, and $\textbf{\textrm{d}}_{i+1}=\odot_{1}(\textbf{\textrm{d}}_{i})$ for $i\in [1,k-1]$, very often, we write $\textbf{\textrm{d}}_k=\odot_{k}(\textbf{\textrm{d}})$, so $\textbf{\textrm{d}}_{i}\rightarrow \textbf{\textrm{d}}_{i+1}$, and $\textbf{\textrm{d}}\rightarrow \textbf{\textrm{d}}_{k}$, refer to Definition \ref{defn:degree-sequence-homomorphism}. If $\textbf{\textrm{d}}_k$ is a degree-sequence, but any $\odot_{1}(\textbf{\textrm{d}}_k)$ is no longer a degree-sequence, we get a set $e_{nd}(\textbf{\textrm{d}})$ of degree-sequences such that the self-contraction degree-sequence $\odot_{1}(\textbf{\textrm{d}}^*)$ of any degree-sequence $\textbf{\textrm{d}}^*\in e_{nd}(\textbf{\textrm{d}})$ is no longer a degree-sequence.\paralled
\end{rem}

\begin{defn}\label{defn:degree-coinciding-operation}
\cite{Yao-Wang-Su-Jianmin-Xie-2021-Conference} For two sequences $\textbf{\textrm{d}}=(a_1,a_2, \dots , a_m)$ and $\textbf{\textrm{d}}\,'=(c_1,c_2, \dots , c_{n})$, make a new sequence by adding a degree component $a_i\in \textbf{\textrm{d}}$ with another degree component $c_j\in \textbf{\textrm{d}}\,'$ together, and put other degree components of two sequences $\textbf{\textrm{d}}$ and $\textbf{\textrm{d}}\,'$ into the new sequence as follows
\begin{equation}\label{eqa:degree-coinciding-operation}
{
\begin{split}
\odot\langle \textbf{\textrm{d}},\textbf{\textrm{d}}\,'\rangle=(a_1,a_2, \dots , a_{i-1},a_{i+1}, \dots ,a_m,c_1,c_2, \dots , c_{j-1},c_{j+1}, \dots ,c_{n},a_i+c_j)
\end{split}}
\end{equation} we call the procedure of obtaining $\odot\langle \textbf{\textrm{d}},\textbf{\textrm{d}}\,'\rangle$ a \emph{degree-coinciding operation}. The following degree-sequence
\begin{equation}\label{eqa:degree-joining-operation}
{
\begin{split}
\ominus\langle \textbf{\textrm{d}},\textbf{\textrm{d}}\,'\rangle=(a_1,a_2, \dots , a_{i-1},1+a_i,a_{i+1}, \dots ,a_m,c_1,c_2, \dots , c_{j-1},1+c_j,c_{j+1}, \dots ,c_{n})
\end{split}}
\end{equation} is the result of a \emph{degree-joining operation}.\qqed
\end{defn}

\begin{thm}\label{thm:degree-coinciding-joining-operation}
\cite{Yao-Wang-Su-Jianmin-Xie-2021-Conference} For two sequences $\textbf{\textrm{d}}=(a_1,a_2, \dots , a_m)$ and $\textbf{\textrm{d}}\,'=(c_1,c_2, \dots , c_{n})$, the sequence $\odot\langle \textbf{\textrm{d}},\textbf{\textrm{d}}\,'\rangle$ (resp. $\ominus\langle \textbf{\textrm{d}}_1,\textbf{\textrm{d}}_2\rangle$) is a degree-sequence if and only if both sequences $\textbf{\textrm{d}}$ and $\textbf{\textrm{d}}\,'$ are degree-sequences.
\end{thm}

\begin{thm}\label{thm:degree-coinciding-joining-operation}
\cite{Yao-Wang-Su-Jianmin-Xie-2021-Conference} If there are four degree-sequences $\textbf{\textrm{d}},\textbf{\textrm{d}}_1,\textbf{\textrm{d}}\,',\textbf{\textrm{d}}_2$ holding $\textbf{\textrm{d}}\rightarrow \textbf{\textrm{d}}_1$ and $\textbf{\textrm{d}}\,'\rightarrow \textbf{\textrm{d}}_2$ true, then we have $\odot\langle \textbf{\textrm{d}},\textbf{\textrm{d}}\,'\rangle\rightarrow \odot\langle \textbf{\textrm{d}}_1,\textbf{\textrm{d}}_2\rangle$, and $\ominus\langle \textbf{\textrm{d}},\textbf{\textrm{d}}\,'\rangle\rightarrow \ominus\langle \textbf{\textrm{d}}_1,\textbf{\textrm{d}}_2\rangle$
\end{thm}

\begin{rem}\label{rem:333333}
The operations on degree-sequences introduced above can produce algorithms for constructing large-scale degree-sequences from small-scale degree-sequences.\paralled
\end{rem}

\subsubsection{Homomorphisms of degree-sequences}

\begin{defn}\label{defn:degree-sequence-homomorphism}
\cite{Yao-Zhang-Wang-Su-Integer-Decomposing-2021} The degree-sequence homomorphism ``$\rightarrow $'' is defined as: For two graphs $G$ and $H$ with $G\not \cong H$, if there exists a mapping $\varphi :V(G)\rightarrow V(H)$ such that each edge $uv\in E(G)$ holds $\varphi(u)\varphi(v)\in E(H)$, where the degree-sequence $\textrm{deg}(G)=\textbf{\textrm{d}}$ and the degree-sequence $\textrm{deg}(H)=\textbf{\textrm{d}}\,'$, then we say that $\textbf{\textrm{d}}$ is \emph{degree-sequence homomorphism} to $\textbf{\textrm{d}}\,'$, denoted as $\textbf{\textrm{d}}\rightarrow \textbf{\textrm{d}}\,'$, and $L_{ength}(\textbf{\textrm{d}})>L_{ength}(\textbf{\textrm{d}}\,')$ if $\textbf{\textrm{d}}\neq \textbf{\textrm{d}}\,'$.\qqed
\end{defn}

\begin{defn}\label{defn:graph-splitting-homomorphism}
\cite{Yao-Wang-2106-15254v1} For a graph homomorphism $G\rightarrow H$, a \emph{graph-split homomorphism} $H\rightarrow_{split} T$ is defined by the vertex-splitting operation, in other words, the \emph{graph-split homomorphism} is the \emph{inverse} of a graph homomorphism $G\rightarrow H$. Similarly, a degree-sequence homomorphism $\textbf{\textrm{d}}\rightarrow \textbf{\textrm{d}}\,'$ is accompanied by a \emph{degree-sequence splitting homomorphism} $\textbf{\textrm{d}}\,'\rightarrow_{split} \textbf{\textrm{d}}$.\paralled
\end{defn}

\begin{lem}\label{thm:degree-sequence-homomorphism}
\cite{Yao-Zhang-Wang-Su-Integer-Decomposing-2021} If a simple graph $G$ holds $N(u)\cap N(v)\neq \emptyset$ for any edge $uv\not \in E(G)$ and $u,v\in V(G)$, then its degree-sequence $\textrm{deg}(G)$ is not \emph{degree-sequence homomorphism} to any degree-sequence, except $\textrm{deg}(G)$ itself.
\end{lem}

\begin{defn}\label{defn:Ds-homomorphism-sequence-graph-set}
\cite{Yao-Wang-Su-Jianmin-Xie-2021-Conference} Suppose that $\{\textbf{\textrm{d}}_1,\textbf{\textrm{d}}_2,\dots, \textbf{\textrm{d}}_m\}$ is a degree-sequence set, and there are degree-sequence homomorphisms $\textbf{\textrm{d}}_{k}\rightarrow \textbf{\textrm{d}}_{k+1}$ with $L_{ength}(\textbf{\textrm{d}}_{k})>L_{ength}(\textbf{\textrm{d}}_{k+1})$ for $k\in [1,m-1]$ (refer to Definition \ref{defn:degree-sequence-homomorphism}), we denote them as $H_{omo}\{\textbf{\textrm{d}}_k\}^m_{k=1}$, and call $H_{omo}\{\textbf{\textrm{d}}_k\}^m_{k=1}$ a \emph{degree-sequence homomorphism sequence} (Ds-homomorphism sequence). Since each $\textbf{\textrm{d}}_{k}\in H_{omo}\{\textbf{\textrm{d}}_k\}^m_{k=1}$ corresponds a set $G_{raph}(\textbf{\textrm{d}}_{k})$ of graphs having their degree-sequences to be the same $\textbf{\textrm{d}}_{k}$, so we define a \emph{graph-set homomorphism} $G_{raph}(\textbf{\textrm{d}}_{k})\rightarrow G_{raph}(\textbf{\textrm{d}}_{k+1})$ if each $H_k\in G_{raph}(\textbf{\textrm{d}}_{k})$ with $\textrm{deg}(H_k)=\textbf{\textrm{d}}_{k}$ is graph homomorphism to a graph $H_{k+1}\in G_{raph}(\textbf{\textrm{d}}_{k+1})$ with $\textrm{deg}(H_{k+1})=\textbf{\textrm{d}}_{k+1}$ for $k\in [1,m-1]$.\qqed
\end{defn}

\begin{defn}\label{defn:Cds-matrix-homomorphism}
\cite{Yao-Wang-Su-Jianmin-Xie-2021-Conference} For a $p$-rank degree-sequence $D_{sc}(\textbf{\textrm{d}})$ based on $\textbf{\textrm{d}}=(a_1,a_2,\dots ,a_p)$, if there is another colored degree-sequence matrix $D_{sc}(\textbf{\textrm{d}}\,')=(\textbf{\textrm{d}}\,',g(\textbf{\textrm{d}}\,'))^T$ based on a $(p-1)$-rank degree-sequence $\textbf{\textrm{d}}\,'=(b_1,b_2,\dots ,b_{p-1})$ of a $(p-1,q)$-graph with $b_j\in Z^0\setminus \{0\}$, such that $b_1=a_1+a_2$ and $g(b_1)=f(a_1)=f(a_2)$, as well as $b_{i-1}=a_i$ and $g(b_{i-1})=f(a_i)$ for $i\in [3,p]$. We say that the colored degree-sequence matrix $D_{sc}(\textbf{\textrm{d}})$ is \emph{degree-sequence matrix homomorphism} to the colored degree-sequence matrix $D_{sc}(\textbf{\textrm{d}}\,')$, and write this fact as $D_{sc}(\textbf{\textrm{d}})\rightarrow^{cdsm} D_{sc}(\textbf{\textrm{d}}\,')$ (resp. the inverse $D_{sc}(\textbf{\textrm{d}}\,')\rightarrow^{cdsm}_{split} D_{sc}(\textbf{\textrm{d}})$).

Each $D_{sc}(\textbf{\textrm{d}})$ corresponds a set $G_{raph}(D_{sc}(\textbf{\textrm{d}}))$ of colored graphs having their colored degree-sequence matrices to be the same $D_{sc}(\textbf{\textrm{d}})$. So we have a \emph{graph-set homomorphism} $G_{raph}(D_{sc}(\textbf{\textrm{d}}))\rightarrow G_{raph}(D_{sc}(\textbf{\textrm{d}}\,'))$ since $D_{sc}(\textbf{\textrm{d}})\rightarrow^{cdsm} D_{sc}(\textbf{\textrm{d}}\,')$.\qqed
\end{defn}

\begin{thm}\label{thm:xxxxxx}
\cite{Yao-Wang-Su-Jianmin-Xie-2021-Conference} In a degree-sequence matrix homomorphism $D_{sc}(\textbf{\textrm{d}})\rightarrow^{cdsm} D_{sc}(\textbf{\textrm{d}}\,')$, if the colored degree-sequence matrix $D_{sc}(\textbf{\textrm{d}})=(\textbf{\textrm{d}},f(\textbf{\textrm{d}}))^T$ holds that $f$ induces a graceful (or odd-graceful) coloring of a simple graph $G$ with $\textrm{deg}(G)=\textbf{\textrm{d}}$, then $D_{sc}(\textbf{\textrm{d}}\,')=(\textbf{\textrm{d}}\,',g(\textbf{\textrm{d}}\,'))^T$ holds that $g$ induces a graceful (resp. odd-graceful) coloring of a simple graph $H$ with $\textrm{deg}(H)=\textbf{\textrm{d}}\,'$, such that $G\rightarrow H$.
\end{thm}

\subsubsection{Degree-sequence lattices}

\textbf{Degree-sequence linear-sum operation ``$\Sigma $''}: Let $\textbf{\textrm{d}}^*=(\textbf{\textrm{d}}\,^1_n,\textbf{\textrm{d}}\,^2_n,\dots,\textbf{\textrm{d}}\,^m_n)$ be a \emph{degree-sequence base}, where each component $\textbf{\textrm{d}}\,^k_n=(b_{k,1},b_{k,2},\dots ,b_{k,n})$ with $k\in [1,m]$ is a degree-sequence. A new integer sequence $(y_{j})^n_{j=1}=(y_1,y_2,\dots ,y_n)=\sum^m_{k=1} \lambda_k\textbf{\textrm{d}}\,^k_n$ with $\lambda_k\in Z^0$, where each component $y_{j}=\sum^m_{k=1}\lambda_kb_{k,j}$ with $j\in [1,n]$. We claim that the integer sequence $(y_{j})^n_{j=1}$ is just a degree-sequence by the vertex-coinciding operation on the graphs having degree-sequences $d\,^k_n$ with $k\in [1,m]$. We call the following set
\begin{equation}\label{eqa:degree-sequence-lattice}
\textbf{\textrm{L}}(Z^0(\Sigma)\textbf{\textrm{d}}^*)=\left \{\sum^m_{k=1} \lambda_k\textbf{\textrm{d}}\,^k_n:\lambda_k\in Z^0,\textbf{\textrm{d}}\,^k_n\in \textbf{\textrm{d}}^*\right \}
\end{equation} a \emph{degree-sequence lattice}.

Let $\textbf{\textrm{Cd}}^*=(\textbf{\textrm{Cd}}\,^1_n,\textbf{\textrm{Cd}}\,^2_n,\dots,\textbf{\textrm{Cd}}\,^m_n)$ be a \emph{complex degree-sequence base}, where each component $\textbf{\textrm{Cd}}\,^k_n=a^k_s\overline{\ominus} c^k_ti^2$ is a complex degree-sequence, where $a^k_s=(a_{k,1},a_{k,2},\dots ,a_{k,s})$ and $c^k_t=(c_{k,1}$, $c_{k,2},\dots ,c_{k,t})$ with $s+t=n$. A \emph{complex degree-sequence lattice} is
\begin{equation}\label{eqa:complex-degree-sequence-lattice}
\textbf{\textrm{L}}(Z^0(\Sigma)\textbf{\textrm{Cd}}^*)=\left \{\sum^m_{j=1} \beta_j\textbf{\textrm{Cd}}\,^j_n:\beta_j\in Z^0,\textbf{\textrm{Cd}}\,^j_n\in \textbf{\textrm{Cd}}^*\right \}
\end{equation} with each complex sequence $(w_1,w_2,\dots ,w_n)=\sum^m_{j=1} \beta_j\textbf{\textrm{Cd}}\,^j_n$, where each component
$${
\begin{split}
w_j&=\beta_j(a^j_s\overline{\ominus} c^j_ti^2)=(\beta_ja_{j,1},\beta_ja_{j,2},\dots ,\beta_ja_{j,s})\overline{\ominus} (\beta_jc_{j,1},\beta_jc_{j,2},\dots ,\beta_jc_{j,t})i^2\\
&=(\beta_ja_{j,1},\beta_ja_{j,2},\dots ,\beta_ja_{j,s},-\beta_jc_{j,1},-\beta_jc_{j,2},\dots ,-\beta_jc_{j,t})
\end{split}}
$$ is a \emph{complex degree-sequence}.

\begin{defn} \label{defn:degree-coinciding-ds-lattices}
\cite{Yao-Wang-Su-Jianmin-Xie-2021-Conference} Let $\textbf{\textrm{I}}=(\textbf{\textrm{d}}_1,\textbf{\textrm{d}}_2,\dots ,\textbf{\textrm{d}}_m)$ be a \emph{degree-sequence vector}, where degree-sequences $\textbf{\textrm{d}}_1,\textbf{\textrm{d}}_2,\dots ,\textbf{\textrm{d}}_m$ are independent from each other under the degree-coinciding operation ``$\odot$'' defined in Definition \ref{defn:degree-coinciding-operation}. For $\sum^m_{k=1} a_k\geq 1$ with $a_k\in Z^0$, we have a set $\{a_1\textbf{\textrm{d}}_1,a_2\textbf{\textrm{d}}_2,\dots ,a_m\textbf{\textrm{d}}_m\}= \{a_k\textbf{\textrm{d}}_k\}^m_{k=1}$, and do the degree-coinciding operation ``$\odot$'' to the elements of the set $\{a_k\textbf{\textrm{d}}_k\}^m_{k=1}$ such that each of $a_k$ degree-sequences $\textbf{\textrm{d}}_k$ appears in some $\odot\langle \textbf{\textrm{d}}_k,\textbf{\textrm{d}}_j\rangle$ if $a_k\neq 0$ and $a_j\neq 0$. The resulting degree-sequences are collected into a set $\odot\langle a_k\textbf{\textrm{d}}_k\rangle^m_{k=1}$. We call the following set
\begin{equation}\label{eqa:555555}
\textbf{\textrm{L}}(Z^0\odot \langle \textbf{\textrm{I}}\rangle )=\{\odot\langle a_k\textbf{\textrm{d}}_k\rangle^m_{k=1}:~a_k\in Z^0,\textbf{\textrm{d}}_k\in \textbf{\textrm{I}}\}
\end{equation} a \emph{degree-coincided degree-sequence lattice}.\qqed
\end{defn}

\begin{defn} \label{defn:degree-joining-ds-lattices}
\cite{Yao-Wang-Su-Jianmin-Xie-2021-Conference} Let $\textbf{\textrm{I}}=(\textbf{\textrm{d}}_1,\textbf{\textrm{d}}_2,\dots ,\textbf{\textrm{d}}_m)$ be a \emph{degree-sequence vector}, where degree-sequences $\textbf{\textrm{d}}_1,\textbf{\textrm{d}}_2,\dots ,\textbf{\textrm{d}}_m$ are independent from each other under the degree-coinciding operation ``$\ominus$'' defined in definition \ref{defn:degree-coinciding-operation}. For $\sum^m_{k=1} a_k\geq 1$ with $a_k\in Z^0$, we have a set $\{a_1\textbf{\textrm{d}}_1,a_2\textbf{\textrm{d}}_2,\dots ,a_m\textbf{\textrm{d}}_m\}= \{a_k\textbf{\textrm{d}}_k\}^m_{k=1}$, and do the degree-coinciding operation ``$\ominus$'' to the elements of the set $\{a_k\textbf{\textrm{d}}_k\}^m_{k=1}$ such that each of $a_s$ degree-sequences $\textbf{\textrm{d}}_s$ appears in some $\ominus\langle \textbf{\textrm{d}}_s,\textbf{\textrm{d}}_i\rangle$ if $a_s\neq 0$ and $a_i\neq 0$. The resulting degree-sequences are collected into a set $\ominus\langle a_k\textbf{\textrm{d}}_k\rangle^m_{k=1}$. We call the following set
\begin{equation}\label{eqa:555555}
\textbf{\textrm{L}}(Z^0\ominus \langle \textbf{\textrm{I}}\rangle )=\{\ominus\langle a_k\textbf{\textrm{d}}_k\rangle^m_{k=1}:~a_k\in Z^0,\textbf{\textrm{d}}_k\in \textbf{\textrm{I}}\}
\end{equation} a \emph{degree-joined degree-sequence lattice}.\qqed
\end{defn}

\begin{defn} \label{defn:degree-sequence-lattices-homomorphism}
\cite{Yao-Wang-Su-Jianmin-Xie-2021-Conference} If a degree-sequence vector $\textbf{\textrm{I}}=(\textbf{\textrm{d}}_1,\textbf{\textrm{d}}_2$, $\dots ,\textbf{\textrm{d}}_m)$ in a degree-coinciding degree-sequence lattice $\textbf{\textrm{L}}(Z^0\odot \langle \textbf{\textrm{I}}\rangle )$ (resp. $\textbf{\textrm{L}}(Z^0\ominus \langle \textbf{\textrm{I}}\rangle )$) is \emph{degree-sequence vector homomorphism} to another degree-sequence vector $\textbf{\textrm{I}}\,'=(\textbf{\textrm{d}}\,'_1,\textbf{\textrm{d}}\,'_2,\dots ,\textbf{\textrm{d}}\,'_m)$ in another degree-coinciding degree-sequence lattice $\textbf{\textrm{L}}(Z^0\odot \langle \textbf{\textrm{I}}\,'\rangle )$ (resp. $\textbf{\textrm{L}}(Z^0\ominus \langle \textbf{\textrm{I}}\,'\rangle )$), such that each degree-sequence homomorphism $\textbf{\textrm{d}}_k\rightarrow \textbf{\textrm{d}}\,'_k$ with $k\in [1,m]$ holds true, then $\textbf{\textrm{L}}(Z^0\odot \langle \textbf{\textrm{I}}\rangle )$ is \emph{homomorphism} to $\textbf{\textrm{L}}(Z^0\odot \langle \textbf{\textrm{I}}\,'\rangle )$ (resp. $\textbf{\textrm{L}}(Z^0\ominus \langle \textbf{\textrm{I}}\rangle )$ is \emph{homomorphism} to $\textbf{\textrm{L}}(Z^0\ominus \langle \textbf{\textrm{I}}\,'\rangle )$), we call
\begin{equation}\label{eqa:ds-lattice-homomorphism}
\textbf{\textrm{L}}(Z^0\odot \langle \textbf{\textrm{I}}\rangle )\rightarrow \textbf{\textrm{L}}(Z^0\odot \langle \textbf{\textrm{I}}\,'\rangle )\quad \big (\textrm{resp. }\textbf{\textrm{L}}(Z^0\ominus \langle \textbf{\textrm{I}}\rangle )\rightarrow \textbf{\textrm{L}}(Z^0\ominus \langle \textbf{\textrm{I}}\,'\rangle )\big )
\end{equation} a \emph{degree-sequence lattice homomorphism}.\qqed
\end{defn}

\begin{rem}\label{rem:ABC-conjecture}
For an operation $(\ast)\in \{\ominus, \odot\}$, each $\textbf{\textrm{L}}(Z^0(\ast) \langle \textbf{\textrm{I}}\rangle )$ corresponds a set $G_{raph}(\textbf{\textrm{L}}(Z^0(\ast) \langle \textbf{\textrm{I}}\rangle ))$ of graphs having their degree-sequences to be in $\textbf{\textrm{L}}(Z^0(\ast) \langle \textbf{\textrm{I}}\rangle )$. So we have a \emph{graph-set homomorphism}
\begin{equation}\label{eqa:555555}
G_{raph}(\textbf{\textrm{L}}(Z^0(\ast) \langle \textbf{\textrm{I}}\rangle ))\rightarrow G_{raph}(\textbf{\textrm{L}}(Z^0(\ast) \langle \textbf{\textrm{I}}\,'\rangle ))
\end{equation} where $G_{raph}(\textbf{\textrm{L}}(Z^0(\ast) \langle \textbf{\textrm{I}}\,'\rangle ))$ is the set of graphs having their degree-sequences to be in $\textbf{\textrm{L}}(Z^0(\ast) \langle \textbf{\textrm{I}}\,'\rangle )$ if the degree-sequence lattice homomorphism $\textbf{\textrm{L}}(Z^0(\ast) \langle \textbf{\textrm{I}}\rangle )\rightarrow \textbf{\textrm{L}}(Z^0(\ast) \langle \textbf{\textrm{I}}\,'\rangle )$ holds true.\paralled
\end{rem}

\begin{problem}\label{qeu:444444}
\textbf{Build} up the connection between degree-sequence lattices and traditional lattices, and plant some problems of traditional lattices into degree-sequence lattices.
\end{problem}

\subsubsection{Every-zero Cds-matrix groups of degree-sequences}

\begin{defn} \label{defn:Cds-matrix-group-definition}
\cite{Yao-Wang-Su-Jianmin-Xie-2021-Conference, Yao-Mu-Sun-Zhang-Wang-Su-Ma-IAEAC-2018} Let $D^1_{sc}(\textbf{\textrm{d}})=(\textbf{\textrm{d}},f_1(\textbf{\textrm{d}}))^T$ be a Cds-matrix defined in Definition \ref{defn:colored-degree-sequence-matrix}, and $M=\max\{a_i\in \textbf{\textrm{d}}:i\in [1,p]\}$ be a positive integer. By $f_j(a_s)=f_1(a_s)+j~(\bmod~M)$ for each degree component $a_s\in \textbf{\textrm{d}}$ with $s\in [1,p]$, we define $f_{j}(\textbf{\textrm{d}})=f_1(\textbf{\textrm{d}})+j$ $(\bmod~M)$ for each Cds-matrix $D^j_{sc}(\textbf{\textrm{d}})=(\textbf{\textrm{d}},f_j(\textbf{\textrm{d}}))^T$. A set of Cds-matrices $D^j_{sc}(\textbf{\textrm{d}})$ is denoted by $\big \{D^{(M)}_{sc}(\textbf{\textrm{d}}),\oplus \big \}$. For a fixed $D^k_{sc}(\textbf{\textrm{d}})\in \big \{D^{(M)}_{sc}(\textbf{\textrm{d}}),\oplus \big \}$ being consisted of a graph set $\{D^k_{sc}(\textbf{\textrm{d}})\}^M_{k=1}$ and an algebraic operation ``$\oplus$'', we have
\begin{equation}\label{eqa:Cds-matrix-group-operation}
f_i(\textbf{\textrm{d}})+f_j(\textbf{\textrm{d}})-f_k(\textbf{\textrm{d}})=f_{\lambda}(\textbf{\textrm{d}})
\end{equation} with $\lambda=i+j-k~(\bmod~M)$ and a \emph{preappointed zero} $f_{k}(\textbf{\textrm{d}})\in \{D^{(M)}_{sc}(\textbf{\textrm{d}}),\oplus \}$, where the formula (\ref{eqa:Cds-matrix-group-operation}) is computed by
\begin{equation}\label{eqa:Cds-matrix-group-operation11}
f_i(a_s)+f_j(a_s)-f_k(a_s)=f_{\lambda}(a_s)
\end{equation} for each degree component $a_s\in \textbf{\textrm{d}}$ with $s\in [1,p]$. We call the set $\big \{D^{(M)}_{sc}(\textbf{\textrm{d}}),\oplus \big \}$ an \emph{every-zero Cds-matrix group}.\qqed
\end{defn}

\section{Randomly topological codes}

The theory of random graphs lies at the intersection between graph theory and probability theory. Random graph is the general term to refer to probability distributions over graphs in mathematics. Random graphs may be described simply by a probability distribution, or by a random process which generates them. From a mathematical perspective, random graphs are used to answer questions about the properties of typical graphs. Its practical applications are found in all areas in which complex networks need to be modeled -- many random graph models are thus known, mirroring the diverse types of complex networks encountered in different areas. In a mathematical context, random graph refers almost exclusively to the Erd\"{o}s-R\'{e}nyi random graph model. In other contexts, any graph model may be referred to as a random graph (Ref. https://encyclopedia.thefreedictionary.com/Random+graph).

\subsection{Randomly adding-edge-removing operation}

\begin{defn} \label{defn:col-pre-e-add-remov-operation}
$^*$ Let $G$ be a $(p,q)$-graph with $2q<p(p-1)$.

(i) The $(p,q)$-graph obtained by removing an edge $xy$ of $E(G)$ and adding a new edge $uv\not \in E(G)$ is denoted as $G+uv-xy$, we call the process of obtaining the graph $H=G+uv-xy$ \emph{adding-edge-removing operation} ($\pm e$-operation), and say that $G$ is \emph{$\pm e$-graph homomorphism} to $H$, denoted as $G\rightarrow _{\pm e}H$.

(ii) Suppose that the $(p,q)$-graph $G$ admits a $W$-type coloring $f$, if the $(p,q)$-graph $H=G+uv-xy$ admits a $W$-type coloring $g$ too, such that $f(V(G))=g(V(H))$ if $f(V(G))\neq \emptyset$ and $f(E(G)\setminus \{xy\})=g(E(H)\setminus \{uv\})$ if $f(E(G))\neq \emptyset$, we call the process of obtaining $H$ \emph{coloring-preserved adding-edge-removing operation} (coloring-preserved $\pm e$-operation), and $G$ is \emph{coloring-preserved $\pm e$-graph homomorphism} to $H$, denoted as $G\rightarrow ^{col}_{\pm e}H$.\qqed
\end{defn}

\begin{example}\label{exa:graceful-coloring}
In Fig.\ref{fig:add-edge-remove-edge}, a $(12,11)$-tree $T$ admits a graceful coloring $f$ since two vertices of $T$ were colored with the same color $4$, and each $(12,11)$-graph $T_i$ made by doing the adding-edge-removing operation to the tree $T$ admits a graceful coloring $f_i$ induced by $f$ for $i\in [1,3]$. Thereby, we have $T\rightarrow ^{col}_{\pm e}T_i$ for $i\in [1,3]$ by Definition \ref{defn:col-pre-e-add-remov-operation}. Another $(12,11)$-tree $H_1$ shown in Fig.\ref{fig:add-edge-remove-edge} is obtained from the $(12,11)$-tree $T$ by doing the adding-edge-removing operation, and $H_1$ admits a graceful coloring $h_1$, and moreover each $(12,11)$-graph $H_j$ based on the adding-edge-removing operation produces the $(12,11)$-graph $H_{j+1}$ admitting a graceful coloring $h_{j+1}$ for $j\in [1,3]$. So we get the following graph homomorphism chain
\begin{equation}\label{eqa:555555}
H_1\rightarrow ^{col}_{\pm e}H_2\rightarrow ^{col}_{\pm e}H_3\rightarrow ^{col}_{\pm e}H_4\rightarrow ^{col}_{\pm e}H_3\rightarrow ^{col}_{\pm e}H_2\rightarrow ^{col}_{\pm e}H_1
\end{equation} according to Definition \ref{defn:col-pre-e-add-remov-operation}.
\end{example}

\begin{figure}[h]
\centering
\includegraphics[width=16.4cm]{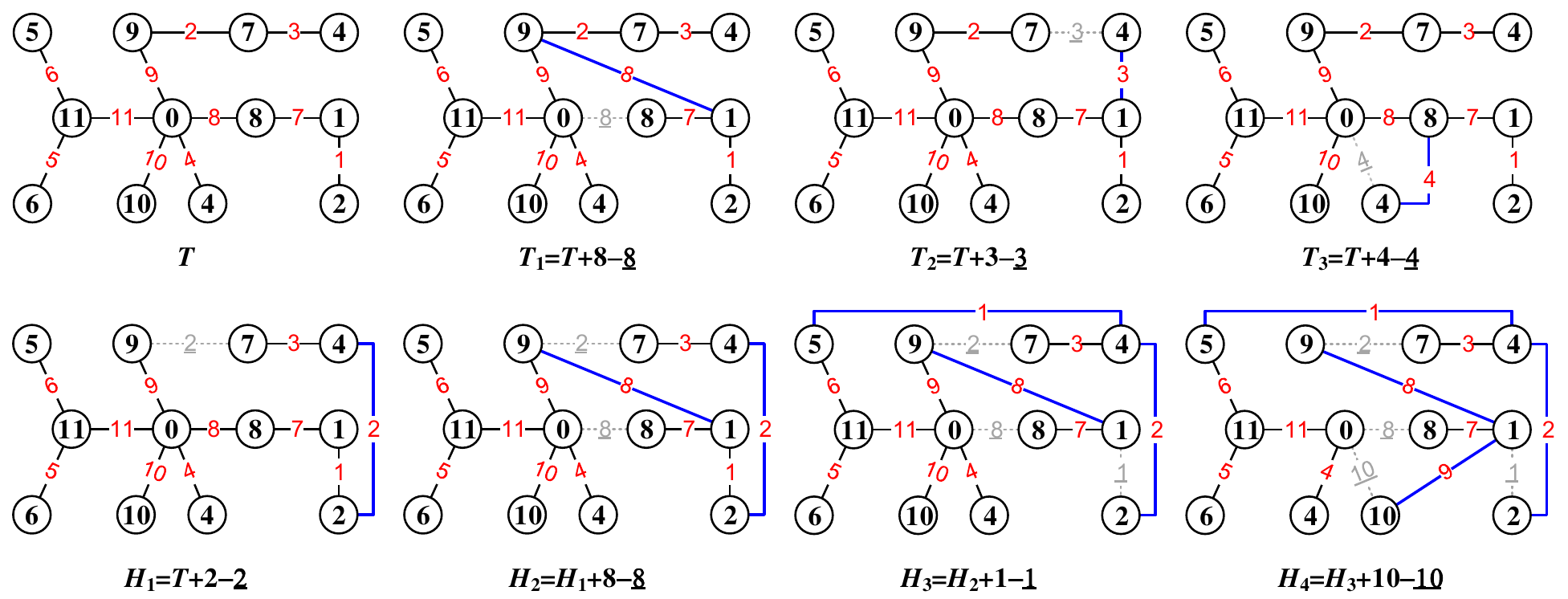}\\
\caption{\label{fig:add-edge-remove-edge}{\small Examples for understanding the randomly adding-edge-removing operation.}}
\end{figure}

\begin{example}\label{exa:strongly-graceful-perfect-matchings}
In Fig.\ref{fig:add-remove-edge-strong-coloring}, a $(16,15)$-tree $G$ admits a strongly graceful labeling $g$, each $(16,15)$-graph $G_i$ admits a strongly graceful labeling $g_i$ for $i\in [1,5]$, and they hold the following graph homomorphism chain
\begin{equation}\label{eqa:555555}
G\rightarrow ^{col}_{\pm e}G_1\rightarrow ^{col}_{\pm e}G_2\rightarrow ^{col}_{\pm e}G_3\rightarrow ^{col}_{\pm e}G_4\rightarrow ^{col}_{\pm e}G_5\rightarrow ^{col}_{\pm e}G_4\rightarrow ^{col}_{\pm e}G_3\rightarrow ^{col}_{\pm e}G_2\rightarrow ^{col}_{\pm e}G_1\rightarrow ^{col}_{\pm e}G
\end{equation} Notice that there are perfect matchings in the tree $G$ and each graph $G_i$.
\end{example}

\begin{figure}[h]
\centering
\includegraphics[width=16.4cm]{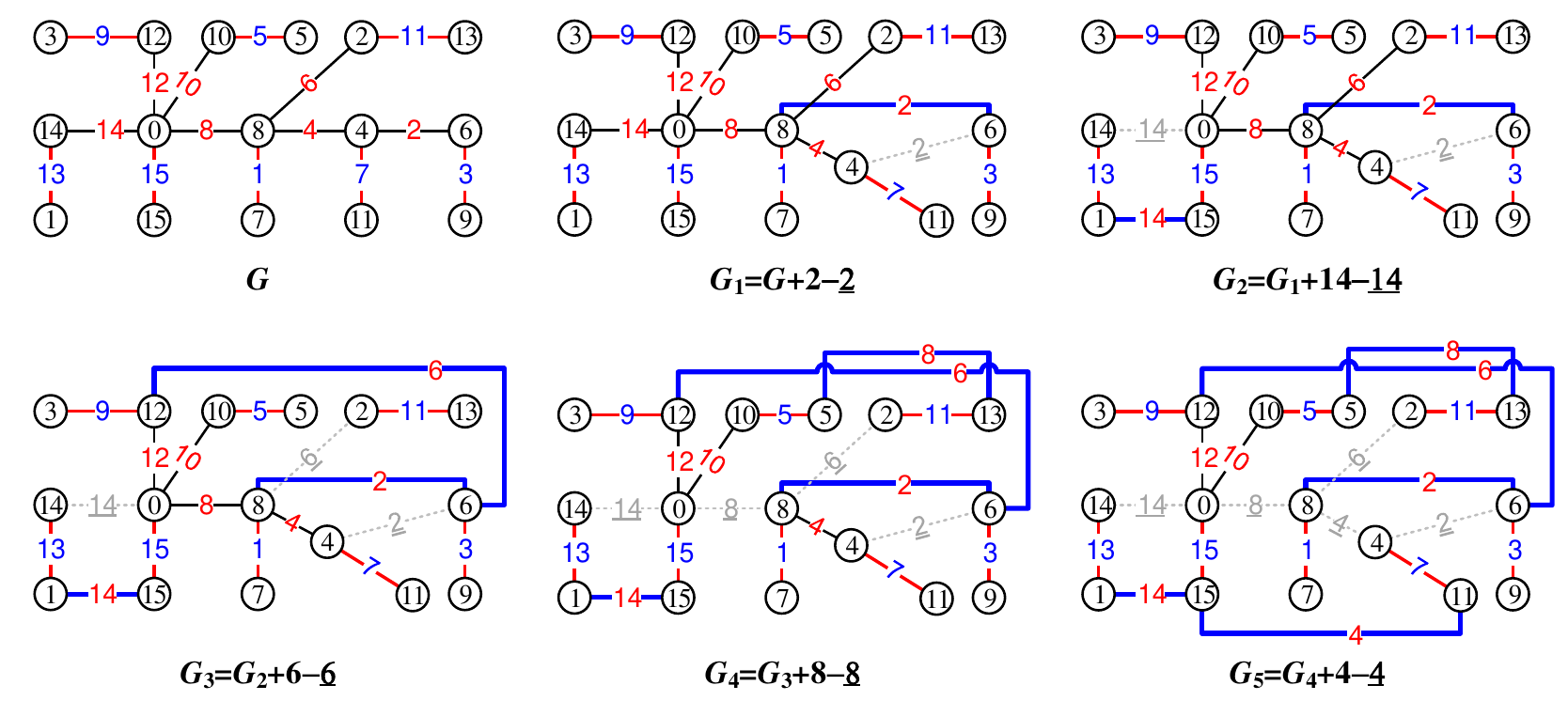}\\
\caption{\label{fig:add-remove-edge-strong-coloring}{\small Keeping the strongly graceful labeling in the adding-edge-removing operation.}}
\end{figure}

\begin{defn}\label{defn:col-pre-e-add-remov-graphs-set}
$^*$ According to Definition \ref{defn:col-pre-e-add-remov-operation}, we define a \emph{coloring-preserved $\pm e$-graph set} $F_{\pm e}\langle G,f\rangle$ in the following way: For any $(p,q)$-graph $T\in F_{\pm e}\langle G,f\rangle$, there are $(p,q)$-graphs $T_i\in F_{\pm e}\langle G,f\rangle$ for $i\in [1,m]$, such that $G\rightarrow ^{col}_{\pm e}T_1$ with $T_1=G+u_1v_1-x_1y_1$ for $x_1y_1\in E(G)$ and $u_1v_1\not\in E(G)$, $T_j\rightarrow ^{col}_{\pm e}T_{j+1}$ with $T_{j+1}=T_j+u_jv_j-x_jy_j$ for $x_jy_j\in E(T_j)$ and $u_jv_j\not\in E(T_j)$ for $j\in [1,m-1]$, and $T_m\rightarrow ^{col}_{\pm e}T$ holding $T=T_m+u_mv_m-x_my_m$ for $x_my_m\in E(T_m)$ and $u_mv_m\not\in E(T_m)$.

Similarly, we have a \emph{$\pm e$-graph set} $F_{\pm e}\langle G\rangle$ by $G\rightarrow _{\pm e}H$ with no any coloring.\qqed
\end{defn}

\begin{thm}\label{thm:666666}
$^*$ Let $F_{\pm e}\langle G,f\rangle$ be a coloring-preserved $\pm e$-graph set based on a $(p,q)$-graph $G$ admitting a $W$-type coloring $f$. Then each graph $G^*\in F_{\pm e}\langle G,f\rangle$ admits a $W$-type coloring $f^*$, so we have $F_{\pm e}\langle G^*,f^*\rangle=F_{\pm e}\langle G,f\rangle$.
\end{thm}

\begin{problem}\label{qeu:col-pre-e-add-remov-graphs-set}
Let $G$ be a $(p,q)$-graph. Using Definition \ref{defn:col-pre-e-add-remov-operation} and Definition \ref{defn:col-pre-e-add-remov-graphs-set}, we have the following questions:
\begin{asparaenum}[\textrm{Ear}-1.]
\item As known, $F_{\pm e}\langle G,f\rangle\subset F_{\pm e}\langle G\rangle$, since $F_{\pm e}\langle G,f\rangle$ is restricted by coloring. But do we have $\bigcup _{f}F_{\pm e}\langle G,f\rangle=F_{\pm e}\langle G\rangle$ for all possible $W$-type colorings of $G$? If it is so, \textbf{find} the smallest number $m$ of different $W$-type colorings such that $\bigcup ^m_{i=1}F_{\pm e}\langle G,f_i\rangle=F_{\pm e}\langle G\rangle$.
\item \textbf{Is} there $F_{\pm e}\langle G,f\rangle=F_{\pm e}\langle G,g\rangle$ for two $W$-type colorings $f$ and $g$ of $G$? For example, $f$ and $g$ both are two graceful labelings of $G$.
\item For a public-key $G\in F_{\pm e}\langle G,f\rangle$, does there exist a private-key $H\in F_{\pm e}\langle G,g\rangle$ for two $W$-type colorings $f$ and $g$ of $G$, such that $H\cong G$?
\item A graph $\Gamma(G,f)$ with its vertex set $F_{\pm e}\langle G,f\rangle$, and $\Gamma(G,f)$ has an edge $H\,'H\,''$ with the ends $H\,'$ and $H\,''$ of $F_{\pm e}\langle G,f\rangle$ if and only if there exists a graph homomorphism $H\,'\rightarrow ^{col}_{\pm e}H\,''$ in $F_{\pm e}\langle G,f\rangle$, clearly, $\Gamma(G,f)$ is connected. \textbf{Characterize} $\Gamma(G,f)$.
\item \textbf{Find} particular subsets $S_j\subset F_{\pm e}\langle G,f\rangle$, so that each graph $H\in S_j$ is a tree, or a connected graph, or a path, \emph{etc.}
\item If $G$ is a tree of $p$ vertices having a perfect matching, and $f$ is a strongly graceful labeling of $G$, does $F_{\pm e}\langle G,f\rangle$ contain a path of $p$ vertices having a perfect matching?
\end{asparaenum}
\end{problem}

\subsection{Randomly leaf-adding algorithms}

In this section, the sentence ``RANDOMLY-LEAF-adding algorithm'' is abbreviated as a short sentence ``RLA-algorithm''.

\begin{defn} \label{defn:twin-k-d-harmonious-labelings}
\cite{Yao-Zhang-Sun-Mu-Sun-Wang-Wang-Ma-Su-Yang-Yang-Zhang-2018arXiv} A $(p,q)$-graph $G$ admits a $(k,d)$-labeling $f$, and another $(p\,',q\,')$-graph $H$ admits another $(k,d)$-labeling $g$. If $(X_0\cup S_{q-1,k,d})\setminus f(V(G)\cup E(G))=g(V(H)\cup E(H))$ with $X_0=\{0,d,2d, \dots ,(q-1)d\}$ and $S_{q-1,k,d}=\{k,k+d, \dots ,k+(p-1)d\}$, then $g$ is called a \emph{complementary $(k,d)$-labeling} of $f$, and both $\langle f, g\rangle $ a \emph{twin $(k,d)$-labeling}.\qqed
\end{defn}

\subsubsection{RLA-algorithm for the odd-graceful labeling}

The RLA-algorithm for the odd-graceful labeling was introduced in \cite{Zhou-Yao-Chen-Tao2012} first.

\vskip 0.4cm

\noindent \textbf{RLA-algorithm for the odd-graceful labeling} \cite{Zhou-Yao-Chen-Tao2012}:

\textbf{Input:} A connected bipartite $(p,q)$-graph $G$ admitting a set-ordered odd-graceful labeling $f$.

\textbf{Output:} A connected bipartite $(p+m,q+m)$-graph $G^*$ admitting an odd-graceful labeling, where $G^*$ is the result of adding randomly $m$ leaves to $G$.

\textbf{Step 1.} By the definition of a set-ordered odd-graceful labeling, so the vertex set $V(G)=X\cup Y$ with $X\cap Y=\emptyset$, where $X=\{x_1,x_2,\dots ,x_s\}$ and $Y=\{y_1,y_2,\dots ,y_t\}$ with $s+t=p=|V(G)|$. Since $\max f(X)<\min f(Y)$, we have
$$
0=f(x_1)<f(x_2)<\cdots <f(x_s)<f(y_1)<f(y_2)<\cdots <f(y_t)=2q-1
$$
so each $f(x_i)$ for $i\in[1,s]$ is even, and each $f(y_j)$ for $j\in[1,t]$ is odd, and
$$
f(E(G))=\{f(x_iy_j)=f(y_j)-f(x_i):x_iy_j\in E(G)\}=[1,2q-1]^o.
$$

\textbf{Step 2.} Adding randomly $a_i$ new leaves $u_{i,k}$ to each vertex $x_i$ by adding new edges $x_iu_{i,k}$ for $k\in [1,a_i]$ and $i\in[1,s]$, and adding randomly $b_j$ new leaves $v_{j,r}$ to each vertex $y_j$ by adding new edges $y_jv_{j,r}$ for $r\in [1,b_j]$ and $j\in[1,t]$, it is allowed that some $a_i=0$ or some $b_j=0$. The resultant graph is denoted as $G^*$.

\textbf{Step 3.} Define a labeling $g$ of $G^*$ in the following way: Color edges $x_iu_{i,k}$ by setting $g(x_1u_{1,k})=2k-1$ for $k\in [1,a_1]$, and
\begin{equation}\label{eqa:555555}
g(x_{i+1}u_{i+1,r})=2r-1+\sum ^{i}_{l=1}2a_l,~r\in [1,a_{i+1}],~i\in [1,s-1]
\end{equation}
and $g(x_{s}u_{s,r})=2r-1+\sum ^{s-1}_{l=1}2a_l$ for $s\in [1,a_s]$, as well as the last edge $x_{s}u_{s,a_s}$ is colored with $g(x_{s}u_{s,a_s})=-1+\sum ^{s}_{l=1}2a_l$.

Let $A=\sum ^{s}_{l=1}2a_l$, we color these edges with $g(y_tv_{t,k})=A-1+2k$ for $k\in [1,b_t]$, $g(y_tv_{t,b_t})=A-1+2b_t$, and

$g(y_{t-1}v_{t-1,k})=A-1+2b_t+2k$ for $k\in [1,b_{t-1}]$, $g(y_{t-1}v_{t-1,b_{t-1}})=A-1+2b_t+2b_{t-1}$, as well as

\begin{equation}\label{eqa:555555}
g(y_{t-j}v_{t-j,k})=2k+A-1+\sum ^{t}_{l=t-j+1}2b_l,~k\in [1,b_{t-j}],~j\in [1,t-2]
\end{equation} the last edge $y_{1}v_{1,b_1}$ is colored with $g(y_{1}v_{1,b_1})=A+B-1$, where $B=\sum ^{t}_{l=1}2b_l$.

\textbf{Step 4.} Color each vertex $x_i\in X$ with $g(x_i)=f(x_i)$ for $i\in [1,s]$; color added leaves $u_{i,k}$ with $g(u_{i,k})=f(x_i)+g(x_iu_{i,k})$ for $k\in [1,a_i]$ and $i\in [1,s]$; and recolor vertices $y_j\in Y$ with $g(y_j)=f(y_j)+2m-1$ for $j\in [1,t]$, where $m=\frac{A+B}{2}=\sum ^{s}_{l=1}a_l+\sum ^{t}_{l=1}b_l$; and color each vertex $v_{j,r}$ with $g(v_{j,r})=g(y_j)-g(y_{j}v_{j,r})$ for $r\in [1,b_{j}]$ and $j\in [1,t]$.

\textbf{Step 5.} Return the odd-graceful labeling $g$ of $G^*$.

\vskip 0.4cm

\begin{example}\label{exa:8888888888}
An example for illustrating the RLA-algorithm for the odd-graceful labeling shown in Fig.\ref{fig:zhou-adding-leaves-algorithms}. The connected bipartite $(10,11)$-graph $S$ admits a set-ordered odd-graceful labeling $f$, we add $m$ leaves to $S$ randomly, the resultant graph is denoted as $S_1=S+E_1$, where $E_1$ is the set of the added edges, and $S_2$ is obtained by coloring the added edges of $E_1$, we color the vertices of $S_2$ to get the desired connected bipartite $S_3$ admitting an odd-graceful labeling. Here, we need to know all set-ordered odd-graceful labelings of the connected bipartite $(10,11)$-graph $S$ if it is as a topological public-key, unfortunately, no polynomial algorithm for finding all set-ordered odd-graceful labelings of a connected bipartite $(p,q)$-graph. Determining the number of ways of adding leaves to $S$ to obtain $S_1=S+E_1$ is not a slight job.
\end{example}

\begin{figure}[h]
\centering
\includegraphics[width=16.4cm]{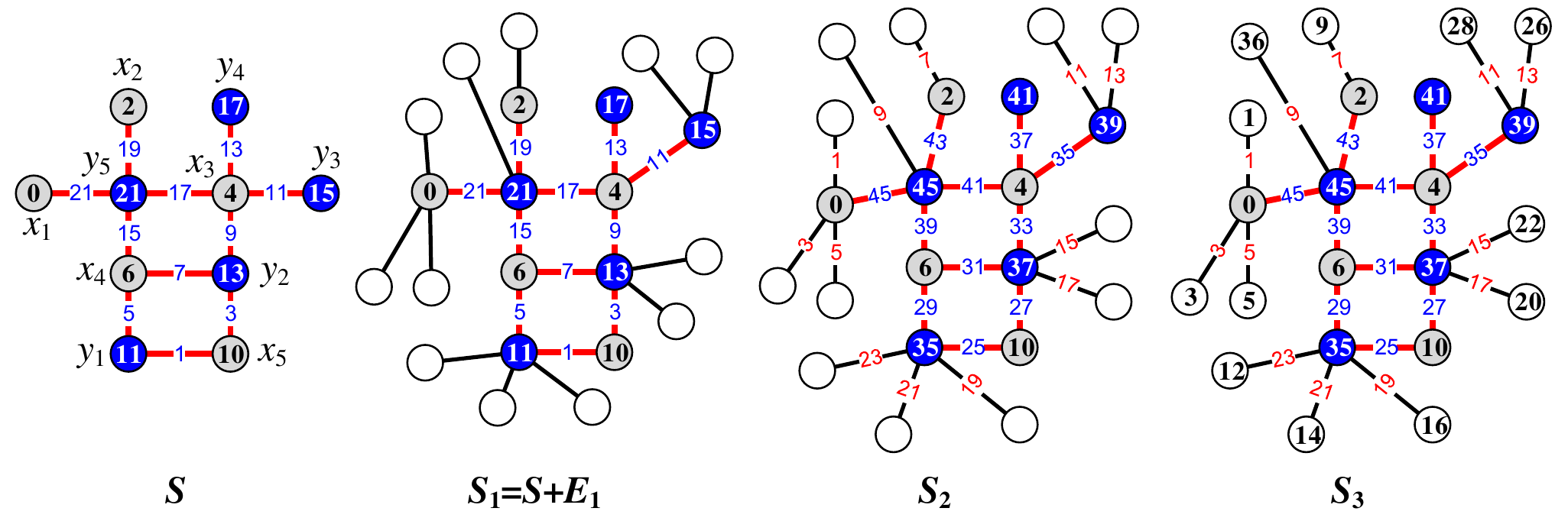}\\
\caption{\label{fig:zhou-adding-leaves-algorithms}{\small An example for illustrating the RLA-algorithm for the odd-graceful labeling.}}
\end{figure}

By the RLA-algorithm for the odd-graceful labeling, we have a result below:

\begin{thm}\label{thm:adding-leaves-bipartite-graph-odd-graceful}
\cite{Zhou-Yao-Chen-Tao2012} If a connected bipartite $(p,q)$-graph $G$ admits a set-ordered odd-graceful labeling, then adding $m$ leaves randomly to $G$ produces a new connected bipartite $(p+m,q+m)$-graph $G^*$ admitting an odd-graceful labeling.
\end{thm}

\begin{rem}\label{rem:adding-leaves-graph-odd-graceful-complexity}
There is a more complexity of adding leaves randomly in Theorem \ref{thm:adding-leaves-bipartite-graph-odd-graceful}: First, we select $k$ vertices from a connected bipartite $(p,q)$-graph $G$ for adding $m$ leaves to them, so we have $A^k_p=p(p-1)\cdots (p-k+1)$ selections in total. A positive integer $m$ is decomposed into a group of $k$ parts $m_1,m_2,\cdots ,m_k$ such that $m=m_1+m_2+\cdots +m_k$ with $m_i>0$. Suppose that there is $P(m,k)$ groups of such $k$ parts. For a group of $k$ parts $m_1,m_2,\cdots ,m_k$, let $m_{i_1},m_{i_2},\cdots ,m_{i_k}$ be a permutation of $m_1,m_2,\cdots ,m_k$, so we have such $k!$ permutations. Then we have
\begin{equation}\label{eqa:c3xxxxx}
A_{leaf}(G,m)=\sum^m_{k=1}A^k_p\cdot P(m,k)\cdot k!=\sum^m_{k=1} P(m,k)\cdot p!
\end{equation} graphs made by adding $m$ leaves to $G$ in total, where $P(m,k)=\sum ^k_{r=1}P(m-k,r)$. Here, computing $P(m,k)$ can be transformed into finding the number $A(m,k)$ of solutions of equation $m=\sum ^k_{i=1}ix_i$. There is a recursive formula
\begin{equation}\label{eqa:c3xxxxx}
A(m,k)=A(m,k-1)+A(m-k,k)
\end{equation}
with $0 \leq k\leq m$. As known, it is not easy to compute the exact value of $A(m,k)$, the authors in \cite{Shuhong-Wu-Accurate-2007} and \cite{WU-Qi-qi-2001} computed exactly
$${
\begin{split}
A(m,6)=&\biggr\lfloor \frac{1}{1036800}(12m^5 +270m^4+1520m^3-1350m^2-19190m-9081)+\\
&\frac{(-1)^m(m^2+9m+7)}{768}+\frac{1}{81}\left[(m+5)\cos \frac{2m\pi}{3}\right ]\biggr\rfloor .
\end{split}}$$

In a topological authentication $\textbf{T}_{\textbf{a}}\langle\textbf{X},\textbf{Y}\rangle$ defined in Eq.(\ref{eqa:topo-authentication-11}) of Definition \ref{defn:topo-authentication-multiple-variables}, if the \emph{public-key} is a connected bipartite $(p,q)$-graph $G$ admitting a set-ordered odd-graceful labeling, then we have $A_{leaf}(G,m)$ connected bipartite $(p+m,q+m)$-graphs admitting odd-graceful labelings to be \emph{private-keys} by the help of Theorem \ref{thm:adding-leaves-bipartite-graph-odd-graceful}.\paralled
\end{rem}

\begin{defn} \label{defn:zhou-odd-elegant-labeling}
\cite{Zhou-Yao-Chen2013} An \emph{odd-elegant labeling} $f$ of a $(p,q)$-graph $G$ holds $f(V(G))\subset [0,2q-1]$, $f(u)\neq f(v)$ for distinct $u,v\in V(G)$, and $f(E(G))=\{f(uv)=f(u)+f(v)~(\bmod~2q):uv\in E(G)\}=[1,2q-1]^o$.\qqed
\end{defn}

\begin{thm}\label{thm:666666}
\cite{Zhou-Yao-Chen2013} If a connected bipartite $(p,q)$-graph $G$ admits a set-ordered odd-graceful labeling, then adding $m$ leaves randomly to $G$ produces a new connected bipartite $(p+m,q+m)$-graph $G^*$ admitting an odd-elegant labeling defined in Definition \ref{defn:zhou-odd-elegant-labeling}.
\end{thm}

\begin{figure}[h]
\centering
\includegraphics[width=16.4cm]{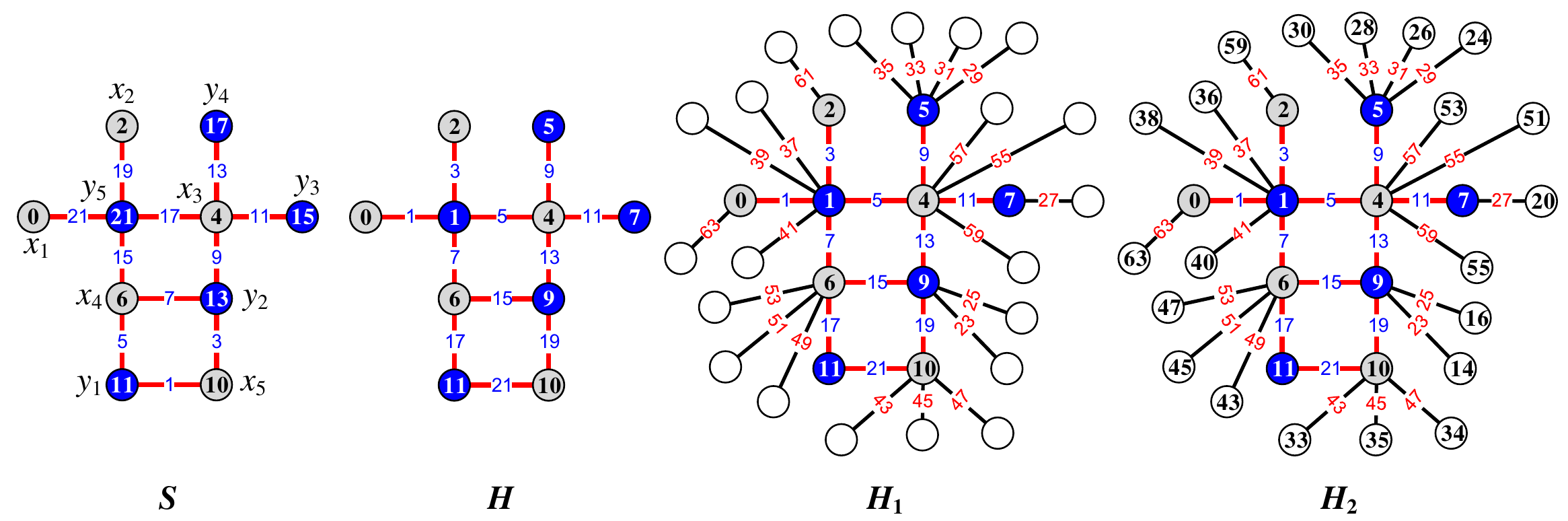}\\
\caption{\label{fig:zhou-add-leaf-odd-elegant}{\small An example from the RLA-algorithm for the odd-elegant labeling defined in Definition \ref{defn:zhou-odd-elegant-labeling}, where $S$ admits a set-ordered odd-graceful labeling, $H$ admits an odd-elegant labeling, $H_1$ is obtained by adding randomly leaves to $G$, and $H_2$ admits an odd-elegant labeling.}}
\end{figure}

\begin{conj}\label{conj:edge-odd-elegant-graph-base}
$^*$ For a bipartite connected graph $G$ admitting an odd-elegant labeling $g$, there exists a group of connected graphs $H_k$ admitting a labeling $h_k$ for $k\in [1,r]$, such that
$$g(V(G))\cup h_1(V(H_1))\cup h_2(V(H_2))\cup \cdots \cup h_r(V(H_r))=[0,M]
$$ with
$$M\leq 2\max\{|E(G)|,|E(H_1)|,|E(H_2)|,\dots ,~|E(H_r)|\},~g(E(G))=[1,2|E(G)|-1]^o
$$ and each edge color set
$$h_k(E(H_k))=\{h_k(xy)=h_k(x)+h_k(y)~(\bmod~2|E(H_k)|):xy\in E(H_k)\}=[1,2|E(H_k)|-1]^o
$$ for $k\in [1,r]$.
\end{conj}

\begin{figure}[h]
\centering
\includegraphics[width=16.4cm]{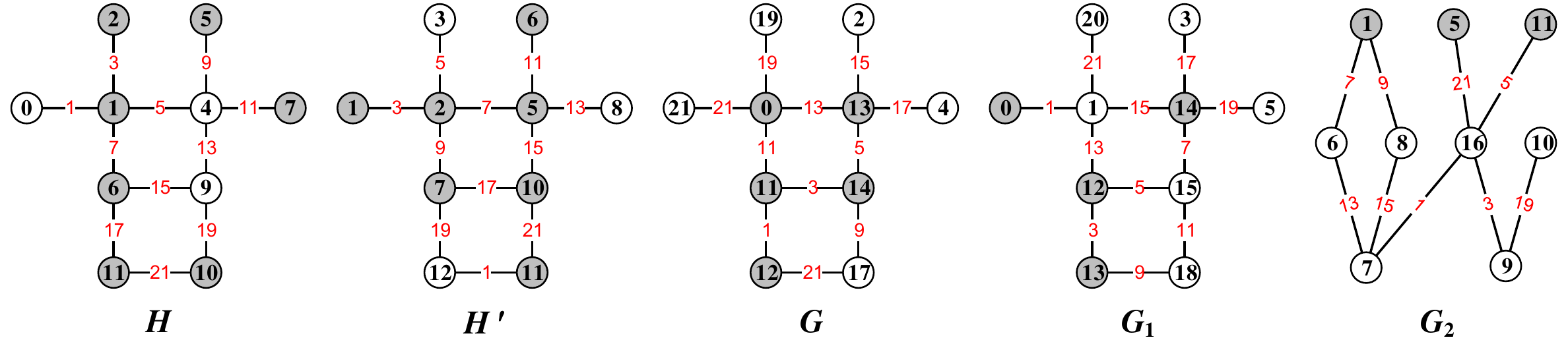}\\
\caption{\label{fig:44-edge-odd-elegant-graph-base}{\small Both graphs $H$ and $G$ admit two odd-elegant labelings, where $H$ has an edge-odd-elegant graph $H\,'$ to form a matching; $G$ has an edge-odd-elegant graph group $(G_1,G_2)$ to form a matching.}}
\end{figure}

Finding all odd-elegant labelings of a bipartite connected graph in Conjecture \ref{conj:edge-odd-elegant-graph-base} is not easy, even very difficult, see Fig.\ref{fig:44-edge-odd-elegant-graph-base} for understanding this conjecture.

We present a result for affirming partly Conjecture \ref{conj:edge-odd-elegant-graph-base} as follows.

\begin{thm}\label{thm:tree-edge-odd-elegant-graph-base}
$^*$ If a $(p,q)$-tree $G$ admits a set-ordered graceful labeling, then there exists at least a connected $(p\,',q\,')$-graph $H$ admitting a labeling $\beta$, and $G$ admits an odd-elegant labeling $\alpha$ such that $\alpha(V(G))\cup \beta(V(H))=[0,M]$ with $M\leq 2\max\{q,q\,'\}$, as well as two edge color sets $\alpha(E(G))=[1,2q-1]^o$ and
$$\beta(E(H))=\{\beta(xy)=\beta(x)+\beta(y)~(\bmod~2q\,'):xy\in E(H)\}=[1,2q\,'-1]^o$$
\end{thm}
\begin{proof}By the definition of a set-ordered graceful labeling $f$ of a $(p,q)$-tree $G$, we have $V(G)=X\cup Y$ with $X\cap Y=\emptyset$, where $X=\{x_1,x_2,\dots ,x_s\}$ and $Y=\{y_1,y_2,\dots ,y_t\}$ with $s+t=p=|V(G)|$. Since $\max f(X)<\min f(Y)$, we have
$$
0=f(x_1)<f(x_2)<\cdots <f(x_s)<f(y_1)<f(y_2)<\cdots <f(y_t)=q.
$$
Since $G$ is a tree, we have $f(x_i)=i-1$ for $i\in[1,s]$, and $f(y_j)=s-1+j$ for $j\in[1,t]$, and
$${
\begin{split}
f(E(G))&=\{f(x_iy_j)=f(y_j)-f(x_i):x_iy_j\in E(G)\}\\
&=\{f(x_iy_j)=s+j-i:x_iy_j\in E(G)\}=[1,q]
\end{split}}
$$

\textbf{Step 1.} We define a new labeling $g$ of $G$ by setting $g(y_j)=f(y_j)$ for $y_j\in Y$, $g(x_i)=f(x_{s-i+1})$ for $x_i\in X$, so the induced edge color
$${
\begin{split}
g(x_iy_j)&=g(x_i)+g(y_j)=f(x_{s-i+1})+f(x_{i})+f(y_j)-f(x_{i})\\
&=s-i+1-1+i-1+f(x_iy_j)=s-1+f(x_iy_j)
\end{split}}
$$ so $g(x_iy_j)=[s,s+q-1]$.

\textbf{Step 2.} We define another new labeling $h$ of $G$ as $h(y_j)=2g(y_j)-1=2(s-1+j)-1=2(s+j)-3$ for $y_j\in Y$, $h(x_i)=2g(x_{s-i+1})=2(s-i+1-1)=2(s-i)$ for $x_i\in X$. Each edge $x_iy_j$ has an induced color $h(x_iy_j)$ defined by
$$h(x_iy_j)=h(x_i)+h(y_j)=2(2s+j-i)-3=2s+2(s+j-i)-3
$$ so the edge color set
$${
\begin{split}
h(E(G))&=\{h(x_iy_j)=h(x_i)+h(y_j):x_iy_j\in E(G)\}\\
&=[2s-1,2q+2s-3]^o~(\bmod~2q)\\
&=[2s-1,2q-1]^o\cup [1,2s-3]^o\\
&=[1,2q-1]^o.
\end{split}}
$$ Thereby, we claim that $h$ is a set-ordered odd-elegant labeling, since $\max h(X)=2(s-1)<2s-1=\min h(Y)$.

\textbf{Step 3.} We define a new labeling $h^*$ of $G$ by $h^*(w)=h(w)+1$ for $w\in V(G)$, immediately, we get the edge color set
$$h^*(E(G))=\{h^*(x_iy_j)=h^*(x_i)+h^*(y_j)~(\bmod~2q):x_iy_j\in E(G)\}=[1,2q]^o
$$ The proof of this theorem is complete, since $H=G$ admits a labeling $h^*$ holding $h^*(E(G))=[1,2q]^o$ and $h(V(G))\cup h^*(V(H))=[0,2q]$. See examples shown in Fig.\ref{fig:edge-odd-elegant-graph-base} (b) and (c).
\end{proof}

\begin{figure}[h]
\centering
\includegraphics[width=16.4cm]{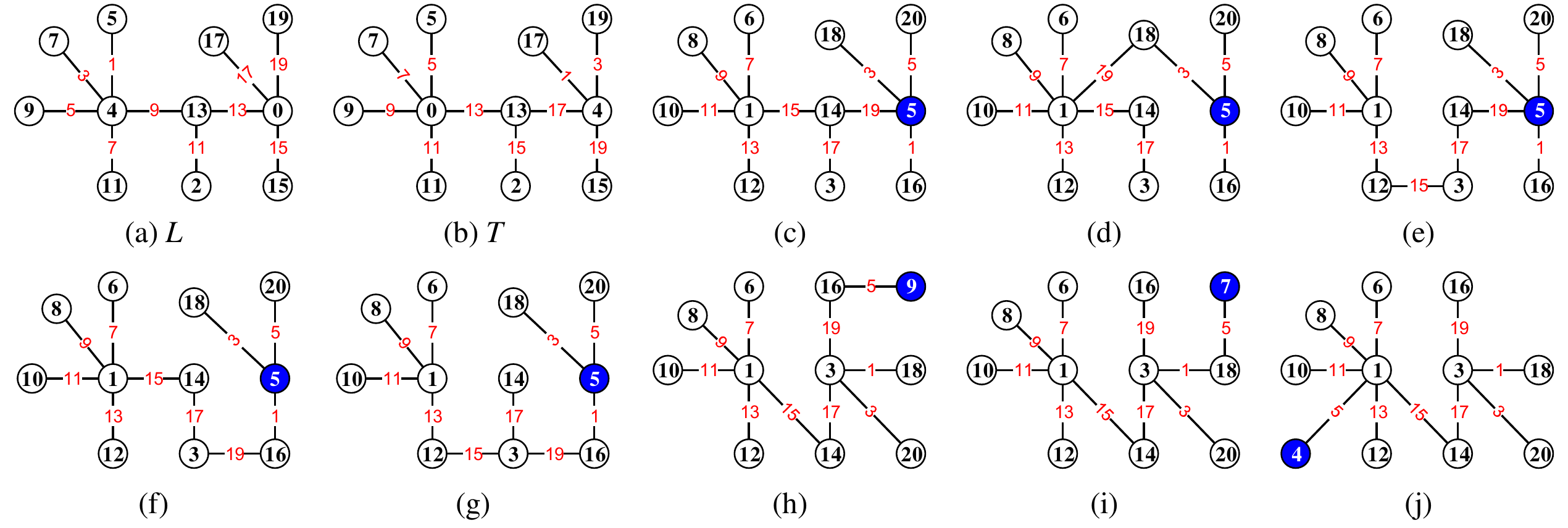}\\
\caption{\label{fig:edge-odd-elegant-graph-base}{\small (a) A tree $L$ admitting an odd-graceful labeling $f$; (b) a tree $T$ admitting an odd-elegant labeling induced by $f$; (c)-(j) form an \emph{edge-odd-elegant graph base} of $T$.}}
\end{figure}

\begin{example}\label{exa:randomly-adding-leaves}
In Fig.\ref{fig:random-cut-one-into-public-private}, we pick up a number-based string $enc(m)$ from a colored connected graph $G\,'$, and then get a public-key $G$ by erasing the colors of vertices and edges of $G\,'$. For topological authentication, one must find a private-key $H$, and coloring the vertices and edges of $H$ produces a colored connected graph $H\,'$ such that $T_{\textrm{au}}=G\,'[\ominus ^W_4]H\,'$ is the desired topological authentication, the deciphering number-based string $dec(n)$ is obtained from the colored connected graph $H\,'$. This is an example for illustrating Definition \ref{defn:gracefully-total-sequence-coloring}.
\end{example}

\begin{figure}[h]
\centering
\includegraphics[width=16.4cm]{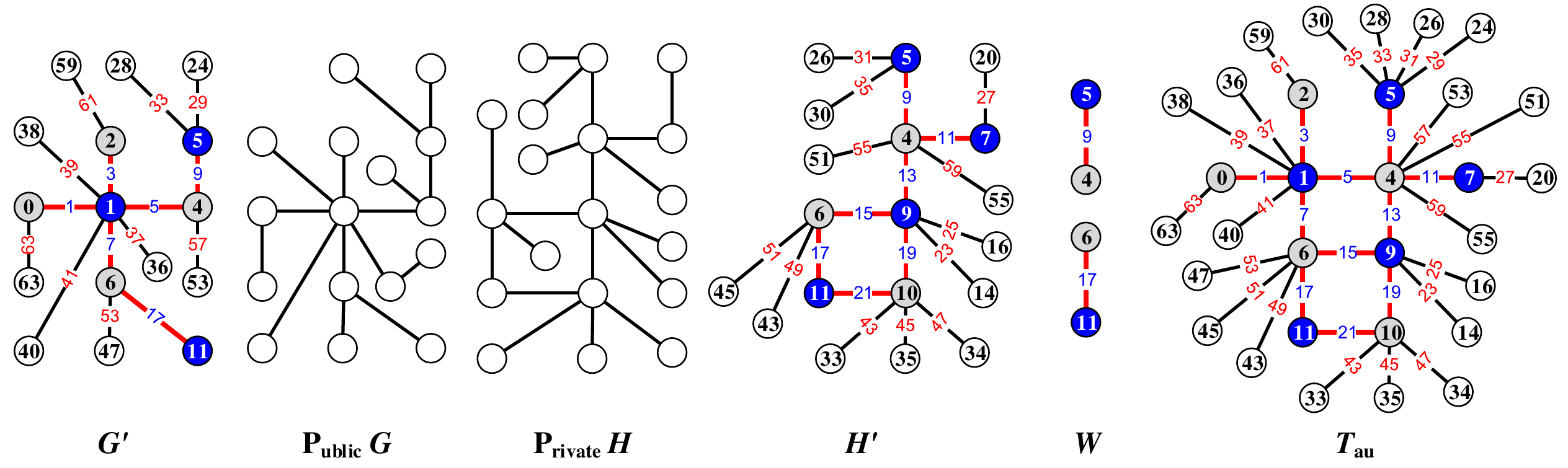}\\
\caption{\label{fig:random-cut-one-into-public-private}{\small A public-key $G$ is obtained by adding randomly leaves, also, a private-key $H$ is obtained by adding randomly leaves.}}
\end{figure}

\subsubsection{RLA-algorithm for the $(k,d)$-harmonious labeling}

\noindent \textbf{$^*$ RLA-algorithm for the $(k,d)$-harmonious labeling}:

\textbf{Input:} A connected bipartite $(p,q)$-graph $G$ admitting a $(k,d)$-harmonious labeling $f$.

\textbf{Output:} A connected bipartite $(p+m,q+m)$-graph $G^*$ admitting a $(k,d)$-harmonious labeling, where $G^*$ is obtained by adding randomly $m$ leaves to $G$.

\textbf{Step 1.} We have $V(G)=X\cup Y$ with $X\cap Y=\emptyset$, where $X=\{x_1,x_2,\dots ,x_s\}$ and $Y=\{y_1$, $y_2,\dots ,y_t\}$ with $s+t=p=|V(G)|$, and without losing generality, let

1. $0=f (x_1)<f (x_2)<\cdots <f (x_s)$, and each vertex color $f (x_i)\in \{0,d,2d,\dots , (q-1)d\}$;

2. $f (y_j)-k<f (y_{j+1})-k$, and each vertex color $f (y_j)\in \{k,k+d,k+2d,\dots , k+(q-1)d\}$ for $j\in [1,t-1]$.

By the definition of a $(k,d)$-harmonious labeling introduced in Definition \ref{defn:3-parameter-labelings}, each edge $x_iy_j$ holds
$$f(x_iy_j)-k~(\bmod~qd)=f(y_j)+f(x_i)-k~(\bmod~qd)\in \{0,d, \dots,(q-1)d\}
$$ and the induced edge color set
\begin{equation}\label{eqa:555555}
f(E(G))=\{f(x_iy_j):x_iy_j\in E(G)\}=\{k,k+d, \dots,k+(q-1)d\}
\end{equation}

\textbf{Step 2.} Adding randomly $a_i$ new leaves $x_{i,1},x_{i,2},\dots ,x_{i,a_i}$ to vertex $x_i$ by adding new edges $x_ix_{i,k}$ for $k\in [1,a_i]$ and $i\in[1,s]$, and next adding randomly $b_j$ new leaves $y_{j,1},y_{j,2},\dots ,y_{j,b_j}$ to vertex $y_j$ by adding new edges $y_jy_{j,r}$ for $r\in [1,b_j]$ and $j\in[1,t]$, notice that there are some $a_i=0$ or some $b_j=0$. Let $A=\sum ^{s}_{i=1}a_i$ and $B=\sum ^{t}_{j=1}b_j$, and $m=A+B$. The resultant graph is denoted as $G^*$, and $G^*$ is a connected bipartite $(p+m,q+m)$-graph.

\textbf{Step 3.} Define a coloring $g$ for $G^*$ as follows:

\textbf{Step 3.1} Set $g(w)=f(w)$ for $w\in X$, $g(z)=f(z)+md$ for $z\in Y$. So, each edge $x_iy_j$ is colored as
$${
\begin{split}
g(x_iy_j)&=f(y_j)+md+f(x_i)\in S^*\\
&=\{k+(j_0+m)d,k+(j_0+m+1)d, \dots,k+(j_0+q+m-1)d\}
\end{split}}
$$ thus,
$${
\begin{split}
S^*~(\bmod~qd)=&\{k,k+d,k+2d,\dots ,k+(j_0-1)d, k+(j_0+m)d,\\
&k+(j_0+m+1)d, \dots,k+(q+m-1)d\}
\end{split}}
$$

We will find an edge set $\{k+j_0d, k+(j_0+m)d, \dots ,k+(j_0+m-1)d\}$ in the following.

\textbf{Step 3.2} Color edges $y_{t}y_{t,r}$ of $G^*$ with $g(y_{t}y_{t,r})=k+(j_0-1)d+rd$ for $r\in [1,b_t]$, and $g(y_{t}y_{t,b_t})=k+(b_t+j_0-1)d$. For other edges $y_jy_{j,r}$ of $G^*$, we have

\begin{equation}\label{eqa:555555}
g(y_{t-j}y_{t-j,r})=k+rd+d\sum ^{t}_{l=t-j}(b_l+j_0-1),~r\in [1,b_{t-j}],~j\in [1,t-2]
\end{equation}
Thereby,
$$
g(y_1y_{1,r})=k+rd+(t-1)(j_0-1)d+d\sum ^{t}_{l=2}b_l,~r\in [1,b_1]
$$ and $g(y_1y_{1,b_t})=k+(j_0-1+B)d$.

\textbf{Step 3.3} Color edges $x_sx_{s,r}$ of $G^*$ by $g(x_sx_{s,r})=k+(j_0-1+B)d+rd$ for $r\in [1,a_s]$, and moreover edges $x_{s-i}x_{s-i,r}$ of $G^*$ are colored as follows
\begin{equation}\label{eqa:555555}
g(x_{s-i}x_{s-i,r})=k+rd+(j_0-1+B)d+d\sum ^{s}_{l=s-i+1}a_l,~r\in [1,a_{s-i}],~i\in [1,s-1]
\end{equation} until to
$$g(x_1x_{1,r})=k+rd+(j_0-1+B)d+d\sum ^{s}_{l=2}a_l,~r\in [1,a_1]
$$ and
$$g(x_1x_{1,a_1})=k+(j_0-1+A+B)d=k+(j_0-1+m)d
$$ and moreover
$$
g(E^*)=\{k+j_0d, k+(j_0+m)d, \dots ,k+(j_0+m-1)d\}
$$ where $E^*=E(G^*)\setminus E(G)$.

\textbf{Step 3.4} Color leaves $x_{i,k}$ with $g(x_{i,k})=g(x_ix_{i,k})-g(x_i)$ for $k\in [1,a_i]$ and $i\in [1,s]$; and color leaves $y_{j,r}$ with
$$
g(y_{j,r})+g(y_j)-k=g(y_jy_{j,r})-k~(\bmod~(q+m)d)
$$ for $r\in [1,b_j]$ and $j\in [1,t]$.

\textbf{Step 4.} Return the $(k,d)$-harmonious labeling $g$ of the connected bipartite $(p+m,q+m)$-graph $G^*$ by the above algorithm and $g(w)\neq g(z)$ for any pair of vertices $w$ and $z$ of $G^*$ (see examples shown in Fig.\ref{fig:k-d-harmonious-11} and Fig.\ref{fig:k-d-harmonious-22}).

\vskip 0.4cm

\begin{figure}[h]
\centering
\includegraphics[width=16.4cm]{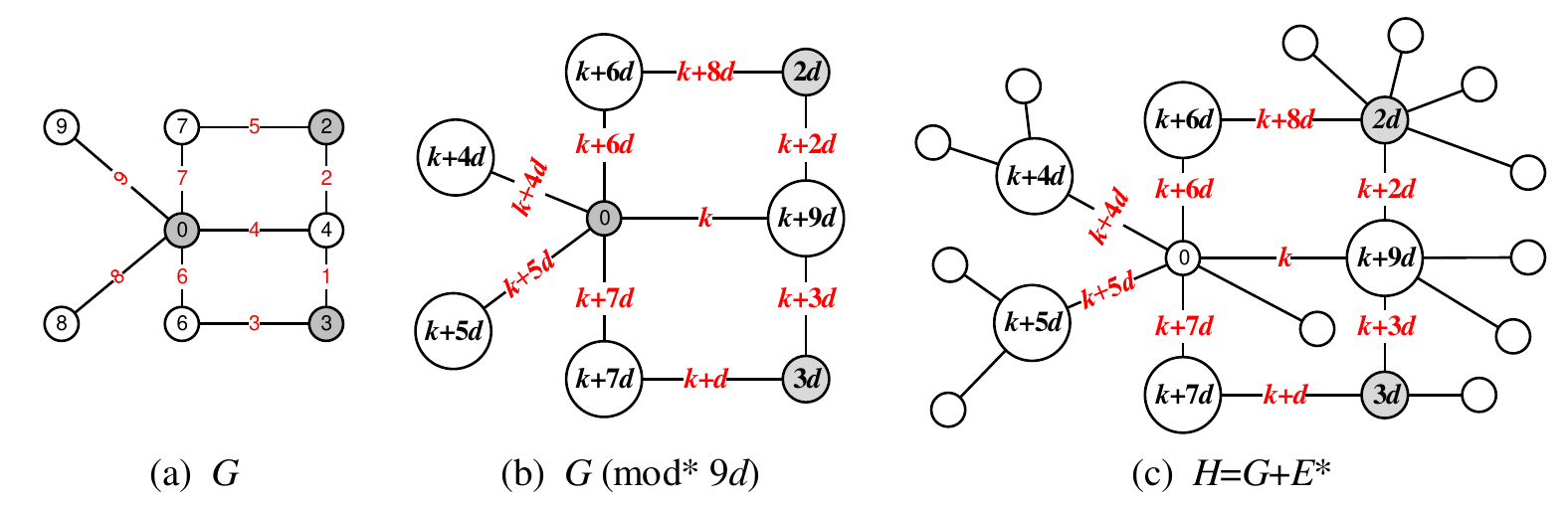}\\
\caption{\label{fig:k-d-harmonious-11}{\small An example for illustrating the RLA-algorithm for the $(k,d)$-harmonious labeling.}}
\end{figure}

\begin{figure}[h]
\centering
\includegraphics[width=16.4cm]{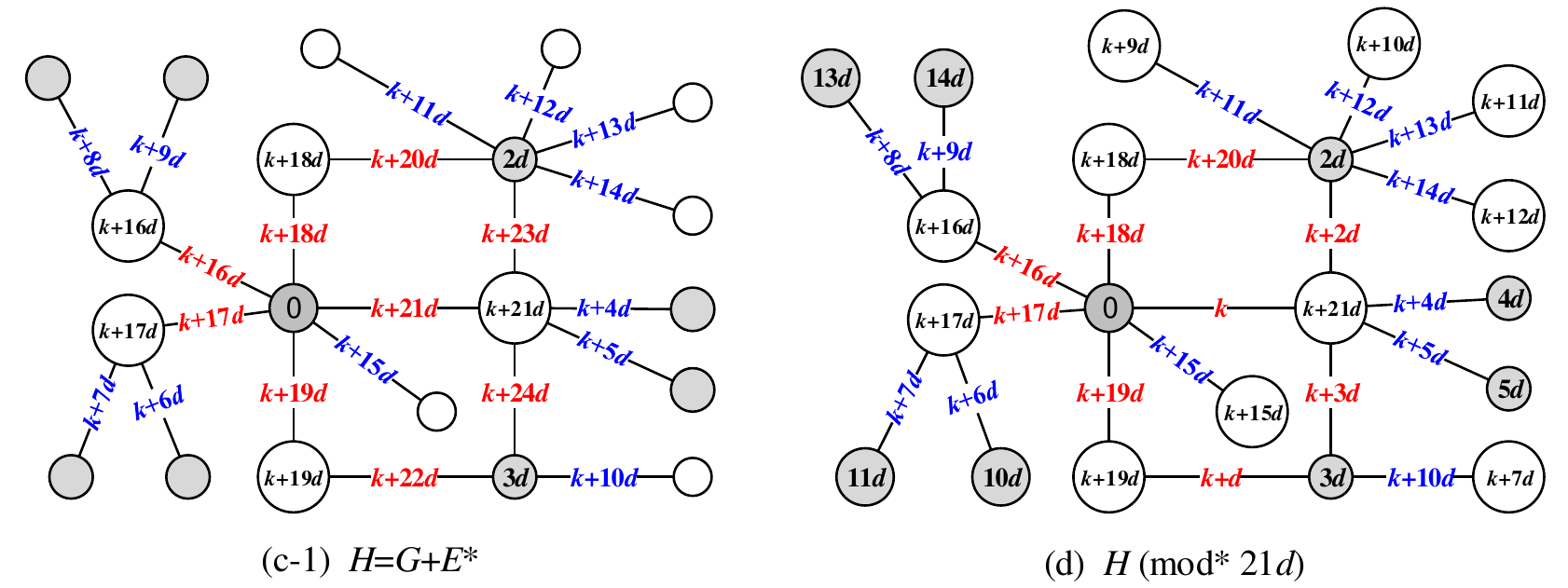}\\
\caption{\label{fig:k-d-harmonious-22}{\small For illustrating the RLA-algorithm for the $(k,d)$-harmonious labeling.}}
\end{figure}

Let $T_0$ be a tree, so $T_1=T_0-L(T_0)$ is a tree too, where $L(T_0)$ is the set of all leaves of the tree $T_0$. Then, we have a tree $T_j=T_{j-1}-L(T_{j-1})$ for $j\in [1, n(T_0)]$, such that $T_{n(T_0)}$ is just a star $K_{1,n(T_0)}$. Since each star admits a set-ordered graceful labeling, by the RLA-algorithm for the $(k,d)$-harmonious labeling, we have proven the following two results:

\begin{thm}\label{thm:tree-k-d-harmonious-labeling}
$^*$ Every tree admits a $(k,d)$-harmonious labeling.
\end{thm}

\begin{thm}\label{thm:adding-leaf-k-d-harmonious-labeling}
$^*$ If a connected bipartite $(p,q)$-graph $G$ admits a set-ordered graceful labeling, then it admits a $(k,d)$-harmonious labeling, and moreover a connected bipartite $(p+m,q+m)$-graph $G^*$ obtained by adding $m$ leaves to $G$ admits a $(k,d)$-harmonious labeling too.
\end{thm}

\subsubsection{RLA-algorithm for the $(k,d)$-elegant labeling}

\begin{defn} \label{defn:k-d-graceful-elegant-labelings}
$^*$ A connected bipartite $(p,q)$-graph $G$ admitting a $(k,d)$-\emph{elegant labeling} $\alpha:V(G)\rightarrow S_{Md}\cup S_{k,(q-1)d}$ for integers $d\geq 1$ and $k\geq 0$ if $\alpha (x)\in \{0,d,2d,\dots , (q-1)d\}$ for $x\in X$, and $\alpha (y)\in \{k,k+d,k+2d,\dots , k+(q-1)d\}$ for $y\in Y$, as well as the induced edge color set
$$\alpha^*(E(G))=\{\alpha (xy)-k=\alpha (y)+\alpha (x)-k~(\bmod~qd):xy\in E(G)\}=[0,q-1]
$$ where $V(G)=X\cup Y$ and $X\cap Y=\emptyset$.\qqed
\end{defn}

\vskip 0.4cm

\noindent \textbf{$^*$ RLA-algorithm for the $(k,d)$-elegant labeling}:

\textbf{Input:} A connected bipartite $(p,q)$-graph $G$ admitting a $(k,d)$-elegant labeling $\alpha$.

\textbf{Output:} A connected bipartite $(p+m,q+m)$-graph $G^*$ admitting a $(k,d)$-elegant labeling, where $G^*$ is obtained by adding $m$ leaves to $G$ randomly.

\textbf{Step 1.} By the definition of a $(k,d)$-elegant labeling defined in Definition \ref{defn:k-d-graceful-elegant-labelings}, there is $V(G)=X\cup Y$ with $X\cap Y=\emptyset$, where $X=\{x_1,x_2,\dots ,x_s\}$ and $Y=\{y_1,y_2,\dots ,y_t\}$ with $s+t=p=|V(G)|$. We have

1. $0=\alpha (x_1)<\alpha (x_2)<\cdots <\alpha (x_s)$, and each vertex color $\alpha (x_i)\in \{0,d,2d,\dots , (q-1)d\}$;

2. $\alpha (y_j)-k<\alpha (y_{j+1})-k$, and each vertex color $\alpha (y_j)\in \{k,k+d,k+2d,\dots , k+(q-1)d\}$ for $j\in [1,t-1]$.

Thereby, we get the edge color set
\begin{equation}\label{eqa:555555}
f^*(E(G))=\{\alpha (x_iy_j)-k=\alpha (y_j)+\alpha (x_i)-k~(\bmod~qd):x_iy_j\in E(G)\}=[0,q-1]
\end{equation}

\textbf{Step 2.} Adding randomly $a_i$ new leaves $x_{i,k}$ to each vertex $x_i$ by adding new edges $x_ix_{i,k}$ for $k\in [1,a_i]$ and $i\in[1,s]$, and adding randomly $b_j$ new leaves $y_{j,r}$ to each vertex $y_j$ by adding new edges $y_jy_{j,r}$ for $r\in [1,b_j]$ and $j\in[1,t]$, it is allowed that some $a_i=0$ or some $b_j=0$. Let $A=\sum ^{s}_{l=1}a_l$ and $B=\sum ^{t}_{l=1}b_l$, and $m=A+B$. The resultant graph is denoted as $G^*$, clearly, $G^*$ is a connected bipartite $(p+m,q+m)$-graph.

\textbf{Step 3.} Define a new labeling $\beta$ for $G^*$ in the following way:

\textbf{Step 3.1} Setting $\beta(w)=\alpha(w)$ for $w\in V(G)\cup E(G)\subset V(G^*)\cup E(G^*)$.

\textbf{Step 3.2} Color edges $y_jy_{j,r}$ of $G^*$ with
\begin{equation}\label{eqa:555555}
\beta(y_{t-j}y_{t-j,r})=rd+k+(q-1)d+d\sum ^{t}_{l=t-j}b_l,~r\in [1,b_{t-j}],~j\in [1,t-2]
\end{equation}
since $\beta(y_{t}y_{t,r})=rd+k+(q-1)d$ for $r\in [1,b_t]$, and $\beta(y_{t}y_{t,b_t})=b_td+k+(q-1)d$ for $r\in [1,b_t]$. Thereby,
$$
\beta(y_1y_{1,r})=rd+k+(q-1)d+d\sum ^{t}_{l=2}b_l,~r\in [1,b_1]
$$ and $\beta(y_1y_{1,b_t})=k+(B+q-1)d$.

\textbf{Step 3.3} Color edges $x_sx_{s,r}$ of $G^*$ by $\beta(x_sx_{s,r})=rd+(B+q-1)d$ for $r\in [1,a_s]$, and moreover we color edges $x_{s-i}x_{s-i,r}$ as follows
\begin{equation}\label{eqa:555555}
\beta(x_{s-i}x_{s-i,r})=rd+(B+q-1)d+d\sum ^{s}_{l=s-i+1}a_l,~r\in [1,a_{s-i}],~i\in [1,s-1]
\end{equation} so $\beta(x_1x_{1,r})=rd+(B+q-1)d+d\sum ^{s}_{l=2}a_l$ and $\beta(x_1x_{1,a_1})=(A+B+q-1)d$.

\textbf{Step 3.4} Color leaves $x_{i,k}$ of $G^*$ with $\beta(x_{i,k})=\alpha (x_ix_{i,k})-\alpha(x_i)$ for $k\in [1,a_i]$ and $i\in [1,s]$; and color leaves $y_{j,r}$ of $G^*$ with $\beta(y_{j,r})=\alpha (y_jy_{j,r})-\alpha(y_j)$ for $r\in [1,b_j]$ and $j\in [1,t]$.

\textbf{Step 4.} Return the $(k,d)$-elegant labeling $\beta$ of the connected bipartite $(p+m,q+m)$-graph $G^*$ by the above algorithm (see examples shown in Fig.\ref{fig:yao-k-d-elegant-labeling} for understanding this algorithm).

\vskip 0.4cm

Let $T_1$ be a tree and let $L(T_1)$ be the set of leaves of $T_1$. We get a tree $T_2=T_1-L(T_1)$ by removing all leaves of $T_1$, and other trees $T_{i+1}=T_i-L(T_i)$ by removing all leaves of $T_i$ for $i\in [1,m_{T_1}-1]$, such that the last tree $T_{m_{T_1}}=K_{1,n}$, where $K_{1,n}$ is a star with $V(K_{1,n})=\{x,y_j:j\in [1,n]\}$ and $E(K_{1,n})=\{xy_j:j\in [1,n]\}$, $\textrm{deg}(x)=n$ and $\textrm{deg}(y_j)=1$ for $j\in [1,n]$. It is obvious, $K_{1,n}$ admits a $(k,d)$-elegant labeling, so by the RLA-algorithm for the $(k,d)$-elegant labeling, each tree $T_{i+1}=T_i-L(T_i)$ admits a $(k,d)$-elegant labeling. Thereby, we present a result as follows:

\begin{thm}\label{thm:k-d-graceful-elegant-labelings}
$^*$ Every tree admits a $(k,d)$-elegant labeling for integers $d\geq 1$ and $k\geq 0$ defined in Definition \ref{defn:k-d-graceful-elegant-labelings}.
\end{thm}

\begin{figure}[h]
\centering
\includegraphics[width=16.4cm]{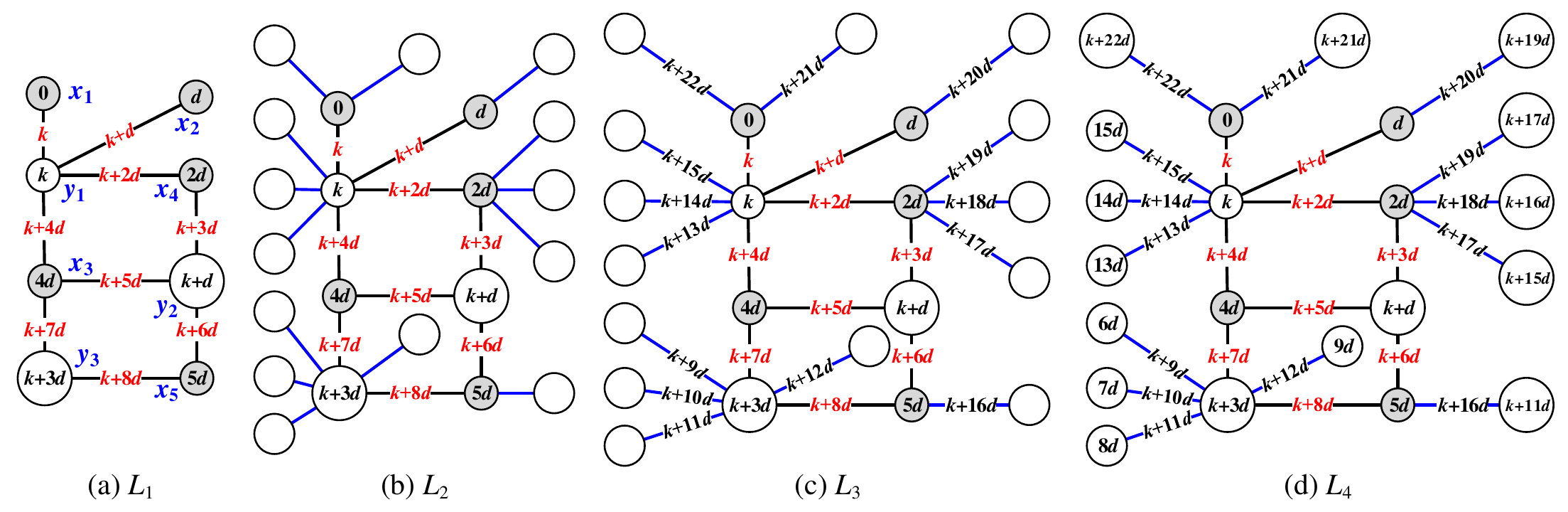}\\
\caption{\label{fig:yao-k-d-elegant-labeling}{\small An example for illustrating the RLA-algorithm for the $(k,d)$-elegant labeling.}}
\end{figure}

As $(k,d)=(1,2)$ in Theorem \ref{thm:k-d-graceful-elegant-labelings}, we have solved the following conjecture:

\begin{conj}\label{conjecture:Zhou-Yao-Chen2013}
\cite{Zhou-Yao-Chen2013} Every tree admits an odd-elegant labeling defined in Definition \ref{defn:zhou-odd-elegant-labeling}.
\end{conj}

\begin{defn} \label{defn:yao-odd-elegant-coloring}
$^*$ An \emph{elegant coloring} $f$ of a $(p,q)$-graph $G$ is defined as $f:V(G)\rightarrow [0,2q-1]$, such that the edge color set
$$\alpha (E(G))=\{\alpha (uv)=\alpha (u)+\alpha (v)~(\bmod~2q):uv\in E(G)\}=[1,2q-1]^o
$$ and $\alpha (x)=\alpha (y)$ for some distinct $x,y\in V(G)$ and $xy\not \in E(G)$.\qqed
\end{defn}

\begin{defn} \label{defn:yao-k-dedge-magic-total-labelings}
$^*$ Suppose a $(p,q)$-graph $G$ admits a mapping $\pi:V(G)\cup E(G)\rightarrow [0,k+(2q-1)d]$ for some integers $d\geq 1$ and $k\geq 0$. If there is a constant $\lambda$ such that $\pi (x)\neq \pi (y)$ for any pair of vertices $x,y\in V(G)$ and $\pi(u)+\pi(uv)+\pi(v)=\lambda$ for each edge $uv\in E(G)$, then we call $\pi$ a \emph{$(k,d)$-edge-magic total labeling}. As $(k,d)=(1,2)$, we call $\pi$ an \emph{odd-edge-magic total labeling}.\qqed
\end{defn}

\begin{example}\label{exa:transformations-four-k-d-labelings}
In Fig.\ref{fig:others-from-odd-elegant}, we can observe:

(i) A connected $(8,9)$-graph $L_1$ admits a $(k,d)$-\emph{elegant labeling} $f_1$ with $f_1(x_iy_j)=f_1(x_i)+f_1(y_j)$ for each edge $x_iy_j\in E(L_1)$;

(ii) a connected $(8,9)$-graph $J$ admits a $(k,d)$-\emph{edge-magic total labeling} $f_2$ holding $f_2(x_i)+f_2(x_iy_j)+f_2(y_j)=2k+8d$ for each edge $x_iy_j\in E(J)$;

(iii) a connected $(8,9)$-graph $D$ admits a $(k,d)$-\emph{graceful difference labeling} $f_3$ with $\big ||f_3(x_i)-f_3(y_j)|-f_3(x_iy_j)\big |=3d$ for each edge $x_iy_j\in E(D)$; and

(iv) a connected $(8,9)$-graph $I$ admits a $(k,d)$-\emph{felicitous difference labeling} $f_4$ with $f_4(x_i)+f_4(x_iy_j)-f_4(y_j)=5d$ for each edge $x_iy_j\in E(I)$ with $x_i\in X$ and $y_j\in Y$, where $V(I)=X\cup Y$ and $X\cap Y=\emptyset$.

Notice that $V(L_1)=V(J)=V(D)=V(I)=X\cup Y$ and $X\cap Y=\emptyset$, and $E(L_1)=E(J)=E(D)=E(I)$, since these four graphs are isomorphic from each other. We have the following transformations among the above four labelings:
\begin{asparaenum}[\textrm{Conn}-1.]
\item $f_2(w)=f_1(w)$ for $w\in V(L_1)$, $f_2(x_iy_j)=\max f_1(E(L_1))+\min f_1(E(L_1))-f_1(x_iy_j)$ for each edge $x_iy_j\in E(I)$
\item $f_3(x_i)=f_1(x_i)$ for $x_i\in X$, $f_3(x_iy_j)=f_1(x_iy_j)$ for each edge $x_iy_j\in E(L_1)$, and $f_3(y_j)=f_1(y_{t-j+1})$ for each vertex $y_j\in Y$ and $j\in [1,t]$.
\item $f_4(x_i)=f_1(x_{s-i+1})$ for $x_i\in X$ and $i\in [1,s]$, $f_4(x_iy_j)=f_1(x_iy_j)$ for each edge $x_iy_j\in E(L_1)$, and $f_4(y_j)=f_1(y_j)$ for each vertex $y_j\in Y$.
\end{asparaenum}
\end{example}

\begin{figure}[h]
\centering
\includegraphics[width=15cm]{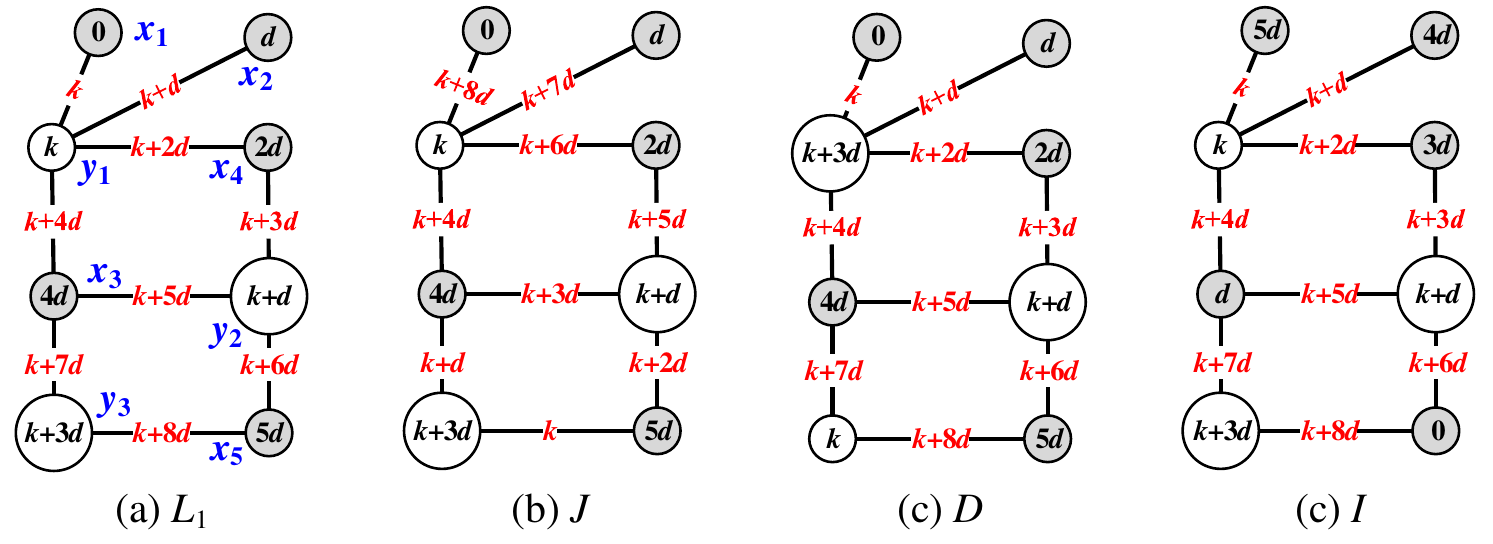}\\
\caption{\label{fig:others-from-odd-elegant}{\small $L_1$ admits a $(k,d)$-elegant labeling $f_1$; $J$ admits a $(k,d)$-edge-magic total labeling $f_2$; $D$ admits a $(k,d)$-graceful difference labeling $f_3$; and $I$ admits a $(k,d)$-felicitous difference labeling $f_4$.}}
\end{figure}

By Definition \ref{defn:yao-k-dedge-magic-total-labelings}, Theorem \ref{thm:k-d-graceful-elegant-labelings} and three transformations introduced in Example \ref{exa:transformations-four-k-d-labelings}, we have

\begin{thm}\label{thm:odd-edge-edge-magic-total-labeling}
$^*$ Each tree $T$ admits a $(k,d)$-elegant labeling if and only if $T$ admits each of a $(k,d)$-edge-magic total labeling (including an odd-edge-magic total labeling), a $(k,d)$-graceful difference labeling and a $(k,d)$-felicitous difference labeling.
\end{thm}

\subsubsection{RLA-algorithm for the $(k,d)$-odd-elegant total coloring}

\begin{defn} \label{defn:k-d-odd-graceful-elegant-labelings}
\cite{Yao-Wang-2106-15254v1} Let $Ao(q;k,d)^o=\{k+d,k+3d,\dots ,k+(2q-1)d\}$ for integers $q,d\geq 1$ and $k\geq 0$. Suppose a $(p,q)$-graph $G$ admits a proper labeling
$\pi:V(G)\rightarrow [0,k+(2q-1)d]$ for integers $d\geq 1$ and $k\geq 0$.

$(i)$ If the induced edge color set
$$\pi(E(G))=\{\pi(uv)=|\pi(u)-\pi(v)|:uv\in E(G)\}=Ao(q;k,d)^o
$$ we say $G$ to be \emph{$(k,d)$-odd-graceful}, and call $\pi$ a \emph{$(k,d)$-odd-graceful labeling} of $G$.

$(ii)$ If the induced edge color set
$$\{\pi(u)+\pi(v)-k~(\bmod~2qd):uv\in E(G)\}=Ao(q;0,d)^o
$$ we say $G$ to be \emph{$(k,d)$-odd-elegant}, and call $\pi$ a \emph{$(k,d)$-odd-elegant labeling} of $G$.\qqed
\end{defn}

\begin{defn} \label{defn:yao-k-d-odd-elegant-coloring}
$^*$ Suppose a $(p,q)$-graph $G$ admits a mapping $\pi:V(G)\rightarrow [0,k+(2q-1)d]$ for some integers $d\geq 1$ and $k\geq 0$. If the induced edge color set
$$\{\pi(u)+\pi(v)-k~(\bmod~2qd):uv\in E(G)\}=\{0,d,3d,\dots ,(2q-3)d\}
$$ and $\alpha (x)=\alpha (y)$ for some distinct $x,y\in V(G)$ and $xy\not \in E(G)$, then we call $\pi$ a \emph{$(k,d)$-odd-elegant coloring} of $G$.\qqed
\end{defn}

The RLA-algorithm for the $(k,d)$-odd-elegant labeling is as the same as the RLA-algorithm for the $(k,d)$-elegant labeling, so we omit to state it, and present an example shown in Fig.\ref{fig:k-d-odd-elegant-add-leaves} for this algorithm.

\begin{figure}[h]
\centering
\includegraphics[width=16.4cm]{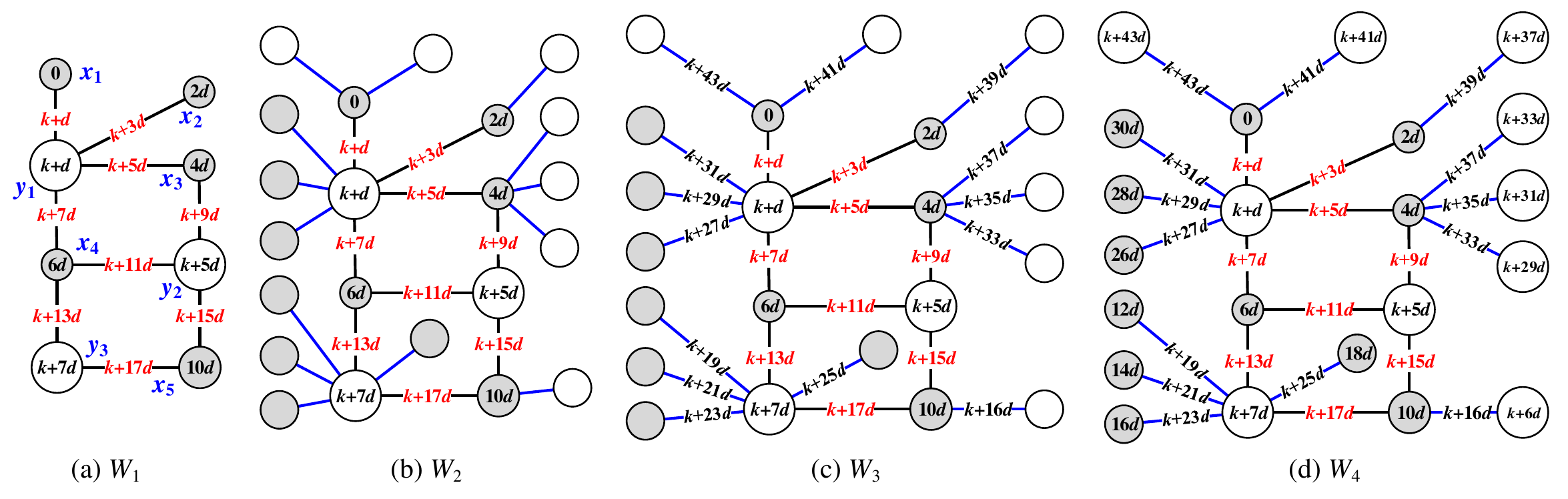}\\
\caption{\label{fig:k-d-odd-elegant-add-leaves}{\small An example for the RLA-algorithm for the $(k,d)$-odd-elegant labeling.}}
\end{figure}

\begin{thm}\label{thm:graphs-odd-elegant-coloring}
$^*$ A connected $(p,q)$-graph $G$ admits an odd-elegant labeling (resp. $(k,d)$-odd-elegant labeling) if and only if there is a subset $S$ of $V(G)$ such that the vertex-split $(q+1,q)$-tree $G\wedge S$ admits an odd-elegant coloring (resp. $(k,d)$-odd-elegant coloring) $f$ defined in Definition \ref{defn:yao-odd-elegant-coloring}, each vertex $u\in S$ corresponds a vertex $u^*\in S$ holding $f(u)=f(u^*)$ true, we have the colored graph homomorphism $G\wedge S\rightarrow G$ under the odd-elegant labeling $f$.
\end{thm}

\begin{defn} \label{defn:55-v-set-e-proper-more-labelings}
\cite{Yao-Wang-2106-15254v1} A \emph{v-set e-proper $\varepsilon$-labeling} (resp. $\varepsilon$-\emph{coloring}) of a $(p,q)$-graph $G$ is a mapping $f:V(G)\cup E(G)\rightarrow \Omega$, where $\Omega$ consists of numbers and sets, such that $f(u)$ is a set for each vertex $u\in V(G)$, and $f(xy)$ is a number for each edge $xy\in E(G)$, and the edge color set $f(E(G))$ satisfies an $\varepsilon$-condition.\qqed
\end{defn}

\begin{thm}\label{thm:graphs-odd-elegant-v-set-coloring}
$^*$ Each connected $(p,q)$-graph $G$ admits a \emph{v-set e-odd-elegant labeling} (resp. \emph{v-set $(k,d)$-odd-elegant labeling}).
\end{thm}

\subsubsection{RLA-algorithm for the $(k,d)$-gracefully total coloring}

\begin{defn} \label{defn:2020arXiv-gracefully-total-coloring}
\cite{Bing-Yao-2020arXiv} Suppose that a connected $(p,q)$-graph $G$ ($\neq K_p$) admits a total coloring $f:V(G)\cup E(G)\rightarrow [1,M]$, and there are $f(x)=f(y)$ for some vertices $x,y\in V(G)$. If $|f(V(G))|< p$, and the edge color set $f(E(G))=\{f(uv)=|f(u)-f(v)|:uv\in E(G)\}=[1,q]$, we call $f$ a \emph{gracefully total coloring}, and moreover $f$ is called a \emph{set-ordered gracefully total coloring} if $\max f(X)<\min f(Y)$ when $G$ is bipartite with bipartition $(X,Y)$ of $V(G)$.\qqed
\end{defn}

\begin{defn} \label{defn:k-d-gracefully-total-colorings}
$^*$ Let $S_{Md}=\{0,d,2d,\dots , Md\}$ and $S_{q-1,k,d}=\{k,k+d,k+2d,\dots , k+(q-1)d\}$ for integers $q\geq 1$, $d\geq 1$ and $k\geq 0$. A connected bipartite $(p,q)$-graph $G$ admitting a $(k,d)$-\emph{gracefully total coloring} $h:V(G)\cup E(G)\rightarrow S_{Md}\cup S_{q-1,k,d}$ if $h(x)\in S_{Md}$ for $x\in X$, and $h(y)\in S_{q-1,k,d}$ for $y\in Y$, and it is allowed that $h(w)=h(z)$ for some vertices $w,z\in V(G)$, as well as the edge color set
$$h(E(G))=\{h(xy)=h(y)-h(x):xy\in E(G)\}=S_{q-1,k,d}
$$ where $(X,Y)$ is the bipartition of $V(G)$ with $X\cap Y=\emptyset$.\qqed
\end{defn}

\vskip 0.4cm

\textbf{RLA-algorithm for the $(k,d)$-gracefully total coloring \cite{Yao-Su-Wang-Hui-Sun-ITAIC2020}}.

\textbf{Input:} A connected bipartite $(p,q)$-graph $G$ admitting a $(k,d)$-gracefully total coloring $f$.

\textbf{Output:} A connected bipartite $(p+m,q+m)$-graph $G^*$ admitting a $(k,d)$-gracefully total coloring $g$, where $G^*$ is obtained by adding $m$ leaves to $G$ randomly.

\textbf{Step 1.} Let $G$ be a bipartite and connected $(p, q)$-graph with its vertex set $V(G)=X\cup Y$ such that $X\cap Y=\emptyset$, $X=\{x_i:i\in [1, s]\}$ and $Y=\{y_j:j\in [1, t]\}$ holding $s+t=p$, and each edge is $x_iy_j$ with $x_i\in X$ and $y_j\in Y$. Since $G$ admits a $(k,d)$-gracefully total coloring $f$, so we have $f:X\rightarrow S_{Md}$ and $f:Y\cup E(G)\rightarrow S_{q-1, k, d}$ with $k\geq 1$, without loss of generality, there are $0=f(x_1)\leq f(x_i)\leq f(x_{i+1})$ for $i\in [1, s-1]$ and $f(x_{s})\leq f(y_j)\leq f(y_{j+1})\leq f(y_{t})=k+(q-1)d$ for $j\in [1, t-1]$, as well as $f(E(G))=S_{q-1, k, d}$.

\textbf{Step 2.} Color each vertex $x_i\in X\subset X\,'$ with $g(x_i)=f(x_i)$ for $i\in [1, s]$, and color each vertex $y_j\in Y\subset Y'$ with $g(y_j)=f(y_j)+[A(s)+B(t)]d$ for $j\in [1, t]$, and color edges $x_iy_j\in E(G)\subset E(G+L_{\textrm{eaf}})$ by
$${
\begin{split}
g(x_iy_j)=g(y_j)-g(x_i)=f(y_j)+[A(s)+B(t)]d-f(x_i)=f(x_iy_j)+[A(s)+B(t)]d
\end{split}}$$ So, we get the edge color set
$${
\begin{split}
g(E(G))=\{k+[A(s)+B(t)]d, k+[A(s)+B(t)+1]d, \dots , k+[A(s)+B(t)+q-1]d\}
\end{split}}
$$

\textbf{Step 3.} Let edges $e_{i_1} e_{i_2} \dots e_{i_{A(s)+B(t)}}$ be a permutation of edges $x_ix_{i, r}$ for $r\in [1, a_i]$ and $i\in [1, s]$ and edges $y_jy_{j, r}$ for $r\in [1, b_{j}]$ and $j\in [1, t]$, that is, $e_{i_1} e_{i_2} \dots e_{i_{A(s)+B(t)}}$ is a permutation of $x_1x_{1, 1}$ $x_1x_{1, 2}$ $\cdots$ $x_1x_{1, a_1}$ $x_2x_{2, 1}$ $x_2x_{2, 2}$ $\cdots$ $x_2x_{2, a_2}$ $\cdots$ $x_ix_{i, 1}$ $x_ix_{i, 2}$ $\cdots$ $x_ix_{i, a_i}$ $\cdots$ $x_sx_{s, 1}$ $x_sx_{s, 2}$ $\cdots$ $x_sx_{s, a_s}$ $y_1y_{1, 1}$ $y_1y_{1, 2}$ $\cdots$ $y_1y_{1, b_1}$ $y_2y_{2, 1}$ $y_2y_{2, 2}$ $\cdots$ $y_2y_{2, b_2}$ $\cdots$ $y_jy_{j, 1}$ $y_jy_{j, 2}$ $\cdots$ $y_jy_{j, b_j}$ $\cdots$ $y_ty_{t, 1}$ $y_ty_{t, 2}$ $\cdots$ $y_ty_{t, b_t}$. We color each edge $e_{i_r}$ with $g(e_{i_r})=k+(r-1)d$ with $r\in [1, A(s)+B(t)]$, and color each vertex $x_{i, r}$ with $g(x_{i, r})=f(x_i)+g(e_{i_r})$ if vertex $x_{i, r}$ is an end of the edge $e_{i_r}$, and color each vertex $y_{j, r}$ with $g(y_{j, r})=g(y_j)-g(e_{i_r})$ if vertex $y_{j, r}$ is an end of the edge $e_{i_r}$.

\textbf{Step 4.} Return the $(k,d)$-gracefully total coloring $g$ of $G^*$.

\vskip 0.4cm

\begin{figure}[h]
\centering
\includegraphics[width=16.4cm]{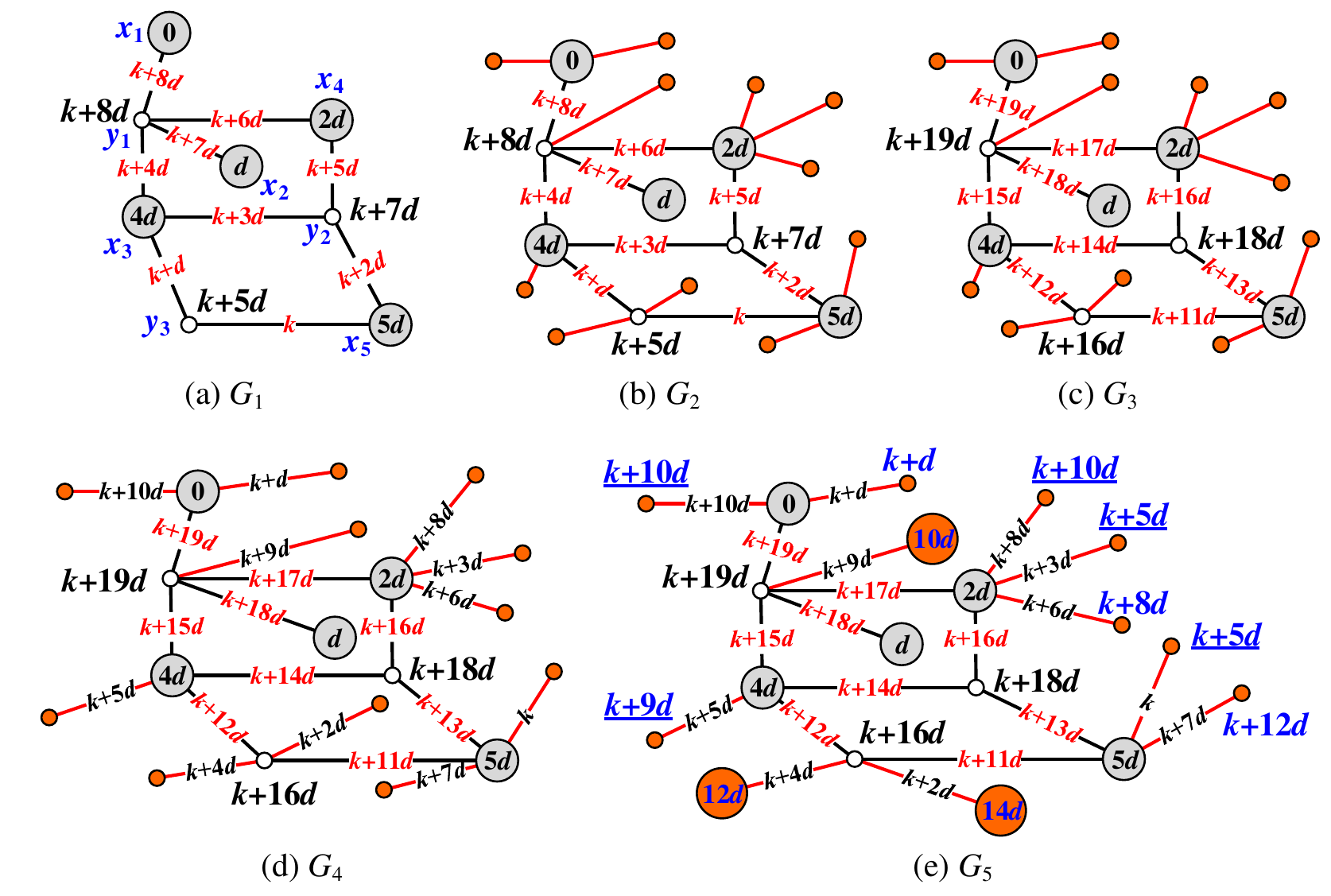}\\
\caption{\label{fig:graceful-sequence-total-coloring}{\small A process for understanding the RLA-algorithm for the $(k,d)$-gracefully total coloring, cited from \cite{Yao-Su-Wang-Hui-Sun-ITAIC2020}.}}
\end{figure}

\begin{figure}[h]
\centering
\includegraphics[width=16.4cm]{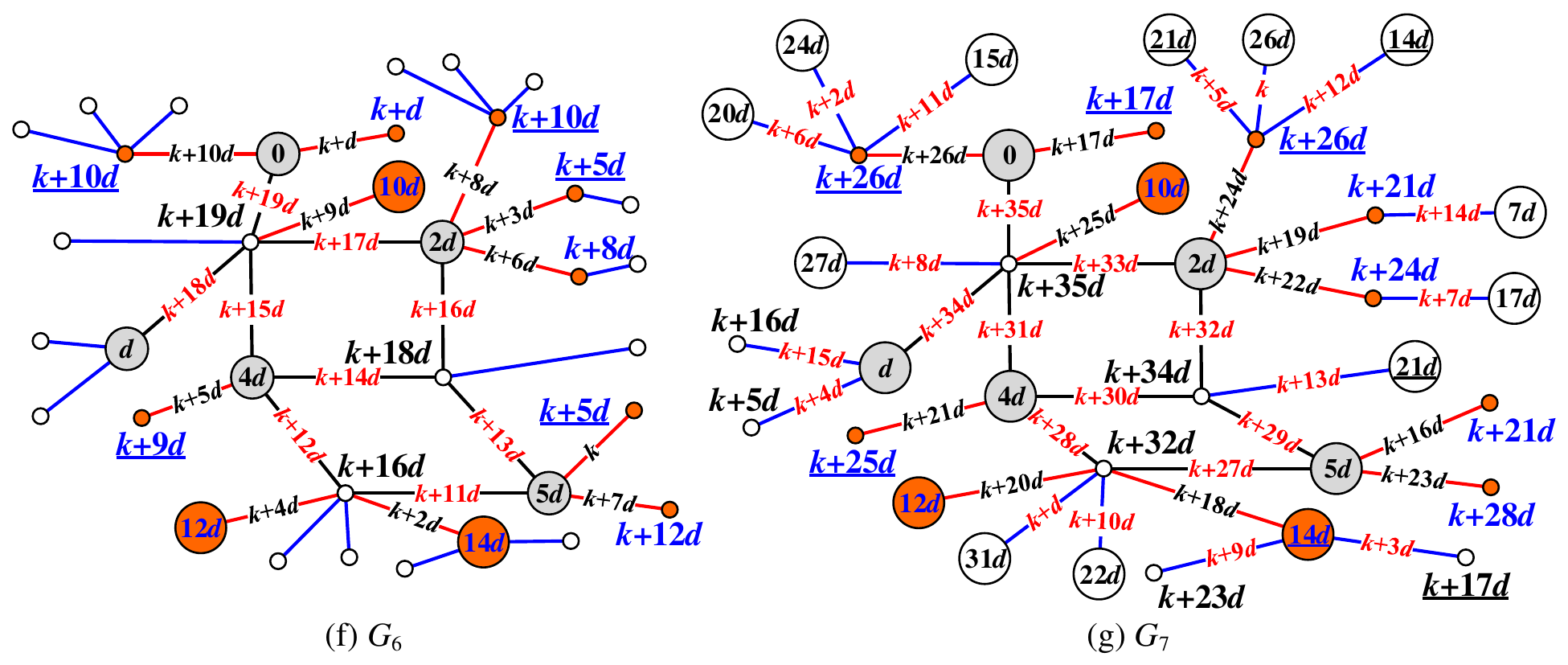}\\
\caption{\label{fig:graceful-sequence-total-coloring11}{\small For understanding the RLA-algorithm for the $(k,d)$-gracefully total coloring, cited from \cite{Yao-Su-Wang-Hui-Sun-ITAIC2020}.}}
\end{figure}

See Fig.\ref{fig:graceful-sequence-total-coloring} and Fig.\ref{fig:graceful-sequence-total-coloring11} for understanding the RLA-algorithm for the $(k,d)$-gracefully total coloring, and by Lemma \ref{thm:adding-leaves-keep-sequence-colorings} and the RLA-algorithm for the $(k,d)$-gracefully total coloring, we get the following results:

\begin{lem}\label{thm:LEAVES-added-algorithm}
\cite{Yao-Su-Wang-Hui-Sun-ITAIC2020} Let $G$ be a bipartite and connected $(p, q)$-graph. If $G$ admits a $(k,d)$-gracefully total coloring, then any leaf-added graph $G+L_{\textrm{eaf}}$ admits a \emph{$(k,d)$-gracefully total coloring} too, where $L_{\textrm{eaf}}$ is the set of leaves added randomly to $G$.
\end{lem}

\begin{thm}\label{thm:tree-graceful-total-coloringss}
\cite{Yao-Su-Wang-Hui-Sun-ITAIC2020} Each tree admits a \emph{$(k,d)$-gracefully total coloring} defined in Definition \ref{defn:kd-w-type-colorings}, also, a \emph{gracefully total coloring} as $(k, d)=(1, 1)$, an \emph{odd-gracefully total coloring} as $(k,d)=(1,2)$.
\end{thm}

Let $N(m)=\big [A(s)+B(t)\big ]!$, there are $N(m)$ permutations of edges $x_ix_{i, r}$ for $r\in [1, a_i]$ and $i\in [1, s]$ and edges $y_jy_{j, r}$ for $r\in [1, b_{j}]$ and $j\in [1, t]$ in Step 2 of the RLA-algorithm for the $(k,d)$-gracefully total coloring. So, we have at least $N(m)$ $(k,d)$-gracefully total colorings of the leaf-added $(p+m,q+m)$-graph $G^*$. Since each star admits a set-ordered graceful labeling, by the RLA-algorithm for the $(k,d)$-gracefully total coloring, then

\begin{thm}\label{thm:permutations-k-d-gracefully-total-coloring}
$^*$ Let $T_0$ be a tree, we have trees $T_j=T_{j-1}-L(T_{j-1})$ for $j\in [1, n(T_0)]$, such that $T_{n(T_0)}$ is just a star $K_{1,n(T_0)}$, and let $c_{j-1}=|L(T_{j-1})|$ for $j\in [1, n(T_0)]$. Then $T_0$ admits at least $N(T_0)$ $(k,d)$-gracefully total colorings for integers $k,d\geq 1$ by the RLA-algorithm for the $(k,d)$-gracefully total coloring, where $N(T_0)=\prod ^{n(T_0)-1}_{j=1}(c_{j-1})!$.
\end{thm}

\begin{problem}\label{qeu:444444}
Since each tree $T$ admits a gracefully total coloring defined in Theorem \ref{thm:tree-graceful-total-coloringss}, let $n_{\textrm{dif}}(f(V(T)))$ be the number of different integers of the color set $f(V(T))$. \textbf{Determine}
$$m_{\textrm{graceful}}(T)=\min_f \{n_{\textrm{dif}}(f(V(T))): f\textrm{ is a gracefully total coloring of }T\}$$
and
$$M_{\textrm{graceful}}(T)=\max_g \{n_{\textrm{dif}}(g(V(T))): g\textrm{ is a gracefully total coloring of }T\}$$
The largest number $M_{\textrm{graceful}}(T)=|V(T)|$ is related with a famous conjectured: \emph{Every tree admits a graceful labeling} (Alexander Rosa, 1966, \cite{Bondy-2008} and \cite{Gallian2021}).
\end{problem}

\begin{conj}\label{conj:rotatable-k-d-gracefully-labeling}
$^*$ Each tree $T$ admits \emph{$0$-rotatable $(k,d)$-gracefully labelings} with $d>k$ for odd $k\geq 1$ and even $d\geq 2$, such that any vertex $u\in V(T)$ is colored with $f(u)=0$ by a $(k,d)$-gracefully labeling $f$ of $T$. (see examples shown in Fig.\ref{fig:rotatable-k-d-gracefully})
\end{conj}

\begin{figure}[h]
\centering
\includegraphics[width=13cm]{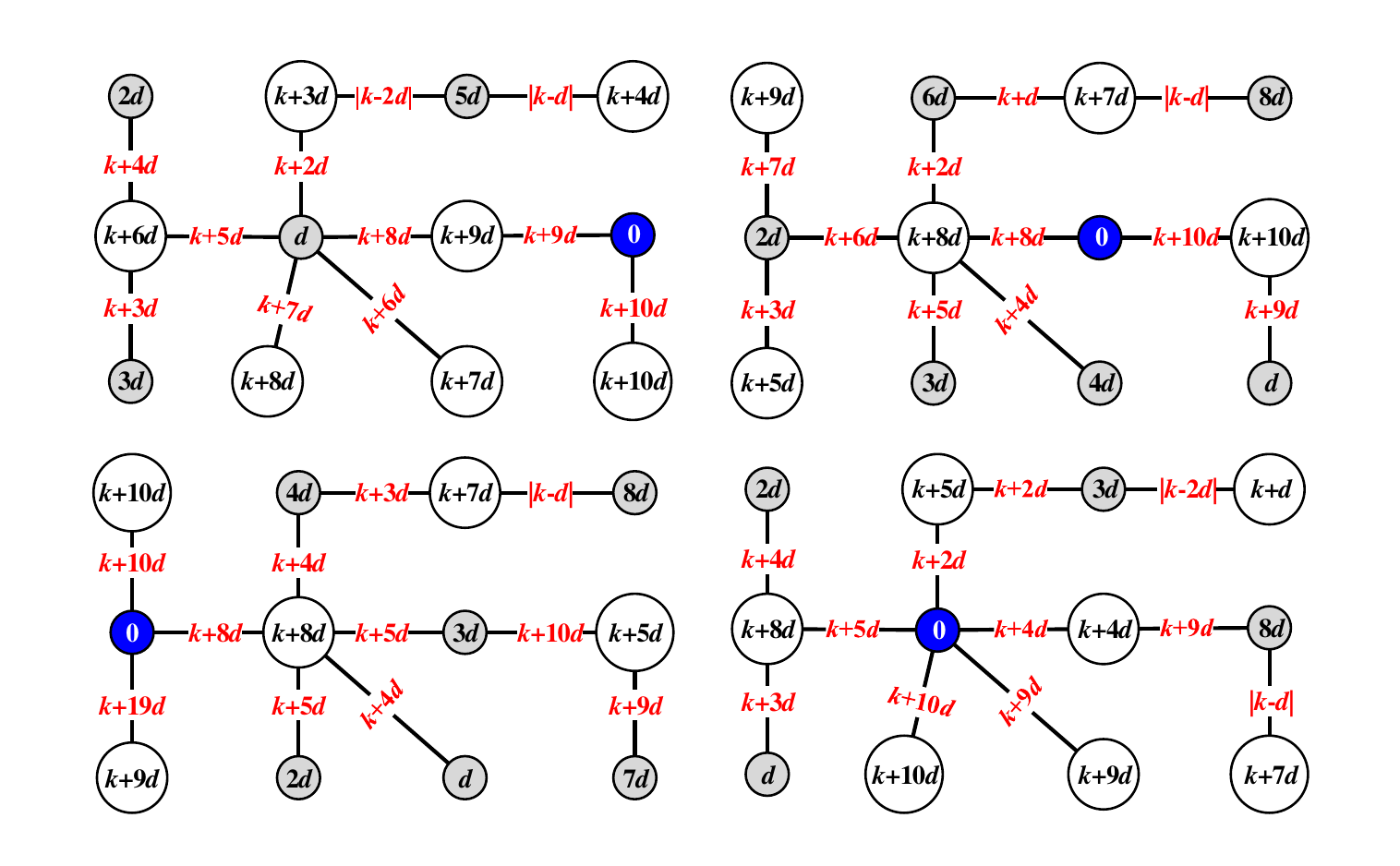}\\
\caption{\label{fig:rotatable-k-d-gracefully}{\small A scheme for illustrating Conjecture \ref{conj:rotatable-k-d-gracefully-labeling}.}}
\end{figure}

\subsubsection{Set-dual $(k,d)$-type total colorings}

By Definition \ref{defn:k-d-gracefully-total-colorings}, we define a group of set-dual $(k,d)$-type total colorings as follows:

\begin{defn} \label{defn:Set-dual-type-k-d-type-total-colorings}
$^*$ Let $S_{Md}=\{0,d,2d,\dots , Md\}$ and $S_{q-1,k,d}=\{k,k+d,k+2d,\dots , k+(q-1)d\}$ for integers $q\geq 1$, $d\geq 1$ and $k\geq 0$. A connected bipartite $(p,q)$-graph $G$ with bipartition $(X,Y)$ admits a $(k,d)$-gracefully total coloring $h:V(G)\cup E(G)\rightarrow S_{Md}\cup S_{q-1,k,d}$ holding $h(X)\subseteq S_{Md}$, $h(Y)\subseteq S_{q-1,k,d}$, and $h(w)=h(z)$ for some vertices $w,z\in V(G)$, as well as the edge color set $h(E(G))=\{h(xy)=h(y)-h(x):xy\in E(G)\}=S_{q-1,k,d}$. Clearly, $\max h(E(G))=k+(q-1)d$ and $\min h(E(G))=k$. Let $\max h(X)=\max \{h(x):x\in X\}$, $\min h(X)=0$, $\max h(Y)=k+(q-1)d$ and $\min h(Y)=\max \{h(x):x\in X\}$.
\begin{asparaenum}[\textbf{\textrm{KD-dual}} 1. ]
\item The dual $(k,d)$-gracefully total coloring $h_{dual}$ of the $(k,d)$-gracefully total coloring $h$ is defined by
$$
h_{dual}(x)=\max h(X)+\min h(X)-h(x)=\max h(X)-h(x),~x\in X
$$ and
$$
h_{dual}(y)=\max h(Y)+\min h(Y)-h(y)=k+(q-1)d+\min h(Y)-h(y),~y\in Y
$$ and $h_{dual}(xy)=h(xy)$ for each edge $xy\in E(G)$. Since
\begin{equation}\label{eqa:555555}
{
\begin{split}
h_{dual}(xy)+|h_{dual}(y)-h_{dual}(x)|=&h(xy)+\big |[k+(q-1)d+\min h(Y)-h(y)]\\
&-[\max h(X)-h(x)]\big |\\
=&k+(q-1)d+\min h(Y)-\max h(X)
\end{split}}
\end{equation} for each edge $xy\in E(G)$, then $h_{dual}$ is called a \emph{$(k,d)$-edge-difference total coloring} of $G$, see an example $Q_1$ shown in Fig.\ref{fig:k-d-set-dual-22}(e).

\quad Moreover, we have another dual $(k,d)$-gracefully total coloring $h\,^*_{dual}$ of $h$ is defined as: $h\,^*_{dual}(w)=h_{dual}(w)$ for $w\in V(G)$, and
$$
h\,^*_{dual}(xy)=\max h(E(G))+\min h(E(G))-h(xy)=2k+(q-1)d-h(xy)
$$ for $xy\in E(G)$. We say that $h\,^*_{dual}$ is a \emph{$(k,d)$-graceful-difference total coloring} of $G$, since
\begin{equation}\label{eqa:555555}
{
\begin{split}
\big |h\,^*_{dual}(y)-h\,^*_{dual}(x)-h\,^*_{dual}(xy)\big |=&\big |[k+(q-1)d+\min h(Y)-\max h(X)-h(xy)]\\
&-[2k+(q-1)d-h(xy)]\big |\\
=&\big |\min h(Y)-k-\max h(X)\big |
\end{split}}
\end{equation} see an example $Q_2$ shown in Fig.\ref{fig:k-d-set-dual-22}(f).

\item The $X$-dual total coloring $D_{setX}$ of the $(k,d)$-gracefully total coloring $h$ is defined as:
$$
D_{setX}(x)=\max h(X)+\min h(X)-h(x)=\max h(X)-h(x),~x\in X
$$ and $D_{setX}(y)=h(y)$ for $y\in Y$, and $D_{setX}(xy)=h(xy)$ for each edge $xy\in E(G)$. Since each edge $xy\in E(G)$ holds
\begin{equation}\label{eqa:555555}
{
\begin{split}
\big | D_{setX}(y)+D_{setX}(x) -D_{setX}(xy)\big |=\big |\max h(X)-h(x)+h(y)-h(xy)\big |=\max h(X)
\end{split}}
\end{equation} so $D_{setX}$ is a \emph{ $(k,d)$-felicitous-difference total coloring} of $G$. (see an example $J_1$ shown in Fig.\ref{fig:k-d-set-dual-11}(a))

\quad We have another $X$-dual total coloring $D\,^*_{setX}$ of $h$ is defined as: $D\,^*_{setX}(w)=D_{setX}(w)$ for $w\in V(G)$, and each edge $xy\in E(G)$ is colored with
$$
D\,^*_{setX}(xy)=\max h(E(G))+\min h(E(G))-h(xy)=2k+(q-1)d-h(xy)
$$ and moreover the following fact
\begin{equation}\label{eqa:555555}
{
\begin{split}
D\,^*_{setX}(x)+D\,^*_{setX}(xy)+D\,^*_{setX}(y)&=\max h(X)-h(x)+2k+(q-1)d-h(xy)+h(y)\\
&=2k+(q-1)d+\max h(X)
\end{split}}
\end{equation} shows that $D\,^*_{setX}$ is a \emph{$(k,d)$-edge-magic total coloring} of $G$. (see an example $J_2$ shown in Fig.\ref{fig:k-d-set-dual-11}(b))

\item The $Y$-dual $(k,d)$-gracefully total coloring $D_{setY}$ of the $(k,d)$-gracefully total coloring $h$ is defined as: $D_{setX}(x)=h(x)$ for $x\in X$,
$$
D_{setX}(y)=\max h(Y)+\min h(Y)-h(y)=k+(q-1)d+\min h(Y)-h(y),~y\in Y
$$ and $D_{setX}(xy)=h(xy)$ for each edge $xy\in E(G)$. $D_{setY}$ is a \emph{$(k,d)$-edge-magic total coloring} of $G$ (see an example $I_1$ shown in Fig.\ref{fig:k-d-set-dual-11}(c)), since
$$
D\,^*_{setX}(xy)=\max h(E(G))+\min h(E(G))-h(xy)=2k+(q-1)d-h(xy)
$$ and the following fact
\begin{equation}\label{eqa:555555}
{
\begin{split}
D_{setY}(x)+D_{setY}(xy)+D_{setY}(y)&=h(x)+h(xy)+k+(q-1)d+\min h(Y)-h(y)\\
&=k+(q-1)d+\min h(Y)
\end{split}}
\end{equation}

\quad We define another $Y$-dual total coloring $D\,^*_{setY}$ of $h$ by setting $D\,^*_{setY}(w)=D_{setY}(w)$ for $w\in V(G)$, and the color of each edge $xy\in E(G)$ is
$$
D\,^*_{setY}(xy)=\max h(E(G))+\min h(E(G))-h(xy)=2k+(q-1)d-h(xy)
$$ immediately, we have
\begin{equation}\label{eqa:555555}
{
\begin{split}
& \quad \big | D\,^*_{setY}(y)+D\,^*_{setY}(x)-D\,^*_{setY}(xy)\big |\\
&=\big |h(x)+k+(q-1)d+\min h(Y)-h(y)-[2k+(q-1)d-h(xy)]\big |\\
&=\min h(Y)-k
\end{split}}
\end{equation} also, $D\,^*_{setY}$ is a \emph{$(k,d)$-edge-magic total coloring} of $G$. See an example shown in Fig.\ref{fig:k-d-set-dual-22}(d).\qqed
\end{asparaenum}
\end{defn}

\begin{figure}[h]
\centering
\includegraphics[width=16.4cm]{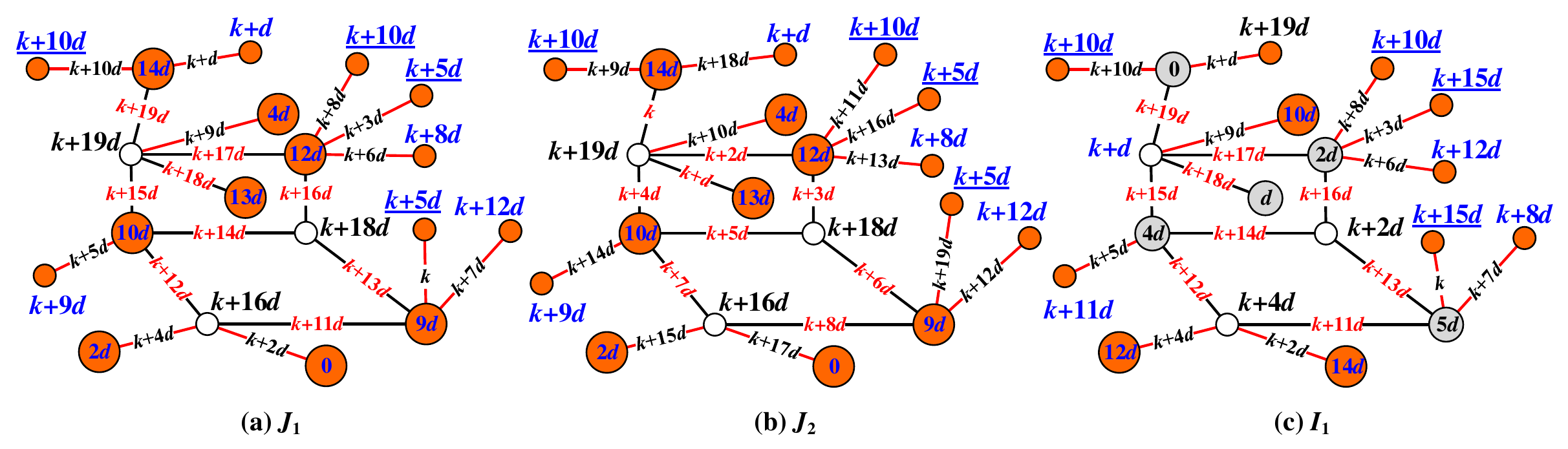}\\
\caption{\label{fig:k-d-set-dual-11}{\small A scheme for set-dual $(k,d)$-type total colorings.}}
\end{figure}

\begin{figure}[h]
\centering
\includegraphics[width=16.4cm]{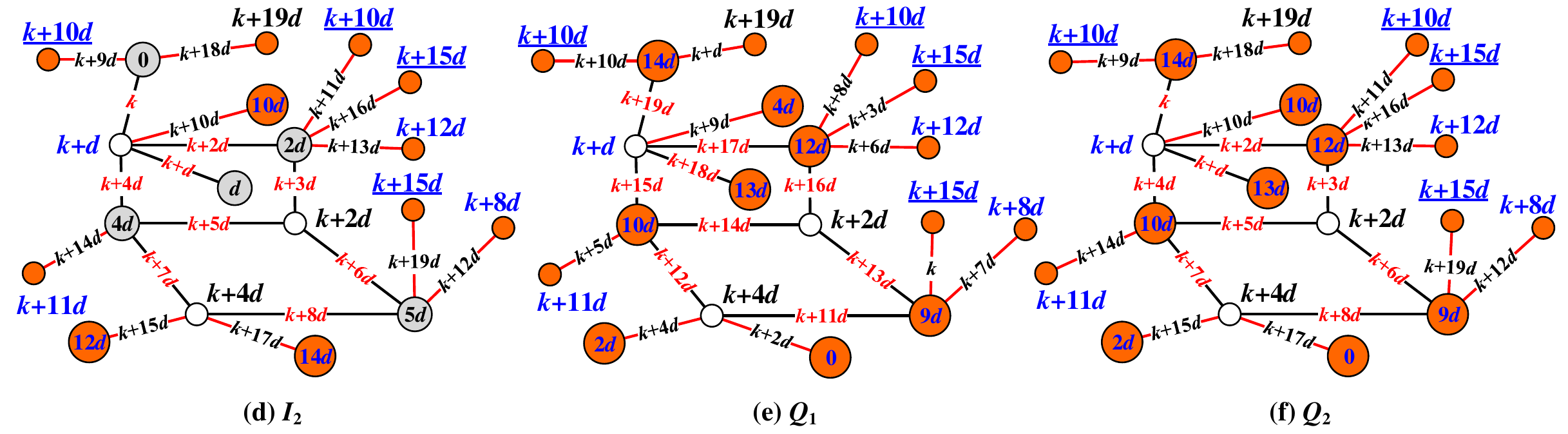}\\
\caption{\label{fig:k-d-set-dual-22}{\small Another scheme for set-dual $(k,d)$-type total colorings.}}
\end{figure}

\subsubsection{RLA-algorithm for the $(k,d)$-gracefully e-image total coloring}

\begin{defn} \label{defn:k-d-gracefully-e-image-total-coloring}
$^*$ Let
$$S^{\pm}_{q-1,k,d}=\{k,k-d,k-2d,\dots, k-(q-1)d\}\cup \{k,k+d,k+2d,\dots, k+(q-1)d\}
$$ for integers $q\geq 1$, $k\geq 0$ and $d\geq 1$. A connected bipartite $(p,q)$-graph $G$ admits a $(k,d)$-gracefully total coloring $f$ (as a \emph{public-key}), and admits another $(k,d)$-total coloring $g$ such that $g(x)=f(x)$ for $x\in V(G)$, $f(E(G))=\{k,k+d,k+2d,\dots, k+(q-1)d\}$ and $g(xy)\in S^{\pm}_{q-1,k,d}$ with $g(xy)\neq g(uv)$ for any pair of edges $xy,uv\in E(G)$, as well as $f(xy)+g(xy)=C(k,d)$, where $C(k,d)$ is a function of variables $k,d$. Then we call $g$ a \emph{$(k,d)$-gracefully e-image total coloring} (as a \emph{private-key}), and $\langle f,g\rangle $ a \emph{$(k,d)$-gracefully e-image matching}. \qqed
\end{defn}

\vskip 0.4cm

\textbf{RLA-algorithm for the $(k,d)$-gracefully e-image total coloring.}

\textbf{Input:} A connected bipartite $(p,q)$-graph $G$ admitting a $(k,d)$-gracefully total coloring $f$.

\textbf{Output:} A connected bipartite $(p+m,q+m)$-graph $G^*$ admitting a $(k,d)$-gracefully e-image total coloring, where $G^*$ is obtained by adding $m$ leaves to $G$ randomly.

\textbf{Step 1.} For integer $m\geq 1$, adding $m$ leaves to $G$ produce a connected bipartite $(p+m,q+m)$-graph $G^*$ with vertex set $V(G^*)=X^*\cup Y^*$ such that $X^*\cap Y^*=\emptyset$, $X^*=\{x_i:i\in [1, s\,']\}$ and $Y^*=\{y_j:j\in [1, t\,']\}$ holding $s\,'+t\,'=p+m$, and each edge is $x_iy_j$ with $x_i\in X^*$ and $y_j\in Y^*$.

By the hypothesis that $G^*$ admits a $(k,d)$-gracefully total coloring $g$, then $g:X^*\rightarrow S_{n,0,d}$ and $g:Y^*\cup E(G)\rightarrow S_{q-1,k,d}$ with $k\geq 1$, without loss of generality, there are $0=g(x_1)\leq g(x_i)\leq g(x_{i+1})$ for $i\in [1, s\,'-1]$ and $g(x_{s})\leq g(y_j)\leq g(y_{j+1})\leq f(y_{t})=k+(q+m-1)d$ for $j\in [1, t\,'-1]$. Thereby, $g(x_iy_j)=g(y_j)-g(x_i)$ for $x_iy_j\in E(G^*)$, such that the edge color set $$g(E(G^*))=\{k,k+d,k+2d,\dots ,k+(q+m-1)d\}=S_{q+m-1,k,d}$$

\textbf{Step 2.} We define a new coloring $h$ for $G^*$ as: $h(x_i)=\max g(X^*)+\min g(X^*)-g(x_i)$ for $i\in [1,s\,']$, $h(y_j)=\max g(Y^*)+\min g(Y^*)-g(y_j)$ for $j\in [1,t\,']$, and $h(x_iy_j)=h(y_j)-h(x_i)$, and
$${
\begin{split}
h(x_iy_j)=h(y_j)-h(x_i)&=\max g(Y^*)+\min g(Y^*)-g(y_j)-[\max g(X^*)+\min g(X^*)-g(x_i)]\\
&=\max g(Y^*)+\min g(Y^*)-[\max g(X^*)+\min g(X^*)]-[g(y_j)-g(x_i)]
\end{split}}
$$ so
$$g(x_iy_j)+h(x_iy_j)=\max g(Y^*)+\min g(Y^*)-[\max g(X^*)+\min g(X^*)]$$

\textbf{Step 3.} Return the \emph{$(k,d)$-gracefully e-image total coloring} $h$ of the graph $G^*$.

\vskip 0.4cm

In Fig.\ref{fig:gracefully-image-total}, the connected bipartite graph $G_a$ (as a \emph{public-key}) admits a $(k,d)$-gracefully total coloring $f_a$, the connected bipartite graph $G_b$ (as a \emph{private-key}) admits a $(k,d)$-gracefully e-image total coloring $f_b$, such that $f_a(xy)+f_b(xy)=2k+6d$ for $xy\in E(G_a)=E(G_b)$, since $G_a\cong G_b$. And in Fig.\ref{fig:gracefully-image-total-11}, there is $G_c\cong G_d$ for two connected bipartite graphs $G_c$ and $G_d$, here, the graph $G_c$ (as a \emph{public-key}) admits a $(k,d)$-gracefully total coloring $f_c$, and the graph $G_d$ (as a \emph{private-key}) admits a $(k,d)$-gracefully e-image total coloring $f_d$, we can see $f_c(uv)+f_d(uv)=2k+9d$ for $uv\in E(G_c)=E(G_d)$.

\begin{figure}[h]
\centering
\includegraphics[width=16.4cm]{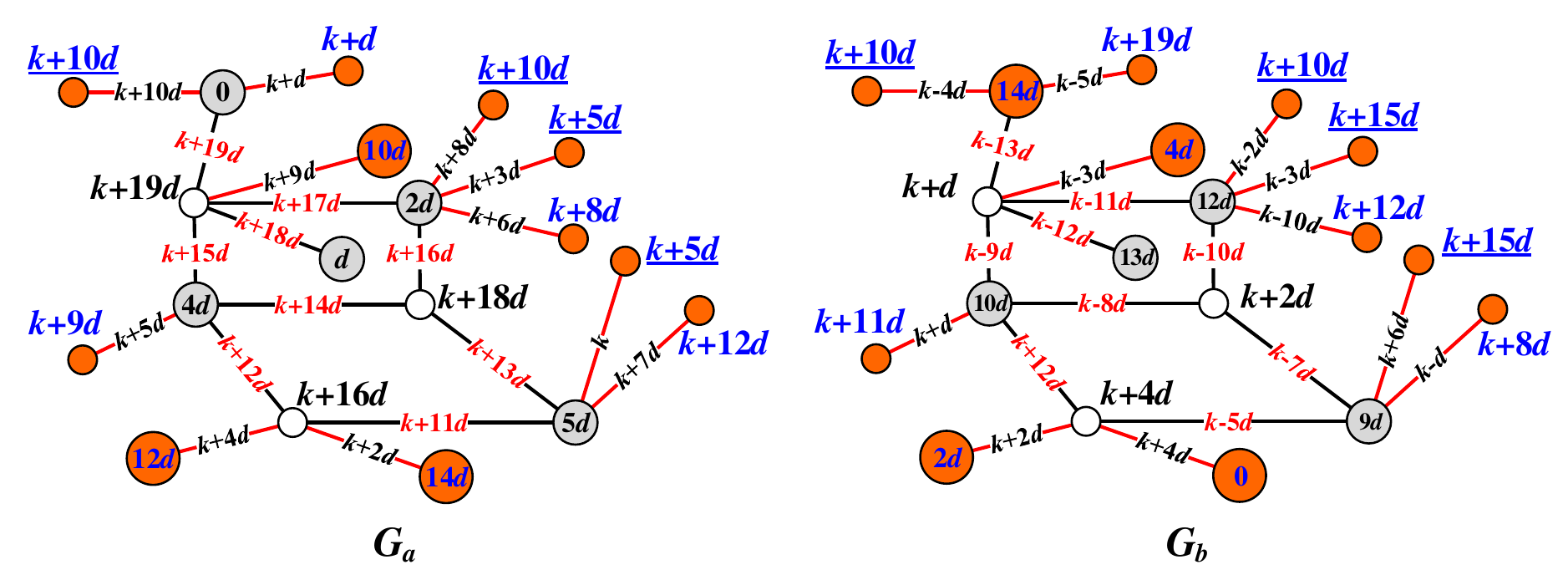}\\
\caption{\label{fig:gracefully-image-total}{\small Two graphs $G_a$ and $G_b$ admit two total colorings $f_a$ and $f_b$ holding $f_a(xy)+f_b(xy)=2k+6d$ for $xy\in E(G_a)=E(G_b)$.}}
\end{figure}

\begin{figure}[h]
\centering
\includegraphics[width=16.4cm]{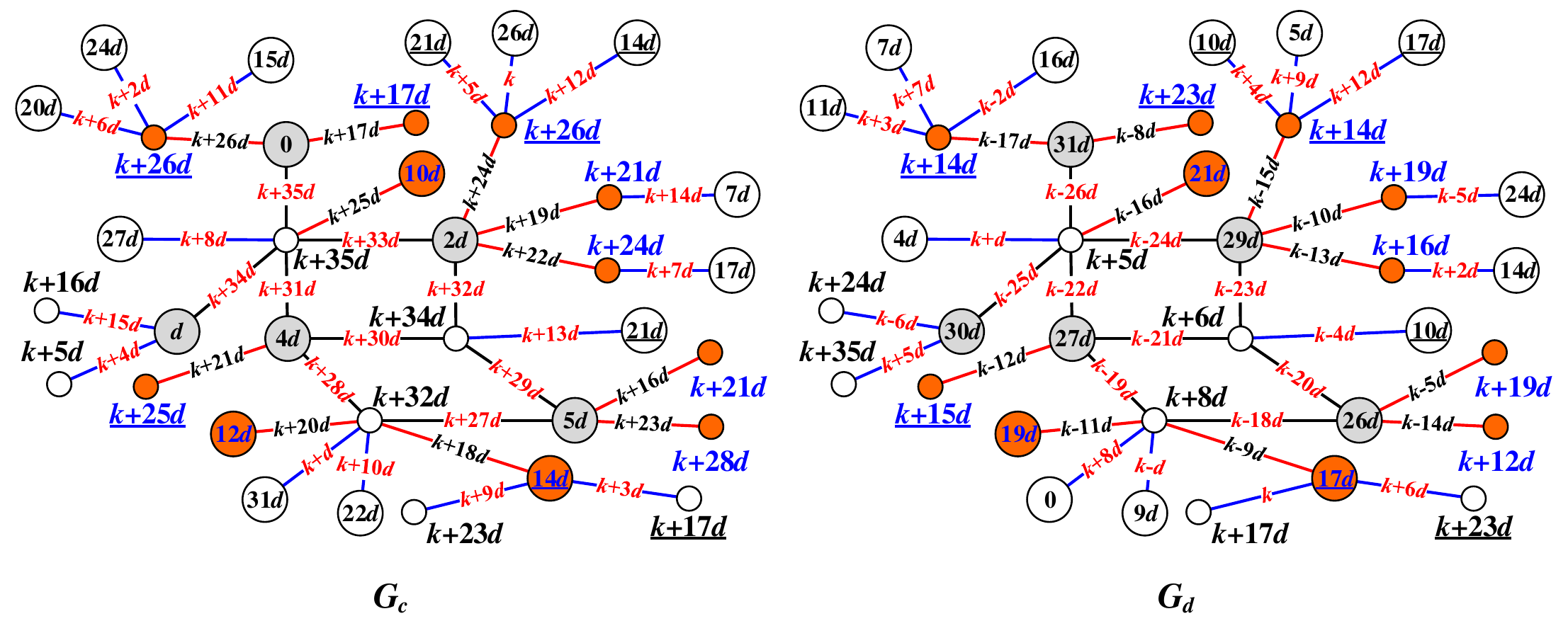}\\
\caption{\label{fig:gracefully-image-total-11}{\small Two graphs $G_c$ and $G_d$ admit two total colorings $f_c$ and $f_d$ holding $f_c(uv)+f_d(uv)=2k+9d$ for $uv\in E(G_c)=E(G_d)$.}}
\end{figure}

By Theorem \ref{thm:permutations-k-d-gracefully-total-coloring}, we have

\begin{cor}\label{thm:k-d-gracefully-e-image-total-coloring}
$^*$ Let $T_0$ be a tree, we have trees $T_j=T_{j-1}-L(T_{j-1})$ for $j\in [1, n(T_0)]$, such that $T_{n(T_0)}$ is just a star $K_{1,n(T_0)}$, and let $c_{j-1}=|L(T_{j-1})|$ for $j\in [1, n(T_0)]$. Then $T_0$ admits at least $N(T_0)$ $(k,d)$-gracefully e-image total colorings for integers $k,d\geq 1$ by the RLA-algorithm for the $(k,d)$-gracefully e-image total coloring, where the number $N(T_0)=\prod ^{n(T_0)-1}_{j=1}(c_{j-1})!$.
\end{cor}

\subsubsection{RLA-algorithm for the strongly edge-magic $(k,d)$-total coloring}

Recall a bipartite and connected $(p,q)$-graph $G$ admits a strongly edge-magic $(k,d)$-total coloring if $f(u)+f(uv)+f(v)=c$ for each edge $uv\in E(G)$, $f(E(G))=S_{q-1,k,0,d}$ and $f(V(G))\subseteq S_{m,0,0,d}\cup S_{q-1,k,0,d}$ defined in Definition \ref{defn:kd-w-type-colorings}, we have

\vskip 0.4cm

\textbf{RLA-algorithm for the strongly edge-magic $(k,d)$-total coloring.}

\textbf{Input:} A connected bipartite $(p,q)$-graph $G$ admitting a $(k,d)$-gracefully total coloring $f$.

\textbf{Output:} A connected bipartite $(p+m,q+m)$-graph $G^*$ admitting a strongly edge-magic $(k,d)$-total coloring, where $G^*$ is obtained by adding $m$ leaves to $G$ randomly.

\textbf{Step 1.} For integer $m\geq 1$, adding $m$ leaves to $G$ produce a connected bipartite $(p+m,q+m)$-graph $G^*$ with vertex set $V(G^*)=X^*\cup Y^*$ such that $X^*\cap Y^*=\emptyset$, $X^*=\{x_i:i\in [1, s\,']\}$ and $Y^*=\{y_j:j\in [1, t\,']\}$ holding $s\,'+t\,'=p+m$, and each edge is $x_iy_j$ with $x_i\in X^*$ and $y_j\in Y^*$.

By the hypothesis that $G^*$ admits a $(k,d)$-gracefully total coloring $g$, then $g:X^*\rightarrow S_{n,0,0,d}$ and $g:Y^*\cup E(G)\rightarrow S_{q-1,k,0,d}$ with $k\geq 1$, without loss of generality, there are $0=g(x_1)\leq g(x_i)\leq g(x_{i+1})$ for $i\in [1, s\,'-1]$ and $g(x_{s})\leq g(y_j)\leq g(y_{j+1})\leq f(y_{t})=k+(q+m-1)d$ for $j\in [1, t\,'-1]$. Thereby, $g(x_iy_j)=g(y_j)-g(x_i)$ for $x_iy_j\in E(G^*)$, such that $g(E(G^*))=\{k,k+d,k+2d,\dots ,k+(q+m-1)d\}=S_{q+m-1,k,0,d}$.

\textbf{Step 2.} We define a new coloring $\alpha$ for $G^*$ as: $\alpha(x_i)=\max g(X^*)+\min g(X^*)-g(x_i)$ for $i\in [1,s\,']$, $\alpha(y_j)=g(y_j)$ for $j\in [1,t\,']$, and $\alpha(x_iy_j)=\max g(E(G^*))+\min g(E(G^*))-g(x_iy_j)$, and
$${
\begin{split}
\alpha(x_i)+\alpha(x_iy_j)+\alpha(y_j)=&\max g(Y^*)+\min g(Y^*)-g(x_i)\\
&+[\max g(E(G^*))+\min g(E(G^*))-g(x_iy_j)]+g(y_j)\\
=&\max g(Y^*)+\min g(Y^*)+\max g(E(G^*))+\min g(E(G^*))
\end{split}}
$$ is a constant $c(k,d)$.

\textbf{Step 3.} Return the strongly edge-magic $(k,d)$-total coloring $\alpha$ of the graph $G^*$.

\vskip 0.4cm

\begin{figure}[h]
\centering
\includegraphics[width=16.4cm]{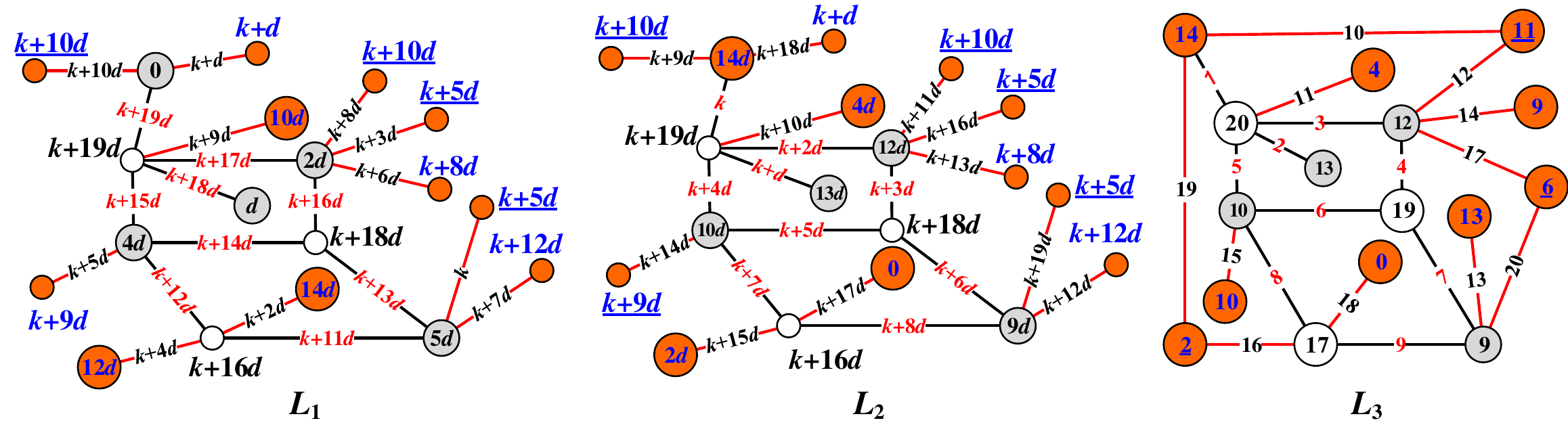}\\
\caption{\label{fig:strongly-edge-magic-k-d-total}{\small An example for illustrating the RLA-algorithm for the strongly edge-magic $(k,d)$-total coloring.}}
\end{figure}

\begin{cor}\label{thm:strongly-edge-magic-k-d-total-coloring}
$^*$ Let $T_0$ be a tree, we have trees $T_j=T_{j-1}-L(T_{j-1})$ for $j\in [1, n(T_0)]$, such that $T_{n(T_0)}$ is just a star $K_{1,n(T_0)}$, and let $c_{j-1}=|L(T_{j-1})|$ for $j\in [1, n(T_0)]$. Then $T_0$ admits at least $N(T_0)$ strongly edge-magic $(k,d)$-total colorings for integers $k,d\geq 1$ by the RLA-algorithm for the strongly edge-magic $(k,d)$-total coloring, where the number $N(T_0)=\prod ^{n(T_0)-1}_{j=1}(c_{j-1})!$.
\end{cor}

\begin{thm}\label{thm:666666}
Each tree admits an edge-magic total coloring.
\end{thm}

\subsection{Random parameters, graph increasing techniques}

\subsubsection{Dynamic parameter colorings and labelings}

\begin{defn} \label{defn:3-parameter-labelings}
\cite{Gallian2021} Let $G_i$ be a $(p,q)$-graph for $i\in [1,m]$, and integers $k_i\geq 1$ and $d_i\geq 1$.

(1) A \emph{$(k_i,d_i)$-graceful labeling} $f$ of $G_i$ hold $f(V(G_i))\subseteq [0, k_i + (q-1)d_i]$, $f(x)\neq f(y)$ for distinct $x,y\in V(G_i)$ and $\pi(E(G_i))=\{|\pi(u)-\pi(v)|:\ uv\in
E(G_i)\}=S(k_i,d_i)^q_1$.

(2) A labeling $f$ of $G_i$ is said to be $(k_i,d_i)$-\emph{arithmetic} if $f(V(G_i))\subseteq [0, k_i+(q-1)d_i]$, $f(x)\neq f(y)$ for distinct vertices $x,y\in V(G_i)$ and $\{f(u)+f(v): uv\in E(G_i)\}=S(k_i,d_i)^q_1$.

(3) A $(k_i,d_i)$-\emph{edge antimagic total labeling} $f$ of $G_i$ hold $f(V(G_i)\cup E(G_i))=[1,p+q]$ and $\{f(u)+f(v)+f(uv): uv\in E(G_i)\}=S(k_i,d_i)^q_1$, and furthermore $f$ is \emph{super} if $f(V(G_i))=[1,p]$.

(4) A \emph{$(k_i,d_i)$-harmonious labeling} of a $(p,q)$-graph $G_i$ is defined by a mapping $h:V(G)\rightarrow [0,k+(q-1)d_i]$ with $k_i,d_i\geq 1$, such that $f(x)\neq f(y)$ for any pair of vertices $x,y$ of $G$, $h(u)+h(v)(\bmod^*~qd_i)$ means that $h(uv)-k=[h(u)+h(v)-k](\bmod~qd_i)$ for each edge $uv\in E(G)$, and the edge color set $h(E(G))=S(k_i,d_i)^q_1$ holds true.

(5) In \cite{Acharya-Hegde-Arithmetic-1990}, a labeling $f$ of $G$ is said to be \emph{$(k,d)$-arithmetic} if the vertex labels are distinct nonnegative integers and the edge labels induced by $f(x) + f(y)$ for each edge $xy$ are $k,k+d,k+2d,\dots,k+(q-1)d$.\qqed
\end{defn}

\begin{rem}\label{rem:333333}
We can use $(k,d)$-harmonious labelings of graphs to design more complicated Topsnut-gpws. Let $\{\langle k_i,d_i\rangle \}_1^n$ be a sequence obtained from two sequences $\{k_1,k_2,\dots ,k_n\}$ and $\{d_1,d_2,\dots ,d_n\}$. We have a Topsnut-gpw set $$F(k_i,d_i)=\{G:~G\textrm{ admits a $(k_i,d_i)$-harmonious labeling}\}$$ for each matching $\langle k_i,d_i\rangle $ with $i\in [1,n]$. There are random \emph{Topsnut-chains} $G_{i_1},G_{i_2}, \dots ,G_{i_m}$ with $G_{i_j}\in \bigcup ^n_{i=1}F(k_i,d_i)$, such that we can apply them to encrypt a network once time.\paralled
\end{rem}

\subsubsection{Random sequence techniques}

\begin{defn} \label{defn:gracefully-total-sequence-coloring}
\cite{Yao-Su-Wang-Hui-Sun-ITAIC2020} Let $G$ be a $(p, q)$-graph, and let an integer sequence $A_M=\{a_i\}_1^M$ hold $0\leq a_i< a_{i+1}$ for $i\in [1, M-1]$ and $p\leq M$, and let $B_q=\{b_j\}_1^q$ be another integer sequence holding $1\leq b_j< b_{j+1}$ for $j\in [1, q-1]$. $G$ admits a \emph{gracefully total sequence coloring} $f:V(G)\cup E(G)\rightarrow A_M\cup B_q$ if $f(u)\neq f(v)$ for any edge $uv\in E(G)$, $f(uv)=|f(u)-f(v)|$ for each edge $uv\in E(G)$ and $f(E(G))=B_q$. Moreover, this gracefully total sequence coloring $f$ is \emph{proper} if $f(u)\neq f(uv)$ and $f(v)\neq f(uv)$ for any edge $uv\in E(G)$\qqed
\end{defn}

We say a sentence ``a \emph{$W$-type sequence coloring}'' to represent one of sequence colorings with no ``set-ordered'', and employ another sentence ``a \emph{set-ordered $W$-type sequence coloring}'' to stand for one of set-ordered sequence colorings defined in \cite{Bing-Yao-2020arXiv}. We have the parameterized labelings: the $(k,d)$-arithmetic labeling, the $(k,d)$-graceful labeling, the $(k,d)$-harmonious labeling and the new $(k,d)$-type labelings introduced in \cite{Sun-Zhang-Yao-ICMITE-2017}. Based on the well-defined parameterized labelings, we present the following parameterized total colorings:

\begin{defn} \label{defn:kd-w-type-colorings}
\cite{Yao-Su-Wang-Hui-Sun-ITAIC2020} Let $G$ be a bipartite and connected $(p,q)$-graph, then its vertex set $V(G)=X\cup Y$ with $X\cap Y=\emptyset$ such that each edge $uv\in E(G)$ holds $u\in X$ and $v\in Y$. Let integers $a,k,m\geq 0$, $d\geq 1$ and $q\geq 1$ in this subsection. We have two parameterized sets
$$S_{m,k,a,d}=\{k+ad,k+(a+1)d,\dots ,k+(a+m)d\},O_{2q-1,k,d}=\{k+d,k+3d,\dots ,k+(2q-1)d\}
$$ The \emph{cardinality} $|S_{m,k,a,d}|=m+1$ and $|O_{2q-1,k,d}|=q$. If there is a coloring $f:X\rightarrow S_{m,0,0,d}=\{0,d,\dots ,md\}$ and $f:Y\cup E(G)\rightarrow S_{n,k,0,d}=\{k,k+d,\dots ,k+nd\}$ with integers $k\geq 0$ and $d\geq 1$, here it is allowed $f(x)=f(y)$ for some distinct vertices $x,y\in V(G)$. Let $c$ be a non-negative integer.
\begin{asparaenum}[\textrm{Ptol}-1. ]
\item If $f(uv)=|f(u)-f(v)|$ for $uv\in E(G)$, $f(E(G))=S_{q-1,k,0,d}$ and $f(V(G)\cup E(G))\subseteq S_{m,0,0,d}\cup S_{q-1,k,0,d}$, then $f$ is called a \emph{$(k,d)$-gracefully total coloring}; and moreover $f$ is called a \emph{$(k,d)$-strongly gracefully total coloring} if $f(x)+f(y)=k+(q-1)d$ for each matching edge $xy$ of a matching $M$ of $G$.
\item If $f(uv)=|f(u)-f(v)|$ for $uv\in E(G)$, $f(E(G))=O_{2q-1,k,d}$ and $f(V(G)\cup E(G))\subseteq S_{m,0,0,d}\cup S_{2q-1,k,0,d}$, then $f$ is called a \emph{$(k,d)$-odd-gracefully total coloring}; and moreover $f$ is called a \emph{$(k,d)$-strongly odd-gracefully total coloring} if $f(x)+f(y)=k+(2q-1)d$ for each matching edge $xy$ of a matching $M$ of $G$.
\item If the color set $$\{f(u)+f(uv)+f(v):uv\in E(G)\}=\{2k+2ad,2k+2(a+1)d,\dots ,2k+2(a+q-1)d\}$$ with $a\geq 0$ and $f(V(G)\cup E(G))\subseteq S_{m,0,0,d}\cup S_{2(a+q-1),k,a,d}$, then $f$ is called a \emph{$(k,d)$-edge antimagic total coloring}.
\item If $f(uv)=f(u)+f(v)~(\bmod^*qd)$ defined by $f(uv)-k=[f(u)+f(v)-k](\bmod ~qd)$ for $uv\in E(G)$, and $f(E(G))=S_{q-1,k,0,d}$, then $f$ is called a \emph{$(k,d)$-harmonious total coloring}.
\item If $f(uv)=f(u)+f(v)~(\bmod^*qd)$ defined by $f(uv)-k=[f(u)+f(v)-k](\bmod ~qd)$ for $uv\in E(G)$, and $f(E(G))=O_{2q-1,k,d}$, then $f$ is called a \emph{$(k,d)$-odd-elegant total coloring}.
\item If $f(u)+f(uv)+f(v)=c$ for each edge $uv\in E(G)$, $f(E(G))=S_{q-1,k,0,d}$ and $f(V(G))\subseteq S_{m,0,0,d}\cup S_{q-1,k,0,d}$, then $f$ is called a \emph{strongly edge-magic $(k,d)$-total coloring}; and moreover $f$ is called an \emph{edge-magic $(k,d)$-total coloring} if $|f(E(G))|\leq q$ and $f(u)+f(uv)+f(v)=c$ for $uv\in E(G)$.
\item If $f(uv)+|f(u)-f(v)|=c$ for each edge $uv\in E(G)$ and $f(E(G))=S_{q-1,k,0,d}$, then $f$ is called a \emph{strongly edge-difference $(k,d)$-total coloring}; and moreover $f$ is called an \emph{edge-difference $(k,d)$-total coloring} if $|f(E(G))|\leq q$ and $f(uv)+|f(u)-f(v)|=c$ for $uv\in E(G)$.
\item If $|f(u)+f(v)-f(uv)|=c$ for each edge $uv\in E(G)$ and $f(E(G))=S_{q-1,k,0,d}$, then $f$ is called a \emph{strongly felicitous-difference $(k,d)$-total coloring}; and moreover $f$ is called a \emph{felicitous-difference $(k,d)$-total coloring} if $|f(E(G))|\leq q$ and $|f(u)+f(v)-f(uv)|=c$ for $uv\in E(G)$.
\item If $\big ||f(u)-f(v)|-f(uv)\big |=c$ for each edge $uv\in E(G)$ and $f(E(G))=S_{q-1,k,0,d}$, then $f$ is called a \emph{strongly graceful-difference $(k,d)$-total coloring}; and $f$ is called a \emph{graceful-difference $(k,d)$-total coloring} if $|f(E(G))|\leq q$ and $\big ||f(u)-f(v)|-f(uv)\big |=c$ for $uv\in E(G)$.\qqed
\end{asparaenum}
\end{defn}

\begin{defn} \label{defn:properk-d-total-colorings}
$^*$ A total coloring $f$ defined in Definition \ref{defn:kd-w-type-colorings} is \emph{proper} if $f(u)\neq f(v)$ for each edge $uv\in E(G)$, and $f(uv)\neq f(uw)$ for any two adjacent edges $uv,uw\in E(G)$, so we call $f$ a \emph{$W$-type proper $(k,d)$-total coloring} of $G$.\qqed
\end{defn}

\begin{rem} \label{rem:kd-w-tupe-colorings-definition}
In general, we call $f$ a \emph{$W$-type $(k,d)$-total coloring} if the constraint function $F(f(u)$, $f(uv)$, $f(v))=0$ and the edge color set $f(E(G))$ hold one of the well-defined total colorings defined in Definition \ref{defn:kd-w-type-colorings}.
\begin{asparaenum}[(1) ]
\item We have some new parameters of graphs based on Definition \ref{defn:kd-w-type-colorings}. For a graph $G$, we have two $W$-type $(k,d)$-total colorings $g_{\min}$ and $g_{\max}$ of $G$ such that
 $$|g_{\min}(V(G))|\leq |g(V(G))|\leq |g_{\max}(V(G))|$$ for each $W$-type $(k,d)$-total coloring $g$ of $G$. As this $W$-type $(k,d)$-total coloring is the graceful $(k,d)$-total coloring, then a graceful $(1,1)$-total coloring $g_{\max}$ of a tree $T$ means that $|g_{\max}(V(T))|$ is the approximate value of a graceful labeling of $T$.

\quad The graceful tree conjecture says: $|g_{\max}(V(T))|=|V(T)|$.
\item In application, we may reduce the restrictions of some $W$-type $(k,d)$-total colorings of graphs, such as a $W$-type $(k,d)$-total coloring $g$ of a bipartite graph $G$ with bipartition $(X,Y)$ holds $g:X\rightarrow \{-md,-(m-1)d,\dots ,-2d,-d\}\cup \{0,d,\dots ,md\}$ and
 $${
\begin{split}
 g:&Y\cup E(G)\rightarrow \{-k,-k-d,\dots ,-k-nd\}\cup \{-k+d,-k+2d,\dots ,-k+nd\}\\
 &\cup \{k-d,k-2d,\dots ,k-nd\}\cup \{k,k+d,\dots ,k+nd\}
 \end{split}}
 $$ with $m,n,k,d>0$.\paralled
\end{asparaenum}
\end{rem}

\begin{thm}\label{thm:graceful-total-sequence-coloring}
\cite{Yao-Su-Wang-Hui-Sun-ITAIC2020} Every tree $T$ with diameter $D(T)\geq 3$ and $s+1=\left \lceil \frac{D(T)}{2}\right \rceil $ admits at least $2^{s}$ different \emph{gracefully total sequence colorings} (refer to Definition \ref{defn:gracefully-total-sequence-coloring}) if two sequences $A_M, B_q$ holding $0<b_j-a_i\in B_q$ for any pair of $a_i\in A_M$ and $b_j\in B_q$.
\end{thm}

\begin{rem}\label{rem:333333}
Let an integer sequence $A_M=\{a_i\}_1^M$ hold $0\leq a_i< a_{i+1}$ for $i\in [1, M-1]$ and $p\leq M$, and let $B_q=\{b_j\}_1^q$ be another integer sequence holding $1\leq b_j< b_{j+1}$ for $j\in [1, q-1]$. We call $(A_M,B_q)$ a \emph{sequence-ordered matching} if $A_M$ and $B_q$ hold $0<b_j-a_i\in B_q$ for any pair of $a_i\in A_M$ and $b_j\in B_q$ according to Theorem \ref{thm:graceful-total-sequence-coloring}.

Consider $A_M=\{a_i\}_1^M$ as a \emph{public-key} and $B_q=\{b_j\}_1^q$ as a \emph{private-key}, so a tree $T$ of $q$ edges admitting a gracefully total sequence coloring defined on the sequence-ordered matching $(A_M,B_q)$ is as a topological authentication of multiple variables, such that the topological authentication $\textbf{T}_{\textbf{a}}\langle\textbf{X},\textbf{Y}\rangle$ has more complexity to resisting attackers, simultaneously, $\textbf{T}_{\textbf{a}}\langle\textbf{X},\textbf{Y}\rangle$ is easier to be made by users.\paralled
\end{rem}

\begin{lem}\label{thm:adding-leaves-keep-sequence-colorings}
\cite{Yao-Su-Wang-Hui-Sun-ITAIC2020} Suppose that a bipartite and connected graph $G$ admits a \emph{gracefully total sequence coloring} based on two sequences $A_M, B_q$ holding $0<b_j-a_i\in B_q$ for $a_i\in A_M$ and $b_j\in B_q$, then a new bipartite and connected graph obtained by adding randomly leaves to $G$ admits a \emph{gracefully total sequence coloring} based on two sequences $A\,'_M, B\,'_q$ holding $0<b\,'_j-a\,'_i\in B\,'_q$ for any pair of $a\,'_i\in A\,'_M$ and $b\,'_j\in B\,'_q$ (refer to Definition \ref{defn:gracefully-total-sequence-coloring}).
\end{lem}

\begin{thm}\label{thm:each-tree-sequence-coloring}
\cite{Yao-Su-Wang-Hui-Sun-ITAIC2020} Each tree on $q$ edges admits a \emph{proper graceful total sequence coloring} based on two sequences $A_M$ and $B_q$ holding $0<b_j-a_i\in B_q$ for $a_i\in A_M$ and $b_j\in B_q$ (refer to Definition \ref{defn:gracefully-total-sequence-coloring}).
\end{thm}

\begin{defn} \label{defn:44-planar-dual-graph}
Let $F(H)=\{f_0,f_1,f_2,\dots,f_m\}$ be the set of all faces of a planar graph $H$, where $f_0$ is the infinite face (outer face) of $H$, and each $f_i$ is an inner face of $H$ for $i\in [1,m]$. The \emph{planar dual graph} of $H$ is denoted as $D_{ual}(H)$, where $V(D_{ual}(H))=F(H)$ and a vertex $f_i$ is adjacent with another vertex $f_j$ in $D_{ual}(H)$ if and only if two faces $f_i$ and $f_j$ have the common part of their boundaries in $H$. We call $\langle H,D_{ual}(H)\rangle $ a \emph{planar dual matching}. \qqed
\end{defn}

\begin{defn} \label{defn:44-planar-dual-graph-colorings}
$^*$ A \emph{ve-set e-proper coloring} $\theta$ of the planar dual graph $D_{ual}(G)$ of a maximal planar graph $G$ is defined as: Suppose that the maximal planar graph $G$ admits a proper vertex coloring $g:V(G)\rightarrow [1,k]$. Since each edge $f_if_j\in E(D_{ual}(G))$ means that two faces $f_i$ and $f_j$ have a common edge of their boundaries in $G$, so $\theta(f_j)=\{g(u_{j,i}): i\in [1,m(f_j)]\}\in [1,k]^2$ for each vertex $f_i\in V(D_{ual}(G))$, where $m(f_j)$ is the number of vertices of the boundary of face $f_i$ of $G$, and each edge $f_if_j\in E(D_{ual}(G))$ is colored with $\theta(f_if_j)\subseteq \theta(f_i)\cup \theta(f_j)$, such that $\theta(uv)\neq \theta(uw)$ for two adjacent edges $uv,uw\in E(D_{ual}(G))$.\qqed
\end{defn}

\begin{thm}\label{thm:ve-set-e-proper-coloring-planar-dual}
The planar dual graph $D_{ual}(G)$ of a maximal planar graph $G$ admits a ve-set e-proper coloring defined on $[1,4]^2$ (refer to Definition \ref{defn:44-planar-dual-graph-colorings}) if and only if $G$ is a 4-colorable.
\end{thm}

\begin{example}\label{exa:8888888888}
Fig.\ref{fig:dual-3-regular-plane} is for understanding Definition \ref{defn:44-planar-dual-graph} and Definition \ref{defn:44-planar-dual-graph-colorings}, in which $G$ is a 4-colored maximal planar graph, and $G_{\textrm{3r}}$ is a 3-regular planar graph, they form a \emph{planar dual matching}. Moreover, the 3-regular planar graph $G_{\textrm{3r}}$ admits a \emph{ve-set e-proper coloring} $h$ such that $h(x)=\{i,j,k\}\subset \{1,2,3,4\}$ for each vertex $x\in V(G_{\textrm{3r}})$, $h(uv)\subseteq h(u)\cap h(v)$ for each edge $uv\in E(G_{\textrm{3r}})$, as well as $h(uv)\neq h(uw)$ for two adjacent edges $uv,uw\in E(G_{\textrm{3r}})$. Obviously, a Topcode-matrix $T_{code}(G_{\textrm{3r}})$ can make number-based strings with longer bytes.
\end{example}

\begin{figure}[h]
\centering
\includegraphics[width=16.4cm]{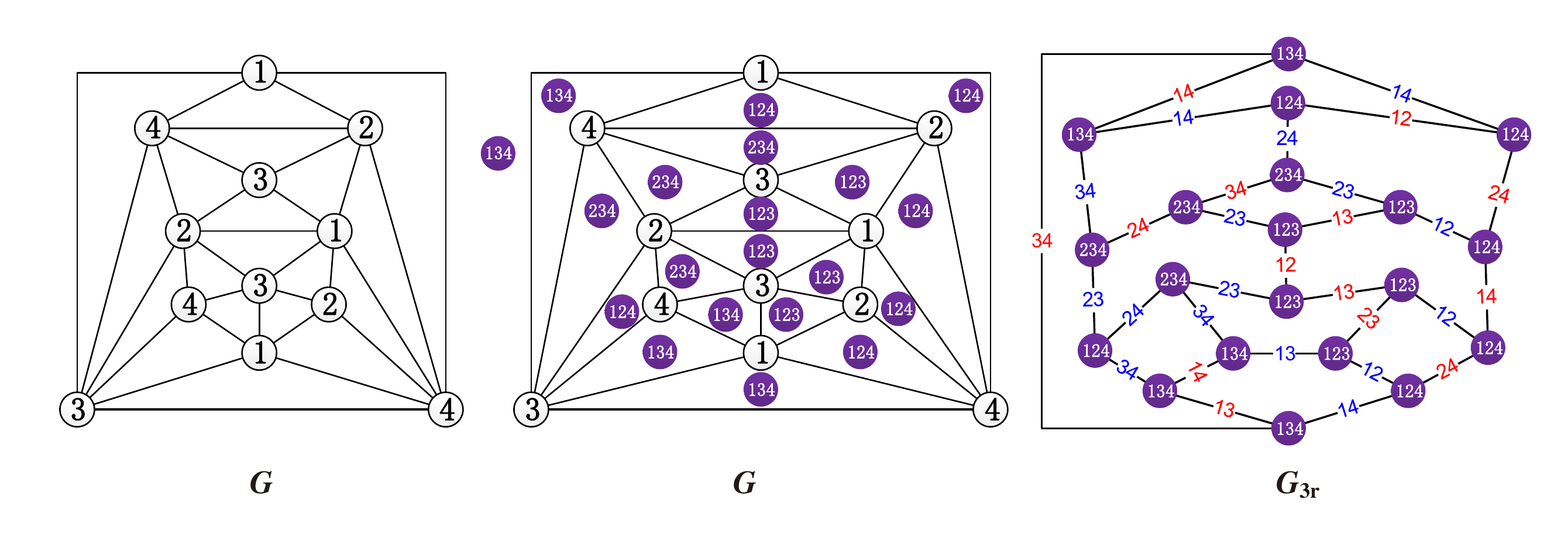}\\
\caption{\label{fig:dual-3-regular-plane} {\small Two planar graphs $G$ and $G_{\textrm{3r}}$ are a \emph{planar dual matching}, refer to Definition \ref{defn:44-planar-dual-graph} and Definition \ref{defn:44-planar-dual-graph-colorings}.}}
\end{figure}

\subsubsection{Edge-magic-type total colorings as randomly increasing techniques}

\begin{defn} \label{defn:combinatoric-definition-total-coloring}
\cite{Bing-Yao-2020arXiv} For a proper total coloring $f:V(G)\cup E(G)\rightarrow [1,M]$ of a graph $G$, we define an \emph{edge-function} $c_f(uv)$ for each edge $uv\in E(G)$, and then have a parameter
\begin{equation}\label{eqa:edge-difference-total-coloring}
B^*_{\alpha}(G,f,M)=\max_{uv \in E(G)}\{c_f(uv)\}-\min_{xy \in E(G)}\{c_f(xy)\}
\end{equation}
If $B^*_{\alpha}(G,f,M)=0$, we call $f$ an \emph{$\alpha$-proper total coloring} of $G$, the smallest number
\begin{equation}\label{eqa:minimum}
\chi\,''_{\alpha}(G) =\min_f \{M:~B^*_{\alpha}(G,f,M)=0\}
\end{equation}
over all $\alpha$-proper total colorings of $G$ is called \emph{$\alpha$-proper total chromatic number}, and $f$ is called a \emph{perfect $\alpha$-proper total coloring} if $\chi\,''_{\alpha}(G)=\chi\,''(G)$. \textbf{Moreover}
\begin{asparaenum}[\textrm{Tcoloring}-1. ]
\item We call $f$ a \emph{(resp. perfect) \textbf{edge-magic proper total coloring}} of $G$ if $c_f(uv)=f(u)+f(v)+f(uv)$, rewrite $B^*_{\alpha}(G,f,M)=B^*_{emt}(G,f$, $M)$, and $\chi\,''_{\alpha}(G)=\chi\,''_{emt}(G)$ is called \emph{edge-magic total chromatic number} of $G$.
\item We call $f$ a \emph{(resp. perfect) \textbf{edge-difference proper total coloring}} of $G$ if $c_f(uv)=f(uv)+|f(u)-f(v)|$, rewrite $B^*_{\alpha}(G,f,M)=B^*_{edt}(G,f$, $M)$, and $\chi\,''_{\alpha}(G)=\chi\,''_{edt}(G)$ is called \emph{edge-difference total chromatic number} of $G$.
\item We call $f$ a \emph{(resp. perfect) \textbf{felicitous-difference proper total coloring}} of $G$ if $c_f(uv)=|f(u)+f(v)-f(uv)|$, rewrite $B^*_{\alpha}(G,f,M)=B^*_{fdt}(G,f,M)$, and $\chi\,''_{\alpha}(G)=\chi\,''_{fdt}(G)$ is is called \emph{ felicitous-difference total chromatic number} of $G$.
\item We refer to $f$ a \emph{(resp. perfect) \textbf{graceful-difference proper total coloring}} of $G$ if $c_f(uv)=\big ||f(u)-f(v)|-f(uv)\big |$, rewrite $B^*_{\alpha}(G,f,M)=B^*_{gdt}(G,f,M)$, and $\chi\,''_{\alpha}(G)=\chi\,''_{gdt}(G)$ is called \emph{graceful-difference total chromatic number} of $G$.\qqed
\end{asparaenum}
\end{defn}

\begin{rem}\label{rem:333333}
In Fig.\ref{fig:spider-edge-magic-22} and Fig.\ref{fig:random-growing-edge-difference}, we can see \emph{randomly growing graphs} admitting edge-magic proper total colorings, or edge-difference proper total colorings defined in Definition \ref{defn:combinatoric-definition-total-coloring}. Clearly, these randomly growing graphs provide us \emph{randomly growing number-based strings} for encrypting or decrypting digital files in dynamic networks.

The notation $\{G_k\}^n_{k=1}$ stands for a sequence of colored graphs, such that each graph $G_k\in \{G_k\}^n_{k=1}$ admits a $W$-type proper total coloring, where $W$-type is one of edge-magic, edge-difference, felicitous-difference and graceful-difference defined in Definition \ref{defn:combinatoric-definition-total-coloring}, and $G_k\subset G_{k+1}$ for $k\in [1,n-1]$. So, $\{G_k\}^n_{k=1}$ forms a \emph{topological authentication chain}.\paralled
\end{rem}

\begin{figure}[h]
\centering
\includegraphics[width=16.4cm]{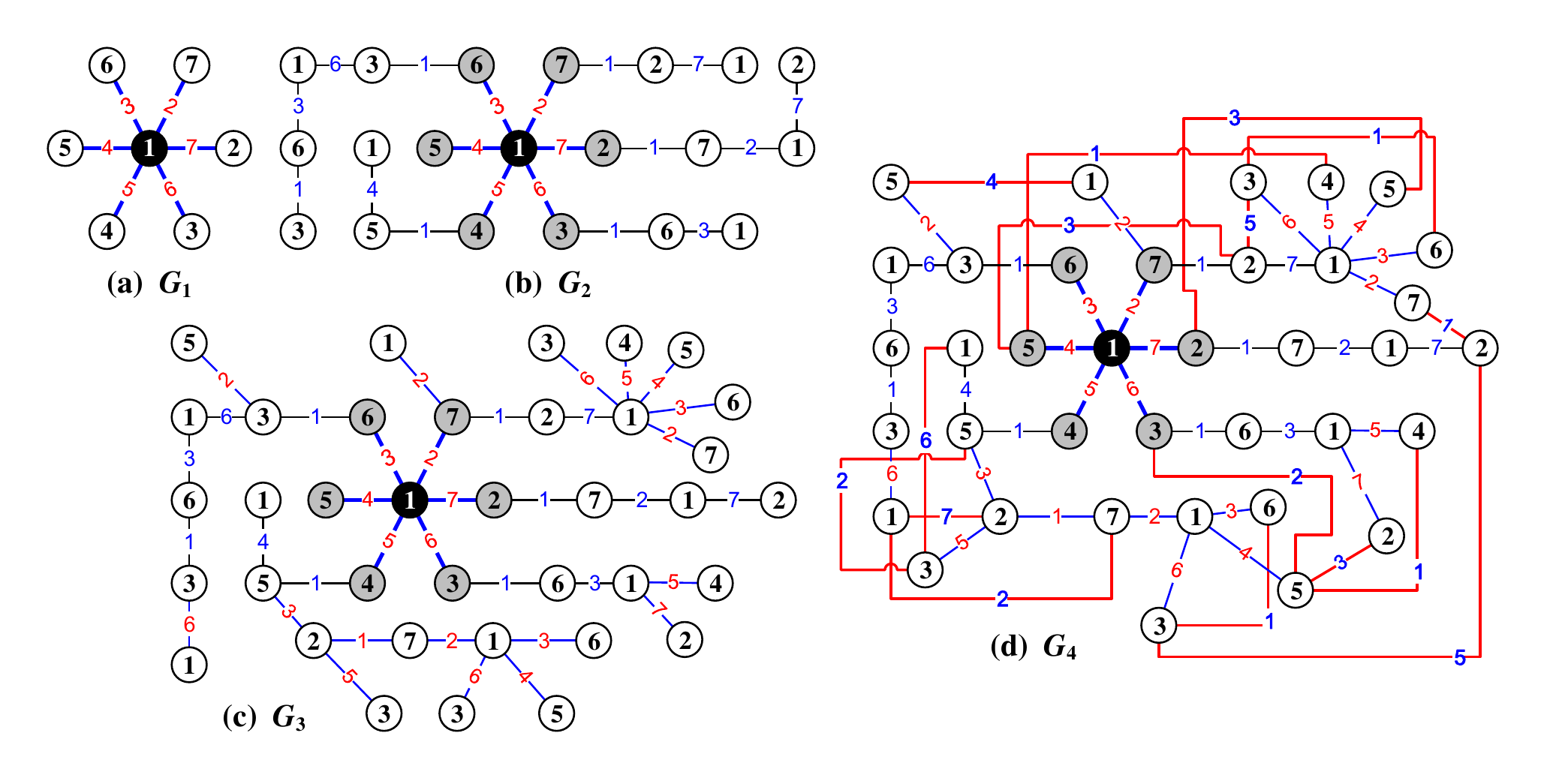}\\
\caption{\label{fig:spider-edge-magic-22}{\small Four graphs admitting the edge-magic proper total colorings holding $f(u)+f(uv)+f(v)=10$, where (a) is a star $K_{1,6}$, called a \emph{root}; (b) is a \emph{rooted spider} $S_{1,3,3,4,3,5}$; (c) is a rooted tree; (d) is a non-planar rooted graph, cited from \cite{Bing-Yao-2020arXiv}.}}
\end{figure}

\begin{figure}[h]
\centering
\includegraphics[width=16cm]{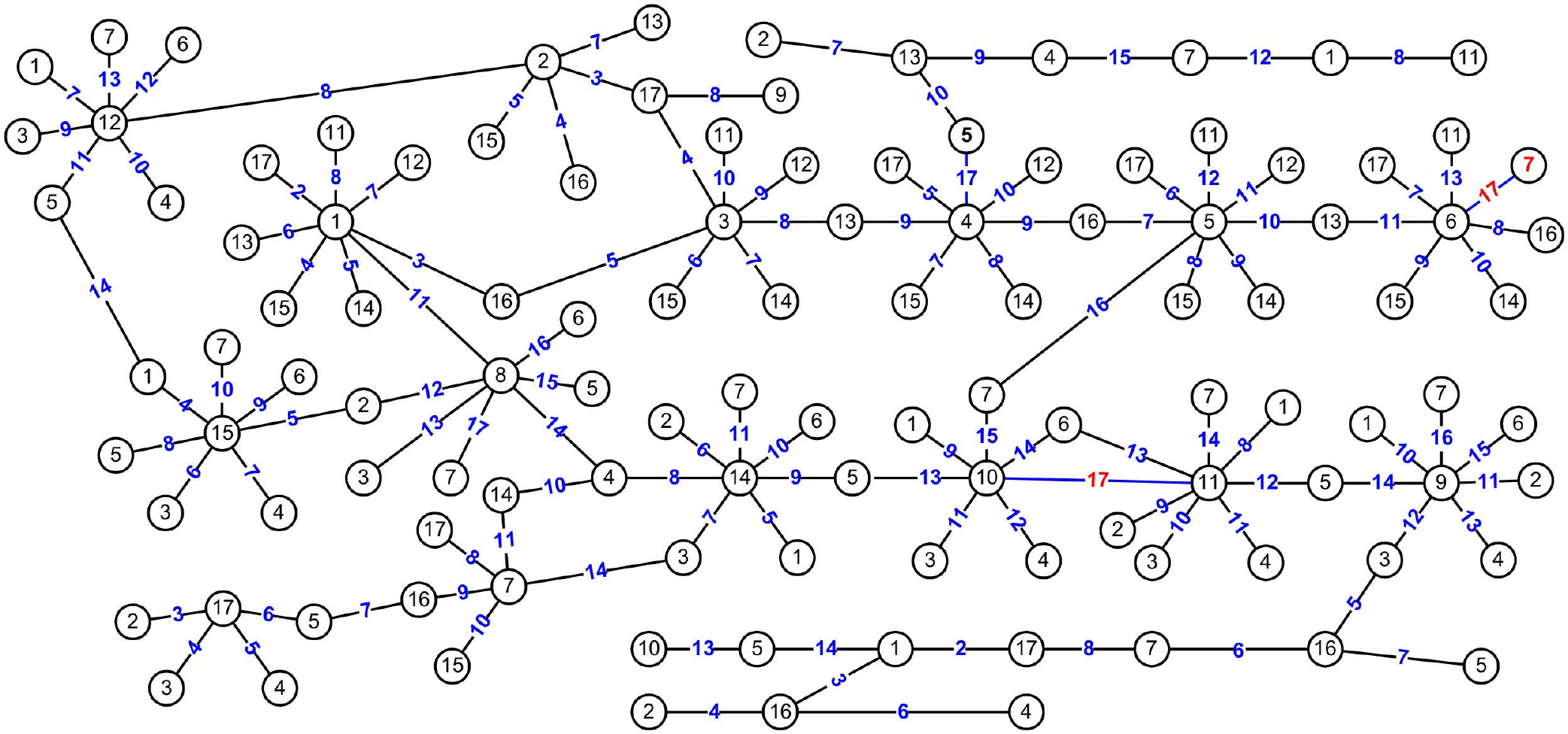}\\
\caption{\label{fig:random-growing-edge-difference} {\small A randomly growing graph $G$ admits an edge-difference proper total coloring holding $f(uv)+|f(u)-f(v)|=18$ for each edge $uv\in E(G)$, cited from \cite{Bing-Yao-2020arXiv}.}}
\end{figure}

\begin{defn} \label{defn:4-dual-total-coloring}
\cite{Bing-Yao-2020arXiv} The \emph{dual colorings} of the colorings defined in Definition \ref{defn:combinatoric-definition-total-coloring-abc} as $(a,b,c)=(1,1,1)$ are in the following:
\begin{asparaenum}[\textrm{\textbf{Dual}}-1. ]
\item For an \emph{edge-magic proper total coloring} $f_{em}$ of a graph $G$, so there exists a constant $k$ such that $f_{em}(u)+f_{em}(uv)+f_{em}(v)=k$ for each edge $uv\in E(G)$. Let $\max f_{em}=\max \{f_{em}(w):w\in V(G)\cup E(G)\}$ and $\min f_{em}=\min \{f_{em}(w):w\in V(G)\cup E(G)\}$. We have the \emph{dual coloring} $g_{em}$ of $f_{em}$ defined as: $g_{em}(w)=(\max f_{em}+\min f_{em})-f_{em}(w)$ for $w\in V(G)\cup E(G)$, and then
\begin{equation}\label{eqa:f-em}
{
\begin{split}
g_{em}(u)+g_{em}(uv)+g_{em}(v)&=3(\max f_{em}+\min f_{em})-[f_{em}(u)+f_{em}(uv)+f_{em}(v)]\\
&=3(\max f_{em}+\min f_{em})-k=k\,'
\end{split}}
\end{equation} for each edge $uv\in E(G)$.
\item For an \emph{edge-difference proper total coloring} $f_{ed}$ of a graph $G$, we have a constant $k$ holding $f_{ed}(uv)+|f_{ed}(u)-f_{ed}(v)|=k$ for each edge $uv\in E(G)$. Let $\max f_{ed}=\max \{f_{ed}(w):w\in V(G)\cup E(G)\}$ and $\min f_{ed}=\min \{f_{ed}(w):w\in V(G)\cup E(G)\}$. We have the \emph{dual coloring} $g_{ed}$ of $f_{ed}$ defined by setting $g_{ed}(x)=(\max f_{ed}+\min f_{ed})-f_{ed}(x)$ for $x\in V(G)$ and $g_{ed}(uv)=f_{ed}(uv)$ for $uv\in E(G)$, and then
\begin{equation}\label{eqa:f-ed}
g_{ed}(uv)+|g_{ed}(u)-g_{ed}(v)|=f_{ed}(uv)+|f_{ed}(u)-f_{ed}(v)|=k
\end{equation} for every edge $uv\in E(G)$.
\item For a \emph{graceful-difference proper total coloring} $f_{gd}$ of a graph $G$, there exists a constant $k$ such that $\big ||f_{gd}(u)-f_{gd}(v)|-f_{gd}(uv)\big |=k$ for each edge $uv\in E(G)$. Let $\max f_{gd}=\max \{f_{gd}(w):w\in V(G)\cup E(G)\}$ and $\min f_{gd}=\min \{f_{gd}(w):w\in V(G)\cup E(G)\}$. We have the \emph{dual coloring} $g_{gd}$ of $f_{gd}$ defined in the way: $g_{gd}(x)=(\max f_{gd}+\min f_{gd})-f_{gd}(x)$ for $x\in V(G)$ and $g_{gd}(uv)=f_{gd}(uv)$ for each edge $uv\in E(G)$, and then
\begin{equation}\label{eqa:f-md}
{
\begin{split}
\big ||g_{gd}(u)-g_{gd}(v)|-g_{gd}(uv)\big |=\big ||f_{gd}(u)-f_{gd}(v)|-f_{gd}(uv)\big |=k
\end{split}}
\end{equation} for each edge $uv\in E(G)$.
\item For a \emph{felicitous-difference proper total coloring} $f_{fd}$ of a graph $G$, we have a constant $k$ such that $|f_{fd}(u)+f_{fd}(v)-f_{fd}(uv)|=k$ for each edge $uv\in E(G)$. Let $\max f_{fd}=\max \{f_{fd}(w):w\in V(G)\cup E(G)\}$ and $\min f_{fd}=\min \{f_{fd}(w):w\in V(G)\cup E(G)\}$. We have the \emph{dual coloring} $g_{fd}$ of $f_{fd}$ defined as: $g_{fd}(w)=(\max f_{fd}+\min f_{fd})-f_{fd}(w)$ for $w\in V(G)\cup E(G)$, and then
\begin{equation}\label{eqa:f-tg}
{
\begin{split}
|g_{fd}(u)+g_{fd}(v)-g_{fd}(uv)|&=|(\max f_{fd}+\min f_{fd})+f_{fd}(u)+f_{fd}(v)-f_{fd}(uv)|\\
&=(\max f_{fd}+\min f_{fd})\pm k
\end{split}}
\end{equation} for each edge $uv\in E(G)$. Here, if $f_{fd}$ is edge-ordered, that is, $f_{fd}(x)+f_{fd}(y)\geq f_{fd}(xy)$ for each edge $xy\in E(G)$, then
$$
|g_{fd}(u)+g_{fd}(v)-g_{fd}(uv)|=(\max f_{fd}+\min f_{fd})+k=k\,'
$$ We have $$|g_{fd}(u)+g_{fd}(v)-g_{fd}(uv)|=(\max f_{fd}+\min f_{fd})-k=k\,'
$$
if $f_{fd}(x)+f_{fd}(y)< f_{fd}(xy)$ for each edge $xy\in E(G)$. \qqed
\end{asparaenum}
\end{defn}

As removing ``proper'' from Definition \ref{defn:combinatoric-definition-total-coloring}, we get

\begin{defn} \label{defn:4-magic-total-colorings}
\cite{Yao-Ma-Wang-ITAIC2020} A randomly growing network model $N(t)$ admits a $W$-type total coloring $f_t:V(N(t))\cup E(N(t))\rightarrow [1,M(t)]$ with $t\in [a,b]^r$, then we have:
\begin{asparaenum}[(\textrm{E}-1) ]
\item An \emph{edge-magic total coloring} holds $f_t(x)+f_t(xy)+f_t(y)=k>0$ true for each edge $xy\in E(N(t))$.
\item A \emph{graceful-difference total coloring} holds $\big ||f_t(x)-f_t(y)|-f_t(xy)\big |=k\geq 0$ true for each edge $xy\in E(N(t))$.
\item An \emph{edge-difference total coloring} holds $f_t(xy)+|f_t(x)-f_t(y)|=k> 0$ true for each edge $xy\in E(N(t))$.
\item A \emph{felicitous-difference total coloring} holds $\big |f_t(x)+f_t(y)-f_t(xy)\big |=k\geq 0$ true for each edge $xy\in E(N(t))$.\qqed
\end{asparaenum}
\end{defn}

We can add another constraint requirement, such that each edge color $f_t(xy)$ is odd in Definition \ref{defn:4-magic-total-colorings}, so we get new colorings: \emph{odd-edge-magic total coloring}, \emph{odd-graceful-difference total coloring}, \emph{odd-edge-difference total coloring}, \emph{odd-felicitous-difference total coloring}, respectively. Next, we add three parameters $a,b,c$ to Definition \ref{defn:combinatoric-definition-total-coloring} if $G$ is bipartite, and get another group of particular total colorings as follows:

\begin{defn} \label{defn:combinatoric-definition-total-coloring-abc}
\cite{Bing-Yao-2020arXiv} Suppose that a bipartite graph $G$ admits a proper total coloring $f:V(G)\cup E(G)\rightarrow [1,M]$. We define an \emph{edge-function} $c_f(uv)(a,b,c)$ with three non-negative integers $a,b,c$ for each edge $uv\in E(G)$, and define a parameter
\begin{equation}\label{eqa:edge-difference-total-coloring}
B^*_{\alpha}(G,f,M)(a,b,c)=\max_{uv \in E(G)}\{c_f(uv)(a,b,c)\}-\min_{xy \in E(G)}\{c_f(xy)(a,b,c)\}
\end{equation}
If $B^*_{\alpha}(G,f,M)(a,b,c)=0$, we call $f$ a \emph{parameterized $\alpha$-proper total coloring} of $G$, the smallest number
\begin{equation}\label{eqa:minimum}
\chi\,''_{\alpha}(G)(a,b,c) =\min_f \{M:~B^*_{\alpha}(G,f,M)(a,b,c)=0\}
\end{equation}
over all parameterized $\alpha$-proper total colorings of $G$ is called \emph{parameterized $\alpha$-proper total chromatic number}, and $f$ is called a \emph{perfect $\alpha$-proper total coloring} if $\chi\,''_{\alpha}(G)(a,b,c)=\chi\,''(G)$.
\begin{asparaenum}[\textrm{TCol}-1. ]
\item We call $f$ a \emph{(resp. perfect) \textbf{parameterized edge-magic proper total coloring}} of $G$ if $c_f(uv)=af(u)+bf(v)+cf(uv)$, rewrite $B^*_{\alpha}(G,f,M)(a,b,c)=B^*_{emt}(G,f$, $M)(a,b,c)$, and $\chi\,''_{\alpha}(G)(a,b,c)=\chi\,''_{emt}(G)(a,b,c)$ is called \emph{parameterized edge-magic total chromatic number} of $G$.
\item We call $f$ a \emph{(resp. perfect) \textbf{parameterized edge-difference proper total coloring}} of $G$ if $c_f(uv)=cf(uv)+|af(u)-bf(v)|$, rewrite $B^*_{\alpha}(G,f,M)(a,b,c)=B^*_{edt}(G,f$, $M)(a,b,c)$, and $\chi\,''_{\alpha}(G)(a,b,c)=\chi\,''_{edt}(G)(a,b,c)$ is called \emph{parameterized edge-difference total chromatic number} of $G$.
\item We call $f$ a \emph{(resp. perfect) \textbf{parameterized felicitous-difference proper total coloring}} of $G$ if $c_f(uv)=|af(u)+bf(v)-cf(uv)|$, rewrite $B^*_{\alpha}(G,f,M)(a,b,c)=B^*_{fdt}(G,f,M)(a,b,c)$, and $\chi\,''_{\alpha}(G)(a,b,c)=\chi\,''_{fdt}(G)(a,b,c)$ is called \emph{parameterized felicitous-difference total chromatic number} of $G$.
\item We refer to $f$ as a \emph{(resp. perfect) \textbf{parameterized graceful-difference proper total coloring}} of $G$ if $c_f(uv)=\big ||af(u)-bf(v)|-cf(uv)\big |$, rewrite
$$
B^*_{\alpha}(G,f,M)(a,b,c)=B^*_{gdt}(G,f,M)(a,b,c)
$$ and $\chi\,''_{\alpha}(G)(a,b,c)=\chi\,''_{gdt}(G)(a,b,c)$ is called \emph{parameterized graceful-difference total chromatic number} of $G$.\qqed
\end{asparaenum}
\end{defn}

\begin{rem}\label{rem:333333}
We can put forward various requirements for $(a,b,c)$ in Definition \ref{defn:combinatoric-definition-total-coloring-abc} to increase the difficulty of attacking topological coding, since the ABC-conjecture (also, Oesterl\'{e}-Masser conjecture, 1985) involves the equation $a+b=c$ and the relationship between prime numbers. Proving or disproving the ABC-conjecture could impact: Diophantine (polynomial) math problems including Tijdeman's theorem, Vojta's conjecture, Erd\"{o}s-Woods conjecture, Fermat's last theorem, Wieferich prime and Roth's theorem, and so on \cite{Cami-Rosso2017Abc-conjecture}.\paralled
\end{rem}

\subsubsection{Parameterized total colorings}

There are colorings introduced in Definition \ref{defn:n-dimension-total-colorings} for making public-keys and private-keys in topological coding.

\begin{defn} \label{defn:n-dimension-total-colorings}
\cite{Yao-Wang-Su-arameterized-2020} A $(p,q)$-graph $G$ admits a coloring $f:S\subseteq V(G)\cup E(G)\rightarrow S(n,Z^0)$. Let $S(n,Z^0)=\{x_1x_2\cdots x_n:~x_i\in Z^0\}$ be the set of all $n$-dimension number strings, and let $f(u)=a_1a_2\cdots a_n$, $f(v)=b_1b_2\cdots b_n$, $f(uv)=c_1c_2\cdots c_n$ for each edge $uv\in E(G)$, and $\gamma$ be a non-negative integer. There are the following constraint conditions:
\begin{asparaenum}[\textrm{Res}-1. ]
\item \label{dimen:v} $S=V(G)$;
\item \label{dimen:e} $S=E(G)$;
\item \label{dimen:ve} $S=V(G)\cup E(G)$;
\item \label{dimen:adjacent-v} $f(u)\neq f(v)$ for each edge $uv\in E(G)$;
\item \label{dimen:adjacent-e} $f(uv)\neq f(uw)$ for any pair of two adjacent edges $uv,uw\in E(G)$;
\item \label{dimen:incident-v-e} $f(u)\neq f(uv)$ and $f(v)\neq f(uv)$ for each edge $uv\in E(G)$;
\item \label{dimen:e-graceful} $c_j=|a_j-b_j|$ with $j\in [1,n]$;
\item \label{dimen:each-odd} $c_j=|a_j-b_j|$, and each $c_j$ is odd with $j\in [1,n]$;
\item \label{dimen:felicitous} $c_j=a_j+b_j~(\bmod~q)$ with $j\in [1,n]$;
\item \label{dimen:edge-magic} $a_j+b_j+c_j=\gamma$ with $j\in [1,n]$;
\item \label{dimen:edge-difference} $c_j+|a_j-b_j|=\gamma$ with $j\in [1,n]$;
\item \label{dimen:graceful-difference} $\big |c_j-|a_j-b_j|\big |=\gamma$ with $j\in [1,n]$;
\item \label{dimen:felicitous-difference} $|a_j+b_j-c_j|=\gamma$ with $j\in [1,n]$;
\item \label{dimen:each-common-factor} $c_j=\textrm{gcd}(a_j,b_j)$ for each $j\in [1,n]$;
\item \label{dimen:common-factor} Some $j\in [1,n]$ holds $c_j=\textrm{gcd}(a_j,b_j)$;
\item \label{dimen:pairwise-distinct} $c_i\neq c_j$ for any pair of $c_i$ and $c_j$;
\item \label{dimen:part-graceful} $f(E(G))=\{f(e_k)=c_{k,1}c_{k,2}\cdots c_{k,n}:k\in[1$, $q]\}$ such that each $k\in[1,q]$ holds $c_{k,j}=k$ for some $c_{k,j}$ of $f(e_k)$;
\item \label{dimen:part-odd-graceful} $f(E(G))=\{f(e_j)=t_{j,1}t_{j,2}\cdots t_{j,n}:j\in[1$, $2q-1]^o\}$ such that each $j\in[1,2q-1]^o$ holds $t_{j,r}=j$ for some $t_{j,r}$ of $f(e_j)$;
\item \label{dimen:uniform-graceful} $f(E(G))=\{f(e_k)=c_{k,1}c_{k,2}\cdots c_{k,n}:k\in[1$, $q]\}$ such that $c_{k,r}=k\in[1,q]$ for each $r\in [1,n]$; and
\item \label{dimen:uniform-odd-graceful} $f(E(G))=\{f(e_j)=t_{j,1}t_{j,2}\cdots t_{j,n}:j\in[1$, $2q-1]^o\}$ such that $t_{j,r}=j\in[1,2q-1]^o$ for each $r\in [1,n]$.
\end{asparaenum}

\noindent \textbf{We call $f$}:

\noindent ------ \emph{traditional types}

\begin{asparaenum}[\textrm{Dtc}-1. ]
\item A \emph{$n$-dimension proper vertex coloring} if Res-\ref{dimen:v} and Res-\ref{dimen:adjacent-v} hold true.
\item A \emph{$n$-dimension proper edge coloring} if Res-\ref{dimen:e} and Res-\ref{dimen:adjacent-e} hold true.
\item A \emph{$n$-dimension proper total coloring} if Res-\ref{dimen:ve}, Res-\ref{dimen:adjacent-v}, Res-\ref{dimen:adjacent-e} and Res-\ref{dimen:incident-v-e} hold true.

\noindent ------ \emph{other proper types}

\item A \emph{subtraction $n$-dimension proper total coloring} if Res-\ref{dimen:ve}, Res-\ref{dimen:adjacent-v}, Res-\ref{dimen:adjacent-e}, Res-\ref{dimen:incident-v-e} and Res-\ref{dimen:e-graceful} hold true.
\item A \emph{factor $n$-dimension proper total coloring} if Res-\ref{dimen:ve}, Res-\ref{dimen:adjacent-v}, Res-\ref{dimen:adjacent-e}, Res-\ref{dimen:incident-v-e} and Res-\ref{dimen:common-factor} hold true.
\item A \emph{felicitous $n$-dimension proper total coloring} if Res-\ref{dimen:ve}, Res-\ref{dimen:adjacent-v}, Res-\ref{dimen:adjacent-e}, Res-\ref{dimen:incident-v-e} and Res-\ref{dimen:felicitous} hold true.

\item A \emph{graceful $n$-dimension proper total coloring} if Res-\ref{dimen:ve}, Res-\ref{dimen:adjacent-v}, Res-\ref{dimen:adjacent-e}, Res-\ref{dimen:incident-v-e}, Res-\ref{dimen:e-graceful} and Res-\ref{dimen:part-graceful} hold true.
\item An \emph{odd-graceful $n$-dimension proper total coloring} if Res-\ref{dimen:ve}, Res-\ref{dimen:adjacent-v}, Res-\ref{dimen:adjacent-e}, Res-\ref{dimen:incident-v-e}, Res-\ref{dimen:e-graceful} and Res-\ref{dimen:part-odd-graceful} hold true.

\item An \emph{e-magic $n$-dimension proper total coloring} if Res-\ref{dimen:ve}, Res-\ref{dimen:adjacent-v}, Res-\ref{dimen:adjacent-e}, Res-\ref{dimen:incident-v-e}, Res-\ref{dimen:e-graceful} and Res-\ref{dimen:edge-magic} hold true.
\item An \emph{e-difference $n$-dimension proper total coloring} if Res-\ref{dimen:ve}, Res-\ref{dimen:adjacent-v}, Res-\ref{dimen:adjacent-e}, Res-\ref{dimen:incident-v-e}, Res-\ref{dimen:e-graceful} and Res-\ref{dimen:edge-difference} hold true.
\item A \emph{graceful-difference $n$-dimension proper total coloring} if Res-\ref{dimen:ve}, Res-\ref{dimen:adjacent-v}, Res-\ref{dimen:adjacent-e}, Res-\ref{dimen:incident-v-e}, Res-\ref{dimen:e-graceful} and Res-\ref{dimen:graceful-difference} hold true.
\item A \emph{felicitous-difference $n$-dimension proper total coloring} if Res-\ref{dimen:ve}, Res-\ref{dimen:adjacent-v}, Res-\ref{dimen:adjacent-e}, Res-\ref{dimen:incident-v-e}, Res-\ref{dimen:e-graceful} and Res-\ref{dimen:felicitous-difference} hold true.
\item An \emph{anti-equitable $n$-dimension proper total coloring} if Res-\ref{dimen:ve}, Res-\ref{dimen:adjacent-v}, Res-\ref{dimen:adjacent-e}, Res-\ref{dimen:incident-v-e}, Res-\ref{dimen:e-graceful} and Res-\ref{dimen:pairwise-distinct} hold true.

\noindent ------ \emph{sub-proper types}

\item A \emph{subtraction $n$-dimension sub-proper total coloring} if Res-\ref{dimen:ve}, Res-\ref{dimen:adjacent-v}, Res-\ref{dimen:adjacent-e} and Res-\ref{dimen:e-graceful} hold true.
\item A \emph{factor $n$-dimension sub-proper total coloring} if Res-\ref{dimen:ve}, Res-\ref{dimen:adjacent-v}, Res-\ref{dimen:adjacent-e} and Res-\ref{dimen:common-factor} hold true.
\item A \emph{felicitous $n$-dimension sub-proper total coloring} if Res-\ref{dimen:ve}, Res-\ref{dimen:adjacent-v}, Res-\ref{dimen:adjacent-e} and Res-\ref{dimen:felicitous} hold true.

\item A \emph{graceful $n$-dimension sub-proper total coloring} if Res-\ref{dimen:ve}, Res-\ref{dimen:adjacent-v}, Res-\ref{dimen:adjacent-e}, Res-\ref{dimen:e-graceful} and Res-\ref{dimen:part-graceful} hold true.
\item An \emph{odd-graceful $n$-dimension sub-proper total coloring} if Res-\ref{dimen:ve}, Res-\ref{dimen:adjacent-v}, Res-\ref{dimen:adjacent-e}, Res-\ref{dimen:e-graceful} and Res-\ref{dimen:part-odd-graceful} hold true.

\item A \emph{twin-graceful $n$-dimension sub-proper total coloring} if Res-\ref{dimen:ve}, Res-\ref{dimen:adjacent-v}, Res-\ref{dimen:adjacent-e}, Res-\ref{dimen:e-graceful}, Res-\ref{dimen:part-graceful} and Res-\ref{dimen:part-odd-graceful} hold true.

\item A \emph{uniformly graceful $n$-dimension sub-proper total coloring} if Res-\ref{dimen:ve}, Res-\ref{dimen:adjacent-v}, Res-\ref{dimen:adjacent-e}, Res-\ref{dimen:e-graceful} and Res-\ref{dimen:uniform-graceful} hold true.
\item A \emph{uniformly odd-graceful $n$-dimension sub-proper total coloring} if Res-\ref{dimen:ve}, Res-\ref{dimen:adjacent-v}, Res-\ref{dimen:adjacent-e}, Res-\ref{dimen:e-graceful} and Res-\ref{dimen:uniform-odd-graceful} hold true.

\item A \emph{uniformly factor $n$-dimension sub-proper total coloring} if Res-\ref{dimen:ve}, Res-\ref{dimen:adjacent-v}, Res-\ref{dimen:adjacent-e}, Res-\ref{dimen:e-graceful} and Res-\ref{dimen:each-common-factor} hold true.

\item An \emph{anti-equitable $n$-dimension sub-proper total coloring} if Res-\ref{dimen:ve}, Res-\ref{dimen:adjacent-v}, Res-\ref{dimen:adjacent-e}, Res-\ref{dimen:e-graceful} and Res-\ref{dimen:pairwise-distinct} hold true.

\item An \emph{e-magic $n$-dimension sub-proper total coloring} if Res-\ref{dimen:ve}, Res-\ref{dimen:adjacent-v}, Res-\ref{dimen:adjacent-e}, Res-\ref{dimen:e-graceful} and Res-\ref{dimen:edge-magic} hold true.
\item An \emph{e-difference $n$-dimension sub-proper total coloring} if Res-\ref{dimen:ve}, Res-\ref{dimen:adjacent-v}, Res-\ref{dimen:adjacent-e}, Res-\ref{dimen:e-graceful} and Res-\ref{dimen:edge-difference} hold true.
\item A \emph{graceful-difference $n$-dimension sub-proper total coloring} if Res-\ref{dimen:ve}, Res-\ref{dimen:adjacent-v}, Res-\ref{dimen:adjacent-e}, Res-\ref{dimen:e-graceful} and Res-\ref{dimen:graceful-difference} hold true.
\item A \emph{felicitous-difference $n$-dimension sub-proper total coloring} if Res-\ref{dimen:ve}, Res-\ref{dimen:adjacent-v}, Res-\ref{dimen:adjacent-e}, Res-\ref{dimen:e-graceful} and Res-\ref{dimen:felicitous-difference} hold true.\qqed
\end{asparaenum}
\end{defn}

\begin{rem}\label{rem:333333}
We, also, call $f$ defined in Definition \ref{defn:n-dimension-total-colorings} a \emph{$W$-type $n$-dimension total coloring}, since $f$ holds a group of restrictive conditions denoted as $f(uv)=W(f(u),f(v))$ for each edge $uv\in E(G)$. By the way, we can add the restrictive conditions Res-\ref{dimen:adjacent-v}, Res-\ref{dimen:adjacent-e} and Res-\ref{dimen:incident-v-e} to a $W$-type $n$-dimension total coloring for getting a \emph{$W$-type $n$-dimension proper total coloring}. See examples shown in Fig.\ref{fig:n-dimension-11} and Fig.\ref{fig:n-dimension-22} for understanding part of colorings defined in Definition \ref{defn:n-dimension-total-colorings}.\paralled
\end{rem}

\begin{figure}[h]
\centering
\includegraphics[width=16cm]{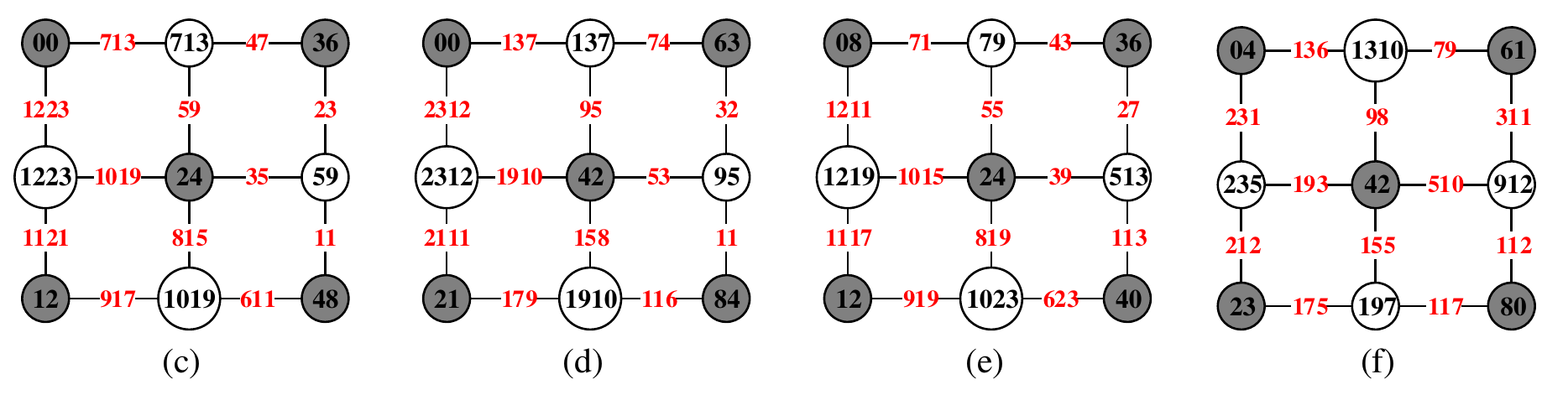}
\caption{\label{fig:n-dimension-11}{\small Four Topsnut-gpws, where both (c) and (d) admit two twin $2$-dimension sub-proper total colorings; both (e) and (f) admit two twin $2$-dimension proper total colorings, cited from \cite{Yao-Wang-Su-arameterized-2020}.}}
\end{figure}

\begin{figure}[h]
\centering
\includegraphics[width=14cm]{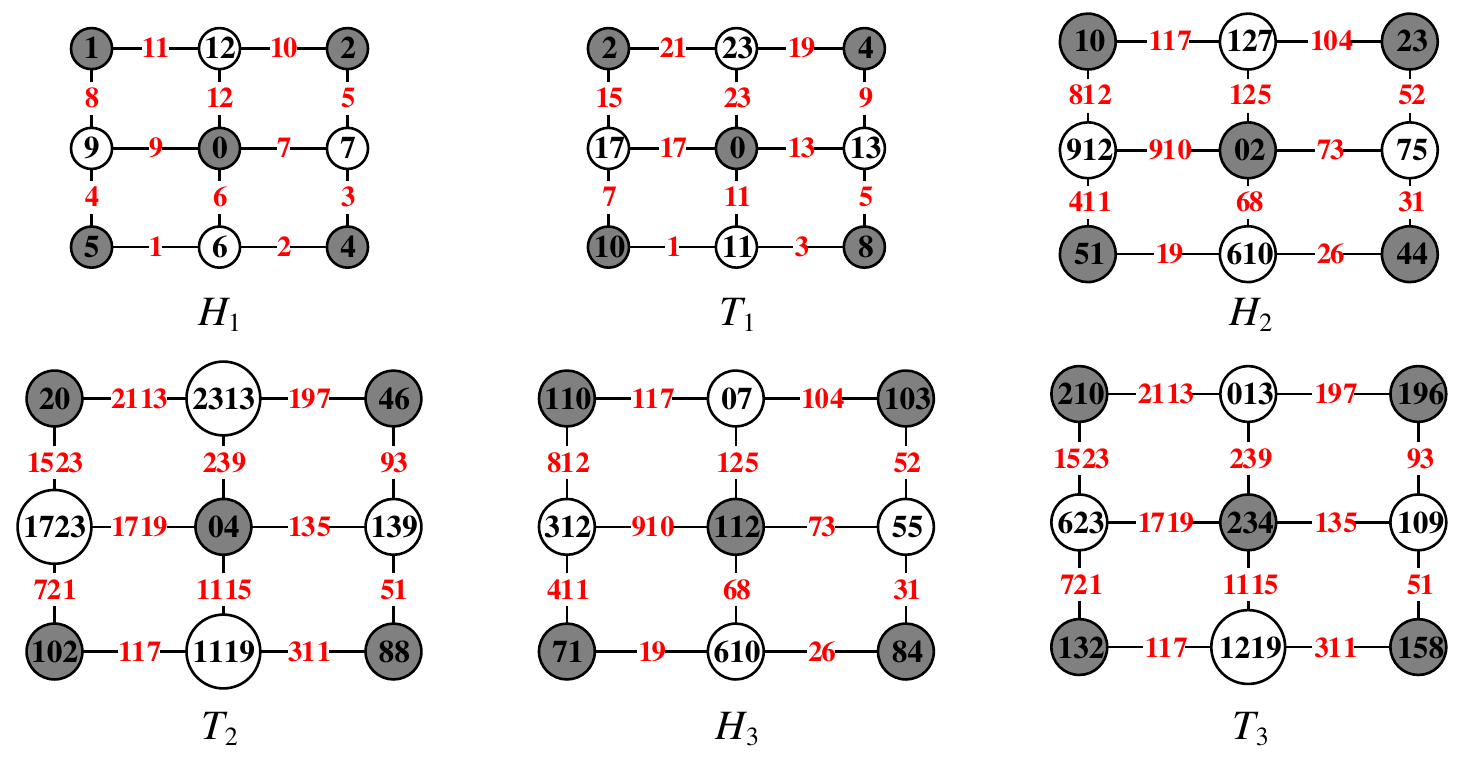}
\caption{\label{fig:n-dimension-22}{\small Six Topsnut-gpws: $H_1$ admits a graceful labeling, $T_1$ admits an odd-graceful labeling; both $H_2$ and $H_3$ admit two graceful $2$-dimension proper total colorings; both $T_2$ and $T_3$ admit two odd-graceful $2$-dimension proper total colorings, cited from \cite{Yao-Wang-Su-arameterized-2020}.}}
\end{figure}

\begin{thm} \label{thm:graph-graceful-n-dimension}
\cite{Yao-Wang-Su-arameterized-2020} Every connected simple graph admits a \emph{graceful $n$-dimension sub-proper total coloring} for some $n\geq 2$.
\end{thm}
\begin{thm} \label{thm:trees-graceful-2-dimension}
\cite{Yao-Wang-Su-arameterized-2020} Every tree admits a \emph{graceful $2$-dimension proper total coloring}.
\end{thm}
\begin{thm} \label{thm:twin-graceful-2-dimension-tree}
\cite{Yao-Wang-Su-arameterized-2020} Every tree admits a \emph{twin-graceful $2$-dimension sub-proper total coloring}.
\end{thm}

We restate the following compounded-type multiple dimension total colorings based on random parameters $k_s$ and $d_s$:

\begin{defn} \label{defn:33multiple-dimension-kd-w-type-colorings}
\cite{Yao-Wang-2106-15254v1} By Definition \ref{defn:kd-w-type-colorings}, let $G$ be a bipartite and connected $(p,q)$-graph, so its vertex set $V(G)=X\cup Y$ with $X\cap Y=\emptyset$ such that each edge $uv\in E(G)$ holds $u\in X$ and $v\in Y$. There are a group of colorings
$$
f_s:X\rightarrow S_{m,0,0,d_s}=\{0,d_s,\dots ,md_s\},~f_s:Y\cup E(G)\rightarrow S_{n,k_s,0,d_s}=\{k_s,k_s+d_s,\dots ,k_s+nd_s\}
$$ here it is allowed $f_s(u)=f_s(w)$ for some distinct vertices $u,w\in V(G)$ for $s\in [1,B]$ with integer $B\geq 2$, such that $f_s\in \{(k_s,d_s)$-strongly gracefully total coloring, $(k_s,d_s)$-strongly odd-gracefully total coloring, $(k_s,d_s)$-edge antimagic total coloring, $(k_s,d_s)$-harmonious total coloring, $(k_s,d_s)$-odd-elegant total coloring, strongly $(k_s,d_s)$-edge-magic total coloring, edge-magic $(k_s,d_s)$-total coloring, strongly edge-difference $(k_s,d_s)$-total coloring, edge-difference $(k_s,d_s)$-total coloring, strongly felicitous-difference $(k_s,d_s)$-total coloring, felicitous-difference $(k_s,d_s)$-total coloring, strongly graceful-difference $(k_s,d_s)$-total coloring, graceful-difference $(k_s,d_s)$-total coloring$\}$ with $s\in [1,B]$. Then $G$ admits a \emph{parameterized compounded $B$-dimension total coloring} $F$ holding
$$F(u)=f_1(u)f_2(u)\cdots f_B(u),~F(uv)=f_1(uv)f_2(uv)\cdots f_B(uv),~F(v)=f_1(v)f_2(v)\cdots f_B(v)
$$ for each edge $uv\in E(G)$.\qqed
\end{defn}

\begin{defn} \label{defn:33multiple-dimension-total-colorings}
\cite{Yao-Wang-2106-15254v1} A $(p,q)$-graph $G$ admits a coloring
$$
g_j:S\subseteq V(G)\cup E(G)\rightarrow S(n,Z^0)=\{x_1x_2\cdots x_n:~x_i\in Z^0\},~j\in [1,B]
$$ and each $g_i$ is one of the total colorings with $i\in [1,27]$ defined in Definition \ref{defn:n-dimension-total-colorings}. Then $G$ admits a \emph{compounded $B$-dimension total coloring} $F$ holding $F(u)=g_1(u)g_2(u)\cdots g_B(u)$, $F(uv)=g_1(uv)g_2(uv)\cdots g_B(uv)$ and $F(v)=g_1(v)g_2(v)\cdots g_B(v)$ for each edge $uv\in E(G)$.\qqed
\end{defn}

\subsubsection{Randomly orienting edges of undirected graphs}

Directed Topcode-matrices are related with directed colorings (resp. labelings) of graphs.

\begin{defn}\label{defn:directed-Topcode-matrix}
\cite{Yao-Zhao-Zhang-Mu-Sun-Zhang-Yang-Ma-Su-Wang-Wang-Sun-arXiv2019} A \emph{directed Topcode-matrix} is defined as
\begin{equation}\label{eqa:Topcode-dimatrix}
\centering
{
\begin{split}
\overrightarrow{T}_{code}= \left(
\begin{array}{ccccc}
x_{1} & x_{2} & \cdots & x_{q}\\
e_{1} & e_{2} & \cdots & e_{q}\\
y_{1} & y_{2} & \cdots & y_{q}
\end{array}
\right)^{+}_{-}=
\left(\begin{array}{c}
X\\
\overrightarrow{E}\\
Y
\end{array} \right)^{+}_{-}=[(X,\overrightarrow{E},Y)^{+}_{-}]^{T}
\end{split}}
\end{equation}\\
where \emph{v-vector} $X=(x_1 , x_2 , \dots ,x_q)$, \emph{v-vector} $Y=(y_1 $, $y_2$, $\dots $, $y_q)$ and \emph{directed-e-vector} $\overrightarrow{E}=(e_1$, $e_2 $, $ \dots $, $e_q)$, such that each arc $e_i$ has its own \emph{head} $x_i$ and \emph{tail} $y_i$ with $i\in [1,q]$, and $q$ is the \emph{size} of $\overrightarrow{T}_{code}$. Moreover, the directed Topcode-matrix $\overrightarrow{T}_{code}$ is \emph{evaluated} if there exists a function $\varphi$ such that $e_i=\varphi(x_i,y_i)$ for $i\in [1,q]$. If a digraph $\overrightarrow{G}$ corresponds its directed Topcode-matrix $\overrightarrow{T}_{code}$ defined in Eq.(\ref{eqa:Topcode-dimatrix}), we denote $\overrightarrow{T}_{code}=T_{code}(\overrightarrow{G})^+_{-}$.\qqed
\end{defn}

Similarly with Definition \ref{defn:topo-authentication-multiple-variables}, we define a directed topological authentication of multiple variables below

\begin{defn} \label{defn:digraph-topo-authentication-multiple-variables}
$^*$ A \emph{directed topological authentication} $\overrightarrow{\textbf{T}}_{\textbf{a}}\langle \overrightarrow{X},\overrightarrow{Y}\rangle$ \emph{of multiple variables} is defined as follows
\begin{equation}\label{eqa:oriented-topo-authentication-11}
\overrightarrow{\textbf{T}}_{\textbf{a}}\langle\overrightarrow{X},\overrightarrow{Y}\rangle =P_{ub}(\overrightarrow{X}) \rightarrow _{\textbf{F}} P_{ri}(\overrightarrow{Y})
\end{equation} where $P_{ub}(\overrightarrow{X})=(\overrightarrow{\alpha}_1,\overrightarrow{\alpha}_2,\dots ,\overrightarrow{\alpha}_m)$ and $P_{ri}(\overrightarrow{Y})=(\overrightarrow{\beta}_1,\overrightarrow{\beta}_2,\dots ,\overrightarrow{\beta}_m)$ both are \emph{variable vectors}, in which both $\overrightarrow{\alpha}_1$ and $\overrightarrow{\beta}_1$ are two digraphs or sets of digraphs (resp. colored digraphs, uncolored digraphs), and $\overrightarrow{F}=(\theta_1,\theta_2,\dots $, $\theta_m)$ is a \emph{directed operation vector}, $P_{ub}(\overrightarrow{X})$ is a \emph{directed topological public-key vector} and $P_{ri}(\overrightarrow{Y})$ is a \emph{directed topological private-key vector}, such that $\theta_k(\overrightarrow{\alpha}_k)\rightarrow \overrightarrow{\beta}_k$ for $k\in [1,m]$ with $m\geq 1$.\qqed
\end{defn}

Orienting the edges of a no-oriented $(p,q)$-graph $G$ produces $2^q$ \emph{pseudo-digraphs} $G_k$ with $k\in [1,2^q]$. Suppose that $G$ admits a total coloring (resp. total labeling) $f$, then each pseudo-digraphs $G_k$ admits a total coloring (resp. total labeling) $f_k$ induced by $f$. So, we have a \emph{semi-Topcode-matrix}
\begin{equation}\label{eqa:44-semi-Topcode-matrix}
ST_{code}(G_k)=T_{code}(G_{k,1})\uplus T_{code}(\overrightarrow{G}_{k,2})^{+}_{-}
\end{equation}where the undirected graph $G_{k,1}$ is an edge induced graph made by all no-oriented edges of $G_k$, the digraph $\overrightarrow{G}_{k,2}$ is an arc induced graph (called digraph) by all oriented edges of $G_k$, such that $G_k=\odot_p\langle G_{k,1},\overrightarrow{G}_{k,2}\rangle$ since $V(G)=V(G_k)=V(G_{k,1})=V(\overrightarrow{G}_{k,2})$.

Thereby, the \emph{directed complexity} $\textrm{O}(\overrightarrow{\textbf{T}}_{\textbf{a}})$ of a directed topological authentication $\overrightarrow{\textbf{T}}_{\textbf{a}}\langle\overrightarrow{X},\overrightarrow{Y}\rangle $ is much greater than the \emph{undirected complexity} $\textrm{O}(\textbf{T}_{\textbf{a}})$ of a undirected topological authentication $\textbf{T}_{\textbf{a}}\langle\textbf{X},\textbf{Y}\rangle $, simply,
\begin{equation}\label{eqa:directed-complexity-vs-undirected-complexity}
\textrm{O}(\overrightarrow{\textbf{T}}_{\textbf{a}})=2^n\cdot \textrm{O}(\textbf{T}_{\textbf{a}})
\end{equation}

\begin{rem}\label{rem:333333}
Here, we can use the undirected graph $G_{k,1}$ as a \emph{topological public-key}, and the digraph $\overrightarrow{G}_{k,2}$ is the desired \emph{topological private-key}, and the pseudo-digraphs $G_k$ is a \emph{semi-directed topological authentication}. By a fixed rule, the Topcode-matrix $T_{code}(G_{k,1})$ in Eq.(\ref{eqa:44-semi-Topcode-matrix}) produces a number-based string $s_{pub}(m)$ to encrypt a digital file, and the directed Topcode-matrix $T_{code}(\overrightarrow{G}_{k,2})^{+}_{-}$ in Eq.(\ref{eqa:44-semi-Topcode-matrix}) produces a number-based string $s_{pri}(n)$ to decrypt the digital file encrypted by $s_{pub}(m)$.

There are many colorings and labelings on digraphs \cite{Yao-Wang-2106-15254v1}, see four colored directed trees shown in Fig.\ref{fig:directed-6c-labelling}.\paralled
\end{rem}

\begin{figure}[h]
\centering
\includegraphics[width=16.4cm]{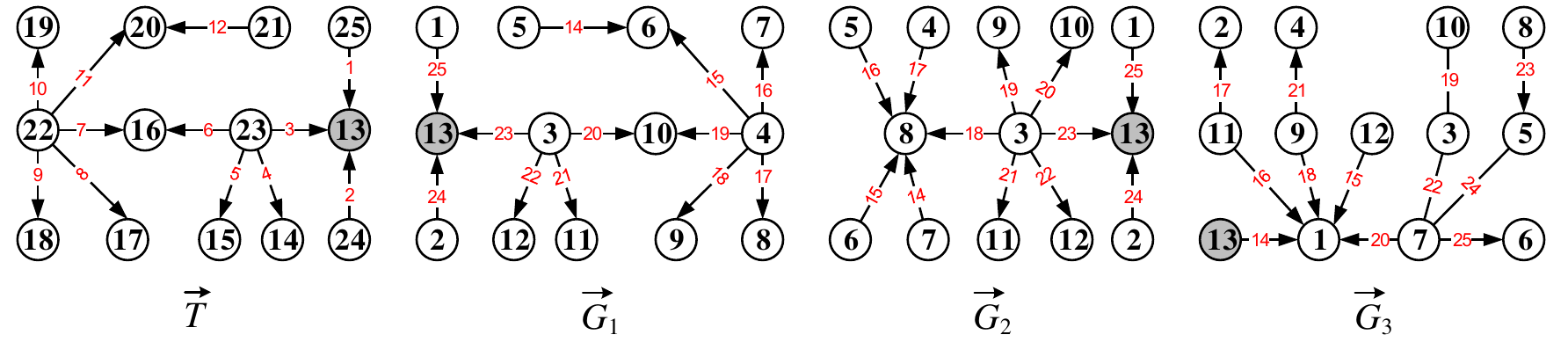}\\
\caption{\label{fig:directed-6c-labelling} {\small Four directed trees admitting 6C-labelings.}}
\end{figure}

\subsubsection{Graphic lattices having coloring closure property}

\begin{lem}\label{thm:four-graceful-constructions}
$^*$ Suppose that a bipartite connected graph $G$ admits a set-ordered graceful labeling, and a connected graph $T$ admits a graceful labeling. Then we have a new connected graph $H$ obtained by adding a new edge to join $G$ and $T$ together, or by vertex-coinciding a vertex of $G$ with a vertex of $T$ into one vertex, such that $H$ admits a graceful labeling too.
\end{lem}
\begin{proof}Let $G$ be a bipartite connected graph of $n$ edges, and $(X,Y)$ be the bipartition of $V(G)$. Suppose that $G$ admits a set-ordered graceful labeling $f$, so $\max f(X)<\min f(Y)$. Without loss of generality, $f(x_i)<f(x_{i+1})$ for $i\in [1,s-1]$, $f(x_s)<f(y_j)<f(y_{j+1})$ for $j\in [1,t-1]$, $X=\{x_i:i\in [1,s]\}$ and $Y=\{y_j: j\in [1,t]\}$, and $s+t=|V(G)|$, where $f(x_1)=0$ and $f(y_t)=n$.

Another bipartite connected graph $T$ of $m$ edges admits a graceful labeling $g$ with $g(w_i)<g(w_{i+1})$ for $w_i\in V(T)=\{w_1,w_2,\dots ,w_p\}$, and $g(E(T))=[1,m]$, where $g(w_1)=0$ and $g(w_p)=m$.

\textbf{Case 1.} We define a labeling $h_1$ as:

(a-1) $h_1(x_i)=f(x_i)$ for $x_i\in X$;

(a-2) $h_1(w_i)=g(w_i)+h_1(x_s)+1$ for $w_i\in V(T)$, $M_1=m+h_1(x_s)+1=m+f(x_s)+1$;

(a-3) $h_1(y_j)=f(y_j)-f(x_s)+M_1=f(y_j)+m+1$ for $y_j\in Y$.

Thereby, we get two edge color sets $h_1(E(T))=[1, m]$ and $h_1(E(G))=[2+m, 1+n+m]$. Since there are $h_1(w_r)-h_1(x_i)=m+1$, or $h_1(y_j)-h_1(w_k)=m+1$, we get a new connected graph $G[\ominus]T$ of $(m+n+1)$ edges obtained by adding a new edge $w_rx_i$ or $y_jw_k$ to join $G$ and $T$ together, and claim that $h_1$ is a graceful labeling of $G[\ominus]T$ with $h_1(E(G[\ominus]T))=[1,n+m+1]$.

\textbf{Case 2.} We define a labeling $h_2$ as:

(b-1) $h_2(x_i)=f(x_i)$ for $x_i\in X$;

(b-2) $h_2(w_i)=g(w_i)+h_2(x_s)$ for $w_i\in V(T)$, $M_2=m+h_2(x_s)+1=m+f(x_s)$;

(b-3) $h_2(y_j)=f(y_j)-f(x_s)+M_2=f(y_j)+m$ for $y_j\in Y$.

Hence, there are $h_2(x_s)=h_2(w_1)$, and $h_2(E(T))=[1, m]$ and $h_2(E(G))=[1+m, n+m]$. We get a new connected graph $\odot\langle G,T\rangle$ of $(m+n)$ edges obtained by vertex-coinciding the vertex $x_s$ with the vertex $w_1$ into one vertex $x_s\odot w_1$, and confirm that $h_2$ is a graceful labeling of $\odot\langle G,T\rangle$ having $h_2(E(\odot\langle G,T\rangle))=[1,n+m]$.

\textbf{Case 3.} We define a labeling $h_3$ as:

(c-1) $h_3(x_i)=f(x_i)$ for $x_i\in X$;

(c-2) $h_3(w_i)=g(w_i)+h_3(x_s)+1$ for $w_i\in V(T)$, $M_3=m+h_3(x_s)+1=m+f(x_s)+1$;

(c-3) $h_3(y_j)=f(y_j)-f(x_s)+M_3-1=f(y_j)+m$ for $y_j\in Y$.

Thereby, we know $h_3(w_p)=h_3(y_1)$, and $h_3(E(T))=[1, m]$ and $h_3(E(G))=[1+m, n+m]$. We get another new connected graph $\odot\langle G,T\rangle$ of $(m+n)$ edges obtained by vertex-coinciding the vertex $w_p$ with the vertex $y_1$ into one vertex $w_p\odot y_1$, and say that $h_3$ is a graceful labeling of $\odot\langle G,T\rangle$ holding $h_3(E(\odot\langle G,T\rangle))=[1,n+m]$ true.

\textbf{Case 4.} We define a labeling $h_4$ as:

(d-1) $h_4(x_i)=f(x_i)$ for $x_i\in X$;

(d-2) $h_4(w_i)=g(w_i)+h_4(x_s)$ for $w_i\in V(T)$, $M_4=m+h_4(x_s)=m+f(x_s)$;

(d-3) $h_4(y_j)=f(y_j)-f(x_s)+M_4-1=f(y_j)+m-1$ for $y_j\in Y$.

There are $h_4(x_s)=h_4(w_1)$, $h_4(w_p)=h_4(y_1)$, and $h_4(x_sy_1)=h_4(w_pw_1)$. We edge-coincide two edges $x_sy_1,w_pw_1$ into one edge $x_sy_1\odot w_pw_1$, the resultant graph is denoted as $\overline{\ominus} \langle G,T\rangle $. Then $h_4$ is a graceful labeling of the edge-coincided graph $\overline{\ominus} \langle G,T\rangle $, and $h_4(E(\overline{\ominus} \langle G,T\rangle ))=[1,n+m-1]$.

We have done the proof of the theorem.
\end{proof}

\begin{figure}[h]
\centering
\includegraphics[width=16.4cm]{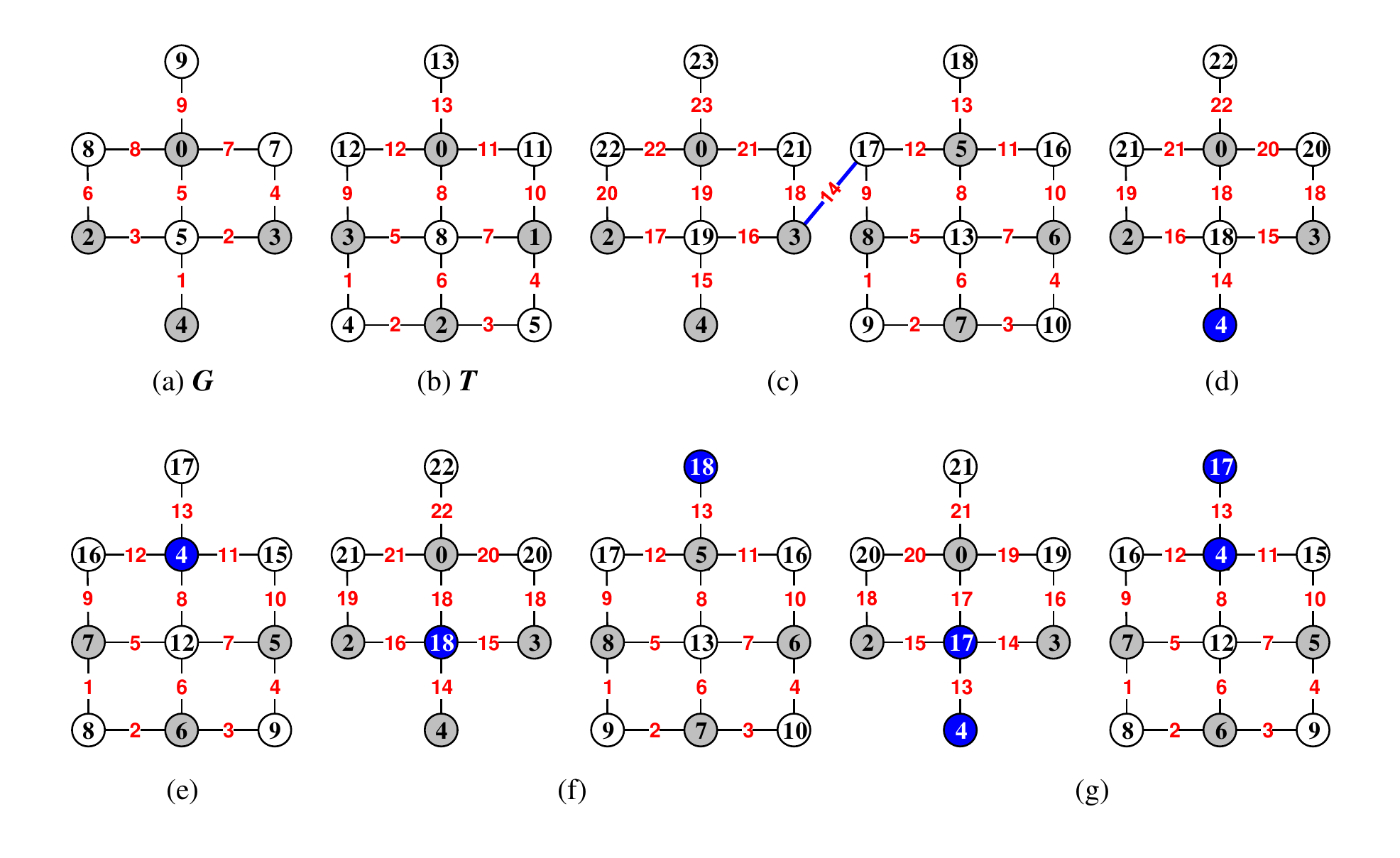}\\
\caption{\label{fig:write-proofs} {\small A scheme for understanding the proof of Theorem \ref{thm:four-graceful-constructions}: (c) $G[\ominus]T$; (d)+(e) $\odot\langle G,T\rangle$; (f) $\odot\langle G,T\rangle$; (g) $\overline{\ominus} \langle G,T\rangle $.}}
\end{figure}

\begin{thm}\label{thm:666666}
$^*$ Suppose that each bipartite connected graph $G_i$ admits $n(G_i)$ different set-ordered graceful labelings for $i\in [1,n]$, and a connected graph $T$ admits a graceful labeling. Then there is a new connected graph $H$ obtained from $G_1,G_2,\dots ,G_n$ and $T$ by the vertex-coinciding operation, the edge-joining operation and the edge-coinciding operation, such that $H$ admits a graceful labeling.
\end{thm}

Let $J_1,J_2,\dots ,J_A$ be a permutation of vertex disjoint gras $a_1G_1,a_2G_2,\dots ,a_nG_n$ with $A=\sum ^n_{k=1}a_k$, if $J_k=G_i$, then $J_k$ admits a set-ordered graceful labeling selected randomly form $n(G_i)$ different set-ordered graceful labelings od $G_i$. For a connected graph $T$ admitting a graceful labeling We construct a connected graph $H_1$ to be one of $T[\ominus]J_1$, $\odot\langle T,J_1\rangle$ and $\overline{\ominus} \langle T,J_1\rangle $, next, we have a connected graph $H_2$ to be one of $H_1[\ominus]J_2$, $\odot\langle H_1,J_2\rangle$ and $\overline{\ominus} \langle H_1,J_2\rangle $; go on in this way, we have connected graphs $H_k$ to be one of $H_{k-1}[\ominus]J_k$, $\odot\langle H_{k-1},J_k\rangle$ and $\overline{\ominus} \langle H_{k-1},J_k\rangle $, and each connected graph $H_k$ admits a graceful labeling for $k\in [1, A]$ by Lemma \ref{thm:four-graceful-constructions}, where $H_0=T$. We write the last connected graph $H_A$ as follows:
\begin{equation}\label{equ:graceful-graph-3-kinds}
H_A=[\ominus \odot \overline{\ominus}]^n_{k=1}a_kG_k
\end{equation} for $\sum ^n_{k=1}a_k\geq 1$ with $a_k\in Z^0$. Moreover, we have a \emph{graceful graph lattice}
\begin{equation}\label{equ:graceful-graph-lattices}
\textbf{\textrm{L}}(Z^0[\ominus \odot \overline{\ominus}]\textbf{\textrm{G}})=\big \{[\ominus \odot \overline{\ominus}]^n_{k=1}a_kG_k:a_k\in Z^0, G_k\in \textbf{\textrm{G}}\big \}
\end{equation} with the \emph{set-ordered graceful graph base} $\textbf{\textrm{G}}=(G_1,G_2,\dots ,G_n)$ has each bipartite connected graph $G_i$ admitting $n(G_i)$ different set-ordered graceful labelings for $i\in [1,n]$.

\begin{rem}\label{rem:333333}
Each connected graph $G$ in a graceful graph lattice $\textbf{\textrm{L}}(Z^0[\ominus \odot \overline{\ominus}]\textbf{\textrm{G}})$ defined in Eq.(\ref{equ:graceful-graph-lattices}) has the following properties: (i) the form $[\ominus \odot \overline{\ominus}]^n_{k=1}a_kG_k$ produces randomly connected graphs by non-negative integers $a_1,a_2,\dots,a_n$ selected from the non-negative integer set $Z^0$ randomly; (ii) admits a graceful labeling made by those set-ordered graceful labelings of the set-ordered graceful graph base $\textbf{\textrm{G}}=(G_1,G_2,\dots ,G_n)$ selected randomly from the $n(G_i)$ different set-ordered graceful labelings of each bipartite connected graph $G_i$ overall $i\in [1,n]$.\paralled
\end{rem}

\begin{problem}\label{problem:666666}
\cite{Yao-Wang-2106-15254v1} For a set-ordered graceful graph base $\textbf{\textrm{G}}=(G_1,G_2,\dots ,G_n)$ with each bipartite connected graph $G_i$ admits $n(G_i)$ different set-ordered graceful labelings for $i\in [1,n]$, how many connected graphs of form $[\ominus \odot \overline{\ominus}]^n_{k=1}a_kG_k$ defined in Eq.(\ref{equ:graceful-graph-3-kinds}) and Eq.(\ref{equ:graceful-graph-lattices}) are there?
\end{problem}

\begin{thm} \label{thm:graceful-total-coloring}
\cite{Zhang-Yang-Yao-Frontiers-Computer-2021} Suppose that a \emph{graph base} $\textbf{\textrm{H}}=(H_k)^n_{k=1}$ consists of $n$ vertex disjoint connected bipartite graph $H_1, H_2,\dots, H_n~(n\geq 2)$ with $e_k=|E(H_k)|$ and $e_1\geq e_2\geq \cdots \geq e_n$. If each connected bipartite graph $H_k$ admits a set-ordered gracefully total coloring, then there exists a set $E^*$ of edges, such that the edge-jointed graph $E^*\oplus^n_{k=1}H_k$ admits a set-ordered gracefully total coloring (refer to Definition \ref{defn:2020arXiv-gracefully-total-coloring}).
\end{thm}

Public-key cryptography plays an important role in modern cryptography, since it consists of public-key and private-key. The public-key is used for encryption or signature verification, and the private-key is used for decryption or signature. Based on the public-key cryptosystem, a connected bipartite graph $F$ of $n$ vertices is as a \emph{public-key}, we will find vertex disjoint connected bipartite graph $H_1, H_2,\dots, H_n~(n\geq 2)$ to be as \emph{private-keys}, so that the new connected bipartite graph $F\odot ^n_{k=1}H_k$ is constructed to complete the encryption, decryption or verification of network security. If a graph $G$ has two vertices $x$ and $y$ holding $N(x)\cap N(y)=\emptyset$, we vertex-coincide these two vertices $x$ and $y$ into one vertex $x\odot y$, the resultant graph is denoted as $G(x\odot y)$.

\begin{thm} \label{thm:coincide-set-ordered-graceful-total-coloring}
\cite{Zhang-Yang-Yao-Frontiers-Computer-2021} Suppose that each connected bipartite graph $H_k$ of a \emph{graph base} $\textbf{\textrm{H}}=(H_k)^n_{k=1}$ admits a set-ordered gracefully total coloring (refer to Definition \ref{defn:2020arXiv-gracefully-total-coloring}), a connected bipartite graph $F$ has $n$ vertices $x_1,x_2,\dots ,x_n$, vertex-coinciding the vertex $x_i$ with a vertex $u_i$ of $H_i$ into one vertex $x_i\odot u_i$ produces an $F$-\emph{graph} $F\odot ^n_{k=1}H_k$. If $F$ admits a set-ordered gracefully total coloring, then the $F$-graph $F\odot ^n_{k=1}H_k$ admits a set-ordered gracefully total coloring too.
\end{thm}

Theorem \ref{thm:coincide-set-ordered-graceful-total-coloring} induces a \emph{set-ordered graceful graph lattice} as follows
\begin{equation}\label{equ:set-ordered-graceful-graph-lattices}
\textbf{\textrm{L}}(\textbf{\textrm{F}}\odot \textbf{\textrm{H}})=\big \{F\odot ^n_{k=1}a_kH_k:a_k\in Z^0, H_k\in \textbf{\textrm{H}}, F\in \textbf{\textrm{F}}\big \}
\end{equation} where $\textbf{\textrm{F}}$ is a set of connected bipartite graphs admitting set-ordered gracefully total colorings. Since each graph of $\textbf{\textrm{L}}(\textbf{\textrm{F}}\odot \textbf{\textrm{H}})$ admits a set-ordered gracefully total coloring, we say that $\textbf{\textrm{L}}(\textbf{\textrm{F}}\odot \textbf{\textrm{H}})$ is closed to the set-ordered gracefully total coloring.

\begin{rem}\label{rem:333333}
Since each graceful graph lattice $\textbf{\textrm{L}}(Z^0[\ominus \odot \overline{\ominus}]\textbf{\textrm{G}})$ defined in Eq.(\ref{equ:graceful-graph-lattices}) and each set-ordered graceful graph lattice defined in Eq.(\ref{equ:set-ordered-graceful-graph-lattices}) are closure to the graceful labeling and the set-ordered gracefully total coloring, they are closed to those $W$-type labelings and set-ordered $W$-type total colorings that are equivalent to graceful labeling, or set-ordered gracefully total coloring.\paralled
\end{rem}

\subsection{Operation graphic lattice and dynamic networked models}

\subsubsection{Operation graphic lattice}

Let $\Lambda=\{O_1,O_2,\dots ,O_n\}$ be the set of graph operations \cite{Yao-Su-Sun-Wang-Graph-Operations-2021}. So, we call a subset $\textbf{\textrm{O}}_i=\{O_{i,1}$, $O_{i,2}$, $\dots ,O_{i,m_i}\}\subset \Lambda$ a \emph{graph operation base} if each graph operation $O_{i,j}$ is not a compound of others. We use symbol $a_kO_{i,k}$ to indicate $a_k$ copies of $O_{i,k}\in \textbf{\textrm{O}}_i$ for $k\in [1,m_i]$, and $O_{j_1},O_{j_2},\dots ,O_{j_{A}}$ is a permutation of these graph operations $a_1O_{i,1}$, $a_2O_{i,2}$, $\dots$, $a_{m_i}O_{i,m_i}$, where $A=\sum ^{m_i}_{k=1}a_k$.

We implement these graph operations $O_{j_1}$, $O_{j_2}$, $\dots $, $O_{j_{A}}$ to a graph $G$, one by one. $O_{j_1}$ is implemented to $G$, the resultant graph is denoted as $O_{j_1}(G)$; and let $G_1=O_{j_1}(G)$, we have $G_2=O_{j_2}(G_1)$, $G_3=O_{j_3}(G_2)$, $\dots$, $G_{A}=O_{j_{A}}(G_{A-1})$. We write $G_{A}=G\lhd |^{m_i}_{k=1}a_kO_{i,k}$ for integrity. In \cite{Yao-Su-Sun-Wang-Graph-Operations-2021}, an \emph{operation graphic lattice} is defined as
\begin{equation}\label{eqa:operation-base-graphic-latticess}
{
\begin{split}
\textbf{\textrm{L}}\big (\textbf{\textrm{F}}\lhd Z^0\textbf{\textrm{O}}_i\big )=\big \{G\lhd |^{m_i}_{k=1}a_kO_{i,k}: a_k\in Z^0,O_{i,k}\in \textbf{\textrm{O}}_i,G\in \textbf{\textrm{F}}\big \}
\end{split}}
\end{equation} with $\sum ^{m_i}_{k=1}a_k\geq 1$, where $\textbf{\textrm{F}}$ is the set of graphs, and $\textbf{\textrm{O}}_i$ is called \emph{operation graphic lattice base} (or \emph{graph operation base} for short). It is noticeable, a graph $G$ in $\textbf{\textrm{F}}$ may be connected, or disconnected with components $H_1,H_2,\dots , H_t$ for $t\geq 2$.

Is $\textbf{\textrm{L}}\big (\textbf{\textrm{F}}\lhd Z^0\textbf{\textrm{O}}_i\big )$ empty? We select $O_{i,1}$ to be the vertex-splitting operation and $O_{i,2}$ to be the vertex-coinciding operation in each graph operation base $\textbf{\textrm{O}}_i$ (refer to Definition \ref{defn:vertex-split-coinciding-operations}), so we can guarantee that $\textbf{\textrm{L}}\big (\textbf{\textrm{F}}\lhd Z^0\textbf{\textrm{O}}_i\big )$ contains at least a graph.

Let $\textbf{\textrm{M}}_{\textrm{axpg}}$ be the set of maximal planar graphs, we get an \emph{operation maximal planar graphic lattice} $\textbf{\textrm{L}}\big (\textbf{\textrm{M}}_{\textrm{axpg}}\lhd Z^0\textbf{\textrm{O}}_i\big )$ by replacing the set $\textbf{\textrm{F}}$ by $\textbf{\textrm{M}}_{\textrm{axpg}}$ in Eq.(\ref{eqa:operation-base-graphic-latticess}). There are many graph operations on planar graphs and maximal planar graphs introduced in \cite{Jin-Xu-Maximal-Science-Press-2019} and \cite{Bing-Yao-Hongyu-Wang-arXiv-2020-homomorphisms}.

Let $\textbf{\textrm{H}}_{\textrm{ami}}$ be the set of Hamilton graphs, we get an \emph{operation Hamilton graphic lattice} obtained by replacing the set $\textbf{\textrm{F}}$ by $\textbf{\textrm{H}}_{\textrm{ami}}$ in Eq.(\ref{eqa:operation-base-graphic-latticess}), see \cite{Yao-Wang-2106-15254v1} for related works.

\begin{problem}\label{problem:xxxxxxxxx}
\cite{Yao-Su-Sun-Wang-Graph-Operations-2021} About $\textbf{\textrm{F}}\subseteq \textbf{\textrm{L}}\big (\textbf{\textrm{F}}\lhd Z^0\textbf{\textrm{O}}_i\big )$, or $\textbf{\textrm{F}}\supseteq \textbf{\textrm{L}}\big (\textbf{\textrm{F}}\lhd Z^0\textbf{\textrm{O}}_i\big )$, or $\textbf{\textrm{F}}=\textbf{\textrm{L}}\big (\textbf{\textrm{F}}\lhd Z^0\textbf{\textrm{O}}_i\big )$, which one is true?
\end{problem}

\subsubsection{Stochastic network lattices}

\textbf{Stochastic operation base}. An \emph{operation stochastic-graphic lattice} $\textbf{\textrm{L}}\big (F_{p,q}\lhd Z^0\textbf{\textrm{O}}_i\big )$ defined in Eq.(\ref{eqa:operation-base-graphic-latticess}) is obtained by picking a graph operation base $\textbf{\textrm{O}}_i\subset \Lambda$ uniformly at random.

Also, if $\textbf{\textrm{F}}_{stoc}$ is the set of stochastic graphs, then $\textbf{\textrm{L}}\big (\textbf{\textrm{F}}_{stoc}\lhd Z^0\textbf{\textrm{O}}_i\big )$ obtained by substituting the set $\textbf{\textrm{F}}$ by $\textbf{\textrm{F}}_{stoc}$ in Eq.(\ref{eqa:operation-base-graphic-latticess}) is an \emph{operation stochastic-graphic lattice}.

\textbf{Scale-free network operation.} \cite{Yao-Su-Sun-Wang-Graph-Operations-2021} One of standard characteristics of a \emph{scale-free network} $N(t)$ has its own degree distribution
\begin{equation}\label{eqa:Barabasi-Albert1999}
P(k)=\textrm{Pr}(x=k)\sim k^{-\lambda}, ~2<\lambda<3
\end{equation}
where $P(k)$ is the \emph{power-law distribution} of a vertex joined with $k$ vertices in $N(t)$ (ref. \cite{Barabasi-Albert1999}). There are three main aspects:

(i) \emph{Growth and preferential attachment}. Add a new vertex $u$ into $N(t-1)$, and join $u$ to vertex $x_i$ of $N(t-1)$ for $i\in [1,m]$ under a preferential attachment $\Pi_i=k_i/\sum_jk_j$, also, a vertex $i$ of degree $k_i$ in $N(t-1)$ was joined with a new vertex under the probability $\Pi_i$.

(ii) \emph{Dynamic equation} $\frac{\partial k_i(t)}{\partial t}=m\Pi_i$, and use the initial condition $k_i(t_i)=m$ to solve degree function $k_i(t)$ from the dynamic equation.

(iii) \emph{Degree distribution}$P(k) \sim k^{-\lambda}$ shown in Eq.(\ref{eqa:Barabasi-Albert1999}) by using a uniformly density function $f(t_i)=(t_i+m_0)^{-1}$ at each time step $t_i$ for computing $P(k)=\frac{\partial P(k_i(t)<k)}{\partial k}$.

Thereby, a \emph{graph operation base} $\textbf{\textrm{O}}_i$ is called a \emph{scale-free network operation base} if each operation $O_{i,k}\in \textbf{\textrm{O}}_i$ is a \emph{scale-free network operation} defined by $P(k) \sim k^{-\lambda_{i,k}}$ and $\lambda_{i,k}>2$. We call the following set
\begin{equation}\label{eqa:scale-free-network-lattice11}
{
\begin{split}
&\textbf{\textrm{L}}\big (\textbf{\textrm{F}}_{scale}(t)\lhd Z^0\textbf{\textrm{O}}_i\big )=\big \{N(t)\lhd |^{m_i}_{k=1}a_kO_{i,k}: a_k\in Z^0, O_{i,k}\in \textbf{\textrm{O}}_i,N(t)\in \textbf{\textrm{F}}_{scale}(t)\big \}
\end{split}}
\end{equation} an \emph{operation scale-free network lattice}, where $\textbf{\textrm{F}}_{scale}(t)$ is the set of dynamic scale-free networks $N(t)$, and $\textbf{\textrm{O}}_i$ is a \emph{scale-free network operation base}.

Suppose that a deterministic growing-model $N(t)$ has $n_v(t)$ vertices and $n_e(t)$ edges at time step $t$. The authors \cite{Yao-Ma-Su-Wang-Zhao-Yao-2016} have considered that the deterministic growing-model $N(t)$ satisfies a system of linear equations (\emph{linear growth})
\begin{equation}\label{eqa:linear-growth}
n_v(t)=a_v t+b_v,\quad n_e(t)=a_e t+b_e
\end{equation}
As known, BA-model holds Eq.(\ref{eqa:linear-growth}), so do the models $N_{PB}(t)$, $T(t)$ and $M(t)$ introduced in the previous section. Very often, a system of non-linear equations (\emph{exponential growth}) appeared in some literature is
\begin{equation}\label{eqa:Only-exponent}
n_v(t)=a_v r^t+b_v,\quad n_e(t)=a_e s^t+b_e
\end{equation}
with $|r|\neq 1$ and $|s|\neq 1$. Furthermore, the \emph{velocity} $V_{el}(N(t))$ of the model $N(t)$ (Ref. \cite{Yao-Wang-Su-Ma-Yao-Zhang-Xie2016}) is equal to
\begin{equation}\label{eqa:velocityss}
V_{el}(N(t))=\sqrt{\left[\frac{\partial n_v(t)}{\partial t}\right]^2+\left[\frac{\partial n_e(t)}{\partial t}\right]^2}=\sqrt{a_v^2+a_e^2}
\end{equation} or
\begin{equation}\label{eqa:velocity11}
V_{el}(N(t))=\sqrt{\left[\frac{\partial n_v(t)}{\partial t}\right]^2+\left[\frac{\partial n_e(t)}{\partial t}\right]^2}=\sqrt{[r^{t}a_v\ln r]^2+[s^{t}a_e\ln r]^2}
\end{equation} according to Eq.(\ref{eqa:linear-growth}) and Eq.(\ref{eqa:Only-exponent}). Furthermore, we have
\begin{equation}\label{eqa:velocity22}
\frac{\partial n_e(t)}{\partial t}\sim \frac{\langle k\rangle }{2}\cdot \frac{\partial n_v(t)}{\partial t}
\end{equation} where $\langle k\rangle =2n_e(t)/n_v(t)$ is the \emph{average degree} of a network model having $n_v(t)$ vertices and $n_e(t)$ edges at each time step $t$. On the other hands, we say a model to be \emph{sparse} if it holds Eq.(\ref{eqa:velocity22}).

As $s=r$ in Eq.(\ref{eqa:Only-exponent}), we have
\begin{equation}\label{eqa:vertices-edges}
n_v(t)=a_v r^t+b_v,\quad n_e(t)=a_e r^t+b_e
\end{equation} with $|r|\neq 1$, and
\begin{equation}\label{eqa:velocity00}
V_{el}(N(t))=\sqrt{\left[\frac{\partial n_v(t)}{\partial t}\right]^2+\left[\frac{\partial n_e(t)}{\partial t}\right]^2}=\sqrt{a_v^2+a_e^2}\cdot r^{t}\ln r
\end{equation}

There are deterministic growing-models holding Eq.(\ref{eqa:vertices-edges}), such as the Sierpinski model $S(t)$ with $r=3$ (Ref. \cite{Zhang-Zhou-Fang-Guan-Zhang2007}), the Recursive tree model $R(t)$ with $r=q+1$ for $q\geq 2$ (Ref. \cite{Comellas-Fertin-Raspaud-2004}), and the Apollonian model $A(t)$ with $r=m(d+1)$ (Ref. \cite{Zhang-Comellas-Fertin-Rong-2006}). Clearly, we can classify deterministic growing-models holding Eq.(\ref{eqa:vertices-edges}) and Eq.(\ref{eqa:velocity00}) by velocity, and call them as \emph{$r$-rank models}. Here, the model $S(t)$ develops in speed that is lower than two models $R(t)$ and $A(t)$. A tripartite model introduced in \cite{Lu-Su-Guo2013} and a 4-partite model appeared in \cite{Zhang-Rong-Guo2006} both are 2-rank models.

We summarize the fundamental characteristics of a \emph{scale-free network model} $N(t)$ having $n_v(t)$ vertices and $n_e(t)$ edges as:

(1) Two \emph{common mechanisms}. \emph{Growth} form
\begin{equation}\label{eqa:555555}
n_v(t)>n_v(t-1),\quad n_e(t)>n_e(t-1)
\end{equation} and \emph{Preferential attachment} by $\Pi_i(t)=k_i/\sum_jk_j$.

(2) A \emph{dynamic equation}
\begin{equation}\label{eqa:c3xxxxx}
\frac{\partial k_i(t)}{\partial t}=m\Pi_i(t),\quad \frac{\partial k_i(t)}{\partial t}=f(n,m,t,k_i(t), \Pi_i(t), \alpha,\beta)
\end{equation} where the function $f$ contains six fundamental variables: $n$ new vertices and $m$ new edges, time step $t$, degree $k_i(t)$, $\alpha$ vertices removed, $\beta$ edges removed and the probability of preferential attachment $\Pi_i(t)$. Using the initial condition $k_i(t_i)=m$ solves degree function $k_i(t)$ from the dynamic equation.

(3) A \emph{sum} $\sum n_i(N(t))\frac{\partial k_i(t)}{\partial t}$ equals to a constant, or approaches to a constant as $t\rightarrow\infty$, where $n_i(N(t))$ is the number of vertices having degree $k_i(t)$ in $N(t)$.

(4) A \emph{degree distribution} $P(k)\sim k^{-\gamma}$ and cumulative distribution $P_{cum}(k)\sim k^{1-\gamma}$ with $2<\gamma<3$ hold
\begin{equation}\label{eqa:c3xxxxx}
P(k)=\frac{\partial P(k_i(t)<k)}{\partial k},\quad P(k)=-\frac{\partial P_{cum}(k)}{\partial k}
\end{equation}The cumulative distribution is defined as
\begin{equation}\label{eqa:Dorogovstev-Goltsev2002}
P_{cum}(k)=\sum_{k\,'\geq k}\frac{|V(k\,',t)|}{n_v(t)}\sim k^{1-\lambda},
\end{equation}
where $|V(k\,',t)|$ is the number of vertices of degree $k\,'$ (Ref. \cite{Dorogovstev-2002}). The \emph{edge-cumulative distribution} is $P_{ecum}(k) =\sum _{k\,'\geq k}\frac{E(k\,',t)}{n_e(t)}\sim k^{1-\delta}$, where $E(k\,',t)$ is the number of edges joined with vertices of degree $k\,'$ greater than $k$ at time step $t$, and the \emph{$d$-edge-cumulative distribution} is $P^d_{ecum}(k) =\sum _{k\,'\geq k}\frac{k\,'E(k\,',t)}{n_e(t)}\sim k^{1-\varepsilon}$ \cite{Yao-Wang-Su-Ma-Yao-Zhang-Xie2016}.

(5) A \emph{velocity} $V_{el}(N(t))=\sqrt{\left [\frac{\partial n_{v}(t)}{\partial t}\right ]^2+ \left [\frac{\partial n_{e}(t)}{\partial t}\right ]^2}$, where $\frac{\partial n_{v}(t)}{\partial t}$ is the \emph{v-velocity} and $\frac{\partial n_{e}(t)}{\partial t}$ is the \emph{e-velocity}, and
\begin{equation}\label{eqa:c3xxxxx}
\frac{\partial n_e(t)}{\partial t}\sim \langle k\rangle\cdot b \cdot \frac{\partial n_v(t)}{\partial t}
\end{equation} where $\langle k\rangle $ is the average degree, and $b$ is a constant.

(6) $N(t)$ is \emph{sparse} (Ref. \cite{Genio-Gross-Bassler2011}), i.e. its average degree $\langle k\rangle $ is approximate to a constant, or
\begin{equation}\label{eqa:c3xxxxx}
n_e(t) \sim \textrm{O}(n_v(t)\ln[n_v(t)])
\end{equation}

Three physical scientists Newman, Barab\'{a}si and Watts (Ref. \cite{Newman-Barabasi-Watts2006}) pointed:``\emph{Pure graph theory is elegant and deep, but it is not especially relevant to networks arising in the real world. Applied graph theory, as its name suggests, is more concerned with real-world network problems, but its approach is oriented toward design and engineering.}''

The technique of adding leaf to scale-free trees can quickly color trees for producing number-based strings, see two scale-free trees shown in Fig.\ref{fig:scale-free-tree-network}.

\begin{figure}[h]
\centering
\includegraphics[width=16.4cm]{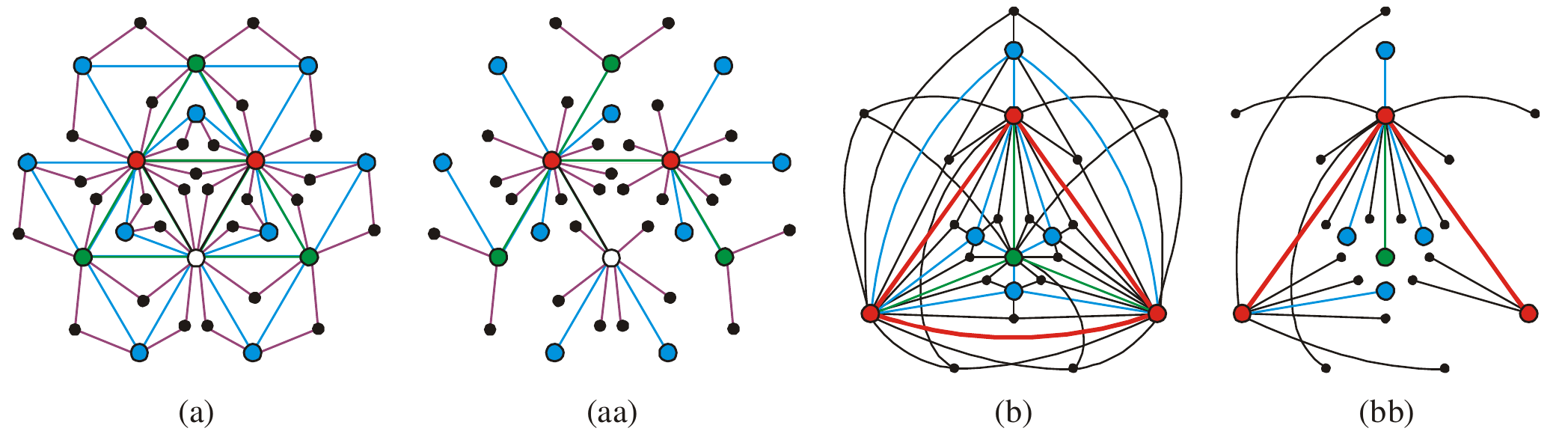}\\
\caption{\label{fig:scale-free-tree-network} {\small Two scale-free trees (aa) and (bb) from two scale-free networks (a) and (b).}}
\end{figure}

We take randomly a family of network operations $O_{i,1},O_{i,2},\dots, O_{i,m}$ from an \emph{network operation base} $\textbf{\textrm{O}}_i$, and do these network operations $O_{i,1},O_{i,2},\dots, O_{i,m}$ to a permutation of vertex disjoint dynamic networks $a_1G_1(t)$, $a_2G_2(t)$, $ \dots $, $a_nG_n(t)$, where $a_kG_k(t)$ stands for $a_k$ copies of $G_k(t)$, the resultant graph $H(t)$ is denoted as
\begin{equation}\label{eqa:randomly-operations}
H(t)=[O_{i,1},O_{i,2},\dots, O_{i,m}]\langle a_1G_1(t),a_2G_2(t), \dots,a_nG_n(t)\rangle=[O_{i,j}]^m_{j=1}\langle a_kG_k(t)\rangle^n_{k=1}
\end{equation} we get a \emph{dynamic network lattice}
\begin{equation}\label{eqa:dynamic-network-lattices}
\textbf{\textrm{L}}([\textbf{\textrm{O}}_i]Z^0\langle \textbf{\textrm{G}}(t)\rangle)=\big \{[O_{i,j}]^m_{j=1}\langle a_kG_k(t)\rangle^n_{k=1}:a_k\in Z^0,G_k(t)\in \textbf{\textrm{G}}(t), O_{i,j}\in \textbf{\textrm{O}}_i\big \}
\end{equation} where $\textbf{\textrm{G}}(t)=(G_1(t),G_2(t),\dots ,G_n(t))$ is \emph{dynamic network base} with each dynamic network $G_k(t)$ admitting a $W$-type coloring at time step $t$.

\textbf{Randomness and complexity of a dynamic network lattice.} Each dynamic network $H(t)\in \textbf{\textrm{L}}([\textbf{\textrm{O}}_i]Z^0\langle \textbf{\textrm{G}}(t)\rangle)$ is random, since the network operations $O_{i,1},O_{i,2},\dots,O_{i,m}$ were randomly taken and each permutation of vertex disjoint dynamic networks $a_1G_1(t)$, $ a_2G_2(t)$, $ \dots $, $a_nG_n(t)$ (there are $\alpha!$ permutations with $\alpha=\sum ^n_{k=1}a_k$) is random, and the dynamic network $H(t)$ was produced randomly from $[O_{i,j}]^m_{j=1}\langle a_kG_k(t)\rangle^n_{k=1}$, notice that there are $m!$ operation permutations of the network operations $O_{i,1},O_{i,2},\dots, O_{i,m}$ in total. Thereby, the \emph{dynamic number-based strings} $s(H(t))$ are obtained randomly from the dynamic network $H(t)$ at time step $t$. Attacking dynamic number-based strings is related with exponential level calculation, even impossible in the age of quantum computer.

\subsection{Counting spanning trees on some networked models}

The total number of all spanning trees on a networked model $N_k(t)$ with $k\in [1,8]$ is denoted as $S_k(t)$ in this subsection.

\subsubsection{Fibonacci-model $N_{1}(t)$}

\textbf{Triangle-growth function.} \cite{Ma-Wang-Wang-Yao-Theoretical-Computer-Science-2018} Given a network model $N_{i}(V_{i},E_{i})$, at time step $i+1$ $(i\geq0)$, add a new vertex for each edge of the edge set $E_{i}$ and then connect the new vertex to two endpoints of that old edge. The above operation is called triangle-growth function

\textbf{Procedure of generating $N_{1}(t)$}

As $t=1$, the initial seminal model $N_{1}(1)$ is a triangle, see the leftmost one in Fig.\ref{fig:PHYSA2019}(a).

As $t=2$, the second model $N_{1}(2)$ is obtained by not only implementing the triangle-growth function to each edge of $N_{1}(1)$ once exactly but also additionally connecting a pendent vertex.

Before starting our next discussions, for arbitrary $t$ $(>2)$, we can call an edge to be active, if it is resulted at time steps $t-1$ or $t-2$. As $t\geq3$, we apply the triangle-growth function to each active edge of $N_{1}(t-1)$ only one time. On the other hand, we need to link each existing vertex $v$ of degree two with a pendent vertex. The described procedure can be carried out until a desired model is obtained.

The Fibonacci-model $N_{1}(t)$ has \emph{sparsity}, \emph{scale-free} feature, \emph{small-world} property and \emph{higher and richer clusters} \cite{Ma-Wang-Wang-Yao-Theoretical-Computer-Science-2018}.

\begin{figure}
\centering
\includegraphics[width=16.4cm]{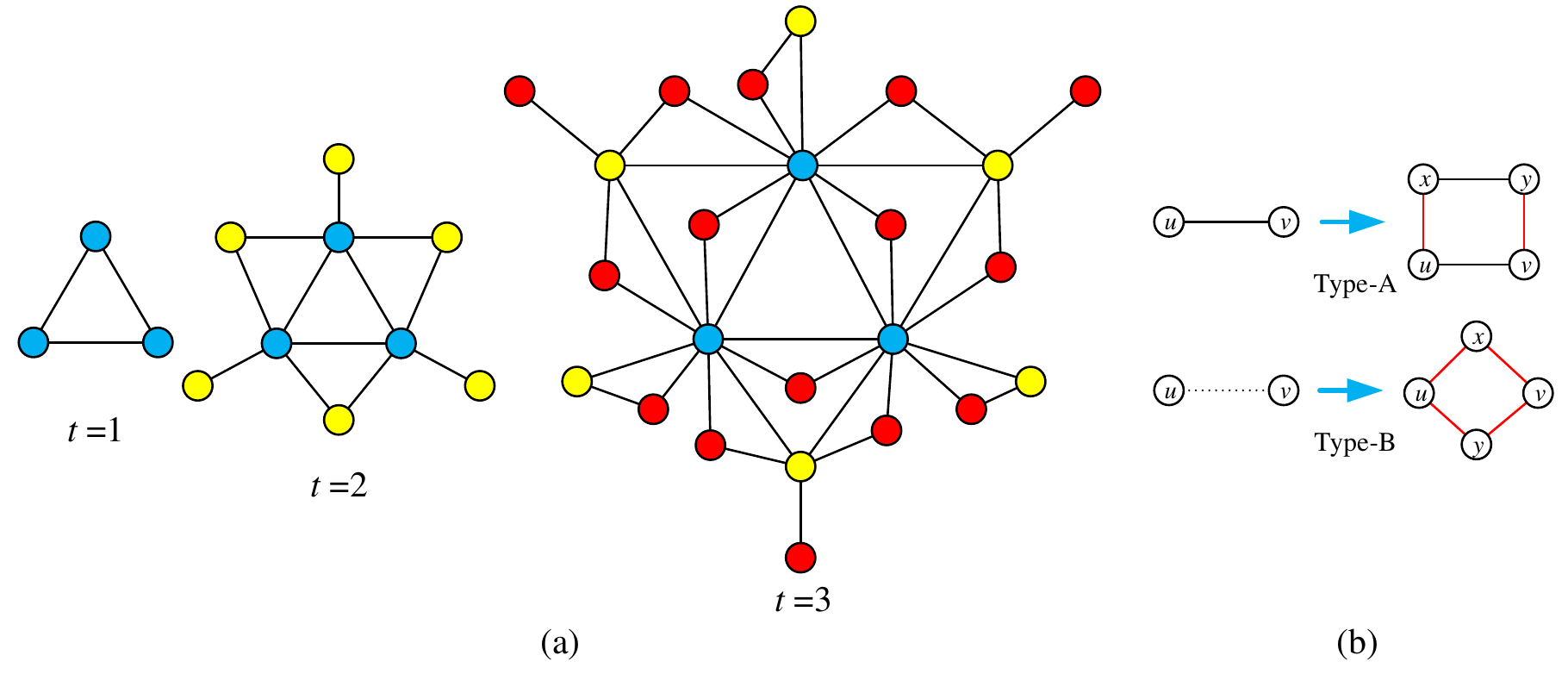}\\
\caption{\label{fig:PHYSA2019}{\small The diagram of the first two growth steps of model $N_{1}(t)$, cited from \cite{Ma-Wang-Wang-Yao-Theoretical-Computer-Science-2018}.}}
\end{figure}

\begin{thm}\label{thm:666666}
\cite{Ma-Wang-Wang-Yao-Theoretical-Computer-Science-2018} The closed-form solution of the total number of model $N_{1}(t)$ is given by
$$S_{1}(t)=3^{\Delta_{v_{1}}(t-1)}\cdot 54^{\Delta_{v_{1}}(t-2)}\cdot M_{1}(t)\prod_{i=3}^{t-2}M_{1}(i)^{\Delta_{v_{1}}(t-i)}$$
in which
\begin{equation}\label{eqa:mf-5-1}
\begin{split}M_{1}(t)=&3\left(4^{\sum_{i=1}^{t-3}\sigma_{i}}\left[\frac{9}{4}\left(-\frac{1}{2}\right)^{t-3}+\sum_{j=0}^{t-4}\left(-\frac{1}{2}\right)^{j}(t-j)\right]^{\sum_{i=1}^{2}\sigma_{i}}\right)^{3}\\
&\times\left(\prod_{i=3}^{t-2}\left[\frac{9}{4}\left(-\frac{1}{2}\right)^{i-2}+\sum_{l=0}^{i-3}\left(-\frac{1}{2}\right)^{l}(i+1-l)\right]^{\sum_{j=t-i-1}^{t-i}\sigma_{j}}\right)^{3}\\
&\times\left[(54\times24^{2})^{2({\sigma_{(t-4)}}+{\sigma_{(t-3)}})}24^{2{\sigma_{(t-3)}}}\left(\frac{9}{4}\right)^{\sigma_{(t-3)}}\right]^{3}\\
&\times\left[\frac{9}{4}\left(-\frac{1}{2}\right)^{t-2}+\sum_{i=0}^{t-3}\left(-\frac{1}{2}\right)^{i}(t+1-i)\right]^{2}
\end{split}
\end{equation}
and $\sigma_{n}=2(\sqrt{3}-1)^{t-2}+\sum_{i=0}^{t-3}\left(\sqrt{3}+1\right)^{t-1-i}(\sqrt{3}-1)^{i}$ with $\sigma_{1}=1$ and $\sigma_{2}=2$, as well as
\begin{equation}\label{eqa:A-8}
\Delta_{v_{1}}(t)=\Delta_{v_{2}}(t-1)+\Delta_{v_{1}}(t-2)=\left\{\begin{split}& \sum_{i=3}^{t/2}\Delta_{v_{2}}(2i-1)+\Delta_{v_{1}}(4), \; t\textrm{ is even}\\
&\sum_{i=2}^{\lfloor t/2\rfloor}\Delta_{v_{2}}(2i)+\Delta_{v_{1}}(3),\qquad \; t\textrm{ is odd}
\end{split}
\right.
\end{equation}

\begin{equation}\label{eqa:A-2}
\left\{\begin{split}& \Delta_{v_{1}}(t)=\Delta_{v_{2}}(t-1)+\Delta_{v_{1}}(t-2) \\
&\Delta_{v_{2}}(t)=2\Delta_{v_{2}}(t-1)+\Delta_{v_{1}}(t-1)+2\Delta_{v_{2}}(t-2)+\Delta_{v_{1}}(t-2)
\end{split}
\right.
\end{equation}

$$\left\{\begin{split}&\Delta_{v_{1}}(2)=3,\qquad \Delta_{v_{1}}(3)=3,\qquad \Delta_{v_{1}}(4)=\Delta_{v_{2}}(3)+\Delta_{v_{1}}(2)=15 \\
&\Delta_{v_{2}}(2)=3, \qquad \Delta_{v_{2}}(3)=12,\qquad \Delta_{v_{2}}(4)=2\Delta_{v_{2}}(3)+\Delta_{v_{1}}(1)+2\Delta_{v_{2}}(2)+\Delta_{v_{1}}(2)=36
\end{split}
\right.$$
\end{thm}

\subsubsection{A networked model $N_{2}(p,q,t-1)$}

\textbf{Type-A operation.} \cite{Ma-Wang-Wang-Chaos-013136-2020} For a given edge $uv$ with two vertices $u$ and $v$, bringing an edge $xy$ on vertices $x$ and $y$ and then connecting vertex $u$ to $x$ and $v$ to $y$ using two new edges, respectively, produces a cycle $C_{4}$. Such a process is Type-A operation, shown in Fig.\ref{fig:PHYSA2019}(b).

\textbf{Type-B operation.} \cite{Ma-Wang-Wang-Chaos-013136-2020} For a given active edge $uv$ with two vertices $u$ and $v$, bringing two vertices $x$ and $y$, connecting vertex $x$ to two endpoints of edge $uv$ by two new edges and similar connections for vertex $y$ and vertex pair $u$ and $v$, as well as deleting active edge $uv$, together produces a cycle $C_{4}$ also. Such a process is Type-B operation, see Fig.\ref{fig:PHYSA2019}(b).

\textbf{Procedure of generating $N_{2}(t)$}

As $t=0$, the seminal graph $N_{2}(0)$ is a cycle $C_{4}$.

As $t\geq1$, the next graph $N_{2}(p,q,t)$ is obtained from $N(p,q,t-1)$ with $t\geq 1$ by applying either Type-A operation to each edge of graph $N_{2}(p,q,t-1)$ with probability $p$ or Type-B operation to the same edge of graph $N_{2}(p,q,t-1)$ with complementary probability $q$.

The networked model $N_{2}(p,q,t-1)$ obeys the \emph{power-law degree distribution}.

\begin{thm}\label{thm:666666}
\cite{Ma-Wang-Wang-Chaos-013136-2020} The closed-form solution of the total number of models $N_{2}(1,0,t)$ and $N_{2}(0,1,t)$ is given by, respectively,
$$S_{2}(1,0,t)=3^{\frac{A-1}{3}}\times\left(4/3\right)^{\frac{A+6t+5}{9}},\quad S_{2}(0,1,t)=4^{\frac{A-1}{3}},~A=4^{t+1}$$
\end{thm}

\subsubsection{Networked models $N_{3}(t)$}

\textbf{Type-A operation.} \cite{Ma-Luo-Wang-Zhu-Chaos-113120-2020} For a given vertex $u$, bringing an edge $xy$ on vertices $x$ and $y$ and then connecting vertex pair $x$ and $y$ with vertex $u$ together produces a cycle $C_{3}$, shown in Fig.\ref{fig:Chaos2020-2}(a). Such a process is viewed as Type-A operation.

\textbf{Type-B operation.} \cite{Ma-Luo-Wang-Zhu-Chaos-113120-2020} For a given active edge $uv$ with two vertices $u$ and $v$, bringing two vertices $x$ and $y$ linked by an edge $xy$, connecting vertex $x$ with one endpoint $u$ of edge $uv$ by a new edge and similar implementation for vertex $y$ and the other endpoint $v$ together produces a cycle $C_{4}$. Such a process is Type-B operation, see Fig.\ref{fig:Chaos2020-2}(b).

\begin{figure}
\centering
\includegraphics[width=16.4cm]{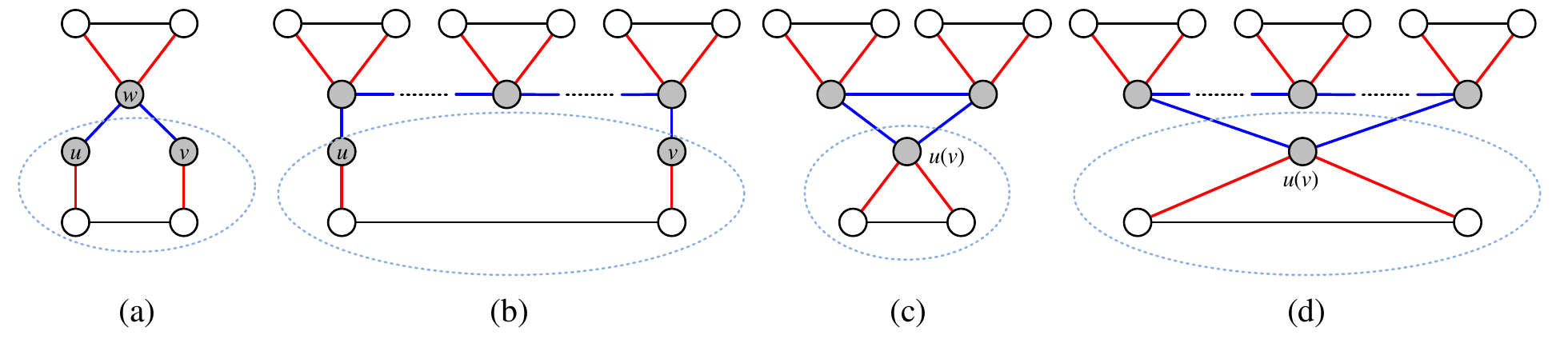}\\
\caption{\label{fig:Chaos2020-2}{\small The diagram for illustrating both Type-A and Type-B operations, cited from \cite{Ma-Luo-Wang-Zhu-Chaos-113120-2020}.}}
\end{figure}

\textbf{Procedure of generating $N_{3}(t)$}

Given a connected graph $\mathcal{G}(V,E)$ and a predefined probability parameter $p$ ($0<p<1$), one will encounter the following four cases plotted in Fig.\ref{fig:Chaos2020-2} after implementing Type-A operation exactly once. For brevity, we only make an elaborated description of \emph{1-1-configuration} here, and the remaining three can be easily understood with the help of both similar explanation to \emph{1-1-configuration} and illustration in Fig.\ref{fig:Chaos2020-2}. For a path $\mathcal{P}$ of length $2$, i.e., $\mathcal{P} :uwv$, vertex $w$ will be considered completely saturated after applying Type-A operation to vertex $w$, that is to say, its degree is changed by plus $2$. The other both vertices, however, are referred to as incompletely saturated ones due to no utilizing such an operation on them. On the other hand, to make them completely saturated ultimately, one needs to take useful advantage of Type-B operation for vertices $u$ and $v$ exactly once, which is highlighted using blue dashed cycle in Fig.\ref{fig:Chaos2020-2}(a). Model $N_{3}(t)$ is obtained by running Algorithm 1.

\begin{algorithm}
\caption{Algorithm $\mathcal{A}$ is to construct a random graph $\mathcal{G}(V(T),E(T))$ with identical exponential degree distribution based on a seminal model $\mathcal{G}(V(0),E(0))$ \cite{Ma-Luo-Wang-Zhu-Chaos-113120-2020}}
	\label{alg:Framwork}
	\begin{algorithmic}[1]
		\Require
		Initial graph $\mathcal{G}( V(0),E(0))$;
		Probability parameter $p$;
		End time $T$
		\Ensure
		$\mathcal{G}( V(T),E(T))$
		\State
			$ (V_{last},E_{last}) \leftarrow (V(0),E(0))$
		\For{$t = \{1,2...,T\}$}
			\State
				$\mathcal{G}( V(t),E(t)) \leftarrow \mathcal{G}( V(t-1),E(t-1))$
			\State
				$E_{now} \leftarrow \emptyset$
			\State
				$\Psi \leftarrow \emptyset$
			\For {\textbf{each} vertex $u \in V_{last}$}
			\If {$rand(0,1) < p$}
			\State
				update $\mathcal{G}( V(t),E(t))$ and $E_{now}$ by applying Type-A operation to vertex $u$
			\State
				add $u$ to $\Psi$
			\EndIf
			\EndFor
			\For {\textbf{each} cycle $uSu \in (V_{last},E_{last})$ s.t. $u \notin \Psi$ and $S\subseteq \Psi$}
			\State
				update $\mathcal{G}( V(t),E(t))$ and $E_{now}$ by applying Type-A operation to vertex $u$
			\EndFor
			\For {\textbf{each} edge $ uv \in E_{last}$ s.t. $u,v \notin \Psi$ }
			\State
				update $\mathcal{G}( V(t),E(t))$ and $E_{now}$ by applying Type-B operation to edge $uv$
			\EndFor
			\For {\textbf{each} path $uSv \in (V_{last},E_{last})$ s.t. $u,v \notin \Psi$ and $S\subseteq \Psi$}
			\State
				update $\mathcal{G}( V(t),E(t))$ and $E_{now}$ by applying Type-B operation to vertex pair $u$ and $v$
			\EndFor
			\State
				$ (V_{last},E_{last}) \leftarrow (V(t),E_{now})$
		\EndFor
		\State
		$\mathcal{G}( V(T),E(T)) \leftarrow \mathcal{G}( V(t),E(t))$
		\State
		\Return $\mathcal{G}( V(T),E(T))$
	\end{algorithmic}
\end{algorithm}

Each networked model $N_{3}(t)$ has the \emph{exponential degree distribution}.

\begin{thm}\label{thm:666666}
\cite{Ma-Luo-Wang-Zhu-Chaos-113120-2020} The closed-form solution of the total number of models $N_{3}(1,0,t)$ and $N_{3}(0,1,t)$ is given by, respectively,
$$S_{3}(1,0,t)=4\prod_{t_{i}=0}^{t-1}3^{|\Delta\mathfrak{V}(t_{i},t)|},$$
where $|\Delta\mathfrak{V}(t_{i},t)|=4\times 3^{t_{i}-1}$, and
$$S_{3}(0,1,t)=4\left[-\lambda+(4/3+\lambda)\left(\frac{1}{3}\right)^{t-1}\right]^{3}\times\left[3^{3^{t-1}}\times\prod_{i=0}^{t-2}\left(3\left[-\lambda+(4/3+\lambda)
\left(\frac{1}{3}\right)^{t-2-i}\right]^{2}\right)^{3^{i}}\right]^{4} $$
where $\lambda=-\frac{3}{2}$.
\end{thm}

\subsubsection{Self-similar fractal models $N_{4}(t)$ }

\begin{defn} \label{defn:111111}
\cite{Ma-Ma-Yao-Information-2018} An operation is called $V-E$ function, shorted as $f(V-E)$, if we just link two endpoints $u_{1}$ and $u_{2}$ of an edge $u_{1}u_{2}$ with a fixed vertex $v$ by two new edges $u_{1}v$ and $u_{2}v$, see Fig.\ref{fig:IPL2018}(a). Similarly, another is called $E-V$ function, denoted by $f(E-V)$, if we not only link a new vertex $u$ with two endpoints $v_{1}$ and $v_{2}$ of a stable edge $v_{1}v_{2}$ by two new edges $uv_{1}$ and $uv_{2}$, but also remove that old edge $v_{1}v_{2}$, as shown in Fig.\ref{fig:IPL2018}(b).\qqed
\end{defn}

\textbf{Procedure of generating $N_{4}(t)$}

For $t=0$, $N_{4}(0)$ is a cycle $C_{3}$ on three vertices and three edges.

For $t=1$, $N_{4}(1)$ is obtained from $N_{4}(0)$ by taking two following operations

\quad $(1)$ Each vertex of $N_{4}(0)$ is applied the $f(V-E)$ only one time;

\quad $(2)$ Every edge of $N_{4}(0)$ is applied the $f(E-V)$ only once.

For $t\geq2$, $N_{4}(t)$ is obtained from $N_{4}(t-1)$ by taking two operations described above.

\begin{figure}
\centering
\includegraphics[height=3cm]{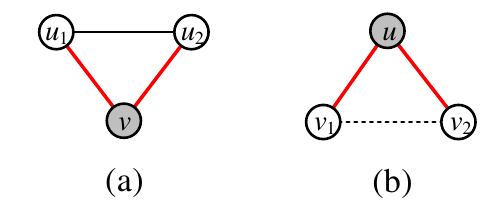}\\
\caption{\label{fig:IPL2018}{\small The diagrams for illustrating the functions $f(V-E)$ and $f(E-V)$, cited from \cite{Ma-Ma-Yao-Information-2018}.}}
\end{figure}

\begin{thm}\label{thm:666666}
\cite{Ma-Ma-Yao-Information-2018} The closed-form solution of the total number of models $N_{4}(t)$ is given by
$$S_{4}(t)=3^{\lambda^{t}}\left(2\times3^{\frac{3-\sqrt{13}}{2}}\right)^{F_{1}}\cdot \left(\frac{1}{3}\right)^{\frac{C_{1}-C_{2}}{\mu-1}}\cdot\left(\frac{1}{2}\right)^{W}$$
where
$$C_{1}=\frac{\mu^{t}-\lambda^{t}}{\mu-\lambda}, \quad C_{2}=\frac{\lambda^{t}-1}{\lambda-1},\quad W=\frac{D_{1}-\mu D_{2}-D_{3}}{\mu-1}$$
$$D_{1}=\frac{\mu^{t}-\lambda^{t}}{(\mu-\lambda)(\mu-1)},\quad
D_{2}=\frac{\lambda^{t}-1}{(\lambda-1)(\mu-1)},\quad D_{3}=\frac{\lambda^{t}-\lambda}{(\lambda-1)^{2}}-\frac{\lambda^{t}}{\lambda-1}-\frac{t-2}{\lambda-1},$$
$$F_{1}=\frac{\mu^{t}-\lambda^{t}}{\mu-\lambda}=C_{1},\quad \lambda=\frac{5+\sqrt{13}}{2},\quad \mu=\frac{5-\sqrt{13}}{2}.$$
\end{thm}

\subsubsection{Network models built by complete graph and iteration-function $N_{5}(t)$ }

\begin{defn} \label{defn:111111}
\cite{Ma-Ma-Yao-Physica-492-2018} Some models can be generated by the following network-iteration-function $\psi$. Let a small graph $G(V, E)$ be a initial model, namely, $G(V, E)$ is the model $N(1)$. We choose $i$ vertices in $N(1)$ to make a vertex set $V_{1}=\{v_{1}, v_{2},\dots,v_{i} \}$. After that, we can obtain a subgraph $N\,'(1)$ induced by the vertex set $V_{1}$ (or called vertex-induced subgraph $N\,'(1)$) if $E_{1}$ is a subset of $E$ such that edge $e\,'=uv$ belongs to $E_{1}$ if and only if $u$ and $v$ are in $V_{1}$. The vertex-induced subgraph $N\,'(1)$ can be called a kernel of $N(1)$, shorted as $\mathcal{K}(1)$. After we do $m$ replicas of initial model $N(1)$, for the same reason, also get $m$ corresponding kernels $\mathcal{K}_{1}(1)$, $\mathcal{K}_{2}(1)$, $\dots$, $\mathcal{K}_{m}(1)$. We will build the second model $N(2)$ by applying network-iteration-function $\psi$ to these $m$ kernels, such as
\begin{equation}\label{eqa:m-f-21}
\psi\left(\mathcal{K}_{1}(1), \mathcal{K}_{2}(1),\dots,\mathcal{K}_{m}(1), p\right)=N\,'(2)=G_{2}(V, E)
\end{equation}
where $V=\bigcup_{i=1}^{m}V_{i}(1)$, $E_{2}=\bigcup_{i=1}^{m}E_{i}(1)\bigcup\Omega(p)$, in which $\Omega(p)$ is a new edge set generated by connecting every vertex of $\mathcal{K}_{i}(1)$ to each one of $\mathcal{K}_{j}(1)$ ($ i, j\in[1,m], i\neq j $) with the linking probability $p ~(0<p\leq 1)$.\qqed
\end{defn}

\textbf{Procedure of generating $N_{5}(t)$}

For $t=1$, $N_{5}(1)$ is a complete graph with $6$ vertices, shorted as $K_{6}$, see Fig.\ref{fig:PHYSA2018-1}.

For $t=2$, $N_{5}(2)$ can be built from three replicas of $N_{5}(1)$ by using this network-iteration-function $\psi$ among every kernel $\mathcal{K}_{i}(1)$ ($ i\in[1,3] $) with the linking probability $p=1$, where each kernel $\mathcal{K}_{i}(1)$ has two vertices and an edge, see Fig.\ref{fig:PHYSA2018-1}.

For $t\geq3$, $N_{5}(t)$ can be built from arbitrary three replicas of $N_{5}(t-1)$ through the preceding same process, see Fig.\ref{fig:PHYSA2018-1}.

\begin{figure}
\centering
\includegraphics[width=15cm]{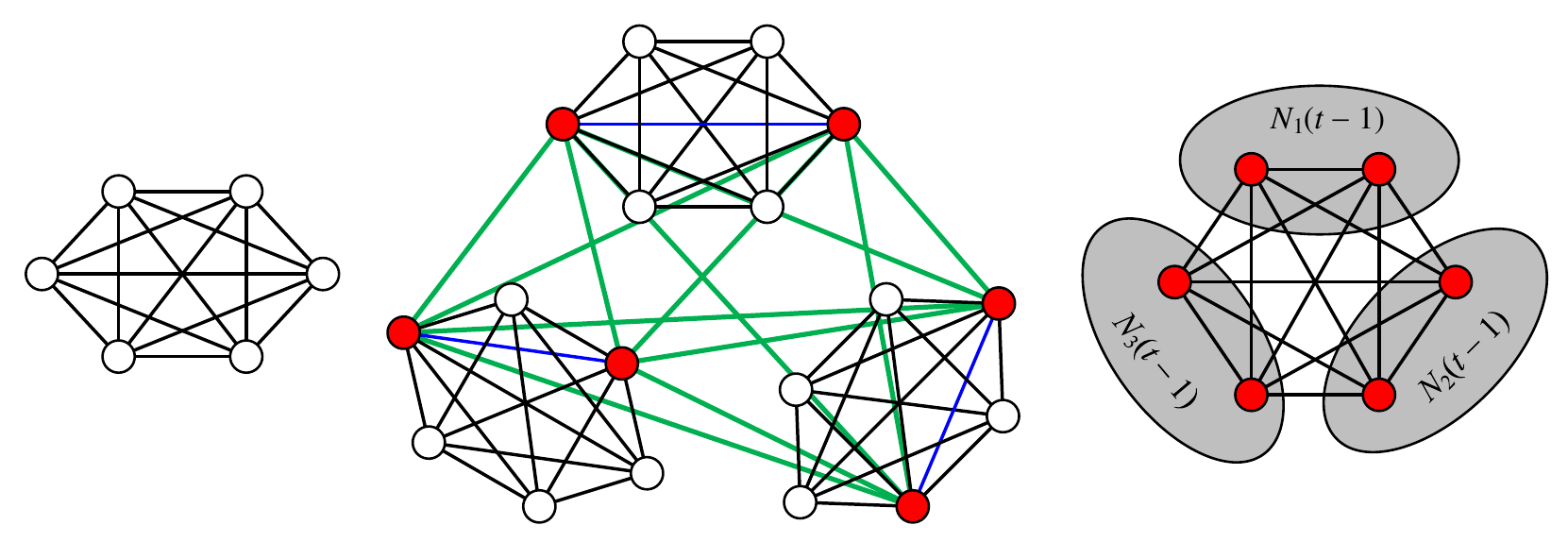}\\
\caption{\label{fig:PHYSA2018-1}{\small The diagrams for illustrating the network models $N_{5}(t)$, cited from \cite{Ma-Ma-Yao-Physica-492-2018}.}}
\end{figure}

Our models $N_{5}(t)$ are exponential-scale. Note that amounts of small-world networks around our real world including WS-network belong to this type.

\begin{thm}\label{thm:666666}
\cite{Ma-Ma-Yao-Physica-492-2018} The closed-form solution of the total number of models $N_{5}(t)$ is given by
\begin{equation}\label{eqa:m-f-471}
S_{5}(t)=48^{\frac{3^{t-1}-1}{2}}486^{3^{t-1}}\left[\frac{8}{3}+2(t-1)\right]\prod_{i=0}^{t-2}\left[\frac{8}{3}+2(t-1-i)\right]^{3^{i}}
\end{equation}
\end{thm}

\subsubsection{Vertex-edge-growth small-world network models $N_{6}(t)$ }

\begin{defn} \label{defn:111111}
\cite{Ma-Su-Hao-Yao-Yan-Physica-492-2018} An operation is called $V-E$ function (upright triangle-growth), shorted as $f(V-E)$, if two endpoints $u_{1}$ and $u_{2}$ of an edge $u_{1}u_{2}$ are linked with a fixed vertex $v$ by two new edges, shown in Fig.\ref{fig:IPL2018}(a). Similarly, another is called $E-V$ function (inverted triangle-growth), shorted as $f(E-V)$, if a new vertex $u$ is connected to two endpoints $v_{1}$ and $v_{2}$ of a stable edge $v_{1}v_{2}$ by two new edges, shown in Fig.\ref{fig:IPL2018}(b).\qqed
\end{defn}

\textbf{Procedure of generating $N_{6}(t)$}

For $t=0$, $N_{6}(0)$ is a complete graph $K_{3}$ which has three vertices and three edges.

For $t=1$, $N_{6}(1)$ is obtained from $N_{6}(0)$ by taking two operations, following

\quad $(1)$ Apply the $f(V-E)$ only one time to each vertex of $N_{6}(0)$;

\quad $(2)$ Apply the $f(E-V)$ only once to every edge of $N_{6}(0)$.

For $t\geq2$, $N_{6}(t)$ is constructed by implementing two operations $f(V-E)$ and $f(E-V)$ onto every vertex and any edge of $N_{6}(t-1)$, respectively.

The model $N_{6}(t)$ follows the \emph{power-law degree distribution} and has scale-free feature.

\begin{thm}\label{thm:666666}
\cite{Ma-Su-Hao-Yao-Yan-Physica-492-2018} The closed-form solution of the total number of models $N_{6}(t)$ is given by
\begin{equation}\label{eqa:theorem-3-1}
S_{6}(t)=2^M\cdot 3^L
\end{equation}
where
$$M=\frac{Q_{91}-3^{3}Q_{92}}{2(3-\beta)}-\frac{Q_{11}-Q_{92}-(\beta-1)Q_{12}}{2\left(\beta-1\right)^{2}}+\alpha^{t-1}+(8-\alpha)Q_{13},$$
$$L=\frac{Q_{91}-3^{3}Q_{92}}{2(3-\beta)}+\frac{Q_{11}-Q_{92}-(\beta-1)Q_{12}}{2\left(\beta-1\right)^{2}}+6\alpha^{t-1}+(29-6\alpha)Q_{13}$$
and
$$Q_{91}=\frac{3^{4}\alpha^{t-2}-3^{t+2}}{\alpha-3},\quad Q_{92}=\frac{\beta\alpha^{t-2}-\beta^{t-1}}{\alpha-\beta};\quad Q_{11}=\frac{\beta^{2}\alpha^{t-2}-\beta^{t}}{\alpha-\beta},$$
$$Q_{12}=\frac{\alpha^{t-2}-1}{\alpha-1}, \quad Q_{13}=\frac{\alpha^{t-1}-\beta^{t-1}}{\alpha-\beta};\quad \alpha=3+\sqrt{3},\quad \beta=3-\sqrt{3}.$$
\end{thm}

\subsubsection{Self-similar small-world scale-free network models $N_{7}(t)$ }

\textbf{Procedure of generating $N_{7}(t)$} \cite{Ma-Yao-Eur-Phys-2018}

For $t=0$, $N_{7}(0)$ is an active path with length $3$ whose two end vertices are excited.

For $t=1$, $N_{7}(1)$ can be obtained from $N_{7}(0)$ by adding two new paths on $5$ vertices and linking the two pendant vertices of every new path with those two excited end vertices of $N_{7}(0)$. The result leads to this initial active path being disable and generates $6$ new active path of length $3$ belonging to $N_{7}(1)$.

For $t\geq 2$, $N_{7}(t)$ can be obtained from $N_{7}(t-1)$ by the construction process narrated in step 2. The six vertices having the largest degree can be called hub-vertex, $H_{i}(t), i\in[1,6]$. See Fig.\ref{fig:EPJB2018} for more information.

\begin{figure}
\centering
\includegraphics[width=14cm]{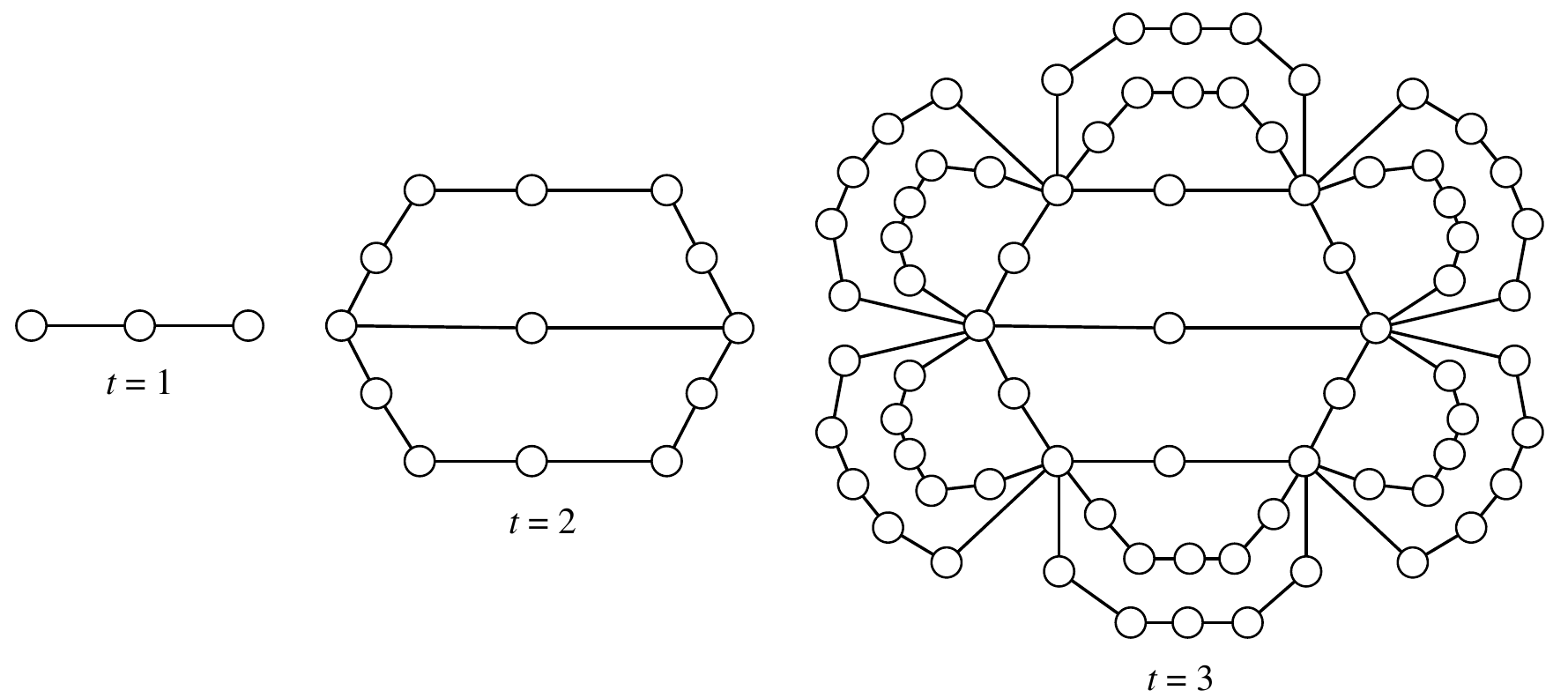}\\
\caption{\label{fig:EPJB2018}{\small The diagrams for illustrating the network model $N_{7}(t)$, cited from \cite{Ma-Yao-Eur-Phys-2018}.}}
\end{figure}

Our \emph{deterministic scale-free network model} $N_{7}(t)$ has degree exponent $\gamma > 3$.

\begin{thm}\label{thm:666666}
\cite{Ma-Yao-Eur-Phys-2018} The closed-form solution of the total number of models $N_{7}(t)$ is given by
\begin{equation}\label{eqa:m-f-726}
S_{7}(t)=60^{6^{t-1}}\prod_{i=1}^{t-1}\left(\left[\frac{47}{12}\left(\frac{5}{3}\right)^{i}-\frac{9}{4}\right]^{-1}\left\{9\left[\frac{47}{12}\left(\frac{5}{3}\right)^{i}-\frac{9}{4}\right]^{-1}+10\right\}\right)^{6^{t-1-i}}
\end{equation}
\end{thm}

\subsubsection{Small-world network models $N_{8}(t)$}

\textbf{Procedure of generating $N_{8}(t)$} \cite{Ma-Yao-Physica-2017}

For $t=1$, $N_{8}(1, m)$ is a complete graph with $m$ vertices, shorted as $K_{m}$.

For $t=2$, $N_{8}(2, m)$ can be obtained from two models $N_{8}(1, m)$ connected by part-join-operation. In order to facilitate the narrative, a $N_{8}(1, m)$ is called $L-N_{8}(1, m)$, another is $R-N_{8}(1, m)$. Two randomly selected vertices of $L-N_{8}(1, m)$ are linked to two arbitrary ones of $R-N_{8}(1, m)$ by four new edges, which leads to the model $N_{8}(2, m)$. Note that four chosen vertices are called hub-vertex $\mathcal{H}(2)$. The resulting new model made up hub-vertex $\mathcal{H}(2)$ and those four relative edges is a kernel $\mathcal{K}(2)$ of $N_{8}(2, m)$.

For $t\geq 3$, $N_{8}(t, m)$ can be obtained from two models $N_{8}(t-1, m)$ in the same way as used above, see Fig.\ref{fig:PHYSA2017}.

\begin{figure}
\centering
\includegraphics[width=16.4cm]{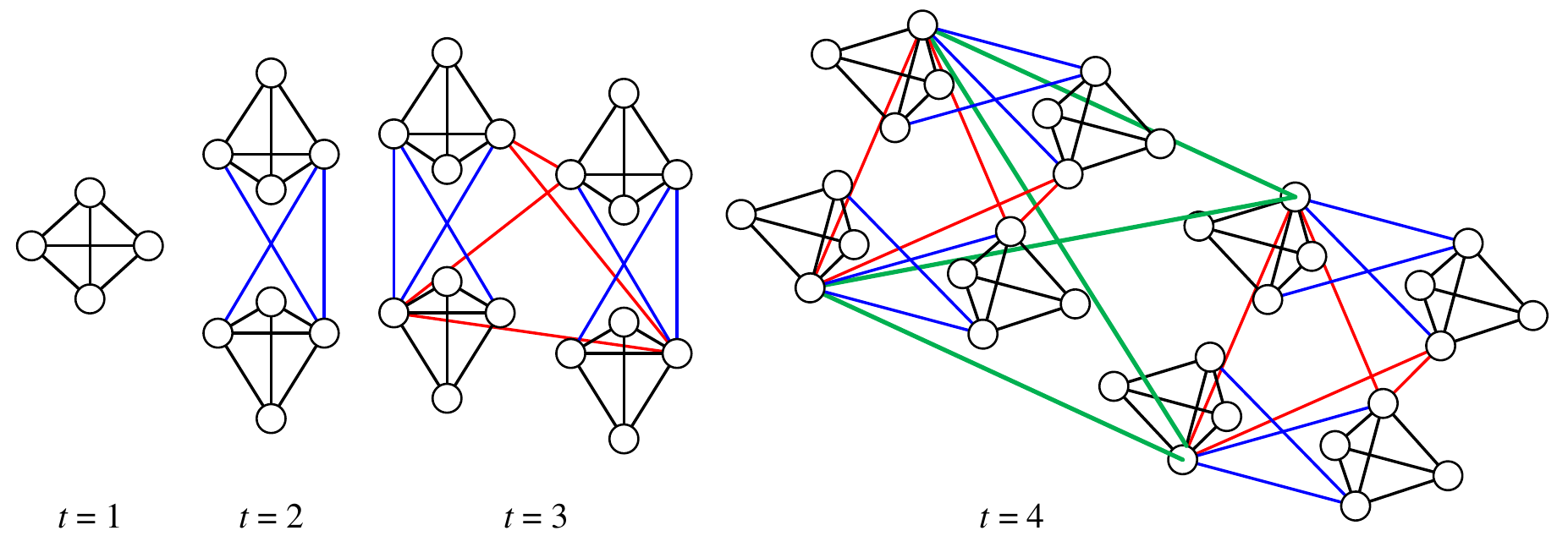}\\
\caption{\label{fig:PHYSA2017}{\small The diagrams for illustrating the network model $N_{8}(t)$, cited from \cite{Ma-Yao-Physica-2017}.}}
\end{figure}

$N_{8}(t)$ is a \emph{scale-free} and \emph{small-world network model}, since $P_{cum}(k)\propto k^{-\frac{\ln 6}{\ln 2}}$.

\begin{thm}\label{thm:666666}
\cite{Ma-Yao-Physica-2017} The closed-form solution of the total number of models $N_{8}(t)$ is given by
\begin{equation}\label{eqa:m-f-6}
\begin{split}S_{8}(t)=&F(1)^{2^{2^{t}-2}}\prod_{i=t-2}^{0}\left [\frac{1}{\beta^{t-1}\left(\frac{1}{2-\lambda}+\alpha\right)-\alpha}-\lambda\right]\cdot \\
&\cdot \left [\left(\frac{1}{\beta^{t-1}\left(\frac{1}{2-\lambda}+\alpha\right)-\alpha}-\lambda\right)^{2}+4\left(\frac{1}{\beta^{t-1}\left(\frac{1}{2-\lambda}+\alpha\right)-\alpha}-\lambda\right)+3\right ]^{2^{t-i-1}}\\
\end{split}
\end{equation}
where $\lambda=\frac{\sqrt{17}-9}{2}, \mu=\frac{\sqrt{17}-1}{2}, \alpha=\frac{2-\lambda}{\mu}$ and $\beta=3-\lambda$.
\end{thm}

\subsection{Self-similar networked modules and lattices}

Self-similarity is a common phenomena between a part of a complex system and the whole of the system. The similarity between the fine structure or property of different parts can reflect the basic characteristics of the whole. In other word, the invariance under geometric or non-linear transformation: the similar properties in different magnification multiples, including geometry. The mathematical expression of self-similarity is defined by
\begin{equation}\label{eqa:self-similarity}
\theta(\lambda r)=\lambda\alpha \theta(r), \textrm{ or } \theta(r)\sim r\alpha,
\end{equation} where $\lambda $ is called \emph{scaling factor}, and $\alpha$ is called \emph{scaling exponent} (fractal dimension) and describes the spatial properties of the structure. The function $\theta(r)$ is a measure of the occupancy number, quantity and other properties of area, volume, mass, \emph{etc} (Wikipedia).

In mathematics, a self-similar object is exactly or approximately similar to a part of itself (i.e. the whole has the same shape as one or more of the parts). Many objects in the real world, such as coastlines, are statistically self-similar: parts of them show the same statistical properties at many scales (Ref. \cite{Mandelbrot-Benoit-B-1967}). Another reason is that many tree-like graphs admit many graph labelings for making Hanzi-gpws easily. We present the following constructive leaf-algorithms for building up particular self-similar tree-like networked modules \cite{Yao-Mu-Sun-Sun-Zhang-Wang-Su-Zhang-Yang-Zhao-Wang-Ma-Yao-Yang-Xie2019}. See several self-similar trees shown in Fig.\ref{fig:2021-12-self-similar-tree} and Fig.\ref{fig:a-self-similar-0}.

\begin{figure}[h]
\centering
\includegraphics[width=14cm]{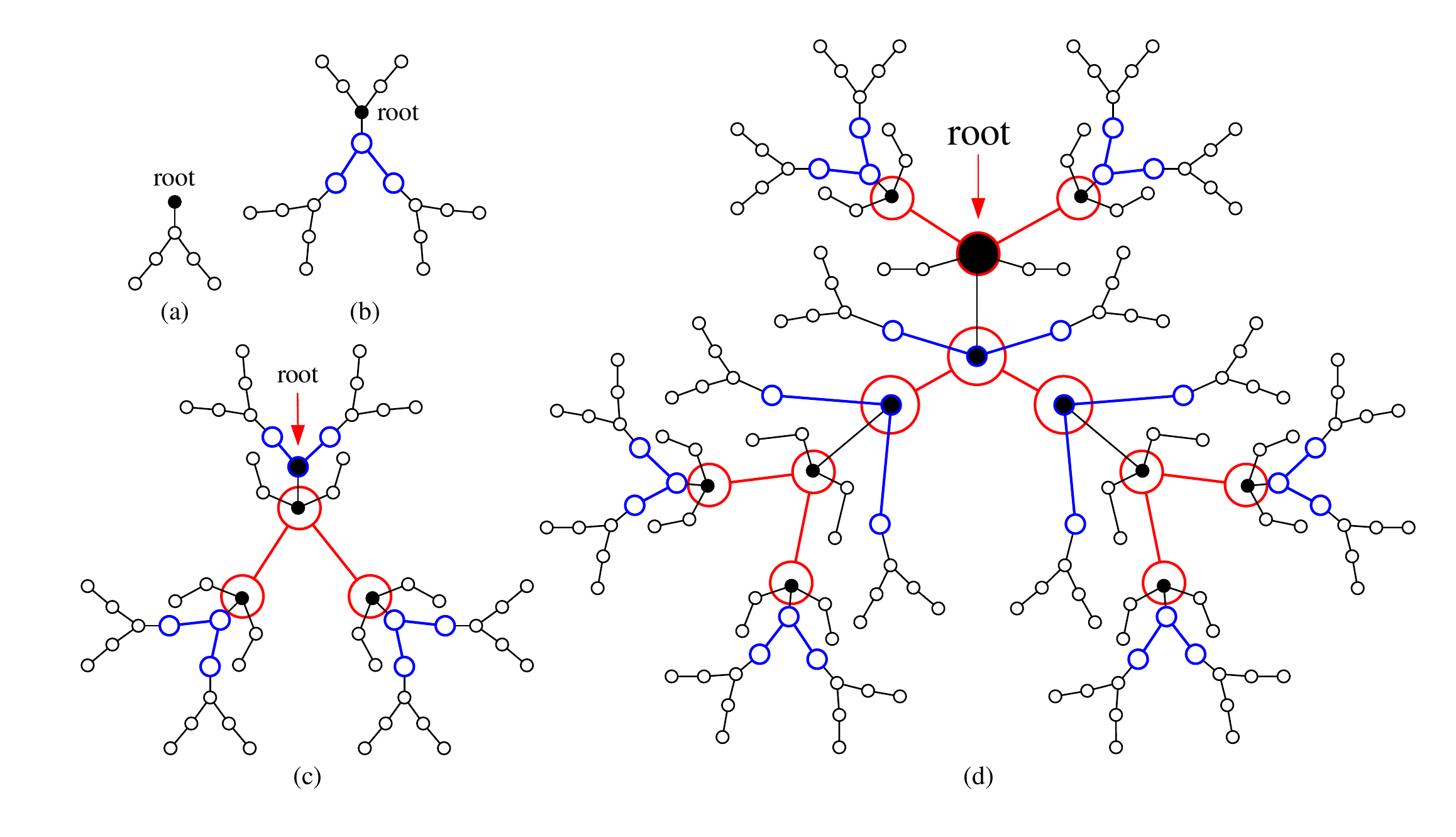}\\
\caption{\label{fig:2021-12-self-similar-tree} {\small Examples for self-similar trees.}}
\end{figure}

\begin{figure}[h]
\centering
\includegraphics[width=9cm]{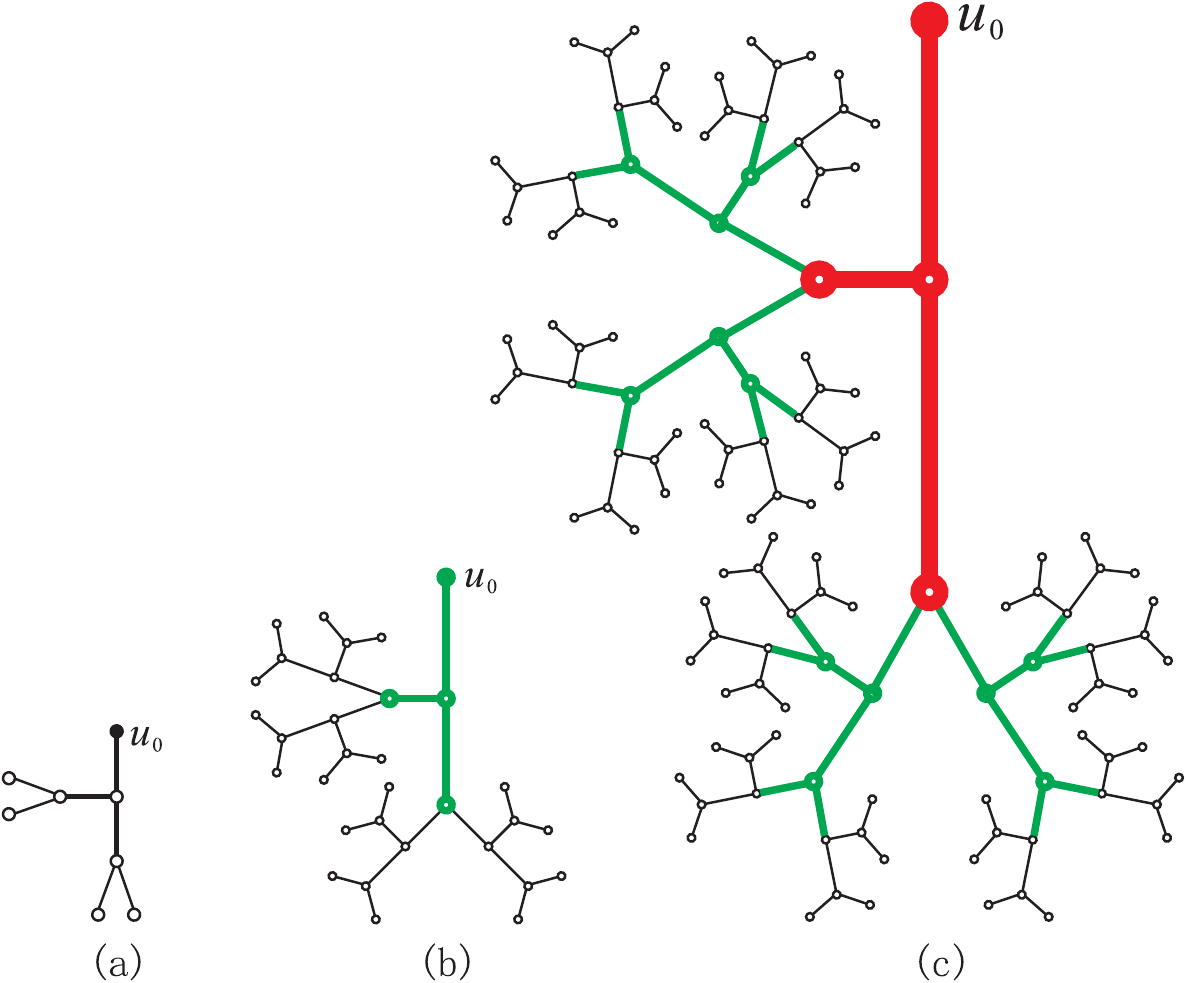}~\includegraphics[width=7cm]{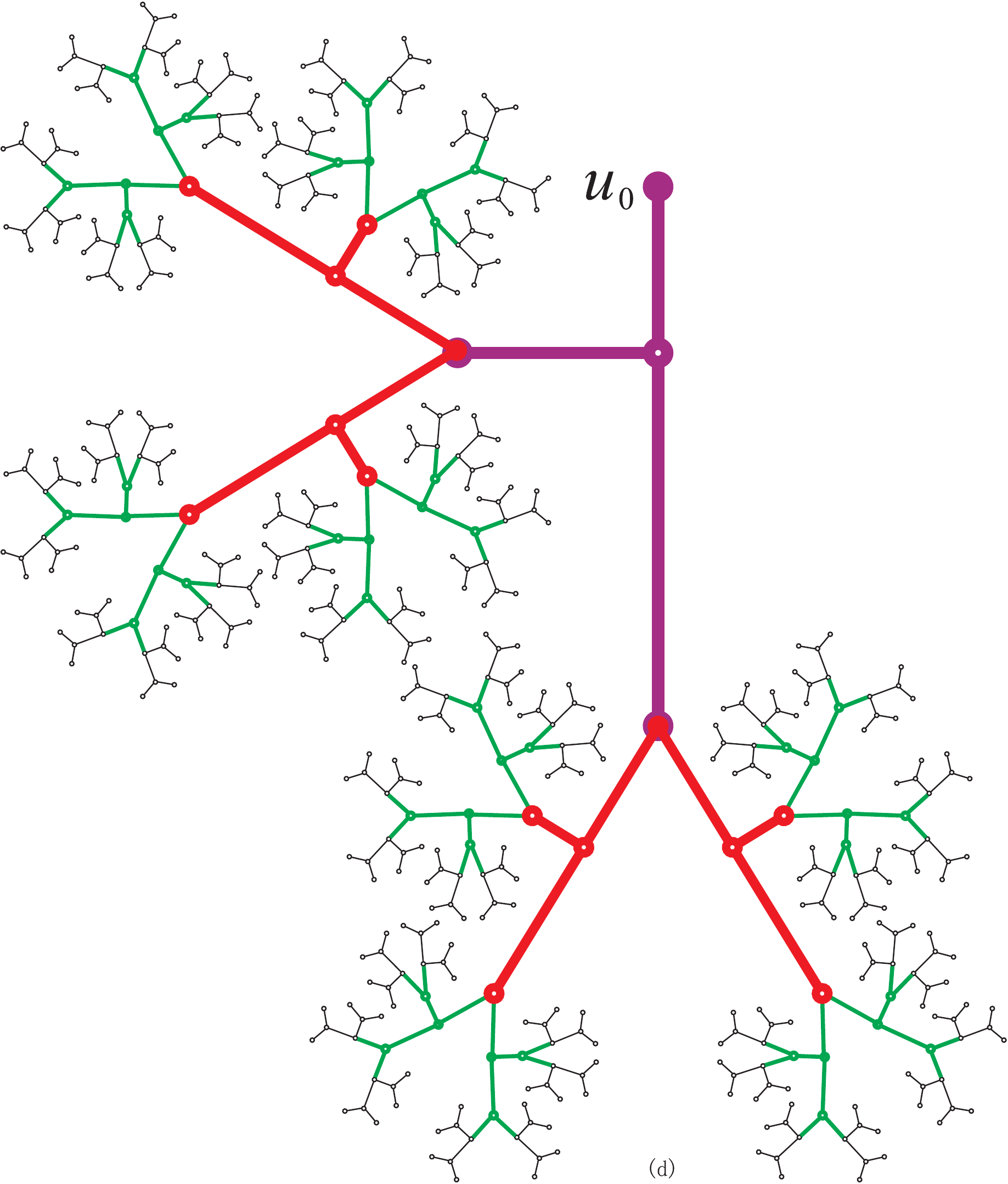}
\caption{\label{fig:a-self-similar-0} {\small (b),(c) and (d) are self-similar trees.}}
\end{figure}

\subsubsection{Leaf-algorithm-A}

Let $N(1)$ be a tree-like networked module on $n~(\geq 3)$ vertices and let $L(N(1))=\{y_i:i\in[1,m]\}$ be the set of leaves of $N(1)$, see an example $L(N(1))=\{y_i:i\in[1,6]\}$ shown in Fig.\ref{fig:Leaf-algorithm-A}(a). We refer a vertex $u_0\in V(N(1))$ to be the \emph{root} of $N(1)$, and write $L(N(1))\setminus \{u_0\}=\{u_j:j\in [1,m]\}$, where $m=|L(N(1))\setminus \{u_0\}|$, and assume that each leaf $y_i\in L(N(1))\setminus \{u_0\}$ is adjacent to $v_i\in V(N(1))\setminus L(N(1))$ with $i\in [1,m]$.

There are the copies $N_i(1)$ of $N(1)$ with $u_{0,i}\in V(N_i(1))$ to be the image of the root $u_0$ with $i\in [1,m]$. We do: (i) delete each leaf $y_i\in L(N(1))\setminus \{u_0\}$; and (ii) vertex-coincide the \textrm{root vertex} $u_{0,i}$ of the tree $N_i(1)$ with $v_i$ into one vertex $u_{0,i}\odot v_i$ for $i\in [1,m]$, the resultant tree is denoted as $N(2)=\odot_A\langle N(1),\{N_i(1)\}^m_1\rangle$ and called a \emph{uniformly $1$-rank self-similar tree} with root $u_0$, see Fig.\ref{fig:Leaf-algorithm-A}(d) and (e). Go on in the way, we have \emph{uniformly $N(1)$-leaf $t$-rank self-similar trees} $N(t)=\odot_A \langle N(1),\{N_i(t-1)\}^m_1\rangle$ with the root $u_0$ and $t\geq 1$. Moreover, we called $N(t)$ a \emph{uniformly $N(1)$-leaf self-similar networked module} with the root at time step $t$.

Obviously, every uniformly $k$-rank self-similar tree $N(k)=\odot_A \langle N(1),\{N_i(k-1)\}^m_1\rangle$ with $k\geq 2$ is similar with $N(1)$ as regarding each $N_i(k-1)$ as a ``\emph{leaf}''. If the root $u_0$ is a leaf of $N(1)$, then the uniformly $k$-rank self-similar trees $N\,'(k)=\odot_A \langle N(),\{N\,'_i(k-1)\}^m_1\rangle$ have some good properties, see Fig.\ref{fig:Leaf-algorithm-A}(d), Fig.\ref{fig:Leaf-algorithm-A-1}(e), Fig.\ref{fig:Leaf-algorithm-A-1}(f) and Fig.\ref{fig:Leaf-algorithm-A-1}(g).

The vertex number $v(N(t))$ and edge number $e(N(t))$ of each uniformly $N(1)$-leaf $t$-rank self-similar tree $N(t)=\odot_A \langle N(1),\{N_i(t-1)\}^m_1\rangle$ can be computed by the following way:
\begin{equation}\label{eqa:Leaf-algorithm-A}
\left\{
{
\begin{split}
v(N(t))&=v(N(1))m^t+[v(N(1))-2m]\sum^{t-1}_{k=0}m^k\\
e(N(t))&=v(N(t))-1
\end{split}}
\right.
\end{equation}
where $m=|L(N(1))\setminus \{u_0\}|$.

\begin{figure}[h]
\centering
\includegraphics[width=16.4cm]{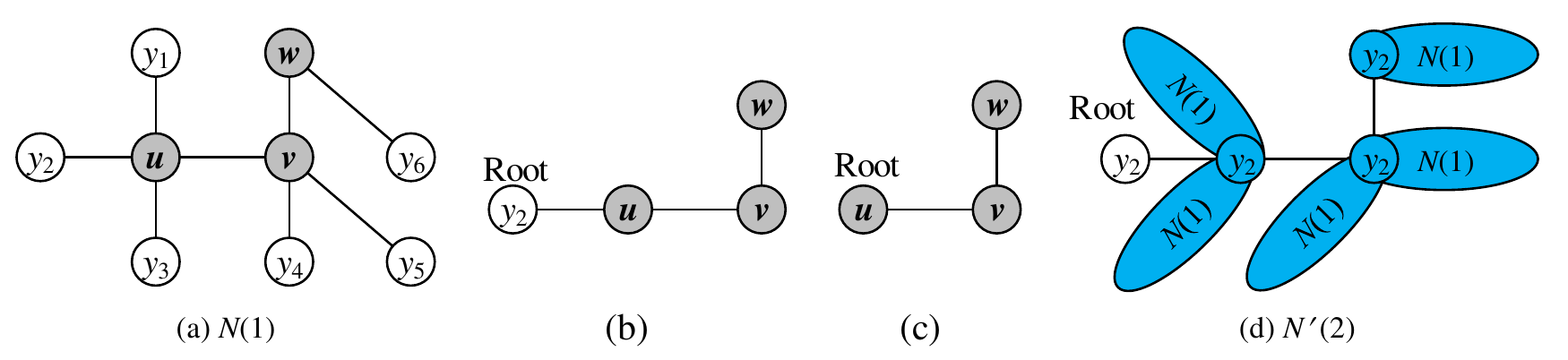}\\
\caption{\label{fig:Leaf-algorithm-A} {\small (a) A tree-like networked module $N(1)$ with the root $y_2$; (b) removing all leaves of $N(1)$ with the root $y_2$, except the root $y_2$; (c) removing all leaves of $N(1)$ with the root $u$.}}
\end{figure}

\begin{figure}[h]
\centering
\includegraphics[width=16.4cm]{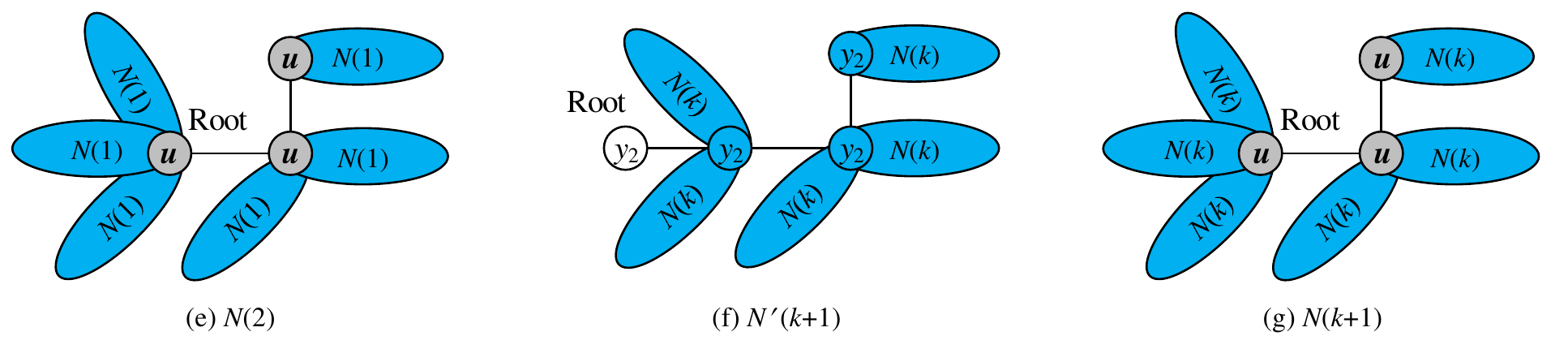}\\
\caption{\label{fig:Leaf-algorithm-A-1} {\small (f) Networked module $N\,'(k+1)$ with the root $y_2$ for $k\geq 1$; (g) networked module $N(k+1)$ with the root $u$ for $k\geq 1$.}}
\end{figure}

\subsubsection{Leaf-algorithm-B}

We take $n$ copies $H_{0,1}$, $H_{0,2}$, $\dots $, $H_{0,n}$ of a tree $H_0$, where $n=|L(H_0)|$ is the number of leaves of $H_0$, and do:

(1) delete each leaf $x_i$ from $H_0$, where $x_i$ is adjacent with $y_i$ such that the edge $x_iy_i\in E(H_0)$, clearly, $y_i$ may be adjacent two or more leaves; then

(2) vertex-coincide some vertex $x_{0,i}\in V(H_{0,i})$ with the vertex $y_i$ into one vertex $x_{0,i}\odot y_i$ for $i\in [1,n]$.

The resultant tree is denoted as $H_1=\odot_B \langle H_0,\{H_{0,i}\}^n_1\rangle $. Proceeding in this way, we get trees $H_j=\odot_B\langle H_0, \{H_{j-1,i}\}^n_1\rangle$ for $j\geq 2$, where $H_{j-1,1}$, $H_{j-1,2}$, $\dots $, $H_{j-1,n}$ are the copies of $H_{j-1}$, and deleting leaves $x_i$ from $H_0$ and vertex-coincide an arbitrary vertex $x_{j-1,i}\in V(H_{j-1,i})$ with the vertex $y_i$ of $H_0$ into one vertex $x_{j-1,i}\odot y_i$ for $i\in [1,n]$. It refers to each tree $H_k=\odot_B\langle H_0, \{H_{k-1,i}\}^n_1\rangle$ as an \emph{$H_0$-leaf $k$-rank self-similar tree} without root with $k\geq 1$. We name $H_k$ as an \emph{$H_0$-leaf self-similar networked module} at time step $k$ if $H_0$ is a graph.

The vertex number $v(H_t)$ and edge number $e(H_t)$ of each $H_0$-leaf $k$-rank self-similar tree $H_t=\odot_B \langle H_{0},\{H_{t-1,i}\}^n_1\rangle$ can be computed in the following way:
\begin{equation}\label{eqa:Leaf-algorithm-B}
{
\begin{split}
v(H_t)=v(H_0)n^t+[v(H_0)-2n]\sum^{t-1}_{k=0}n^k,\quad e(H_t)=v(H_t)-1
\end{split}}
\end{equation}

\subsubsection{Leaf-algorithm-C}

Let each of vertex disjoint trees $G_{0,1}$, $G_{0,2}$, $\dots $, $G_{0,m(0)}$ be a copy of a tree $G_0$, where $m(0)=|L(G_0)|$ is the number of leaves of $G_0$. We delete each leaf $x_i$ from $G_0$ and vertex-coincide some vertex $x_{0,i}$ of $G_{0,i}$ with the vertex $y_i$ into one vertex $x_{0,i}\odot y_i$ for $i\in [1,m(0)]$, where the edge $x_iy_i\in E(G_0)$. The resultant tree is denoted as $G_1=\odot _C\big \langle G_0,\{G_{0,i}\}^{m(0)}_1\big \rangle $. Proceeding in this way, each tree $G_k=\odot_C\big \langle G_{k-1}, \{G_{k-1,i}\}^{m(k-1)}_1\big \rangle$ with $k\geq 2$ is obtained by removing each leaf $w_i$ of $G_{k-1}$, and then vertex-coincide some vertex $z_{k-1,i}$ of $G_{k-1,i}$ being a copy of $G_{k-1}$ with the vertex $w\,'_i$ of $G_{k-1}$ into one vertex $z_{k-1,i}\odot w\,'_i$ for $i\in [1,m(k-1)]$, where the leaf $w_i$ is adjacent with $w\,'_i$ in $G_{k-1}$, and $m(k-1)=|L(G_{k-1})|$ is the number of leaves of $G_{k-1}$. We refer to each tree $G_k$ as a \emph{leaf-$k$-rank self-similar tree} with $k\geq 1$, and call $G_k=\odot_C\big \langle G_{k-1}, \{G_{k-1,i}\}^{m(k-1)}_1\big \rangle$ a \emph{leaf-$k$-rank self-similar networked module} at time step $k\geq 1$.

Let $(a_1,a_2,\dots ,a_{k})$ be a combinator selected from integer numbers $0,1,\dots ,s-1$ with $s>1$, so we have the number ${s\choose k}$ of different combinators $(a_1,a_2,\dots ,a_{k})$ in total, and put them into a set $F_k$. Each leaf-$k$-rank self-similar tree $G_t$ has its vertex number $v(G_t)$ and edge number $e(G_t)$ as follows:
\begin{equation}\label{eqa:Leaf-algorithm-B}
{
\begin{split}
v(G_t)&=v(G_0)\prod^{s-1}_{k=0}[1+m(k)]-2\sum^{s-1}_{k=1}\sum^{{s\choose k}}_{(a_1,a_2,\dots ,a_{k})\in F_k}m(a_1)m(a_2)\cdots m(a_k)\\
e(G_t)&=v(G_t)-1
\end{split}}
\end{equation}
Moreover, if the vertex $z_{k-1,i}$ of $G_{k-1,i}$ is not a leaf of $G_{k-1,i}$, then each $G_{k-1}$ has $m(k)=[m(0)]^{2^k}$ leaves in total.

\subsubsection{Self-similar networked lattices}

Let
$$
\textbf{\textrm{N}}(t)=(N_1(t),N_2(t),\dots ,N_m(t))=(N_1,N_2,\dots ,N_m)(t)=(N_i)^m_{i=1}(t)
$$ be a \emph{$T$-self-similar networked module base} at time step $t$, where each networked module $N_i(t)$ is a $T$-self-similar networked module, that is, $N_i(t)$, by the vertex-splitting operation and the edge-splitting operation, can be split into vertex-disjoint graphs $N_{i,1}(t),N_{i,2}(t),\dots ,N_{i,a(i)}(t)$ at time step $t$, such that each graph $N_{i,j}(t)\cong T$ for $j\in [1,a(i)]$, and $E(N_i(t))=\bigcup^{a(i)}_{j=1}E(N_{i,j}(t))$. We rearrange randomly $b_1N_1(t),b_2N_2(t),\dots ,b_mN_m(t)$ into $M_1(t),M_2(t),\dots ,M_A(t)$ with $A=\sum ^m_{k=1}b_k$, and introduce the following operations:

\textbf{Oper-1.} We vertex-coincide a vertex $w_1$ of $M_1(t)$ with a vertex $w_2$ of $M_2(t)$ into one vertex $w_1\odot w_2$, the resultant networked module is denoted as $L_{1,2}(t)=M_1(t)\odot M_2(t)$, and we vertex-coincide a vertex $w_{1,2}$ of $L_{1,2}(t)$ with a vertex $w_3$ of $M_3(t)$ into one vertex $w_{1,2}\odot w_3$, the resultant networked module is written as $L_{2,3}(t)=L_{1,2}(t)\odot M_3(t)$, go on in this way, we get networked modules $L_{i+1,i+2}(t)=L_{i,i+1}(t)\odot M_{i+3}(t)$ with $i\in [1,A-3]$. For simplicity, we write $L_{A-2,A-1}(t)=L_{A-3,A-2}(t)\odot M_{A}(t)$ as
$$
L_{A-2,A-1}(t)=\odot |^A_{k=1} M_k(t)=\odot |^m_{k=1} b_kN_k(t)
$$ Clearly, each networked module $\odot |^m_{k=1} b_kN_k(t)$ is a $T$-self-similar networked module. We get a \emph{$T$-self-similar networked lattice} as follows
\begin{equation}\label{eqa:self-similar-networked-lattice-11}
\textbf{\textrm{L}}(\odot Z^0\textbf{\textrm{N}}(t))=\big \{\odot |^m_{k=1} b_kN_k(t):b_k\in Z^0,N_k(t)\in \textbf{\textrm{N}}(t)\big \}
\end{equation} based on the $T$-self-similar networked module base $\textbf{\textrm{N}}(t)=(N_i)^m_{i=1}(t)$.

\textbf{Oper-2.} We vertex-coincide a vertex $z_1$ of $H(t)$ with a vertex $w_1$ of $M_1(t)$ into one vertex $z_1\odot w_1$, the resultant networked module is denoted as $H_{1}(t)=H(t)\odot M_1(t)$, and then vertex-coincide a vertex $z_2$ of $H(t)$ with a vertex $w_2$ of $M_2(t)$ into one vertex $z_1\odot w_2$, the resultant networked module is written as $H_{2}(t)=H_{1}(t)\odot M_2(t)$, in general, we obtain $H_{i}(t)=H_{i-1}(t)\odot M_{i}(t)$ with $i\in [1,A]$, where $H_{0}(t)=H(t)$, such that each $H_{i}(t)$ is a $T$-self-similar networked module. The last networked module $H_{A}(t)$ is written as $H_{A}(t)=H(t)\odot |^m_{k=1} b_kN_k(t)$, we call the following set
\begin{equation}\label{eqa:self-similar-networked-lattice-22}
\textbf{\textrm{L}}(\textbf{\textrm{F}}(t)\odot Z^0\textbf{\textrm{N}}(t))=\{H(t)\odot |^m_{k=1} b_kN_k(t):b_k\in Z^0,N_k(t)\in \textbf{\textrm{N}}(t),H(t)\in \textbf{\textrm{F}}(t)\}
\end{equation} a \emph{$T$-self-similar networked lattice} based on the $T$-self-similar networked module base $\textbf{\textrm{N}}(t)=(N_i)^m_{i=1}(t)$, and $\textbf{\textrm{F}}(t)$ is a set of networked modules.

Similarly, the edge-coinciding operation enables us to have $T$-self-similar networked modules $\ominus |^m_{k=1} b_kN_k(t)$ and $H(t)\ominus |^m_{k=1} b_kN_k(t)$ based on a $T$-self-similar networked module base $\textbf{\textrm{N}}(t)=(N_i)^m_{i=1}(t)$, as well as two $T$-self-similar networked lattices
\begin{equation}\label{eqa:self-similar-networked-lattice-33}
{
\begin{split}
\textbf{\textrm{L}}(\ominus Z^0\textbf{\textrm{N}}(t))&=\big \{\ominus |^m_{k=1} b_kN_k(t):b_k\in Z^0,N_k(t)\in \textbf{\textrm{N}}(t)\big \}\\
\textbf{\textrm{L}}(\textbf{\textrm{F}}(t)\ominus Z^0\textbf{\textrm{N}}(t))&=\big \{H(t)\ominus |^m_{k=1} b_kN_k(t):b_k\in Z^0,N_k(t)\in \textbf{\textrm{N}}(t),H(t)\in \textbf{\textrm{F}}(t)\big \}
\end{split}}
\end{equation} based on the edge-coinciding operation.

Since the networked modules $b_1N_1(t),b_2N_2(t),\dots ,b_mN_m(t)$ can be arranged into $M_1(t)$, $M_2(t)$, $\dots $, $M_A(t)$ with $A=\sum ^m_{k=1}b_k$, the networked module $H_1(t)=M_1(t)[\ominus\odot]M_2(t)$ is obtained by doing one operation of two operations $\ominus$ and $\odot$ to $M_1(t)$ and $M_2(t)$, so we have networked modules $H_k(t)=H_{k-1}(t)[\ominus\odot]M_{k+1}(t)$ with $k\in [1,A-1]$ and $H_0(t)=M_1(t)$ at time step $t$. Obviously, each networked module $H_k(t)$ is a $T$-self-similar networked module. So we can write
$$
H_{A-1}(t)=[\ominus\odot]^{A}_{i=1}M_i(t)=[\ominus\odot]^m_{k=1}N_k(t)
$$ and get the following $T$-self-similar networked lattices
\begin{equation}\label{eqa:self-similar-networked-lattice-mixed}
{
\begin{split}
\textbf{\textrm{L}}([\ominus\odot] Z^0\textbf{\textrm{N}}(t))&=\big \{[\ominus\odot]^m_{k=1} b_kN_k(t):b_k\in Z^0,N_k(t)\in \textbf{\textrm{N}}(t)\big \}\\
\textbf{\textrm{L}}(\textbf{\textrm{F}}(t)[\ominus\odot] Z^0\textbf{\textrm{N}}(t))&=\big \{H(t)[\ominus\odot]^m_{k=1} b_kN_k(t):b_k\in Z^0,N_k(t)\in \textbf{\textrm{N}}(t),H(t)\in \textbf{\textrm{F}}(t)\big \}
\end{split}}
\end{equation} based on the mixed operation consisted of the vertex-coinciding operation and the edge-coinciding operation.

\subsubsection{Three algorithms for finding spanning trees}

In \cite{Yao-Zhang-Sun-Mu-Sun-Wang-Wang-Ma-Su-Yang-Yang-Zhang-2018arXiv}, the authors presented the following algorithm for finding particular spanning trees of graphs and networked modules:

\vskip 0.4cm

\noindent \textbf{LARGEDEGREE-NEIGHBOR-FIRST Algorithm}

Let $N(X)$ and $N(u)$ be the sets of neighbors of a vertex $u$ and a set $X$. A vertex set $S$ of a $(p,q)$-graph $G$ is called a \emph{dominating set} of $G$ if each vertex $x\in V(G)\setminus S$ is adjacent with some vertex $y\in S$, and moreover the dominating set $S$ is \emph{connected} if the induced graph over $S$ is a connected subgraph of $G$.

\vskip 0.4cm

\textbf{Input.} A connected and simple graphs $G=(V, E)$.

\textbf{Output.} A spanning tree and a connected dominating set of $G$.

\textbf{Step 1.} Let $S_1\leftarrow N(u_1)\cup \{u_1\}$, degree $\textrm{deg}_G(u_1)\leftarrow \Delta(G)$, and $T_1$ is a tree with vertex set $S_1$, $k\leftarrow 1$.

\textbf{Step 2.} If $Y_k=V\setminus (S_k\cup N(S_k))\neq \emptyset$, goto Step 3, otherwise Step 4.

\textbf{Step 3.} Select a vertex $u_{k+1}\in L(T_k)$ holding its degree $\textrm{deg}_G(u_{k+1})\geq \textrm{deg}_G(x)$ ($x\in L(T_k)$), and let $S_{k+1}\leftarrow S_k\cup N^*(u_{k+1})$, where

$N^*(u_{k+1})=N(u_{k+1})\setminus (N(u_{k+1})\cap S_k)$,

$T_{k+1}\leftarrow T_k+\{u_{k+1}u':u '\in N^*(u_{k+1})\}$,

$k\leftarrow k+1$, goto Step 2.

\textbf{Step 4.} For $Y_k=\emptyset$, $y\in V\setminus V(T_{k+1})$, \textbf{do}: $y$ is adjacent with $v\in V(T_{k+1})$ when $\textrm{deg}_{T_{k+1}}(v)\geq \textrm{deg}_{T_{k+1}}(x)$, $xy\in E$. The resulting tree is denoted as $T^*$.

\textbf{Step 5.} Return the connected dominating set $S_k=V(T_k)$ and the spanning tree $T^*$.

\vskip 0.4cm

\noindent \textbf{PREDEFINED-NODES Algorithm}

\textbf{Input.} A connected graph $G=(V, E)$, and indicate a subset $S=\{u_1, u_2, \dots , u_m\}$ of $G$.

\textbf{Output.} A connected dominating set $X$ of $G$ such that $S\subseteq X$.

Step 1. Add new vertices $\{v_1, v_2, \dots , v_m\}$ to $G$, and join
$v_i$ with a vertex $u_i$ of $S$ by an edge, $i\in [1, m]$. The resulting graph is denoted as $G^*$, such that $V(G^*)=V(G)\cup \{v_i:i\in [1, m]\}$ and $E(G^*)=E(G)\cup \{u_iv_i:i\in [1, m]\}$.

Step 2. Find a connected dominating set $X'$ of $G^*$.

Step 3. Return the connected dominating set $X=X'$ of $G$.

\vskip 0.4cm

\noindent \textbf{LARGEDEGREE-PRESERVE Algorithm}

\textbf{Input.} A scale-free network $\mathcal {N}(t_0)=(p(u, k, t_0), G(t_0))$. $G=G(t_0)$ with $n$ vertices, and degrees $\textrm{deg}_G(v_i)\geq \textrm{deg}_G(v_{i+1})$ with $i\in [1,n-1]$, $\Delta(G)=\textrm{deg}_G(v_1)>1$. A constant $k$ satisfies $\delta(G)<k<\Delta(G)$, degree $\textrm{deg}_G(v_l)\geq k$, but degree $\textrm{deg}_G(v_{l+1})<k$.

\textbf{Output.} A spanning tree $T^*$ of $G$, such that the vertices of $T^*$ hold degrees $\textrm{deg}_G(v_i)\geq \textrm{deg}_G(v_{i+1})$ with $i\in [1, l-1]$, and $k> \textrm{deg}_G(v)$, $v\in V(G)\setminus \{v_1, v_2, \dots , v_l\}$.

Step 1. Let $W_1=N[v_1]\leftarrow N(v_1)\cup \{v_1\}$, an induced graph $G[W_1]$ over $W_1$.

Step 2. If $N_{i, i+1}=N[v_{i+1}]\cap W_{i}=\emptyset $, let
$W_{i+1}\leftarrow W_i\cup N[v_{i+1}]$, and an induced graph $G[W_{i+1}]$ over $W_{i+1}$; if $N_{i,
i+1}\neq \emptyset $, take a vertex $x_{i, i+1}\in N(v_{i+1})\cap
W_{i}$ with degrees $\textrm{deg}_G(x_{i, i+1})\geq \textrm{deg}_G(x)$ for $x\in N(v_{i+1})\cap W_{i}$, and an induced graph $G[W_{i+1}]$, where $W_{i+1}\leftarrow W_{i}\cup W_{i, i+1}$
and $W_{i, i+1}\leftarrow (N[v_{i+1}]\setminus N_{i, i+1})\cup \{x_{i,i+1}\}$.

Step 3. If degrees $\textrm{deg}_G(v_i)\geq k$, goto Step 2, and goto Step 4, otherwise.

Step 4. Apply modified BFS-algorithm (Breadth-First Search Algorithm). Let $S\leftarrow W_i$, $R\leftarrow \{v_0\}$ for $v_0\in N(W_i)$, degree $\textrm{deg}(v_0, v_0)\leftarrow 0$.

Step 5. If $R=\emptyset$, denote the found spanning tree as $T^*$, goto Step 7. If $R\neq \emptyset$, goto Step 6.

Step 6. The vertex $v$ is the first vertex of $R$, take $y\in N(v)\setminus(R\cup S)$ holding degrees $\textrm{deg} _G(y)\geq \textrm{deg} _G(x)$ ($x\in N(v)$); put $y$ into $R$, such that is the last of $R$; and take $v$ from $R$, and then put $v$ into $S$.
Let degree $\textrm{deg} (v_0, y)\leftarrow \textrm{deg} (v_0, v)+1$, goto Step 5.

Step 7. Return the spanning tree $T^*$.

\vskip 0.4cm

The spanning tree $T^*$ found by LARGEDEGREE-PRESERVE Algorithm has its own number of leaves to be approximate to the number of leaves of each spanning tree $T^{\max}$ having maximal leaves, and the connected dominating set $S^*=V(T^*)\setminus L(T^*)$ approximates to $D^+$-minimal dominating set \cite{Yao-Chen-Yang-Wang-Zhang-Zhang2012}. As we have known, no polynomial algorithm for finding:

(i) $L^+$-balanced set;

(ii) optimal cut set $H[S^-]$;

(iii) $D^-$-minimal dominating set and $D^+$-minimal dominating set;

(iv) $k$-distance dominating set. Untersch\"{u}tz and Turau \cite{Unterschutz-Turau2012} have shown the \emph{probabilistic self-stabilizing algorithm} (PSS-algorithm) for looking connected dominating set (Ref. \cite{Blum-Ding-Thaeler-Cheng2004}).

As known, PSS-algorithm is suitable large scale of networks, especially good for those networks having larger degree vertices. PSS-algorithm consists of three subprogrammes: Finding maximal independent set (MIS) first, and find weak connected set (WCDS), the last step is for finding connected dominating set (CDS).

\subsection{Algorithms for Maximum Leaf Spanning Trees}

The Maximum Leaf Spanning Tree (MLST) problem, which asks to find, for a given graph, a spanning tree with as many leaves as possible, is one of the classical NP-complete problems in Ref. \cite{Garey-Johnson1979}. Fernau \emph{et al.} \cite{Fernau-Kneis-Kratsch-Langer-Liedloff-Raible-Rossmanith2011}, investigate MLST based on an exponential time viewpoint that is equivalent to the Connected Dominating Set problem (CDSP), and present a branching algorithm whose running time of $O(1.8966^n)$ has been analyzed using the Measure-and-Conquer technique as well as a lower bound of $\Omega(1.4422^n)$ for the worst case running time of their algorithm. In real networks, it is difficult to employ MLS-trees when investigating topological properties of growing networks.

\subsubsection{Growing Sierpinski network model}

\textbf{Growing Sierpinski network model (GS-network).} Let $N(0)$ be the first network pictured in the second graph in Figure \ref{fig:Sierpinsk-000}, and let $V(0)$ be the node set of $N(0)$. We define a labeling $f$ such that $f(\alpha)=0$ for each node $\alpha\in V(0)$. Do a fractal-operation to the inner face of the first network $N(0)$ by adding a new triangle produces the second network $N(1)$, and label $f(\beta)=1$ for every node $\beta\in V(1)\setminus V(0)$. To form the third network $N(2)$ from $N(1)$, we do a fractal-operation to each inner triangle $\Delta uvw$ of $N(1)$ if no $f(u)=f(v)=f(w)$, and label each node $x\in V(2)\setminus V(1)$ as $f(x)=2$. Thereby, every GS-network $N(t)$ can be obtained from the previous GS-network $N(t-1)$ for $t\geq 2$ by doing a fractal-operation to each inner triangle $\Delta xyz$ of $N(t-1)$ if no $f(x)=f(y)=f(z)$, and label each node $w\in V(t)\setminus V(t-1)$ as $f(w)=t$, see Fig.\ref{fig:Sierpinsk-000}.

\begin{figure}[h]
\centering
\includegraphics[width=14cm]{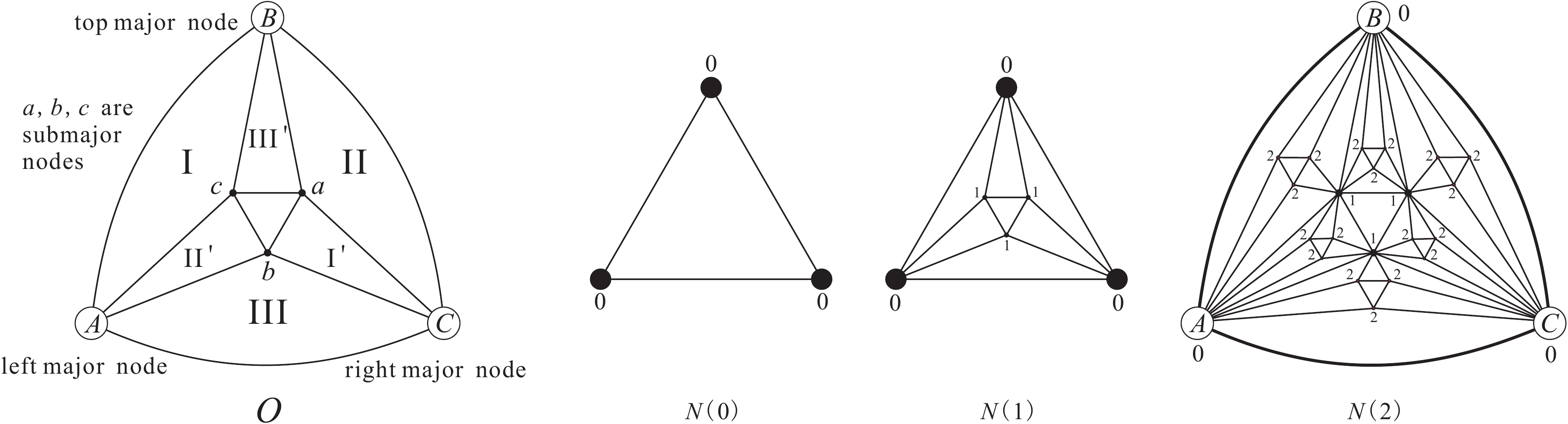}
\caption{\label{fig:Sierpinsk-000} {\small A growing Sierpinski network model $N(t)$}}
\end{figure}

\begin{figure}[h]
\centering
\includegraphics[width=16.4cm]{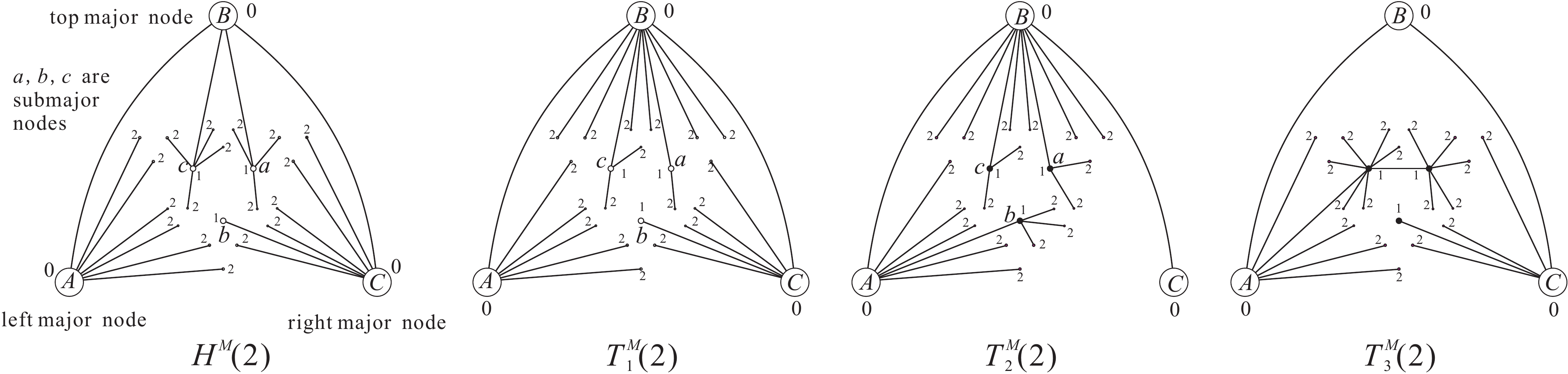}
\caption{\label{fig:Sierpinsk-001} {\small Four spanning trees having the maximum leaves in a growing Sierpinski network model $N(2)$ with diameters $D(H^M(2))=D(T^M_1(2))=4$, $D(T^M_2(2))=5$ and $D(T^M_3(2))=6$.}}
\end{figure}

\begin{figure}[h]
\centering
\includegraphics[width=16.4cm]{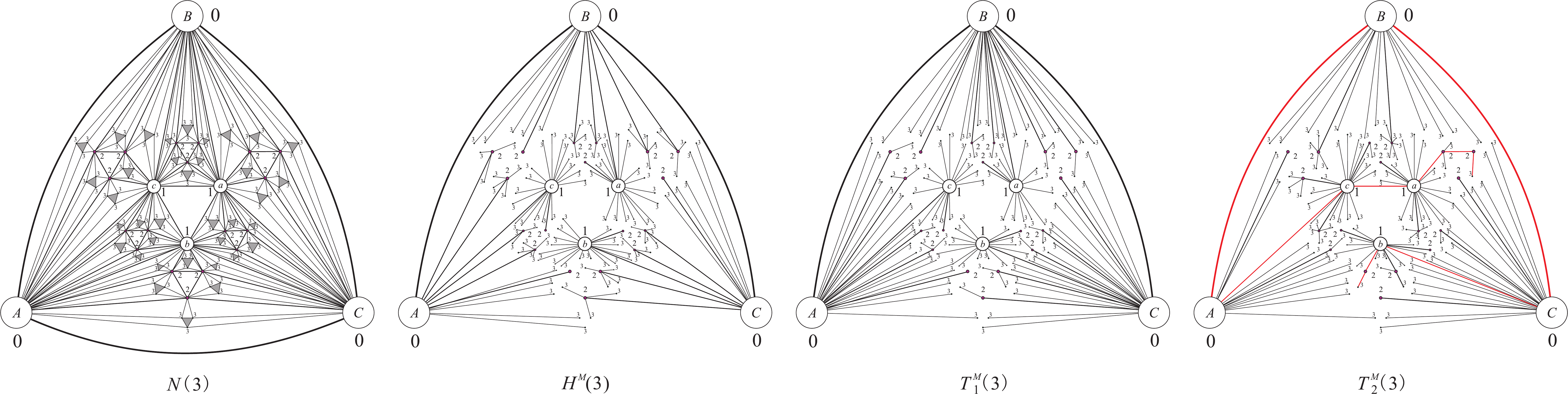}
\caption{\label{fig:Sierpinsk-002} {\small A growing Sierpinski network model $N(3)$ and its ML-spanning trees with diameters $D(H^M(3))=D(T^M_1(3))=6$,$D(T^M_2(3))=10$.}}
\end{figure}

GS-network $N(t)$ obeys a power law degree distribution $p(k)\sim 6(k-1)^\alpha$, where $\alpha=1+\frac{\ln 2}{\ln 3}$.

\vskip 0.4cm

\textbf{MAXILEAF-ST algorithm for GS-network model}

\textbf{Input:} A GS-network $N(t)$ for $t\geq 2$.

\textbf{Output:} An ML-spanning tree $T^M(t)$ having $19\cdot 6^{t-2}$ leaves and diameter $2t$.

\textbf{Step 1.} For $H^M(2)$ with $f(A)=f(B)=f(C)=0$ in $N(2)$, do: $A(2)\leftarrow A$, $B(2)\leftarrow B$ and $C(2)\leftarrow C$; $i\leftarrow 2$.

\textbf{Step 2.} If $i<t$ go to 3; otherwise $T^M(t)\leftarrow H^M(i)$, go to 5.

\textbf{Step 3.} \textbf{Do}:

(1) $I$-$H^M(i)$ is a copy of $H^M(i)$, and its left major node $c\leftarrow A(i)$, its top major node $A\leftarrow B(i)$, and its right major node $B\leftarrow C(i)$;

(2) $III'$-$H^M(i)$ is obtained from a copy of $H^M(i)$ by deleting two edges $A(i)B(i),B(i)C(i)$, and its left major node $c\leftarrow A(i)$, its top major node $B\leftarrow B(i)$, and its right major node $a\leftarrow C(i)$;

(3) $II$-$H^M(i)$ is a copy of $H^M(i)$, and its left major node $a\leftarrow A(i)$, its top major node $B\leftarrow B(i)$, and its right major node $C\leftarrow C(i)$;

(4) $I'$-$H^M(i)$ is obtained from a copy of $H^M(i)$ by deleting two edges $A(i)B(i),B(i)C(i)$, and its left major node $a\leftarrow A(i)$, its top major node $C\leftarrow B(i)$, and its right major node $b\leftarrow C(i)$;

(5) $III$-$H^M(i)$ is obtained from a copy of $H^M(i)$ by deleting one edge $B(i)C(i)$, and its left major node $b\leftarrow A(i)$, its top major node $C\leftarrow B(i)$, and its right major node $A\leftarrow C(i)$;

(6) $II'$-$H^M(i)$ is obtained from a copy of $H^M(i)$ by deleting two edges $A(i)B(i),B(i)C(i)$, and its left major node $b\leftarrow A(i)$, its top major node $A\leftarrow B(i)$, and its right major node $c\leftarrow C(i)$.

\textbf{Step 4.} Identify the major nodes of the above $I$-$H^M(i)$, $III'$-$H^M(i)$, $II$-$H^M(i)$, $I'$-$H^M(i)$, $III$-$H^M(i)$ and $II'$-$H^M(i)$ having the same labels into one node, respectively, produces a ML-spanning tree $H^M(i+1)$ of $N(i+1)$; here, three major nodes of $H^M(i+1)$ are the left major node $A$, the top major node $B$ and the right major node $C$, and three submajor nodes are $a,b,c$; $H^M(i+1)$ has two edges $AB,BC$ and has no the edge $AC$.

Let $f(A)\leftarrow 0$, $f(B)\leftarrow 0$, $f(C)\leftarrow 0$, $f(a)\leftarrow 1$, $f(b)\leftarrow 1$, $f(c)\leftarrow 1$; $f(x)\leftarrow f(x)+1$ for $x\in V(H^M(i+1))\setminus \{a,b,c,A,B,C\}$; $A(i+1)\leftarrow A$, $B(i+1)\leftarrow B$, $C(i+1)\leftarrow C$, $i\leftarrow i+1$; go to 2.

\textbf{Step 5.} Return the desired ML-spanning tree $T^M(t)$.

\vskip 0.4cm

See examples shown in Fig.\ref{fig:Sierpinsk-001}, Fig.\ref{fig:Sierpinsk-002} and Fig.\ref{fig:Sierpinsk-003}.

\begin{figure}[h]
\centering
\includegraphics[width=16.4cm]{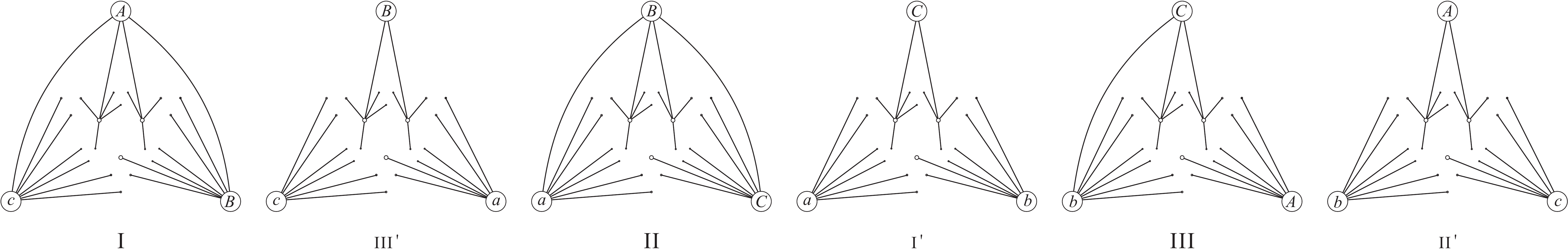}
\caption{\label{fig:Sierpinsk-003} {\small A diagram for illustrating MAXILEAF-ST algorithm.}}
\end{figure}

\subsubsection{Nested growing network models}

In \cite{Yao-Liu-Zhang-Chen-Yao-2014}, the authors have investigated panning trees with maximal leaves in network models. Let $T^{\max}$ (resp. $T^{\min}$) stand for a spanning tree with maximal leaves (resp. minimal leaves) in a scale-free network $\mathcal {N}$. Let $\mathcal {S}(G;T^{\max})$ be the set of spanning trees of a connected graph $G$ having maximal leaves. A spanning tree $T^{\max}_{\text{dia}}$ is one member of $\mathcal {S}(G;T^{\max})$ with the shortest diameter.

For finding spanning trees in bound growing network models, we will apply the Bread-first Search Algorithm (BFSA) in \cite{Bondy-2008} to make our Dynamic First-first BFSA algorithm (DFF-BFSA algorithm) based on the motivation of the linear preferential attachment rule. The \emph{children} of a predecessor are searched before the children of the \emph{successors} of the \emph{predecessor}, namely, ``priority has priority''.

\textbf{NGN-models.} Let $n_v(t)$ and $n_e(t)$ be the numbers of nodes and edges of $M_t$, respectively. If a class of network models $M_t$ hold $n_v(t-1)<n_v(t)$ at each time step $t$, we call them \emph{growing network models} (GN-models). A class of \emph{nested growing network models} (NGN-models) $M_t$ is a particular class of growing networks such that $M_{t-1}\subset M_t$ at each time step $t$.

\vskip 0.4cm

\textbf{DFF-BFSA algorithm for NGN-model}

\textbf{Input:} A NGN-model $M_{t}$ for $t\geq 1$.

\textbf{Output:} A spanning tree $T(t)$ of $M_{t}$.

\textbf{1.} For the NGN-model $M_1$, BFSA outputs a spanning tree $T(1)$ with $V(1)=\bigcup ^{m(1)}_{j=0}V^{(0)}_j$ and a level function $l$ such that $l(x)=j$ for $x\in V^{(1)}_j$. The nodes of $V^{(0)}_j$ are ordered by BFSA. Let $nei(x,k)$ be the neighborhood of a node $x$ of $M_{k}$ at time step $k$.

\textbf{2.} Let $\textrm{nei}(x,k)$ be the neighborhood of a node $x$ of $N(k)$ at time step $k$. At time step $k+1$, $V(k)=\bigcup^{k}_{l=0} \bigcup^{m(l)}_{j=m(l-1)+1} V^{(l)}_j$ (here, $m(-1)=-1$). Implementing BFSA, \textbf{do}:

\textbf{2.1.} For every ordered set $V^{(l)}_j=\big \{x^{(l)}_{j,1},x^{(l)}_{j,2},\dots ,x^{(l)}_{j,m(l,j)}\big \}$ with $l\leq k$, from $i=1$ to $i=m(l,j)$, scan $y\in \textrm{nei}(x^{(l)}_{j,i},k+1)\setminus \big \{w\in V(t): l(w)\textrm{ exists}\big \}$, $b\leftarrow l\big (x^{(l)}_{j,i}\big )+1$, $l(y)\leftarrow m(k)+b$.

\textbf{2.2.} \textbf{Add} $y$ to the ordered set $V^{(k+1)}_{m(k)+b}$ as the last node, and add the node $y$ and the edge $yx^{(l)}_{j,i}$ to the spanning tree $T(k)$ in order to form new spanning tree $T(k+1)$.

\textbf{2.3.} \textbf{Do} $m(k+1)\leftarrow \max \{m(k)+l(x)+1:x\in V(k)\}$;

$V(k+1)\setminus V(k)\leftarrow \bigcup^{m(k+1)}_{j=m(k)+1} V^{(k+1)}_j$;

$V(k+1)\leftarrow \bigcup^{k+1}_{l=0} \bigcup^{m(l)}_{j=m(l-1)+1} V^{(l)}_j$;

go to 3.

\textbf{3.} If $k+1=t$, $T(t)\leftarrow T(k+1)$, go to 4; otherwise go to 2.

\textbf{4.} Return the spanning tree $T(t)$ with a level function $l$.

\subsubsection{Edge-growing network models}
The recursive construction applied here is introduced by Comellas \emph{et al.} in \cite{Comellas-Fertin-Raspaud-2004}. \textbf{A-operation}: Add a new vertex $w$ to each edge $uv\in E(t-1)$, and join $w$ with two ends $u$ and $v$ of the edge $uv$ in order to product two new edges $wu$ and $wv$, so the resultant model is just $N(t)$, called \emph{complete edge-growing network model} (CEGN-model). We design two algorithms for finding Maximum Leaf Spanning Trees (MLS-strees) of CEGN-models.

\vskip 0.4cm

1. \textbf{LFS-Algorithms} (Level-first searching algorithm)

\textbf{Input:} A CEGN-model $N(t)$ with $t\geq 3$.

\textbf{Output:} All MLS-trees of $N(t)$.

\textbf{Step 1.} $F^M(k)\leftarrow \{\textrm{all MLS-trees of }N(k)\}$, $F(k)\leftarrow \{\textrm{all spanning trees of }N(k)\}$ for $k =1,2$.

\textbf{Step 2.} $k\leftarrow 1$ if $t$ is odd, otherwise $k\leftarrow 2$; $F^M(k+2)\leftarrow \emptyset$.

For each spanning tree $T_0\in F(k)$, \textbf{do}: $V(T_0)=\{u^{(k)}_i:i =1, 2,\dots, n_{v,k}\}$.

From $i=1$ to $i=n_{v,k}$, \textbf{do}:
$${
\begin{split}
E\,'(T_{i-1})\leftarrow \big \{&u^{(k)}_iu^{(k+1)}_s\in E(k+1), u^{(k)}_iu^{(k+2)}_l\in E(k+2): u^{(k+1)}_s\in V(k+1)\setminus V(T_{i-1}),\\
&u^{(k+2)}_l\in V(k+2)\setminus V(T_{i-1})\big \}
\end{split}}$$

$E(T_{i})\leftarrow E(T_{i-1})\cup E\,'(T_{i-1})$; $V(T_{i})\leftarrow V(T_{i-1})\cup \big \{x: xu^{(k)}_i\in E\,'(T_{i-1})\big \}$.

$T^M(k+2)\leftarrow T_{n_{v,k}}$;

$F^M_i(k+2)\leftarrow F^M_{i-1}(k+2)\cup T^M(k+2)$;

$F^M(k+2)\leftarrow F^M(k+2)\cup F^M_{n_{v,k}}(k+2)$.

$F(k+2)\leftarrow \{\textrm{all spanning trees of }N(k)\}$.

\textbf{Step 3.} If $t=k+2$, go to Step 6; otherwise, go to Step 4.

\textbf{Step 4.} If $s<t-2$, \textbf{do}: $F^M(s+1)\leftarrow \emptyset$; for every spanning tree $T_0\in F(s)$,

$F^M_0(s+1)\leftarrow \emptyset$, $V(T_0)=\big \{u^{(s)}_i: i=1, 2,\dots, n_{v,k}\big \}$.

From $i=1$ to $i =n_{v,k}$, \textbf{do}:

$E\,'(T_{i-1})\leftarrow \big \{u^{(s)}_iu^{(s+1)}_r\in E(s+1): u^{(s+1)}_r\in V(s+1)\setminus V(T_{i-1})\big \}$;

$E(T_{i})\leftarrow E(T_{i-1})\cup E\,'(T_{i-1})$;

$V(T_{i})\leftarrow V(T_{i-1})\cup \{yu^{(s)}_r: yu^{(s)}_r\in E\,'(T_{i-1})\}$;

$T^M(s+1)\leftarrow T_{n_{v,k}}$;

$F^M_i(k+2)\leftarrow F^M_{i-1}(k+2)\cup \{T^M(s+1)\}$;

$F^M(s+1)\leftarrow F^M(s+1)\cup F^M_{n_{v,k}}(k+1)$;

$F(k+1)\leftarrow \{\textrm{all spanning trees of } N(s+1)\}$.

\textbf{Step 5.} If $s\neq t-2$, go to Step 4.

If $s=t-2$, \textbf{do}: $F^M(t)\leftarrow \emptyset$; for every spanning tree $T_0\in F(t-2)$;

$F^M_{0}(t)\leftarrow \emptyset$, $V(t-2)=\big \{u^{(t-2)}_i: i=1,2,\dots, n_{v,t-2}\big \}$.

From $i=1$ to $i=n_{v,t-2}$, \textbf{do}:
$${
\begin{split}
E\,'(T_{i-1})\leftarrow \big \{&u^{(t-2)}_iu^{(t-1)}_r\in E(t-1),u^{(s)}_iu^{(t)}_l\in E(t): u^{(t-1)}_r\in V(t-1)\setminus V(T_{i-1}),\\
&u^{(t)}_l\in V(t)\setminus V(T_{i-1})\big \}
\end{split}}$$

$E(T_{i})\leftarrow E(T_{i-1})\cup E\,'(T_{i-1})$;

$V(T_{i})\leftarrow V(T_{i-1})\cup \big \{y : yu^{(t-2)}_i\in E\,'(T_{i-1})\big \}$;

$T^M(t)\leftarrow T_{n_{v,k}}$;

$F^M_i(t)\leftarrow F^M_{i-1}(t)\cup \big \{T^M(t)\big \}$;

$F^M(t)\leftarrow F^M(t)\cup F^M_{n_{v,k}}(t)$.

\textbf{Step 6.} Return $F^M(t)$.

\vskip 0.4cm

2. \textbf{BDFLS-algorithm} (Big-degree-first level searching algorithm)

\textbf{Input:} A CEGN-model $N(t)$ with $t\geq 3$.

\textbf{Output:} All MLS-trees of $N(t)$.

\textbf{Step 1.} $F^M(k)\leftarrow \{\textrm{all MLS-trees of }N(k)\}$, $F(k)\leftarrow \{\textrm{all spanning trees of }N(k)\}$ for $k =1, 2$.

\textbf{Step 2.} $k\leftarrow 1$ if $t$ is odd, otherwise $k\leftarrow 2$; $F^M(k+2)\leftarrow \emptyset$. For each spanning tree $T_0\in F(k)$, \textbf{do}: make the degrees of vertices of $V(T_0)=\{u^{(k)}_i: i=1,2,\dots, n_{v,k}\}$ to be a decreasing sequence $d\big (u^{(k)}_{i+1}\big )\leq d\big (u^{(k)}_i\big )$ with $i\in [1,n_{v,k-1}]$.

From $i=1$ to $i=n_{v,k}$, \textbf{do}:
$${
\begin{split}
E\,'(T_{i-1})\leftarrow \big \{&u^{(k)}_iu^{(k+1)}_s\in E(k+1), u^{(k)}_iu^{(k+2)}_l\in E(k+2): u^{(k+1)}_s\in V(k+1)\setminus V(T_{i-1}),\\
& u^{(k+2)}_l\in V(k+2)\setminus V(T_{i-1})\big \}
\end{split}}$$

$E(T_{i})\leftarrow E(T_{i-1})\cup E\,'(T_{i-1})$;

$V(T_{i})\leftarrow V(T_{i-1})\cup \{x: xu^{(k)}_i\in E\,'(T_{i-1})\}$;

$T^M(k+2)\leftarrow T_{n_{v,k}}$;

$F^M_i(k+2)\leftarrow F^M_{i-1}(k+2)\cup \{T^M(k+2)\}$;

$F^M(k+2)\leftarrow F^M(k+2)\cup F^M_{n_{v,k}}(k+2)$;

$F(k+2)\leftarrow \{\textrm{all spanning trees of } N(k)\}$.

\textbf{Step 3.} If $t =k+2$, go to Step 6; otherwise, go to Step 4.

\textbf{Step 4.} If $s<t-2$, \textbf{do}: $F^M(s+1)\leftarrow \emptyset$; for every spanning tree $T_0\in F(s)$, make the degrees of vertices of $V(T_0)=\big \{u^{(s)}_i: i\in [1, n_{v,k}]\big \}$ to be a decreasing sequence $d\big (u^{(k)}_{i+1}\big )\leq d\big (u^{(k)}_i\big )$ with $i\in [1,n_{v,k-1}]$ with $i\in [1, n_{v,s-1}]$.

From $i=1$ to $i=n_{v,s-2}$, \textbf{do}:

$E\,'(T_{i-1})\leftarrow \big \{u^{(s)}_iu^{(s+1)}_r\in E(s+1): u^{(s+1)}_r\in V(s+1)\setminus V(T_{i-1})\big \}$;

$E(T_{i})\leftarrow E(T_{i-1})\cup E\,'(T_{i-1})$;

$V(T_{i})\leftarrow V(T_{i-1})\cup \big \{y : yu^{(s)}_i\in E\,'(T_{i-1})\big \}$;

$T^M(s+1)\leftarrow T_{n_{v,k}}$;

$F^M_i(s+2)\leftarrow F^M_{i-1}(s+2)\cup \{T^M(s+1)\}$;

$F^M(s+1)\leftarrow F^M(s+1)\cup F^M_{n_{v,k}}(s+1)$;

$F(k+1)\leftarrow \{\textrm{all spanning trees of } N(s+1)\}$.

\textbf{Step 5.} If $s\neq t-2$, go to Step 4.

If $s=t-2$, $F^M(t)\leftarrow \emptyset$, \textbf{do}: make the degrees of vertices of $V(T_0)=\big \{u^{(s)}_i: i=1, 2,\dots, n_{v,t-2}\big \}$ of $T_0\in F(t-2)$ to be a decreasing sequence $d\big (u^{(k)}_{i+1}\big )\leq d\big (u^{(k)}_i\big )$ with $i\in [1, n_{v,t-2}-1]$, $F^M_0(t)\leftarrow \emptyset$.

From $i=1$ to $i=n_{v,t-2}$, \textbf{do}:

$${
\begin{split}
E\,'(T_{i-1})\leftarrow \big \{&u^{(t-2)}_iu^{(t-1)}_r\in E(t-1),u^{(s)}_iu^{(t)}_l\in E(t): u^{(t-1)}_r\in V(t-1)\setminus V(T_{i-1}),\\
&u^{(t)}_l\in V(t)\setminus V(T_{i-1})\big \}
\end{split}}$$

$E(T_{i})\leftarrow E(T_{i-1})\cup E\,'(T_{i-1})$;

$V(T_{i})\leftarrow V(T_{i-1})\cup \big \{y : yu^{(t-2)}_i\in E\,'(T_{i-1})\big \}$;

$T^M(t)\leftarrow T_{n_{v,t-2}}$;

$F^M_i(t)\leftarrow F^M_{i-1}(t)\cup \big \{T^M(t)\big \}$;

$F^M(t)\leftarrow F^M(t)\cup F^M_{n_{v,t-2}}(t)$.

\textbf{Step 6.} Return $F^M(t)$.

\section{Realization of topological authentication}

As known, there are many polynomial algorithms for some problems of complete graphs, trees and planar graphs that can be scanned easily into computer, and obey many mathematical constraints from graph colorings and labelings, as well as are suitable for making public-keys and private-keys. We will introduce various technique for generating number-based strings, topological public-keys and private-keys in this subsection.

\subsection{Number-based strings generated from colored graphs}

\subsubsection{Basic operations on number-based strings}

A \emph{number-based string} $s(m)=c_1c_2\cdots c_m$ with $c_i\in [0,9]$ has its own \emph{reciprocal number-based string} defined by $s^{-1}(m)=c_mc_{m-1}\cdots c_2c_1$, also, we say both $s(m)$ and $s^{-1}(m)$ match from each other. We can consider that $s(m)$ is a \emph{public-key}, and $s^{-1}(m)$ is a \emph{private-key} in simpler topological authentications. Let $S(n)=\{s(n)=c_1c_2\cdots c_n:c_i\in [0,9]\}$ in the following discussion.

\begin{defn} \label{defn:operations-nb-strings}
$^*$ For two number-based strings $A_n=a_1a_2\cdots a_n$ and $B_n=b_1b_2\cdots b_n$ with $a_i,b_j\in [0,9]$, let $a\,'_1a\,'_2\cdots a\,'_n$ be a permutation of $a_1a_2\cdots a_n$, and $b\,'_1b\,'_2\cdots b\,'_n$ be a permutation of $b_1b_2\cdots b_n$, and a string $Z_n=z_1z_2\cdots z_n$ is the result of an operation $O\langle A_n, B_n\rangle$. We define several operations $O\langle A_n, B_n\rangle =f_1(a\,'_1,b\,'_1)f_2(a\,'_2,b\,'_2)\cdots f_n(a\,'_n,b\,'_n)$ with $O\in \{[+],[-],[\times],[\ominus],[\cup]\}$ as follows:
\begin{equation}\label{eqa:operations-number-based-trings}
{
\begin{split}
[+]\langle A_n, B_n\rangle &=(a\,'_1+b\,'_1)(a\,'_2+b\,'_2)\cdots (a\,'_n+b\,'_n)=s_1s_2\cdots s_n,~s_i=a\,'_i+b\,'_i,~i\in [1,n];\\
[-]\langle A_n, B_n\rangle &=|a\,'_1-b\,'_1|\cdot |a\,'_2-b\,'_2|\cdot \cdots \cdot |a\,'_n-b\,'_n|=t_1t_2\cdots t_n,~t_i=|a\,'_i-b\,'_i|,~i\in [1,n];\\
[\times]\langle A_n, B_n\rangle &=(a\,'_1\cdot b\,'_1)(a\,'_2\cdot b\,'_2)\cdots (a\,'_n\cdot b\,'_n)=r_1r_2\cdots r_n,~r_i=a\,'_i\cdot b\,'_i,~i\in [1,n];\\
[\ominus]\langle A_n, B_n\rangle &=a\,'_1b\,'_1a\,'_2b\,'_2\cdots a\,'_n b\,'_n=w_1w_2\cdots w_n,~w_i=a\,'_i b\,'_i,~i\in [1,n];\\
[\cup]\langle A_n, B_n\rangle &=z_1z_2\cdots z_n,~z_j\in A^*_n\cup B^*_n,~Z^*_n\cap A^*_n\neq \emptyset,~Z^*_n\cap B^*_n\neq \emptyset
\end{split}}
\end{equation} where $Z^*_n=\{z_1,z_2,\dots ,z_n\}$, $A^*_n=\{a_1,a_2,\dots ,a_n\}$ and $B^*_n=\{b_1,b_2,\dots ,b_n\}$.\qqed
\end{defn}

\begin{rem}\label{rem:333333}
About Definition \ref{defn:operations-nb-strings}, we have:
\begin{asparaenum}[(i)]
\item Notice that some $s_i,t_j$ and $r_k$ are not in $[0,9]$. If a string $Z_n=z_1z_2\cdots z_n$ is defined as $r_j=c_{j,1}c_{j,2}\cdots c_{j,m_j}$ with $c_{j,k}\in [0,9]$ for $k\in [1,m_j]$ and $j\in [1,n]$, we call $Z_n$ a \emph{number-based hyper-string}.
\item The complexity of each operation $O\langle A_n, B_n\rangle$ defined in Eq.(\ref{eqa:operations-number-based-trings}) is $O(2^n)$, since there are $(n!)^2$ results from one operation $O\langle A_n, B_n\rangle$.
\item The number-based string $[\ominus]\langle A_n, B_n\rangle$ has the length $2n$.
\item One can define other operations $O\langle A_n, B_n\rangle$, for instance,
\begin{equation}\label{eqa:other-operations-number-based-trings}
{
\begin{split}
[+^2]\langle A_n, B_n\rangle &=\big [(a\,'_1)^2+(b\,'_1)^2\big ]\big [(a\,'_2)^2+(b\,'_2)^2\big ]\cdots \big [(a\,'_n)^2+(b\,'_n)^2\big ];\\
[-^2]\langle A_n, B_n\rangle &=(a\,'_1-b\,'_1)^2 (a\,'_2-b\,'_2)^2 \cdots (a\,'_n-b\,'_n)^2;\\
[\times^2]\langle A_n, B_n\rangle &=(a\,'_1\cdot b\,'_1)^2(a\,'_2\cdot b\,'_2)^2\cdots (a\,'_n\cdot b\,'_n)^2
\end{split}}
\end{equation} and more complex operations $O\langle A_n, B_n\rangle$ as users' like.\paralled
\end{asparaenum}
\end{rem}

\begin{defn} \label{defn:number-based-string-coloring}
$^*$ Coloring the vertices of a graph $G$ with number-based strings is defined as: $f(x)=s_x(n)$ for $x\in V(G)$ and $s_x(n)\in S(n)$, and color each edge $uv$ by $f(uv)=O\langle f(u),f(v) \rangle$ to be one of the basic operations $[+]\langle A_n, B_n\rangle$, $[-]\langle A_n, B_n\rangle$, $[\times]\langle A_n, B_n\rangle$, $[\ominus]\langle A_n, B_n\rangle$ and $[\cup]\langle A_n, B_n\rangle$ define in Eq.(\ref{eqa:operations-number-based-trings}). We call $f$ a \emph{number-based string coloring} of $G$.\qqed
\end{defn}

Suppose that a number-based string set $S(n)$ forms an every-zero group $\{S(n);\oplus\}$ under the addition operation ``$\oplus ~(\bmod~M)$'' on the number-based strings of $S(n)$, in Definition \ref{defn:number-based-string-coloring}, then each edge $uv\in E(G)$ is colored with $f(uv)$ defined as follows
\begin{equation}\label{eqa:basic-every-zero-group}
{
\begin{split}
&f(uv)=O\langle f(u),f(v) \rangle=O\langle A_i,B_j \rangle=s_{\lambda}(n)=d_{\lambda,1}d_{\lambda,2}\cdots d_{\lambda,n}\in S(n)\\
&a_{i,t}+b_{j,t}-c_{k,t}=d_{\lambda,t},~ \lambda=i+j-k~(\bmod~M),~t\in [1,n]
\end{split}}
\end{equation} where $f(u)=A_i=a_{i,1}a_{i,2}\cdots a_{i,n}\in S(n)$, $f(v)=B_j=b_{j,1}b_{j,2}\cdots b_{j,n}\in S(n)$, $C_k=c_{k,1}c_{k,2}\cdots c_{k,n}\in S(n)$ is a preselected arbitrarily \emph{zero} in the addition operation (\ref{eqa:basic-every-zero-group}).

\subsubsection{Algorithms for Topcode-type of number-based strings}

In Fig.\ref{fig:TB-strings-from-basic-4-curves}, we give four groups of basic line-ways Vo-$k$ and its reciprocal Vo-$k$-r and inverse Vo-$k$-i for $k\in [1,4]$ in the following polynomial NBSTRING algorithms for generating number-based strings from Topcode-matrices.

\begin{figure}[h]
\centering
\includegraphics[width=16.4cm]{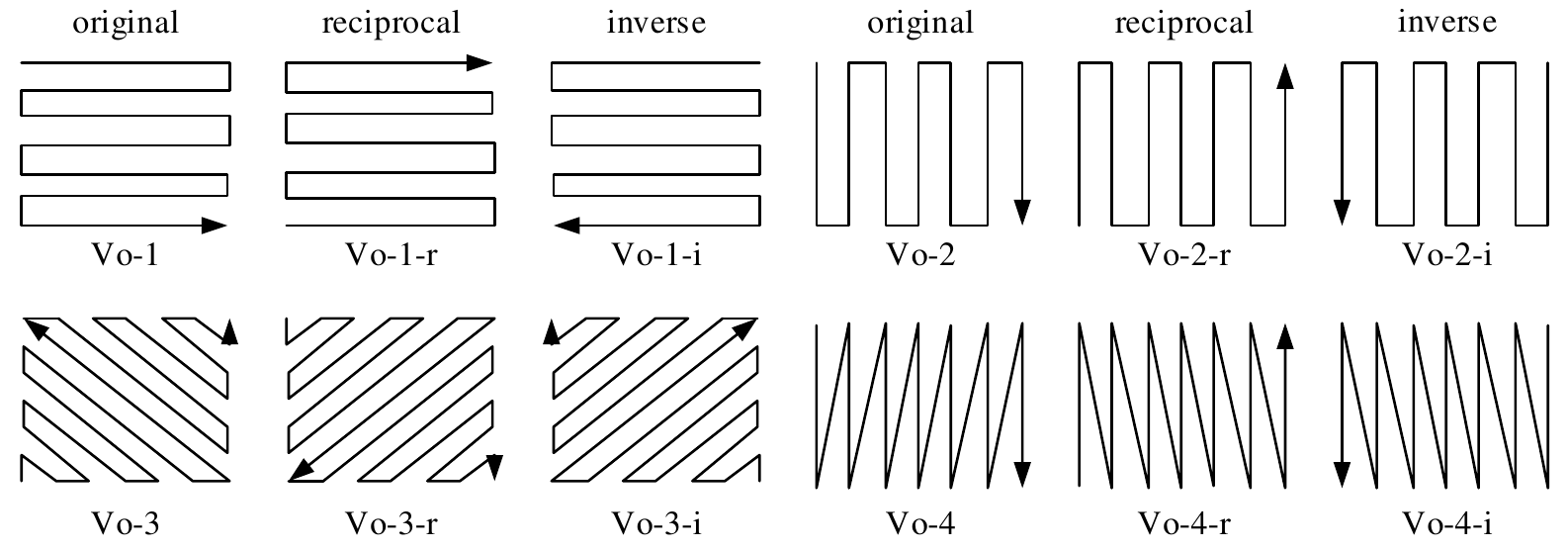}\\
\caption{\label{fig:TB-strings-from-basic-4-curves}{\small Four basic Vo-$k$ with $k\in [1,4]$ used in general matrices, cited from \cite{Yao-Mu-Sun-Sun-Zhang-Wang-Su-Zhang-Yang-Zhao-Wang-Ma-Yao-Yang-Xie2019}.}}
\end{figure}

\vskip 0.4cm

\noindent \textbf{NBSTRING algorithm-I (Vo-1) \cite{Yao-Mu-Sun-Sun-Zhang-Wang-Su-Zhang-Yang-Zhao-Wang-Ma-Yao-Yang-Xie2019}}

\textbf{Input:} A $(p,q)$-graph $G$ and its Topcode-matrix $T_{code}(G)$ defined in Definition \ref{defn:topcode-matrix-definition}.

\textbf{Output:} Number-based strings (NB-strings): Vo-1 NB-string $S_{\textrm{Vo-1}}$, Vo-1-$r$ NB-string $S_{\textrm{Vo-1-}r}$ and Vo-1-$i$ NB-string $S_{\textrm{Vo-1-}i}$ as follows:
\begin{equation}\label{eqa:555555}
{
\begin{split}
&S_{\textrm{Vo-1}}=x_1x_2\cdots x_qe_qe_{q-1}\cdots e_2e_1y_1 y_2\cdots y_q\\
&S_{\textrm{Vo-1-}r}=y_1 y_2\cdots y_qe_qe_{q-1}\cdots e_2e_1x_1x_2\cdots x_q\\
&S_{\textrm{Vo-1-}i}=x_qx_{q-1}\cdots x_2x_1 e_1 e_2\cdots e_qy_qy_{q-1}\cdots y_2y_1
\end{split}}
\end{equation}

\vskip 0.4cm

\noindent \textbf{NBSTRING algorithm-II (Vo-2)}

\textbf{Input:} A $(p,q)$-graph $G$ and its Topcode-matrix $T_{code}(G)$ defined in Definition \ref{defn:topcode-matrix-definition}.

\textbf{Output:} Number-based strings (NB-strings): Vo-2 NB-string $S_{\textrm{Vo-2}}$, Vo-2-$r$ NB-string $S_{\textrm{Vo-2-}r}$ and Vo-2-$i$ NB-string $S_{\textrm{Vo-2-}i}$ as follows:
\begin{equation}\label{eqa:555555}
{
\begin{split}
S_{\textrm{Vo-2}}&=x_1e_1y_1y_2e_2x_2x_3e_3y_3y_4\dots x_{q-1}x_qe_qy_q &(\textrm{odd }q)\\
S_{\textrm{Vo-2}}&=x_1e_1y_1y_2e_2x_2x_3e_3y_3y_4\dots y_{q-1}y_qe_qx_q &(\textrm{even }q)\\
S_{\textrm{Vo-2-}r}&=y_1e_1x_1x_2e_2y_2y_3\dots y_{q-1}y_qe_qx_q &(\textrm{odd }q)\\
S_{\textrm{Vo-2-}r}&=y_1e_1x_1x_2e_2y_2y_3\dots x_{q-1}x_qe_qy_q &(\textrm{even }q)\\
S_{\textrm{Vo-2-}i}&=x_qe_qy_qy_{q-1}e_{q-1}x_{q-1}x_{q-2}\dots x_1e_1y_1 &(\textrm{odd }q)\\
S_{\textrm{Vo-2-}i}&=x_qe_qy_qy_{q-1}e_{q-1}x_{q-1}x_{q-2}\dots y_1e_1x_1 &(\textrm{even }q)
\end{split}}
\end{equation}

\vskip 0.4cm

\noindent \textbf{NBSTRING algorithm-III (Vo-3)}

\textbf{Input:} A $(p,q)$-graph $G$ and its Topcode-matrix $T_{code}(G)$ defined in Definition \ref{defn:topcode-matrix-definition}.

\textbf{Output:} Number-based strings (NB-strings): Vo-3 NB-string $S_{\textrm{Vo-3}}$, Vo-3-$r$ NB-string $S_{\textrm{Vo-3-}r}$ and Vo-3-$i$ NB-string $S_{\textrm{Vo-3-}i}$ as follows:
\begin{equation}\label{eqa:555555}
{
\begin{split}
S_{\textrm{Vo-3}}&=y_2y_1e_1x_1e_2y_3y_4e_3x_2x_3e_4y_5y_6\dots y_qe_{q-1}x_{q-2}x_{q-1}x_qe_q\\
S_{\textrm{Vo-3-}r}&=x_2x_1e_1y_1e_2x_3x_4e_3y_2y_3e_4x_5x_6\dots x_qe_{q-1}y_{q-2}y_{q-1}y_qe_q\\
S_{\textrm{Vo-3-}i}&=y_{q-1}y_qe_qx_qe_{q-1}y_{q-2}y_{q-3}e_{q-2}x_{q-1}x_{q-2}\dots y_1e_2x_2x_1e_1
\end{split}}
\end{equation}

\vskip 0.4cm

\noindent \textbf{NBSTRING algorithm-IV (Vo-4)}

\textbf{Input:} A $(p,q)$-graph $G$ and its Topcode-matrix $T_{code}(G)$ defined in Definition \ref{defn:topcode-matrix-definition}.

\textbf{Output:} Number-based strings (NB-strings): Vo-4 NB-string $S_{\textrm{Vo-4}}$, Vo-4-$r$ NB-string $S_{\textrm{Vo-4-}r}$ and Vo-4-$i$ NB-string $S_{\textrm{Vo-4-}i}$ as follows:
\begin{equation}\label{eqa:555555}
{
\begin{split}
S_{\textrm{Vo-4}}&=x_1e_1y_1x_2e_2y_2\dots y_{q-1}x_{q}e_qy_q\\
S_{\textrm{Vo-4-}r}&=y_1e_1x_1y_2e_2x_2\dots x_{q-1}y_{q}e_qx_q \\
S_{\textrm{Vo-4-}i}&=x_q e_q y_{q}x_{q-1}e_{q-1}y_{q-1}x_{q-2}\dots x_2e_2y_2x_1e_1 y_1
\end{split}}
\end{equation}

\vskip 0.4cm

\begin{rem}\label{rem:333333}
The above NBSTRING algorithms can be used to adjacent matrix, Topcode-matrix, adjacent e-value matrix and adjacent ve-value matrix defined in Definition \ref{defn:topcode-matrix-definition}, Definition \ref{defn:e-value-matrix} and Definition \ref{defn:ve-value-matrix}.\paralled
\end{rem}

By Fig.\ref{fig:six-lines-strings}, we have the following number-based strings:
{\small
$${
\begin{split}
T^{(a)}_b=&5758292928606061625487888827283031323334353681828384853030596060272626268686443\\
T^{(b)}_b=&5727303028582930596031292832602733606034262635616236268681548286483878884438588\\
T^{(c)}_b=&3030275728596030582931602732292833262634606035268638616281864825483438487888885\\
T^{(d)}_b=&5727305828302930592931602832606033276034266135266236265818648386878348884488853\\
T^{(e)}_b=&8888874562616060282929585727303028305960313260273334262635362686818286483844385\\
T^{(f)}_b=&3027305928306060313227263334262635368686818244838438588888745626160602829295857\\
T^{(g)}_b=&5758282730305960313029292860333260272626353460616258136268686483824878888858443\\
T^{(h)}_b=&3030272857583059606031292928323327262634606035268681366162582864483487888438588
\end{split}}$$
}
\begin{figure}[h]
\centering
\includegraphics[width=16.4cm]{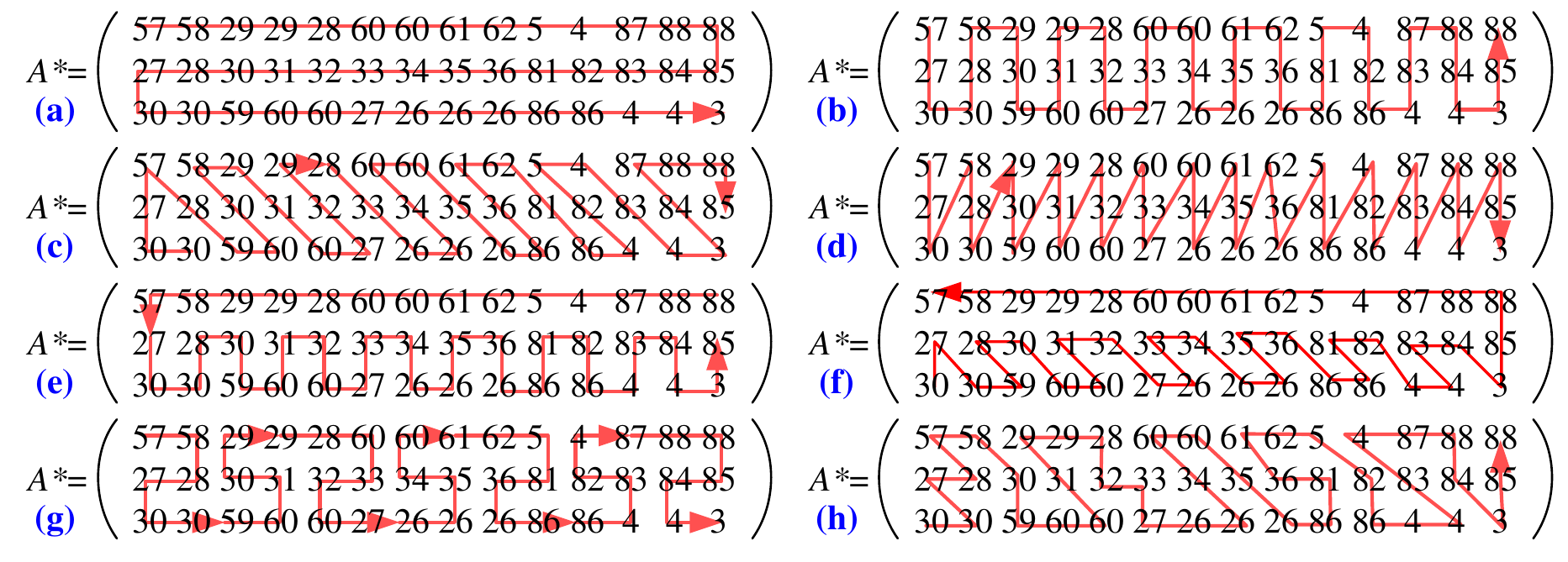}\\
\caption{\label{fig:six-lines-strings}{\small The lines in (a), (b), (c) and (d) are the NBSTRING algorithms for producing number-based strings; others are examples for showing there are many line-ways for generating number-based strings, cited from \cite{Yao-Mu-Sun-Sun-Zhang-Wang-Su-Zhang-Yang-Zhao-Wang-Ma-Yao-Yang-Xie2019}.}}
\end{figure}

\begin{rem}\label{rem:333333}
A Topsnut-matrix $T_{code}(G)$ of a $(p,q)$-graph $G$ may has $N_{fL}(m)$ groups of disjoint fold-lines $L_{j,1},L_{j,2},\dots, L_{j,m}$~($=\{L_{j,i}\}^m_1$) for $j\in [1,N_{fL}(m)]$ and $m\in [1,M]$, where each fold-line $L_{j,i}$ has own initial point $(a_{j,i},b_{j,i})$ and terminal point $(c_{j,i},d_{j,i})$ in $xOy$-plane, and $L_{j,i}$ is internally disjoint, such that each element of $T_{code}(G)$ is on one and only one of the disjoint fold-lines $=\{L_{j,i}\}^m_1$ after we put the elements of $T_{code}(G)$ into $xOy$-plane. Notice that each fold-line $L_{j,i}$ has its initial and terminus points, so $2m\leq 3q$, and then $M=\lfloor 3q/2\rfloor $. Each group of disjoint fold-lines $\{L_{j,i}\}^m_1$ can distributes us $m!$ number-based strings, so we have at least $N_{fL}(m)\cdot m!$ number-based strings with $m\in [1,M]$. Thereby, the graph $G$ gives us the number $N^*(G)$ of number-based strings in total as follows
\begin{equation}\label{eqa:random-routes}
{
\begin{split}
N^*(G)=q!\cdot 2^q\sum^{M}_{m=1}N_{fL}(m)\cdot m!
\end{split}}
\end{equation}

Each tree admits a regular odd-rainbow intersection total set-labeling based on a \emph{regular odd-rainbow set-sequence} $\{R_k\}^{q}_1$ defined as: $R_k=[1,2k-1]$ with $k\in [1,q]$, where $[1,1]=\{1\}$. Moreover, we can define a \emph{regular Fibonacci-rainbow set-sequence} $\{R_k\}^{q}_1$ by $R_1=[1,1]$, $R_2=[1,1]$, and $R_{k+1}=R_{k-1}\cup R_{k}$ with $k\in [2,q]$; or a $\tau$-term Fibonacci-rainbow set-sequence $\{\tau,R_i\}^{q}_1$ holds: $R_i=[1,a_i]$ with $a_i>1$ and $i\in [1,q]$, and $R_k=\sum ^{k-1}_{i=k-\tau}R_i$ with $k>\tau$ \cite{Ma-Wang-Wang-Yao-Theoretical-Computer-Science-2018}.\paralled
\end{rem}

\begin{problem}\label{problem:xxxxxx}
\cite{Yao-Zhang-Sun-Mu-Sun-Wang-Wang-Ma-Su-Yang-Yang-Zhang-2018arXiv} \textbf{Fold-lines covering all points on geometric lattices in $xOy$-plane.} Let $P_3\times P_q$ be a lattice in $xOy$-plane. There are points $(i,j)$ on the lattice $P_3\times P_q$ with $i\in[1,3]$ and $j\in [1,q]$. If a continuous fold-line $L$ with initial point $(a,b)$ and terminal point $(c,d)$ on $P_3\times P_q$ is internally disjoint and contains all points $(i,j)$ of $P_3\times P_q$, we call $L$ a \emph{total number-based string line}. \textbf{Find} all possible total number-based string lines in the lattice $P_3\times P_q$.
\end{problem}

\subsubsection{Tb-string Algorithms for general matrices}

Here ``Tb-string'' is the abbreviation of ``text-based string''. A general matrix is defined as follows:

\begin{equation}\label{eqa:general-matrix}
\centering
{
\begin{split}
A= \left(
\begin{array}{ccccc}
x_{1,1} & x_{1,2} & \cdots & x_{1,n}\\
x_{2,1} & x_{2,2} & \cdots & x_{2,n}\\
\cdots & \cdots & \cdots & \cdots\\
x_{m,1} & x_{m,2} & \cdots & x_{m,n}\\
\end{array}
\right)_{m\times n}=
\left(\begin{array}{c}
X_1\\
X_2\\
\cdots \\
X_m
\end{array} \right)=(X_1,~X_2,~\cdots ,~X_m)^{T},~m,n\geq 2
\end{split}}
\end{equation}

By four group of basic line-ways Vo-$k$ and its reciprocal Vo-$k$-r and inverse Vo-$k$-i for $k\in [1,4]$ shown in Fig.\ref{fig:TB-strings-from-basic-4-curves}, we design the following Tb-string Algorithm-Vo-$k$ for $k=$I, II, III, IV as follows.

\vskip 0.4cm

\noindent \textbf{$^*$ Tb-string Algorithm-Vo-I}

\textbf{Input:} A general matrix $A=(X_1,~X_2,~\cdots ,~X_m)^{T}$ defined in Eq,(\ref{eqa:general-matrix}).

\textbf{Output:} Text-based strings (Tb-strings): Vo-1 Tb-string $D_{\textrm{Vo-1}}$, Vo-1-$r$ Tb-string $D_{\textrm{Vo-1-}r}$ and Vo-1-$i$ Tb-string $D_{\textrm{Vo-1-}i}$ as follows:
{\small
\begin{equation}\label{eqa:general-matrix-Vo-11}
{
\begin{split}
&D_{\textrm{Vo-1}}=x_{1,1}x_{1,2}\cdots x_{1,n} x_{2,n}x_{2,n-1}\cdots x_{2,2}x_{2,1}x_{3,1}\cdots \cdots x_{m,1}x_{m,2}\cdots x_{m,n} &(\textrm{odd }m)\\
&D_{\textrm{Vo-1}}=x_{1,1}x_{1,2}\cdots x_{1,n} x_{2,n}x_{2,n-1}\cdots x_{2,2}x_{2,1}x_{3,1}\cdots \cdots x_{m,n}x_{m,n-1} \cdots x_{m,2}x_{m,1} &(\textrm{even }m)\\
&D_{\textrm{Vo-1-}r}=x_{m,1}x_{m,2}\cdots x_{m,n}x_{m-1,n} x_{m-1,n-1}\cdots x_{m-1,2}x_{m-1,1}\cdots \cdots x_{1,1}x_{1,2}\cdots x_{1,n} &(\textrm{odd }m)\\
&D_{\textrm{Vo-1-}r}=x_{m,1}x_{m,2}\cdots x_{m,n}x_{m-1,n} x_{m-1,n-1}\cdots x_{m-1,2}x_{m-1,1}\cdots \cdots x_{1,n}x_{1,n-1}\cdots x_{1,2}x_{1,1} &(\textrm{even }m)\\
&D_{\textrm{Vo-1-}i}=x_{1,n}x_{1,n-1}\cdots x_{1,2}x_{1,1} x_{2,2}x_{2,1}\cdots x_{2,n}\cdots \cdots x_{m,n}x_{m,n-1} \cdots x_{m,2}x_{m,1} &(\textrm{odd }m)\\
&D_{\textrm{Vo-1-}i}=x_{1,n}x_{1,n-1}\cdots x_{1,2}x_{1,1} x_{2,2}x_{2,1}\cdots x_{2,n}\cdots \cdots x_{m,n}x_{m,n-1} \cdots x_{m,1}x_{m,2}\cdots x_{m,n} &(\textrm{even }m)\\
\end{split}}
\end{equation}
}

\noindent \textbf{$^*$ Tb-string Algorithm-Vo-II}

\textbf{Input:} A general matrix $A=(X_1,~X_2,~\cdots ,~X_m)^{T}$ defined in Eq,(\ref{eqa:general-matrix}).

\textbf{Output:} Text-based strings (Tb-strings): Vo-2 Tb-string $D_{\textrm{Vo-2}}$, Vo-2-$r$ Tb-string $D_{\textrm{Vo-2-}r}$ and Vo-2-$i$ Tb-string $D_{\textrm{Vo-2-}i}$ as follows:
{\small
\begin{equation}\label{eqa:general-matrix-Vo-22}
{
\begin{split}
&D_{\textrm{Vo-2}}=x_{1,1}x_{2,1}\cdots x_{m,1} x_{m,2}x_{m-1,2}\cdots x_{2,2}x_{2,1}\cdots \cdots x_{n,1}x_{n,2}\cdots x_{m,n} &(\textrm{odd }n)\\
&D_{\textrm{Vo-2}}=x_{1,1}x_{2,1}\cdots x_{m,1} x_{m,2}x_{m-1,2}\cdots x_{2,2}x_{2,1}\cdots \cdots x_{m,n}x_{m-1,n}\cdots x_{2,n}x_{1,n} &(\textrm{even }n)\\
&D_{\textrm{Vo-2-}r}=x_{m,1}x_{m,2}\cdots x_{m,n}x_{m-1,n} x_{m-1,n-1}\cdots x_{m-1,2}x_{m-1,1}\cdots x_{1,1}x_{1,2}\cdots x_{1,n} &(\textrm{odd }n)\\
&D_{\textrm{Vo-2-}r}=x_{m,1}x_{m-1,1}\cdots x_{2,1}x_{1,1} x_{2,1}x_{2,2}\cdots x_{m,2}\cdots \cdots x_{m,n}x_{m-1,n} \cdots x_{2,n}x_{1,n} &(\textrm{even }n)\\
&D_{\textrm{Vo-2-}i}=x_{1,n}x_{2,n}\cdots x_{m,n}x_{m,n-1} x_{m-1,n-1}\cdots x_{2,n-1}x_{1,n-1}\cdots \cdots x_{1,1}x_{1,2} \cdots x_{1,m-1}x_{1,m} &(\textrm{odd }n)\\
&D_{\textrm{Vo-2-}i}=x_{1,n}x_{2,n}\cdots x_{m,n}x_{m,n-1} x_{m-1,n-1}\cdots x_{2,n-1}x_{1,n-1}\cdots \cdots x_{1,m}x_{1,m-1} \cdots x_{1,2}x_{1,1} &(\textrm{even }n)\\
\end{split}}
\end{equation}
}

\noindent \textbf{$^*$ Tb-string Algorithm-Vo-III}

\textbf{Input:} A general matrix $A=(X_1,~X_2,~\cdots ,~X_m)^{T}$ defined in Eq,(\ref{eqa:general-matrix}).

\textbf{Output:} Text-based strings (Tb-strings): Vo-3 Tb-string $D_{\textrm{Vo-3}}$, Vo-3-$r$ Tb-string $D_{\textrm{Vo-3-}r}$ and Vo-3-$i$ Tb-string $D_{\textrm{Vo-3-}i}$ as follows:
\begin{equation}\label{eqa:general-matrix-Vo-33}
{
\begin{split}
&D_{\textrm{Vo-3}}=x_{m,1}x_{m-1,1}x_{m,2} x_{m,3}x_{m-1,2} x_{m-3,1}x_{m-4,1}\cdots \cdots x_{3,n}x_{2,n-1}x_{1,n-3}x_{1,n-2} x_{2,n}x_{1,n}\\
&D_{\textrm{Vo-3-}r}=x_{1,1}x_{1,2}x_{1,3}x_{2,2}x_{1,3}\cdots \cdots x_{m-3,n}x_{m-2,n-1} x_{m,n-2}x_{m,n-1}x_{m-1,n}x_{m,n}\\
&D_{\textrm{Vo-3-}i}=x_{m,n}x_{m-1,n}x_{m,n-1}x_{m,n-2}x_{m-1,n-1}x_{m-3,n}\cdots \cdots x_{3,1}x_{2,2} x_{1,3}x_{1,2}x_{2,1}x_{1,1}\\
\end{split}}
\end{equation}

\noindent \textbf{$^*$ Tb-string Algorithm-Vo-IV}

\textbf{Input:} A general matrix $A=(X_1,~X_2,~\cdots ,~X_m)^{T}$ defined in Eq,(\ref{eqa:general-matrix}).

\textbf{Output:} Text-based strings (Tb-strings): Vo-4 Tb-string $D_{\textrm{Vo-4}}$, Vo-4-$r$ Tb-string $D_{\textrm{Vo-4-}r}$ and Vo-4-$i$ Tb-string $D_{\textrm{Vo-4-}i}$ as follows:
\begin{equation}\label{eqa:general-matrix-Vo-43}
{
\begin{split}
D_{\textrm{Vo-4}}=&x_{1,1}x_{2,1}\cdots x_{m,1} x_{1,2}x_{2,2}\cdots x_{m-1,2}x_{m,1}\cdots \cdots x_{n,1}x_{n,2}\cdots x_{m,n}\\
D_{\textrm{Vo-4-}r}=&x_{m,1}x_{m-1,1}x_{m-3,1}\cdots x_{2,1}x_{1,1}x_{m,2}x_{m-1,2}x_{m-3,2}\cdots x_{2,2}x_{1,2}\cdots \cdots \\
&x_{m,n}x_{m-1,n}x_{m-3,n} x_{2,n}x_{1,n}\\
D_{\textrm{Vo-4-}i}=&x_{n,1}x_{n,2}\cdots x_{m,n} x_{n-1,1}x_{n-1,2}\cdots x_{m,n-1}\cdots \cdots x_{1,1}x_{2,1}\cdots x_{m,1}\\
\end{split}}
\end{equation}

\begin{example}\label{exa:8888888888}
By an adjacent ve-value matrix $VE_{col}(H_{4043})$ shown in Eq.(\ref{eqa:4-matrices-Hanzi-graph-H-4043-22}) and the formula $D_{\textrm{Vo-3}}$ and $D_{\textrm{Vo-3-}i}$ of the Tb-string Algorithm-Vo-III shown in Eq.(\ref{eqa:general-matrix-Vo-33}) based on the lines shown in Fig.\ref{fig:TB-strings-from-basic-4-curves}, we have two strings

\begin{equation}\label{eqa:adjacent-ve-value-matrix11}
{
\begin{split}
D_{\textrm{Vo-3}}& =54540320300020100010000000000001000302304545\\
D_{\textrm{Vo-3-}i}&=000000000040004531135502020540000430003101000.
\end{split}}
\end{equation}
By upsetting the order of numbers of each of the above two number-based strings in Eq.(\ref{eqa:adjacent-ve-value-matrix11}) we get two public-keys below
\begin{equation}\label{eqa:adjacent-ve-value-matrix22}
{
\begin{split}
D\,'_{\textrm{Vo-3}}& =00050004000500043232323111000004000500040005\\
D\,'_{\textrm{Vo-3-}i}&=000011000022000033000044000055000055011033044.
\end{split}}
\end{equation} for real application. However, finding $D_{\textrm{Vo-3}}$ from $D\,'_{\textrm{Vo-3}}$ is impossible, and moreover finding the adjacent ve-value matrix $VE_{col}(H_{4043})$ from $D_{\textrm{Vo-3}}$ is not easy by cutting $D_{\textrm{Vo-3}}$ into elements of $VE_{col}(H_{4043})$.
\end{example}

\begin{figure}[h]
\centering
\includegraphics[width=16.4cm]{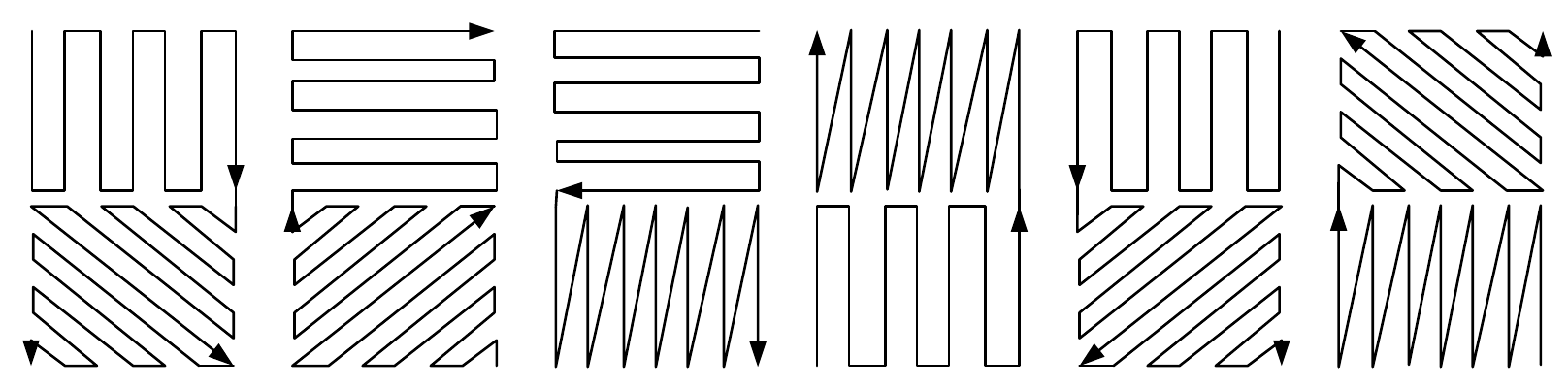}\\
\caption{\label{fig:basic-4-curves-combi}{\small Six combinatorics of four basic Vo-$k$ with $k\in [1,4]$ shown in Fig.\ref{fig:TB-strings-from-basic-4-curves} for generating number-based strings and Tb-strings.}}
\end{figure}

\subsubsection{Generating number-based strings from super-shape colored graphs}

There are many expressions of vertices, edges and letters, see examples shown in Fig.\ref{fig:various-edges-vertices-colored}, and we can color them with \emph{shape-based numbers} in order to add these shape numbers into the number-based strings generated from Topcode-matrices.

\begin{defn} \label{defn:super-shape-colored-graphs}
$^*$ A \emph{super-shape colored graph} $G^{\textrm{ss}}$ has its own \emph{super-shape vertex set}
$$
V(G^{\textrm{ss}})=\{u(\textrm{g-shape, v-color, math}):u\in V(G)\}
$$ both $G$ and $G^{\textrm{ss}}$ are isomorphic from each other in topological structure, that is, $G\cong G^{\textrm{ss}}$; similarly, we have the \emph{super-shape edge set} of $G^{\textrm{ss}}$ as follows
$$
E(G^{\textrm{ss}})=\{xy(\textrm{g-shape, v-color, math}):xy\in E(G)\}
$$ where notation $u$(g-shape, v-color, math) (resp. $xy$(\textrm{g-shape, v-color, math})) is a \emph{super-shape vertex} (resp. \emph{super-shape edge}) consisted of geometric shapes, visual colors $(r,g,b)$ of shapes and letters, and mathematical restriction (see Fig.\ref{fig:various-edges-vertices-colored} and Fig.\ref{fig:Pan-various-Topsnut-gpws}).\qqed
\end{defn}

For example, some visual colors $(r, g, b)$ are: white = (255, 255, 255), gray = (127, 127, 127), black = (0, 0, 0), red = (255, 0, 0), green = (0, 255, 0), cyan blue = (0, 255, 255), magenta = (255, 0, 255), yellow = (255, 255, 0), orange = (255, 127, 0), purple = (127, 0, 255), pink green = (0, 225, 128), lake blue = (0, 128, 255), grass green = (128, 255, 0), rose red = (255, 0, 128), and so on. Adjusting the numbers in $(r, g, b)=$(0-255, 0-255, 0-255) gets different shades of visual colors. Thereby, those number-based strings generated from a super-shape colored graph $G$ are more complicated than the number-based strings generated from Topcode-matrix $T_{code}(G)$.

\begin{defn} \label{defn:super-shape-topo-authentication-multiple-variables}
$^*$ A \emph{super-shape topological authentication} $\textbf{T}_{\textbf{ss-a}}\langle\textbf{A},\textbf{B}\rangle$ \emph{of multiple variables} is defined as follows
\begin{equation}\label{eqa:super-shape-topo-authentication-11}
\textbf{T}_{\textbf{ss-a}}\langle\textbf{A},\textbf{B}\rangle =P_{\textrm{ss-ub}}(\textbf{A}) \rightarrow _{\textbf{F}_{\textrm{ss-a}}} P_{\textrm{ss-ri}}(\textbf{B})
\end{equation} where $P_{\textrm{ss-ub}}(\textbf{A})=(a_1,a_2,\dots ,a_m)$ and $P_{\textrm{ss-ri}}(\textbf{B})=(b_1,b_2,\dots ,b_m)$ both are \emph{variable vectors}, in which both $a_1$ and $b_1$ are two \emph{super-shape graphs} or sets of \emph{super-shape graphs}, and $\textbf{F}_{\textrm{ss-a}}=(f_1$, $f_2,\dots $, $f_m)$ is an \emph{operation vector}, $P_{\textrm{ss-ub}}(\textbf{A})$ is a \emph{topological public-key vector} and $P_{\textrm{ss-ri}}(\textbf{B})$ is a \emph{topological private-key vector} holding $f_k(a_k)\rightarrow b_k$ for $k\in [1,m]$ with $m\geq 1$.\qqed
\end{defn}

\begin{figure}[h]
\centering
\includegraphics[width=16.4cm]{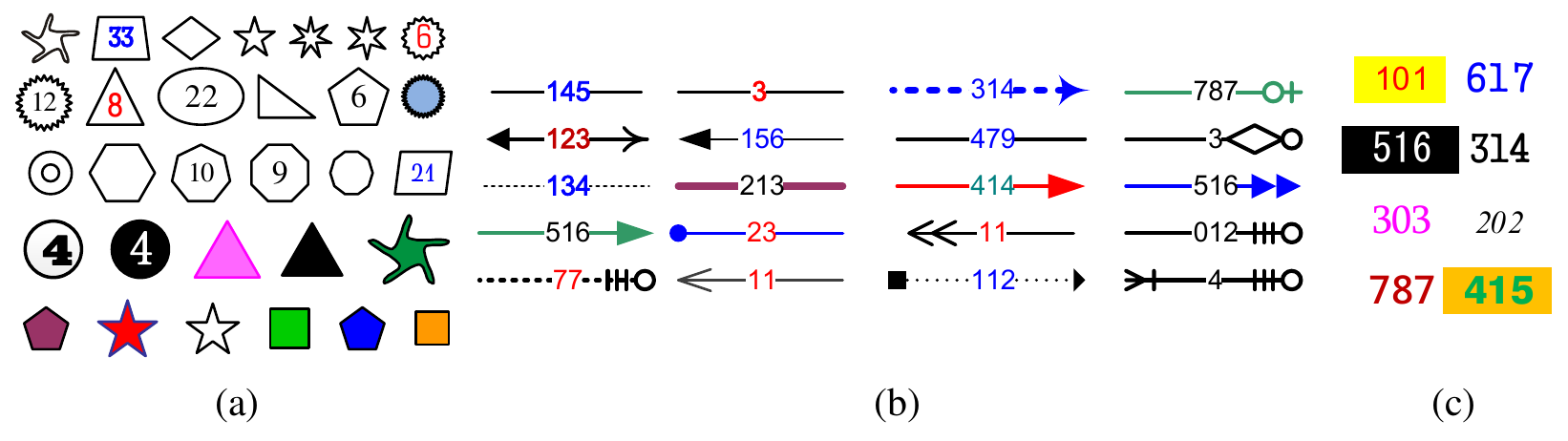}\\
\caption{\label{fig:various-edges-vertices-colored}{\small Various shapes of edges, vertices and letters with variety of colors.}}
\end{figure}

\begin{figure}[h]
\centering
\includegraphics[width=12cm]{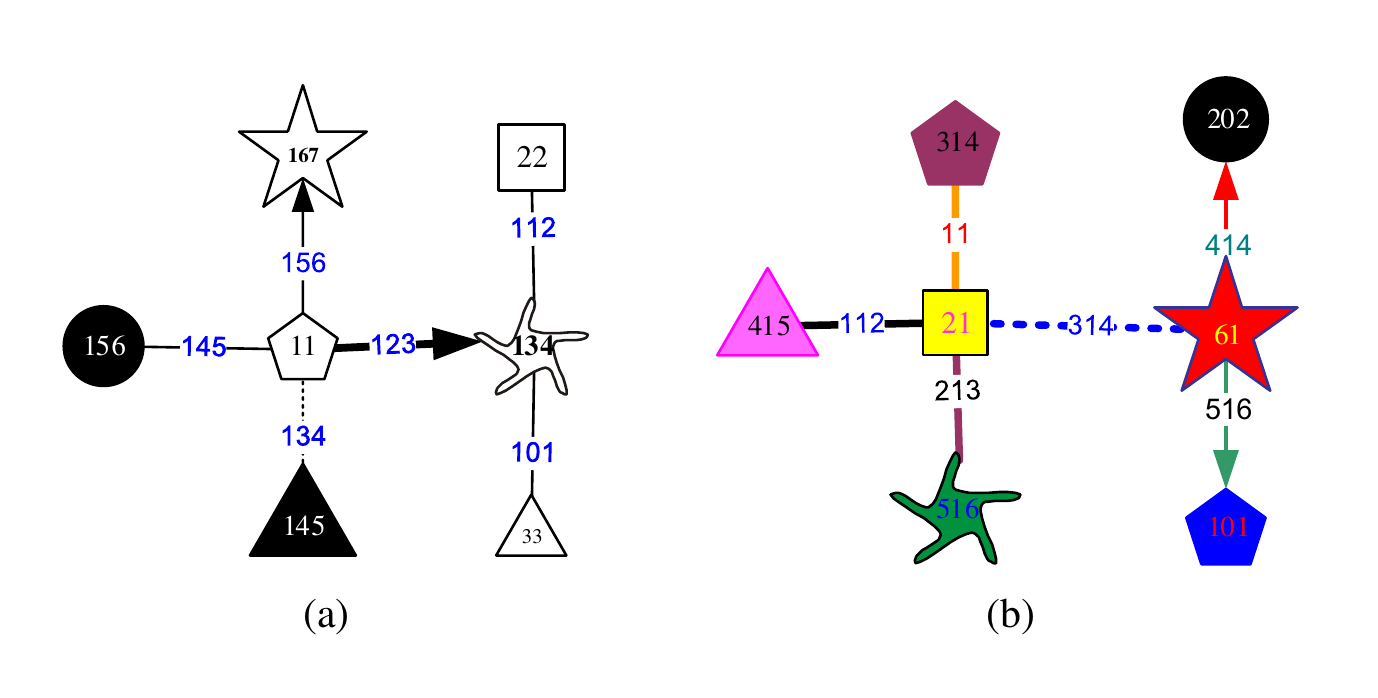}\\
\caption{\label{fig:Pan-various-Topsnut-gpws}{\small Two super-shape colored graphs, cited from \cite{Yao-Wang-2106-15254v1}.}}
\end{figure}

\subsection{Complete graphs}

Recall, an uncolored graph $G$ and its own \emph{complement} (or \emph{inverse}) $\overline{G}$ form a \emph{topological authentication} $K_n$ such that $E(K_n)=E(G)\cup E(\overline{G})$ with $E(G)\cap E(\overline{G})=\emptyset$ and $V(K_n)=V(G)=V(\overline{G})$, where $n=|V(G)|=|V(\overline{G})|$, and $G$ is called a \emph{topological public-key} and $\overline{G}$ is called a \emph{topological private-key}, write this case as $\langle G, \overline{G}\mid K_n\rangle$. There are the following graphs related with complete graphs:

(i) Tournaments.

(ii) Auto-isomorphic groups, automorphism. In the analysis of algorithms on graphs, the distinction between a graph and its complement is an important one, because a sparse graph (one with a small number of edges compared to the number of pairs of vertices) will in general not have a sparse complement, and so an algorithm that takes time proportional to the number of edges on a given graph may take a much larger amount of time if the same algorithm is run on an explicit representation of the complement graph.

(iii) A \emph{self-complementary graph} is a graph that is isomorphic to its own complement, some examples include the four-vertex path graph and five-vertex cycle graph. Any induced subgraph of the complement graph of a graph $G$ is the complement of the corresponding induced subgraph in $G$.

\subsubsection{Spanning trees of complete graphs}

As known, a colored complete graph $K_n$ holds the Cayley's formula \cite{Bondy-2008} below
\begin{equation}\label{eqa:55-Cayley-formula}
\tau(K_n) = n^{n-2}
\end{equation} where $\tau(K_n)$ is the number of spanning trees in the complete graph $K_n$ admitting a bijection $f:V(K_n)\rightarrow [1,n]$ with $f(V(K_n))=[1,n]$.

Let $F(n)$ be the number of all forests of the colored complete graph $K_n$. In \cite{Jianfang-Wang-Hypergraphs-2008}, we have
\begin{equation}\label{eqa:555555}
F(n)=\frac{n!}{n+1}\sum^{\lfloor \frac{n}{2}\rfloor}_{j=0}(-1)^j\frac{(2j+1)(n+1)^{n-2j}}{2^j\cdot j!(n-2j)!}
\end{equation}

\begin{rem}\label{rem:333333}
In Fig.\ref{fig:Cayley-formula-more}, we defined each edge $x_ix_j$ colored by $f(x_ix_j)=f(x_i)+f(x_j)$, also, we can color each edge $x_ix_j$ of the complete graph $K_4$ by a function $\varphi$ holding $f(x_ix_j)=\varphi(f(x_i),f(x_j))$ for producing more colorings of $K_n$, such as

$f(x_ix_j)=|f(x_i)-f(x_j)|$, $f(x_ix_j)=\textrm{gcd}(f(x_i),f(x_j))$,

$f(x_ix_j)=[f(x_i)+f(x_j)]^2$, $f(x_ix_j)=[f(x_i)-f(x_j)]^2$,

$f(x_ix_j)=\max \big \{f(x_i)^{f(x_j)},f(x_j)^{f(x_i)}\big \}$, and $f(x_ix_j)=af(x_i)+bf(x_j)$\\
with integers $a,b\geq 1$, and so on.\paralled
\end{rem}

\begin{example}\label{exa:spanning-tree-K-graph-decomposition}
The colored complete graph $K_4$ has 16 spanning trees shown in Fig.\ref{fig:Cayley-formula-more}, where there are four stars $S_1,S_2,S_3,S_4$ and 12 paths $P_j$ for $j\in [1,12]$ shown in Eq.(\ref{eqa:12-paths-Cayley-formula}). Four Topcode-matrices of four stars $S_1,S_2,S_3,S_4$ are shown in Eq.(\ref{eqa:stars-Cayley-formula-matrices}). Clearly, by a fixed NBSTRING algorithm-$k$ with $k=$I, II, III, IV, the number-based strings generated from four Topcode-matrices $T(S_j)$ with $j\in [1,4]$ are different from each other, even though $S_i\cong S_j$ for $1\leq i,j\leq 4$.
\end{example}

\begin{figure}[h]
\centering
\includegraphics[width=16.4cm]{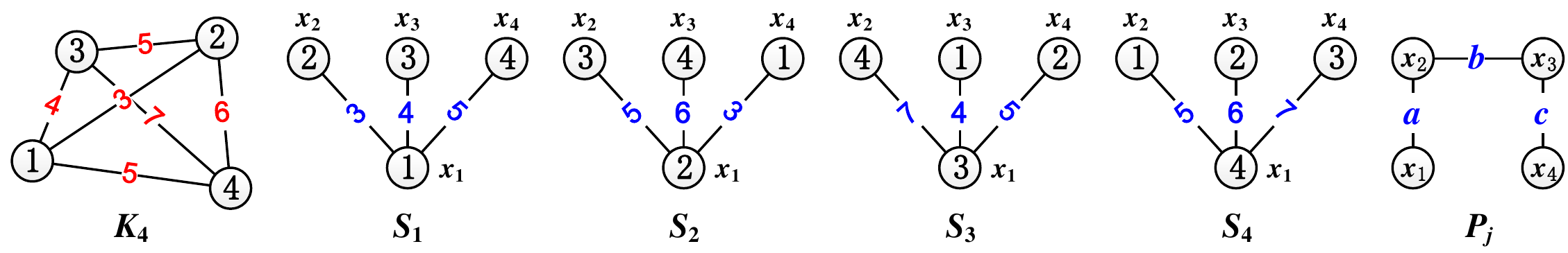}\\
\caption{\label{fig:Cayley-formula-more}{\small There are 16 spanning trees of $K_4$ by the Cayley-formula.}}
\end{figure}

\begin{equation}\label{eqa:12-paths-Cayley-formula}
\centering
{
\begin{split}
\begin{array}{c|ccccccc}
P_j & x_1 & \textcolor[rgb]{0.00,0.00,1.00}{a} & x_2 & \textcolor[rgb]{0.00,0.00,1.00}{b} & x_3 & \textcolor[rgb]{0.00,0.00,1.00}{c} & x_4\\
\hline
P_1 & 1 & \textcolor[rgb]{0.00,0.00,1.00}{\textbf{3}} & 2 & \textcolor[rgb]{0.00,0.00,1.00}{\textbf{5}} & 3 & \textcolor[rgb]{0.00,0.00,1.00}{\textbf{7}} & 4\\
P_2 & 1 & \textcolor[rgb]{0.00,0.00,1.00}{\textbf{3}} & 2 & \textcolor[rgb]{0.00,0.00,1.00}{\textbf{6}} & 4 & \textcolor[rgb]{0.00,0.00,1.00}{\textbf{7}} & 3\\
P_3 & 1 & \textcolor[rgb]{0.00,0.00,1.00}{\textbf{4}} & 3 & \textcolor[rgb]{0.00,0.00,1.00}{\textbf{5}} & 2 & \textcolor[rgb]{0.00,0.00,1.00}{\textbf{6}} & 4\\
P_4 & 1 & \textcolor[rgb]{0.00,0.00,1.00}{\textbf{4}} & 3 & \textcolor[rgb]{0.00,0.00,1.00}{\textbf{7}} & 4 & \textcolor[rgb]{0.00,0.00,1.00}{\textbf{6}} & 2\\
P_5 & 1 & \textcolor[rgb]{0.00,0.00,1.00}{\textbf{5}} & 4 & \textcolor[rgb]{0.00,0.00,1.00}{\textbf{6}} & 2 & \textcolor[rgb]{0.00,0.00,1.00}{\textbf{5}} & 3\\
P_6 & 1 & \textcolor[rgb]{0.00,0.00,1.00}{\textbf{5}} & 4 & \textcolor[rgb]{0.00,0.00,1.00}{\textbf{7}} & 3 & \textcolor[rgb]{0.00,0.00,1.00}{\textbf{5}} & 2
\end{array}
\end{split}}\qquad
{
\begin{split}
\begin{array}{c|ccccccc}
P_j & x_1 & \textcolor[rgb]{0.00,0.00,1.00}{a} & x_2 & \textcolor[rgb]{0.00,0.00,1.00}{b} & x_3 & \textcolor[rgb]{0.00,0.00,1.00}{c} & x_4\\
\hline
P_7 & 2 & \textcolor[rgb]{0.00,0.00,1.00}{\textbf{3}} & 1 & \textcolor[rgb]{0.00,0.00,1.00}{\textbf{4}} & 3 & \textcolor[rgb]{0.00,0.00,1.00}{\textbf{7}} & 4\\
P_8 & 2 & \textcolor[rgb]{0.00,0.00,1.00}{\textbf{3}} & 1 & \textcolor[rgb]{0.00,0.00,1.00}{\textbf{5}} & 4 & \textcolor[rgb]{0.00,0.00,1.00}{\textbf{7}} & 3\\
P_9 & 3 & \textcolor[rgb]{0.00,0.00,1.00}{\textbf{4}} & 1 & \textcolor[rgb]{0.00,0.00,1.00}{\textbf{3}} & 2 & \textcolor[rgb]{0.00,0.00,1.00}{\textbf{6}} & 4\\
P_{10} & 3 & \textcolor[rgb]{0.00,0.00,1.00}{\textbf{4}} & 1 & \textcolor[rgb]{0.00,0.00,1.00}{\textbf{5}} & 4 & \textcolor[rgb]{0.00,0.00,1.00}{\textbf{6}} & 2\\
P_{11} & 4 & \textcolor[rgb]{0.00,0.00,1.00}{\textbf{5}} & 1 & \textcolor[rgb]{0.00,0.00,1.00}{\textbf{3}} & 2 & \textcolor[rgb]{0.00,0.00,1.00}{\textbf{5}} & 3\\
P_{12} & 4 & \textcolor[rgb]{0.00,0.00,1.00}{\textbf{5}} & 1 & \textcolor[rgb]{0.00,0.00,1.00}{\textbf{4}} & 3 & \textcolor[rgb]{0.00,0.00,1.00}{\textbf{5}} & 2
\end{array}
\end{split}}
\end{equation}

We define an operation ``$\oplus$'' on the set $F(S)=\{S_1,S_2,S_3,S_4\}$, where each star $S_i$ admits a vertex coloring $g_i$ induced by the coloring $f$ of $K_4$, by selecting arbitrarily $S_k$ as the zero, thus, we have
\begin{equation}\label{eqa:555555}
f_i(x_t)+f_j(x_t)-f_k(x_t)=f_{\lambda}(x_t),~t\in [1,4]
\end{equation} where $\lambda=i+j-k~(\bmod~ 4)$, and $f_i(x_t)\in f_i(S_i)$, $f_j(x_t)\in f_j(S_j)$ and $f_{\lambda}(x_t)\in f_{\lambda}(S_{\lambda})$, such that ``$S_i\oplus S_j=S_{\lambda}\in F(S)$''. So, $F(S)$ forms an \emph{every-zero graphic group}, denoted as $\{F(S);\oplus\}$, since $F(S)$ holds the Zero, the Inverse, the Uniqueness and Closure law, the Associative law \cite{Yao-Sun-Zhao-Li-Yan-ITNEC-2017}. Clearly, the Topcode-matrix set $T(S)=\{T(S_1),T(S_2),T(S_3),T(S_4)\}$ forms an \emph{every-zero Topcode-matrix group}, see Eq.(\ref{eqa:stars-Cayley-formula-matrices}).

{\small
\begin{equation}\label{eqa:stars-Cayley-formula-matrices}
\centering
{
\begin{split}
T(S_1)= \left(
\begin{array}{ccccc}
1 & 1 & 1\\
3 & 4 & 5\\
2 & 3 & 4
\end{array}
\right)~T(S_2)= \left(
\begin{array}{ccccc}
2 & 2 & 2\\
5 & 6 & 3\\
3 & 4 & 1
\end{array}
\right)~T(S_3)= \left(
\begin{array}{ccccc}
3 & 3 & 3\\
7 & 4 & 5\\
4 & 1 & 2
\end{array}
\right)~T(S_4)= \left(
\begin{array}{ccccc}
4 & 4 & 4\\
5 & 6 & 7\\
1 & 2 & 3
\end{array}
\right)
\end{split}}
\end{equation}
}

The paths $P_j$ for $j\in [1,12]$, to form an every-zero graphic group in Eq.(\ref{eqa:path-group-paths-Cayley-formula}), shown in Eq.(\ref{eqa:12-paths-Cayley-formula}) of Table-3 in Appendix A. We get all every-zero graphic groups made by the spanning-tree groups of $K_4$:

$\{F(P_1);\oplus\}$ with $F(P_1)=\{P_{1,r}:r\in [0,3]\}=\{P_1,P_6,P_8,P_{11}\}$,

$\{F(P_2);\oplus\}$ with $F(P_2)=\{P_{2,r}:r\in [0,3]\}=\{P_2,P_{12}\}$,

$\{F(P_3);\oplus\}$ with $F(P_3)=\{P_{3,r}:r\in [0,3]\}=\{P_3,P_4,P_9,P_{10}\}$ and

$\{F(P_5);\oplus\}$ with $F(P_5)=\{P_{5,r}:r\in [0,3]\}=\{P_5,P_7\}$, as well as $\{F(S);\oplus\}$.

In other words, all spanning trees of $K_4$ are calcified into the every-zero graphic groups $\{F(S);\oplus\}$ and $\{F(P_j);\oplus\}$ with $j\in [1,3]$ and $\{F(P_5);\oplus\}$ (Ref. \cite{Wang-Su-Sun-Yao-submitted-ITOEC2020}).

\begin{problem}\label{qeu:spanning-tree-groups-K-n}
Since Cayley's formula $\tau(K_n) = n^{n-2}$ tells us there are $n^{n-2}$ spanning trees of $K_n$, \textbf{find} all every-zero graphic groups $\{F(S_k);\oplus \}$ made by the spanning trees of $K_n$. The complete graph $K_{26}$ has
$$26^{24}=9,106,685,769,537,220,000,000,000,000,000,000$$ colored spanning trees of 26 vertices. Notice that, there are $t_{26}=279,793,450$ non-isomorphic spanning trees of 26 vertices, and there are
$$
26^{24}\div t_{26}=32,547,887,627,595,300,000,000,000
$$ trees being isomorphic to others.

If $G$ is the complete bipartite graph $K_{m,n}$ then the number of all spanning trees of $K_{m,n}$ is $\tau (G)=m^{n-1}n^{m-1}$.
\end{problem}

\begin{defn} \label{defn:spanning-tree-K-graph-decomposition}
$^*$ For a group of integers $n_1,n_2,\dots ,n_m$, where each integer $n_j\geq 2$ for $j\in [1,m]$ with integer $m\geq 2$, decompose a colored complete graph $K_n$ into edge-disjoint connected subgraphs $H_1,H_2,\dots ,H_m$ such that $V(K_n)=V(H_j)$, $E(K_n)=\bigcup ^m_{j=1}E(H_j)$ and each connected subgraph $H_j$ contains just $n_j$ spanning trees of $K_n$ with $j\in [1,m]$. If $n^{n-2}=\sum ^m_{j=1}n_j$, we call $\langle H_1,H_2,\dots ,H_m\mid K_n\rangle $ a \emph{spanning tree $K$-graph decomposition}. Here, the group $(n_1,n_2,\dots ,n_m)$ is as a topological public-key group, and the edge-disjoint connected graph group $(H_1,H_2,\dots ,H_m)$ is as a \emph{topological private-key group}.\qqed
\end{defn}

\begin{problem}\label{qeu:spanning-tree-K-graph-decomposition}
By Definition \ref{defn:spanning-tree-K-graph-decomposition}, \textbf{find} all spanning tree $K$-graph decompositions $\langle H_1,H_2,\dots ,H_m\mid K_n\rangle $ for all possible integer groups $(n_1,n_2,\dots ,n_m)$ with $n_j\geq 2$ and $j\in [1,m]$, as well as $n^{n-2}=\sum ^m_{i=1}n_i$. Refer to Remark \ref{rem:adding-leaves-graph-odd-graceful-complexity} for the complexity of this problem.
\end{problem}

\begin{rem}\label{rem:333333}
A \emph{single spanning tree} of a graph can be found in linear time by either \emph{depth-first search} or \emph{breadth-first search} in graph theory. Both of these algorithms explore the given graph, starting from an arbitrary vertex $v$, by looping through the neighbors of the vertices they discover and adding each unexplored neighbor to a data structure to be explored later. They differ in whether this data structure is a stack (in the case of depth-first search) or a queue (in the case of breadth-first search). In either case, one can form a spanning tree by connecting each vertex, other than the root vertex $v$, to the vertex from which it was discovered. This tree is known as a depth-first search tree or a breadth-first search tree according to the graph exploration algorithm used to construct it (https://encyclopedia.thefreedictionary.com/Spanning+tree+(mathematics)).

\emph{Because of a graph may have exponentially many spanning trees}, it is not possible to list them all in polynomial time. \paralled
\end{rem}

\subsubsection{Tree decomposition of complete graphs}

\begin{conj} \label{conj:c2-KT-conjecture}
(Gy\'{a}r\'{a}s and Lehel, 1978; B\'{e}la Bollob\'{a}s, 1995) For integer $n\geq 3$, given $n$ vertex disjoint trees $T_k$ of $k$ vertices with respect to $1\leq k\leq n$. A complete graph $K_n$ can be decomposed into the union of $n$ edge-disjoint trees $H_k$, namely $K_n=\bigcup ^{n}_{k=1}H_k$, such that $T_k\cong H_k$ whenever $1\leq k\leq n$. Also, write this conjecture in a short notation $\langle T_1,T_2,\dots, T_n\mid K_n\rangle$.
\end{conj}

\begin{defn} \label{defn:totally-graceful-split-tree-groups}
$^*$ Suppose that a complete graph $K_n$ admits a total coloring $f:V(K_n)\cup E(K_n)\rightarrow [1,n]$ holding $|f(V(K_n))|=n$ and $f(E(K_n))=\{f(uv)=|f(u)-f(v)|:uv\in E(K_n)\}$. Doing the vertex-split operation to $K_n$ produces a group of trees $T_1,T_{2},T_{3},\dots ,T_{n}$, where each tree $T_{k}$ has just $k$ vertices for $k\in [1,n]$, such that $K_n=\odot|^n_{k=1}T_{k}$, also $\langle T_1,T_{2},T_{3},\dots ,T_{n}\mid K_n\rangle $ (refer to Conjecture \ref{conj:c2-KT-conjecture}). We call the set $S=\{T_1,T_{2},T_{3},\dots ,T_{n}\}$ a \emph{vertex-split tree-group} of $K_n$. Let $f_k$ be the total coloring of each tree $T_{k}$, so $f_k(x)\neq f_k(y)$ for $x,y\in V(T_{k})$, if
$$f_k(E(T_{k}))=\{f_k(uv)=|f_k(u)-f_k(v)|:uv\in E(T_{k})\}=[1,k-1]
$$ we call vertex-split tree-group $S=\{T_1,T_{2},T_{3},\dots ,T_{n}\}$ a \emph{totally graceful vertex-split tree-group} of $K_n$.
\qqed
\end{defn}

\begin{figure}[h]
\centering
\includegraphics[width=16.4cm]{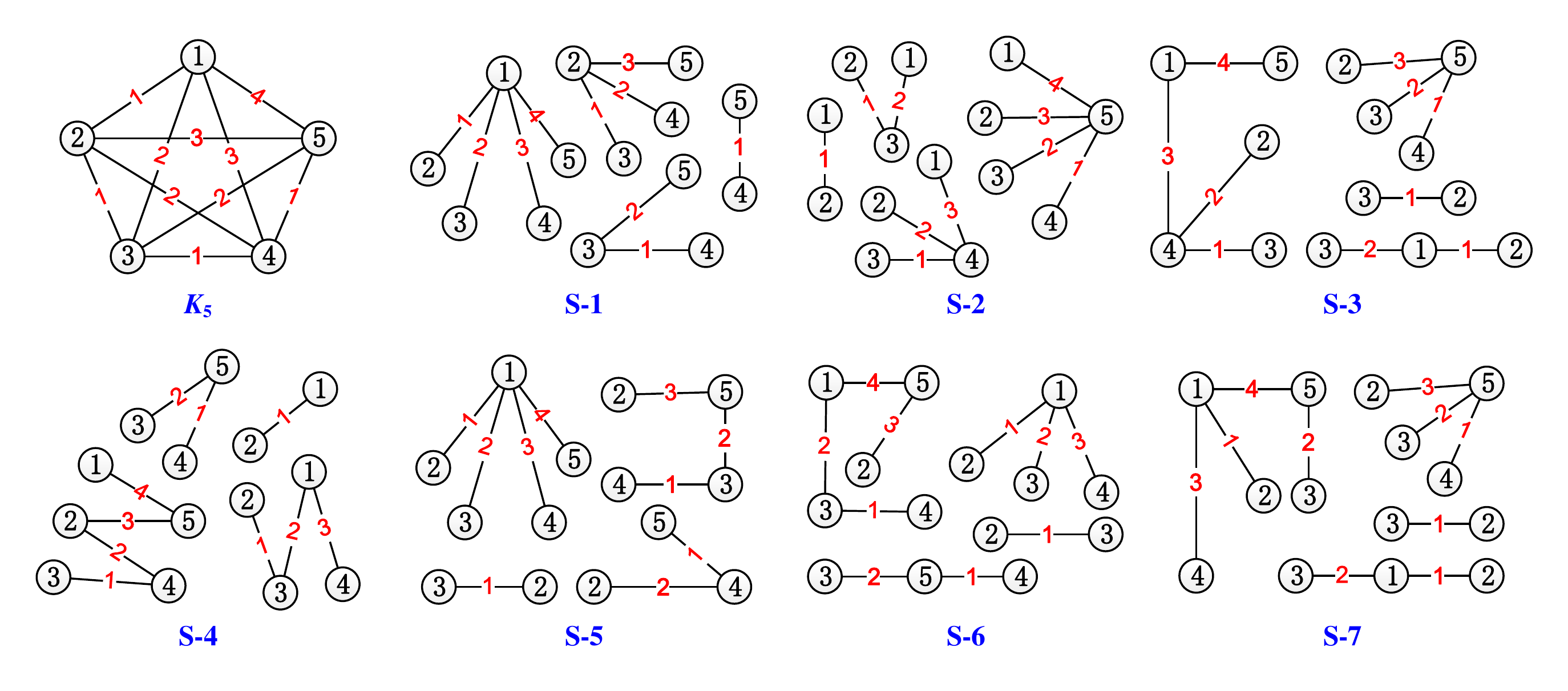}\\
\caption{\label{fig:total-graceful-group}{\small Seven totally graceful vertex-split tree-groups of $K_5$ for understanding Definition \ref{defn:totally-graceful-split-tree-groups}.}}
\end{figure}

\begin{thm}\label{thm:666666}
Suppose that a complete graph $K_n$ admits a total coloring $f:V(K_n)\cup E(K_n)\rightarrow [1,n]$ holding $|f(V(K_n))|=n$ and $f(E(K_n))=\{f(uv)=|f(u)-f(v)|:uv\in E(K_n)\}$. Then $K_n$ has at least a \emph{totally graceful vertex-split tree-group} $S_{plit}$ defined in Definition \ref{defn:totally-graceful-split-tree-groups}, and the dual $S^{dual}_{plit}$ of $S_{plit}$ is a totally graceful vertex-split tree-group too.
\end{thm}

In Fig.\ref{fig:total-graceful-group}, each tree $T_{k,j}$ of each vertex-split tree-group $S\textrm{-}k=\{T_{k,2},T_{k,3},T_{k,4},T_{k,5}\}$ has $j$ vertices and admits a total coloring $f_{k,j}$ holding $|f_{k,j}(V(T_{k,j}))|=j$ and
$$f_{k,j}(E(T_{k,j}))=\{f_{k,j}(uv)=|f_{k,j}(u)-f_{k,j}(v)|:uv\in E(T_{k,j})\}=[1,j-1]
$$ so each tree-group $S\textrm{-}k$ is a \emph{totally graceful vertex-split tree-group} of $K_5$ for $k\in [1,7]$. Clearly, $\langle S\textrm{-}k\mid K_5\rangle $ with $k\in [1,7]$.

Notice that $K_5$ has many vertex-split tree-groups that are not totally graceful vertex-split tree-groups, since $K_5$ has 125 different colored spanning trees. Each totally graceful vertex-split tree-group $S\textrm{-}k$ shown in Fig.\ref{fig:total-graceful-group} has its own dual $S_{dual}\textrm{-}k=\{T^{dual}_{k,2},T^{dual}_{k,3},T^{dual}_{k,4},T^{dual}_{k,5}\}$ with each tree $T^{dual}_{k,j}$ admitting a total coring $f^{dual}_{k,j}$ for $j\in[1,5]$ and $k\in [1,7]$, such that $S\textrm{-}k \cong S_{dual}\textrm{-}k$, and $f^{dual}_{k,j}$ is defined by
$$
f^{dual}_{k,j}(w)=\max f_{k,j}(V(S\textrm{-}k))+\min f_{k,j}(V(S\textrm{-}k))-f_{k,j}(w),
$$ and the edge color set
$$
f^{dual}_{k,j}(E(T^{dual}_{k,j}))=\{f^{dual}_{k,j}(uv)=|f^{dual}_{k,j}(u)-f^{dual}_{k,j}(v)|:uv\in E(T^{dual}_{k,j})\}=[1,j-1]
$$ for $j\in[1,5]$ and $k\in [1,7]$. Each dual $S_{dual}\textrm{-}k$ of $S\textrm{-}k$ for $k\in [1,7]$ is a totally graceful vertex-split tree-group too.

\begin{rem}\label{rem:333333}
In the techniques of topologically asymmetric cryptosystem, we have

(i) The group of vertex disjoint trees $T_1,T_2,\dots, T_n$ in Conjecture \ref{conj:c2-KT-conjecture} can be considered a group of topological private-keys. Or we take a part of vertex disjoint trees $T_{j_1},T_{j_2},\dots, T_{j_s}$ to be a group of public-keys, the remainder is as a group of private-keys. Let $t_k$ be the number of trees of $k$ vertices, so there are $\prod ^n_{k=1} (t_k)!$ groups of vertex disjoint trees $T_1,T_2,\dots, T_n$ for $\langle T_1,T_2,\dots, T_n\mid K_n\rangle$ defined in Conjecture \ref{conj:c2-KT-conjecture}, refer to Appendix.

(ii) For any tree group $S_{plit}=\{T_{2},T_{3},\dots ,T_{n}\}$ with each $T_k$ is a tree of $k$ vertices, the tree group $S_{plit}$ and its dual $S^{dual}_{plit}=\{T^{dual}_2,T^{dual}_3,\dots ,T^{dual}_n\}$ can be consider as a topological public-key and a topological private-key, respectively.

(iii) The \emph{Topcode-matrix group} $T_{code}(S_{plit})=\{T_{code}(T_{2}),T_{code}(T_{3})$, $\dots $, $T_{code}(T_{n})\}$ induces a group of number-based strings $s(m_2),s(m_3),\dots ,s(m_n)$ differs from the group of number-based strings $s(m\,'_2),s(m\,'_3),\dots ,s(m\,'_n)$ generated by the \emph{dual Topcode-matrix group}
\begin{equation}\label{eqa:555555}
T_{code}(S^{dual}_{plit})=\big \{T_{code}(T^{dual}_{2}),T_{code}(T^{dual}_{3}),\dots ,T_{code}(T^{dual}_{n})\big \}.
\end{equation} So, we can encrypt a group of digital files $F_2,F_3,\dots $, $F_n$ by a tree group $S_{plit}$ and its Topcode-matrix group $T_{code}(S_{plit})$ once time. Since there are transformations
\begin{equation}\label{eqa:555555}
\theta_i\langle T_{i}\rightarrow T^{dual}_{i}\rangle,\quad \theta\,'_i\langle T_{code}(T_{i})\rightarrow T_{code}(T^{dual}_{i})\rangle
\end{equation}

We can use the dual tree group $S^{dual}_{plit}$ and its Topcode-matrix group $T_{code}(S^{dual}_{plit})$ to decrypt the group of encrypted digital files $F^*_2,F^*_3,\dots ,F^*_n$ for obtaining the original files $F_2,F_3,\dots ,F_n$.\paralled
\end{rem}

The origins of graceful labelings lie in the problem of packing isomorphic copies of a given tree into a complete graph (Gerhard Ringel and Anton Kotzig, 1963; Alexander Rosa, 1967). If each tree admits a graceful labeling, then this will settle a longstanding and well-known Ringel-Kotzig Decomposition Conjecture:
\begin{conj}\label{conj:Ringel-Kotzig-Conjecture}
\cite{Ringel-G} A complete graph $K_{2n+1}$ can be decomposed into $2n+1$ subgraphs that are all isomorphic with a given tree having $n$ edges. (see Fig.\ref{fig:Ringel-Kotzig-conjecture})
\end{conj}

\begin{figure}[h]
\centering
\includegraphics[width=16.4cm]{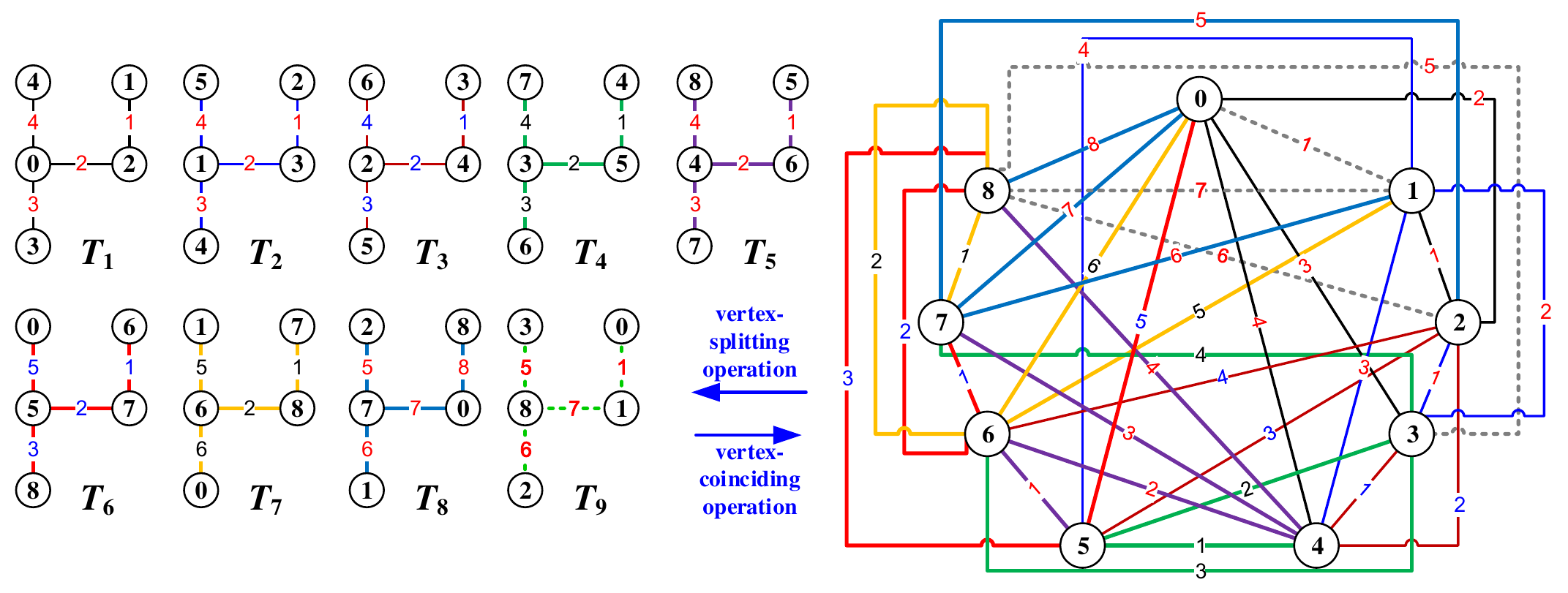}\\
\caption{\label{fig:Ringel-Kotzig-conjecture}{\small An example for illustrating the Ringel-Kotzig conjecture, where $K_9=\odot |^{9}_{j=1}T_j$, and an every-zero graphic group $\{F(T);\oplus\}$ with $F(T)=\{T_j:j\in [1,9]\}$.}}
\end{figure}

\begin{conj} \label{conj:chapter3-perfect-1-factor}
(Anton Kotzig, 1964; Perfect 1-Factorization Conjecture) For any $n \geq 2$, $K_{2n}$ can be decomposed into $2n-1$ perfect matchings such that the union of any two matchings forms a hamiltonian cycle of $K_{2n}$.
\end{conj}

\subsubsection{Complete graphs and cycles}

\begin{thm} \label{thm:complete-into-h-cycle}
\cite{Sun-Zhang-Yao-IAEAC-2017} Each complete graph $K_{2n+1}$ is obtained by doing a series of non-common vertex-coinciding operations to the cycle $C_{n(2n+1)}$ of $n(2n+1)$ vertices. Conversely, doing a series of non-common vertex-splitting operations to $K_{2n+1}$ produces $C_{n(2n+1)}$.
\end{thm}

In Fig.\ref{thm:complete-into-h-cycle}, doing the vertex-splitting operation to a complete graph $K_7$ produces $E_1$, we write this fact as $K_7\rightarrow _{split}E_1$, and $E_1$ is graph homomorphism to $K_7$, that is, $E_1\rightarrow K_7$, and moreover $E_i\rightarrow _{split}E_{i+1}$ and $E_{i+1}\rightarrow E_i$ for $i\in [1,4]$.

\begin{figure}[h]
\centering
\includegraphics[width=14.4cm]{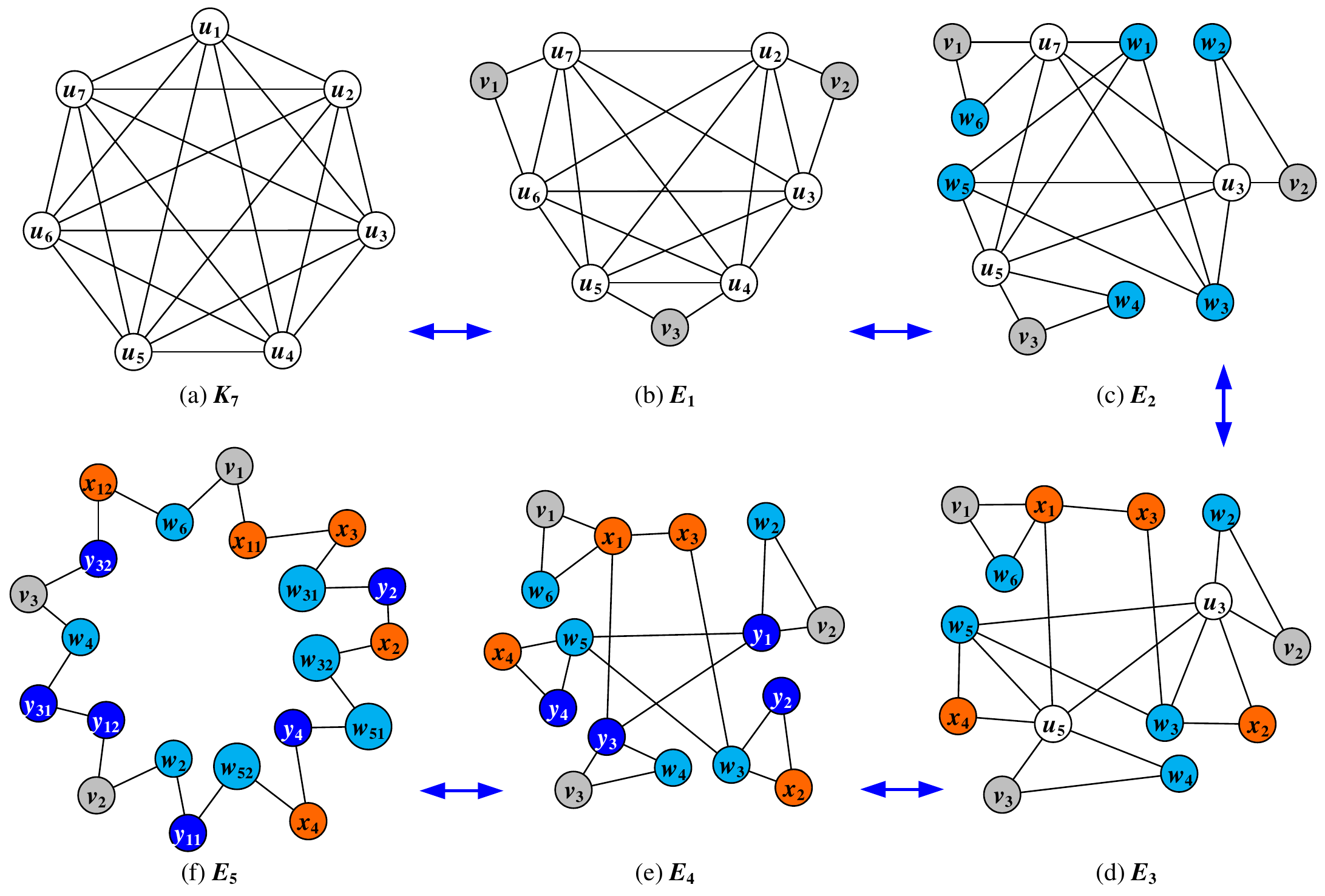}\\
\caption{\label{fig:K7-vs-H-cycle} {\small An example for Theorem \ref{thm:complete-into-h-cycle}.}}
\end{figure}

\begin{problem}\label{qeu:444444}
There are connected graphs $G_{i_1},G_{i_2},\dots ,G_{i_{m_i}}$ such that $K_{2n+1}\rightarrow _{split}G_{i_1}$, $G_{i_{j}}\rightarrow _{split}G_{i_{j+1}}$ for $j\in [1,m_i-1]$, and $G_{i_{m_i}}\rightarrow _{split}C_{n(2n+1)}$. \textbf{Find} all graphs $G_{i_1},G_{i_2},\dots ,G_{i_{m_i}}$.
\end{problem}

\subsection{Trees}

Trees are graphs that are non-cycle, connected and planar in graph theory. For another important reason, trees admit many labelings and colorings for generating number-based strings in application of cryptography. Some researching topics related with trees are: Abstract syntax tree, B-tree, Binary tree, Binary search tree, Self-balancing binary search tree, AVL tree, Red-black tree, Splay tree, T-tree, Binary space partitioning, Full binary tree, B*-tree, Binary heap, Binomial heap, Fibonacci heap, 2-3 heap, Kd-tree, Cover tree, Decision tree, Empty tree, Evolutionary tree, Exponential tree, Fault tree, Free tree, Game tree, K-ary tree, Octree, Parse tree, Phylogenetic tree, Polytree, Positional tree, PQ tree, R-tree, Rooted tree, Ordered tree, Recursive tree, SPQR tree, Suffix tree, Technology tree, Patricia trie, Spanning tree, Minimum spanning tree (Boruvka's algorithm, Kruskal's algorithm, Prim's algorithm), Steiner tree, Quadtree (Wikipedia).

\subsubsection{Basic properties of trees}

\begin{thm}\label{thm:Yao-tree-if-and-only-if}
\cite{Yao-Zhang-Yao-2007, Yao-Zhang-Wang-Sinica-2010} A connected graph $T$ is a tree if and only if
\begin{equation}\label{eqa:Yao-tree-iff}
n_1(T)=2+\sum _{d\geq 3}(d-2)n_d(T)
\end{equation} where $n_d(T)$ is the number of vertices of $d$ degree in $T$.
\end{thm}

\begin{cor}\label{thm:666666}
$^*$ A connected graph $T$ is a tree if and only if its vertex number holds
\begin{equation}\label{eqa:another-iff-condition-tree}
|V(T)|=n_1(T)+n_2(T)+\sum _{d\geq 3}n_d(T)
\end{equation}
\end{cor}
\begin{proof}Since
$${
\begin{split}
2|E(T)|&=n_1(T)+2n_2(T)+\sum _{d\geq 3}dn_d(T)\\
&=n_1(T)+2n_2(T)+\sum _{d\geq 3}(d-2)n_d(T)+\sum _{d\geq 3}dn_d(T)-\sum _{d\geq 3}(d-2)n_d(T)\\
&=2n_1(T)+2n_2(T)-2+2\sum _{d\geq 3}n_d(T)
\end{split}}
$$ and $|V(T)|=|E(T)|+1$, we are done.
\end{proof}

\begin{cor}\label{thm:666666}
A $(p,q)$-graph $G$ is a connected graph with $\delta(G)\geq 2$ if and only if $q=p+\frac{1}{2}\sum _{d\geq 3}(d-2)n_d(G)$, where $n_d(G)$ is the number of vertices of $d$ degree in $G$.
\end{cor}

\begin{thm}\label{thm:adding-edge-removing-sets}
$^*$ Let $T_{\pm e}(\leq n)$ be the set of trees of $p$ vertices with $p\leq n$. Then each tree $H\in T_{\pm e}(\leq n)$ is a star $K_{1,p-1}$, or corresponds another tree $T\in T_{\pm e}(\leq n)$ holding $H=T+uv-xy$ for $xy\in E(T)$ and $uv\not \in E(T)$.
\end{thm}
\begin{proof}For a $(p,q)$-graph $G$ with $p\leq n$, if $G\neq K_p$, then doing the \emph{adding-edge-removing operation} defined in Definition \ref{defn:col-pre-e-add-remov-operation} to $G$ produces a set $H_{\pm e}(G)$ of graphs $G+uv-xy$ for $xy\in E(G)$ and $uv\not \in E(G)$. So, a set $U_{\pm e}(\leq n)$ contains all elements of each set $H_{\pm e}(G)$, naturally, $T_{\pm e}(\leq n)\subset U_{\pm e}(\leq n)$, we are done.
\end{proof}

\begin{thm}\label{thm:adding-edge-removing-w-type-sets}
$^*$ Let $G_{tree}(\leq n)$ be the set of trees of $p$ vertices with $p\leq n$. Then each tree $H\in G_{tree}(\leq n)$ admits a $W$-type coloring and corresponds another tree $T\in G_{tree}(\leq n)$ admitting a $W$-type coloring holding $H=T+uv-xy$ for $xy\in E(T)$ and $uv\not \in E(T)$.
\end{thm}
\begin{proof}Let a tree $G$ of $p$ vertices with $p\leq n$ admit a $W$-type coloring. Doing the \emph{adding-edge-removing operation} defined in Definition \ref{defn:col-pre-e-add-remov-operation} to $G$ produces a set $L_{\pm e}(G)$ of trees $G+uv-xy$ for $xy\in E(G)$ and $uv\not \in E(G)$, such that each tree $G+uv-xy$ admits a $W$-type coloring too. We then have a set $W_{\pm e}(\leq n)$ contains all elements of each $L_{\pm e}(G)$, immediately, $G_{tree}(\leq n)\subset W_{\pm e}(\leq n)$, the proof is complete.
\end{proof}

\begin{rem}\label{rem:333333}
In the proof of Theorem \ref{thm:adding-edge-removing-sets}, the set $H_{\pm e}(G)$ has $\big [\frac{1}{2}p(p-1)-q\big ]$ graphs like $G+uv-xy$ for $xy\in E(G)$ and $uv\not \in E(G)$. For the set $L_{\pm e}(G)$ in the proof of Theorem \ref{thm:adding-edge-removing-w-type-sets}, it is not easy to determine the \emph{cardinality} of $L_{\pm e}(G)$, since there are hundreds of colorings in graph theory, and it is unable to find all of a tree's colorings in a valid time period.

Clearly, two sets $T_{\pm e}(\leq n)$ and $G_{tree}(\leq n)$ are suitable to make public-keys and private-keys in topological authentication.

Each connected graph $G$ can be vertex-split into trees $H_1,H_2,\dots ,H_{m(G)}$, such that $H_i\not \cong H_j$ for $i\neq j$, so each coloring/labeling $f_k$ of $G$ induces a coloring/labeling $f_{k,i}$ of $H_i$. If $G$ is a public-key, then finding a particular private-key $H_i$ is not easy, however finding $G$ from a colored graph $H_i$ is difficult, since there are huge number of colorings and labelings, and it is extremely difficult to determine the coloring/labeling $f_{k,i}$. \paralled
\end{rem}

\subsubsection{Colorings and labelings of trees}

We show some colorings and labelings on trees, often used. In \cite{Yao-Liu-Yao-2017}, the authors have proven the following mutually equivalent labelings:

\begin{thm} \label{thm:connections-several-labelings}
\cite{Yao-Liu-Yao-2017} Let $T$ be a tree on $p$ vertices, and let $(X,Y)$ be its
bipartition. For all values of integers $k\geq 1$ and $d\geq 1$, the following assertions are mutually equivalent:

$(1)$ $T$ admits a set-ordered graceful labeling $f$ with $f(X)<f(Y)$.

$(2)$ $T$ admits a super felicitous labeling $\alpha$ with
$\alpha(X)<\alpha(Y)$.

$(3)$ $T$ admits a $(k,d)$-graceful labeling $\beta$ with
$\beta(x)<\beta(y)-k+d$ for all $x\in X$ and $y\in Y$.

$(4)$ $T$ admits a super edge-magic total labeling $\gamma$ with
$\gamma(X)<\gamma(Y)$ and a magic constant $|X|+2p+1$.

$(5)$ $T$ admits a super $(|X|+p+3,2)$-edge antimagic total
labeling $\theta$ with $\theta(X)<\theta(Y)$.

$(6)$ $T$ has an odd-elegant labeling $\eta$ with
$\eta(x)+\eta(y)\leq 2p-3$ for every edge $xy\in E(T)$.

$(7)$ $T$ has a $(k,d)$-arithmetic labeling $\psi$ with
$\psi(x)<\psi(y)-k+d\cdot |X|$ for all $x\in X$ and $y\in Y$.

$(8)$ $T$ has a harmonious labeling $\varphi$ with
$\varphi(X)<\varphi(Y\setminus \{y_0\})$ and $\varphi(y_0)=0$.
\end{thm}

By Lemma \ref{thm:graceful-image-labeling} and Theorem \ref{thm:connections-several-labelings}, if a tree $T$ admits a set-ordered graceful labeling, then we have the following results:

\begin{thm}\label{thm:10-image-labelings}
\cite{Yao-Zhang-Sun-Mu-Sun-Wang-Wang-Ma-Su-Yang-Yang-Zhang-2018arXiv} If a tree $T$ admits a set-ordered graceful labeling, then $T$ admits a matching of $SKD$ image-labelings, where $SKD\in \{$set-ordered graceful, set-ordered odd-graceful, edge-magic graceful, set-ordered felicitous, set-ordered odd-elegant, super set-ordered edge-magic total, super set-ordered edge-antimagic total, set-ordered $(k,d)$-graceful, $(k,d)$-edge antimagic total, $(k,d)$-arithmetic total, harmonious, $(k,d)$-harmonious $\}$.
\end{thm}

\begin{thm}\label{thm:inverse-reciprocal}
\cite{Yao-Zhang-Sun-Mu-Sun-Wang-Wang-Ma-Su-Yang-Yang-Zhang-2018arXiv} If two trees of $p$ vertices admit set-ordered graceful labelings, then they are inverse matching to each other under the edge-magic graceful labelings.
\end{thm}

\begin{thm}\label{thm:Yao-tree-admits-colorings-group-11}
Let $T$ be a tree on $p$ vertices, and let $(X,Y)$ be its bipartition of vertex set $V(T)$, and integers $k\geq 1$ and $d\geq 1$. If $T$ admits a set-ordered graceful labeling, then it admits each one of the following labelings that can be found in \cite{Gallian2021, Yao-Wang-2106-15254v1}:
\begin{asparaenum}[\textrm{Class}-1. ]
\item graceful labeling, set-ordered graceful labeling, graceful-intersection total set-labeling, graceful group-labeling.
\item odd-graceful labeling, set-ordered odd-graceful labeling, edge-odd-graceful total labeling, odd-graceful-intersection total set-labeling, odd-graceful group-labeling, perfect odd-graceful labeling.
\item elegant labeling, odd-elegant labeling.
\item 6C-labeling, 6C-complementary matching.
\item edge-magic total labeling, super edge-magic total labeling, super set-ordered edge-magic total labeling, edge-magic total graceful labeling.
\item relaxed edge-magic total labeling.
\item odd-edge-magic matching labeling, ee-difference odd-edge-magic matching labeling.
\item 6C-labeling, odd-6C-labeling.
\item image-labeling, harmonious image-labeling.
\item harmonious labeling $\varphi$ with $\max \varphi(X)<\min \varphi(Y\setminus \{y_0\})$ and $\varphi(y_0)=0$.
\item an ee-difference graceful-magic matching labeling.
\item difference-sum labeling, felicitous-sum labeling.
\item multiple edge-meaning vertex labeling.
\item perfect $\varepsilon$-labeling.
\item image-labeling, $(k,d)$-harmonious image-labeling.
\item super felicitous labeling $\alpha$ with $\max \alpha(X)<\min \alpha(Y)$.
\item super edge-magic total labeling $\gamma$ with $\max \gamma(X)<\min \gamma(Y)$ and a magic constant $|X|+2p+1$.
\item super $(|X|+p+3,2)$-edge antimagic total labeling $\theta$ with $\max \theta(X)<\min \theta(Y)$.
\item odd-elegant labeling $\eta$ with $\eta(x)+\eta(y)\leq 2p-3$ for every edge $xy\in E(T)$.
\item $(k,d)$-graceful labeling $\beta$ with $\beta(x)<\beta(y)-k+d$ for all $x\in X$ and $y\in Y$.
\item $(k,d)$-arithmetic labeling $\psi$ with $\max \psi(x)<\min \psi(y)-k+d\cdot |X|$ for all $x\in X$ and $y\in Y$.
\item $(k,d)$-edge antimagic total labeling, $(k,d)$-arithmetic.
\item $(k,d)$-arithmetic labeling $\psi$ with $\max \psi(x)<\min \psi(y)-k+d\cdot |X|$ for all $x\in X$ and $y\in Y$.
\item $(k,d)$-harmonious labeling $\varphi$ with $\max \varphi(X)<\min \varphi(Y\setminus \{y_0\})$ and $\varphi(y_0)=0$.
\item twin $(k,d)$-labeling, twin Fibonacci-type graph-labeling, twin odd-graceful labeling.
\end{asparaenum}
\end{thm}

By RLA-algorithm for the odd-graceful labeling, RLA-algorithm for the $(k,d)$-harmonious labeling, RLA-algorithm for the $(k,d)$-odd-elegant labeling and RLA-algorithm for the gracefully total coloring introduced in Section 4, we have the following results:

\begin{thm}\label{thm:Yao-tree-admits-colorings-group-22}
Adding $m$ leaves to a connected bipartite $(p,q)$-graph $G$, the resultant graph is a new connected bipartite $(p+m,q +m)$-graph $G^*$. If $G$ admits a set-ordered odd-graceful labeling, then
\begin{asparaenum}[\textrm{Tcl}-1. ]
\item $G^*$ admits an odd-graceful labeling, an odd-elegant labeling.
\item $G$ admits a $(k,d)$-harmonious labeling, and $G^*$ admits a $(k,d)$-harmonious labeling.
\item Every tree admits a $(k,d)$-harmonious labeling.
\item Every tree admits a $(k,d)$-odd-elegant labeling for integers $d\geq 1$ and $k\geq 0$.
\item Every tree admits a $(k,d)$-elegant labeling
\item Every tree admits a $(k,d)$-edge-magic total labeling (including an odd-edge-magic total labeling)
\item Every tree admits a $(k,d)$-graceful difference labeling.
\item Every tree admits a $(k,d)$-felicitous difference labeling.
\item Every tree admits a gracefully total coloring.
\item Every tree admits an odd-gracefully total coloring.
\item Every tree admits a $(k,d)$-gracefully total coloring.
\item Every connected $(p,q)$-graph $G$ admits a \emph{v-set e-odd-elegant labeling} (resp. \emph{v-set $(k,d)$-odd-elegant labeling}) defined in Definition \ref{defn:55-v-set-e-proper-more-labelings}.
\end{asparaenum}
\end{thm}

\begin{defn} \label{defn:totally-k-d-sequential-labeling}
$^*$ For integers $k\geq 1$ and $d\geq 1$, a \emph{totally $(k,d)$-sequential labeling} $g$ of a $(p,q)$-graph $G$ is a mapping $g:V(G) \cup E(G)\rightarrow S_{k,d}$ with $S_{k,d}=\{k, k+d, \dots, k+(p+q-1)d\}$, such that $g(xy) =k-d+|g(x)-g(y)|$ and $g(V(G) \cup E(G))=S_{k,d}$.\qqed
\end{defn}

\begin{figure}[h]
\centering
\includegraphics[width=16.4cm]{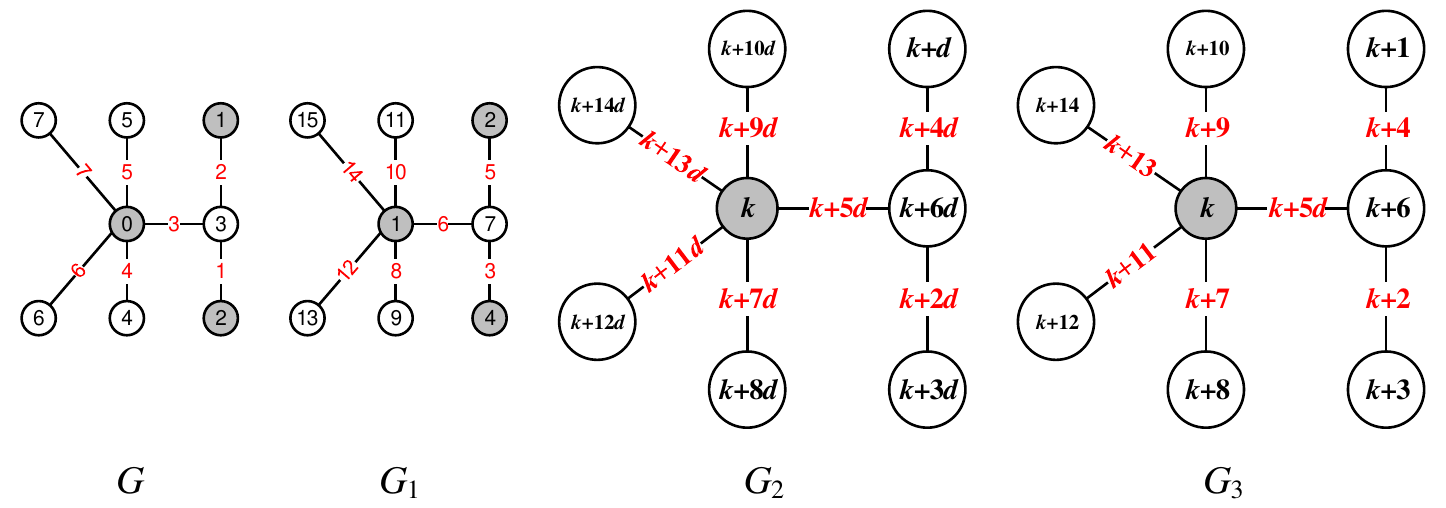}\\
\caption{\label{fig:totally-k-d-sequential} {\small For understanding Definition \ref{defn:totally-k-d-sequential-labeling}, a tree $G_2$ admits a totally $(k,d)$-sequential labeling, and another tree $G_3$ admits a totally $k$-sequential labeling as $(k,d)=(k,1)$.}}
\end{figure}

\begin{thm}\label{thm:tree-totally-k-d-sequential-labeling}
$^*$ If a tree admits a set-ordered graceful labeling, then it admits a totally $(k,d)$-sequential labeling defined in Definition \ref{defn:totally-k-d-sequential-labeling}.
\end{thm}
\begin{proof} By the definition of a set-ordered graceful labeling of a tree $T$, we have $V(T)=X\cup Y$ with $X\cap Y=\emptyset$, where $X=\{x_1,x_2,\dots ,x_s\}$ and $Y=\{y_1,y_2,\dots ,y_t\}$ with $s+t=p=|V(T)|=q+1$. Since $\max f(X)<\min f(Y)$, we have $f(x_r)=r-1$ for $r\in [1,s]$, $f(y_j)=s+j-1$ for $j\in [1,t]$, and
$$
f(E(T))=\{f(x_iy_j)=f(y_j)-f(x_i)=s+j-i:x_iy_j\in E(G)\}=[1,q]
$$ We then define a new labeling $h$ for the tree $T$ as:

(i) $h(x_1)=k$, $h(x_i)=k+(2i-3)d$ for $i\in [2,s]$, $h(y_j)=k+2(s+j-1)d$ for $j\in [1,t]$; and

(ii) For each edge $x_iy_j\in E(T)$, we set
$${
\begin{split}
h(x_iy_j)&=k-d+h(y_j)-h(y_j)=k-d+[k+2(s+j-1)d-k-(2i-3)d]\\
&=k+[2(s+j-1)-(2i-3)-1]d=k+2(s+j-i)d
\end{split}}
$$

So, we get two vertex color sets
$$h(X)=\{k,k+d,k+3d,\dots ,k+(2s-3)d\},\quad h(Y)=\{k+(2s+1)d,k+(2s+3)d,\dots, k+(2q+1)d\}
$$ and an edge color set $h(E(T))=\{k+2d,k+4d,\dots ,k+2(p-1)d\}$. Thereby,
$$h(X)\cup h(Y)\cup h(E(T))=S_{k,d}$$

If $h(x_i)=k+(2i-3)d=k+2(s+j-i)d=h(x_iy_j)$, which induces $2i-3=2(s+j-i)$, thus $4i-3=2(s+j)$, an obvious contradiction.

If $h(y_j)=k+2(s+j-1)d=k+2(s+j-i)d=h(x_iy_j)$, then we get $2(s+j-1)=2(s+j-i)$, and $-2=-2i$,a contradiction, since $i\in [2,s]$, .

By the above deduction, we claim that $T$ admits a totally $(k,d)$-sequential labeling for integers $k\geq 1$ and $d\geq 1$.
\end{proof}

\begin{rem}\label{rem:333333}
There is a group of odd-edge-magic total coloring, odd-graceful-difference total coloring, odd-edge-difference total coloring and odd-felicitous-difference total coloring defined in Definition \ref{defn:4-magic-total-colorings}, and another group of parameterized odd-edge-magic total coloring, parameterized odd-graceful-difference total coloring, parameterized odd-edge-difference total coloring, parameterized odd-felicitous-difference total coloring is defined in Definition \ref{defn:combinatoric-definition-total-coloring-abc}.\paralled
\end{rem}

\subsubsection{Rotatable labelings}

\begin{defn}\label{defn:mf-graceful-mf-odd-graceful}
\cite{Yao-Mu-Sun-Sun-Zhang-Wang-Su-Zhang-Yang-Zhao-Wang-Ma-Yao-Yang-Xie2019} For any vertex $u$ of a connected and bipartite $(p,q)$-graph $G$, there exist a vertex labeling $f:V(G)\rightarrow [0,q]$ (resp. $f:V(G)\rightarrow [0,2q-1]$) such that

(i) $f(u)=0$;

(ii) $f(E(G))=\{f(xy)=|f(x)-f(y)|: ~xy\in E(G)\}=[1,q]$ (resp. $f(E(G))=[1,2q-1]^o$);

(iii) the bipartition $(X,Y)$ of $V(G)$ holds $\max f(X)<\min f(Y)$.

Then we say $G$ admits a \emph{$0$-rotatable set-ordered system of graceful labelings} (resp. \emph{$0$-rotatable set-ordered system of odd-graceful labelings}).\qqed
\end{defn}

Recall an \emph{edge-symmetric graph} $G\ominus G\,'$ is obtained by join a vertex $u$ of a graph $G$ with its image $u\,'$ of a copy $G\,'$ of $G$ by a new edge $uu\,'$, see examples shown in Fig.\ref{fig:0-rotatable-system}, where $H_j=T_1\ominus T_2$ for $j=1,2$, and $T_1\cong T_2$, and $H_1\cong H_2$.

\begin{lem}\label{thm:symmetric-tree}
\cite{Yao-Mu-Sun-Sun-Zhang-Wang-Su-Zhang-Yang-Zhao-Wang-Ma-Yao-Yang-Xie2019} If a tree $T$ admits a $0$-rotatable system of (odd-)graceful labelings, then its edge symmetric tree $T\ominus T\,'$ admits a $0$-rotatable set-ordered system of (odd-)graceful labelings.
\end{lem}

By Lemma \ref{thm:symmetric-tree}, we can prove the following results (see an example shown in Fig.\ref{fig:0-rotatable-system}):

\begin{thm}\label{thm:0-rotatable-set-ordered-system00}
\cite{Yao-Mu-Sun-Sun-Zhang-Wang-Su-Zhang-Yang-Zhao-Wang-Ma-Yao-Yang-Xie2019} Suppose that a connected and bipartite $(p,q)$-graph $G$ admits a $0$-rotatable set-ordered system of (odd-)graceful labelings. Then the edge symmetric graph $G\ominus G\,'$ admits a $0$-rotatable set-ordered system of (odd-)graceful labelings too, where $G\,'$ is a copy of $G$, and $G\ominus G\,'$ is obtained by joining a vertex of $G$ with its image in $G\,'$ by a new edge.
\end{thm}

\begin{figure}[h]
\centering
\includegraphics[width=16.4cm]{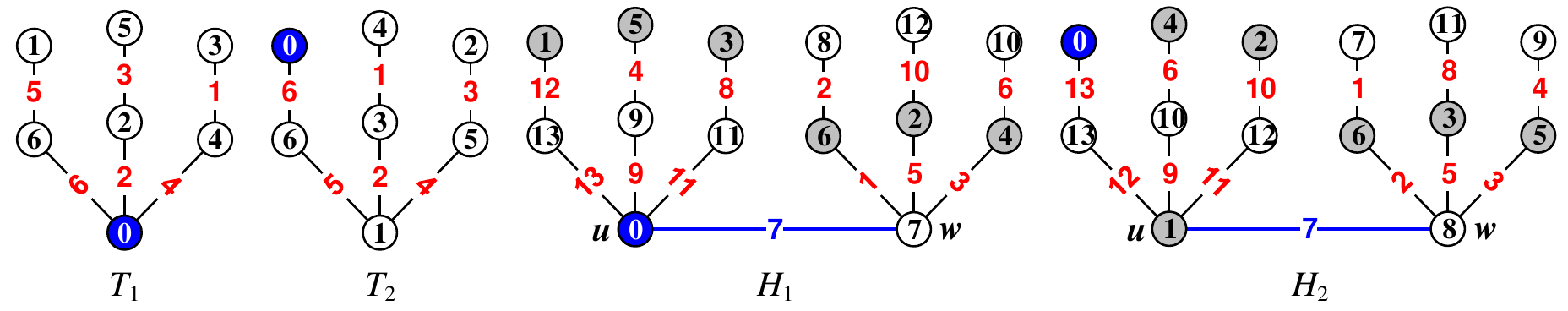}\\
\caption{\label{fig:0-rotatable-system} {\small Two trees $T_1,T_2$ admit two \emph{$0$-rotatable graceful labelings}, but set-ordered. Two tree $H_1$ and $H_2$ admit two \emph{set-ordered $0$-rotatable graceful labelings}, cited from \cite{Yao-Mu-Sun-Sun-Zhang-Wang-Su-Zhang-Yang-Zhao-Wang-Ma-Yao-Yang-Xie2019}.}}
\end{figure}

\begin{thm}\label{thm:0-rotatable-set-ordered-system11}
\cite{Yao-Mu-Sun-Sun-Zhang-Wang-Su-Zhang-Yang-Zhao-Wang-Ma-Yao-Yang-Xie2019} There are infinite graphs admit $0$-rotatable set-ordered systems of (odd-)graceful labelings.
\end{thm}

\begin{rem}\label{thm:CCCCCC}
Notice that a $0$-rotatable set-ordered system of graceful labelings defined in Definition \ref{defn:mf-graceful-mf-odd-graceful} is a set of set-ordered graceful labelings. We can consider other $k$-rotatable system of $W$-type labelings with $k\geq 0$, such as: edge-magic total labeling, elegant labeling, odd-elegant labeling, felicitous labeling, $(k,d)$-graceful labeling, edge antimagic total labeling, $(k, d)$-arithmetic, harmonious labeling, odd-edge-magic matching labeling, relaxed edge-magic total labeling, 6C-labeling, odd-6C-labeling, and so on.\paralled
\end{rem}

\subsubsection{Multiple dimension colorings on trees}

We will use some $n$-dimension colorings defined in Definition \ref{defn:n-dimension-total-colorings} to color trees for making hyper-strings.

\begin{defn} \label{defn:tree-n-dimension-colorings}
$^*$ Suppose that a tree $T$ admits $n$ different colorings $f_1,f_2,\dots ,f_n$ with $f_i(V(T))\subseteq [0,M^v_i]$ and
$$f_i(E(T))=\{f_i(uv)=\varphi_i\langle f_i(u),f_i(v)\rangle:uv\in E(T)\}\subseteq [0,M^e_i]$$

Let $M^v=\max\{M^v_1,M^v_2,\dots ,M^v_n\}$ and $M^e=\max\{M^e_1,M^e_2,\dots ,M^e_n\}$, then the tree $T$ admits a $n$-dimension coloring $F:V(T)\rightarrow [0,M^v]$ holding
\begin{equation}\label{eqa:tree-n-dimension-colorings}
{
\begin{split}
F(u)&=f_1(u)f_2(u)\cdots f_n(u),~u\in V(T)\\
F(uv)&=\varphi_1\langle f_1(u),f_1(v)\rangle \varphi_2\langle f_2(u),f_2(v)\rangle \cdots \varphi_n\langle f_n(u),f_n(v)\rangle,~uv\in E(T)
\end{split}}
\end{equation}
And the tree $T$ has its own Topcode-matrix
\begin{equation}\label{eqa:555555}
T_{code}(T,F)=\uplus |^q_{i=1}(F(u_i),F(u_iv_i),F(v_i))^T
\end{equation} defined by Eq.(\ref{eqa:tree-n-dimension-colorings}), where $E(T)=\{u_iv_i:~i\in [1,q]\}$.\qqed
\end{defn}

\begin{example}\label{exa:complete-graph-trees-n-dimension}
By the Cayley's formula $\tau(K_m)=m^{m-2}$, a complete graph $K_m$ vertex-colored by a vertex coloring $f$ has $\tau(K_m)$ different colored spanning trees, we collect these colored trees into the set $S_{pan}(K_m)$. Since each spanning tree $H\in S_{pan}(K_m)$ has been vertex-colored by a vertex coloring $f_H$ induced by the vertex coloring $f$ of $K_m$, so $H$ has its own Topcode-matrix $T_{code}(H,f_H)$, then we recolor this spanning tree $H$ by a $n$-dimension coloring $F_H$ defined in Definition \ref{defn:tree-n-dimension-colorings}, which produces a Topcode-matrix $T_{code}(H,F)$ of $H$. Thereby, this colored spanning tree $H$ distributes us a \emph{hyper-string} $s(M_H)$ made by $T_{code}(H,f_H)$ and $T_{code}(H,F)$, or $T_{code}(H,f_H)\uplus T_{code}(H,F)$, or $T_{code}(H,f_H)\cup T_{code}(H,F)$.
\end{example}

\begin{rem}\label{rem:333333}
(i) In Definition \ref{defn:tree-n-dimension-colorings}, a Topcode-matrix $T_{code}(T,F)$ can induce $(3q)!$ hyper-strings $h_{string}(T)=A_1A_2\cdots A_{3q}$ in total, where $A_j=F(u)$, or $A_j=F(uv)$ defined in Eq.(\ref{eqa:tree-n-dimension-colorings}).

(ii) In Example \ref{exa:complete-graph-trees-n-dimension}, we suppose that there are $n(K_m)$ different vertex colorings for the complete graph $K_m$, so the Topcode-matrix $T_{code}(H,f_H)$ of each spanning tree $H\in S_{pan}(K_m)$ is $(m-1)$-rank, and $T_{code}(H,F)$ is $(m-1)$-rank too. Thereby, each colored spanning tree $H$ distributes us at least $N_H$ hyper-strings $s(M_H)$ based on a $n$-dimension coloring $F_H$ of $H$, where $N_H=n(K_m)[(m-1)!]^2$.

(iii) We can select particular spanning trees from the spanning tree set $S_{pan}(K_m)$ to produce public-keys and private-keys with particular properties and structures of graph theory.\paralled
\end{rem}

\subsubsection{Star-tree systems}

\begin{defn} \label{defn:uncolored-ice-flower-systems}
\cite{Bing-Yao-2020arXiv, Yao-Wang-Liu-ice-flower-2020arXiv} \textbf{Ice-flower systems.} A star $K_{1,m_{i}}$ is a tree having $m_{i}$ leaves $y_{i,r}$ for $r\in [1,m_{i}]$ and a unique center $x_i$ of $m_{i}$ degree with $m_{i}\geq 1$, that is, $V(K_{1,m_{i}})=\{x_i,y_{i,r}:r\in [1,m_{i}]\}$. Doing the leaf-coinciding operation to a \emph{star base} $\textbf{\textrm{K}}=(K_{1,m_{1}},K_{1,m_{2}},\dots ,K_{1,m_{n}})=(K_{1,j})^n_{j=1}$ in the following way: We leaf-coincide a \emph{leaf-edge} $x_iy_{i,r}$ of a star $K_{1,m_{i}}$ with a \emph{leaf-edge} $x_{j}y_{j,s}$ of another star $K_{1,m_{j}}$ into one edge $x_ix_{j}=x_iy_{i,r}\overline{\ominus} x_{j}y_{j,s}$ for $i\neq j$, which joins $K_{1,m_{i}}$ and $K_{1,m_{j}}$ together, the resultant graph is denoted as $K_{1,m_{i}}\overline{\ominus} K_{1,m_{j}}$. It is allowed some star $K_{1,m_{i}}$ participates the leaf-coinciding operation with other stars more than twice, we get graphs $\overline{\ominus}| ^n_{i=1}a_iK_{1,m_{i}}$ for $a_i\in Z^0$, called \emph{leaf-coincided graphs} if they have not any leaf of each star $K_{1,m_{i}}$ of the star base $\textbf{\textrm{K}}$, and we call $\textbf{\textrm{K}}$ an \emph{ice-flower system}.\qqed
\end{defn}

\begin{defn} \label{defn:colored-ice-flower-systems}
$^*$ A \emph{colored ice-flower system} $\textbf{\textrm{K}}^c$ is defined as: Each star $K^c_{1,m_j}$ of a colored ice-flower system $\textbf{\textrm{K}}^c=(K^c_{1,m_j})^n_{j=1}$ admits a proper vertex coloring $g_j$ such that $g_j(x)\neq g_j(y)$ for any pair of vertices $x,y$ of $K^c_{1,m_j}$. If $g_i(y_{i,r})=g_j(x_{j})$ and $g_i(x_i)=g_j(y_{j,s})$ for two colored stars $K^c_{1,m_{i}}$ and $K^c_{1,m_{j}}$ for $i\neq j$, so we have $x_ix_{j}=x_iy_{i,r}\overline{\ominus} x_{j}y_{j,s}$ in $K^c_{1,m_{i}}\overline{\ominus} K^c_{1,m_{j}}$ admitting a proper vertex coloring $g_{i,j}=g_i\overline{\ominus}g_j$. If each graph $T=\overline{\ominus}|^n_{j=1}b_jK^c_{1,m_j}$ for $b_j\in Z^0$ has not any leaf of each star $K^c_{1,m_{i}}$ of the star base $\textbf{\textrm{K}}^c$, then we call it a \emph{colored leaf-coincided graph} and call $\textbf{\textrm{K}}^c$ a \emph{colored ice-flower system}.\qqed
\end{defn}

For an uncolored ice-flower system $\textbf{\textrm{K}}$, we call the following set
\begin{equation}\label{eqa:55-star-graphic-lattice}
\textbf{\textrm{L}}(\overline{\ominus} (m)\textbf{\textrm{K}})=\big \{\overline{\ominus}|^n_{j=1}a_iK_{1,m_j},~a_i\in Z^0,~K_{1,m_j}\in \textbf{\textrm{K}}\big \}
\end{equation} a \emph{star-graphic lattice}; and moreover the following set
\begin{equation}\label{eqa:55-colored-star-graphic-lattice}
\textbf{\textrm{L}}(Z^0\overline{\ominus}\textbf{\textrm{K}}^c)=\big \{\overline{\ominus}|^n_{j=1}b_jK^c_{1,m_j},~b_j\in Z^0,~K^c_{1,m_j}\in \textbf{\textrm{K}}^c\big \}
\end{equation} is called a \emph{colored star-graphic lattice}, where each graph $T=\overline{\ominus}|^n_{j=1}b_jK^c_{1,m_j}$ admits a proper vertex coloring defined by $\overline{\ominus}|^n_{j=1}b_jg_j$.

\begin{rem}\label{rem:spanning-tree-finite-lattice}
In general, we point that ``$\overline{\ominus}| ^n_{i=1}K_{1,m_{i}}$'' defined in Definition \ref{defn:uncolored-ice-flower-systems} produces a set of graphs with the same degree-sequence $(m_1,m_2, \dots , m_n)$ and vertex set $\{x_1,x_2,\dots ,x_n\}$, where each vertex $x_i$ is the center of the star $K_{1,m_{i}}\in \textbf{\textrm{K}}$ for $i\in [1,n]$.

If $[1,m]\subset \{m_1,m_2, \dots , m_n\}$ defined in Definition \ref{defn:colored-ice-flower-systems}, then there is a subset $S_{pt}(m)\subset \textbf{\textrm{L}}(Z^0\overline{\ominus}\textbf{\textrm{K}}^c)$ containing colored leaf-coincided graphs of $m$ vertices. If each colored leaf-coincided graph $T\in S_{pt}(m)$:

(1) is connected;

(2) holds $n_1(T)=2+\Sigma_{d\geq 3}(d-2)n_d(T)$ (refer to Theorem \ref{thm:Yao-tree-if-and-only-if}), where $n_d(T)$ is the number of vertices of $d$ degree in $T$ (Ref. Theorem \ref{thm:Yao-tree-if-and-only-if}); and

(3) admits a proper vertex coloring $f=\overline{\ominus}|^n_{j=1}g_j$ with $f(u)\neq f(w)$ for distinct vertices $u,w\in V(T)$.

We call $S_{pt}(m)$ a \emph{spanning-tree finite lattice}, since $|S_{pt}(m)|=m^{m-2}$.\paralled
\end{rem}

\subsubsection{Edge-asymmetric trees}

\begin{defn} \label{defn:totally-graceful-labeling}
\cite{Subbiah-Pandimadevi-Chithra2015} Let $G$ be a $(p,q)$-graph. If there exist a bijection $f:V(G)\cup E(G)\rightarrow[1, p+q]$ such that$f(uv)=|f(u)-f(v)|$ for every edge $uv\in E(G)$, then we say $f$ a \emph{totally graceful labeling} of $G$. Moreover, if $f(E(G))=[1, q]$, we call $f$ a \emph{super totally graceful labeling}, if $G$ is a bipartite graph with bipartition $(X, Y)$, and $f(X)<f(Y)$, we call $f$ a \emph{set-ordered totally graceful labeling}; moreover if $f(E(G))=[1, q]$ and $G$ is a bipartite graph with bipartition $(X, Y)$ such that $f(X)<f(Y)$, we call $f$ a \emph{super set-ordered totally graceful labeling}.\qqed
\end{defn}

In \cite{Wang-Xu-Yao-2016}, the authors show the following tree constructions and some results on trees.

(1) An \emph{edge-symmetric graph} $G$ is a connected graph having a proper subgraph $T\subset G$ such that the graph $G-E(T)$ is not connected and has exactly $m~(\geq 2)$ components $H_1,H_2,\dots , H_m$ with $H_i\cong H$ for a given graph $H$ and $i\in [1,m]$, write $G=\langle T; H,m\rangle$.

(2) An \emph{edge-asymmetric graph} $G$ is a connected graph having a proper subgraph $T\subset G$ such that the remainder $G-E(T)$ is not connected and has exactly $m~(\geq 2)$ components $H_1,H_2,\dots, H_m$ holding $H_i\not \cong H_j$ for some two components $H_i$ and $H_j$; $|V(H_i)|=a$ and $|E(H_i)|=b$ for two fixed integers $a,b>0$ and $i\in [1,m]$, denoted as $G=\langle T;H_1,H_2,\dots, H_m\rangle$.

(3) Let $T$ be a connected graph having vertices $x_1,x_2,\dots,x_p$, and let $H_i$ be a copy of a connected graph $H$ for $i\in [1,p]$. A fixed vertex $u_0\in V(H)$ corresponds to its image vertex $u_{i,0}\in V(H_i)$ for $i\in[1,p]$. We identify each vertex $u_{i,0}$ with the vertex $x_i\in V(T)$ into one vertex for $i\in [1,p]$, and then obtain a particularly edge-symmetric graph $G$, called a \emph{uniformly edge-symmetric graph} and denoted as $G=\langle T;H,p\rangle_{uni}$.

\begin{lem}\label{thm:lemma-1}
(i) \cite{Wang-Xu-Yao-2016} A tree is graceful if and only if it admits a super totally graceful labeling.

(ii) \cite{Wang-Xu-Yao-2016} A tree is set-ordered graceful if and only if it admits a super set-ordered totally graceful labeling.
\end{lem}

\begin{thm} \label{thm:symmetric-lemma}
\cite{Wang-Xu-Yao-2016} Every uniformly edge-symmetric graph $G=\langle K_{2};T, 2\rangle_{uni}$ admits a super set-ordered totally graceful labeling if the tree $T$ admits a graceful labeling.
\end{thm}

\begin{thm} \label{thm:m-distinct-trees}
\cite{Wang-Xu-Yao-2016} Suppose that $T_{1},T_{2},\dots, T_m$ are vertex disjoint trees, and each tree $T_{i}$ admits a set-ordered totally graceful labeling for $i\in [1,m-1]$, and the tree $T_m$ admits a graceful labeling. Then there are vertices $u_i\in
V(T_i)$ with $i\in [1,m]$ such that joining $u_j\in V(T_j)$ and $u_{j+1}\in V(T_{j+1})$ for $j\in [1,m-1]$ produces a new tree $T$ admitting a super totally graceful labeling.
\end{thm}

\begin{thm} \label{thm:main-theorem}
\cite{Wang-Xu-Yao-2016} Suppose that $T$ and $H_1,H_2,\dots, H_p$ are mutually vertex disjoint trees, where $p=|V(T)|$. Let $A=\lceil\frac{p-1}{2}\rceil$, $B=\lceil\frac{p+1}{2}\rceil$ and $C=\lceil\frac{p+3}{2}\rceil$. If the tree $T$ admits a graceful labeling, and each tree $H_k$ with $k\in[1,A]\cup[C, p]$ admits a set-ordered graceful labeling $f_k$ and its own bipartition $(X_k,Y_k)$ holding $X_k=\{u_{k,i}:~ i\in [1,s]\}$ and
$Y_k=\{v_{k,j}:~ j\in [1,t]\}$ for two fixed integers $s,t\geq 1$ and $s+t=n=|V(H_{B})|$. Then edge-asymmetric tree $G=\langle T;H_1,H_2,\dots, H_{p}\rangle$ admits a super totally graceful labeling.
\end{thm}

\begin{cor} \label{thm:corollary}
\cite{Wang-Xu-Yao-2016} Suppose that $T$ and $H_1,H_2,\dots, H_p$ are mutually vertex disjoint trees, where $p=|V(T)|$. If the tree $T$ admits a set-ordered totally graceful labeling and each tree $H_k$ with $k\in[1, p]$ admits a set-ordered totally graceful labeling and its own bipartition $(X_k,Y_k)$ holding $X_k=\{u_{k,i}:~ i\in [1,s]\}$ and $Y_k=\{v_{k,j}:~ j\in [1,t]\}$ for two fixed integers $s,t\geq 1$. Then the edge-asymmetric tree $\langle T;H_1,H_2,\dots, H_{p}\rangle$ admits a super set-ordered totally graceful labeling.
\end{cor}
\subsubsection{Hanzi-trees, Self-similar Hanzi-trees}

\begin{defn} \label{defn:self-similar-Hanzi-trees}
$^*$ If a tree $T$ can be vertex-split (resp. edge-split) into several Hanzi-graphs $H_1,H_2,\dots ,H_n$, then we call $T$ a \emph{Hanzi-tree}, and moreover we call $T$ a \emph{self-similar Hanzi-tree} if these Hanzi-graphs $H_1,H_2,\dots ,H_n$ hold $H_i\cong H_j$ true for $1\leq i,j\leq n$.\qqed
\end{defn}

\begin{figure}[h]
\centering
\includegraphics[width=16.4cm]{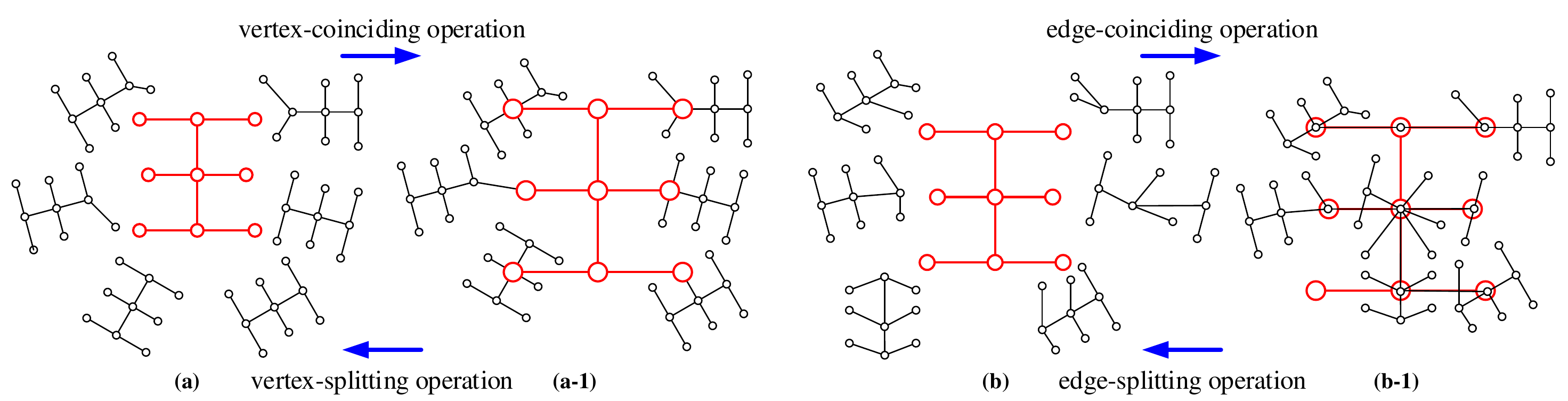}\\
\caption{\label{fig:Hanzi-tree-operation}{\small Four operations on Hanzi-trees: The vertex-coinciding operation from (a) to (a-1); the vertex-splitting operation from (a-1) to (a); the edge-coinciding operation from (b) to (b-1); and the edge-splitting operation from (b-1) to (b).}}
\end{figure}

Four operations on Hanzi-trees shown in Fig.\ref{fig:Hanzi-tree-operation} are for understanding Definition \ref{defn:self-similar-Hanzi-trees}. Some self-similar Hanzi-trees are shown in Fig.\ref{fig:Hanzi-tree-operation}, Fig.\ref{fig:tree-vertex-coincid-self} and Fig.\ref{fig:tree-edge-coincid-self}. And three self-similar Hanzi-trees $N(1)$, $N(2)$ and $N(3)$ shown in Fig.\ref{fig:tree-vertex-coincid-self} are the first three items of a \emph{self-similar network} $N(t)$, and each $N(t)$ is obtained by adding some copies of $N(1)$ to $N(t-1)$ using the vertex-coinciding operation. Similarly, another group of three self-similar Hanzi-trees $L(1)$, $L(2)$ and $L(3)$ shown in Fig.\ref{fig:tree-edge-coincid-self} are the first three items of a \emph{self-similar network} $L(t)$. Clearly, two self-similar networks $N(t)$ and $L(t)$ are \emph{random} since there are no fixed way and no fixed number of copies of $N(1)$ or $L(1)$ added to $N(t-1)$ or $L(t-1)$ at time step $t$.

Hanzi-trees admit some colorings and labelings shown in Theorem \ref{thm:Yao-tree-admits-colorings-group-11} and Theorem \ref{thm:Yao-tree-admits-colorings-group-22}, and moreover there are polynomial algorithms for coloring Hanzi-trees with particular graph colorings and graph labelings for the application of topological authentication.

\begin{figure}[h]
\centering
\includegraphics[width=16.4cm]{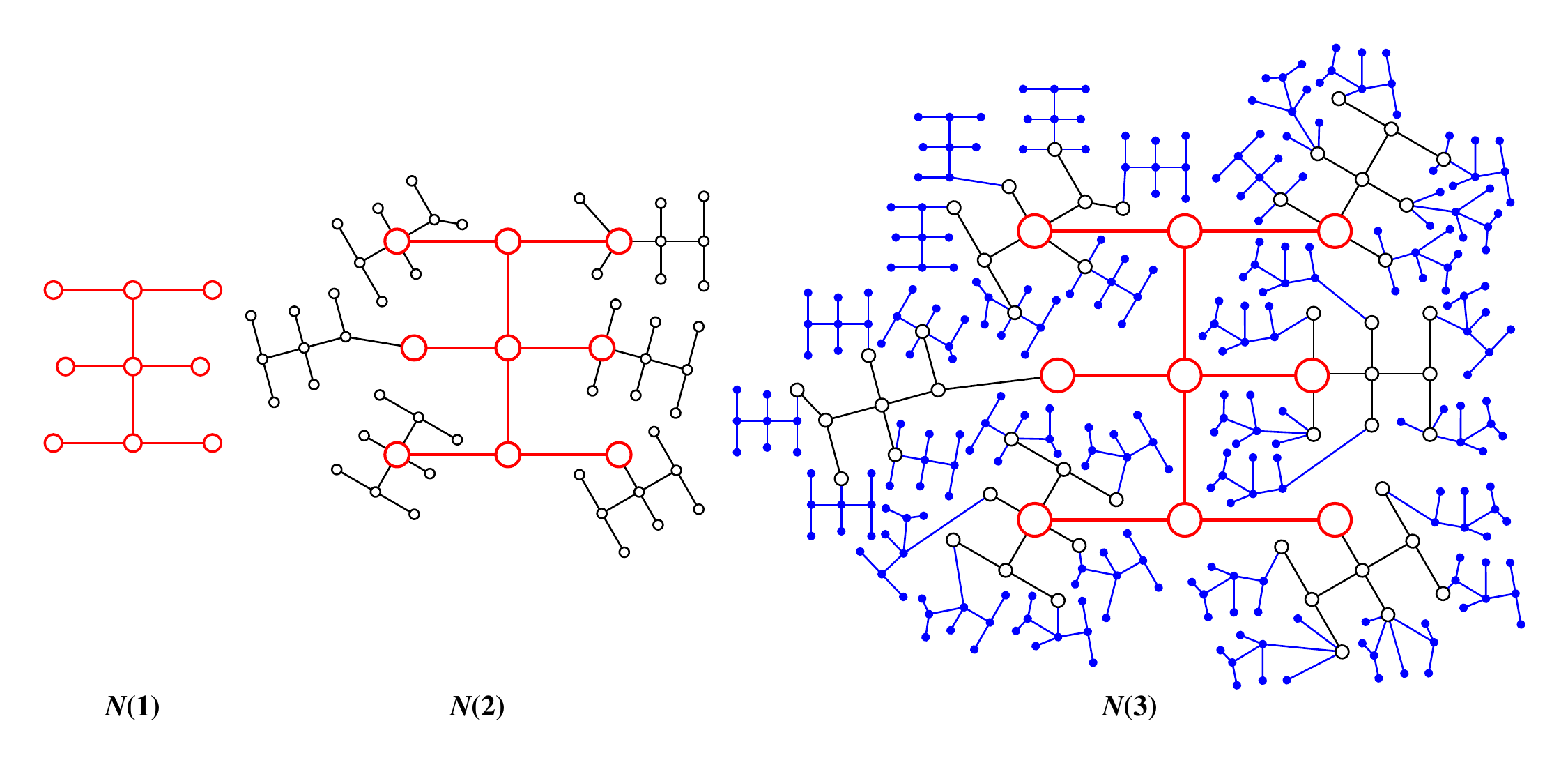}\\
\caption{\label{fig:tree-vertex-coincid-self}{\small Three self-similar Hanzi-trees $N(1)$, $N(2)$ and $N(3)$ are from a \emph{self-similar network} $N(t)$ as $t=1,2,3$.}}
\end{figure}

\begin{figure}[h]
\centering
\includegraphics[width=16.4cm]{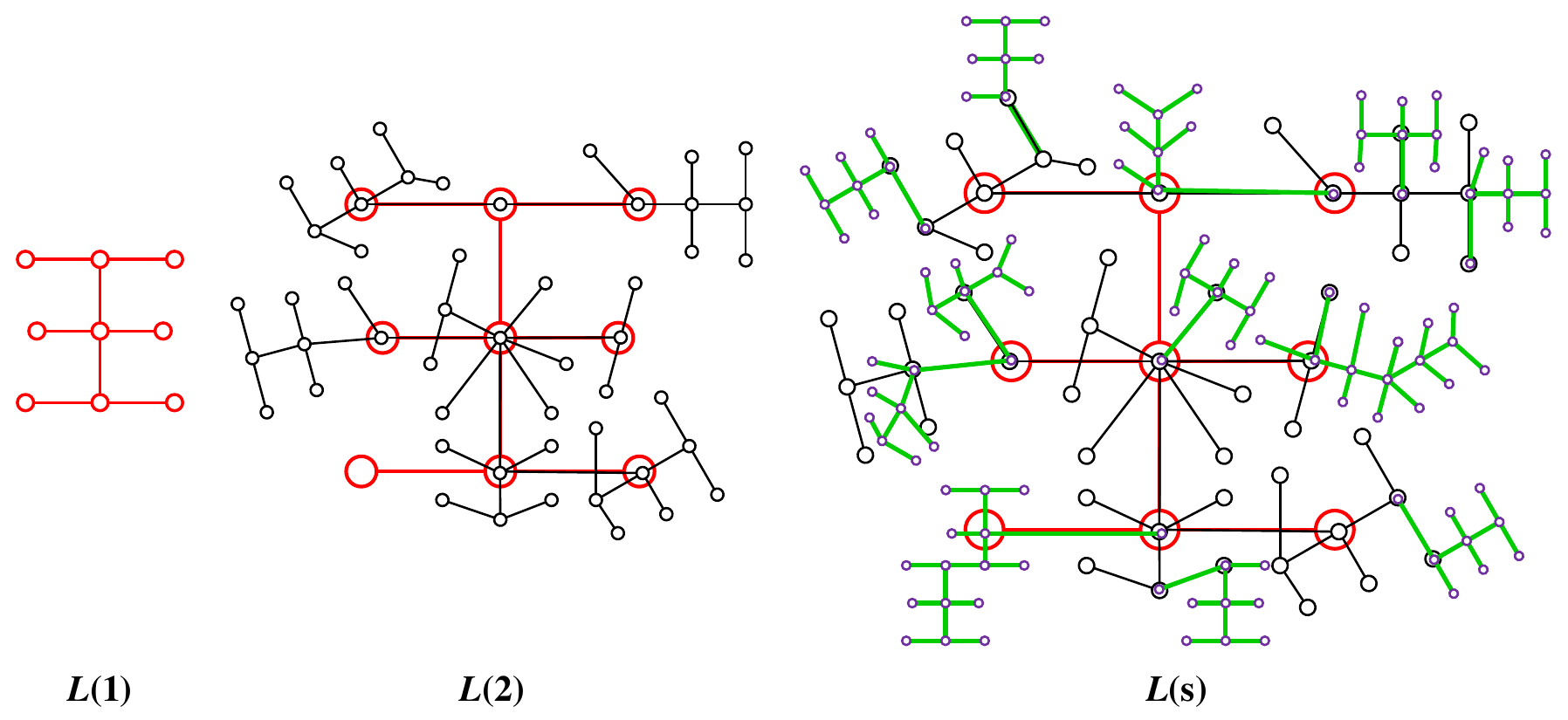}\\
\caption{\label{fig:tree-edge-coincid-self}{\small Three self-similar Hanzi-trees $L(1)$, $L(2)$ and $L(3)$ are from a \emph{self-similar network} $L(t)$ as $t=1,2,3$.}}
\end{figure}

\subsubsection{Trees obtained from connected graphs}

A connected $(p,q)$-graph $G$ can be vertex-split into trees $T_1,T_2,\dots, T_m$ of $q+1$ vertices. Suppose that each tree $T_i$ admits different proper total colorings $f_{i,1},f_{i,2},\dots, f_{i,a_i}$ with $i\in [1,m]$, and let $T_{i,j}$ be the tree obtained by coloring totally the vertices and edges of the each tree $T_i$ with the proper total coloring $f_{i,j}$ with $j\in [1,a_i]$ and $i\in [1,m]$. So, vertex-coinciding each colored tree $T_{i,j}$ produces the original connected $(p,q)$-graph $G$ admitting a total coloring $F_{i,j}$, there are the following cases:

(i) $F_{i,j}(uv)$ for each edge $uv\in E(G)$ is a number, however it may happen $F_{i,j}(uv)=F_{i,j}(uw)$ for $v,w\in N(u)$ for some vertex $u\in V(G)$;

(ii) $F_{i,j}(x)$ for each vertex $x\in V(G)$ is a set, since it may happen that $x$ is the result of vertex-coinciding two or more vertices of the tree $T_{i,j}$, that is, $x=x_{i,1}\odot x_{i,2}\odot \cdots \odot x_{i,b}$ for $x_{i,s}\in V(T_{i,j})$.

(iii) \textbf{Good-property}: $F_{i,j}(V(G)\cup E(G))=[1,\chi\,''(G)]$, where $\chi\,''(G)$ is the total chromatic number of the connected $(p,q)$-graph $G$; $|F_{i,j}(y)|=1$ for each vertex $y\in V(G)$; and $F_{i,j}(xy)\neq F_{i,j}(xz)$ for $y,z\in N(x)$ for each vertex $x\in V(G)$.

Let $V_{\textrm{sp-co}}(G)=\{T_{i,j}:j\in [1,a_i],i\in [1,m]\}$. Each totally colored tree $T_{i,j}\in V_{\textrm{sp-co}}(G)$ is \emph{colored graph homomorphism} to $G$, also, $T_{i,j}\rightarrow G$.

\begin{problem}\label{qeu:444444}
\textbf{Find} totally colored trees $T_{i,j}$ from $V_{\textrm{sp-co}}(G)$, such that $T_{i,j}$ is colored graph homomorphism to $G$ admitting a total coloring $F_{i,j}$, which has the Good-property.
\end{problem}

We get a topological authentication: The connected $(p,q)$-graph $G$ is as a \emph{public-key}, then there are many \emph{private-keys} in $V_{\textrm{sp-co}}(G)$, and each private-key $T_{i,j}\in V_{\textrm{sp-co}}(G)$ is colored graph homomorphism to $G$, such that $G$ admits a total coloring $F_{i,j}$, which produces a Topcode-matrix $T_{code}(F_{i,j},G)$ for making number-based strings.

\subsubsection{Conjectures and problems on trees}

Despite the simplicity and linearity of trees, there are a lot of conjectures and difficult problems on trees, which make colorings and labelings of trees to become available techniques of information security in the quantum age.

\begin{conj} \label{conj:famois-conjecture-trees}
There are many open problems related with trees. The following labelings' definitions can be found in \cite{Gallian2021}:
\begin{asparaenum}[\textbf{Conj-}1.]
\item (Erd\"{o}s-S\'{o}s Conjecture, 1963) Given a tree $T_n$ on $n$ vertices, then every $(p,q)$-graph $G_p$ with $q> \frac{1}{2}p(n-2)$ edges contains $T_n$. In general, every graph of average degree grater than $k-1$ contains every tree of order $k+1$.
\item (Gerhard Ringel and Anton Kotzig, 1963; Alexander Rosa, 1967) $K_{2n+1}$ can be decomposed into $2n+1$ subgraphs which are isomorphic with a given tree with $n$ edges.
\item (Alexander Rosa, 1966) Each tree is graceful.
\item (Bermond, 1979) Every lobster is graceful.
\item (Truszczy\'{n}ski, 1984) All connected unicyclic graphs except $C_n$, $n = 1$ or $2$ $(\emph{mod}\ 4)$, are graceful.
\item (H. J. Broersma and C. Hoede, 1999) Every tree containing a perfect matching is strongly graceful.
\item (R. L. Graham and N. J. A. Sloane, 1980) Every nontrivial tree is harmonious.
\item (R.B. Gnanajothi, 1991) Any tree is odd-graceful.
\item (Cahit, 1966) Every tree is $k$-equitable for any $k$.
\item (Sin-Min Lee, 1989) All trees of odd order are edge-graceful.
\item (R. H\"{a}ggkvist, 1989) Every tree on $m+1$ vertices decomposes $K_{rm+1}$ for every natural number $r\geq 2$ provided that $r$ and $m+1$ are not both odd.
\item (R. H\"{a}ggkvist, 1989) Every $(m+1)$-order tree decomposes every $2m$-regular graph.
\item (R. H\"{a}ggkvist, 1989) Every $(m+1)$-order tree decomposes every $m$-regular bipartite graph.
\item (Anton Kotzig and Alex Rosa, 1970) Every tree admits an edge-magic total labeling.
\item (H. Enomoto, A. S. Llado, T. Nakamigawa, and G. Ringel, 1998) Every tree admits a super edge-magic total labeling.
\item (N. Hartsfield and G. Ringel, 1990) Every simple connected graph, other than $K_2$, is antimagic.
\item (Graham and Haggkvist, 1988) Every tree with $n$ edges decomposes any bipartite $n$-regular graph.
\item (Graham and Sloane, 1980) Every tree is felicitous.
\item A wheel $W_n$ admits an edge-magic total labeling if $n\neq 4k$ for any integer $k>0$.
\item Each tree admits a super $(k,\lambda)$-magic labeling for some $k$ and $\lambda$.
\item Every tree is $1$-dimension $(s,t)$-magical graph.
\item For $n\geq 3$ and $m\geq 4$, $P_m\cup C_{n}$ is felicitous.
\item All trees have balanced $(k,d)$-graceful labelings for some values of $k$ and $d$.
\item Each tree is $k$-graceful and $(k,d)$-graceful for some $k,d>1$.
\end{asparaenum}
\end{conj}

\begin{conj} \label{conj:Bing-Yao-conjectures}
The following open problems can be found in \cite{Yao-Wang-2106-15254v1}:
\begin{asparaenum}[\textbf{Ycon-}1.]
\item \cite{Yao-Chen-Yao-Cheng2013JCMCC} \textbf{A conjecture on $(k,\lambda _k)$-magically total labellins of paths}. For each $k\in [2,2(n-1)]$, each path $P_n$ of $n$ vertices admits a magic total labeling $f_k:V(P_n)\cup E(P_n)\rightarrow [1,2n-1]$ such that $f_k(u)+f_k(v)=k+f_k(uv)$ for each edge $uv\in E(P_n)$. For example, a path on $n$ vertices is denoted as $P_n=u_1e_1u_2e_2\cdots u_{n-1}e_{n-1}u_n$, and we color each vertex $u_i$ of $P_n$ with a small circle ${\Large \textcircled{\small $m$}}$ colored with a color $m$ and we color each edge $e_i$ of $P_n$ with a number, respectively. For a path $P_7$, $\lambda _k=1$ and each integer $k\in [2,12]$, we have:

\begin{tabular}{ll}
\quad $k=2$:\quad ${\Large \textcircled{\small 1}}2{\Large
\textcircled{\small 3}}10{\Large \textcircled{\small 9}}13{\Large
\textcircled{\small 6}}12{\Large \textcircled{\small 8}}11{\Large
\textcircled{\small 5}}7{\Large \textcircled{\small 4}}$,\quad
&$k=3$:\quad ${\Large \textcircled{\small 5}}3{\Large
\textcircled{\small 1}}2{\Large \textcircled{\small 4}}10{\Large
\textcircled{\small 9}}13{\Large \textcircled{\small 7}}12{\Large
\textcircled{\small 8}}11{\Large \textcircled{\small 6}}$\\[4pt]
\quad $k=4$:\quad ${\Large \textcircled{\small 7}}12{\Large
\textcircled{\small 9}}13{\Large \textcircled{\small 8}}6{\Large
\textcircled{\small 2}}1{\Large \textcircled{\small 3}}4{\Large
\textcircled{\small 5}}11{\Large \textcircled{\small 10}}$,\quad
&$k=5$:\quad ${\Large \textcircled{\small 1}}5{\Large
\textcircled{\small 9}}11{\Large \textcircled{\small 7}}12{\Large
\textcircled{\small 10}}13{\Large \textcircled{\small 8}}6{\Large
\textcircled{\small 3}}2{\Large \textcircled{\small 4}}$\\[4pt]
\quad $k=6$:\quad ${\Large \textcircled{\small 3}}2{\Large
\textcircled{\small 5}}7{\Large \textcircled{\small 8}}12{\Large
\textcircled{\small 10}}13{\Large \textcircled{\small 9}}4{\Large
\textcircled{\small 1}}6{\Large \textcircled{\small 11}}$,\quad
&$k=7$:\quad ${\Large \textcircled{\small 3}}6{\Large
\textcircled{\small 10}}12{\Large \textcircled{\small 9}}4{\Large
\textcircled{\small 2}}8{\Large \textcircled{\small 13}}7{\Large
\textcircled{\small 1}}5{\Large \textcircled{\small 11}}$\\[4pt]
\quad $k=8$:\quad ${\Large \textcircled{\small 4}}8{\Large
\textcircled{\small 12}}7{\Large \textcircled{\small 3}}5{\Large
\textcircled{\small 10}}11{\Large \textcircled{\small 9}}2{\Large
\textcircled{\small 1}}6{\Large \textcircled{\small 13}}$,\quad
&$k=9$:\quad ${\Large \textcircled{\small 6}}7{\Large
\textcircled{\small 10}}9{\Large \textcircled{\small 8}}12{\Large
\textcircled{\small 13}}5{\Large \textcircled{\small 1}}3{\Large
\textcircled{\small 11}}4{\Large \textcircled{\small 2}}$\\[4pt]
\quad $k=10$:\quad ${\Large \textcircled{\small 2}}5{\Large
\textcircled{\small 13}}4{\Large \textcircled{\small 1}}3{\Large
\textcircled{\small 12}}8{\Large \textcircled{\small 6}}7{\Large
\textcircled{\small 11}}10{\Large \textcircled{\small 9}}$
&$k=11$:\quad ${\Large \textcircled{\small 4}}5{\Large
\textcircled{\small 12}}2{\Large \textcircled{\small 1}}3{\Large
\textcircled{\small 13}}11{\Large \textcircled{\small 9}}6{\Large
\textcircled{\small 8}}7{\Large \textcircled{\small 10}}$\\[4pt]
\quad $k=12$:\quad ${\Large \textcircled{\small 4}}3{\Large
\textcircled{\small 11}}12{\Large \textcircled{\small 13}}9{\Large
\textcircled{\small 8}}2{\Large \textcircled{\small 6}}1{\Large
\textcircled{\small 7}}5{\Large \textcircled{\small 10}}$
\end{tabular}
\item \cite{Yao-Chen-Yao-Cheng2013JCMCC} Every tree admits a super $(k,\lambda)$-magically total labeling for some integers $k$ and $\lambda\neq 0$
\item $^*$ A tree $T$ with $n$ vertices contains a certain largest matching $M$ and admits a graceful labeling $f$ such that $f(u)+f(v)=n-1$ for each matching edge $uv\in M$.
\item $^*$ For a bipartite connected graph $G$ admitting an odd-elegant labeling $g$, there exists a group of connected graphs $H_k$ admitting a labeling $h_k$ for $k\in [1,r]$, such that
$$g(V(G))\cup h_1(V(H_1))\cup h_2(V(H_2))\cup \cdots \cup h_r(V(H_r))=[0,M]
$$ with
$$
M\leq 2\max\{|E(G)|,|E(H_1)|,|E(H_2)|,\dots ,|E(H_r)|\}
$$ and $g(E(G))=[1,2|E(G)|-1]^o$, and moreover, for $k\in [1,r]$, we have edge color sets
$$
h_k(E(H_k))=\{h_k(xy)=h_k(x)+h_k(y)~(\bmod~2|E(H_k)|):xy\in E(H_k)\}=[1,2|E(H_k)|-1]^o
$$
\item $^*$ Each tree $T$ admits \emph{$0$-rotatable $(k,d)$-gracefully labelings} with $d>k$ for odd $k\geq 1$ and even $d\geq 2$, such that any vertex $u\in V(T)$ is colored with $f(u)=0$ by a $(k,d)$-gracefully labeling $f$ of $T$. (see examples shown in Fig.\ref{fig:rotatable-k-d-gracefully})
\item $^*$ All every-zero graphic groups $\{F(S_k);\oplus \}$ made by the spanning trees of $K_n$ holds $\bigcup _{k}\{F(S_k);\oplus \}=S_{\textrm{pan-tree}}(K_n)$, where $S_{\textrm{pan-tree}}(K_n)$ is the set of all spanning trees of $K_n$.
\item \textbf{A conjecture on the v-set e-proper labeling.} Each connected graph with no multiple edges and self-loops admits a v-set e-proper (odd-)graceful labeling.
\item $^*$ \textbf{A coloring expression of Conjecture \ref{conj:c2-KT-conjecture}}. For any tree group $S_{plit}=\{T_{2},T_{3}$, $\dots $, $T_{n}\}$ defined by vertex-disjoint trees $T_{k}$ of $k$ vertices with $k\in [2,n]$, then each tree $T_{k}$ admits a proper vertex coloring $f_k: V(T_{k})\rightarrow [1,n]$ holding $f_k(x)\neq f_k(y)$ for distinct $x,y\in V(T_{k})$, such that doing the vertex-coinciding operation to the vertices (colored with the same color) of trees of the tree group $S_{plit}$ produces just a complete graph $K_n$, we write this case as $\langle S_{plit}\mid K_n\rangle$.
\end{asparaenum}
\end{conj}

\begin{problem}\label{qeu:444444}
How many trees on $p$ vertices have exactly $k$ leaves? \textbf{Determine} all non-isomorphic spanning trees of a connected graph $G$.
\end{problem}

\begin{problem} \label{prob:444444}
In 2003, Leizhen Cai, in his paper ``\emph{The complexity of the locally connected spanning tree problem}'', defined that a \textbf{locally connected spanning tree} $T$ of a graph $G$ is a spanning tree of $G$ with the following property: for every vertex, its neighborhood in $T$ induces a connected subgraph in $G$. The existence of such a spanning tree in a network ensures, in case of site and line failures, effective communication amongst operative sites as long as these failures are isolated. The problem of determining whether a graph contains a locally connected spanning tree is NP-complete, even when input graphs are restricted to planar graphs or split graphs.
\end{problem}

\begin{problem} \label{prob:c2-decomposition-problem-smaller-tree}
Let $\text{tree}(G)$ be the minimum number of edge-disjoint trees $T_i$ of $G$ such that $E(G)=\bigcup ^{\text{tree}(G)}E(T_i)$. \textbf{Determine} $\text{tree}(G)$ for each connected graph $G$. If we restrict us for looking for edge-disjoint paths or cycles rather than trees, thus, $\text{tree}(G)=pc(G)$.
\end{problem}

\begin{problem} \label{prob:c2-decomposition-problem-tree-sequence}
If a graph $G$ can be decomposed into $m$ prescribed edge-disjoint trees $T_k$ of $k\ (k\in [2, m])$ vertices such that $G=\bigcup ^{m}_{k=2}T_k$, then what structure \textbf{does} $G$ have? We call this problem the \emph{\textbf{decomposition problem of tree sequence}}. We have a weaker form $G=\bigcup ^{m}_{k=2}a_kT_k$ for integers $a_k\geq 0$ $(k\in [2, m])$.
\end{problem}

\begin{problem} \label{prob:c2-decomposition-problem-tree-sequence}
Given $m$ edge-disjoint trees $T_k$ of $k\ (2\leq k\leq m)$ vertices we have a resulting graph $G$ by vertex-coinciding these trees only on vertices such that $G$ has no parallel edges and $|E(G)|=\frac{1}{2}m(m-1)$, of course, $G$ is required to be connected. It is not hard to estimate $m\leq |V(G)|\leq \frac{1}{2}m(m+1)-1-(m-1)$, called the \emph{\textbf{assembling tree sequence problem}}. Show characterizations of this graph $G$.
\end{problem}

\begin{problem} \label{prob:box}
A graph $G$ on $n$ vertices is called \textbf{arbitrarily vertex decomposable} if for each sequence $(n_1,n_2,\dots, n_k)$ of positive integers such that $n_1 +n_2 +\cdots+n_k = n$, there exists a partition $(V_1, V_2,\dots, V_k)$ of $V(G)$ such that for each $i\in [1,k]$ then $V_i$ induces a connected subgraph of $G$ on $n_i$ vertices. Obviously, a complete graph $K_n$ is arbitrarily vertex decomposable. The problem of deciding whether a given graph is arbitrarily vertex decomposable is NP-complete (M. Robson, 1998). However, each traceable graph is arbitrarily vertex decomposable. It is also clear that a connected graph is arbitrarily vertex decomposable if its spanning tree is arbitrarily vertex decomposable.
\end{problem}

\begin{problem} \label{thm:newproblems1}
\textbf{ST-balance set}. Let $G$ be a connected graph, and let $L(G)$ be the set of all leaves of $G$. A subset $S$ of $V(G)$ is called a \textbf{spanning tree balance set} of $G$ if for any two spanning trees $T$ and $T\,'$ of $G$ the following identity
\begin{equation}\label{eqa:chapter2-controlset00}
n_1(T)-|L(T)\cap S|=n_1(T\,')-|L(T\,')\cap S|
\end{equation} holds true, where $n_1(T)$ is the number of leaves of $T$. \textbf{Determine} each smallest ST-balance set $S^*$ of $G$, such that $|S^*|\leq |S|$ for any ST-balance set of $G$.
\end{problem}

\begin{problem} \label{prob:c2bbb}
It has been proved that if a graph $G$ contains $k$ edge-disjoint spanning trees if and only if for each $n$-partition $(V_i)^n_{i=1}=(V_1,V_2,\dots,V_n)$ of $V$, the number of edges which have ends in different parts of the partition $(V_i)^n_{i=1}$ is at least $k(n-1)$. \textbf{What} properties does $G$ have such that $G$ contains a spanning caterpillar (or a spanning spider, or a spanning lobster, or a spanning banana, or a spanning firecracker)?
\end{problem}

\begin{problem} \label{prob:c2bbb}
\textbf{Determine} a largest caterpillar $T\,^*$ in a connected graph $G$ such that $|V(T\,^*)|\geq |V(T\,')|$ for any caterpillar $T\,'$ of $G$.
\end{problem}

\begin{problem} \label{prob:c2bbb}
\textbf{What} properties does a graph $G$ have such that it contains just $k\ (>1)$ edge-disjoint spanning caterpillars, or spanning spiders, or spanning lobster?
\end{problem}

\begin{problem} \label{prob:c2bbb}
\textbf{Caterpillar-covering Problem}. Let $F=(T_1,T_2,\dots , T_k)$ be a caterpillar sequence of a connected graph $G$ such that any $u\in V(G)$ is in a certain caterpillar $T_i\in F$. For any $k$-set $S$ of $V(G)$, \textbf{is} there $S\subseteq V(T_i)$ for a certain caterpillar $T_i\in F$?
\end{problem}

\begin{problem} \label{prob:box}
A \textbf{panconnected graph} $G$ of order $n$ satisfies that for any pair of vertices $u$ and $v$ then $G$ contains all paths of lengths from $\text{d} _G(u,v)$ to $n-1$, \textbf{characterize} panconnected graphs.
\end{problem}

\begin{problem} \label{prob:box}
\textbf{Decompose} the edge set of a graph $G$ into the union of the edge sets of $k$ edge-disjoint paths (or cycles), of course, need not require that each path must be a Hamilton path. Notice that the maximum $k=|E(G)|$, it is wanted to determine the minimum $k$. This is also called a \textbf{path decomposition}, or \textbf{cycle decomposition}.
\end{problem}

\subsection{Planar graphs}

\subsubsection{Vertex colorings of planar graphs}

\begin{defn}\label{defn:planar-graphs-color-sets-4-operations}
\cite{Yao-Sun-Wang-Su-Maximal-Planar-Graphs-2021} Suppose that a planar graph $T$ admits a proper vertex $4$-coloring $g$, so its vertex set $V(T)$ can be divided into four subsets, namely, $V(T)=\bigcup^4_{k=1} V_k(T)$ such that each vertex $x_{k,j}\in V_k(T)$ is colored with color $k_j$ for $j\in [1,m_k]$ and $k\in [1,4]$, where $m_k=|V_k(T)|$. So we have four vertex color sets $C_k(T)=\{g(x_{k,j}):x_{k,j}\in V_k(T)\}=\{k_1,k_2,\dots ,k_{m_k}\}$ for $k\in [1,4]$. And for two ends of each edge $uv\in E(G)$ colored with $g(u)=a_i$ and $g(v)=b_j$, we define the edge color $g(uv)$ with one of \textbf{addition} $a_i(+)b_j=(a+b)_{(i+j)}$, \textbf{multiplication} $a_i(\cdot ) b_j=ab_{ij}$ and \textbf{subtraction} $a_i(-)b_j=|a-b|_{|i-j|}$.\qqed
\end{defn}

For example, in Fig.\ref{fig:mpg-4-color-authen}(c-1), $V(T_2)=\bigcup^4_{i=1} V_i(T_2)$, then vertex color sets $C_1(T_2)=\{1_1,1_2,1_3,1_4\}$, $C_2(T_2)=\{2_1,2_2,2_3\}$, $C_3(T_2)=\{3_1,3_2,3_3\}$ and $C_4(T_2)=\{4_1,4_2,4_3,4_4\}$. By the addition, multiplication and subtraction defined in Definition \ref{defn:planar-graphs-color-sets-4-operations}, the planar graph $H_{1,1}$ shown in Fig.\ref{fig:4-colors-edge-coloring} has three Topcode-matrices $M_{1,1}$, $M^+_{1,1}$, $M^{\times}_{1,1}$ and $M^{|-|}_{1,1}$ shown in Eq.(\ref{eqa:addition-times-absolut-matrix}), we get four number-based strings

\begin{equation}\label{eqa:order-four-strings-matrix}
\centering
{
\begin{split}
\begin{array}{ll}
&s_{1,1}(24)=144243121232421313433214\\
&s^+_{1,1}(36)=144243121232464455565556421313433214\\
&s^{\times}_{1,1}(36)=144243121232383446494648421313433214\\
&s^{|-|}_{1,1}(36)=144243121232 222031303132 421313433214
\end{array}
\end{split}}
\end{equation}
by a fixed regular algorithm. Next, we have another group of four no-ordered number-based strings below

\begin{equation}\label{eqa:no-order-four-strings-matrix}
\centering
{
\begin{split}
\begin{array}{ll}
&\widetilde{s}_{1,1}(24)=121212121212333333444444\\
&\widetilde{s}^+_{1,1}(36)=141414141431222222466644455555533333\\
&\widetilde{s}^{\times}_{1,1}(36)=831644341312123344648421343234294214\\
&\widetilde{s}^{|-|}_{1,1}(36)=2331 42131222442043234231232 313031114
\end{array}
\end{split}}
\end{equation}
by upsetting the order of numbers of each of four number-based strings shown in Eq.(\ref{eqa:order-four-strings-matrix}). Clearly, it is not easy to reconstruct four Topcode-matrices $M_{1,1}$, $M^+_{1,1}$, $M^{\times}_{1,1}$ and $M^{|-|}_{1,1}$ by the non-ordered number-based strings shown in Eq.(\ref{eqa:no-order-four-strings-matrix}), even impossible for number-based strings with fairly long bytes.

\begin{figure}[h]
\centering
\includegraphics[width=13.8cm]{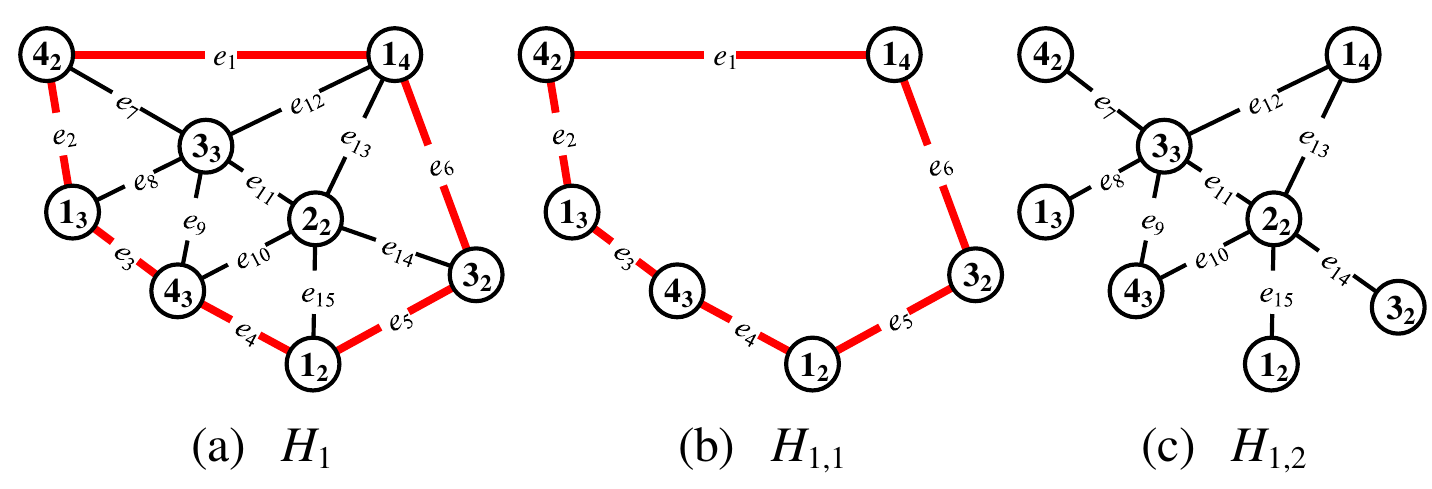}\\
\caption{\label{fig:4-colors-edge-coloring}{\small For illustrating the addition, multiplication and subtraction operations defined in Definition \ref{defn:planar-graphs-color-sets-4-operations}, cited from \cite{Yao-Sun-Wang-Su-Maximal-Planar-Graphs-2021}.}}
\end{figure}

\begin{figure}[h]
\centering
\includegraphics[width=16.4cm]{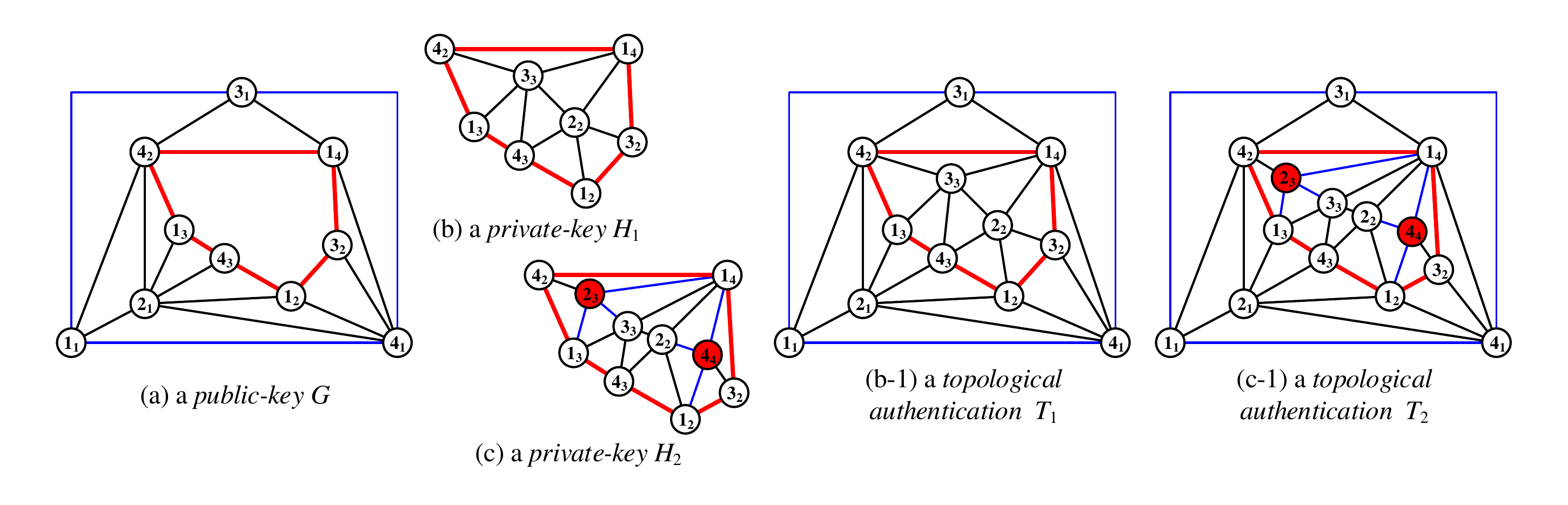}\\
\caption{\label{fig:mpg-4-color-authen}{\small Examples for the topological authentication based on maximal planar graphs, cited from \cite{Yao-Sun-Wang-Su-Maximal-Planar-Graphs-2021}.}}
\end{figure}

\begin{equation}\label{eqa:cycle-topcode-matrix}
\centering
{
\begin{split}
T_{code}(H_{1,1})= \left(
\begin{array}{ccccccccccccccc}
1_{4} & 4_{2} & 4_{3} & 1_{2} & 1_{2} & 3_{2} \\
e_{1} & e_{2} & e_{3} & e_{4} & e_{5} & e_{6}\\
4_{2} & 1_{3} & 1_{3} & 4_{3} & 3_{2} & 1_{4}
\end{array}
\right)
\end{split}}
\end{equation}

\begin{equation}\label{eqa:inside-topcode-matrix}
\centering
{
\begin{split}
T_{code}(H_{1,2})= \left(
\begin{array}{ccccccccccccccc}
4_{2} & 3_{3} & 3_{3} & 2_{2} & 3_{3} & 3_{3} & 1_{4} & 2_{2} & 2_{2}\\
e_{7} & e_{8} & e_{9} & e_{10} & e_{11} & e_{12} & e_{13} & e_{14} & e_{15}\\
3_{3} & 1_{3} & 4_{3} & 4_{3} & 2_{2} & 1_{4} & 2_{2} & 3_{2} & 1_{2}
\end{array}
\right)
\end{split}}
\end{equation}

\begin{equation}\label{eqa:addition-times-absolut-matrix}
\centering
{
\begin{split}
& M^+_{1,1}= \left(
\begin{array}{ccccccccccccccc}
1_{4} & 4_{2} & 4_{3} & 1_{2} & 1_{2} & 3_{2} \\
5_{6} & 5_{5} & 5_{6} & 5_{5} & 4_{4} & 4_{6}\\
4_{2} & 1_{3} & 1_{3} & 4_{3} & 3_{2} & 1_{4}
\end{array}
\right)~
M^{\times}_{1,1}= \left(
\begin{array}{ccccccccccccccc}
1_{4} & 4_{2} & 4_{3} & 1_{2} & 1_{2} & 3_{2} \\
4_{8} & 4_{6} & 4_{9} & 4_{6} & 3_{4} & 3_{8}\\
4_{2} & 1_{3} & 1_{3} & 4_{3} & 3_{2} & 1_{4}
\end{array}
\right)\\
& M^{|-|}_{1,1}= \left(
\begin{array}{ccccccccccccccc}
1_{4} & 4_{2} & 4_{3} & 1_{2} & 1_{2} & 3_{2} \\
3_{2} & 3_{1} & 3_{0} & 3_{1} & 2_{0} & 2_{2}\\
4_{2} & 1_{3} & 1_{3} & 4_{3} & 3_{2} & 1_{4}
\end{array}
\right)
\end{split}}
\end{equation}

Each colored graph $G$ admitting a coloring (resp. labeling) $f$ corresponds a Topcode-matrix $T_{code}(G)$. By the colored topological authentications $T_i=G[\ominus ^{cyc}_6] H_i$ for $i=1,2$ (see Fig.\ref{fig:mpg-4-color-authen}), we get two \emph{Topcode-matrix authentications} $T_{code}(T_i)=T_{code}(G)\uplus T_{code}(H_i)$ for $i=1,2$.

\begin{thm}\label{thm:planar-graph-M-coloring plans}
$^*$ Since there are $(m_k)!$ permutations of $k_1,k_2,\dots ,k_{m_k}$ to color the vertices of $V_k(T)$ of a planar graph $T$ in Definition \ref{defn:planar-graphs-color-sets-4-operations}, so we have $\prod^4_{k=1} (m_k)!$ coloring plans under a proper vertex $4$-coloring $g$. Suppose that the planar graph $T$ admits $|C^0_4(T)|$ different proper vertex $4$-colorings, then we have $M(T)$ different coloring plans for $T$, where the number
$$
M(T)=|C^0_4(T)|\cdot \prod^4_{k=1} (m_k)!
$$ refer to \cite{Jin-Xu-55-56-configurations-arXiv-2107-05454v1}.
\end{thm}

\begin{rem}\label{rem:333333}
In Definition \ref{defn:planar-graphs-color-sets-4-operations}, we have defined the edge color of an edge $uv$ with two ends colors $g(u)=a_i$ and $g(v)=b_j$ under a proper vertex $4$-coloring $g$. Furthermore, we have
\begin{asparaenum}[\textbf{Oper}-1. ]
\item \textbf{Three proper operations}: Addition $a_i(+)b_j=(a+b)_{(i+j)}$, multiplication $a_i(\cdot ) b_j=ab_{ij}$, subtraction $a_i(-)b_j=|a-b|_{|i-j|}$.
\item \textbf{Mixed operation 1}: $a_i(+\times )b_j=(a+b)_{ij}$, $a_i(+|-|)b_j=(a+b)_{|i-j|}$.
\item \textbf{Mixed operation 2}: $a_i(\times +)b_j=(a\cdot b)_{i+j}$, $a_i(\times |-|)b_j=(a\cdot b)_{|i-j|}$.
\item \textbf{Mixed operation 3}: $a_i(|-|\times )b_j=|a-b|_{ij}$, $a_i(|-|+)b_j=|a-b|_{i+j}$.
\item \textbf{Operation based on Klein four-group}: Let $0:=1$, $a:=2$, $b:=3$, $c:=4$ in the table $K^+_{\textrm{lein}}$ shown in Eq.(\ref{eqa:Klein-field-four-group}). We define $(K^+_{\textrm{lein}})$ addition with the commutative law $r_i(K^+_{\textrm{lein}})s_j=s_j(K^+_{\textrm{lein}})r_i$ in the following computations: $1_i(K^+_{\textrm{lein}})1_j=1_{i+j}$, $1_i(K^+_{\textrm{lein}})2_j=2_{i+j}$, $1_i(K^+_{\textrm{lein}})3_j=3_{i+j}$, $1_i(K^+_{\textrm{lein}})4_j=4_{i+j}$; $2_i(K^+_{\textrm{lein}})2_j=2_{i+j}$, $2_i(K^+_{\textrm{lein}})3_j=4_{i+j}$, $2_i(K^+_{\textrm{lein}})4_j=3_{i+j}$; $3_i(K^+_{\textrm{lein}})3_j=1_{i+j}$, $3_i(K^+_{\textrm{lein}})4_j=2_{i+j}$; and $4_i(K^+_{\textrm{lein}})4_j=1_{i+j}$.

\quad We define $(K^{\times}_{\textrm{lein}})$ multiplication with the commutative law $r_i(K^{\times}_{\textrm{lein}})s_j=s_j(K^{\times}_{\textrm{lein}})r_i$ as follows: $1_i(K^{\times}_{\textrm{lein}})1_j=1_{ij}$, $1_i(K^{\times}_{\textrm{lein}})2_j=1_{ij}$, $1_i(K^{\times}_{\textrm{lein}})3_j=1_{ij}$, $1_i(K^{\times}_{\textrm{lein}})4_j=1_{ij}$; $2_i(K^{\times}_{\textrm{lein}})2_j=2_{ij}$, $2_i(K^{\times}_{\textrm{lein}})3_j=3_{ij}$, $2_i(K^{\times}_{\textrm{lein}})4_j=4_{ij}$; $3_i(K^{\times}_{\textrm{lein}})3_j=3_{ij}$, $3_i(K^{\times}_{\textrm{lein}})4_j=2_{ij}$; and $4_i(K^{\times}_{\textrm{lein}})4_j=3_{ij}$.\paralled
\end{asparaenum}
\end{rem}

\begin{rem}\label{rem:333333}
Let $K_{\textrm{lein}}= \{0, a, b, c\}$ be \emph{Klein four-group} with addition ``$+$'' in $K^+_{\textrm{lein}}$ shown in Eq.(\ref{eqa:Klein-field-four-group}). The Klein four-group can be extended to a \emph{finite field}, called the \emph{Klein field}, where multiplication is added as a second operation, with $0$ as the zero element and $a$ as the identity element. The multiplication table is $K^{\times}_{\textrm{lein}}$ shown in Eq.(\ref{eqa:Klein-field-four-group}). Multiplication ``$\times$'' and addition ``$+$'' obey the distributive law.

\begin{equation}\label{eqa:Klein-field-four-group}
\centering
K^+_{\textrm{lein}}=
\begin{array}{c|cccc|c}
+ & 0 & a & b & c & \textrm{substitution}\\
\hline
0 & 0 & a & b & c\\
a & a & 0 & c & b & (1,2)(3,4)\\
b & b & c & 0 & a & (1,3)(2,4) \\
c & c & b & a & 0 & (1,4)(2,3)
\end{array}
\qquad K^{\times}_{\textrm{lein}}=
\begin{array}{c|ccccc}
\times & 0 & a & b & c \\
\hline
0 & 0 & 0 & 0 & 0 \\
a & 0 & a & b & c \\
b & 0 & b & c & a \\
c & 0 & c & a & b
\end{array}
\end{equation}

A Klein four-group, also, is a normal subgroup of the \emph{alternating group} $A_4$ and of the \emph{symmetric group} $S_4$ over four letters.\paralled
\end{rem}

\begin{problem}\label{qeu:444444}
For positive integer sets $C_i=\{i_1,i_2,\dots ,i_{m_i}\}$ with $i\in [1,4]$, \textbf{does} there exist a maximal planar graph $G$ such that $V(G)=\bigcup ^4_{i=1}V_i$, and $m_i=|V_i|$ with $i\in [1,4]$? \textbf{Determine} all such maximal planar graphs $G$ if it is so.
\end{problem}

\subsubsection{Maximal planar graphs with 2-color-unchanged cycles}

\begin{defn} \label{defn:2-color-unchanged-cycle-mpg}
\cite{Jin-Xu-55-56-configurations-arXiv-2107-05454v1} Let $G$ be a maximal planar graph admitting a proper vertex $4$-coloring $f$ and let $C$ be a $2k$-cycle of $G$. If $f(V(C))=\{i,j\}$, and we exchange $a,b\in \{1,2,3,4\}\setminus \{i,j\}$ in one or two of two subgraphs $G^{C}_{out}-V(C)$ and $G^{C}_{in}-V(C)$ for $G=G^C_{out}[\ominus^{cyc}_k]G^C_{in}$, we get a new proper vertex $4$-coloring $g$ of $G$, and call $C$ a \emph{2-color-unchanged cycle} of $G$ if $|g(V(C))|=2$, and $g$ a \emph{Kempe transformation} of $f$.\qqed
\end{defn}

Five graphs shown in Fig.\ref{fig:2-color-no-change-edge-coin} are maximal planar graphs, and they contain 2-color-unchanged cycles, where $G_1$ contains a 2-color-unchanged cycle $C_1=x_1x_2x_3x_4x_3x_2x_1$, $G_2$ containing 2-color-unchanged cycle $C_2$ is obtained by exchanging the inner face $\Delta wx_3x_4$ with the outer face $\Delta wx_1y$ of $G_1$, and then $G_3$ is obtained by flipping $G_2$, and $T$ contains three 2-color-unchanged cycles. We do $W$-coinciding operation defined in Definition \ref{defn:W-splitting-coinciding-operation} to a cycle $wx_1x_4w$ in $G_2$ and another cycle $wx_1x_4w$ in $G_1$, and obtain the maximal planar graph $H=G_1[\ominus ^{cyc}_3]G_2$ as a \emph{topological authentication} based on two maximal planar graphs $G_1$ (as a \emph{topological public-key}) and $G_2$ (as a \emph{topological private-key}). The maximal planar graph $T$ contains three 2-color-unchanged cycles having a common edge. And another maximal planar graph $T=G_1[\ominus ^{cyc}_3]G_3$ shown in Fig.\ref{fig:2-color-no-change-edge-coin}(d) contains three 2-color-unchanged cycles.

\begin{figure}[h]
\centering
\includegraphics[width=16.4cm]{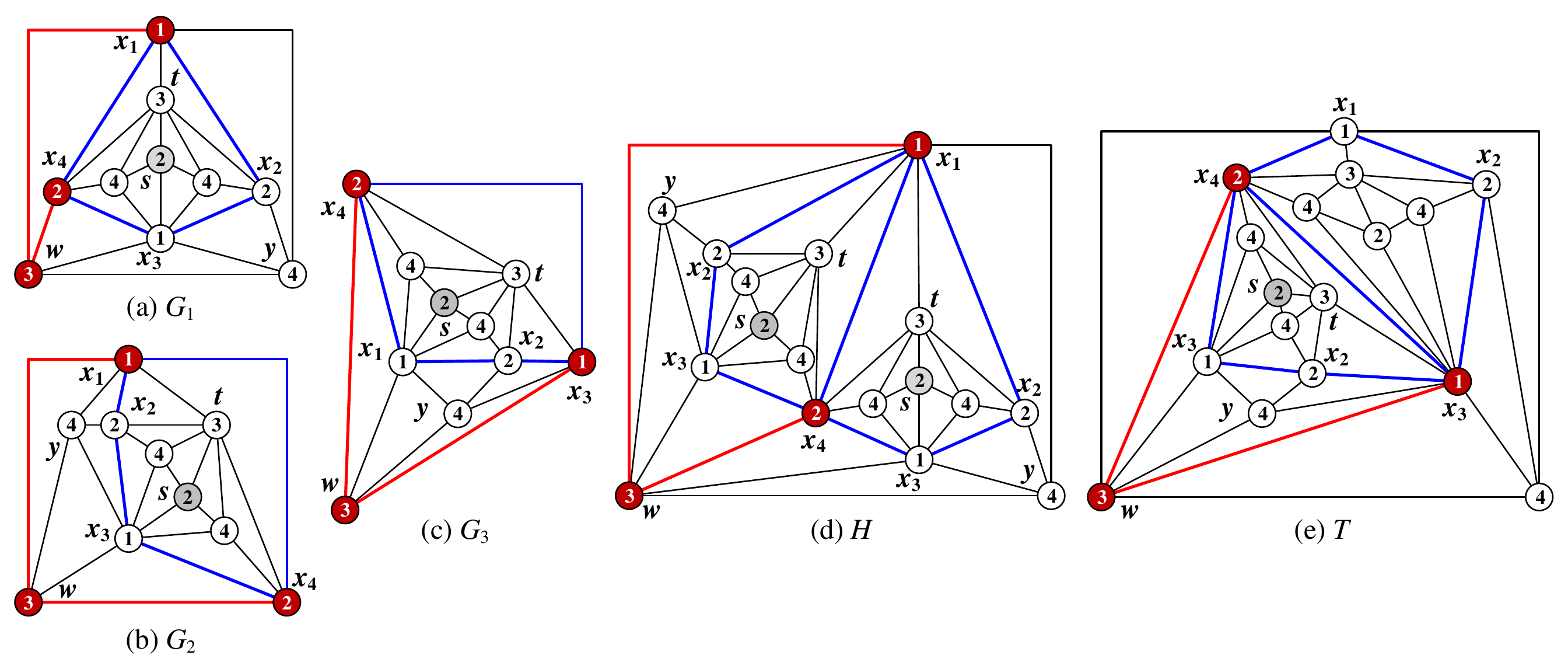}\\
\caption{\label{fig:2-color-no-change-edge-coin}{\small (a) A topological public-key $G_1$; (b) a topological private-key $G_2$ holding $G_2\cong G_1$; (c) a topological private-key $G_3$ holding $G_3\cong G_2\cong G_1$; (d) a topological authentication $H=G_1[\ominus ^{cyc}_3]G_2$; (e) a topological authentication $T=G_1[\ominus ^{cyc}_3]G_3$.}}
\end{figure}

Another topological authentication $J$ shown in Fig.\ref{fig:2-color-no-change-vertex-coin} has four 2-color-unchanged cycles with common vertices, and
\begin{equation}\label{eqa:555555}
J=\Big (\big (G[\ominus ^{cyc}_3]G\big )[\ominus ^{cyc}_3]G_1\Big )[\ominus ^{cyc}_3]G_2
\end{equation}

\begin{figure}[h]
\centering
\includegraphics[width=16.4cm]{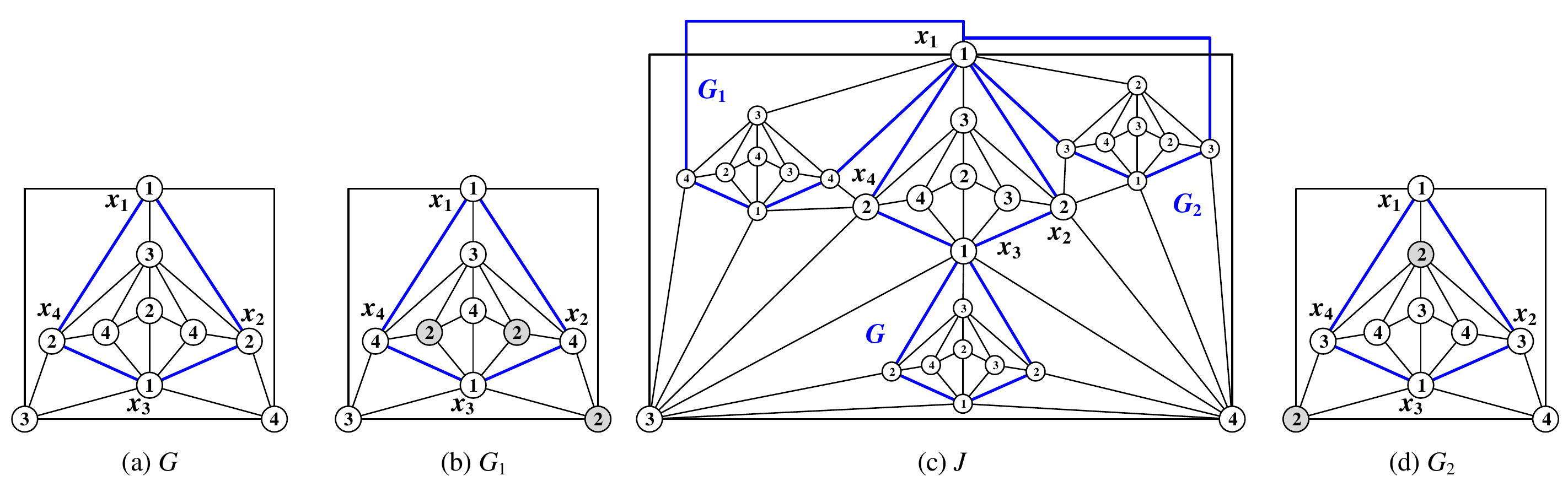}\\
\caption{\label{fig:2-color-no-change-vertex-coin}{\small (a) A topological public-key $G$; (b) and (d) two topological private-keys $G_1$ and $G_2$ holding $G\cong G_1\cong G_2$; (c) a topological authentication $J$.}}
\end{figure}

\begin{thm}\label{thm:infinite-2-color-unchanged-cycles}
$^*$ There are infinite maximal planar graphs containing 2-color-unchanged cycles, in which some 2-color-unchanged cycles are vertex disjoint, or have common vertices, or have common edges.
\end{thm}

\begin{thm}\label{thm:unchanged-bichromatic-cycle-basic}
\cite{Jin-Xu-55-56-configurations-arXiv-2107-05454v1} Let $G$ be a maximal planar graph admitting a proper vertex $4$-coloring $f$. Then, a cycle $C$ of $G$ is a 2-color-unchanged cycle of $k$ vertices if and only if for any 2-color-unchanged cycle $C\,'$ of $G$, both $C$ and $C\,'$ are not internally intersecting from each other under the proper vertex $4$-coloring $f$, that is,

(i) $C\,'$ is a subgraph of one of two graphs $G^{C}_{out}-V(C)$ and $G^{C}_{in}-V(C)$ for $G=G^{C}_{in}[\ominus^C_k]G^{C}_{out}$;

(ii) $C$ is a subgraph of one of two graphs $G^{C\,'}_{out}-V(C\,')$ and $G^{C\,'}_{in}-V(C\,')$ for $G=G^{C\,'}_{in}[\ominus^{C\,'}_k]G^{C\,'}_{out}$.
\end{thm}

\begin{thm}\label{thm:tree-type-2-color-unchanged-cycle-mpg}
\cite{X-Q-Liu-J-Xu-extending-contracting-2017} For any integer $n \geq 2$, there exists a tree-type 2-color-unchanged cycle maximal planar graph $G$ of order $4n$, such that $G$ admits $(2^{n-2}+2^{n-1})$ proper vertex $4$-colorings, where the numbers of tree-colorings and 2-color-unchanged cycle colorings are $2^{n-2}$ and $2^{n-1}$, respectively.
\end{thm}

\begin{rem}\label{rem:333333}
Theorem \ref{thm:tree-type-2-color-unchanged-cycle-mpg} provides us the theoretical basis of using maximal planar graphs to make topological codes with complexity $O(2^n)$.\paralled
\end{rem}

\begin{problem}\label{qeu:444444}
For complex topological authentications, by Definition \ref{defn:2-color-unchanged-cycle-mpg} of the 2-color-unchanged cycle, it may be interesting to consider the following problems:
\begin{asparaenum}[\textbf{2c-un-cyc}-1. ]
\item \textbf{Characterize} a maximal planar graph with exactly $k$ 2-color-unchanged circles with $k\geq 1$.
\item \textbf{Characterize} a maximal planar graph containing a 2-color-unchanged cycle of length $m \geq k$ for a given integer $k\geq 4$.
\item \textbf{Characterize} maximal planar graphs containing the maximum number of 2-color-unchanged cycles.
\item \textbf{Does} there exist the necessary and sufficient conditions for maximal planar graphs containing 2-color-unchanged cycles?
\item Let $n_{2cyc}(G,f)$ be the number of 2-color-unchanged cycles in a maximal planar graph $G$ admitting a proper vertex $4$-coloring $f$. If $G$ contains 2-color-unchanged cycles, \textbf{do} we have $n_{2cyc}(G,f)=n(G,f\,')$ for some pair of different proper vertex $4$-colorings $f$ and $f\,'$?
\item Let $\{C_1,C_2,\dots ,C_m\}$ be the set of 2-color-unchanged cycles of a maximal planar graph $G$ if any proper vertex $4$-coloring $f$ of $G$ holds $n_{2cyc}(G,f)\geq 1$ true. We have a $C_1$-split graph $G_1=G\wedge C_1$ having two components $G^{C_1}_{out}$ and $G^{C_1}_{in}$. Without loss of generality, $C_2$ is a proper subgraph of $G^{C_1}_{out}$, we have a $C_2$-split graph $G_2=G_1\wedge C_2$ having three components $G^{C_2}_{out}$, $G^{C_2}_{in}$ and $G^{C_1}_{in}$, go on in this way, we get $G_{k+1}=G_k\wedge C_{k+1}$ having $(k+2)$ components, and $G_{m}=G_{m-1}\wedge C_{m}$ having $(m+1)$ components, in which each component contains no any one of $C_1,C_2,\dots ,C_m$ as a proper subgraph. For simplicity, we write
 \begin{equation}\label{eqa:55-2-color-unchanged-cycle}
 G_{m}=G\wedge ^m_{k=1}C_k
 \end{equation} For any permutation $C_{i_1},C_{i_2},\dots ,C_{i_m}$ of the 2-color-unchanged cycles $C_1,C_2,\dots ,C_m$, then
 $$
 G\wedge ^m_{k=1}C_k\cong G\wedge ^m_{k=1}C_{i_k}
 $$ \textbf{Do} we have $G\wedge ^s_{k=1}C_k\cong G\wedge ^s_{k=1}C_{i_k}$ with $2\leq s<m$?
\end{asparaenum}
\end{problem}

\begin{problem}\label{qeu:2-color-unchanged-cycles-authentication}
A topological authentication on maximal planar graphs with 2-color-unchanged cycles under proper vertex $4$-colorings: Given a maximal planar graph $G$ (as a topological public-key) with $n_{2cyc}(G,f)\geq 1$, where $n_{2cyc}(G,f)$ is the number of 2-color-unchanged cycles, but the vertices of $G$ are not colored, \textbf{find}:

(i) A topological private-key $G\wedge C_{i}$ is colored by a proper vertex $4$-coloring $g$, where $C_i$ is a 2-color-unchanged cycle of $G$, such that $G$ admits a proper vertex $4$-coloring $f$ holding $f(w)=g(w)$ for $w\in V(G\wedge C_{i})$.

(ii) By Eq.(\ref{eqa:55-2-color-unchanged-cycle}) a topological private-key $G\wedge ^s_{j=1}C_{j}$ is colored by a proper vertex $4$-coloring $g$, where $C_1,C_2$, $\dots $, $C_{s}$ with $s\geq 2$ are the 2-color-unchanged cycles of $G$, such that $G$ admits a proper vertex $4$-coloring $f$ holding $f(w)=g(w)$ for $w\in V(G\wedge ^s_{j=1}C_{j})$.
\end{problem}

\subsubsection{Rhombus algorithms of planar graphs}

See subsections Uncolored-rhombus algorithms and Colored-rhombus algorithms in Section 3.

\subsubsection{Coloring characterized graphs}

\begin{defn} \label{defn:4-coloring-characterized-graph}
\cite{Jin-Xu-55-56-configurations-arXiv-2107-05454v1} Let $C_4(G)=\{f_1, f_2,\dots ,f_n\}$ be the set of all different proper vertex $4$-colorings of a planar graph $G$. A \emph{$4$-coloring characterized graph} $C_c(G)$ has its own vertex set $V(C_c(G))=C_4(G)$, and $C_c(G)$ has an edge $f_if_j$ with two ends $f_i,f_j\in V(C_c(G))$ if $f_j$ is obtained by doing the Kempe transformation to $f_i$ defined in Definition \ref{defn:2-color-unchanged-cycle-mpg}, or exchange the colors of some vertices of $G$ under $f_i$ to make the proper vertex $4$-coloring $f_j$ of $G$.\qqed
\end{defn}

See Fig.\ref{fig:4-coloring-character} and Fig.\ref{fig:Kempe-type-4-base-module} for understanding Definition \ref{defn:4-coloring-characterized-graph}.

\begin{figure}[h]
\centering
\includegraphics[width=16.4cm]{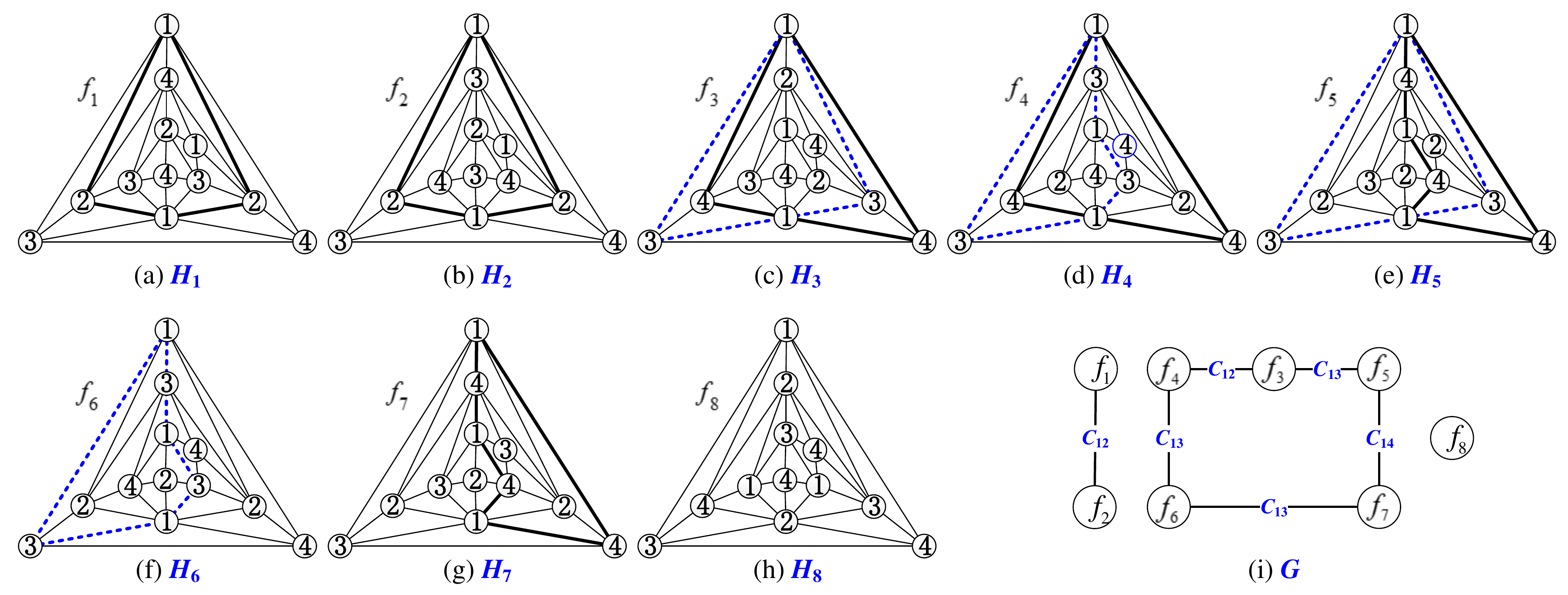}\\
\caption{\label{fig:4-coloring-character} {\small (a)-(h) all proper vertex $4$-colorings of a hybrid-type maximal planar graph having 2-color-unchanged cycles; (i) the proper vertex $4$-coloring characterized graph $G$, cited from \cite{Jin-Xu-55-56-configurations-arXiv-2107-05454v1}.}}
\end{figure}

\begin{figure}[h]
\centering
\includegraphics[width=16.4cm]{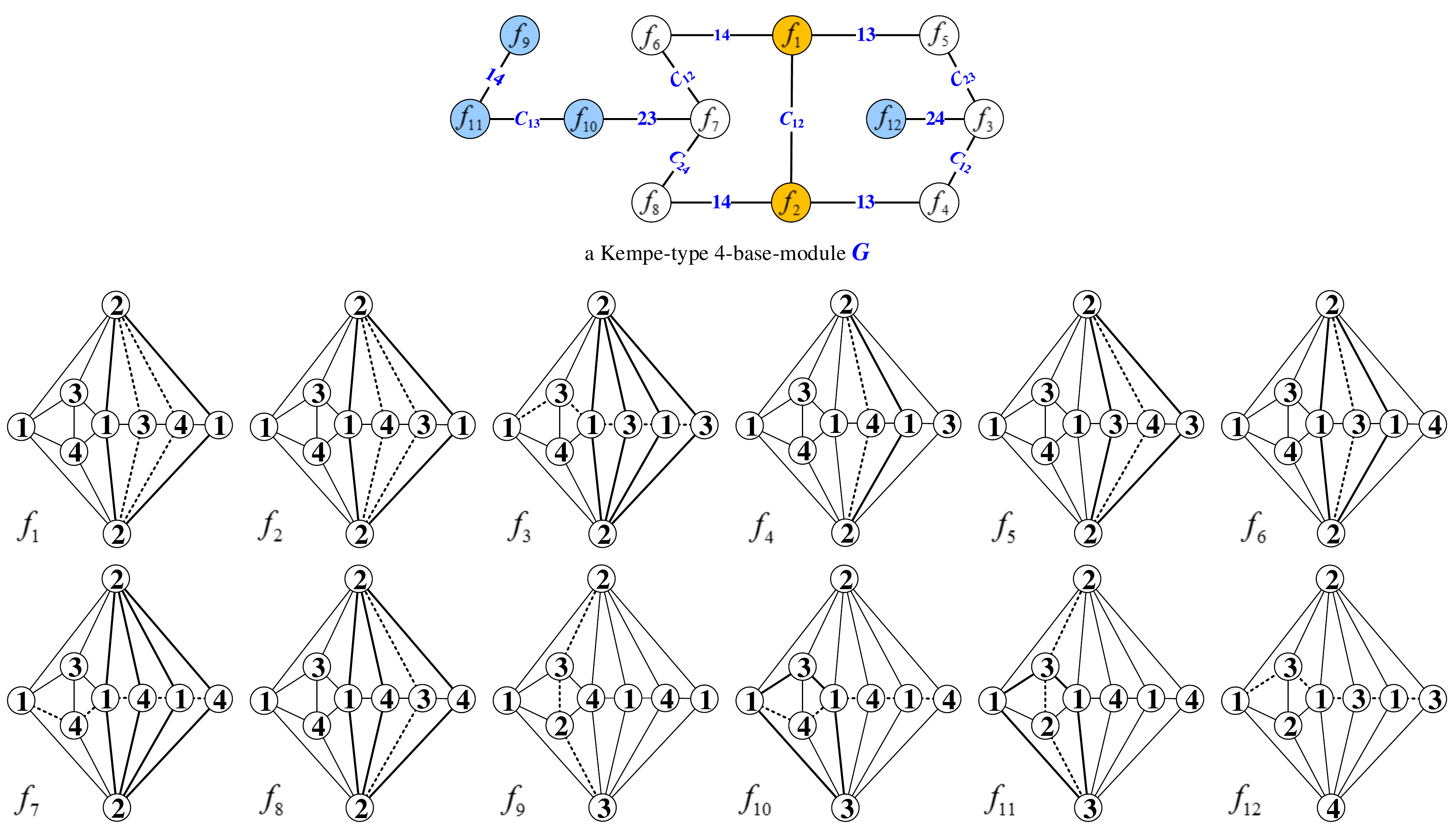}\\
\caption{\label{fig:Kempe-type-4-base-module} {\small A Kempe-type 4-base-module with all proper vertex $4$-colorings and its proper vertex $4$-coloring characterized graph $G$, cited from \cite{Jin-Xu-55-56-configurations-arXiv-2107-05454v1}.}}
\end{figure}

\begin{problem}\label{qeu:characterized-graph-public-key}
Given a proper vertex $4$-coloring characterized graph $H$ (as a topological public-key, refer to Definition \ref{defn:4-coloring-characterized-graph}), \textbf{find} a maximal planar graph $G$ (as a topological private-key), such that its proper vertex $4$-coloring characterized graph $C_c(G)$ (as a topological authentication) contains $H$ as a subgraph.
\end{problem}

\begin{rem}\label{rem:333333}
In fact, each edge $f_if_j$ of the proper vertex $4$-coloring characterized graph $C_c(G)$ of a planar graph $G$ defined in Definition \ref{defn:4-coloring-characterized-graph} means a \emph{vertex-colored graph homomorphism} from $G$ into $G$ itself, also, refer to Definition \ref{defn:vertex-colored-graph-anti-homomorphisms}.
\paralled
\end{rem}

\begin{problem}\label{qeu:444444}
Let $S_4\langle J,C\rangle $ be the set of all proper vertex $4$-colorings of a planar graph $J$ such that a cycle $C$ of $J$ is a 2-color-unchanged cycle for any proper vertex $4$-coloring of $S_4\langle J,C\rangle $, and this cycle $C$ is not a 2-color-unchanged cycle for any proper vertex $4$-coloring in $C_4(J)\setminus S_4\langle J,C\rangle $, where $C_4(J)$ is the set of all proper vertex $4$-colorings of $J$. \textbf{Compare} $S_4\langle J,C\rangle $ and $C_4(J)\setminus S_4\langle J,C\rangle $.
\end{problem}

\begin{problem}\label{qeu:444444}
Let $C_4(H)$ be the set of proper vertex $4$-colorings of a maximal planar graph $H$ that is not a cycle. There are the following problems:

(i) \textbf{Does} there exist a maximal planar graph $H$ containing a cycle $C^*$ such that $C^*$ is a 2-color-unchanged cycle for any proper vertex $4$-coloring of $C_4(H)$?

(ii) For topological authentications, we can consider 2-color-unchanged path, 2-color-unchanged tree (forest), 3-color-unchanged cycle (path, tree), 4-color-unchanged cycle (path, tree) \emph{etc.}

(iii) For any proper vertex $4$-coloring of $C_4(H)$, \textbf{do} we have one of a 2-color-unchanged cycle $C_m$, a 2-color-unchanged path $P_m$ and a 2-color-unchanged tree $T_m$ of $m$ vertices? If it is so, \textbf{determine} the maximal value of $m$.
\end{problem}

\begin{conj}\label{conj:2-color-tree-forest}
$^*$ There is a proper vertex $4$-coloring $f$ of a maximal planar graph $G$ of $p$ vertices such that a 2-color tree (or forest) $T$ holds $|V(T)|\geq \frac{p}{2}$ true. It is related with the Big Forest Conjecture of planar graphs \cite{Albertson-Berman-large-forest-1979}.
\end{conj}

\begin{rem}\label{rem:333333}
About Conjecture \ref{conj:2-color-tree-forest}, we recall the Big Forest Conjecture of planar graphs is proposed in \cite{Albertson-Berman-large-forest-1979}: ``Every planar graph of order $n$ contains an induced forest of order at least $\frac{n}{2}$.'' In 1986, Erd\"{o}s \emph{et al.}, in \cite{Erdos-Saks-Sos-1986}, proved that the sum of the dwindling number of a graph $G$ and the order of the maximum induced forest of the graph $G$ is exactly equal to the order of $G$. If the Big Forest Conjecture of planar graphs is true, then any planar graph of order $n$ contains an independent set having at least $\frac{n}{4}$ vertices, which can be obtained directly from the Four Color Problem on planar graphs. \paralled
\end{rem}

\begin{defn} \label{defn:W-type-coloring-characterized-graph}
$^*$ Let $C_{olor}(G)=\{f_1, f_2,\dots ,f_m\}$ be the set of all different $W$-type colorings of a graph $G$, and let $O_W=\{O_1, O_2,\dots ,O_m\}$ be the set of operations between $W$-type colorings. A \emph{$W$-type coloring characterized graph} $C^W_{\textrm{harac}}(G)$ has its own vertex set $V(C^W_{\textrm{harac}}(G))=C_{olor}(G)$, and $C^W_{\textrm{harac}}(G)$ has an edge $f_if_j$ with two ends $f_i,f_j\in V(C^W_{\textrm{harac}}(G))$ if $f_j$ is obtained by doing an operation $O_s\in O_W$ to $f_i$, and vice versa.\qqed
\end{defn}

\begin{problem}\label{qeu:authenticationproblem-NP}
By Definition \ref{defn:W-type-coloring-characterized-graph}, for a given public-key $H^*$, \textbf{find} a private-key $G$, such that the $W$-type coloring characterized graph $C^W_{\textrm{harac}}(G)=H^*$. This problem has:

(i) No directed information for the private-key $G$;

(ii) what is $W$-type coloring? and

(iii) finding all $W$-type colorings of $G$ is quit difficult, since no polynomial algorithm for determining all $W$-type colorings for each of the existing $W$-type coloring and each graph $G$.

Thereby, this problem induces a topological authentication $\textbf{T}_{\textbf{a}}\langle\textbf{X},\textbf{Y}\rangle $ to be $O(\textrm{NP})$ (refer to Definition \ref{defn:topo-authentication-multiple-variables}).
\end{problem}

\begin{problem}\label{qeu:444444}
\textbf{Maximal-planar-graph authentication.}

\textbf{Public-key input:} $m$ triangles with $m\geq 3$.

\textbf{Private-key output:} A maximal planar graph $G$ having $m$ inner triangular faces and an outer face, a proper vertex $4$-coloring of $G$, and a number-based string generated from the Topcode-matrix $T_{code}(G)$.

The complexity of the Maximal-planar-graph authentication is as follows:

(i) \textbf{Determine} the number of maximal planar graphs containing $m$ inner triangular faces;

(ii) \textbf{find} out all proper vertex $4$-colorings of each maximal planar graph $G$ containing just $m$ inner faces;

(iii) \textbf{select} a desired Topcode-matrix $T_{code}(G)$ based on a proper vertex $4$-coloring of $G$;

(iv) \textbf{produce} a number-based string from $T_{code}(G)$ by a pre-specified algorithm on Topcode-matrices.
\end{problem}

\begin{defn} \label{defn:multiple-edge-complete-graphs}
\cite{Yao-Wang-2106-15254v1} A \emph{multiple-edge complete graph} $K^{mul}_n$ of $n$ vertices has its vertex set $V(K^{mul}_n)=\{x_1,x_2,\dots ,x_n\}$, each pair of vertices $x_i$ and $x_j$ for $i\neq j$ is joined by $a_{i,j}$ edges $e_{i,j}(k)$ with $k\in [1,a_{i,j}]$ for $a_{i,j}\geq 1$, and there exists some $a_{r,s}\geq 2$ for $r\neq s$. (see a multiple-edge complete graph $K^{mul}_4$ shown in Fig.\ref{fig:vertex-split-problem})\qqed
\end{defn}

\begin{figure}[h]
\centering
\includegraphics[width=14cm]{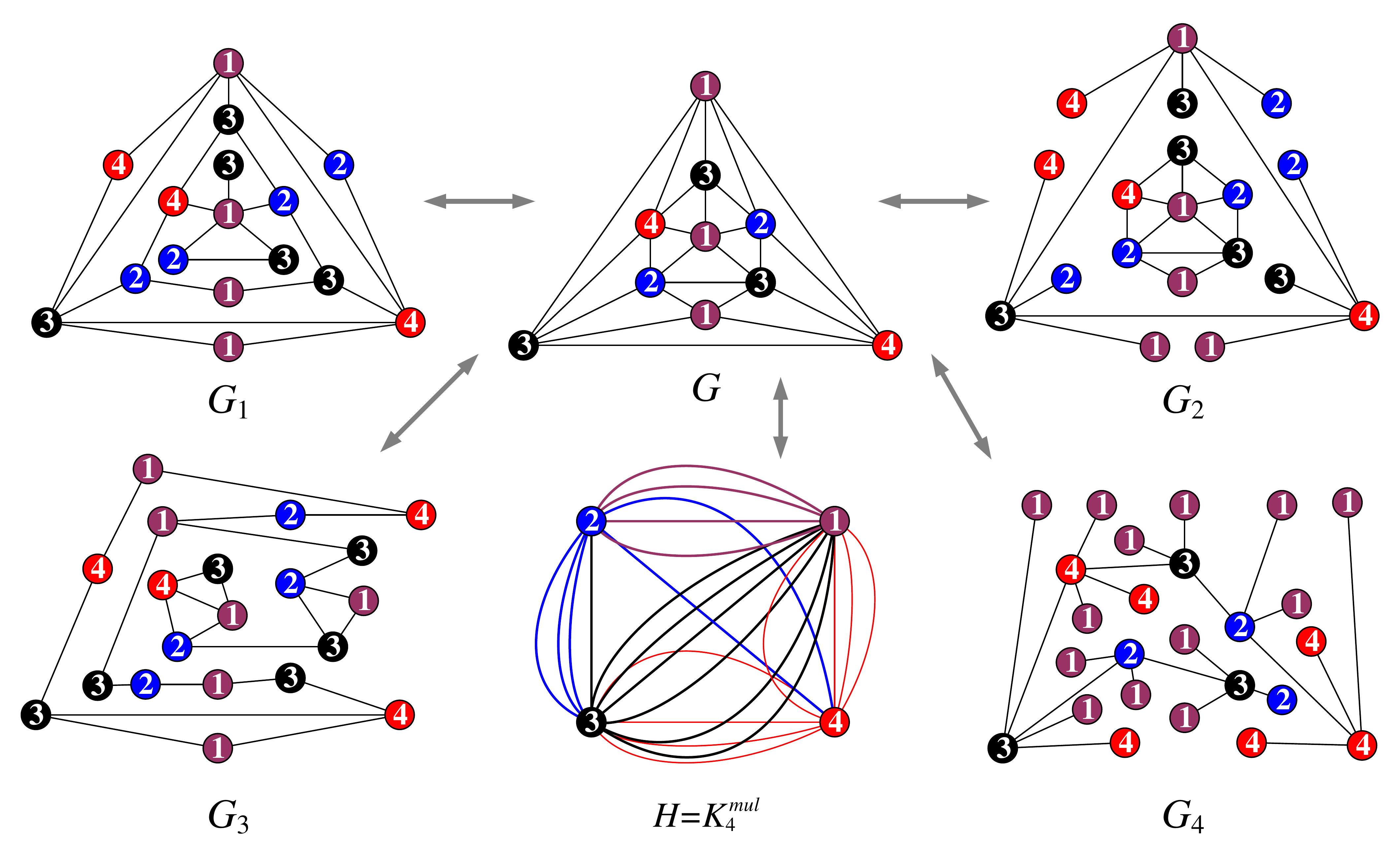}\\
\caption{\label{fig:vertex-split-problem}{\small Simple graph homomorphisms: $G_k\rightarrow G$ for $k\in [1,4]$, and a multiple graph homomorphism $G\rightarrow H=K^{mul}_4$. Conversely, $H\rightarrow_{split} G$, and $G\rightarrow_{split} G_k$ for $k\in [1,4]$, cited from \cite{Yao-Wang-2106-15254v1}.}}
\end{figure}

\begin{thm}\label{thm:666666}
A connected graph $G$ admits a proper vertex $n$-coloring with $n=\chi(G)$ if and only if $G$ is graph homomorphism into a multiple-edge complete graph $K^{mul}_n$ defined in Definition \ref{defn:multiple-edge-complete-graphs}.
\end{thm}

\begin{conj} \label{conj:c4-Reed-B-conjecture-X}
(Bruce Reed, 1998) $ \chi(G)\leq \big \lceil \frac{1}{2}(\Delta(G)+1+K(G))\big \rceil$.
\end{conj}

Andrew King proved that Conjecture \ref{conj:c4-Reed-B-conjecture-X} holds when $G$ is a line graph. Reed has proved the following two results:

\begin{thm} \label{thm:ReedB-1998}
(Bruce Reed, 1998) There exists a positive number $\lambda$ such that $\chi(G)\leq (1-\lambda)(\Delta(G)+1)+\lambda K(G)$.
\end{thm}

\begin{thm} \label{thm:box}
(Bruce Reed, 1998) If there exists a positive number $\zeta$ such that $K(G)\geq (1-\zeta)(\Delta(G)+1)$, then $\chi(G)=\left \lceil \frac{1}{2}(\Delta(G)+1+K(G))\right \rceil$.
\end{thm}

\begin{problem}\label{problem:xxxxxx}
\cite{Yao-Wang-2106-15254v1} By Definition \ref{defn:multiple-edge-complete-graphs}, for \textbf{what conditions held} by $a_{i,j}$ with $i,j\in [1,n]$ and $i\neq j$, \textbf{can} a multiple-edge complete graph $K^{mul}_n$ (as a \emph{public-key}) be vertex-split into a particular graph $H$ (as a \emph{private-key}) such that there is a graph homomorphism $H\rightarrow K^{mul}_n$?
\end{problem}

\begin{problem}\label{problem:xxxxxx}
\cite{Yao-Wang-2106-15254v1} By Definition \ref{defn:multiple-edge-complete-graphs}, for $n=4$, \textbf{what conditions do} $a_{1,2},a_{1,3},a_{1,4},a_{2,3},a_{2,4}$ and $a_{3,4}$ hold such that $K^{mul}_4$ can be vertex-split into maximal planar graphs, or out-planar graphs, or semi-maximal planar graphs for making \emph{public-keys} and \emph{private-keys}?
\end{problem}

\begin{problem}\label{qeu:444444}
In Fig.\ref{fig:Triangle-split-coincide}, doing the triangle-splitting operation to a maximal planar graph $G$ produces a disconnected graph $H$ with 11 components, we write this process as $G\rightarrow_{\textrm{t-split}} H$. Conversely, the maximal planar graph $G$ can be obtained by doing the \emph{cycle-coinciding operation} to the disconnected graph $H$, that is, $H$ is \emph{graph homomorphism} to $G$, denoted as $H\rightarrow G$. Another operation is the \emph{triangular system} $\textbf{\textrm{T}}_{sys}=\{T_{S,1},T_{S,2},T_{S,3},T_{S,4}\}$ shown in Fig.\ref{fig:Triangle-split-coincide} to tile completely a maximal planar graph for getting a proper vertex $4$-coloring of the maximal planar graph. So, we can conjecture: \emph{Each uncolored maximal planar graph can be tiled completely by the triangular system} $\textbf{\textrm{T}}_s$.
\end{problem}

\begin{figure}[h]
\centering
\includegraphics[width=16.4cm]{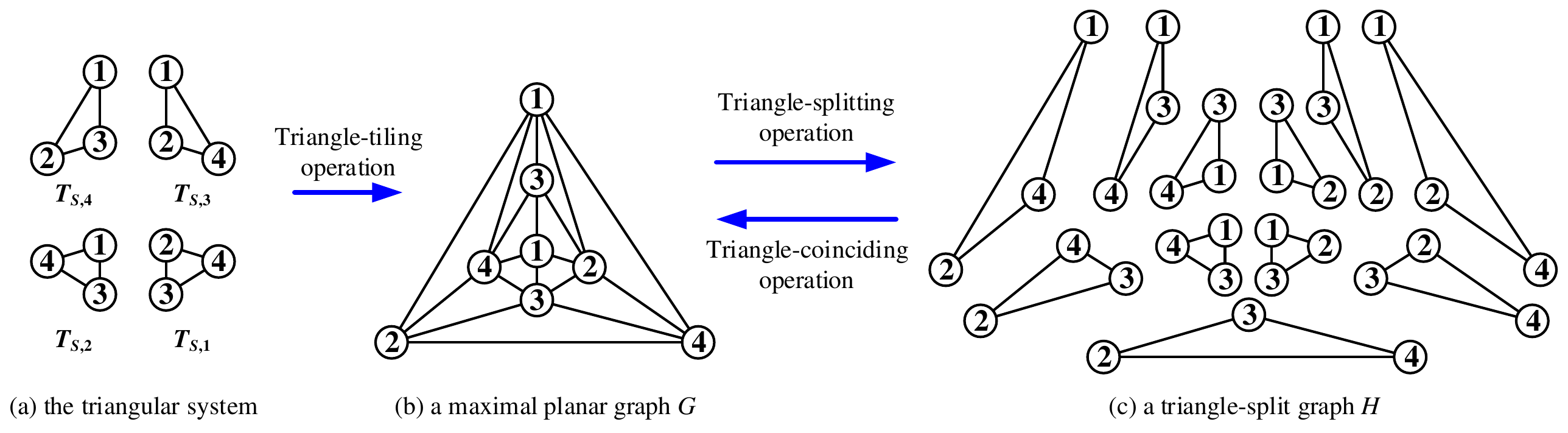}\\
\caption{\label{fig:Triangle-split-coincide}{\small The triangular system, the triangle-splitting operation and the triangle-coinciding operation.}}
\end{figure}

\subsubsection{Mpg-nested maximal planar graphs}

\begin{defn} \label{defn:complete-triangular-coincided-mpgs}
$^*$ Let $G$ be a maximal planar $(p,q)$-graph admitting a $4$-coloring, and let $G$ have $m$ inner triangular faces $C^3_1,C^3_2,\dots ,C^3_m$, here $m=F(G)-1$ based on Euler's formula $p-q+F(G)=2$, and $C^3_k$ be the bound of $k$th inner triangular face. We take randomly $m$ maximal planar $(p_k,q_k)$-graphs $G_k$ admitting a 4-coloring $f_k$ with $p_k\geq 4$ such that the bound of outer face of each maximal planar graph $G_k$ is colored the same colors as that of the bound $C^3_k$, where $C^3_k$ is the bound of outer face of the maximal planar graph $G_k$. We do a \emph{colored cycle-coinciding operation} to the triangle $C^3_1$ of $G$ and the triangular outer face $C^3_1$ of $G_1$, the resultant maximal planar graph is denoted as $H_1=G\big [\ominus^{C^3_1}_3\big ]G_1$, and moreover $H_2=H_1\big [\ominus^{C^3_2}_3\big ]G_2$, go on in this way, we get $H_j=H_{j-1}\big [\ominus^{C^3_j}_3\big ]G_j$ for $j\in [1,m]$, where $H_0=G$, and we rewrite the last maximal planar graph by $H_m=G\big [\ominus^{C^3_j}_3\big ]^m_{j=1}G_j$, called a \emph{complete mpg-nested maximal planar graph}. \qqed
\end{defn}

Motivated from Definition \ref{defn:complete-triangular-coincided-mpgs}, we define a \emph{mpg-nested maximal planar graph} as:
\begin{defn} \label{defn:mpg-nested-maximal-planar-graphs}
$^*$ Let $G$ be a maximal planar graph admitting a $4$-coloring $h$, and let $C_{3,i}=x_{3,i,1}x_{3,i,2}x_{3,i,3}x_{3,i,1}$ be the bound of some triangular inner face $f^i_{ace}$. We take randomly another maximal planar graph $H$ with the outer face $f^H_{ace}$, where $f^H_{ace}$ has its own bound $C_{out}=y_1y_2y_3y_1$ and $H$ admits a $4$-coloring $g$. Obviously, we can switch the colors of the $4$-coloring $g$ of $H$ to obtain another $4$-coloring $g^*$ of $H$, such that $h(x_{3,i,j})=g^*(y_j)$ with $j\in [1,3]$. Then we do a colored cycle-coinciding operation to $G$ and $H$, such that the \emph{mpg-nested maximal planar graph} $G[\ominus ^{cyc}_3]H$ admits a $4$-coloring, so that $C_{3,i}[\ominus ^{cyc}_3]C_{out}$ is a triangular inner face of $G[\ominus ^{cyc}_3]H$, we say that $H$ is \emph{nested} into $G$.\qqed
\end{defn}

Let $\textbf{\textrm{M}}_{\textbf{pg}}=(G_1,G_2,\dots ,G_m)$ be a \emph{maximal planar graph base} with each graph $G_i$ to be a maximal planar graph, and let $\textbf{\textrm{M}}^c_{\textbf{pg}}=(G^c_1,G^c_2,\dots ,G^c_m)$ be a \emph{$4$-colored maximal planar graph base}, where each maximal planar graph $G^c_i$ admits a $4$-coloring. We come to construct the so-called mpg-nested maximal planar graphic lattices in the following.

We write $a_k$ copies of a maximal planar graph $G_i$ of the base $\textbf{\textrm{M}}_{\textbf{pg}}$ as $a_kG_i$, and arrange $a_1G_1$, $a_2G_2$, $\dots$, $a_mG_m$ into a permutation $T_1,T_2,\dots, T_{A}$, where $A=\sum ^m_{k=1}a_k$. By Definition \ref{defn:mpg-nested-maximal-planar-graphs}, we do a colored cycle-coinciding operation to $G$ and $T_1$, such that $T_1$ is \emph{nested} into $G$, and we get a mpg-nested maximal planar graph $H_1=G[\ominus ^{cyc}_3]T_1$, and then we do a colored cycle-coinciding operation to $H_1$ and $T_2$, and obtain the second mpg-nested maximal planar graph $H_2=H_1[\ominus ^{cyc}_3]T_2$; go on in this way we have mpg-nested maximal planar graphs $H_k=H_{k-1}[\ominus ^{cyc}_3]T_k$ with $k\in [1,m]$ and $H_0=G$. For simplicity, we write $H_m$ as $H_m=G[\ominus ^{cyc}_3]^m_{k=1}T_k$, and moreover $H_m=G[\ominus ^{cyc}_3]^m_{k=1}\langle a_kG_k\rangle$ in general. We get a \emph{mpg-nested maximal planar graphic lattice} as follows
\begin{equation}\label{eqa:mpg-nested-maximal-planar-graphic-lattice}
\textbf{\textrm{L}}(F_{mpg}(n)[\ominus ^{cyc}_3]Z^0\textbf{\textrm{M}}_{\textbf{pg}})=\left \{G[\ominus ^{cyc}_3]^m_{k=1}\langle a_kG_k\rangle :a_k\in Z^0,G_k\in \textbf{\textrm{M}}_{\textbf{pg}}, G\in F_{mpg}(n)\right \}
\end{equation} with $\sum ^m_{k=1} a_k\geq 1$, where $F_{mpg}(n)$ is the set of maximal planar graphs of orders $\leq n$.

Similarly, we define the following maximal planar graph set
\begin{equation}\label{eqa:colored-mpg-nested-maximal-planar-graphic-lattice}
\textbf{\textrm{L}}(F^c_{mpg}(n)[\ominus ^{cyc}_3]Z^0\textbf{\textrm{M}}^c_{\textbf{pg}})=\left \{H[\ominus ^{cyc}_3]^m_{k=1}\langle b_kG^c_k\rangle :b_k\in Z^0,G^c_k\in \textbf{\textrm{M}}^c_{\textbf{pg}}, H\in F^c_{mpg}(n)\right \}
\end{equation} as a \emph{$4$-colored mpg-nested maximal planar graphic lattice} with $\sum ^m_{k=1} b_k\geq 1$, where $F^c_{mpg}(n)$ is the set of $4$-colored maximal planar graphs of orders no more than $n$.

\begin{problem}\label{qeu:444444}
\textbf{Configuration problem of maximal planar graphs}. By the mpg-nested maximal planar graphic lattice defined in Eq.(\ref{eqa:mpg-nested-maximal-planar-graphic-lattice}) and the $4$-colored mpg-nested maximal planar graphic lattice defined in Eq.(\ref{eqa:colored-mpg-nested-maximal-planar-graphic-lattice}), if two maximal planar graphs $G_j$ and $G^c_j$ hold $G_j\cong G^c_j$ with $j\in [1,m]$ for two bases $\textbf{\textrm{M}}_{\textbf{pg}}$ and $\textbf{\textrm{M}}^c_{\textbf{pg}}$, \textbf{does} each maximal planar graph $G^*\in \textbf{\textrm{L}}(F_{mpg}(n)[\ominus ^{cyc}_3]Z^0\textbf{\textrm{M}}_{\textbf{pg}})$ correspond a maximal planar graph $H^*\in \textbf{\textrm{L}}(F^c_{mpg}(n)[\ominus ^{cyc}_3]Z^0\textbf{\textrm{M}}^c_{\textbf{pg}})$ holding $G^*\cong H^*$ true, and vice versa? Moreover, if the base $\textbf{\textrm{M}}_{\textbf{pg}}$ contains all configurations of maximal planar graphs, is there a graph $J^*\in \textbf{\textrm{L}}(F_{mpg}(n)[\ominus ^{cyc}_3]Z^0\textbf{\textrm{M}}_{\textbf{pg}})$ to be a new configuration?
\end{problem}

\begin{figure}[h]
\centering
\includegraphics[width=14.4cm]{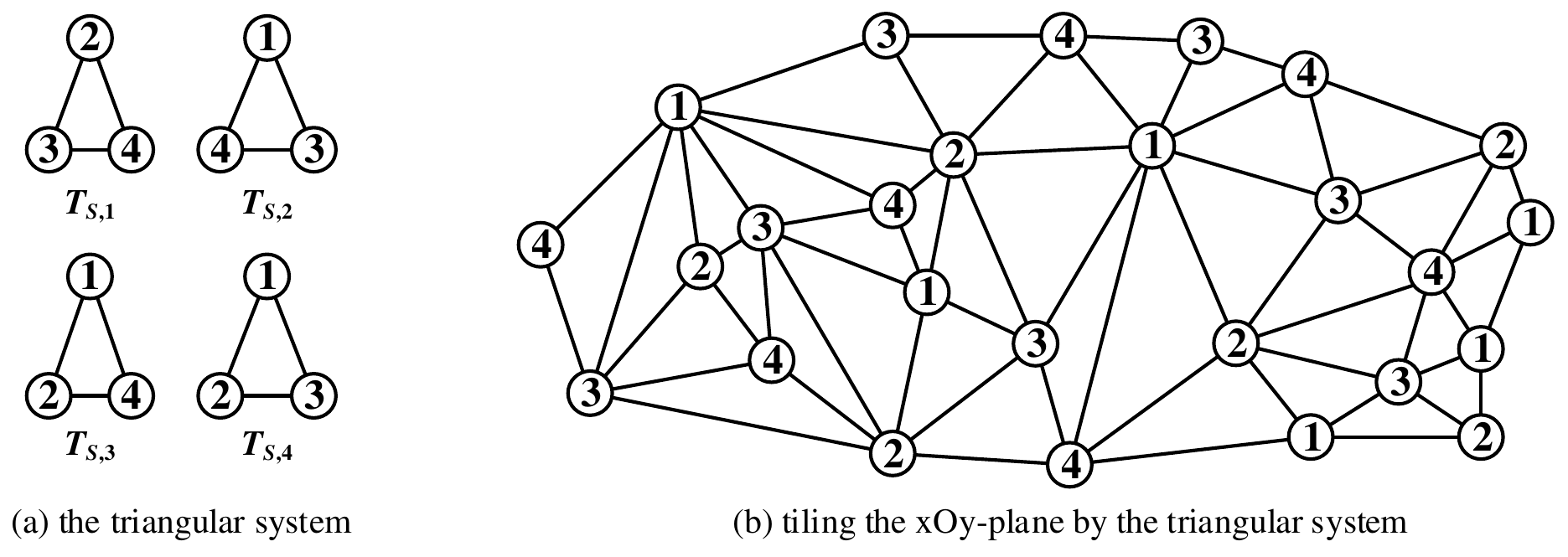}\\
\caption{\label{fig:tiling-xOy-plane}{\small (a) The triangular system; (b) tiling the $xOy$-plane by the triangular system, the vertex-coinciding operation and the edge-coinciding operation.}}
\end{figure}

\begin{problem}\label{qeu:444444}
By the triangular system $\textbf{\textrm{T}}_{sys}=\{T_{S,1},T_{S,2},T_{S,3},T_{S,4}\}$ shown in Fig.\ref{fig:tiling-xOy-plane}, the vertex-coinciding operation and the edge-coinciding operation, we tile the $xOy$-plane, the resultant plane is called a \emph{triangular $4$-colored-plane}, denoted as $T4$-plane. At the infinity of the $T4$-plane, is the $T4$-plane's boundary a vertex? Or is the $T4$-plane's boundary a triangle? Or does the $T4$-plane's boundary have finite vertices? Or do there are infinite vertices on the $T4$-plane's boundary?
\end{problem}

We apply the triangular system $\textbf{\textrm{T}}_{sys}=\{T_{S,1},T_{S,2},T_{S,3},T_{S,4}\}$ shown in Fig.\ref{fig:tiling-xOy-plane} to build up a graphic lattice in the following way: By the vertex-coinciding operation and the edge-coinciding operation, we do the edge-coinciding operation to a group of colored triangles $T_1,T_2,\dots ,T_M$, where $T_1,T_2,\dots ,T_M$ are a permutation of colored triangles $a_1T_{S,1},a_2T_{S,2},a_3T_{S,3},a_4T_{S,4}$ from the triangular system $\textbf{\textrm{T}}_{sys}$ with $a_k\in Z^0$ and $M=\sum ^4_{i=1}a_i$, the resultant graphs are planar and denoted as $\overline{\ominus } |^4_{k=1}a_kT_{S,k}$. So, we have a \emph{triangular system lattice}
\begin{equation}\label{eqa:triangular-system-lattice-edge-coinciding}
\textbf{\textrm{L}}(Z^0\overline{\ominus }\textbf{\textrm{T}}_{sys})=\big \{\overline{\ominus } |^4_{k=1}a_kT_{S,k}:~a_k\in \in Z^0,T_{S,k}\in \textbf{\textrm{T}}_{sys}\big \}
\end{equation} with $\sum^4_{k=1}a_k\geq 3$, such that each graph of $\textbf{\textrm{L}}(Z^0\overline{\ominus }\textbf{\textrm{T}}_{sys})$ is planar.

\begin{problem}\label{qeu:triangular-system-lattice-edge-coinciding}
What conditions do integers $a_1,a_2,a_3,a_4$ shown in Eq.(\ref{eqa:triangular-system-lattice-edge-coinciding}) satisfy to guarantee that the resultant graphs $\overline{\ominus } |^4_{k=1}a_kT_{S,k}$ to be maximal planar graphs and $4$-colorable well?
\end{problem}

\subsubsection{Examples of topological authentications based on planar graphs}

\begin{example}\label{exa:8888888888}
In Fig.\ref{fig:44-more-public-keys}, we have eight topological authentications $H_i=P_i[\ominus ^{cyc}_{m_i}]T_i$ for $i\in [1,8]$, where each public-key $P_i\in G_{pub}$ and each private-key $T_i\in G_{pri}$. In practical application, we show the first public-key $G$ to be a \emph{proper vertex $4$-coloring characterized graph} shown in Fig.\ref{fig:4-coloring-character}, and then give the public-key set $G_{pub}$, as well as the private-key set is $G_{pri}$ that will produce the desired number-based string for decryption. The 4-coloring $f_i$ of each topological authentication $H_i=P_i[\ominus ^{cyc}_{m_i}]T_i$ for $i\in [1,8]$ is determined by the 4-coloring $g_i$ of each public-key $P_i$ and the 4-coloring $h_i$ of each public-key $T_i$, such that $\{f_1,f_2,\dots ,f_8\}$ is the vertex set of the first public-key $G$ shown in Fig.\ref{fig:Kempe-type-4-base-module}.

In a topological authentication $\textbf{T}_{\textbf{a}}\langle\textbf{X},\textbf{Y}\rangle $ defined in Definition \ref{defn:topo-authentication-multiple-variables}, we have the variable vectors $P_{ub}(\textbf{X})=(G_{pub}, g_1,g_2$, $\dots $, $g_8)$ and $P_{ri}(\textbf{Y})=(G_{pri}, h_1,h_2,\dots ,h_8)$, and the operation vector $\textbf{F}=([\ominus^{cyc}_{m_i}],f_1,f_2,\dots ,f_8 )$, such that $G_{pub}\rightarrow G_{pri}$ made by $H_i=P_i[\ominus^{cyc}_{m_i}]T_i$ with 4-coloring $f_i=g_i\uplus h_i$ for $m_i\geq 3$ and $i\in [1,8]$.
\end{example}

\begin{figure}[h]
\centering
\includegraphics[width=15.4cm]{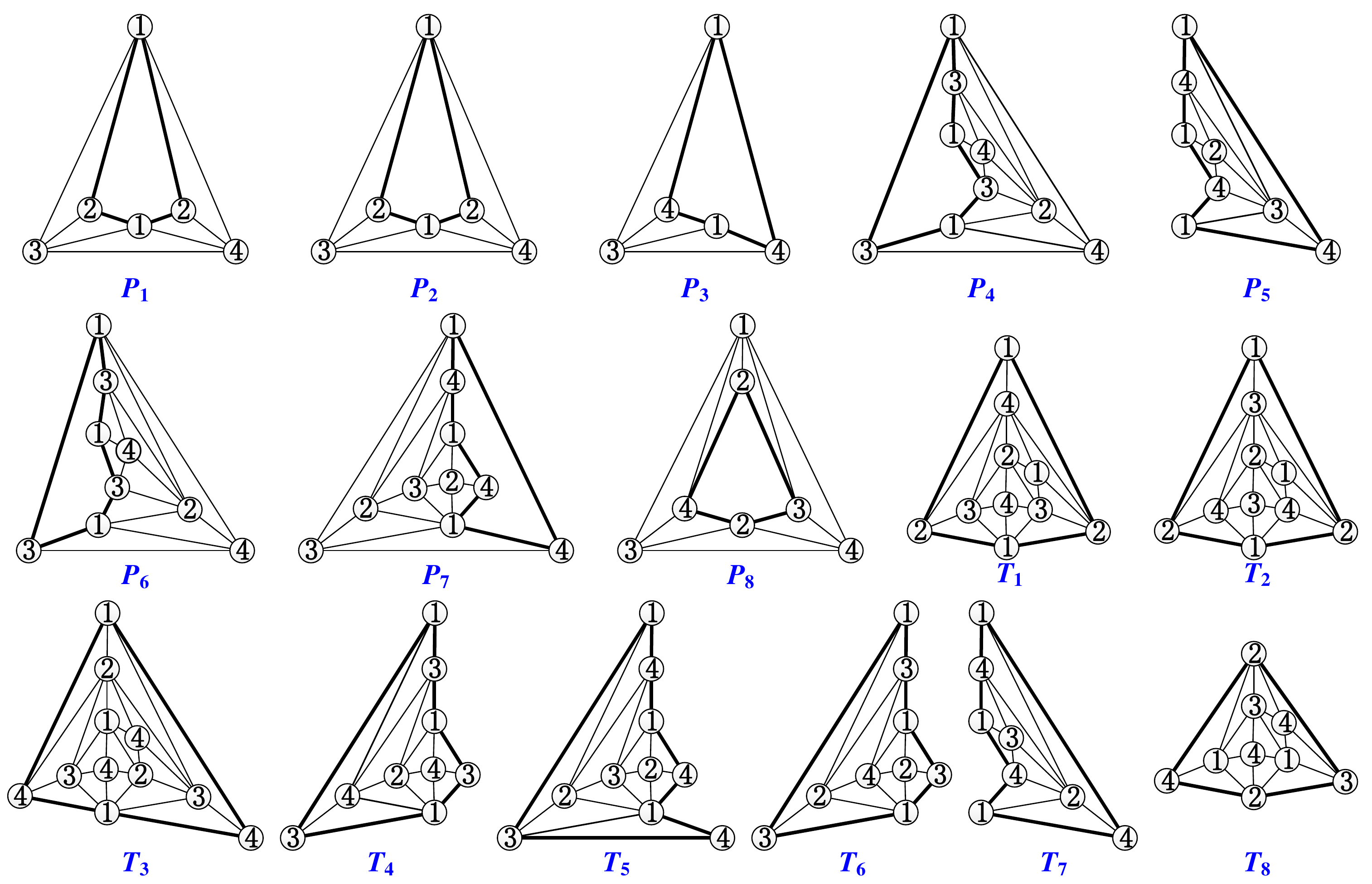}\\
\caption{\label{fig:44-more-public-keys} {\small A public-key group $G_{pub}=(P_1,P_2,\dots ,P_8)$ and a private-key group $G_{pri}=(T_1,T_2,\dots ,T_8)$.}}
\end{figure}

\begin{example}\label{exa:semi-maximal-planar-graphs-4-coloring}
Use the notation and the terminology appeared in Definition \ref{defn:topo-authentication-multiple-variables}. A maximal planar graph $G=G^C_{out}[\ominus^C_k]G^C_{in}$ admits a proper vertex $4$-coloring $f$, so the semi-maximal planar graph $G^C_{out}$ admits a proper vertex $4$-coloring $f_{out}$ induced by $f$, and the semi-maximal planar graph $G^C_{in}$ admits a proper vertex $4$-coloring $f_{in}$ induced by $f$. So, we have a topological public-key
\begin{equation}\label{eqa:semi-maximal-planar-graph-au-11}
P_{ub}(\textbf{X})=\big (G^C_{out}, f_{out}, T_{code}(G^C_{out}),A(G^C_{out})\big )
\end{equation}
and a topological private-key
\begin{equation}\label{eqa:semi-maximal-planar-graph-au-22}
P_{ri}(\textbf{Y})=\big (G^C_{in},f_{in}, T_{code}(G^C_{in}),A(G^C_{in})\big )
\end{equation}
and an operation vector
\begin{equation}\label{eqa:semi-maximal-planar-graph-au-33}
\textbf{F}=\left (\theta_1,\theta_2,\theta_3,\theta_4\right)=\Big ([\ominus^C_k],f_{out}\cup f_{in},T_{code}(G^C_{out})\uplus T_{code}(G^C_{in}), A(G)\Big )
\end{equation}
where $A(G^C_{out})$, $A(G^C_{in})$ and $A(G)$ are the adjacent matrices of three graphs $G$, $G^C_{out}$ and $G^C_{in}$. Notice that two adjacent matrices $A(T)$ and $A(L)$ are not similar from each other if two non-isomorphic graphs $T$ and $L$, there is no matrix $B$ holding $B^{-1}A(T)B=A(L)$ true.

By Eq.(\ref{eqa:semi-maximal-planar-graph-au-11}), Eq.(\ref{eqa:semi-maximal-planar-graph-au-22}) and Eq.(\ref{eqa:semi-maximal-planar-graph-au-33}), we get a topological authentication $P_{ub}(\textbf{X}) \rightarrow _{\textbf{F}} P_{ri}(\textbf{Y})$ made by

$G=\theta_1(G^C_{out},G^C_{in})=G^C_{out}[\ominus^C_k]G^C_{in}$, $\theta_2(f_{out}, f_{in})=f_{out}\cup f_{in}$,

$\theta_3(T_{code}(G^C_{out}), T_{code}(G^C_{in}))=T_{code}(G^C_{out})\uplus T_{code}(G^C_{in})$, and

$A(G)=\theta_4(A(G^C_{out},A(G^C_{in}))=A(G^C_{out}\cup A(G^C_{in})$.

\vskip 0.4cm

Moreover, we may meet a set of topological public-keys consisted of maximal planar graphs $G_i=G^C_{out}[\ominus^C_k]G^C_{in}(i)$ for $i\in [1,m]$ with $m\geq 2$, and each semi-maximal planar graph $G^C_{in}(i)$ admits a proper vertex $4$-coloring $f^i_{in}$, so we get a group of topological private-key vectors
\begin{equation}\label{eqa:555555}
P_{ri}(\textbf{Y}_i)=\Big (G^C_{in}(i),f^i_{in}, T_{code}(G^C_{in}(i)),A(G^C_{in}(i))\Big )
\end{equation}
and a group of operation vectors
\begin{equation}\label{eqa:555555}
\textbf{F}_i=\left (\theta^i_1,\theta^i_2,\theta^i_3,\theta^i_4\right)=\Big ([\ominus^C_k],f_{out}\cup f^i_{in},T_{code}(G^C_{out})\uplus T_{code}(G^C_{in}(i)), A(G_i)\Big )
\end{equation} for $i\in [1,m]$. Thereby, we get a group of topological authentications
\begin{equation}\label{eqa:topo-authentication-group}
\textbf{T}_{\textbf{a}}\langle\textbf{X},\textbf{Y}_i\rangle =P_{ub}(\textbf{X}) \rightarrow _{\textbf{F}_i} P_{ri}(\textbf{Y}_i),~i\in [1,m]
\end{equation}
\end{example}

\begin{example}\label{exa:characterized-graph-vs-mpgs}
Use the concepts and notations in Definition \ref{defn:topo-authentication-multiple-variables} and Definition \ref{defn:4-coloring-characterized-graph}. Let $C_c(G)$ be a proper vertex $4$-coloring characterized graph of a planar graph $G$ with its coloring set $C_4(G)=\{f_1, f_2,\dots ,f_n\}$ of all different proper vertex $4$-colorings. We select randomly a graph $J$ admitting a proper vertex coloring $F$, and make a topological public-key
\begin{equation}\label{eqa:555555}
P_{ub}(\textbf{X}_{\textrm{characg}})=(J, F)
\end{equation}
where ``characg = characterized graph'', and \textbf{find} out a maximal planar graph $H$ with its 4-coloring set $C_4(H)$ as a topological private-key below
\begin{equation}\label{eqa:555555}
P_{ri}(\textbf{Y}_{\textrm{mpg}})=(H,C_4(H))
\end{equation} where ``mpg = maximal planar graph'', and \textbf{determine} an operation vector $\textbf{F}=(\theta_1,\theta_2)$, where $\theta_1(J)\rightarrow H$, $\theta_2(F)\rightarrow C_4(H)$, that is, $J=C_c(H)$ and $F(J)=C_4(H)$. Thereby, we obtain a topological authentication
\begin{equation}\label{eqa:characterized-graph-vs-mpgs}
\textbf{T}_{\textbf{a}}\langle\textbf{X}_{\textrm{characg}},\textbf{Y}_{\textrm{mpg}}\rangle =P_{ub}(\textbf{X}_{\textrm{characg}}) \rightarrow _{\textbf{F}} P_{ri}(\textbf{Y}_{\textrm{mpg}})
\end{equation} However, realizing the topological authentication $\textbf{T}_{\textbf{a}}\langle\textbf{X}_{\textrm{characg}},\textbf{Y}_{\textrm{mpg}}\rangle$ defined in Eq.(\ref{eqa:characterized-graph-vs-mpgs}) is very difficult.
\end{example}

\begin{example}\label{exa:characterized-graph-vs-general-graph}
Let $C^W_{\textrm{harac}}$ be a $W$-type coloring characterized graph admitting a coloring $F$ defined on a $W$-type coloring set $C_{olor}=\{f_1, f_2,\dots ,f_m\}$ (refer to Definition \ref{defn:W-type-coloring-characterized-graph}). We have a topological public-key vector $P_{ub}(\textbf{X}_W)=(C^W_{\textrm{harac}}, F)$, and we will \textbf{determine} a topological private-key graph $H$ admitting $W$-type colorings to form a topological private-key vector $
P_{ri}(\textbf{Y}_W)=(H,C_W(H))$, where $C_W(H)$ is the set of all different $W$-type colorings of $H$; and \textbf{find} an operation vector $\textbf{F}_W=(\theta_1,\theta_2)$ with
$$
\theta_1 \big (C^W_{\textrm{harac}}\big )\rightarrow H,~\theta_2(F)\rightarrow C_W(H)
$$ that is, $C^W_{\textrm{harac}}=C^W_{\textrm{harac}}(H)$ and $F\big (C^W_{\textrm{harac}}\big )=C_{olor}=C_W(H)$, we get a topological authentication
\begin{equation}\label{eqa:characterized-graph-vs-general-graph}
\textbf{T}_{\textbf{a}}\langle\textbf{X}_W,\textbf{Y}_W\rangle =P_{ub}(\textbf{X}_W) \rightarrow _{\textbf{F}_W} P_{ri}(\textbf{Y}_W)
\end{equation}
\end{example}

\begin{defn}\label{defn:99-4-color-star-graphic-lattice}
\cite{Bing-Yao-2020arXiv} \textbf{The 4-color ice-flower system.} Each star $K_{1,d}$ with $d\in [2,M]$ admits a proper vertex-coloring $f_i$ with $i\in [1,4]$ defined as $f_i(x_0)=i$, $f_i(x_j)\in [1,4]\setminus\{i\}$, and $f_i(x_s)\neq f_i(x_t)$ for some $s,t\in [1,d]$, where $V(K_{1,d})=\{x_0,x_j:j\in [1,d]\}$. For each pair of $d$ and $i$, $K_{1,d}$ admits $n(d,i)$ proper vertex-colorings like $f_i$ defined above. Such colored star $K_{1,d}$ is denoted as $P_dS_{i,k}$, we have a set $(P_{d}S_{i,k})^{n(a,i)}_{k=1}$ with $i\in [1,4]$ and $d\in [2,M]$, and moreover we obtain a \emph{4-color ice-flower system} $\textbf{\textrm{I}}_{ce}(PS,M)=I_{ce}(P_{d}S_{i,k})^{n(a,i)}_{k=1})^{4}_{i=1})^{M}_{d=2}$, which induces a \emph{$4$-color star-graphic lattice}
\begin{equation}\label{eqa:4-color-star-system-lattices}
\textbf{\textrm{L}}(\Delta\overline{\ominus} \textbf{\textrm{I}}_{ce}(PS,M))=\left \{\Delta\overline{\ominus}|^{A}_{(d,i,k)} a_{d,i,k}P_{d}S_{i,j}:~a_{d,i,k}\in Z^0,~P_{d}S_{i,j}\in \textbf{\textrm{I}}_{ce}(PS,M)\right \}
\end{equation} for $\sum ^{A}_{(d,i,k)} a_{d,i,k}\geq 3$ with $A=|\textbf{\textrm{I}}_{ce}(PS$, $M)|$, and the operation ``$\Delta\overline{\ominus}$'' is doing a series of leaf-coinciding operations to colored stars $a_{d,i,k}P_{d}S_{i,j}$ such that the resultant graph to be a planar graph with each inner face being a triangle.\qqed
\end{defn}

\begin{conj}\label{conj:0000000000}
$^*$ The planar dual graph $G^*$ of a maximal planar graph $G$ can be decomposed into a spanning tree $T$ and a perfect matching $M$ such that $E(G^*)=E(T)\cup M$.
\end{conj}

\begin{problem}\label{qeu:444444}
Adding the edges of an edge set $E^*\not \subset E(G)$ to a maximal planar graph $G$ admitting 4-colorings produces an edge-added graph $G+E^*$ admitting a 4-coloring. \textbf{Determine} the edge-added graph set $\{G+E^*\}$.
\end{problem}

\subsubsection{Planar graphs and Hanzi-graphs}

In Fig.\ref{fig:planar-split-hanzi-graphs}, there are a vertex-coincided graph $G=H_{4043}\odot H_{5448}\odot H_{3164}\odot H_{2201}$ and a graph $H$ that can be vertex-split into four Hanzi-graphs $H_{4043}, H_{5448}, H_{3164}, H_{2201}$. Clearly, coloring two graphs $H^*$ and $G^*$ is difficult although $\langle H_{4043}, H_{5448}, H_{3164}, H_{2201}\mid H^*\rangle$ and $\langle H_{4043}, H_{5448}, H_{3164}, H_{2201}\mid G^*\rangle$, since there are a huge number of colorings and labelings in graph theory.

\begin{rem}\label{rem:333333}
We point, in general, that vertex-splitting a given graph (as a \emph{public-key}) into a family of Hanzi-graphs $H_1, H_2, \dots, H_m$ (as a family of \emph{private-keys}) is not a slight work, even a challenging job. Another reason there are many Chinese sentences made by a fixed group of Hanzi-graphs, see three Chinese sentences $H_{3164}H_{5448}H_{4043}H_{2201}$, $H_{4043}H_{2201}H_{3164}H_{5448}$ and $H_{4043}H_{5448}H_{3164}H_{2201}$ shown in Fig.\ref{fig:planar-split-hanzi-graphs}.\paralled
\end{rem}

\begin{figure}[h]
\centering
\includegraphics[width=16.4cm]{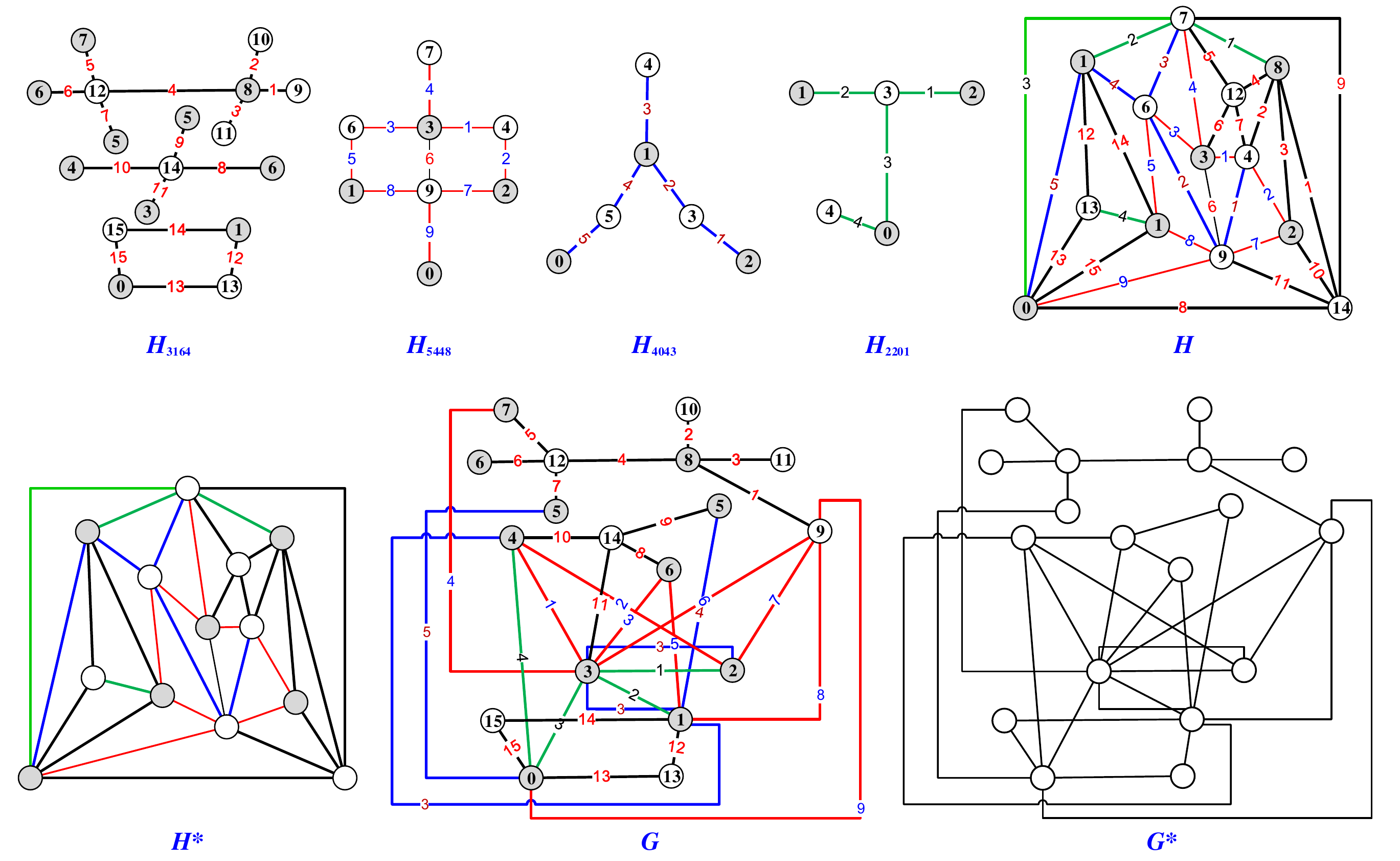}\\
\caption{\label{fig:planar-split-hanzi-graphs}{\small A maximal planar graph $H$ can be vertex-split into four Hanzi-graphs $H_{3164}, H_{5448}, H_{4043}$ and $H_{2201}$, and a non-planar graph $G$ is obtained by these Hanzi-graphs by the vertex-coinciding operation.}}
\end{figure}

\begin{problem}\label{qeu:444444}
\textbf{Vertex-split} a planar graph $G$ (as a \emph{public-key}) into $k$ different Hanzi-graph groups $L_{\textrm{Hanzi}}(j,\wedge G)=\{H_{j,1},H_{j,2} ,\dots ,H_{j,m_j}\}$ with $j\in [1,k]$, as a group of \emph{public-keys}, such that each Hanzi-graph group $L_{\textrm{Hanzi}}(j,\wedge G)$ expresses a complete sentence $H_{j,1}H_{j,2}\cdots H_{j,m_j}$ in Chinese.
\end{problem}

\begin{problem}\label{qeu:444444}
For a given Chinese sentence $W=H_{a_1b_1c_1d_1}H_{a_2b_2c_2d_2}\cdots H_{a_mb_mc_md_m}$ with each $a_jb_jc_jd_j$ defined in \cite{GB2312-80}, let $T_{ile}(W)$ be the set of planar graphs obtained by tiling the Chinese sentence $W$ into planar graphs by the vertex-splitting operation and the vertex-coinciding operation, \textbf{determine}

(i) the graph set $T_{ile}(W)$; and

(ii) all maximal planar graphs in $T_{ile}(W)$.
\end{problem}

\subsection{Methods of generating public-keys}

There are many ways for producing number-based strings in topological coding, for example,
\begin{asparaenum}[$\bullet$]
\item Degree-sequence.
\item Degree-sequence lattice.
\item Complex degree-sequence.
\item Degree-sequence homomorphisms.
\item Various matrices of colored graphs.
\item Graphic groups.
\item Graphic lattices.
\item Graph homomorphisms.
\item Hanzi codes.
\item Codes of shapes and visual colors of vertices and edges.
\item Multiple combinatoric technique.
\end{asparaenum}

\subsubsection{Number-based strings for public-keys}

\begin{asparaenum}[(i)]
\item For a number-based string $s(n)=c_1c_2\cdots c_n$ with $c_i\in [0,9]$, we may recognize $s(n)$ as one of a degree-sequence, a vector, a number-based string generated from some matrix or more, \emph{etc}.
\item A number-based string $s(n)$ matches with another number-based string $s\,'(t)$ such that they form a $W$-type coloring matching $\langle s(n), s\,'(t)\rangle$ in graph theory.
\item Decompose a positive integer $m$ into a group of positive integers $m_1,m_2,\dots ,m_r$, that is, $m=m_1+m_2+\cdots +m_r$, such that these positive integers form a vector $(m_1,m_2,\dots ,m_r)$, or a degree-sequence, a number-based string $s_j(r)=m_{j_1}m_{j_2}\cdots m_{j_r}$ made by a permutation $m_{j_1},m_{j_2},\dots ,m_{j_r}$ of the positive integers $m_1,m_2,\dots ,m_r$.
\item For a \emph{number-based hyper-string} $R_n=r_1r_2\cdots r_n$ is defined as $r_j=c_{j,1}c_{j,2}\cdots c_{j,m_j}$ with $c_{j,k}\in [0,9]$ for $k\in [1,m_j]$ and $j\in [1,n]$, we rewrite $R_n$ as $R_n=[+]\langle A_n, B_n\rangle$, or $R_n=[-]\langle A_n, B_n\rangle$, or $R_n=[\times]\langle A_n, B_n\rangle$, or $R_{2m}=[\ominus]\langle A_m, B_m\rangle$, and or $R_n=[\cup]\langle A_n, B_n\rangle$ define in Eq.(\ref{eqa:operations-number-based-trings}), however, there is no general way for finding two number-based strings $ A_n$ and $B_n$.
\end{asparaenum}

\subsubsection{Dual-type, twin-type, image-type labelings and colorings}

An object and its own \emph{dual} or \emph{complimentary} of graph theory form an evident pair of public-key and private-key. Very often, $W$-type matching labelings give expression to these objects and their own dual or complimentary in graph theory, such as dual-type labelings including edge-dual labeling, vertex-dual labeling, total dual labeling (see Definition \ref{defn:total-image-dual}), dual colorings (see Definition \ref{defn:universal-label-set}, Definition \ref{defn:4-dual-total-coloring}), planar dual graph and planar dual matching $\langle H,D_{ual}(H)\rangle $ defined in Definition \ref{defn:44-planar-dual-graph}, dual Topcode-matrix group, dual tree group, twin odd-graceful and odd-elegant labelings, inverse labeling, 6C-complimentary matching, edge-image labeling, vertex-image labeling, totally image labeling (see Definition \ref{defn:total-image-dual}), $W$-type e-image total coloring (see Definition \ref{defn:k-d-gracefully-e-image-total-coloring}), twin $(k,d)$-labelings, and so on.

\subsubsection{Set-dual labelings, Edge-difference magic-type total labelings}

\begin{defn} \label{defn:k-sedge-difference-magically-total-labeling}
\cite{Yao-Wang-2106-15254v1} Let $G$ be a connected $(p,q)$-graph. If there are integers $k$ and $\lambda~(\neq 0)$ such that $G$ admits a proper total labeling $h:V(G)\cup E(G)\rightarrow [1,p+q]$ holding $|h(u)-h(v)|=k+\lambda h(uv)$ for each edge $ uv\in E(G)$ true, then we call $h$ a $(k,\lambda)$-\emph{edge-difference magically total labeling} of $G$, $k$ a \emph{magical constant} and $\lambda$ a \emph{balanced number}, and moreover $f$ is \emph{super} if the vertex color set $f(V(G))=[1,p]$.\qqed
\end{defn}

\begin{rem}\label{rem:edge-difference-magically-total}
It is noticeable: (i) A $(k,\lambda )$-magically total labeling implies some \emph{edge-magic total labeling} and some \emph{felicitous-difference total coloring}, since $f(u)+f(v)=k+\lambda f(uv)$ is related with $f(u)+f(uv)+f(v)=k$ and $f(u)+f(v)-f(uv)=k$ as $\lambda=-1$.

(ii) Because of $|h(u)-h(v)|=k+\lambda h(uv)$ is related with $h(uv)+|h(u)-h(v)|=k$ and $\big |h(uv)-|h(u)-h(v)|\big |=k$ as $\lambda=-1$, a $(k,\lambda)$-edge-difference magically total labeling defined in Definition \ref{defn:k-sedge-difference-magically-total-labeling} implies some edge-magic graceful total labeling, some edge-difference magic total coloring and some graceful ev-difference magic total coloring.\paralled
\end{rem}

\begin{defn} \label{defn:edge-difference-magic-k-d-labeling}
$^*$ A bipartite connected $(p,q)$-graph $G$ with $V(G)=X\cup Y$ and $X\cap Y=\emptyset$ admits a labeling $f:V(G)\cup E(G)\rightarrow S_{Md}\cup S_{k,(q-1)d}$, such that $f(x)\in S_{Md}=\{0,d,2d,\dots , Md\}$ for $x\in X$, and $f(y)\in S_{k,(q-1)d}=\{k,k+d,k+2d,\dots , k+(q-1)d\}$ for $y\in Y$, as well as there are two constants $k^*\geq 0$ and $\lambda \neq 0$ holding $|f(y)-f(x)|=k^*+\lambda f(xy)$ for each edge $xy\in E(G)$. We call $f$ a $(k,d)$-\emph{edge-difference $(k^*,\lambda)$-magically total labeling} of $G$.\qqed
\end{defn}

\begin{thm}\label{thm:666666}
$^*$ Suppose that a bipartite $(p,q)$-graph $G$ admits a set-ordered graceful labeling, then $G$ admits one of a set-ordered pan-graceful labeling and a $(k,d)$-edge-difference $(k^*,\lambda)$-magically total labeling.
\end{thm}
\begin{proof} Suppose that a bipartite $(p,q)$-graph $G$ admits a set-ordered graceful labeling $f$, we have $V(G)=X\cup Y$ with $X\cap Y=\emptyset$, where $X=\{x_1,x_2,\dots ,x_s\}$ and $Y=\{y_1,y_2,\dots ,y_t\}$ with $s+t=p=|V(G)|$. Since $\max f(X)<\min f(Y)$, we have
$$
0=f(x_1)<f(x_2)<\cdots <f(x_s)<f(y_1)<f(y_2)<\cdots <f(y_t)=q.
$$

(A) If $G$ is a tree, then $f(x_i)=i-1$ for $i\in[1,s]$, and $f(y_j)=s-1+j$ for $j\in[1,t]$, and
\begin{equation}\label{eqa:555555}
{
\begin{split}
f(E(G))&=\{f(x_iy_j)=f(y_j)-f(x_i):x_iy_j\in E(G)\}\\
&=\{f(x_iy_j)=s+j-i:x_iy_j\in E(G)\}\\
&=[1,q]
\end{split}}
\end{equation}

We define another labeling $g$ for $G$ as: $g(x_i)=f(x_s)-f(x_i)=(s-1)-(i-1)=s-i$ for $i\in [1,s]$, $g(y_j)=f(y_1)+f(y_t)-f(y_j)=s+q-(s-1+j)=q+1-j$ for $j\in [1,t]$, and

$$g(x_iy_j)=g(y_j)-g(x_i)=q+1-j-s+i=q+1-(s+j-i)\rightarrow [1,q]
$$ since $f(x_iy_j)=s+j-i$ for $x_iy_j\in E(G)$.

Clearly, $g$ is a set-ordered graceful labeling of $G$, and we call this labeling $g$ \emph{set-dual labeling} of the set-ordered graceful labeling $f$ of $G$.

(B) If $G$ is not a tree, we define a new labeling $h$ for $G$ as: $h(x_i)=f(x_s)+f(x_1)-f(x_i)=f(x_s)-f(x_i)$ for $i\in [1,s]$, $h(y_j)=f(y_1)+f(y_t)-f(y_j)=q+f(y_1)-f(y_j)$ for $j\in [1,t]$, and

$${
\begin{split}
h(x_iy_j)&=h(y_j)-h(x_i)=q+f(y_1)-f(y_j)-[f(x_s)-f(x_i)]\\
&=q+f(y_1)-f(x_s)-[f(y_j)-f(x_i)]\\
&=q+f(y_1)-f(x_s)-f(x_iy_j) \rightarrow [f(y_1)-f(x_s),f(y_1)-f(x_s)+q-1]
\end{split}}
$$ since $f(x_iy_j)\rightarrow [1,q]$ for $x_iy_j\in E(G)$.

We call $h$ a \emph{set-ordered pan-graceful labeling} of $G$ because
$$h(E(G))=\{h(x_iy_j)=h(y_j)-h(x_i):x_iy_j\in E(G)\}=[f(y_1)-f(x_s),f(y_1)-f(x_s)+q-1]
$$ also, $h$ is a \emph{set-dual labeling} of the set-ordered graceful labeling $f$ of $G$.

(C) According to a set-ordered graceful labeling $f$ of a bipartite $(p,q)$-graph $G$ with $V(G)=X\cup Y$ and $X\cap Y=\emptyset$, we define a $(k,d)$-\emph{graceful labeling} $f^*$ as: $f^*(x_i)=f(x_i)\cdot d$ for $x_i\in X$, and $f^*(y_j)=k+[f(y_j)-1]\cdot d$ for $y_j\in Y$, the induced edge color
$$
f^*(x_iy_j)=k+[f(y_j)-1]\cdot d-f(x_i)\cdot d=k+[f(x_iy_j)-1]\cdot d
$$ for $x_iy_j\in E(G)$.

(C-1) Thereby, using the $(k,d)$-\emph{graceful labeling} $f^*$ defined above, we have a $(k,d)$-\emph{pan-graceful labeling} $h\,'$ of $G$ defined as: $h\,'(x_i)=f^*(x_s)+f^*(x_1)-f^*(x_i)=f^*(x_s)-f^*(x_i)$ for $i\in [1,s]$, $h\,'(y_j)=f^*(y_1)+f^*(y_t)-f^*(y_j)=k+(q-1)d+f^*(y_1)-f^*(y_j)$ for $j\in [1,t]$, and
$${
\begin{split}
h\,'(x_iy_j)&=h\,'(y_j)-h\,'(x_i)=k+(q-1)d+f^*(y_1)-f^*(y_j)-[f^*(x_s)-f^*(x_i)]\\
&=k+(q-1)d+f^*(y_1)-f^*(x_s)-f^*(x_iy_j)
\end{split}}$$ for $x_iy_j\in E(G)$, also, we call $h\,'$ a \emph{set-dual labeling} of the $(k,d)$-graceful labeling $f^*$.

(C-2) The $(k,d)$-\emph{graceful labeling} $f^*$ defined above enables us to define a $(k,d)$-\emph{edge-difference $(k^*,\lambda)$-magically total labeling} $h\,''$ of $G$ by $$h\,''(x_i)=f^*(x_s)+f^*(x_1)-f^*(x_i)=f^*(x_s)-f^*(x_i)
$$ for $i\in [1,s]$, and
$$h\,''(y_j)=f^*(y_1)+f^*(y_t)-f^*(y_j)=k+(q-1)d+f^*(y_1)-f^*(y_j)
$$ for $j\in [1,t]$, and $h\,''(x_iy_j)=f^*(x_iy_j)$ for $x_iy_j\in E(G)$. So
$${
\begin{split}
h\,''(y_j)-h\,''(x_i)&=k+(q-1)d+f^*(y_1)-f^*(y_j)-[f^*(x_s)-f^*(x_i)]\\
&=k+(q-1)d+f^*(y_1)-f^*(x_s)-f^*(x_iy_j)\\
&=k^*+\lambda h\,''(x_iy_j)
\end{split}}$$ with $k^*=k+(q-1)d+f^*(y_1)-f^*(x_s)$ and $\lambda=-1$.

This theorem had been proven.
\end{proof}

\subsubsection{Subgraphs of maximal planar graphs}

We decompose a maximal planar graph $M$ shown in Fig.\ref{fig:mpg-spanning-decom} into three edge-disjoint spanning forests $F_1,F_2$ and $F_3$, such that $E(M)=E(F_1)\cup E(F_2)\cup E(F_3)$, and $V(M)=V(F_1)=V(F_2)=V(F_3)$. Another group of three graphs $M_1,M_2, M_3$ shown in Fig.\ref{fig:mpg-spanning-decom} holds $M_i=M-E(F_i)$ and $M=\odot \langle M_i,F_i\rangle$ for $i\in [1,3]$, and moreover $M=\odot \langle F_1,F_2,F_3\rangle$. Notice that $M_1=\odot \langle F_{2},F_{3}\rangle$, $M_2=\odot \langle F_{1},F_{3}\rangle$ and $M_3=\odot \langle F_{2},F_{1}\rangle$.

\begin{figure}[h]
\centering
\includegraphics[width=16.4cm]{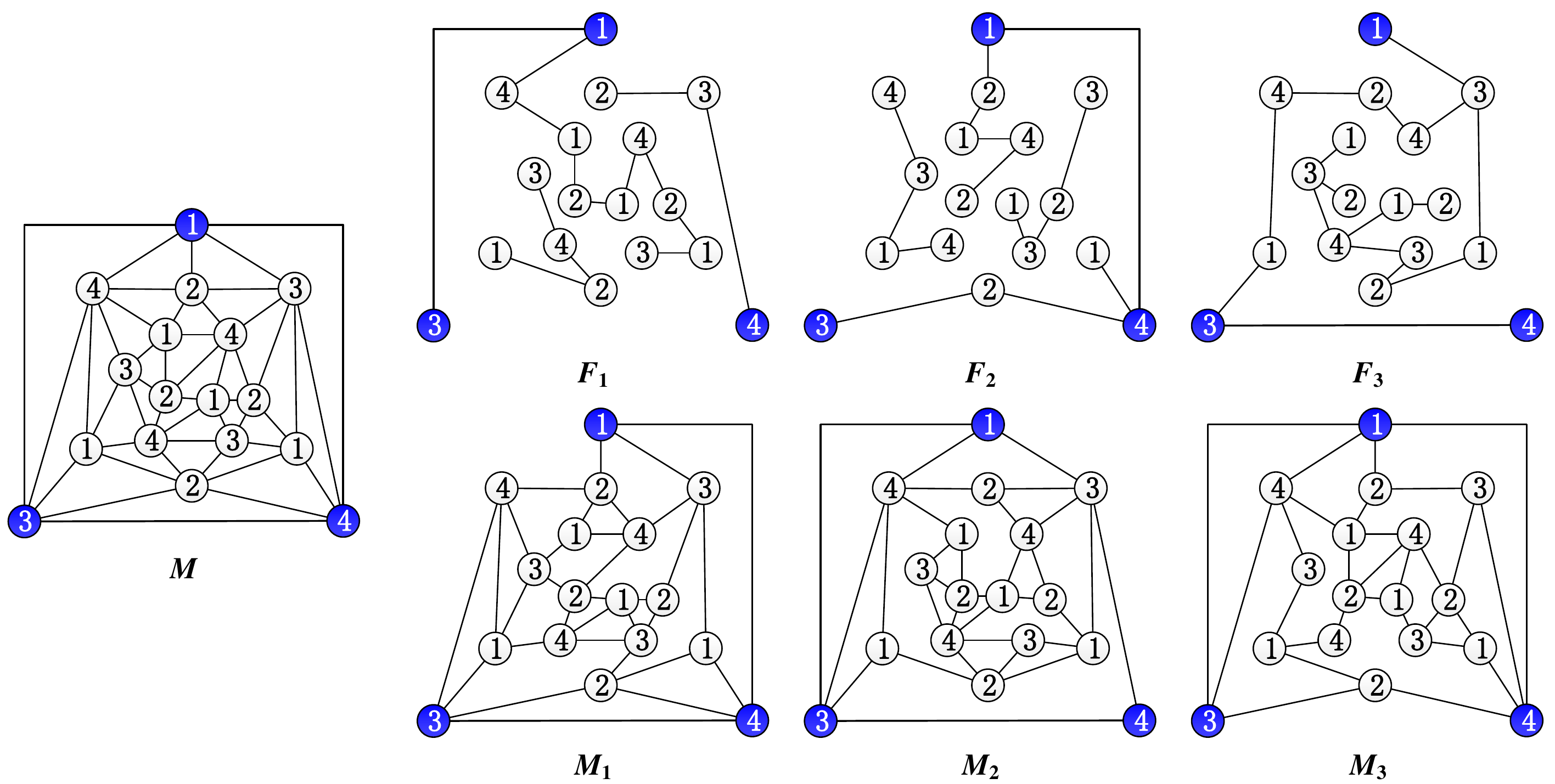}\\
\caption{\label{fig:mpg-spanning-decom} {\small Decomposing maximal planar graphs into spanning forests.}}
\end{figure}

\begin{problem}\label{qeu:444444}
Let $n_{spa}(G;F_i)$ be the number of vertex-disjoint components of a spanning forests $F_i$ of a maximal planar graph $G$, and let $S^k_{pan}(G)$ be the set of all edge-disjoint spanning forests of $G$, so we have two numbers as follows
$$
N_{spa}(G,k)=\sum_{F_i\in S^k_{pan}(G)} n_{spa}(G;F_i),~E_{spa}(G,k)=\left |\bigcup_{F_i\in S^k_{pan}(G)} E(F_i)\right |
$$ \textbf{Determine} the smallest \emph{component-number sum} $\min \{N_{spa}(G,k): S^k_{pan}(G)\}$, and the largest \emph{edge number} $\max \{E_{spa}(G,k): S^k_{pan}(G)\}$, notice that a planar $(p,q)$-graph $G$ is maximal if and only if $q=3p-6$.
\end{problem}

We present another way for making public-keys and their private-keys as follows: Taking arbitrarily a proper vertex subset $V_i=\{x_{i,1},x_{i,2},\dots ,x_{i,a_i}\}$ of a maximal planar $(p,q)$-graph $G$ admitting a proper vertex coloring $f$ with $1\leq a_i\leq \frac{p}{2}$, and each neighbor set $N(x_{i,j})=\{y_{i,j,k}: x_{i,j}y_{i,j,k}\in E(G), k\in [1,b_j]\}$ induces a number-based string $D(x_{i,j})=f(y_{i,j,1})f(y_{i,j,2})\cdots f(y_{i,j,b_j})$, there are $(b_j)!$ different number-based strings like $D(x_{i,j})$. So, we have a number-based string
\begin{equation}\label{eqa:55-vertex-coloring-neighbor}
D(V_i)=D(x_{i,1})\uplus D(x_{i,2})\uplus \cdots \uplus D(x_{i,a_i})
\end{equation}
and there are $(a_i)!$ different number-based strings like $D(V_i)$ shown in Eq.(\ref{eqa:55-vertex-coloring-neighbor}). Similarly, the subset $V(G)\setminus V_i=\{w_{i,1},w_{i,2},\dots ,w_{i,c_i}\}$ induces a number-based string
\begin{equation}\label{eqa:666-vertex-coloring-neighbor}
D(V(G)\setminus V_i)=D(w_{i,1})\uplus D(w_{i,2})\uplus \cdots \uplus D(w_{i,c_i})
\end{equation}
with $D(w_{i,j})=f(z_{i,j,1})f(z_{i,j,2})\cdots f(z_{i,j,d_j})$ for $z_{i,j,k}\in N(w_{i,j})=\{z_{i,j,k}: w_{i,j}z_{i,j,k}\in E(G), k\in [1,d_j]\}$, where there are $(d_j)!$ different number-based strings like $D(w_{i,j})$, also, there are $(c_i)!$ different number-based strings like $D(V(G)\setminus V_i)$ shown in Eq.(\ref{eqa:666-vertex-coloring-neighbor}). Then we have at least $n(V_i)$ different number-based strings, where the number $n(V_i)=(a_i)!(b_j)!(c_i)!(d_j)!$ for a proper vertex subset $V_i$. Clearly, we have $p \choose a_i$ proper vertex subsets like $V_i$, so a maximal planar $(p,q)$-graph $G$ can distributes us $n(V_i)\cdot $$p \choose a_i$ number-based strings for an integer $a_i\in [1, \lfloor \frac{p}{2}\rfloor ]$.

\subsubsection{Public-keys based on algebraic technique}

For making number-based strings with longer bytes, we present the following graphic Topcode-matrix:

\begin{defn} \label{defn:graphic-topcode-matrix}
\cite{Yao-Wang-2106-15254v1} Let $G_{x,i},G_{e,j}$ and $G_{y,k}$ be colored graphs with $i,j,k\in [1,q]$. A \emph{graphic Topcode-matrix} $G_{code}$ is defined as
\begin{equation}\label{eqa:graphic-topcode-matrix}
\centering
{
\begin{split}
G_{code}= \left(
\begin{array}{ccccc}
G_{x,1} & G_{x,2} & \cdots & G_{x,q}\\
G_{e,1} & G_{e,2} & \cdots & G_{e,q}\\
G_{y,1} & G_{y,2} & \cdots & G_{y,q}
\end{array}
\right)=
\left(\begin{array}{c}
G_X\\
G_E\\
G_Y
\end{array} \right)=(G_X,~G_E,~G_Y)^{T}
\end{split}}
\end{equation} we call $G_{x,i}$ and $G_{y,i}$ the \emph{ends} of $G_{e,i}$, and call $G_X=(G_{x,1}, G_{x,2},\dots ,G_{x,q})$, $G_E=(G_{e,1}, G_{e,2}$, $\dots $, $G_{e,q})$, $G_Y=(G_{y,1}, G_{y,2},\dots ,G_{y,q})$ \emph{graphic vectors}. Moreover the graphic Topcode-matrix $G_{code}$ is \emph{functional valued} if there is a group of functions $\varphi_1,\varphi_2,\dots ,\varphi_m$ such that $G_{e,i}=\varphi_j(G_{x,i},G_{y,i})$ for some $j\in [1,m]$ and each $i\in [1,q]$.\qqed
\end{defn}

\begin{rem}\label{rem:333333}
The graphic Topcode-matrix $G_{code}$ is a three-dimension matrix, since each graph of $G_{x,i},G_{e,j}, G_{y,k}$ corresponds its own Topcode-matrix, in other word, $G_{code}$ is a matrix of matrices, like \emph{Tensor matrix}.\paralled
\end{rem}

\begin{defn} \label{defn:stringic-topcode-matrix}
$^*$ Let $S_{x,i},S_{e,j}$ and $S_{y,k}$ be number-based strings with $i,j,k\in [1,q]$. A \emph{string Topcode-matrix} $S_{code}$ is defined by
\begin{equation}\label{eqa:stringic-topcode-matrix}
\centering
{
\begin{split}
S_{code}= \left(
\begin{array}{ccccc}
S_{x,1} & S_{x,2} & \cdots & S_{x,q}\\
S_{e,1} & S_{e,2} & \cdots & S_{e,q}\\
S_{y,1} & S_{y,2} & \cdots & S_{y,q}
\end{array}
\right)=
\left(\begin{array}{c}
S_X\\
S_E\\
S_Y
\end{array} \right)=(S_X,~S_E,~S_Y)^{T}
\end{split}}
\end{equation} $S_{x,i}$ and $S_{y,i}$ are called the \emph{ends} of $S_{e,i}$, and we call $S_X=(S_{x,1}, S_{x,2},\dots ,S_{x,q})$, $S_E=(S_{e,1}, S_{e,2}$, $\dots$, $S_{e,q})$, $S_Y=(S_{y,1}, S_{y,2},\dots ,S_{y,q})$ \emph{string vectors}. Moreover the string Topcode-matrix $S_{code}$ is \emph{functional valued} if there is a group of functions $F_1,F_2,\dots ,F_n$ such that $S_{e,i}=F_k(S_{x,i},S_{y,i})$ for some $k\in [1,n]$ and each $i\in [1,q]$.\qqed
\end{defn}

\begin{rem}\label{rem:333333}
Each \emph{graph-based string} $D_k(3q)=H_{k,1}H_{k,2}\cdots H_{k,3q}$ with $k\in [1,(3q)!]$ generated from a graphic Topcode-matrix $G_{code}$ defined in Definition \ref{defn:graphic-topcode-matrix} holds each colored graph $H_{k,i}\in \{G_{x,i},G_{e,i},G_{y,i}:i\in [1,q]\}$ and $H_{k,i}\neq H_{k,j}$ for $i\neq j$. Since the Topcode-matrix $T_{code}(H_{k,i})$ of each $H_{k,i}$ produces a number-based string $s(n_{k,i})$ with length $n_{k,i}$, then we get a \emph{number-based hyper-string} $D_k=s(n_{k,1})s(n_{k,2})\cdots s(n_{k,3q})$ with length $L_{\textrm{ength}}(D_k)=\sum ^{3q}_{i=1} n_{k,i}$ with $k\in [1,(3q)!]$.

We have compressed a number-based hyper-string $D_k$ with longer bytes into a graph-based string $D_k(3q)$ with smaller bytes, more or less like the effect of Hash function.

The elements of a graphic Topcode-matrix $G_{code}$ may belong to an \emph{every-zero graphic group} $\{F(G);\oplus \}$, so the elements of a string Topcode-matrix $S_{code}$ defined in Definition \ref{defn:stringic-topcode-matrix} may belong to an \emph{every-zero string group} $\{S(G);\oplus \}$.\paralled
\end{rem}

\subsubsection{Topological structures serving public-keys}

In topological authentication, we can design a topological public-key by one of the following cases:
\begin{asparaenum}[Pubkey-1. ]
\item An uncolored graph (resp. uncolored digraph), see two uncolored graphs $G^*$ and $H^*$ shown in Fig.\ref{fig:planar-split-hanzi-graphs}.
\item A colored graph (resp. colored digraph) $K^{mul}_4$ shown in Fig.\ref{fig:vertex-split-problem} is a public-key, and $G\rightarrow K^{mul}_4$ is a private-key, other private-keys $G_k\rightarrow G$ for $k\in [1,4]$.
\item An uncolored graph (resp. uncolored digraph) and a number-based string.
\item A colored graph (resp. colored digraph) and a number-based string.
\item A (an uncolored) colored graph (resp. digraph) and a text-based string.
\item A colored graph (resp. colored digraph) and a matrix, or a family of matrices.
\item A Topcode-matrix of a colored graph, or a family of Topcode-matrices of colored graphs.
\item A graphic group (resp. digraph group), see an example shown Fig.\ref{fig:44-more-public-keys}.
\item A graphic lattice (resp. digraph lattice).
\item A graph homomorphism chain, or a number-based string homomorphism chain, see examples shown in Fig.\ref{fig:graph-operation-homo}.
\item If $G$ admits two different $W_f$-type coloring $f$ and $g$ with no $f(uv)=g(uv)$ for $uv \in E(G)$, then $G_f\rightarrow G_g$ is a \emph{self-similar graph homomorphism authentication}.
\item Taking a number-based string along a path $P$ or a cycle $C$ in a colored graph $G$ according to a direction of clockwise and counterclockwise.
\item Taking a number-based string from a group of Hanzi-graphs, or from a larger matrix with elements being numbers.
\item One part was cut from a whole thing.
\end{asparaenum}

\begin{example}\label{exa:8888888888}
There are several topological public-keys shown in Fig.\ref{fig:various-public-keys} and Fig.\ref{fig:three-authentications}, we can see that three maximal planar graphs $P_1,P_2, P_3$ as \emph{topological public-keys} tell us no which semi-maximal planar graphs are as \emph{topological private-keys}, and four semi-maximal planar graphs $P_4,P_5, P_6,P_7$ are not complete topological public-keys, only semi-maximal planar graph $P_8$ is the complete topological public-key. $P_3$ can produce many maximal planar graphs as topological public-keys; there are many 4-colorings for two uncolored planar graphs $P_1$ and $P_4$; and planar graphs $P_2,P_5, P_6,P_7$ have uncolored vertices.

For the public-key $P_8$, three private-keys $B_1,B_2, B_3$ are shown in Fig.\ref{fig:three-authentications}, so we have three topological authentications $A_1,A_2$ and $A_3$. This is the case of one public-keys corresponding more private-keys. Clearly, there are infinite private-keys made by colored semi-maximal planar graphs, like $B_1,B_2, B_3$, for one public-key $P_8$.
\end{example}

\begin{figure}[h]
\centering
\includegraphics[width=16.4cm]{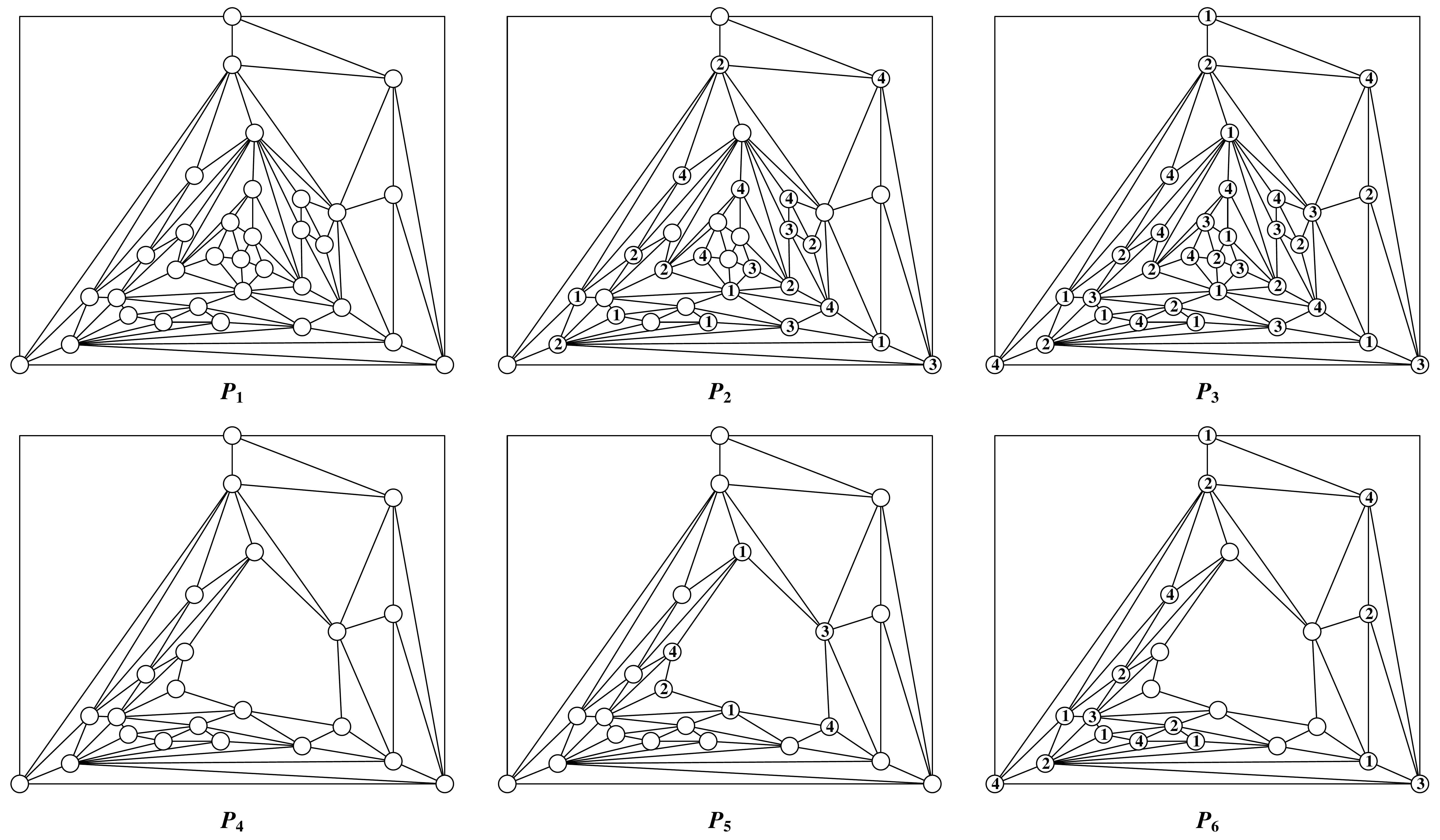}\\
\caption{\label{fig:various-public-keys}{\small Various topological public-keys.}}
\end{figure}

\begin{figure}[h]
\centering
\includegraphics[width=16.4cm]{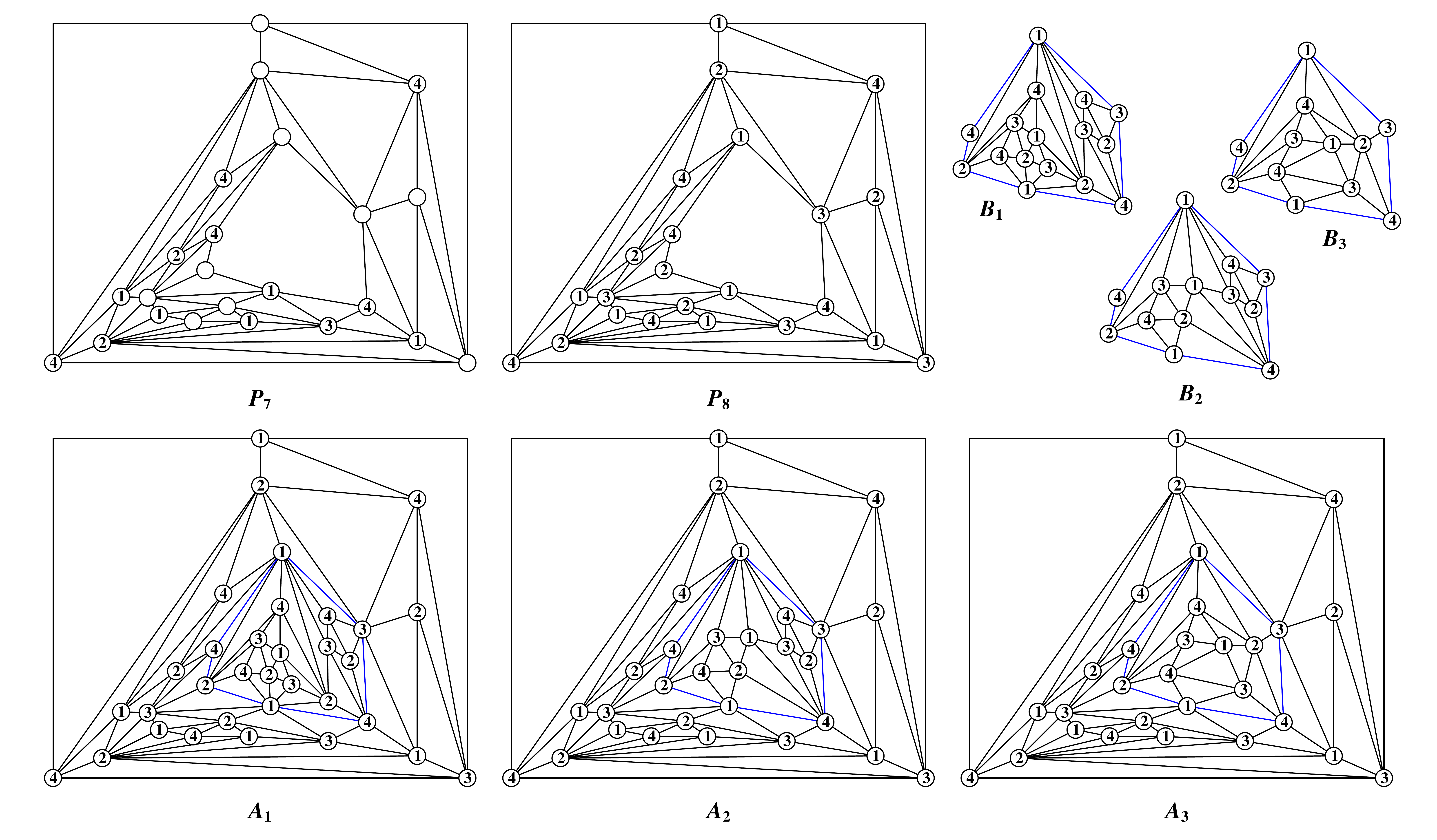}\\
\caption{\label{fig:three-authentications}{\small Three topological authentications $A_1,A_2$ and $A_3$ made by the public-key $P_8$ and the private-keys $B_1,B_2, B_3$.}}
\end{figure}

In Fig.\ref{fig:55-string-degree-sequence}, a graceful tree $T$ induces a number-based string $s(18)=$414423235641157611, we rewrite $s(18)$ as two degree-sequences $\textrm{deg}(G)=$(12, 11, 7, 6, 6, 5, 5, 4, 4, 4, 4, 3, 3, 2, 1, 1) and $\textrm{deg}(H)=$(12, 7, 6, 6, 5, 5, 4, 4, 4, 4, 3, 3, 2, 1, 1, 1, 1) of two graphs $G$ and $H$.

\begin{figure}[h]
\centering
\includegraphics[width=15.4cm]{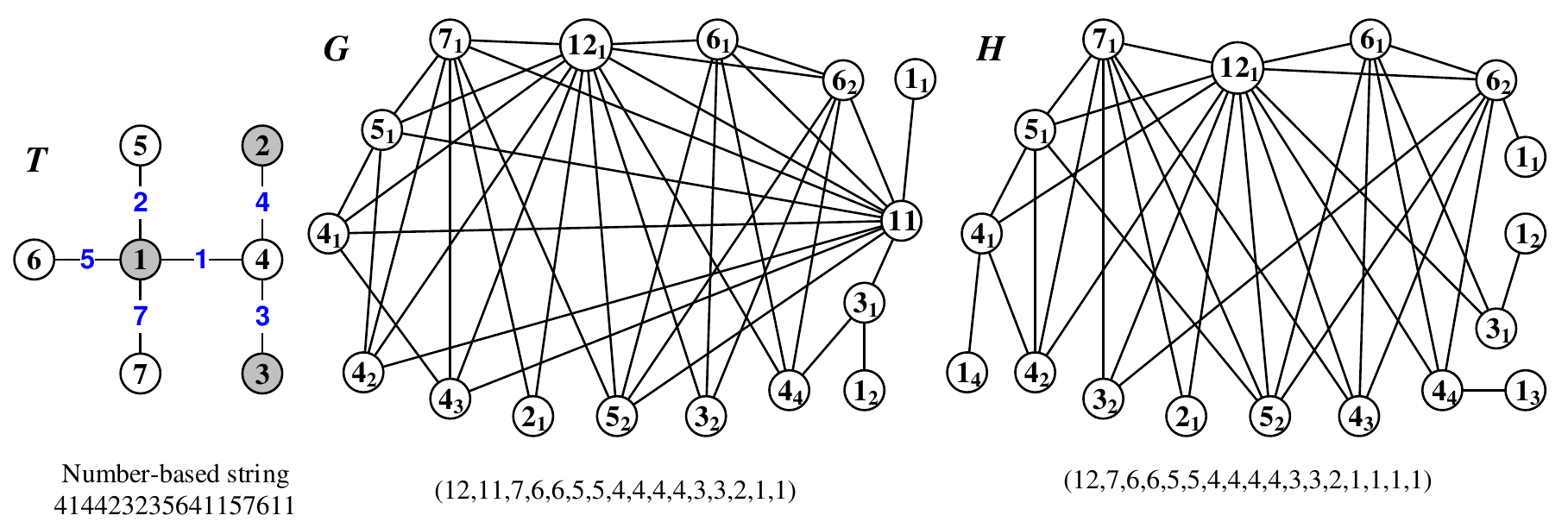}\\
\caption{\label{fig:55-string-degree-sequence}{\small A number-based string corresponds two degree-sequences $\textrm{deg}(G)=$(12, 11, 7, 6, 6, 5, 5, 4, 4, 4, 4, 3, 3, 2, 1, 1) and $\textrm{deg}(H)=$(12, 7, 6, 6, 5, 5, 4, 4, 4, 4, 3, 3, 2, 1, 1, 1, 1).}}
\end{figure}

\begin{defn} \label{defn:degree-vertex-total-coloring}
$^*$ Suppose that a graph $G$ has its own vertex set $V(G)=\bigcup _{\delta\leq k\leq \Delta}V^d_k$, such that $V^d_i\cap V^d_j=\emptyset $ for $i\neq j$, each vertex $u_{k,i}\in V^d_k$ is colored as $f(u_{k,i})=k_i$ for $i\in [1,m_k]$ with $m_k=|V^d_k|$ and $\delta\leq k\leq \Delta$. We call $f$ a \emph{degree-vertex coloring} of $G$. The edge induced coloring $f^*(xy)=F\langle f(x),f(y)\rangle$ under a function $F$ for each edge $xy\in E(G)$, so $h_F=\langle f,f^* \rangle$ is called a \emph{degree-vertex total coloring} of $G$.\qqed
\end{defn}

\begin{rem}\label{rem:333333}
In Fig.\ref{fig:55-string-degree-sequence}, two graphs $G$ and $H$ admit degree-vertex colorings. About Definition \ref{defn:degree-vertex-total-coloring}, we have

(i) There are many functions $F$ for $f^*(xy)=F\langle f(x),f(y)\rangle$, refer to Definition \ref{defn:operations-nb-strings}, Eq.(\ref{eqa:operations-number-based-trings}) and Eq.(\ref{eqa:other-operations-number-based-trings}).

(ii) The number-based strings generated from the Topcode-matrix $T_{code}(G)$ are more complex for severing topological authentications.

Clearly, a number-based string $s(M)$ induced from the Topcode-matrix $T_{code}(H_m)$ defined in Definition \ref{defn:complete-triangular-coincided-mpgs} is complex than any number-based string $s(n)$ generated from $T_{code}(G)$, and $s(M)$ is compressed into $s(n)$ like Hash function. Conversely, the number-based string $s(n)$ can be expended into $s(M)$. \paralled
\end{rem}

\section*{Acknowledgment}
The author, \emph{Bing Yao}, was supported by the National Natural Science Foundation of China under grants No. 61163054, No. 61363060 and No. 61662066.

My students have done a lot of works on graph labelings and graph colorings, the main members are: Dr. Xiangqian Zhou (School of Mathematics and Statistics, Huanghuai University, Zhu Ma Dian, Henan); Postdoctoral Hongyu Wang, Dr. Xiaomin Wang, Dr. Fei Ma, Dr. Jing Su, Dr. Hui Sun (School of Electronics Engineering and Computer Science, Peking University, Beijing); Postdoctoral Xia Liu (School of Mathematics and Statistics, Northwestern Polytechnical University, Xian); Dr. Chao Yang (School of Mathematics, Physics and Statistics, Shanghai University of Engineering Science, Shanghai); Inst. Meimei Zhao (College of Science, Gansu Agricultural University, Lanzhou); Inst. Sihua Yang (School of Information Engineering, Lanzhou University of Finance and Economics, Lanzhou); Inst. Jingxia Guo (Lanzhou University of technology, Lanzhou); Inst. Wanjia Zhang (College of Mathematics and Statistics, Hotan Teachers College, Hetian); Dr. Xiaohui Zhang (Computer College of Qinghai Normal University, Xining); Dr. Lina Ba (School of Mathematics and Statistics, Lanzhou University, Lanzhou); Inst. Lingfang Jiang, Haixia Tao, Jiajuan Zhang, Xiyang Zhao, Dr. Tianjiao Dai (Mathematics School of Shandong University, Jinan), Yaru Wang, Dr. Yichun Li (Mathematics School of Shandong University, Jinan), Inst. Yarong Mu (Lanzhou College of Information Science and Technology, Lanzhou).

Thanks for the teachers of Graph Coloring/labeling Symposium, they are: Prof. Mingjun Zhang and Prof. Guoxing Wang (School of Information Engineering, Lanzhou University of Finance and Economics, Lanzhou); Prof. Ming Yao (College of Information Process and Control Engineering, Lanzhou Petrochemical Polytechnic, Lanzhou); Inst. Yirong Sun (College of Mathematics and Statistics, Northwest Normal University, Lanzhou); Prof. Lijuan Qi (Department of basic courses, Lanzhou Institute of Technology, Lanzhou); Prof. Jianmin Xie (College of Mathematics, Lanzhou City University, Lanzhou).

\newpage

\textbf{Appendix A}

\begin{center}
\textbf{Table-1.} The number $G_p$ of graphs of $p$ vertices \cite{Harary-Palmer-1973}
\begin{tabular}{c|l|l}
$p$&$G_p$&bits\\
\hline
6&156&7\\
7&1044&10\\
8&12346&14\\
9&274668&18\\
10&12005168&24\\
11&1018997864&30\\
12&165091172592&37\\
13&50502031367952&46\\
14&29054155657235488&55\\
15&31426485969804308768&65\\
16&64001015704527557894928&76\\
17&245935864153532932683719776&88\\
18&1787577725145611700547878190848&100\\
19&24637809253125004524383007491432768&114\\
20&645490122795799841856164638490742749440&129\\
21&32220272899808983433502244253755283616097664&145\\
22&3070846483094144300637568517187105410586657814272&161\\
\end{tabular}
\end{center}

where $G_{p}\approx 2^{\textrm{bits}}$ for $p=6,7,\dots ,22$, and another two numbers are
$${
\begin{split}
G_{23}&=559946939699792080597976380819462179812276348458981632\approx 2^{179}\\
G_{24}&=195704906302078447922174862416726256004122075267063365754368\approx 2^{197}.
\end{split}}$$

\begin{center}
\textbf{Table-2.} The numbers of digraphs and connected digraphs of $p$ vertices \cite{Harary-Palmer-1973}
\begin{tabular}{c|ll}
$p$ & Digraphs & Connected digraphs\\
\hline
1&1&1\\
2&3&2\\
3&16&13\\
4&218&199\\
5&9,608&9,364\\
6&1,540,944&1,530,843\\
7&882,033,440&880,471,142\\
8&1,793,359,192,848&1,792,473,955,306\\
\end{tabular}
\end{center}


\begin{center}
\textbf{Table-3.} The numbers of trees of $p$ vertices \cite{Harary-Palmer-1973}
\end{center}
\begin{center}
\begin{tabular}{c|ll}
$p$&$t_p$&$T_p$\\
\hline
7&11&48\\
8&23&115\\
9&47&286\\
10&106&719\\
11&235&1,842\\
12&551&4,766\\
13&1,301&12,486\\
14&3,159&32,973\\
15&7,741&87,811\\
16&19,320&235,381\\
\end{tabular}\qquad
\begin{tabular}{c|ll}
$p$&$t_p$&$T_p$\\
\hline
17&48,629&634,847\\
18&123,867&1,721,159\\
19&317,955&4,688,676\\
20&823,065&12,826,228\\
21&2,144,505&35,221,832\\
22&5,623,756&97,055,181\\
23&14,828,074&268,282,855\\
24&39,299,897&743,724,984\\
25&104,636,890&2,067,174,645\\
26&279,793,450&5,759,636,510\\
\end{tabular}
\end{center}
where $t_p$ is the number of trees of $p$ vertices, and $T_p$ is the number of rooted trees of $p$ vertices.

\begin{center}
\textbf{Table-3.} The paths $P_j$ for $j\in [1,12]$ for forming every-zero graphic groups.
{\small
\begin{equation}\label{eqa:path-group-paths-Cayley-formula}
\centering
{
\begin{split}
\begin{array}{c|ccccccc}
\textrm{Group} & x_1 & \textcolor[rgb]{0.00,0.00,1.00}{a} & x_2 & \textcolor[rgb]{0.00,0.00,1.00}{b} & x_3 & \textcolor[rgb]{0.00,0.00,1.00}{c} & x_4\\
\hline
P_j=P_{j,0} & g_j(x_1) & g_j(x_1x_2) & g_j(x_2) & g_j(x_2x_3) & g_j(x_3) & g_j(x_3x_4) & g_j(x_4)\\
P_{j,r} & g_{j}(x_1)+r & g_{j,r}(x_1x_2) & g_{j}(x_2)+r & g_{j,r}(x_2x_3) & g_{j}(x_3)+r & g_{j,r}(x_3x_4) & g_{j}(x_4)+r\\
\hline
P_{1,0} & 1 & \textcolor[rgb]{0.00,0.00,1.00}{\textbf{3}} & 2 & \textcolor[rgb]{0.00,0.00,1.00}{\textbf{5}} & 3 & \textcolor[rgb]{0.00,0.00,1.00}{\textbf{7}} & 4\\
P_{1,1} & 2 & \textcolor[rgb]{0.00,0.00,1.00}{\textbf{5}} & 3 & \textcolor[rgb]{0.00,0.00,1.00}{\textbf{7}} & 4 & \textcolor[rgb]{0.00,0.00,1.00}{\textbf{5}} & 1\\
P_{1,2} & 3 & \textcolor[rgb]{0.00,0.00,1.00}{\textbf{7}} & 4 & \textcolor[rgb]{0.00,0.00,1.00}{\textbf{5}} & 1 & \textcolor[rgb]{0.00,0.00,1.00}{\textbf{3}} & 2\\
P_{1,3} & 4 & \textcolor[rgb]{0.00,0.00,1.00}{\textbf{5}} & 1 & \textcolor[rgb]{0.00,0.00,1.00}{\textbf{3}} & 2 & \textcolor[rgb]{0.00,0.00,1.00}{\textbf{5}} & 3\\
\hline
P_{2,0} & 1 & \textcolor[rgb]{0.00,0.00,1.00}{\textbf{3}} & 2 & \textcolor[rgb]{0.00,0.00,1.00}{\textbf{6}} & 4 & \textcolor[rgb]{0.00,0.00,1.00}{\textbf{7}} & 3\\
P_{2,1} & 2 & \textcolor[rgb]{0.00,0.00,1.00}{\textbf{5}} & 3 & \textcolor[rgb]{0.00,0.00,1.00}{\textbf{4}} & 1 & \textcolor[rgb]{0.00,0.00,1.00}{\textbf{5}} & 4\\
P_{2,2} & 3 & \textcolor[rgb]{0.00,0.00,1.00}{\textbf{7}} & 4 & \textcolor[rgb]{0.00,0.00,1.00}{\textbf{6}} & 2 & \textcolor[rgb]{0.00,0.00,1.00}{\textbf{3}} & 1\\
P_{2,3} & 4 & \textcolor[rgb]{0.00,0.00,1.00}{\textbf{5}} & 1 & \textcolor[rgb]{0.00,0.00,1.00}{\textbf{4}} & 3 & \textcolor[rgb]{0.00,0.00,1.00}{\textbf{5}} & 2\\
\hline
P_{3,0} & 1 & \textcolor[rgb]{0.00,0.00,1.00}{\textbf{4}} & 3 & \textcolor[rgb]{0.00,0.00,1.00}{\textbf{5}} & 2 & \textcolor[rgb]{0.00,0.00,1.00}{\textbf{6}} & 4\\
P_{3,1} & 2 & \textcolor[rgb]{0.00,0.00,1.00}{\textbf{6}} & 4 & \textcolor[rgb]{0.00,0.00,1.00}{\textbf{7}} & 3 & \textcolor[rgb]{0.00,0.00,1.00}{\textbf{4}} & 1\\
P_{3,2} & 3 & \textcolor[rgb]{0.00,0.00,1.00}{\textbf{4}} & 1 & \textcolor[rgb]{0.00,0.00,1.00}{\textbf{5}} & 4 & \textcolor[rgb]{0.00,0.00,1.00}{\textbf{6}} & 2\\
P_{3,3} & 4 & \textcolor[rgb]{0.00,0.00,1.00}{\textbf{6}} & 2 & \textcolor[rgb]{0.00,0.00,1.00}{\textbf{3}} & 1 & \textcolor[rgb]{0.00,0.00,1.00}{\textbf{4}} & 3\\
\hline
P_{5,0} & 1 & \textcolor[rgb]{0.00,0.00,1.00}{\textbf{5}} & 4 & \textcolor[rgb]{0.00,0.00,1.00}{\textbf{6}} & 2 & \textcolor[rgb]{0.00,0.00,1.00}{\textbf{5}} & 3\\
P_{5,1} & 2 & \textcolor[rgb]{0.00,0.00,1.00}{\textbf{3}} & 1 & \textcolor[rgb]{0.00,0.00,1.00}{\textbf{4}} & 3 & \textcolor[rgb]{0.00,0.00,1.00}{\textbf{7}} & 4\\
P_{5,2} & 3 & \textcolor[rgb]{0.00,0.00,1.00}{\textbf{5}} & 2 & \textcolor[rgb]{0.00,0.00,1.00}{\textbf{6}} & 4 & \textcolor[rgb]{0.00,0.00,1.00}{\textbf{5}} & 1\\
P_{5,3} & 4 & \textcolor[rgb]{0.00,0.00,1.00}{\textbf{7}} & 3 & \textcolor[rgb]{0.00,0.00,1.00}{\textbf{4}} & 1 & \textcolor[rgb]{0.00,0.00,1.00}{\textbf{3}} & 2\\
\end{array}
\end{split}}
\end{equation}
}
\end{center}

\newpage

\textbf{Appendix B}

\begin{tabular}{ll}
\textrm{Symbol} & \textrm{Description}\\
\hline
$V(G)$& vertex set of $G$\\
$E(G)$& edge set of $G$\\
$(p,q)$-graph $G$ & $p=|V(G)|$ and $q=|E(G)|$\\
$\overline{G}$ & complement of $G$\\
$\delta(G)$& minimum degree of $G$\\
$\Delta(G)$& maximum degree of $G$\\
$\chi (G)$& chromatic number of $G$\\
$\chi\,'(G)$& chromatic index of $G$\\
$\chi\,''(G)$& total chromatic number of $G$\\
$n_d(G)$& number of vertices of degree $d$ in $G$\\
$D_G(u,v)$ & distance between two vertices $u$ and $v$ in $G$\\
$D(G)$& diameter of $G$\\
$\textbf{T}_{\textbf{a}}\langle\textbf{X},\textbf{Y}\rangle$& topological authentication of multiple variables\\
$P_{ub}(\textbf{X})$& topological public-key vector\\
$P_{ri}(\textbf{Y})$& topological private-key vector\\
$\textbf{F}=(\theta_1,\theta_2,\dots ,\theta_m)$& operation vector\\
CIA-bases & Confidentiality, Integrity and Availability\\
$s(n)=c_1c_2\cdots c_n$ & number-based string with $c_i\in [0,9]$\\
$G\cong H$ & $G$ is isomorphic to $H$\\
$G\rightarrow H$ & $G$ is graph homomorphism into $H$\\
$G-uv$ & edge-removed graph\\
$G+xy$ & edge-added graph\\
$G+E\,'$ & edge-set-added graph, $E\,'\subset E(\overline{G})$\\
$G\ominus G\,'$ & edge symmetric graph, where $G\cong G\,'$ \\
$G\wedge u$ & vertex-split graph\\
$\odot_k \langle G,H\rangle=G[\odot_k]H$ & vertex-coincided graph of two vertex-disjoint graphs $G$ and $H$\\
$G\wedge uv$ &  edge-split graph\\
$G\wedge W$ & $W$-split graph\\
$G=G_1[\ominus ^W_k]G_2$ & $W$-coincided graph of two vertex-disjoint graphs $G_1$ and $G_2$ with\\
&  $W\subset G_i$ for $i=1,2$\\
$G=G_1[\ominus ^{cyc}_k]G_2$ & $W$ is a cycle\\
$G=G_1[\ominus ^{path}_k]G_2$ & $W$ is a path\\
$G=G_1[\ominus ^{tree}_k]G_2$ & $W$ is a tree\\
$G=G_1[\ominus ^{K_n}_k]G_2$ & $W$ is a complete graph $K_n$\\
$G=G^C_{out}[\ominus^C_k]G^C_{in}$ & $G^C_{out}$ and $G^C_{in}$ are semi-maximal planar graphs of a maximal\\
&  planar graph $G$ having a $k$-cycle $C$\\
\end{tabular}

\newpage

\begin{tabular}{ll}
\textrm{Symbol} & \textrm{Description}\\
\hline
$\kappa_{W}(G)$ & $W$-type graph-split connection number of $G$\\
$\kappa_{path}(G)$ & path-split connection number of $G$\\
$\kappa_{cycle}(G)$ & cycle-split connection number of $G$\\
$\kappa_{tree}(G)$ & tree-split connection number of $G$\\
$G(W\rightarrow u_0)$ & $W$-contracted graph\\
$G(u_0\rightarrow W)$ & $W$-expanded graph\\
$T_{code}(G)$ & Topcode-matrix of $G$\\
$T_{code}(G)\uplus T_{code}(H)$ & union of two Topcode-matrices\\
$E_{col}(G)=(c_{i,j})_{p\times p}$ & adjacent e-value matrix of $G$\\
$VE_{col}(G)=(a_{i,j})_{(p+1)\times (p+1)}$ & adjacent ve-value matrix of $G$\\
$D_{sc}(\textbf{\textrm{d}})$ & colored degree-sequence matrix of a degree-sequence $\textbf{\textrm{d}}$\\
$H\overline{\ominus} iG$ & complex graph\\
$H\rightarrow_{anti} G$ & $H$ is graph anti-homomorphism to $G$\\
$H\rightarrow_{oper} G$ & $H$ is graph-operation homomorphism to $G$\\
$H\rightarrow_{\textrm{v-coin}} G$ & $H$ is vertex-coincided homomorphism to $G$\\
$H\rightarrow_{split}  G$ & $H$ is vertex-split homomorphism to $G$\\
\end{tabular}


{\footnotesize
\begin{thebibliography}{9}

\setlength{\parskip}{0pt}

\bibitem{Albertson-Berman-large-forest-1979} Albertson M.O. and Berman D.M. A conjecture on planar graphs. in: J. A. Bondy, U.S.R.Murty (Eds.), Graph Theory and Related Topics, \textbf{357}, 1979.
\bibitem{Barabasi-Albert1999} Albert-L\'{a}szl\'{o} Barab\'{a}si and Reka Albert. Emergence of scaling in random networks. Science \textbf{286} (1999) 509-512.
\bibitem{Acharya-Hegde-Arithmetic-1990} B. D. Acharya and S. M. Hegde, Arithmetic graphs, J. Graph Theory, \textbf{14} (1990) 275-299.
\bibitem{Bang-Jensen-Gutin-digraphs-2007}Jorgen Bang-Jensen, Gregory Gutin. Digraphs Theory, Algorithms and Applications. Springer-Verlag, 2007.
\bibitem{Bernstein-Buchmann-Dahmen-Quantum-2009}Daniel J. Bernstein, Johannes Buchmann, Erik Dahmen. Post-Quantum Cryptography. Springer-Verlag Berlin Heidelberg, 2009. (245 pages, 110 reference papers)
\bibitem{Bennett-DiVincenzo-2000-Nature} Bennett, C., DiVincenzo, D. Quantum information and computation. Nature \textbf{404}, 247-255 (2000). https://doi.org/10.1038/35005001
\bibitem{Mandelbrot-Benoit-B-1967}Mandelbrot, Benoit B. (5 May 1967). ``How long is the coast of Britain? Statistical self-similarity and fractional dimension''. Science. New Series. \textbf{156} (3775): 636-638. Bibcode: 1967 Sci. 156. 636M. doi: 10.1126/science.156.3775.636. PMID 17837158. Retrieved 11 January 2016.

\bibitem{Blum-Ding-Thaeler-Cheng2004}J. Blum, M. Ding, A. Thaeler, and X. Cheng. Connected dominating set in sensor networks and manets. Handbook of Combinatorial Optimization (2004) 329-369.



\bibitem{Bondy-2008} J. A. Bondy, U. S. R. Murty. Graph Theory. Springer London, 2008. DOI: 10.1007/978-1-84628-970-5. J. Adrian Bondy and U. S. R. Murty, Graph Theory with Application. The MaCmillan Press Ltd., London, 1976.
\bibitem{Comellas-Fertin-Raspaud-2004}Francesc Comellas, Guillaume Fertin, Andre Raspaud. Recursive graphs with small-world scale-free properties, Physical Review E \textbf{69}, 03710-1, 037104-4 4 (2004).
\bibitem{Erdos-Saks-Sos-1986} Erd\"{o}s P., Saks M. and Sos V. Maximum induced trees in graphs. J. Combin. Ser. B, 1986, \textbf{41}: 61-79.
\bibitem{Gallian2021} Joseph A. Gallian. A Dynamic Survey of Graph Labeling. The electronic journal of combinatorics, \# DS6, Twenty-fourth edition, December 9, 2021. (576 pages, 3028 reference papers, over 200 graph labellings)

\bibitem{Gena-Hahn-Claude-Tardif-1997} Ge\v{n}a Hahn and Claude Tardif. Graph homomorphisms structure and symmetry. Graph Symmetry. NATO Adv. Sci. Inst. Ser. C. Math. Phys. Sci. \textbf{497}, 107-166 (1997). (60 pages, 142 reference papers)
\bibitem{Harary-Palmer-1973}Harary F. and Palmer E. M. Graphical enumeration. Academic Press, 1973.
\bibitem{Daniele-Micciancio-Oded-Regev} Daniele Micciancio and Oded Regev. Lattice-based Cryptography. November 7, 2008, preprint.
\bibitem{Greenwell-Kronk-Uniquely-4c-1973} Greenwell D. and Kronk H V. Uniquely Line-colorable Graphs, Canadian Mathematics Bulletin, 1973, \textbf{16} (4): 525-529. DOI: 10.4153/CMB-1973-086-2.

\bibitem{Marumuthu-G-2015}Marumuthu. G. Super Edge Magic Graceful Labeling of Generalized Petersen Graphs, Discrete Mathematics (2015): 235-241.

\bibitem{Cami-Rosso2017Abc-conjecture}Cami Rosso. Abc conjecture-The Enormity of Math. 2/24/2017, https://www.linkedin.com/pulse/abc-conjecture-enormity-math-cami-rosso
\bibitem{Burris-Schelp} A. C. Burris and R. H. Schelp. Vertex-Distingushing proper edge-colorings. Journal of Graph Theory, \textbf{26} (2) (1997) 73-82.
\bibitem{Ringel-G} Gerhard Ringel. Problem 25, Theory of Graphs and Its Applications. Proc. Int. Symp. Smolenice (June 1963), Czech Acad. Sci. Prague, Czech. (1964), 162.
\bibitem{Deo-Nikoloski-Suraweera-2002} N. Deo, Z. Nikoloski, F. Suraweera. Generation of Graceful Trees. The 33rd Southeastern International Conference on Combinatorics, Graph Theory and Computing, Boca Raton, FL, Mar. (2002), 4-8.
\bibitem{Sheppard-D-A-1976} Sheppard, D. A. The factorial representation of balanced labeled graphs. Discrete Math. \textbf{15} (1976), 379-388.
\bibitem{Garey-Johnson1979} M. Garey, D. Johnson, Computers and Intractability: A Guide to the Theory of NP-completeness, Freeman, 1979.
\bibitem{Fernau-Kneis-Kratsch-Langer-Liedloff-Raible-Rossmanith2011}Henning Fernau, Joachim Kneis, Dieter Kratsch, Alexander Langer, Mathieu Liedloff, Daniel Raible, Peter Rossmanith, An exact algorithm for the Maximum Leaf Spanning Tree problem, Theoretical Computer Science \textbf{412} (2011) 6290-6302.


\bibitem{X-Q-Liu-J-Xu-extending-contracting-2017} Xiaoqing Liu and Jin Xu. A special type of domino extending-contracting operations. Journal of Electronics and Information Technology \textbf{39} (1), 221-230 (2017).
\bibitem{Jin-Xu-Maximal-Science-Press-2019}Jin Xu. Maximal Planar Graph Theory-Structure, Construction, Coloring (Volume One, Chinses), Science Press, Beijing, January 2019.
\bibitem{Jin-Xu-55-56-configurations-arXiv-2107-05454v1} Jin Xu. 55- and 56-configurations are reducible. arXiv:2107.05454v1 [math.CO] 4 Jul 2021.



\bibitem{GB2312-80} ``GB2312-80 Encoding of Chinese characters'' cited from The Compilation Of National Standards For Character Sets And Information Coding, China Standard Press, 1998.


\bibitem{Xiang-en-CHEN-LI-2018}Xiang'en Chen, Ting Li. The Structure of $(k,l)$-recursive Maximal Planar Graph. Journal of Electronics and Information Technology, 2018, \textbf{40} (9): 2281-2286. DOI: 10.11999/JEIT171021.


\bibitem{Ma-Wang-Wang-Yao-Theoretical-Computer-Science-2018}Fei Ma, Ping Wang, Bing Yao. Generating Fibonacci-model as evolution of networks with vertex-velocity and time-memory. Physica A \textbf{527} (2019) 121295 (Published by Elsevier B.V.). https://doi.org/10.1016/j.physa.2019.121295

\bibitem{Ma-Wang-Wang-Chaos-013136-2020} F. Ma, X. M. Wang and P. Wang. An ensemble of random graphs with identical degree distribution. Chaos. \textbf{30}, 013136 (2020).
\bibitem{Ma-Luo-Wang-Zhu-Chaos-113120-2020} F. Ma, X. D. Luo, P. Wang and R. B. Zhu. Random growth networks with exponential degree distribution. Chaos. \textbf{30}, 113120 (2020).

\bibitem{Ma-Ma-Yao-Information-2018} F. Ma and B. Yao. The number of spanning trees of a class of self-similar fractal models. Information Processing Letters, \textbf{136}, 64-69 (2018).
\bibitem{Ma-Ma-Yao-Physica-492-2018} F. Ma and B. Yao. A family of small-world network models built by complete graph and iteration-function. Physica A, \textbf{492}, 2205-2219 (2018).


\bibitem{Ma-Su-Hao-Yao-Yan-Physica-492-2018} F. Ma, J. Su, Y. X. Hao, B. Yao and G. H. Yan. A class of vertex-edge-growth small-world network models having scale-free, self-similar and hierarchical characters. Physica A, \textbf{492}, 1194-1205 (2018).
\bibitem{Ma-Yao-Eur-Phys-2018} F. Ma and B. Yao. A recursive method for calculating the total number of spanning trees and its applications in self-similar small-world scale-free network models. Eur. Phys. J. B. \textbf{91}: 82 (2018).
\bibitem{Ma-Yao-Physica-2017} F. Ma and B. Yao. The relations between network-operation and topological-property in a scale-free and small-world network with community structure. Physica A, \textbf{484}, 182-193 (2017).


\bibitem{Subbiah-Pandimadevi-Chithra2015} S. P. Subbiah, J. Pandimadevi, R. Chithra. Super total graceful graphs. Electronic Notes in Discrete Mathematics \textbf{48} (2015), 301-304.


\bibitem{Sun-Zhang-Zhao-Yao-2017}Hui Sun, Xiaohui Zhang, Meimei Zhao and Bing Yao. New Algebraic Groups Produced By Graphical Passwords Based On Colorings And Labellings. ICMITE 2017, MATEC Web of Conferences \textbf{139}, 00152 (2017), DOI: 10. 1051/matecconf/201713900152
\bibitem{Sun-Zhang-Yao-ICMITE-2017}Hui Sun, Xiaohui Zhang, and Bing Yao. Strongly $(k,d)$-Graphical Labellings For Designing Graphical Passwords In Communication. ICMITE 2017, MATEC Web of Conferences \textbf{139}, 00204 (2017). DOI: 10.1051/matecconf/201713900204
\bibitem{Sun-Zhang-Yao-IAEAC-2017}Hui Sun, Xaohui Zhang, Bing Yao. On Operation Phenomena In Networks With Hub-Rings. 2017 IEEE 2nd Advanced Information Technology, Electronic and Automation Control Conference (IEEE IAEAC 2017), March 25-26, 2017, Chongqing China, 2017: 82-85. ISBN-13:9781467389778


\bibitem{su-yan-yao-2018}Jing Su, Guanghui Yan, Bing Yao. Generalized Edge Magic Labellings Of Apollonian Network Models. Journal of Jilin University, 2018, \textbf{56} (3): 567-572.
\bibitem{Su-Wang-Yao-Image-labelings-2021-MDPI} Jing Su, Hongyu Wang and Bing Yao. Matching-Type Image-Labelings of Trees. Mathematics 2021,1,0. https://doi.org/10.3390/math1010000. Licensee MDPI, Basel, Switzerland.


\bibitem{Tian-Li-Peng-Yang-2021-102212}Yanzhao Tian, Lixiang Li, Haipeng Peng and Yixian Yang. Achieving flatness: Graph labeling can generate graphical honeywords. Computers and Security, \textbf{104} (2021) 102212. DOI: 10.1016/j.cose.2021.102212.

\bibitem{Unterschutz-Turau2012}S. Untersch\"{u}tz, and V. Turau. Construction of Connected Dominating Sets in Large-Scale MANETs Exploiting Self-Stabilization. http://www.ti5.tu-harburg.de/research/projects/heliomesh/


\bibitem{Bing-Yao-Fei-Ma-13354v1-2022}Bing Yao, Fei Ma. Graph Set-colorings And Hypergraphs In Topological Coding. arXiv: 2201.13354v1 [cs.CR] 31 Jan 2022.
\bibitem{Yao-Zhang-Wang-Su-Integer-Decomposing-2021}Bing Yao, Wanjia Zhang Hongyu Wang, Jing Su. Integer-Decomposing Topological Authentication Problem For Post-Quantum Cryptosystem. 2021 submitted to the 2021 IEEE 4th Advanced Information Management, Communicates, Electronic and Automation Control Conference (IEEE IEEE IMCEC 2021).
\bibitem{Wang-Yao-Su-Wanjia-Zhang-2021-IMCEC} Hongyu Wang, Bing Yao, Jing Su, Wanjia Zhang. Number-Based Strings/Passwords From Imaginary Graphs Of Topological Coding For Encryption. submitted to IMCEC 2021.
\bibitem{Yao-Wang-Ma-Wang-Degree-sequences-2021}Bing Yao, Xiaomin Wang, Fei Ma, Hongyu Wang. Number-Based Strings And Degree-sequences Of Topological Cryptography. 2021 IEEE 5th Advanced Information Technology, Electronic and Automation Control Conference (IAEAC 2021) will be held on March 12-14, 2021 in Chongqing China.
\bibitem{Yao-Wang-Su-arameterized-2020}Bing Yao, Hongyu Wang, Jing Su. Parameterized Total Colorings Of Trees. submitted 2020.
\bibitem{Yao-Wang-Su-Jianmin-Xie-2021-Conference}Bing Yao, Hongyu Wang, Jing Su, Jianmin Xie. Degree-sequence Homomorphisms For Homomorphic Encryption Of Information. 2021 submitted to the 2021 IEEE 5th Information Technology,Networking,Electronic and Automation Control Conference (IEEE ITNECTNECTNEC 2021).
\bibitem{Yao-Wang-Liu-ice-flower-2020arXiv}Bing Yao, Hongyu Wang, Xia Liu, Xiaomin Wang, Fei Ma, Jing Su, Hui Sun. Ice-Flower Systems And Star-graphic Lattices. arXiv: 2005.02823v1 [Math.CO] 6 May 2020.
\bibitem{Yao-Yang-Yao-2020-distinguishing}Bing Yao, Chao Yang, Ming Yao. Coding Techniques From Distinguishing Colorings In Topological Coding. 2020 IEEE 9th Joint International Information Technology and Artificial Intelligence Conference (ITAIC 2020): 77-83.
\bibitem{Bing-Yao-2020arXiv}Bing Yao. Graphic Lattices and Matrix Lattices Of Topological Coding. arXiv: 2005.03937v1 [cs.IT] 8 May 2020.
\bibitem{Yao-Su-Wang-Hui-Sun-ITAIC2020}Bing Yao, Jing Su, Hongyu Wang, Hui Sun. Sequence-type Colorings of Topological Coding Towards Information Security. 2020 IEEE 9th Joint International Information Technology and Artificial Intelligence Conference (ITAIC 2020): 1226-1231.
\bibitem{Yao-Sun-Wang-Su-Maximal-Planar-Graphs-2021}Bing Yao, Hui Sun, Hongyu Wang, Jing Su. Maximal Planar Graphs As Topological Authentications For Asymmetric Encryption. 2021 IEEE 2nd International Conference on Information Technology, Big Data and Artificial Intelligence (ICIBA 2021) pp 133-138.
\bibitem{Yao-Su-Sun-Wang-Graph-Operations-2021} Bing Yao, Jing Su, Hui Sun, Hongyu Wang. Graph Operations For Graphic Lattices and Graph Homomorphisms In Topological Cryptosystem. submitted to 2021 IEEE 5th Information Technology, Networking, Electronic and Automation Control Conference (ITNEC 2021).
\bibitem{Yao-Wang-2106-15254v1}Bing Yao, Hongyu Wang. Recent Colorings And Labelings In Topological Coding. arXiv: 2106.15254v1 [cs.IT] 29 Jun 2021.
\bibitem{Bing-Yao-Hongyu-Wang-arXiv-2020-homomorphisms}Bing Yao, Hongyu Wang. Graph Homomorphisms Based On Particular Total Colorings of Graphs and Graphic Lattices. arXiv: 2005.02279v1 [math.CO] 5 May 2020.
\bibitem{Yao-Wang-Ma-Su-Wang-Sun-2020ITNEC}Bing Yao, Hongyu Wang, Fei Ma, Jing Su, Xiaomin Wang, Hui Sun. On Real-Valued Total Colorings Towards Topological Authentication In Topological Coding. 2020 IEEE 4th Information Technology, Networking, Electronic and Automation Control Conference (ITNEC 2020). pp 1751-1756.
\bibitem{Yao-Wang-Su-Ma-Wang-Sun-ITNEC2020}Bing Yao, Hongyu Wang, Jing Su, Fei Ma, Xiaomin Wang, Hui Sun. Especial Total Colorings Towards Multiple Authentications In Network Encryption. 2020 IEEE 4th Information Technology, Networking, Electronic and Automation Control Conference (ITNEC 2020). pp 1617-1623.
\bibitem{Yao-Zhao-Zhang-Mu-Sun-Zhang-Yang-Ma-Su-Wang-Wang-Sun-arXiv2019}Bing Yao, Meimei Zhao, Xiaohui Zhang, Yarong Mu, Yirong Sun, Mingjun Zhang, Sihua Yang, Fei Ma, Jing Su, Xiaomin Wang, Hongyu Wang, Hui Sun. Topological Coding and Topological Matrices Toward Network Overall Security. arXiv:1909.01587v2 [cs.IT] 15 Sep 2019.


\bibitem{Yao-Zhang-Sun-Mu-Sun-Wang-Wang-Ma-Su-Yang-Yang-Zhang-2018arXiv}Bing Yao, Xiaohui Zhang, Hui Sun, Yarong Mu, Yirong Sun, Xiaomin Wang, Hongyu Wang, Fei Ma, Jing Su, Chao Yang, Sihua Yang, Mingjun Zhang. Text-based Passwords Generated From Topological Graphic Passwords. arXiv: 1809.04727v1 [cs.IT] 13 Sep 2018.
\bibitem{Yao-Sun-Zhang-Mu-Sun-Wang-Su-Zhang-Yang-Yang-2018arXiv}Bing Yao, Hui Sun, Xiaohui Zhang, Yarong Mu, Yirong Sun, Hongyu Wang, Jing Su, Mingjun Zhang, Sihua Yang, Chao Yang. Topological Graphic Passwords And Their Matchings Towards Cryptography. arXiv: 1808.03324v1 [cs.CR] 26 Jul 2018.
\bibitem{YAO-SUN-WANG-SU-XU2018arXiv}Bing Yao, Hui Sun, Hongyu Wang, Jing Su, Jin Xu. Graph Theory Towards New Graphical Passwords In Information Networks. arXiv: 1806.02929v1 [cs.CR] 8 Jun 2018.
\bibitem{Yao-Mu-Sun-Sun-Zhang-Wang-Su-Zhang-Yang-Zhao-Wang-Ma-Yao-Yang-Xie2019}Bing Yao, Yarong Mu, Yirong Sun, Hui Sun, Xiaohui Zhang, Hongyu Wang, Jing Su, Mingjun Zhang, Sihua Yang, Meimei Zhao, Xiaomin Wang, Fei Ma, Ming Yao, Chao Yang, Jianming Xie. Using Chinese Characters To Generate Text-Based Passwords For Information Security. arXiv:1907.05406v1 [cs.IT] 11 Jul 2019.


\bibitem{Yao-Mu-Sun-Zhang-Wang-Su-Ma-IAEAC-2018}Bing Yao, Yarong Mu, Hui Sun, Xiaohui Zhang, Hongyu Wang, Jing Su, Fei Ma. Algebraic Groups For Construction Of Topological Graphic Passwords In Cryptography. 2018 IEEE 3rd Advanced Information Technology, Electronic and Automation Control Conference (IAEAC 2018), 2211-2216.
\bibitem{Yao-Mu-Sun-Zhang-Wang-Su-2018} Bing Yao, Yarong Mu, Hui Sun, Xiaohui Zhang, Hongyu Wang, Jing Su. Connection Between Text-based Passwords and Topological Graphic Passwords. 2018 4th Information Technology and Mechatronics Engineering Conference (IEEE ITOEC 2018) Chongqing, Dec. 14-16, 2018, 1090-1096.
\bibitem{Yao-Sun-Zhang-Li-Zhang-Xie-Yao-2017-Tianjin-University}Bing Yao, Hui Sun, Xiaohui Zhang, Jingwen Li, Mingjun Zhang, Jianmin Xie, Ming Yao. Applying graph theory to graphical passwords. 2017 Academic Annual Conference, Specialized Committee Of Graph Theory And System Optimization, Chinese Society Of Electronics, Circuits And Systems, Tianjin University, August, 2017: 12-13.

\bibitem{Yao-Sun-Zhao-Li-Yan-ITNEC-2017} Bing Yao, Hui Sun, Meimei Zhao, Jingwen Li, Guanghui Yan. On Coloring/Labelling Graphical Groups For Creating New Graphical Passwords. (ITNEC 2017) 2017 IEEE 2nd Information Technology, Networking, Electronic and Automation Control Conference. 2017: 1371-1375.

\bibitem{Bing-Yao-Cheng-Yao-Zhao2009}Bing Yao, Hui Cheng, Ming Yao and Meimei Zhao. A Note on Strongly Graceful Trees. Ars Combinatoria \textbf{92} (2009), 155-169.
\bibitem{Yao-Liu-Yao-2017}Bing Yao, Xia Liu and Ming Yao. Connections between labellings of trees. Bulletin of the Iranian Mathematical Society, ISSN: 1017-060X (Print) ISSN: 1735-8515 (Online), Vol. \textbf{43} (2017), 2, 275-283.

\bibitem{Yao-Ma-Wang-ITAIC2020}Bing Yao, Fei Ma, Xiaomin Wang. Optimal Design And Randomly Topological Coloring Of Dynamic Networks. 2020 IEEE 9th Joint International Information Technology and Artificial Intelligence Conference (ITAIC 2020): 226-231.
\bibitem{Yao-Sun-Zhang-Mu-Wang-Jin-Xu-2018} Bing Yao, Hui Sun, Xiaohui Zhang, Yarong Mu, Hongyu Wang, Jin Xu. New-type Graphical Passwords Made By Chinese Characters With Their Topological Structures. Procceding of 2018 2nd IEEE Advanced Information Management,Communicates,Electronic and AutoMation Control Conference (IMCEC 2018), 1606-1610.


\bibitem{Yao-Zhang-Yao-2007}Bing Yao, Zhong-fu Zhang and Ming Yao. A Class of Spanning Trees. International Journal of Computer, Mathematical Sciences and Applications, 1. 2-4, April-December 2007, 191-198.
\bibitem{Yao-Zhang-Wang-Sinica-2010}Bing Yao, Zhong-fu Zhang and Jian-fang Wang. Some results on spanning trees. Acta Mathematicae Applicatae Sinica, English Series, 2010, \textbf{26} (4), 607-616. DOI:10.1007/s10255-010-0011-4

\bibitem{Yao-Chen-Yao-Cheng2013JCMCC} Bing Yao, Xiang'en Chen, Ming Yao, Hui Cheng. On $(k,\lambda)$-magically total labeling of graphs. JCMCC (Journal of Combinatorial Mathematics and Combinatorial Computing, ISSN:8353-3026) 87(2013): 237-253.
\bibitem{Yao-Chen-Yang-Wang-Zhang-Zhang2012}B. Yao, X.-E Chen, C. Yang, H.-Y Wang, J.-J Zhang, X.-M Zhang. Spanning Trees And Dominating Sets In Scale-Free Networks. Proceeding of 2012 IET International Conference on Information Science and Control Engineering (ICISCE 2012), December 2012, Shenzhen, China. 111-115.
\bibitem{Yao-Ma-Su-Wang-Zhao-Yao-2016}Bing Yao, Fei Ma, Jing Su, Xiaomin Wang, Xiyang Zhao, Ming Yao. Scale-Free Multiple-Partite Models Towards Information Networks. Proceedings of 2016 IEEE Advanced Information Management, Communicates, Electronic and Automation Control Conference (IMCEC 2016) pp 549-554.

\bibitem{Yao-Wang-Su-Ma-Yao-Zhang-Xie2016}Bing Yao, Xiaomin Wang, Jing Su, Fei Ma, Ming Yao, Mingjun Zhang, and Jianmin Xie. Methods And Problems Attempt in Scale-Free Models From Complex Networks. Joint International Information Technology, Mechanical and Electronic Engineering Conference (JIMEC 2016), ISSN 2352-5401, Volume 59, pp57-61. ISBN: 978-94-6252-234-3. DOI: 10.2991/jimec-16.2016.11.
\bibitem{Yao-Liu-Zhang-Chen-Yao-2014}Bing Yao, Xia Liu, Wanjia Zhang, Xiang'en Chen, Ming Yao. Nested Growing Network Models for Researching The Internet of Things. Proceeding of 2014 IEEE 7th Joint International Information Technology and Artificial Intelligence Conference (ITAIC 2014). December 20-21, 2014, pp 450-454, Chongqing. IEEE Catalog Number: CFP1419L-PRT. ISBN: 978-1-4799-4420-0



\bibitem{Wang-Su-Yao-mixed-difference-2019}Hongyu Wang, Jing Su, Bing Yao. Proper Mixed-Difference Total Coloring Technique In Topological Coding And Network Encryption. submitted 2019.
\bibitem{Wang-Su-Sun-Yao-submitted-ITOEC2020}Hongyu Wang, Jing Su, Hui Sun, Bing Yao. Graphic Groups Towards Cryptographic Systems Resisting Classical And Quantum Computers. 2020 IEEE 5th Information Technology and Mechatronics Engineering Conference (ITOEC 2020), 1804-1808.
\bibitem{Wang-Xu-Yao-Key-models-Lock-models-2016}Hongyu Wang, Jin Xu, Bing Yao. The Key-models And Their Lock-models For Designing New Labellings Of Networks. Proceedings of 2016 IEEE Advanced Information Management, Communicates, Electronic and Automation Control Conference (IMCEC 2016) 565-5568.
\bibitem{Hongyu-Wang-2018-Doctor-thesis} Hongyu Wang. The Structure And Theoretical Analysis On Topological Graphic Passwords. Doctor's thesis. School of Electronics Engineering and Computer Science, Peking University, 2018.
\bibitem{Wang-Xu-Yao-2017-Twin}Hongyu Wang, Jin Xu, Bing Yao. Twin Odd-Graceful Trees Towards Information Security. Procedia Computer Science \textbf{107} (2017) 15-20. DOI: 10.1016/j.procs.2017.03.050
\bibitem{Wang-Xu-Yao-2016} Hongyu Wang, Jin Xu, Bing Yao. Exploring New Cryptographical Construction Of Complex Network Data. IEEE First International Conference on Data Science in Cyberspace. IEEE Computer Society, (2016):155-160.


\bibitem{Jianfang-Wang-Hypergraphs-2008} Jianfang Wang. The Information Hypergraph Theory. Science Press, Beijing, 2008.


\bibitem{Shuhong-Wu-Accurate-2007}Shuhong Wu. The Accurate Formulas of $A(n,k)$ and $P(n,k)$. Journal Of Mathematical Research And Exposition, \textbf{27}, NO.2 (2007) 437-444.

\bibitem{WU-Qi-qi-2001}Qiqi Wu. On the law of left shoulder and law of oblique line for constructing a large table of $P(n,k)$ quickly (in Chinese). Sinica (Chin. Ser.), 2001, \textbf{44} (5): 891-898.


\bibitem{Wang-Xiao-Yun-Liu-2014}Xiao-Yun Wang and Ming-Jie Liu. Survey of Lattice-based Cryptography. Journal of Cryptologic Research, 2014, \textbf{1} (1):13-27.


\bibitem{Wang-Su-Yao-2021-computer-science}Xiaoming Wang, Jing Su, Bing Yao. Algorithms Based on Lattice Thought for Graph Structure Similarity. Computer Science (Cinese) \textbf{48}, No.6A, 2021, 543-551. DOI: 10.11896/jsjkx.201100167.



\bibitem{Zhang-Yang-Yao-Frontiers-Computer-2021}Mingjun Zhang, Sihua Yang, Bing Yao. Exploring Relationship Between Traditional Lattices and Graph Lattices of Topological Coding. Journal of Frontiers of Computer Science and Technology (Chinese). 2021, \textbf{15} (11) 1-13. doi: 10.3778/j.issn.1673-9418.2010072.





\bibitem{Dorogovstev-2002}S. N. Dorogovstev, A. V. Goltsev, J. F. F. Mendes. Pseufractal scale-free web. Physiacal review 2002, \textbf{65}, 066122-066125.
\bibitem{Genio-Gross-Bassler2011}Charo I. Del Genio, Thilo Gross, and Kevin E. Bassler, Physical Review Letters \textbf{107}, 178701 (2011)

\bibitem{Lu-Su-Guo2013}Zhe-Ming Lu, Yu-Xin Su and Shi-Ze Guo(2013). Deterministic scale-free small-world networks of arbitrary order. Physica A. \textbf{392} (17):3555-3562.
\bibitem{Newman-Barabasi-Watts2006} M. E. J. Newman, A. -L. Barab\'{a}si, and D. J. Watts, The Structure and Dynamics of Networks. Princeton University Press, Princeton (2006).
\bibitem{Zhang-Zhou-Fang-Guan-Zhang2007}Zhang Zhongzhi, Zhou Shuigeng, Fang Lujun, Guan Jihong, Zhang Yichao. Maximal planar scale-free Sierpinski networks with small-world effect and power-law strength-degree correlation. EPL (Europhysics Letters), 2007, \textbf{79}: 38007.
\bibitem{Zhang-Comellas-Fertin-Rong-2006}Zhang Zhongzhi, Comellas Francesc, Fertin Guillaume, Rong Lili. High dimensional Apollonian networks. Journal of Physics A: Mathematical and General, 2006, \textbf{39} (8): 1811-1818.
\bibitem{Zhang-Rong-Guo2006}Zhongzhi Zhang, Lili Rong, Chonghui Guo. A deterministic small-world network created by edge iterations. Physica A \textbf{363} (2006) 567-572.






\bibitem{Zhou-Yao-Chen2013}Xiangqian Zhou, Bing Yao, Xiang'en Chen. Every Lobster Is Odd-elegant. Information Processing Letters \textbf{113} (2013): 30-33.
\bibitem{Zhou-Yao-Chen-Tao2012}Xiangqian Zhou, Bing Yao, Xiang'en Chen and Haixia Tao. A proof to the odd-gracefulness of all lobsters. Ars Combinatoria \textbf{103} (2012): 13-18.



\end{thebibliography}


}
\end{document}